\documentclass[conference,idxtotoc,10pt]{IEEEtran}   % ,page-numbering

\usepackage[T2A]{fontenc}
\usepackage[utf8]{inputenc}

\usepackage{tikz-uml}
\usepackage{tikz}
\tikzstyle{line}=[draw]
\tikzstyle{arrow}=[draw, -latex]

\usetikzlibrary{arrows,shapes}
\usetikzlibrary{arrows.meta}
\usetikzlibrary{positioning}

\let\umlnote\undefined  % resolve an overdefinition issue coming from tikz-uml
 % ...
\usepackage{pgf-umlcd}

\usepackage{rotating}

\usepackage{acronym}
\usepackage{algorithm}
\usepackage{algpseudocode}
\usepackage{amsmath}
\usepackage{amssymb}
\usepackage{amsthm}
\usepackage{arydshln}
\usepackage[figurename=Figure]{caption}   % workaround for renewcommand \figurename
\usepackage[short,nodayofweek,level,12hr]{datetime}  %%
\usepackage{flushend}
\usepackage[unicode]{hyperref}
\usepackage[acronym,toc,nonumberlist]{glossaries}  % nonumberlist drops referenced pages, add references by \index{} instead !
\usepackage{graphicx}
\usepackage{listings}
\usepackage{makeidx}
\usepackage{marvosym}
\usepackage{semantic}
\usepackage{stmaryrd}
\usepackage{subfigure}
\usepackage{syntax}   %% BUG: conflicts with \begin{xy}-environments!
\usepackage{textcomp}
\usepackage{thmtools}
\usepackage{ulem}
\usepackage[matrix,arrow,frame,curve]{xy}

\theoremstyle{plain} % definition
\newtheorem{thm}{Theorem}[section]

\newenvironment{gref}[1]{$\uparrow$\textit{#1}}{}
\newtheorem{theorem}[thm]{Theorem}
\newtheorem{definition}[thm]{Definition}
\newtheorem{corollary}[thm]{Corollary}
\newtheorem{lemma}[thm]{Lemma}
\newtheorem{conventions}[thm]{Convention}
\newtheorem{observation}[thm]{Observation}
\newtheorem{example}[thm]{Example}
\newtheorem{thesis}[thm]{Thesis}

\newcommand{\myglsentry}[1]{\texttt{\textit{#1}}.\\}
\newcommand{\myglsdesc}[1]{\textit{#1}}

\makeindex

\begin{document}
\xyoption{all}

\title{A Logical Programming Language as an Instrument\\ for Specifying and Verifying Dynamic Memory\\\includegraphics[width=3cm]{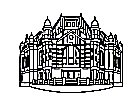}}

\author{
  \IEEEauthorblockN{Ren\'{e} Haberland}
  \IEEEauthorblockA{Saint Petersburg Electrotechnical University "\textit{LETI}"\\
  Saint Petersburg, Russia\\
  email: haberland1@mail.ru\\
  (original compatible with standard GOST R.7.0.11-2011)\\
  (permission to air granted on 29th January 2021; self-translated into English)\\
  \huge{2017}}
}

\maketitle

\tableofcontents

\pagestyle{plain} % page numbering

\begin{abstract}
%%%%%%%%%%%%%%%%%%%%%%%%%%%%%%%%%%%% coming from MOTIVATION:
This work proposes a Prolog-dialect for the found and prioritised problems on expressibility and automation.
Given some given C-like program, if dynamic memory is allocated, altered and freed on runtime, then a description of desired dynamic memory is a heap specification.
The check of calculated memory state against a given specification is dynamic memory verification.
This contribution only considers formal specification and verification in a Hoare calculus.\\

Issues found include: invalid assignment, (temporary) unavailable data in memory cells, excessive memory allocation, (accidental) heap alteration in unexpected regions and others.
Excessive memory allocation is nowadays successfully resolved by memory analysers like Valgrind.
Essentially, papers in those areas did not bring any big breakthrough.
Possible reasons may also include the decrease of tension due to more available memory and parallel threads.
However, starting with Apt, problems related to variable modes have not yet been resolved -- neither entirely nor in an acceptable way.
Research contributions over the last decades show again and again that heap issues remain and remain complex and still important.
A significant contribution was reached in 2016 by Peter O'Hearn, who accepted the G\"{o}del prize for his parallel approach on a spatial heap operation.\\

An essential issue of dynamic memory is when two languages are involved: one for describing and another for checking.
Languages are often inhomogenous.
Models that are supposed to correspond are tough to see are plausible.
Languages include functional or other very special, which end up in limitations, e.g., both sides' expressibility.
Other issues regard variability and extensibility of languages involved as well as models.
The results are bloated notations and complex checks, lack of expressibility or it may eventually require specific non-common conventions.

Restrictions include, e.g. the use of symbols (not variables in an imperative sense), expressibility of heap terms, logical conjuncts and less intuitive heap model and rules.
Thus, a verification may become very hard to perform and debug, even with elementary examples due to the full context.

Furthermore, the input programming language and other involved languages, are often (artificially limited) by language-specific features present or absent and become the verification rules.
Universal approaches seem to be hard to implement.
Separation of concerns between all languages involved looks urgent and challenging.
In real-life applications, the input language is inhomogenous and often represented by triads or the Polish-inverse notation, which bloats specification and verification rules.
%
%%%%%%%%%%%%%%%%%%%%%%%%%%%%%%%%%%%%
\end{abstract}

\begin{flushleft}
\textbf{Keywords.}
\textit{\textbf{heap logic, points-to, heap specification and verification,
heap conjunction, heap disjunction,
Hoare calculus,
partial heap specification,
dynamic memory,
variables mode,
Pointer calculus,
axiomatic semantics,
Prolog,
Horn-rules,
attributed grammar,
Warren's abstract machine,
semi-structured data,
term transformation,
spatial heap operation,
operation ambiguity,
memory models,
object calculi,
heap graph,
verifier architecture,
unification of processes,
narrowing down languages,
abstract predicates,
completeness by incompletness,
fold and unfold,
static analysis,
multi-paradigmal programming,
parsing as proving.}}
\end{flushleft}

\IEEEpeerreviewmaketitle

\section*{Actuality} % not novelty here
Errors related to dynamic memory are one of the most expensive and intricate to locate for several decades.
Often the program location where a problem arises, if any, does not correspond to the real.
Consequently, development and approaches in improving error localisation and avoidance are up to date research objectives.

Expressibility of assertions, completeness of rules and automation define success for heap verification.
Assertion expressibility refers to assertion formalism and verification rules.
Recently, it is found descriptions of dynamic memory to be rather intricate and not intuitive.
Logical reasoning improvements towards correctness and completeness often represent isolated implementations that not always cooperate in theory and practice with recent approaches.
Due to an ever-growing number of conventions specification and verification languages, they do
not grow in expressibility opportunities.
Even the opposite is the case.

Verification often remains non-trivial and overwhelming.
All the unique but different languages involved make it counter-intuitive and very challenging to handle in practice.

\section*{Research Objective}
In this work, the expressibility level is increased for describing and checking dynamic memory by narrowing down heap specification and verification, excluding ambiguities and prototypical implementation.

\section*{Dedication}
The work was funded by research grant no.2.136.2014/K from the Ministry of Science of the Russian Federation.
I would like to thank I.L. Bratchikov, N.K. Kosovsky and A.N. Terekhov (Saint Petersburg) for the excellent preparation and motivation given prior to taking the big challenge.
I would like to thank Andrew Pitts (Cambridge), Peter O'Hearn (London), Uday Reddy (Birmingham) and Graham Hutton (Nottingham) for their very motivating talks and valuable discussions.
I would like to thank Paul Bowen-Hugget, Con Bradley, Rob Lougher (Bristol) and Andy Thomason (Oxford) for being excellent working colleagues. Special thank to Con for recommending never to give up.
Thank you and forever in rememberance S.A. Ivanovskiy.
Apart from the need to earn for a living, unfortunately, I had to fight with (inner-)German racism without which this works would not be.

\section*{Declaration}
The research presented is the result of only my own initiative during the 'spare time' whilst working at the same time full-time on an unrelated job.
This thesis is my own work and contains nothing which is the outcome of work done in collaboration with others, except as specified in the text.
It is not substantially the same as any that I have submitted for a degree, and does not exceed 250 pages.

Saint Petersburg, 2017

%%%%%%%%%%%%%%%%%%%% MAIN PART:
  % -- MAIN PART: all chapters --
  
\section{Introduction}
\label{chapter:intro}
%  (40p.)
% APPENDICES

This section introduces this dissertation's topic: "\textit{A Logical Programming Language as an Instrument for Specifying and Verifying Dynamic Memory}". It introduces Hoare's calculus, basic definitions and existing approaches.

For comparison and discussion, the \index{quality ladder} quality ladder from fig.\ref{fig:QALadder} is referred to, where the higher level includes the lower level.
This ladder is the author's favourite clustering applying to any computer program.
The more a program obeys the listed \index{quality criteria} criteria, the higher is its quality.
For example, if a program miscalculates a result, completeness and optimality may not be that important anymore.
Though, an essential bit in that example would be correctness.
If software obeys correctness only then, the reached level is \index{basically stable} "basically stable".

Once correctness proves to an acceptable level, additional features may be requested from a function then and only then, including a broader input variable domain.
Once the basic set of input proves to work correctly, all other input variables follow for consideration -- this is the second quality level named \index{completely stable} "\textit{completely stable}".
The second criterion includes the first and therefore is stricter.
If the first two criteria hold, then the considered function is sound and complete from a computability's perspective.
If so, then the third quality level asks for how good a function computes.
Satisfying a given specification of the third level achieves \index{optimal stable} "\textit{optimal stability}".
Completeness \index{completeness} and \index{soundness} correctness remain invariant to optimality.
Optimally stable implementations may include visual effects and processing files but are not considered any further in this work.

\begin{figure}[h]
 \begin{center}
  \begin{tabular}{c|c|c}
    Level & Criterion & Quality Predicate\\
    \hline
    3 & Optimality & optimally stable\\
    2 & Completeness & completely stable\\
    1 & Correctness & basically stable
  \end{tabular}
 \end{center}
 \caption{Prioritised quality ladder}
 \label{fig:QALadder}
\end{figure}

Each of the proposed criteria may be resolved and accompanied accordingly by appropriate metrics to measure achievement.
It appears pointless in practice to discuss, for instance, soundness and optimality when completeness violates.
The reason may be an unsound implementation, which for an essential domain is just not visible and though working correctly, even when all other cases seem to work.
A bad case may require a total rewriting of the considered program.
Alternatively, as a compromise, the input domain may be bound, and the callers' arrangements may be required.
Next, this quality ladder is applied numerous times as an orientation gadget in favour or against discussed features throughout the following techniques and sections.

This work considers imperative and hypothetical programming languages (PL) at a draft stage.
For application purposes, \index{object orientation} object-orientation is the focus.
Hence this section also introduces object calculi (OC), particularly to the \index{theory of objects} \textit{Theory of Objects}.
Next, theoretical as well as practical frontiers on dynamic memory are considered.
Later \index{dynamic memory} representation models such as \index{SA} \textit{Shape Analysis} (SA), \index{RC} \textit{Region Calculus} (RC) and others are discussed.
Separation Logic (SL) is considered essential for the theoretical model of dynamic memory w.r.t. dynamic memory because of the locality property.
Afterwards, an overview is given on proof automation, and its restrictions are discussed.
Hoare calculus is chosen to understand the proofs and their confluency better, and abstraction is introduced.
A short agenda on these competing issues is found in \cite{haberland16-5}.
Neighbouring disciplines like \index{alias analysis} alias analysis (AA) (sec.\ref{sect:AliasAnalysis}), \index{GC} garbage collection (GC) and \index{introspection} code introspection (CI) are briefly introduced.
Existing provers are briefly introduced at the end of this section.

%\newpage

%%%%%%%%%%%%%%%%%%%%%%%%%%%%%%%%%%%%%%%%%%%%%%%%%%%%%%%%%%%%%%%%%%%%%%%%%%%%%%%%%%%%%%%%%%%%%%%%%%%%%%%%%%%%%%%%%%%%%%

%\newpage
\subsection{Hoare calculus}
\label{sect:HoareCalc}

The \textit{Hoare calculus} is a \index{verification} formal verification method originating back to the computer scientist Charles Anthony Hoare, which allows checking correctness according to a given \textit{specification}, namely, a formal description of a program's properties.
Specifications denote \index{calculation state} calculation states by mathematical formulae since those are believed to be at most precise and where a conclusion is made according to a given set of rules and axioms within a defined course of discussion.

The objective of a program verification using Hoare's calculus is the formal check whether acclaimed properties are correct indeed or not.
By doing so, claims towards quality may be confirmed or rejected.
If a program obeys previously claimed properties, the program may be compared with other programs and other routines to perform a more complex task.

A rule in Hoare's calculus \index{judgement} is a logical judgement of the kind $A \Rightarrow B$ reads as: "\textit{if} "antecedent \index{antecedent} $A$ is true, \textit{then} the \index{consequent} consequence can be derived".
Let us consider fig.\ref{fig:HoareCalc}.

\begin{figure}[h]
 \begin{center}
\begin{tabular}{ccc}
  \inference[(P1)]{A}{B} & \inference[(P2)]{B \vee \neg B}{C} & \inference[(P3)]{\{P\}C\{Q\}}{\{P'\}C'\{Q'\}}\\\\
  \multicolumn{3}{c}{ \inference[(P4)]{\{P \wedge \Pi\}A\{Q\}\qquad (P \wedge \neg \Pi) \Rightarrow Q}{\{P\} \ \texttt{if} \ (\Pi) \ A \ \{Q\}}  }
\end{tabular}
 \end{center}
 \caption{Logical rules in Hoare's calculus}
 \label{fig:HoareCalc}
\end{figure}

Rule $(P1)$ is an \index{axiom} axiom.
Let $A$ locally be an empty antecedent.
The main idea behind an axiomatic system is a check whether the finite application of rules ends in axioms or not.
We do not distinguish between \index{rule!verification} rules and axioms for simplicity from now on since the latter specialises earlier.
Rule $(P2)$ is an axiom if \textit{antecedent} $B \vee \neg B$ can always be derived in analogy to \textit{propositional logic} and can be replaced by \index{truth} "$true$" or just be dropped.
However, initially, no artificial restrictions existed in the Hoare calculus, e.g. regarding judgements.
Predicates may replace assertions.
A \index{logical derivation} logical derivation becomes a \index{judgement} judgement over predicates.
Arbitrary methods may be applied as logical reasoning \index{predicate!first-order} as first-order predicates, for instance, \index{natural deduction} natural deduction \cite{troelstra00}, \index{resolution} Robinson's resolution or the \index{Tableaux method} Tableaux method \cite{davis94}, \cite{troelstra00}.
The named methods are founded upon \index{deduction} deduction.
Besides deduction, \index{abduction} abduction and \index{induction}induction may be used as a verification primitive.
Induction proposes the introduction of not ultimately yet proven assertions that are hard to prove directly but make sense to assume are right in a context.
An induced assertion can be considered as part of the verification rule set.
An inductive rule added may appear strange on the first look but often may not be rejected due to a lack of \index{counter-example} contradiction, though it may be beneficial on a proof.
A classic example of induction is \index{Popper's method} Popper's method of falsification: Given the sentence: "\textit{All swans are white}" may not always be falsifiable in practice because the number of entities to be checked may either be immense or (temporarily) unavailable (in that classic example only in Australia seem to habit black swans naturally).
More general, the problem applied to logic states: occasionally, not all entities may be checked in order to decide a predicate is true or not.
Therefore, an inductive statement may be chosen as a "right" rule until proven differently in the area of discourse.
For example, if our discourse is about Europe only then, we are safe in asserting such an induction unless we extend our discourse to Australia --- this would finally falsify our induction. 
Induction allows bypassing very hard or even unexplainable phenomena within a discourse by introducing a new rule or axiom, in which terms our model seems complete.
So, induction allows us potentially infinite structures to be described and finally be checked by finite formulae.
Whenever at least one counter-example is discovered, we either are required to reformulate our rule set or discard it entirely.

In contrast to that, abduction searches for necessary and sufficient conditions, s.t. the consequence is derivable from a fixed rule set.

A Hoare calculus may be characterised as a \index{deduction} deductive series of judgements applied to an imperative \index{program statement} program statement (as it was proposed by Hoare initially).
The state of calculation describes each derivable judgment before and after executing that particular program statement (see section \ref{Intro:HoareTriple}).
Regarding \index{confluency} confluency of proof (see later, sec.\ref{fig:CRTonHoareTriples}), abduction and induction may be included as proof primitives for the sake of automation.
Otherwise, a missing rule may disrupt a proof unfinished. 
Whenever the reasoned universe is generated by strict application of rules, and there is no interleaving of antecedents, the \index{universe} set of derivable consequences is closed (regarding the description).
For example, the interpretations of possible logical statements from $(P3)$ and $(P4)$ from fig.\ref{fig:HoareCalc} is enclosed.\\

\subsubsection{Hoare Triple}
\label{Intro:HoareTriple}

Hoare \cite{hoare69} introduces a triple for specifying \index{formal verification} the calculation states and the program statement.

\begin{definition}[Hoare Triple]
\label{def:HoareTriple}
A Hoare triple $\{P\}C\{Q\}$ consists of a \textit{precondition} $P$ before and a \textit{postcondition} $Q$ after invoking a program statement $C$.
\end{definition}

Declarative \index{paradigm!declarative} PLs, e.g. \textit{functional}, differ from \index{language!programming} \index{language!imperative} imperative ones by varying calculation states depending on the variables' content.
Unless said otherwise, it is agreed that by default, \index{evaluation ordering} the evaluation order is non-strict.
The \index{paradigm!declarative} declarative programming paradigm differs from the \index{paradigm!imperative} imperative paradigm, for instance, memory use, especially on symbols and variables.
By default, this work considers imperative PLs only as input PLs.
Initially, Hoare recommended an elementary imperative PL close to \textit{Pascal} and the slightly different triple notation $P\{C\}Q$.
Since superfluous parentheses became unpopular, the positions at which parentheses are put switched as introduced at the beginning around $P$ and $Q$.
Both $P$ and $Q$ \index{calculation state} denote the \textit{calculation state} as assertions before and after the execution of $C$.
$C$ may be composed of an arbitrary number of imperative statements.

So it becomes clear why $C$ change the calculation state stepwise.
By (logical) \index{RAM} random-access memory (RAM), we understand \index{register} \textit{a CPU registers file}, \index{stack} \textit{stack} and \index{dynamic memory} \textit{dynamic memory} as depicted in fig.\ref{fig:ProcessSectionLoader}.
\index{Hoare triple} A Hoare triple is interpreted as the following:
Assuming a calculation state is denoted by $P$, and a \index{program statement} program statement $C$ is given, the corresponding computation is terminated and transits into $Q$.
In other words, $P$ and $Q$ describe the calculation step before and after $C$.
If the Hoare triple is sound, then the transition between memory states is sound.
Otherwise, the following reasons may be given for unsound behaviour:

\begin{enumerate}
 \item If $C$ does not terminate, then the rule set is "\textit{incomplete}", and state $Q$ is unreachable.
It can formally be described as $\vdash \{P\}C\{Q\} \rightarrow (\llbracket C \rrbracket \neq \bot) \wedge Q$, where $\llbracket . \rrbracket$ is the denotational semantics for a given statement \cite{allison89}, \cite{abramsky94}, \cite{winskel93}.
 \item The axiomatic rules are incomplete.
 There is no such rule that applies for a given state $P$.
 \item The obtained actual calculation state  $Q$ is not the expected state $Q'$.
\end{enumerate}

Fig.\ref{fig:ProofRulesEx1} shows essential axioms and rules for imperative PLs.
Programs listed in an imperative paradigm may be rewritten in a \index{paradigm!declarative} declarative paradigm as long as both have the equivalent potency of computation model, e.g. by simulating a universal \index{Turing machine} Turing machine.
A \index{CFG}control-flow graph (CFG) may abstract an imperative program. 
Subroutine calls can be performed on common stack \index{stack} architectures.
Hoare's initial calculus \cite{hoare69} does not impose additional restrictions on the expressibility of specifications as this was primarily not addressed, nor is it a topic in his successors' work, e.g. by Apt \cite{apt93}.
Apt suggests a Hoare calculus for single- and multi-threaded application. Furthermore, he suggests a classification of assertions of input and output variables.
Later we are going to assess the advantages and disadvantages of such propositions in more detail.

\begin{figure}[h]
\begin{tabular}{c}
  \inference[(SEQ)]{\{P\}A_1\{Q\}\quad\{Q\}A_2\{R\}}  {\{P\}A_1;A_2\{R\}}\\\\
  \inference[(LOOP)]{\{P\wedge B\}S\{P\}}{\{P\}\texttt{while} \ B \ \texttt{do} \ S \ \{\neg B \wedge P\}}\\\\
  
  \inference[(ASN)]{}{\{P[e/x]\}x:=e\{P\}}\\\\
  \inference[(CONSEQ)]{P' \Rightarrow P \quad \{P\}C\{Q\} \quad Q \Rightarrow Q'}{\{P'\}C\{Q'\}}\\\\

  \inference[(SP)]{R \Rightarrow P\quad \{P\}A\{Q\}}{\{P\}A\{Q\}}\\\\
  \inference[(WP)]{\{P\}A\{Q\}\quad Q \Rightarrow R}{\{P\}A\{Q\}}
\end{tabular}
 \caption{An incomplete rule set for verification}
 \label{fig:ProofRulesEx1}
\end{figure}

For example, the \index{rule!sequential} sequential rule (SEQ) from fig.\ref{fig:ProofRulesEx1} requires if the compound statement $A_1;A_2$ is given with precondition $P$ and postcondition $R$, then in order to prove \index{soundness} correctness, the existence of an intermediate $Q$ shall be shown, s.t., first, $A_1$ with \index{precondition} precondition $P$, and, second, that $Q$ serves as a precondition for the proof of $A_2$.
A sequence's correctness is proven after $A_1$ is proven, and $A_2$ is proven to have the postcondition $R$.

The rule of logical consequence \index{rule!logical consequence} (CONSEQ) serves for demonstration purposes as a second example.
This rule is a generalisation of (SP) and (WP).
Concretisation \index{concretisation} either strengthens a precondition or generalises a postcondition.
If the predicate is "\textit{true}" in general, then for some arbitrary well-defined (sub-)set $V$, the predicate is also true $\forall v \in V$.
Assume, $V_1 \subseteq V$.
$V$ is a \index{generalisation} \textit{generalisation} of $V_1$.
In this case, $V_1$ is a \textit{concretisation} of $V$.
It can easily be grasped the \index{implication} implication $V \Rightarrow V_1$ holds.
However, $V_1 \Rightarrow V$ does, in general, not hold.
W.l.o.g. the correctness of a predicate may be concluded for a concrete case.
Therefore, rules (SP) and (WP) hold.

When unifying (SP) and (WP) (CONSEQ) may be obtained.
However, renaming must be taken into consideration carefully.

As a third example, the \index{rule!loop} loop rule (LOOP) can be considered a critical feature in any minimalistic and Turing-mighty PL.
This rule states that if we have a loop as a program statement with $P$ as precondition and $Q$ as postcondition, then in order to prove correctness, it is sufficient to prove the correctness of the loop's body $S$ accordingly with $P$ as precondition and $\neg Q$ \index{expression negation} as postcondition, where $\neg$ negates the postcondition.

It is worth noting that the loop from (LOOP) and the \index{rule!assignment} \textit{assignment rule} both together allow expressing all other kinds of loops due to expressibility aspects (cf. fig.\ref{RulesLoopReplacements}).

\begin{figure}[h]
\begin{center}
\begin{tabular}{c}
  \inference[(REP)]{\{P\}S\{P'\} \qquad \{P'\}\texttt{while} \ B \ \texttt{do} \ S \ \{Q \wedge \neg B \}}{\{P\}\texttt{do} \ S \ \texttt{while} \ B\{Q \wedge \neg B\}}\\\\
  \inference[(FOR)]{\{P\}i:=1; \texttt{while} \ B \ \texttt{do} \ (S;i:=i+1) \ \{Q\}}{\{P \}\texttt{for} \ i \ \texttt{from} \ 1 \ \texttt{to} \ n \ \texttt{do} \ S \;\{Q \}}
\end{tabular}
\end{center}
 \caption{Loop rules replacing \texttt{while}}
 \label{RulesLoopReplacements}
\end{figure}

For simplicity, let us agree upon (FOR) $n$ denotes the number of iterations that is known before executing the loop, $i$ denotes a new variable in $P$ and $Q$.
If $i$ was not fresh, then another new variable name should be chosen according to all stacked variables' visibility scope.
This issue is profoundly investigated to a full extent within the so-called "\index{numeral logic} Numeral Logic in \cite{pitts02}, \cite{debruijn72}.
So, (REP) may be transformed w.l.o.g. into (LOOP) and vice versa.
However, in general, (FOR) should be transformed only unidirectional since for the opposite transformation, only a pre-calculated integer should bound the number of iterations which is often not the case, therefore (FOR) should then not be used. 
In the case of primitive recursion, it actually may not even be the case.
An in-depth insight into generalised \textit{$\mu$-recursions} and termination issues with mutual recursion is best considered in \cite{bekic84}.

A \index{conditional branch} conditional branch operator is not considered separately because w.l.o.g. may entirely be replaced by (LOOP).
Further reading on the minimalistic but complete PL "\textit{PCF}" \cite{plotkin77}, \cite{cohn83} is recommended.
The essential concept of simulating Turing-computable functions lies in a minimalistic set of program statements.
The crucial difference between a simple statement and a loop is its body is repeated zero or more times, and assigned variables in the body change every time the body is entered.
All variable \index{data dependency} dependencies need to be analysed to catch all modifications in a basic block (BB) until the loop exits.
Often this requires a profound abstraction.
Usually, this can be done manually only.
Such an abstraction also requires auxiliary variables to be introduced to keep the invariant part of that body description smooth and straightforward.
Therefore, the description approximates some general target function containing target variables (cf. fig.\ref{fig:CFGEx1}), making up the \textit{loop invariant}.
The generation of an \textit{invariant} often is not a trivial part and can often not be decided.

For this reason, a manual refinement of the invariant specification is vital.
The postcondition of the loop considered represents an invariant.
In analogy to a fixed mapping \index{projective geometry} from projective geometry, when one point remains fixed, a loop invariant is an assertion that remains invariant regardless of how often the loop is iterated.
The invariant formula $\Phi$ must obey $\Phi \circ Y = Y \circ \Phi \circ Y$, where $Y$ is a \index{fix-point combinator} \textit{fix-point combinator} and $\circ$ a binary function concatenation.
$Y$ is some syntactic notation simulating repetition.
Its objective is to define some search minimum, widely used with \index{$\lambda$-calculus} $\lambda$-calculi (for more details, refer to \cite{barendregt93}).
Alternatively, but equivalent w.r.t. expressibility, Kleene's \textit{$\mu$-operator} \cite{davis94} may implicitly be applied to variables to be minimised, s.t. a predicate still holds.

Let us consider the example from fig.\ref{fig:ExampleSimpleGCD}.

\begin{figure}[h]
\begin{center}
\begin{minipage}{6cm}
\begin{verbatim}
a:=0; b:=x;
while b>=y do  b:=b-y; a:=a+1; do
\end{verbatim}
\end{minipage}
\end{center}
 \caption{Code example for the remainder of two integers}
 \label{fig:ExampleSimpleGCD}
\end{figure}

Here, $a \cdot y + b = x^b \ge 0$ may be invariant since this equality describing the \index{remainder} remainder on \index{division} division by integer $y$ does not alter.
$y$ is the divisor of $x\ge 0$, $a$ an integer, and $b$ is the remainder of $x$ divided by $y$.
The \index{rule!assignment} assignment rule (ASN) denotes: if a well-defined expression $e$ is assigned to a compatible variable $x$, then before and after assignment, the state in $P$ remains without changes.
The state before assigning requires a variable binding extension, including $x$.
In case of a \index{name clashing} name clashing, it is required to perform a renaming first in the specification.\\

\index{completeness} \textit{Completeness} depends on the completeness of \index{Hoare triple} Hoare triples of kind $\{P\}C\{Q\}$, which depends on complete coverage of $C$ and all fully covered preconditions of $P$ (cf. def.\ref{def:HoareCalculusCompleteness}).
Postconditions like $Q$ are logical consequences if met by $P$ and $C$.
Since verification rules may be applied in potentially any order they occur, the question arises: Can one rule accidentally or intentionally exclude or overlap any other rule?
It can.

Nevertheless, the answer on how to decide to do this automatically does not seem too trivial.
Naturally, if a given $Q$ and given program statement $C$, different preconditions $P_1$ and $P_2$ are inferred, then this clearly shows the rule set is not sound.
Besides correctness and completeness, according to the \index{quality ladder} quality ladder from fig.\ref{fig:QALadder}, optimality matters most.
So, the question related to an effective rule set concerns compactness and simplicity as well.
If the question is about automation, then the question of the chosen architecture is essential too.
This section shows that the invariant generation may be tricky because an accurate and the most general conclusion requires all variables modified in the loop body and could include all referred variables in memory.
Let us also pay attention to automatically allocated variables with a different \index{visibility scope} scope than dynamically allocated variables.
Moreover, it can be noticed the comparison between a given and expected specification may require some previously agreed formalisation, which enables us to decide equality.\\

\subsubsection{Logical Reasoning}
\label{Intro:LogicalReasoning}

The rule $(P1)$ from fig.\ref{fig:HoareCalc} represents the most general logical rule ever.
Verification means a formalised process (see obs.\ref{obs:DeductionWithBacktracking}).
A \index{verification result} \index{verification} verification is either valid or not.
For the latter, when execution is undetermined, it is the same as when the postcondition is not specified or a statement does not terminate.
Otherwise, verification is false.
When applying rules, \index{proof!unfolding} the antecedent may branch into several (sub-)verifications, which all need to be verified.
Hence, any formal verification is \textit{tree-structured}.

\begin{definition}[Input Programming Language]
\label{def:IncomingProgrammingLanguage}

\index{language!programming} A \textit{PL} is a \index{language!formal} formal language whose \index{word} words are programs consisting of statements.
The execution of a \index{program statement} program statement may alter \index{memory state} memory (stack, heap).
By default, a C-like imperative PL with object-oriented extensions is considered, though not \index{C++} C++ nor Peyton-Jones' C\--\--.
\end{definition}

An imperative language is selected for several reasons.
C or Java are among the most popular PLs according to the TIOBE index, and why redefining apriori more and more new languages?
If there is no need, then it would certainly be better not.
However, the difficulty seems to decide whether this is the case or not.
Second, in this work, a C-dialect is chosen to minimise primarily (re-)defining legacy syntax and semantics from a fragment of existing languages.
Naturally, C is platform dependent -- this seems to be a disadvantage.
However, the dependency will be discussed further to get a qualitative judgement on this possible issue.
The selected C-dialect is a subset of ISO C99 (cf. fig.\ref{ScreenshotGUI}, see \cite{haberland19-1}, \cite{haberland16-6}).
As being PCF-compliant \cite{plotkin77}, it supports a basic while-loop, as discussed in fig.\ref{RulesLoopReplacements}.
It supports object calculus type (1) according to fig.\ref{fig:ObjectCalculi}

\begin{definition}[Specification Language]
\label{def:SpecificationLanguage}

\index{language!specification} Any \textit{specification language} for dynamic memory considered here is formal.
This language refers to variables and the input language's features, \index{symbol} \textit{symbolic expressions}, \textit{quantifiers} and \textit{auxiliary program units} for facilitating proofs.
A specification language must follow apriori agreed \textit{formal logic}.
In contrast to an input PL, the specification language is of \index{paradigm!declarative} declarative paradigm, but indeed not of imperative paradigm (cf. sec.\ref{sect:PrologAsReasoningSystem}).
\end{definition}

\begin{figure}[h]
 \begin{center}
   \begin{tabular}{c}
     \inference{\inference{$A_1$ \qquad \inference{\texttt{false \textbf{(!)}}}{A_2}}{A_3}\qquad \inference{\inference{\texttt{true}}{B_1}}{B_2}}{B}
   \end{tabular}
 \end{center}
 \caption{Example of a proof refutation}
 \label{fig:ProofTreeEx1}
\end{figure}

A specification language describes each BB of the associated CFG (cf. fig.\ref{fig:CFGEx1}) the calculation state based on memory (see fig.\ref{fig:ProcessSectionLoader}).
Let us choose an arbitrary \index{derivation tree} derivation tree from fig.\ref{fig:ProofTreeEx1} with the assertions $\{A_1,A_2,A_3,B_1,B_2,B\}$.
Initially, the \index{Hoare triple} Hoare triple $B$ ought to be shown accordingly to def.\ref{def:HoareTriple}.
Thus, the given rule is applied.
The antecedent shall contain $A_3$ and $B_2$, both of which need to be proven separately using respective replacements and strictly according to given rules, (e.g. local symbols), shall be generated using $B$ as its \textit{consequent}.

Furthermore, applying the rules, we show that $B_1$ is a precondition to $B_2$ and $B_1$ according to the rule is a tautology, e.g. $\{n=0 \wedge n \ge 0\}a=5;\{n=0\}$ whenever it becomes evident that $a$ \index{variable!free} is not bound (but a \textit{free} variable with $n$).
Next, $A_3$ is proven.
However, this would imply that $A_2$ contradicts itself.
This condition is enough in order $A_3$ to derive to false and as such $B$ to the same.
So, we have just shown $B$ is incorrect, and the reason for that lies in $A_2$.
In the example, there is no need to track further down the proof on $A_1$, even if it is correct or not.

\begin{definition}[Logical Consequence]
\label{def:LogicalJudgement}

A \textit{logical consequence} $A\vdash B$ denotes an assertion $A$ to which some given rule from a given axiom system is applied once.
It leads to assertion $B$ (according to Frége \cite{frege}).
If $B$ is obtained after rule application several times, including none at all, the consequence is noted as $A \vdash^{*} B$.
%%%

If we want to express that, a triple $A$ is always valid under a given rule set $\Gamma$.
Then we use $\models A$ if it is clear from the context that only $\Gamma$ is considered or $\models_{\Gamma} A$.
\end{definition}

\begin{observation}[Proving as Searching]
\label{obs:ProofAsSearching}
For a given rule set $\Gamma$ a proof of some theorem $B$ in Hoare calculus can be formulated as a \textit{search for axioms}, so $\models_{\Gamma} B$.
\end{observation}

First, we intentionally do so for the sake of a better understanding and approximation, not exclude possible inaccuracies in the rule set to be considered.
For example, referring to $\Gamma$, we also refer to a  \index{formal logic} formal logic, which consists of a closed set of dependencies regarding given rules and a \textit{carrier set} and base logical constants.
By default, we agree upon any arbitrary rule set used for logical reasoning.
We insist on a well-defined and sound Hoare calculus with triples according to def.\ref{def:IncomingProgrammingLanguage} based on def.\ref{def:LogicalJudgement}.
This observation is a precursor to syntax definition as proof (see thes.\ref{thes:ReasoningAsProving}).

\begin{definition}[Soundness of a Hoare Calculus]
\label{def:HoareCalcCorrectness}
A Hoare calculus is \textit{sound} whenever cases are guaranteed to be excluded when a syntactical sound program $C$ with a rule set $\Gamma$ derives (semantically) different results.
\end{definition}

If for one valid logical reasoning results in $B_1$, and another valid reasoning results in $B_2$, although neither of which may derive from the other, then $\Gamma$ is not sound.
This derivation is denoted as
$\{P\}C\{Q\} \vdash^{*} \{P1\}C_1\{Q1\}$ and $\{P\}C\{Q\} \vdash^{*} \{P1\}C_2\{Q2\}$ where $\{P1\}C1\{Q1\} \nvdash^{*} \{P2\}C_2\{Q2\}$ and $\{P2\}C2\{Q2\} \nvdash^{*} \{P1\}C_1\{Q1\}$ (see fig.\ref{fig:CRTonHoareTriples}).

The theorem just formulated upon a Hoare calculus is often called \index{diamond property} \textit{diamond-property} due to its shape and originates back to rule-based rewriting systems and $\lambda$-calculi \cite{barendregt93}.
That theorem goes back to \index{Church-Rosser theorem} Church-Rosser, hence called CRT (see \cite{peirce10}).
If at least one assertion whilst reasoning may result in ambiguity, then the entire rule set $\Gamma$ is \textit{not sound} (see def.\ref{def:HoareCalculusCompleteness}).
\index{confluency} \textit{Proof confluency} --- a severer restriction than soundness, is not always obeyed in practice, for instance, due to general non-decidability on the grounds of non-termination (see \cite{steinbach94}).
Soundness is a required precondition to confluency.
If some calculus is not sound, then at least one case exists, so two or more different results may be derived -- and that is certainly not correct.
According to Steinbach's \cite{steinbach94} \index{term rewriting system} \textit{term rewriting system} \cite{baader98}, decidability on termination strongly determines confluency.
This outcome is a bold statement.
He represents logical rules as term rewriting rules and uses bound \textit{partially sets} to prove termination of his newly obtained rewriting system. However, that system does not work for all cases since termination is only semi-decidable due to foundational limitation.
It can be noticed that the approximation of an upper bound for the reasoning is in general not bound due to \index{term!self-applicable} self-application or to recursive symbolic terms \cite{plaisted85}.
The idea of a \index{descending chain} descending \index{poset} poset chain is strongly intertwined with \index{domain} Domain Theory's concept \cite{scott76}.
Steinbach applies poset chains for determining the least upper bounds on rule sets to decide termination, which is a precondition to soundness.

\begin{figure}[h]
\begin{center}
\begin{tabular}{c}
\xymatrix{
  & \{P\}C\{Q\} \ar@{~>}[dl]^{\vdash^{*}} \ar@{~>}[dr]^{\vdash^{*}} & \\
  \{P1\}C_1\{Q1\} \ar@{~>}[dr]^{\vdash^{*}} &  & \{P1\}C_2\{Q2\} \ar@{~>}[dl]^{\vdash^{*}}\\
  & \{P3\}C_3\{Q3\} &
}
\end{tabular}
\end{center}
 \caption{Church-Rosser theorem applied to Hoare triples}
 \label{fig:CRTonHoareTriples}
\end{figure}
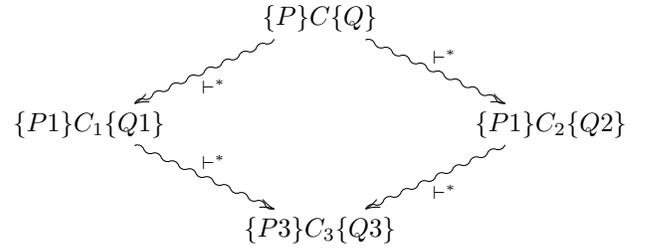

Cook \cite{cook78} investigates both, soundness and completeness of Hoare calculi as competing quality measures.
It is worth remarking that tiny changes towards one objective may heavily change the other objective, for example, in modifying variables in procedures.
Cook's sense of completeness is the full coverage of some function's domain, so a \textit{total mapping} of all input programs, including those earlier mentioned in \cite{apt81}.
"\textit{Soundness}" is defined as an equivalence between \textit{observed and expected behaviour} as recommended by \cite{davis94}, \cite{nielson99}.
There it is defined by \index{operational semantics} operational semantics over Hoare triples.
\index{automaton!abstract} Abstract automata then interpret calculations.
A convenient solution by Cook means either a complete or a sound calculus.
His propositions to overcome some imperative input PL restrictions are supposed to back the objective up, such as:

\begin{itemize}
 \item \textbf{Restriction 1} \index{variable!global} Global variables are allowed.
 However, it is strictly prohibited to pass them to procedures as actual parameters.
 
 \item \textbf{Restriction 2} \index{parameter!by call} \index{parameter!by reference} Parameters by call (or by reference) are forbidden.

 \item \textbf{Restriction 3} No recursive procedures and no \index{functional} \textit{functionals} (\textit{higher-order functions}).
\end{itemize}

As some of those essential restrictions may come as a surprise (but they undoubtedly exist after all), how can disobedience of those lead to unsoundness or incompleteness, since those restrictions must be taken into consideration when designing a PL most fit for a heap verification?
A classic stack requires \textit{parameters by-reference} to address where to return once a stack window is fully processed. 
As soon as a return is initiated, all existing local variables stop existing.
A pass by reference allows altering a unique content at different places among the \index{stack!window} stack windows, but the \index{pointer} pointer to an object may change during a subroutine call.
If we permit \textit{parameters by call} or recursive functions, then data may alter but not how the stack is processed on calls.
For example, \textit{globals} may alter almost everywhere.
In such cases, "na\"{\i}ve" specifications may quickly become invalid in general.
Invalidation may be excluded by restricting input and output parameters to a procedure.
\index{hierarchical specification} Hierarchical specifications \cite{schwinghammer09}, \cite{birkedal06} may blow up specifications.
Moreover, since those specifications may occur anywhere in a program, the use of finer-grained specifications might diminish quickly to nought, just because the program's location may dramatically change and expand even more level.
Consequently, the depth of those nested specifications increases too.
So, in the end, finer modularity may become useless in practice.
For example, built-in procedures would need to be fully specified.
Furthermore, objects may lose their identity and closure because passing a procedure as a parameter may finally lead to unpredictable behaviour towards a class-objects memory layout.\\

Clarke \cite{clarke79} outlines only limitations still a problem to date w.r.t. Hoare calculi.
Rigid constraints refer to:

\begin{itemize}
 \item \textbf{Restriction 4} \index{expressibility} Expressibility and \index{completeness} completeness of the \textit{input PL}

 \item \textbf{Restriction 5} Loop \index{invariant} invariants (e.g. transformation)

 \item \textbf{Restriction 6} \index{recursion} Recursive procedures which do (not) use globals and \textit{static variables}
 
 \item \textbf{Restriction 7} \index{co-procedure} Co-procedures as input and output parameter to a procedure
 
 \item \textbf{Restriction 8} \index{variable!dynamic} Dynamic allocation and deallocation (also see \cite{raman12})
\end{itemize}

Both Cook and Clarke understand guarantees under soundness, s.t. all syntactical sound theorems derive correctly and all unsound theorems false.
As any formal system, Hoare calculi underpin theoretical limitations.
For example, \index{theory of integers} one can always define different theories on integers which will eventually be sound but cannot be proven under the surrounding proof facilities if, e.g. certain axioms are modified or dropped.
This phenomenon is better known under \index{theorem! G\"{o}del's} \textit{G\"{o}del's incompleteness theorem}.
Clarke notices as Cook before him that if a Hoare calculus contains a non-terminating cycle, then there must be a (mutual) recursion in the rule set.
He suggests banning \index{alias} aliasing pointers (see later), and for the sake of soundness to ban unbound recursion.
Clarke's problem may be relaxed if substitution rules in the Hoare rules are applied from outermost to the innermost direction (this corresponds to lazy evaluation).
At the design level of a PL as an input language, exactly this \index{evaluation ordering} evaluation order should be taken into account, ideally.

PLs with \index{lazy evaluation} strict evaluation, like \index{OCaml} OCaml, in contrast to non-strict languages like \index{Haskell} Haskell, should ideally extend its parameter passing method and ABI.
According to Clarke, expressibility always primarily refers to the assertion language, which can be noticed from his work effectively.
Restriction 6 on recursive procedures may be excluded in two cases.
(i) When recursion steadily terminates, so a \index{descending chain} descending chain always exists \cite{steinbach94}.
(ii) parameter passing ordering and parameter \index{typing} type checking \cite{cardelli96-2} on a procedure call must always precisely match between caller and callee, and \index{variable!automated} variables from the current stack window shall permanently be excluded, so non-global and non-static variables, for instance.

Kaufman \cite{kaufmann04} considers expressibility and practical need as the most critical problems that hinder a popular currently vital break-through of \index{formal method} formal verification based on Hoare calculi, especially w.r.t. specification and PL.
He sees in \index{induction} induction and abstraction the single most potent verification techniques.
To distinguish better between a calculation error committed by an input program and an inaccurate specification, he urges a most generic method for constructing \index{counter-example} counter-examples automatically. 
His motivation is to automatically also derive a concrete example as a counter-example for better understanding and faster troubleshooting.

Gerhart \cite{gerhart76} compares existing tools to that time, which today are still not resolved nor resolved satisfactorily.
To name the most important open problems:

\begin{enumerate}
 \item Lack of generalised approach in decidability of program \index{termination} termination
 \item \index{variable!static} Support for distinct variable modes (like static, globals, cf.\cite{clarke79})
 \item A handy usage and generalised IR of all components of a verification system
\end{enumerate}
She notices \index{inductive definition} inductive definitions are an exemplary apparatus to numerous issues simultaneously, e.g. with problems on expressibility.
A further insight makes one agree with Gerhart's remark in practice.
Gerhart insists on simple proofs must primarily refer to universal and global assertions.
From \cite{clarke79} also follows that problems from (ii) are defying and of global nature, and it can hardly be believed that a tiny modification in a Hoare calculus will suddenly allow a solution.

\begin{definition}[Completeness of a Hoare Calculus]
\label{def:HoareCalculusCompleteness}
According to its rules, a \index{Hoare calculus} Hoare calculus is \textit{complete} whenever syntactically well-defined program \index{soundness} soundness can be shown, and unsoundness of not well-defined programs can be rejected.
\end{definition}

In case of lack of at least one rule till termination, the calculation is considered \textit{incomplete} and the verification in total \textit{undetermined}.
In practice, a tiny modification in an (input) program (or specification) may already significantly change a program's behaviour \cite{clarke79, cook78, cook71}.
Not surprisingly, this can happen and witnesses a high complexity of the verification has to take into consideration.
Example 1: According to restriction 6, the absence of static variables and recursive procedures can still lead to a complete program, which depends on the program fragment is considered.
Example 2: The completeness of program verification with \index{inner procedure} internal procedures is violated whenever restriction 6 or restriction 7 is violated.
Whenever two rules lead to two different results, state $A$ implies $A \vdash B_1$ and $A \vdash B_2$, where $B_1$ and $B_2$ are syntactically different, but \index{diamond property} the diamond-property is obeyed from fig.\ref{fig:CRTonHoareTriples}, then $B_1$ and $B_2$ denote just intermediate states and both states confluence.
From a practical perspective, a full \index{confluency} confluency check may, in general, become quite complicated due to the exponential rise of the number of sub-verifications to be performed.
This problem may significantly be reduced by determining all rules.
Clarke also advises imposing the following restriction in order to fight incompleteness:
%%%%%%%%%%%%%%%%

\begin{itemize}
 \item \textbf{Restriction 9} Introduction of partially-computed data-structures.
\end{itemize}

For \index{infinite data structure} partially-computed (possibly infinite) data structures, all \index{object field} fields are calculated only on accessing those (this is a lazy evaluation).
An example is taken from \cite{thompson97} in \index{Haskell} Haskell and deals with \index{linear list} linear lists.

\begin{center}
\begin{tabular}{l}
\begin{minipage}{7cm} 
\begin{verbatim}
 take 10 [ (i,j) | i <- [1..], let k = i*i, j <- [1..k] ]
\end{verbatim}
\end{minipage}
\end{tabular}
\end{center}

, which computes
$$\texttt{[(1,1),(2,1),(2,2),(2,3),(2,4),(3,1),(3,2),(3,3),(3,4),(3,5)]}.$$

However, this \index{linear list} linear list definition, as \index{\texttt{take}} \texttt{take}'s second argument, does not have an upper bound.
The obtained data structure only has a lower bound (the integer 1).
Thanks to the \index{infinite data structure} partially-computed data structure, a procedure's termination behaviour to be investigated on each occasion individually may finally terminate in case of \index{strict evaluation} non-strict evaluation.
"\textit{strict}" means a given procedure evaluates all incoming parameters \index{parameter!incoming} first and only then starts to \index{stack!push} store to memory.
Non-strict \index{non-strict} procedures mean \index{parameter!incoming} incoming procedure parameters are evaluated partially and only if needed in the current algorithm (see \cite{thompson91}, \cite{thompson97}).
Thus, for instance, \index{termination} termination may be checked using a potentially infinite input data structure as a quick test.

Wand \cite{wand76} understands under a \textit{complete function} the same as \index{Cook's completeness} Cook does.
Namely, every sound input parameter corresponds to a sound output parameter, and that every unsound input parameter corresponds to an error.
The transition function between input and output
parameters must be \index{totality} total and \index{partial correctness} non-termination apriori excluded (see def.\ref{def:HoareCalcCorrectness}).
Wand shows that when a program's specification is represented immediately by a \index{CFG} CFG, then the specification may not be complete in general or just not always defined properly whenever \index{higher-order logic} higher-order predicates are used.
Wand urges \index{predicate!quantified} quantified predicates used in specifications must give an intuitive explanation in the \index{pattern matching} best way possible, s.t. a human reading that specification has a full understanding of the predicate and its dependencies immediately.

Cook \cite{cook71} compares the SAT-problem of boolean functions  \index{boolean denotation} \index{SAT-problem} with theoretical estimates on associated problems of formal verification.
The paper gives an overview of the theoretical complexity involved.
From a practical point of view, this paper, unfortunately, is futile for two reasons.
First, the established levels are too coarse and, therefore, cannot be helpful in practice.
Second, the paper's main objective seems to be epistemological rather than have concrete results that could be applied in practice, not even as a design recommendation for future projects.

Landin \cite{landin64} suggests turning formal calculi elements in an expression containing \index{$\lambda$-term} $\lambda$-terms as it was proposed \index{Church typing} initially.
A Hoare calculus is indeed a formal calculus.
Landin demonstrates, by example, \index{conditional branch} conditional branches and \index{recursion} recursion are tractable and tractable elegantly.
From a computability perspective, this is already sufficient.
Moreover, he shows, \index{paradigm} \index{paradigm!functional} functional programs \cite{thompson97}, \cite{bird88} may be represented in an imperative program with the help of \index{closure} "\textit{closures}" which are well-founded on \index{operational semantics} operational semantics.
Generalised \index{model of computation} models of computation are presented that attract Hoare calculi as well.\\

The best practices in verification in Hoare calculi gained over past decades is carefully analysed by Apt \cite{apt81} towards \index{completeness} soundness and \index{soundness} completeness of a \index{C} C-like dialect and with nearly minimal modifications \index{program statement} on all aspects considered, mainly procedures and variable modes.
In addition to all reviewed papers, Apt sees general recursion as the central issue in divining the matching calculation state of an executed program w.r.t. specification.
Hence, he suggests restricting ourselves to \index{recursion!primitive} primitive recursion whenever possible to reduce the overall verification complexity.
Also, Apt recommends restricting procedure calls to those in which actual and \index{parameter!formal} formal parameters match and in which so-called "\textit{incorrect}" parameters are \index{parameter!incorrect} excluded because those can lead to a series of anomalies.
Harmless looking program features were proven wrong and therefore are suggested for a ban: 
for example, \index{uninitialised} uninitialised fields, \index{functional} functionals gained by undermined parameters can even totally change a function's semantic and syntax but at the same time does not extend expressibility in its foundation.
\index{parameter cutting}

Cook \cite{cook78} and authors already consider two problems as the most important:

\begin{itemize}
 \item \textbf{Completeness 1} \index{non-termination} Procedure Non-Termination.
 \item \textbf{Completeness 2} \index{expressibility} Expressibility restrictions \index{language!assertion} in the assertion language (for instance, predicates \index{invariant} and invariants).
\end{itemize}

A sound and \index{Cook's completeness} complete example is chosen in fig.\ref{fig:SoundNCompleteExample} to discuss modifications and their effect on completeness and soundness.
$A_j$ denotes program statements, $D$ denotes a declaration block of local variables, $\sigma$ denotes a variable environment and $\star$ denotes \textit{Kleene's star}.
$\sigma$ is of type "$variable\ name \rightarrow value$".

\begin{figure}[h]
\begin{center}
\begin{tabular}{c}
  \inference[(VAR)]{P\ y/x\{begin\ D^{*};\ A^{*}\ end\}\ Q\ y/x}{P\ \{begin\ new\ x;\ D^{*};\ A^{*}\ end\}\ Q}\\\\
  \inference[(SEQ)]{P > R,\ R\{A\}S,\ S > Q}{P\{A\}Q}\\\\
  
  \inference[(ITE-1)]{P\{A\}Q,\ Q\{begin\ A^{*}\ end\}R}{P\{begin\ A;A^{*}\ end\}R}\\\\
  \inference[(ITE-2)]{}{P\{begin\ end\}P}\\\\
  
  \inference[(CON)]{P\ \& R\{A_1\}Q,\  P\ \&\ \neg R\{A_2\}Q}{P\{if\ R\ then\ A_1\ else\ A_2\}\ Q}\\\\
  \inference[(ASN)]{}{P\ e/x\{x:=e\}\ P}\\\\

  \inference[(CALL-2)]{p(x:v)\ proc\ K,P\{K\}Q}{P\{call\ p(x:v)\}Q}\\\\
  \inference[(LOOP)]{P\ \&\ Q\ \{A\}\ P}{P\{while\ Q\ do\ A\}\ P\ \&\ \neg Q}\\\\
  
  \inference[(PAR)]{P\{call\ p(x:v')\}Q}{P\ u,e/x',v'\ \{call\ p(u:e)\}\ Q\ u,e/x',v'}\\\\
  \inference[(CALL)]{P\{call\ p(u:e)\}Q \qquad $where$\ \sigma = z'/z}{P\sigma\ \{call\ p(u:e)\}\ Q\sigma}\\\\
  \inference[(PROC)]{D,\ P\{begin\ D^{*};\ A^{*}\ end\}Q}{ P\{begin\ D;D^{*};A^{*}\ end\}Q}
\end{tabular}
\end{center}
 \caption{Example of a complete and sound rule set taken from \cite{cook78}}
 \label{fig:SoundNCompleteExample}
\end{figure}

In addition to restriction 6, a further restriction remark shall be noted:

\begin{itemize}
 \item \textbf{Restriction 10} Non-automatically allocated variables
\end{itemize}

As mentioned earlier, the introduction of global, static and dynamic variables may invalidate rule sets in Hoare calculi.
Moreover, \index{variable!local} locals in \index{thread} multi-threaded applications may also invalidate rule sets because parallel execution (\textit{thread-local variable mode}) are not considered by Cook.
This work is dedicated to single-threaded execution only.

For a fully-fledged industrial use, additional language features would be required not yet considered and not in the next future, namely, \index{exception} exceptions that seem to disintegrate with any Hoare calculus (Clarke mentioned).
So there is no wonder why a sound and incomplete rule set integrating exceptions may be challenging since the stack would have to behave fundamentally different when invoking \index{program statement} program statements inside a try-catch block, for instance (see \cite{dedinechin00}, \cite{goodenough75}).
Up to date, not a single proposition was made in terms of a \index{Hoare calculus} Hoare calculus allowing a sound or/and complete \index{stack!rewinding} stack rewinding that would include, for example, dynamic variables \cite{goodenough75}.

Hoare \cite{hoare69} considers \index{abstraction} abstraction the essential verification technique, especially for young engineers who are new to verification and would like to play with some simple examples for the beginning.
From Hoare's original work, it can be implied that "goto" labels, unconditional jumps and passing parameters by name are hard to formalise.
Hoare does not speak in favour or against higher-order logics.
However, he sees a declarative assertion specification as a pivotal element to success for all kind of Hoare calculi in general.\\

%%%%%%%%%%%%%%%%%%%%%%%%%%%%%%%%%%%%%%%%%%%%%%%%%%%%%%%%%%%%%%%%%%%%%%%%%%%%%%%%%%%%%%%%%%%%%%%%%%

\subsubsection{Automated Reasoning}
\label{sect:LogicalReasoningAutomation}

This section aims to introduce and briefly illustrate theoretical problems \index{automated verification} on proof automation.
In this section, examples with the theorem prover \index{Coq} Coq are considered.
Coq \cite{bertot04} can prove previously specified assertions by theorems, different kind of \index{inductive structure} inductively defined structures and commands based on the typed \index{$\lambda$-expression} $\lambda$-expression calculus of second order.
Assertions are specified in the functional PL \index{Gallina} "\textit{Gallina}".
Proofs are specified in the command language \index{Vernacular} "\textit{Vernacular}".
\textit{Coq} often cannot independently search for proofs automatically.
It is a proof assistant that helps the user debug and record proofs.
Coq has a core asset of \index{formal theory} theories, e.g. on integers or \index{real numbers} floating-point numbers, which may swiftly be loaded to the \index{verification core} verifier core as modules.
The \index{verification assistant} proof assistant allows tracking and record-keeping proofs and finds out intermediate proof states.

The sequence of recorded proof commands may manually be enriched by \index{tactics} tactical commands (short for "\textit{tactics}").
Tactics try to apply several default proof steps, including normalisation and simplification, at once in order to reduce or even bring the current proof to the desired goal to be found.
Coq is a rare, compelling and widely accepted \index{verification platform} proof assistant platform.
Successful applications include compiler phase verifications \cite{rideau08}, \cite{leroy09}, \cite{blazy06}, \cite{blazy05}, \cite{leroy06}, \cite{leroy09-2}, \cite{raman12}.

Formulae based on predicate logic \index{predicate logic} are specification expressions.

\begin{definition}[Logical Formulae in the First-Order Predicate Logic]
\index{logical formula} A formula $\Phi$ in first-order predicate logic is defined as

$\Phi ::= true \ | \ false \ | \ x \ | \ REL(f(\vec{x})) \ | \ P(\vec{x}) \ | \
\neg \Phi \ | \ \Phi \circ \Phi \ | \ \forall x.\Phi[x] \ | \ \exists x.\Phi[x]$

where $x$ \index{variable!boolean} \index{boolean denotation} is a \textit{boolean variable}, $f$ \index{functor} is a functor, $P$ \index{predicate!assertion} an assertion predicate, $REL$ a relation having term arguments \index{relation} (relation), and $\circ$ is logical \index{conjunction} conjunction, or operator, or \index{disjunction} disjunction.
Vector $\vec{x}$ defines a term vector, which by default contains components regardless of which context they are.
It is assumed formulae containing quantifiers bind free occurrences of $x$ in $\Phi$.
\label{def:FirstOrderPredicateLogicFormula}
\end{definition}

Coq makes use of normalised formulae on reduction.
These formulae may become undetermined or \index{specification!undetermined} partially determined as the calculation proceeds.
Coq's proof schemas \index{lazy evaluation} are founded on \textit{lazy reduction} \index{$\lambda$-calculus} equivalent to the second-order $\lambda$-calculus (see \cite{cardelli96-2}, \cite{peirce10}, \cite{mitchell96}).
Because rules may be arbitrary at their definition, \index{completeness} completeness is not guaranteed.
Typing \index{typing} given $\lambda$-expressions in Coq, as is by default, avoids \index{typing paradox} antinomies that would otherwise arise due to inductive definitions over \index{Cantor set} Cantorian sets (cf.\cite{barendregt93}, \cite{bird97}), e.g. \index{term!self-applicable} self-applicable terms.
The absence of typing leads to typing ambiguities related to \index{Russel's paradox} \textit{paradox}.

\begin{definition}[Terms $T_{\lambda2}$]
\label{def:TypedLambda2ndOrder}
The type $T_{\lambda2}$ is a set in second-order $\lambda$-calculus types and is defined as
\begin{center}
\begin{tabular}{lcl}
 $t\in T_{\lambda2}$, & \qquad & $t\in V$\\
 $(t_1 \rightarrow t_2)\in T_{\lambda2}$, & \qquad & $t_1,t_2\in T_{\lambda2}$\\
 $\forall a.t\in T_{\lambda2}$, & \qquad & $a\in V,t\in T_{\lambda2}$
\end{tabular}
\end{center}
\end{definition}

, where $V$ denotes the set of variable symbols, $\rightarrow$ denotes type application.
Examples of sound $T_{\lambda2}$-types include, for example, $\forall a.a$ or $(\forall a_1.a_1\rightarrow a_1)\rightarrow (\forall a_2.a_2\rightarrow a_2)$.

\begin{definition}[Set of $\Lambda_{T_{\lambda2}}$ terms]
\label{def:Lambda2ndOrderTermTypes}
%%%
$T_{\lambda2}$-terms $\Lambda_{T_{\lambda2}}$ are defined as
\begin{center}
\begin{tabular}{lcl}
  $\Lambda_{T_{\lambda2}}$ & $::=$ & $V \ |$\\
  && $\Lambda_{T_{\lambda2}} \Lambda_{T_{\lambda2}} \ |$\\
  && $\lambda x:t\in T_{\lambda2}.\Lambda_{T_{\lambda2}} \ |$\\
  && $\Lambda x.\Lambda_{T_{\lambda2}} \ |$\\
  && $\Lambda_{T_{\lambda2}} t\in T_{\lambda2}$
\end{tabular}
\end{center}
%%%
\end{definition}

, where $\Lambda x.t$ denotes an abstraction of term $t$ by some variable $x\in V$.
Sound definitions may be $\Lambda_{T_{\lambda2}}$-types, such as $\Lambda a.\lambda x: a.x$ or 
$\lambda x:(\forall a.a\rightarrow a).x (\forall a.a\rightarrow a) x.$
Reduction ($\beta$-reduction, see \cite{barendregt93}) of $\lambda$-terms is an application of a possibly undefined term to a given $\lambda$-abstraction.
For example,
$$(\Lambda a.\lambda  x: a.x) \ Int \ 3$$
can be reduced to $(\lambda x: Int.x)\ 3$, which finally reduces to "$3$", where $Int$ denotes integers (equals $V$ from def.\ref{def:Lambda2ndOrderTermTypes}).

The objective of \index{term reduction} reducing $T_{\lambda2}$-terms is: (1) evaluate a final result and (2) \index{type checking} type checking.
Type checking differs from verification by the absence of the meaning of a state and expressibility limitations.

\begin{definition}[Type Checking]
\label{def:TypeChecking}
Type checking (after Hindley-Milner \cite{hindley69}) of a term $e$ having type $t$ for a given rule set $\Gamma$ is noted as $\Gamma \vdash e:t$.\\\\
The type of a term is checked by \textit{structural typing rules} and further by \textit{reduction} rules for a given formal theory (skipped here):
\begin{center}
\begin{tabular}{c}
\inference[($\forall$-Intro)]{\Gamma \vdash e:t}{\Gamma \vdash \Lambda a.e:\forall a.t}\\\\
\inference[($\lambda$-Intro)]{\Gamma, x:t_1\vdash e:t_2}{\Gamma \vdash \lambda x:t_1.e:t_1\rightarrow t_2}\\\\
\inference[($\forall$-Elem)]{\Gamma \vdash e:\forall a.t}{\Gamma \vdash e \ t':t[a:=t']}\\\\
\inference[($\lambda$-Elem)]{\Gamma \vdash e_1:t_1 \rightarrow t_2\quad \Gamma \vdash e_2:t_1}{\Gamma \vdash e_1 \ e_2:t_2}
\end{tabular}
\end{center}
\end{definition}

For the sake of simplicity, in the definition, base rules are skipped because universality \index{untypable term} allows them to define more generalised and in non-typed $\lambda$-calculi (see \cite{barendregt93}).
For instance, the axiom $\overline{\Gamma \vdash x:t}$, the rule ($\exists$-Intro) or ($\exists$-Elim), can be defined in analogy to ($\forall$-Intro) or ($\forall$-Elem).

Since Coq bases on $T_{\lambda2}$ and \index{redex} \textit{redex} reduction goes from outermost to innermost, type checking becomes decidable.
The checking complexity is linear.
For reductions from innermost to the outmost, type checking is in general undecidable (cf. \cite{peirce10}, \cite{bertot04}).\\

Let us recapitulate the previously said with fig.\ref{code:CoqPeircesTautology}.
For example, the tautology $p \vee \neg p$ seems intuitively straightforward, but when the \index{Tableaux method} Tableaux method is absent and when implications only may be used to reason, any successful proof may become anything than trivial --- so it is with an \index{intuitionistic judgement} \textit{intuitionistic judgement attempt} here.
Depending on the rules chosen, \textit{Peirce's tautology} may become impossible or hard to prove, as shown in fig.\ref{code:CoqPeircesTautology} facilitating a Coq-proof.

\begin{figure}[h]
\begin{center}
\begin{minipage}{10cm}
\begin{verbatim}
 Definition peirce := forall (p q: Prop), ((p->q)->p)->p.
   
 Definition lem := forall p, p \/ ~p.
 
 Theorem peirce_equiv_lem: peirce <-> lem.
 
 Proof.
   unfold peirce, lem.
   firstorder.
   apply H with (q:=~(p \/ ~p)).
   firstorder.
   destruct (H p).
   assumption.
   tauto.
 Qed.
\end{verbatim}
\end{minipage}
\end{center}
 \caption{\index{Peirce's theorem} Peirce's theorem about the \index{excluded third} excluded third in Coq}
 \label{code:CoqPeircesTautology}
\end{figure}

We need to prove the first definition "\texttt{peirce}" can turn into the definition "\texttt{lem}", where "\texttt{peirce}" may not contain disjunction nor negation but may contain implications of a kind $\rightarrow$.
Let us pay attention to both definitions contain quantified assertions.
Any recorded proof would be a sequence of proof commands in Coq.
Any sequence would start with the keyword \index{\texttt{Proof}} "\texttt{Proof}" and finish with the keyword \index{\texttt{Qed}} "\texttt{Qed}" (short form of "\textit{quod erat demonstrandum}", which translates into "\textit{what was to demonstrate}").
Initially, there is only the equality of both theorems to be shown.
Then both sides of the equality are unfolded.
The latter "\texttt{first-order}"-keyword tries to replace $\forall$-quantified assertions into \index{Skolem normal-form} Skolemised \textit{normal-form}.
We need to show on the right-hand side of the definition "\texttt{lem}" for any true proposition "$p$": $p \vee \neg p$ holds.
In this case, the left-hand side is a \index{hypothesis} provable hypothesis $H = \forall p,q. ((p\rightarrow q)\rightarrow p)\rightarrow q$, which must be shown for any propositions $p$ and $q$.
Now, $p \vee \neg p$ is replaced by $p$, and $q$ is replaced by $p \vee \neg p$.
Thus, we obtain $(p \vee \neg p \rightarrow \neg (p \vee \neg p)) \rightarrow p \vee \neg p$ as a hypothesis to be proven still under the condition that $p$ holds.
The consequence of "\texttt{peirce}" must be separated from its precondition in order to prove the hypothesis.
It is required to abstract the hypothesis and then skolemize it to normal-form, s.t. $q$ is no more quantified.
Now we obtain another more comfortable hypothesis, $H_0 = (p \rightarrow q) \rightarrow p$, which can prove if we assume $H = \forall p. p \vee \neg p$, s.t. the right-hand side of the disjunction holds.
So, we introduce a new hypothesis, $H_1 = \neg p$.
Such a choice is arbitrary, and once applied $\neg p$ to $H_0$, we immediately find $p$ holds because the remaining \index{implication} implication always holds as soon as the left-hand side of the implication is false.
Since we assumed $p$ holds initially, we have just confirmed the theorem holds.
The \index{tactics} tactic's application \index{\texttt{tauto}} helps them prove to accomplish since simple mechanical normalisations are sufficient to quit the proof successfully.

Although this example is intuitive and candid, proofs still require extra-ordinary resources and are not intuitive whatsoever, for example,  the definition of intermediate \index{term normalisation} normal-forms and hypotheses.

The application of tactics is no guarantee at all.
Tactics do not recognise non-obvious \index{abstraction} abstractions.
It is not hard to notice that automating this kind of \textit{non-standard transformation} must eventually fail and most likely not even be recognised in the first place.

\begin{definition}[Formal Proof]
\label{def:FormalProof}
A proof is a sequence of applied equalities (rules and axioms) applied to a given theorem based only on valid axioms.
\end{definition}

Remark: If Hoare rules (rules of some \index{Hoare calculus} Hoare calculus) are complete and proof exists upon defined axioms, proof can be \textit{automated}.
If rules are complete, then any valid theorem may be proven within the considered context (w.r.t. exceptions related to G\"odel's incompleteness theorem).
\index{semiotics} Semiotics, \index{syntax} syntax (e.g. see def.\ref{def:FirstOrderPredicateLogicFormula}), and \index{semantics} semantics must be defined to express some formal theory.
Whenever languages are meant, not only a formal theory is determined by expressions, but also the \index{pragmatics} pragmatics of that language shall be defined.
Semantics, e.g. \textit{axiomatic semantics}, decide which meaning is derivable from a formal system or not, particularly by a Hoare calculus.
Signs \index{sign} and the interconnection between formulae are abstractions of some tangible object representatives or physics.
Hence, for historical reasons, \index{meta-physics} meta-physics often replace \textit{physics}.
An abstracted physics, or \index{logics} logics, whenever not a factual matter is meant, but properties about them indeed -- as is the case with Hoare calculi describes program properties.
It must be noted that any \index{formal logic} \textit{formal logic} is also, by definition, \index{algebra!formal} \index{formal algebra} \textit{formal algebra}.

In analogy \index{modular programming} to modular programming, proofs may be composed to simplify and strengthen the essential thoughts from a designer's perspective.
A more formal definition is given in the coming sections.

For analogy purposes, the units used in a proof are \index{lemma} \textit{lemmas} (are auxiliary theorems), \index{theorem} \textit{theorems} and \index{inductive definition} \textit{inductive definitions}.
The simplest example of inductively defined sets are natural \index{natural numbers} numbers (see fig.\ref{code:NaturalNumbers}).

\begin{figure}[h]
$$\texttt{Inductive nat : Set := O : nat | S : nat -> nat}$$
 \caption{Inductive Coq definition of Church's numbers}
 \label{code:NaturalNumbers}
\end{figure}

Since \index{tactics} tactics are arbitrary sequences of verification commands, we will not consider these in more detail here.
They are just like macros, just another modular unit.

In contrast to that, predicates' abstraction has a definite practical meaning: the more straightforward a proof is, the better.
Symbols and predicates may be \index{unfold} unfolded by need.
Facing the current situation and accompanied rules leading to proof would be highly recommended to automate the quest.
One essential question remains: When to fold and when to unfold an inductive definition, how often may this be required?
Unfortunately, the practice taught us a lesson that may be challenging to decide and is almost impossible to predict on an automated schedule.
Due to the reasons mentioned earlier, it is better to reduce terms \index{lazy evaluation} lazily and in outermost-to-innermost order.
Otherwise, the proof may suddenly stop unfinished even if redex(es) still exist.

\begin{observation}[Verification Model of Computation]
\label{obs:ModelOfComputationVerification}

Verification units remind one unit (or modules) in programming.
Lemmas correspond to procedures.
The main theorem to be proven corresponds to the main-function of a program.
\end{observation}

From fig.\ref{code:CoqPeircesTautology}, the following problems may be derived:

\begin{enumerate}
 \item A description of a problem competes with its expressibility.
 This race can lead to severe expressibility limitations and bloated descriptions.
 \item Often, a simplification of equalities of a related theory (e.g. integer theory) bloats the amount and complexity of the significant theory (e.g. theory on dynamic variables).
Although, e.g. an \index{theory of integers} arithmetic theory is not directly related to memory layout or any specific programming statements, this can also be considered a significant hinder to prove automation.
 \item On proof refutation, a reasonable explanation desires better or is absent.
 \index{counter-example} Counter-examples as explanations, if any, are not intuitive and not general enough.
\end{enumerate}

Wos \cite{wos91}, \cite{wos88} gives a more detailed explanation of automated theorem proving regardless of age.
The problems Wos addresses can be categorised into three classes: (i) \index{assertion} assertion model representation, (ii) logical rules representation, and (iii) choice of optimal \index{strategy} proof strategy and tactics.
Wos suggests parallelism, data/knowledge base indexation, and considering modified reasoning techniques to raise proof efficiency.
A more in-depth analysis of articles related to Wos' remarks, the following can be summarised:

\begin{itemize}
 \item \index{type checking} Type checking is helpful to avoid static analysis errors at an earlier stage.
 However, its use beyond that stage tends to be, in fact, nearly zero.
 \item The actual formula notation is not essential as long as the \index{semantics} program semantics (of an imperative program) is caught.
 \item Even if Robinson's \index{resolution method} resolution method is widely discussed in Wos' contributions, the Natural Deduction \index{natural deduction} can be considered more important when it comes to verification in Hoare calculi, presumably due to the proximity to definitive proofs from classic math proofs and is therefore preferred.
 \item Competition can be witnessed in proof searches between local and global optimal values.
\end{itemize}

Wos \cite{wos91} supersedes all others from all problems, namely the need to strip bloated formulae in specifications and repetitions from verification descriptions.
Wos's motto is: "\textit{It is better to simplify assertions than to analyse all different variants}".
Nobody can disagree with such a position.
Wos's heuristic suggests applying simplifications at most with a polynomial complexity in the worst case rather than risk an exponential rise in search space.
Wos expects a reduction of approximately 90\% on average.
Preston \cite{preston88} refers to Wos \cite{wos88}, expecting good research on all of Wos' raised issues will eventually fill a PhD thesis each, especially since not much work concentrated on those topics.
Overbeek \cite{overbeek88} refers to Woss too.
He provides for each of Wos' issue at least one concise practical example.
Mackock's \cite{macock75} contribution may seem old, but the addressed issues are close to Wos' and are, nevertheless, still up to date.
Gallier \cite{gallier03} is a more recent monography on automated theorem proving.
The issues he raises are mostly congruent with Wos'.
Leroy \cite{leroy11} questions the overall verification process altogether.
He provides critical measures from a practical view.
His remarks are heavily covered by the mentioned.
He, Appel and Dockings \cite{appel07} give the prover developers and researchers recommendations to obey common principles they suggest for prover development so that researchers can compare existing provers years later, mainly focusing on rule categorisation.\\

\textbf{\underline{Abstract Interpretation}.}
The Cousots \cite{cousot77}, \cite{cousot92} suggest for the first time a \index{formal method} formal method (see def.\ref{def:FormalProof}) \index{abstract interpretation} a universal \index{static analysis} static code analysis based on approximation of stepwise interpretation calculations for a given program.
The method analyses a CFG (cf.\cite{nielson99}).
If the CFG cannot be normalised (e.g. no unique entry or exit), it is normalised, e.g. by adding auxiliary uniting nodes.
The \index{program statement} program statements are partially ordered in a non-strict way according to an \index{lattice} algebraic lattice in the order of the operator's appearance.
Abstraction helps to approximate the limit \cite{abramsky94} by introducing new dependent parameters using bound intervals (\textit{narrowing down visibility interval}) and by analysing the most common interval (\textit{widening visibility}).
Thus, uninitialised variables are assigned the interval $- \infty \ to \ \infty$.
\index{incremental approximation} Stepwise approximation of interpretations terminates as soon as two consecutive steps yield the same result.
Interpretations renew after each program statement.
Due to the \textit{Halting problem}, interpretations may have unpredictable, so arithmetically huge 
\index{limit} limits.
Alternatively, limits might be tiny, for instance, when entering a loop with an apriori undetermined value.
% Use:
\textit{Static evaluation of allowed array bounds} \cite{nielson99} is a classic application for it.
The family of methods \index{Cousout's method} may be altered and applied to very different applications, for instance, in \cite{gcc15} for an optimised branching based on heuristics in compiled code.
%ISOMORPHISM:
In order to compare two interpretations, an \index{isomorphic mapping} isomorphism over mappings can be defined.
Two interpretations, $I_1, I_2$, are considered isomorphic, iff two \index{mapping} mappings $\alpha, \alpha^{-1}$ exist that are \index{injection} injective, \index{surjection} surjective and for which $I_1$ maps onto $I_2$ by using $\alpha$, and $I_2$ maps back onto $I_1$ by $\alpha^{-1}$.
%End:
Even if Cousots' method is universal, it still lacks restrictions that would allow proof automation for many reasons.
First, the \index{approximation} approximation proposed mismatches the exact nature of Hoare's calculus and the exact symbolic values.
One could consider Hoare calculus as a singularity of Cousouts' method.
However, their method focuses, after all, on arithmetic approximation.
Second, within a Hoare calculus, verification bases on terms and assertions, a specification language, and so forth – this all is not the case with Cousouts' method.
A commented example for their \index{abstract interpretation} application can be found in fig.\ref{fig:CFGEx1}.

\begin{figure}[h]
\xymatrixrowsep{15pt}
\xymatrixcolsep{40pt}
\begin{center}
\begin{tabular}{lll}

\multicolumn{3}{c}{
\xymatrix{
 \txt{$\triangledown$} \ar[d]^{C_0}\\
 *+[F]\txt{x:=1} \ar[d]^{C_1}\\
 *+[o]+[F]\txt{} \ar[d]^{C_2}\\
 *+[F-:<8pt>]\txt{$x\le 100$} \ar[r]^{C_5}_{false} \ar[d]^{true}_{C_3} & \txt{$\vartriangle$}\\
 *+[F]\txt{x:=x+1} *\ar@/^3pc/[uu]^{C_4} 
}}\\\\
%
%%%%%
\multicolumn{3}{c}{\parbox[t]{7cm}{The method \index{program statement} of assigning assertions to \index{CFG} CFGs was first proposed by Floyd \cite{floyd67}.}}\\\\

\quad \begin{tabular}{l}
 $C_0 = [,]$\\
 $C_1 = [1,1]$\\
 $C_2 = C_1 \cup C_4$
\end{tabular} &&
\begin{tabular}{l}
 $C_3 = C_2 \cap [-\infty,100]$\\
 $C_4 = C_3 + [1,1]$\\
 $C_5 = C_2 \cap [101,+\infty]$\\
\end{tabular}\\\\
\multicolumn{3}{c}{\parbox[t]{7cm}{The method \index{program statement} assigning assertions to a \index{CFG} CFGs was first proposed by Floyd \cite{floyd67}.}}
\end{tabular}\\
\end{center}
 \caption{CFG with assigned intervals}
 \label{fig:CFGEx1} 
\end{figure}

Steinbach \cite{steinbach94} illustrates \index{termination} termination as a firm criterion, even a \index{precondition} precondition for \index{confluency} proof confluency (e.g. applying \index{Knuth-Bendix algorithm} the Knuth-Bendix algorithm) and \index{completeness} completeness.
Termination may be \index{decidability} decidable in some instances.

Apart from the mentioned restrictions, another theoretical restriction exists, dating back to \index{arithmetic!Presburger} \index{Presburger's arithmetics} Presburger's arithmetic.
Presburger's arithmetic is a real subset of \index{arithmetic!Peano} \index{Peano's arithmetics} Peano's very popular arithmetic on \index{natural numbers} natural numbers.
Peano's arithmetic only knows "+" (or its inverse operation "-") as an operator \cite{presburger29}.
This discrete structure leads to \textit{usual incoming programs} containing a "+" that may be decided in a finite amount of time in first-order \index{predicate logic} predicate logic (e.g. when calculating object offsets, see sec.\ref{chapter:expression}, sec.\ref{chapter:logical}).
However, this cannot be generalised for expressions under Peano's arithmetic or its extensions.
Notably, this is of interest for automated proof resolvents \cite{cooper72}, \cite{cherniavsky76} (see obs.\ref{obs:ModelOfComputationVerification}).
Expressions in Presburger's arithmetic are decidable with complexity $\Theta(m,n)=2^{2^{n \cdot m}}$, where $n$ denotes the minimal amount of search branches and $m\geq 1$ is some constant factor (see \cite{fisher74}).
It can easily be noticed complexity raises fast for incrementing $n$.
It is also important to notice that ignoring the expressibility of Presburger's arithmetics, e.g. including further operators, may cause in practice non-decidability in certain cases.
However, obeying Presburger arithmetic axioms will \index{supremum} \index{limit} dramatically reduce the upper bound for theorems to be proven.

Pommerell \cite{pommerell94} evaluated theoretical complexity and finds ineffective implementations of related \index{formal theory} formal theories in theorem provers is the biggest hinder in an overall verification barrier.
For example, he demonstrates encoding as \index{Hoare triple} huge matrices to be solved rather than ineffective, despite scarce matrix algorithms.
The situation can be overcome when applied theories can be separated from the core verification rule set.
It is recommended to use \index{SMT-solver} \index{SMT-solver} SMT-solvers instead, which ease tension w.r.t. \index{completeness} completeness issues and further structural rules.
Practice shows \cite{nelson78} solvers may significantly improve the resolution of related theories.
A solver could be initiated whilst verification by need.
Those solvers include \index{language!functional} \index{Why} "\textit{Why}" \cite{why15}, which is, in fact, a functional PL or its successor \index{Y-not} "\textit{Y-not}" \cite{nanevski08-2}.\\

\subsubsection{Alternative Approaches}

As mentioned at the beginning, \index{verification} verification is a \index{formal method} formal method that checks if a given program follows a given specification or not.
Verification \index{Hoare triple} of Hoare triples is a static method that does not require the given input program to be run.
It is worth noting that reducing the rule set may always happen when rules are unified, generalised or both.
Apart from verification presented earlier, other approaches like \index{static analysis} Static Analysis or \index{model checking} Model Checking or localising problems based on stripping syntactic and environmental \index{error localisation} code.
Other approaches exist, e.g. dynamic approaches like AT, bug-free development (hypothetical), and static approaches like some \index{type checking} type checking.

\textbf{Approach 1 --- Testing.} \index{testing} AT is a dynamic attempt in checking cases that used to be problematic or are suspected to be such.
AT checks whether a specified scenario is calculated correctly by a program or not with an expected result.
Scenarios are implemented as test cases by the test developer.
Each scenario should be minimal and self-containing.
One or more tests are associated with a function that all need to be successful to state anything about quality.
As soon as at least one test fails, the overall testing fails.
When a testing scenario is known, all input and output data is specified and then compared.
If a test requires manual interaction, then the test is not automated.
Automated tests significantly contribute to software development as they minimise risks on each stage of development since there is always one step back possible to find a potential code modification that broke the system.
\index{automated testing} Automating Testing (AT) \cite{beck15} yields significant advantages: high efficiency, so for a short period, a considerable amount of tests may be checked.
So, errors based on incorrect or fuzzy input data may immediately be detected.
If a scenario requires many objects involved, this could indeed happen but would require explicit specification.
If tests are simple, then modules can easily be understood, and the risk for errors is lowered.
Ideally, each test specification is simple.
If not, this is always a reason to believe something is over-engineered or error-prone -- this comes from practical experiences.

For a given program $p$ and a test asset $\forall i.t_i(p)$, each test is launched for $p$, and results are observed.
A disadvantage of AT is the manual test setting which can be pretty challenging, especially when the CFG is getting complex per $p$.
The number $n$ of needed tests $p_1,p_2,...p_n$ to cover all \index{test} test branches rises exponentially.
The more complicated a \index{CFG} CFG gets, the more edges linking back to higher BBs it contains, the more tests are needed, and the more complex tests will get in general.

\textbf{Approach 2 --- Automated Model Checking.}
The main problem of approach 1 is the exponential rise in \index{test} test cases needed, which is hard to automate to get a representable amount of tests.
The idea behind Model Checking lies in a formula describing a given program using assertions and temporary expressions, which are checked stepwise.
The problem is to find a solution for the discrete formula depending on multiple (logical) variables systematically.
If the formula is probed too detailed, the AT will not even terminate within a long period.
If, in contrast, the formula is analysed only superficially, then the result will be useless.
The responsibility of certain constraints as equalities is to direct Model Checking.
If equalities are too much restricted, then results may be skipped.
If equalities are too generous, then results may be skipped too.
If there are too few restrictions, AT may take too long or find too many possible unimportant problems first.
An ineffective search strategy might be equivalent to non-termination or an avoidable, very slow \index{model checking} model checking \cite{kohli94}.
Several modifications can be exploited for higher performance or expressibility, namely \index{symbol} introducing formulae \index{model checking} \cite{clarke99} or static methods deciding which part of the model formula to be tested in more detail next based on heuristics \cite{kwiatkowska07}.
Implementations include \index{VDM++} VD\-M++ \cite{weissenbacher01}, \cite{weissenbacher04}.

So, effective alternatives to Hoare-based verification are AT \cite{marinov12}, Model Checking as introduced first by Peled and Clarke \cite{clarke99} 
(\cite{kroening09} \cite{weissenbacher04} give an actual review on recent tools).
The advantage of testing in contrast to Model Checking is its simplicity and execution without explicit formulae, but the problem of a full test specification and execution every time whilst software development is required.

\textbf{Approach 3 --- Bug-free software development.}

Naturally, bug-free software development has only advantages and resolves all other problems presented and, most importantly, does not require verification at all.
After all, this is only a dream since humans make errors.
However, model-based software development approaches try to exploit the minimal distance between a design blueprint and an actual reference implementation -- this difference would be much easier to check, and \index{pattern} patterns seem to have a single significant impact on this \cite{kerievsky05}.
As long as the number of reported errors is tiny and (or) mainly related to minor issues, e.g. regarding performance in non-critical code areas only, then the overall number should approximate zero over time.
Another crucial observation is that the shorter a program gets, the smaller the number of committed errors are on average during the development process.
A program with zero lines has guaranteed zero errors -- this natural law always holds independently from a concrete implementation and does not require further explanation: an empty program has a minimal number of errors, full-stop.
Humans always introduce errors, and by the nature of programming --- this will not be avoided regardless of how strong we try to.
Even if a program was perfect, it might be the specification then which might be not accurate, and as a consequence, both parts require corrections.
Humans do debugging and error localisation, and it cannot be fully automated.
\index{ADT} Abstract data types (ADT), algorithms and intuition, in general, are all human reasoning too.
The software development process, module description and interfaces are all defined by humans.
So, defining a brand-new, fully automated approach  (even with a very high level of artificial intelligence involved) will still not be feasible by computer software only.\\

\textbf{\underline{Software Development}.} 
The developer should think about robust and at most error-free design and stability issues at an early stage supported by the modelling language \index{UML} \textit{UML} and its extension \textit{OCL}. 
The practice has shown on multiple occurrences that proper modelling can avoid crucial design flaws and keep expenses low when it comes to maintenance or new feature implementation.
Unfortunately, not all issues can be resolved apriori, and this is why:
\begin{itemize}
 \item Every developer is an individual, software development.
 After all, it requires creativity, and many skills computers currently cannot be done yet.
 \item It is hard to predict the area where an error may occur
in the future.
 Apart from that, even bug localisation can be tricky sometimes.\\
\end{itemize}

\textbf{\underline{Model transformation}.} 

\index{graph transformation} Graph transformation means a correct initial program is given and needs to be converted into another correct program.
A graph transformation is just for this job, and the challenge is to keep the transformation process safe \cite{dodds08}.
This approach is similar to the last approach presented earlier on approaches on dynamic memory verification.
However, Dodds' approach is based on a \index{term rewriting system} term rewriting system only and therefore is beyond the scope of this work.
Another term rewriting system, but with solid task priority queues, is \index{Stratego XT} \textit{Stratego XT}, \cite{johann03} and \cite{haberland08-1}.
\cite{katz73} suggests heuristics for applying rewrite rules for better confluency.

\textbf{Approach 4 --- Type Checking.} 
\index{type checking} Type checking excludes numerous mistakes in expressions at an early development stage due to incompatible types.
For example, an assignment of a \index{alignment} 32-bit integer to an 8-bit floating-point, in the best case, copies just that small value and causes no harm, if lucky.
In a bad case, it may store any pseudo-random value.

If type checking succeeds, then the proceeding semantic analysis gets valid input.
However, type checking cannot state, e.g. a loop is correctly implemented.
To state just that, one would first specify what "\textit{correct}" means.
Second, check this specification then.
Unfortunately, type checking is insufficient for verification purposes and can only be considered a required precondition.
Type checking e.g. does not contain information \index{data dependency} about data dependencies.
It also does not know which variables will be allocated.
It also does not know how many objects might be in memory. 
The \index{static analysis} static analysis does this, which is an independent phase during software development:\\

\begin{figure}[h]
\begin{center}
\begin{tabular}{l|l|l}
  No. & Phase & Input representation\\
  \hline 
  0 & & command line\\
  1 & Syntax analyis & tokens, derivation tree\\
  2 & Semantic analysis  & yes/no\\
  3 & Code Generation & derivation tree\\
  4 & Static Analysis & yes/no, any data structure\\
  4* & Linking & Target file\\
  5 & Program Invocation & terminal text and errors
\end{tabular}
\end{center}
 \caption{Phases until code generation}
 \label{fig:PhasesOfCodeGeneration}
\end{figure}

Type checking may be formalised as $\Gamma \vdash e:t$, where $\Gamma$ is a term environment containing an associative map onto types, $e$ is a term expression to be checked, and $t$ is an associated type to be checked against (see \cite{barendregt93}).
$e$ requires decomposition to its base statements, and each statement requires a separate check to validate whether a given compound statement $e$ is of type $t$ in $\Gamma$.

In contrast to Hoare triples, terms corresponding to statements are assigned a type.
In contrast to \index{Hoare calculus} Hoare's calculus, \index{type checking} type checking only checks if expressions have a specific type that must match a declaration.

On type checking, \index{program statement} program statements and expressions are checked against expected types for compatibility.
Type checking does not know of \index{calculation state} a calculation state as verification does.

In analogy to Hoare calculi, type checking is a bottom-up approach, so the checking starts with a given expression $e$ and terminates in case of success. 
Type checking, if considered as proof of program properties, is bound to a tiny check only.
It is a single and mandatory check whilst semantic analysis (cf.\cite{isocpp14}).
It does not cover nor overlap with \index{dynamic memory} verifying dynamic memory since type checking is not extensible and has no easy way to represent the dynamic memory model (see later).

\begin{observation}[Phase Type Checking]
\label{obs:TypeCheckingPhases}
Type checking \index{type checking} must be performed before all other semantic phases (especially before verification), which rely on correct typing available.
\end{observation}

When \index{Hindley-Milner calculus} type checking (after Hindley and Millner)  \cite{hindley69}, \cite{peirce10}, \cite{bruce02} \index{Hoare triple} Hoare triples, then it must be noted that "\textit{true}" is the most common precondition.
Furthermore, type checking has fundamental restrictions equivalent to simplified Hoare triple checks (see sec.\ref{chapter:logical}, cf. fig.\ref{fig:HoareCalcVSTypeSystem}).

\begin{figure}[h]

\begin{center}
\begin{tabular}{|cc|c|}
 \hline
 \parbox{6cm}{\ \\Type checking (see \cite{milner78}): $\vdash_{\Gamma} e:t.$\\} && 1\\
 \hline
 \hline
 \parbox{6cm}{\ \\Type checking. Given $\Gamma,e,t$, show $\vdash_{\Gamma} e:t$.\\} && 2\\
 \hline
 \parbox{6cm}{\ \\Type inference. Given $\Gamma,e$, find $t$ ?\\} && 3\\
 \hline
 \parbox{6cm}{\ \\Inhabitant problem. Given $\Gamma,t$, find $e$ ?\\} && 4.\\
 \hline
\end{tabular}
\end{center}
%%%%%%%%%%%%%%%%%%%%%%%%%%%
\begin{center}
\begin{tabular}{|cc|c|}
 \hline
1 && \parbox{7cm}{\ \\Hoare Calculus: $\vdash_{\Gamma} \{P\}C\{Q\} \equiv B$\\}\\
 \hline
 \hline
 2 && \parbox{7cm}{\ \\Proof checking, where $\overline{A_{i,j}}$ is an axiom, $A_{i,j}$ is a rule: $B \vdash_{\Gamma} A_{n,0..} \vdash_{\Gamma} A_{n-1,..} \vdash_{\Gamma} \overline{A_{0,..}}$, given: everything.}\\
 \hline
 3 && \parbox{7cm}{\ \\Logical inference, given: $B$, $\Gamma$, wanted:\\$B \vdash_{\Gamma} A_{n,0..} \vdash_{\Gamma} A_{n-1,..} \vdash_{\Gamma} \overline{A_{0,..}}$.}\\
 \hline
 4 && \parbox{7cm}{\ \\Inhabitant problem: $\{P\}$, $\{Q\}$, wanted: $C$}\\
 \hline
\end{tabular}
\end{center}
 \caption{Comparison between type checking and Hoare calculus}
 \label{fig:HoareCalcVSTypeSystem}
\end{figure}
Extensions are made, which now allow checks on \index{record} structs, type strengthening, properties of non-assignment, and static bounds of struct size \cite{flanagan02}.
Regrettably, \index{dynamic memory} dynamic memory properties may not be checked either.
For example, let us \index{sequentialiser} consider

\begin{center}
\begin{tabular}{c}
  \inference[($\rightarrow$-elim)]{\Gamma \vdash e_1:(s\rightarrow t)\ \Gamma \vdash e_2:s}{\Gamma \vdash (e_1 e_2):t}
\end{tabular}
\end{center}

The rule ($\rightarrow$-elim) demonstrates a type-specification (e.g. $s \rightarrow t$) may be used as a simple specification (see \cite{chatterjee00}, \cite{nanevski08-1}).

In \cite{barendregt93}, all crucial theorems and \index{lemma} lemmas to the \index{$\lambda$-calculus} $\lambda$-calculus are introduced and their proofs.
The paper \cite{plaisted85} is dedicated to \index{non-decidable} non-decidable problems with $\lambda$-terms in untyped calculi (e.g. in \index{Curry's calculus} Curry's calculus, but not in \index{Church's calculus} Church's calculus).
For instance, $\lambda x.x x$ cannot be \index{typing} typed and cannot generally be solved in applications.
Let us assume that $\lambda x.x x$ was typable, then $x$ would have type $t_2$, and the first $x$ would be a \index{functional} functional type $t_0 \rightarrow t_1$.
However, that does not apply to $t_2$ in the first-order typed $\lambda$-calculus, in which \index{type} types are not further parameterised. 
Hence, $\lambda x.x x$ \index{untyped term} is a not typable term.
However, if assuming infinite types \cite{macqueen84}, then an approximating \index{limit} limit may be found.
This type corresponds to the recursively defined \index{record} ADT, which has further (possibly again recursive) dependencies and may be resolved by a union type (a struct or record as in C or Pascal).

\begin{observation}[Verification Non-Reducibility towards Typing]
\label{obs:NonReducibilityToTypeChecking}
In general, type checking can be reduced to verification in a \index{Hoare calculus} Hoare calculus. However, the opposite does not hold.
Fig.\ref{ExampleProofDerivation} clarifies this for the ALC from fig.\ref{fig:ObjectTermsAbadiLeinoSimplified}.
\end{observation}

\textbf{Approach 5 --- Automated problem reduction.}

An \textit{automated problem reduction} \index{reduction} \index{program slicing} tries to cut off fragments (e.g. lines) from a given program listing, s.t. a given error or any other specifiable program behaviour remains.
Other specifiable function behaviour might be, e.g., "\textit{function test123 always returns 5}".
If a program calculates some mapping, but we are interested only in a single value from its domain, then all code may be replaced by a fixed constant except for that singleton.
Often a program error is related to small program fragments due to modularity design.
Unfortunately, bulky or unknown programs may dramatically increase error localisation efforts, which is often the reality.

Assume we build a software project by a single shell command.
For the sake of simplicity, \index{GNU make} \textit{GNU make} \cite{make16} is replaced by a primary \index{program simplification} project build tool, called \index{builder} "\textit{builder}", which may be used for any software project to be built \cite{haberland15-4}.
The slicing tool \index{shrinker} "\textit{shrinker}" \cite{haberland15-3} is used to automate an incoming program reduction.
The program works repeatedly: it builds the software project from sources, analyses it and launches it, then it compares \index{observed behaviour} observed with expected behaviour.
If the program after reduction still builds, runs, and comparison succeeds, then reduction succeeds, and another reduction step is tried recursively.
If the reduction fails, then another reduction is tried.
If no other reduction is available, the overall reduction exits with the previous successful reduction as final.
The program under investigation dumps all it would usually do to \index{stdout} stdout  and \index{stderr} stderr.
By default and w.l.o.g. it is agreed upon that any behaviour may be directed to stdout or stderr -- this may, under the circumstances, may require tracing information issued directly by the program at that place that interest when analysing the problematic code regions (so-called symptom).

A symptom is specified behaviour of a program, either desired behaviour or undesired.
If a symptom still appears after a reduction step, then behaviour remains \index{invariant} invariant, and we continue with reduction until the invariant is violated.
This behaviour is very close and related to how a while-loop is defined (see axiomatic semantic of (LOOP) in fig.\ref{fig:ProofRulesEx1}, cf. fig.\ref{RulesLoopReplacements}).
This approach is best described in algo.\ref{algo:AlgorithmProblemReduction}.
The algorithm is na\"{\i}ve and not optimal as e.g. line selection $linesseq$ is not determined and may contain any number of lines equals to at least one or more lines.
It must be stated, the number of selected lines may fluctuate over time depending on if reduction succeeds or fails for several consecutive times. 
The default reduction strategy may also be tweaked.

\begin{algorithm}[h]
 \caption{Na\"{\i}ve line-based reduction attempt \cite{haberland15-3}}
 \label{algo:AlgorithmProblemReduction}
\begin{algorithmic}[1]
\Procedure {naiveShrinker}{$prog$, $symptom$}
 \State $newProg \leftarrow \emptyset$
 \State $oldProg \leftarrow prog$
 \ForAll {$loc(newProg) \le loc(oldProg)$}
   \State $(newProg, oldProg) \leftarrow (\texttt{REDUCE}(oldProg, symptom), newProg)$
 \EndFor
\EndProcedure
\Procedure {reduce}{$prog$, $symptom$}
 \State $\exists lines \leftarrow lines\-seq(prog)$
 \If {$\texttt{RUN}(prog \setminus lines) == symptom$}
   \State $\texttt{return}\ prog \setminus lines$
 \Else
   \State $\texttt{return} \ prog$
 \EndIf
\EndProcedure
\end{algorithmic}
\end{algorithm}

%\newpage

\subsection{Object Calculi}
\label{sect:TheoryOfObjects}

\index{theory of objects} \index{object calculus} OCs are formal rules that impose how representations of and operation on objects such as type checking and assertion checks.
OCs \index{Abadi-Cardelli calculus} can be divided into two groups, namely Abadi-Cardelli calculi (ACC) and \index{Abadi-Leino calculus} Abadi-Leino calculi (ALC) (cf. fig.\ref{fig:ObjectCalculi}).

\begin{figure}[h]
 \begin{center}
  \begin{tabular}{|c|l|}
    \hline & \\
    Type & Calculus\\
    \hline & \\
    (1) & \parbox[t]{6.7cm}{Abadi-Cardelli \cite{cardelli96}, class object, Example:}\\
    & \\
    & 
\begin{minipage}{6.7cm}
\begin{verbatim}
class MyClass{
  int A; int B; MyClass C;
};
MyClass o1 = new MyClass();
o1.A=1; o1.B=2; o1.C=NULL;
MyClass o2= new MyClass();
o2.A=5; o2.B=3;
o1.C=o2;
\end{verbatim}
\end{minipage}\\
    \hline & \\
    (2) & \parbox[t]{5cm}{Abadi-Leino \cite{leino98}, \cite{abadi97}, \cite{abadi93}, non-class instance, Example:}\\
    & \\
    &
\begin{minipage}{6.7cm}
\begin{verbatim}
object o1 = [ A=1; B=2; C=[A=5; B=3] ];
\end{verbatim}
\end{minipage}\\
    \hline
  \end{tabular}
 \end{center}
 \caption{Object calculi}
 \label{fig:ObjectCalculi}
\end{figure}

So, why do we need a formalisation of objects?
Why can objects not just be replaced by bare variables instead?
%
%
% MEANING OF OBJECT CALCULI:
Second question first: yes, it could be replaced at the end.
Nevertheless, as objects have fields and more interaction may happen, it is, after all, better to abstract more and therefore to use objects as an abstraction of some data container and operations over it.
Algorithms and specification may greatly simplify, and in terms of verification, another indirection may be required --- but there is no apparent reason that could outweigh the advantages.
A class of an object is an \index{ADT} (ADT) and not just a data container.
ADT is also a \index{carrier set} carrier set of objects foreseen for complex problem modelling (see \cite{leino02}).
Operations over objects are closures.
Thus, an object may not suddenly decompose itself while calling a method.
Hence, a developer defines an object only with operations, which are characteristic and presumably affect only its inner state.
Encapsulation is very close to a split-up as a modelling technique with union-types in pre-object-oriented time.
Operations characteristic to objects are in practice class methods, where an object is a class instance.
The \index{paradigm} \index{paradigm!object-oriented} \textit{object-oriented paradigm} has undoubtedly been prevalent, as can be checked periodically with industry's TIOBE index.
The ADT concept can be assessed at a notable single successful contribution -- independently belief or personal preferences anyone may have gathered over the last years.
ADT indeed supersedes all possible refinements, such as encapsulation, polymorphism, message sending, event handling.

% OOA + OO-stuff

So, one advantage of ADT is an improved modularisation concept which allows building even more reliable, flexible, more intuitive, more complex and faster software due to a well-understood set of modelling techniques applicable to objects, for instance,  \textit{object-oriented modelling} (OOM) \cite{dsouza98}, \cite{rumbaugh91} or patterns \cite{kerievsky05}.
Those best-practices are well accepted and wide-spread in the industry.
% patterns and OO-M
The importance of patterns can hardly be overestimated.
Positive effects include the decline in error probability compared to traditional development approaches due to a strict separation of concerns.
Patterns are \index{stereotype} stereotypes.
"\textit{Stereotypes}" are \index{programming idiom} idioms that allow humans to grasp
particular objects and their dependencies on other objects faster.
However, it is an essential precondition for a better understanding of software, including software maintenance and development -- without thoroughly reading the code.
Many philosophical concepts can be considered a precursor of object-orientation, as (logical) \index{atomism} atomism, \index{connectionism} connectionism and many others.
Without any claim of completeness, \index{pattern} patterns could intuitively be interpreted as object comparison OOM with roles assigned to \index{actor} \textit{actors} of a scene: each actor acts according to a global and well-understood written manuscript, a specification for a full shot.
This picture is what a pattern makes a software pattern as we know it.
When it is said, a software developer thinks in patterns, the developer thinks in pictures, which can be fully grasped within just a few seconds.
It is no mystery nor magic at all.
It is common human psychology -- this also counts for how objects shall be structured and how many fields or methods shall be allowed, which entities or compartments shall be chosen better as long as each unit does not have more than seven everything seems fully compatible with the human mind.
A pattern's \index{pragmatism} pragmatism increases when a scenario (e.g. use-cases or tests) can be described as more comfortable and clearer.

% performance issues with objects/records
The \index{object calculus} OC can be used for \index{semantic analysis} semantic analyses, e.g. for type checking or verification.
Apart from that, it may be facilitated to analyse \index{object!field} object field dependencies, e.g. by a compiler \index{register allocation} to ease the tension regarding spills on processor registers \cite{muchnick07}.
Occasionally it may be useful to merge object fields into a \index{array} contiguous array block or split them into \index{variable!local} local variables if both array bound and object size are known apriori to be small (cf.\cite{bornat00}).
The data dependency determines register candidates (may include general-purpose registers and stack cells) whilst \index{register} register allocation.
CPU registers are fast.
However, they occur in a tiny amount and can usually accommodate only tiny words regardless of new architectures trying to lower the pressure.
On the other side, CPU registers depending on a concrete \index{hardware architecture} platform may, under the \index{ABI} \textit{ABI} (application binary interface, see \index{GCC} \cite{gcc15}, \cite{llvm15}, \cite{leroy11}), store only a \index{ABI} few words.
Words are restricted, and even so, an \textit{ABI} usually does not split contiguous blocks of data chunks representing class objects because compilation makes it nearly impossible to predict which object fields will be used and how their memory layout finally looks.
The C (and \index{C++} ISO C++, too) specifications state that objects shall be treated as contiguous memory lump and further details remain unspecified \cite{isocpp14}.
However, if there was an effective way to predict reliably, then the current, solely conservative \textit{GCC} approach could significantly improve the runtime and reduce memory consumption in the ABI on the caller's side.
However, the reality keeps objects always contiguous, even if 99\% of all fields may not even be touched.
If memory would lead to a speed issue, then all memory access would add additional stack store and load operations, strictly speaking, which may not be necessary.

% diff to record calculi
Ehrig's records in his "\textit{Record Calculus}" \cite{ehrig80} only differ by definition.
They are not ADT but just data containers.
His record is formalised as a \index{tuple} tuple (related to Codd's relations) and applied to parallel executions.

Scott \cite{scott76} considers data types (generalised as tuples) as algebraic \index{lattice} lattice, which has initially often an \index{uninitialised} uninitialised \index{limit} \index{infimum} infimum, a final value as \index{supremum} supremum, and specific Galois-joints defined upon set (tuple component) inclusion.
All values a variable may be assigned to between those extremums, infimum at its lower bound and supremum at its upper bound, is encoded as a state tuple that obeys \index{evaluation ordering} evaluation order.
For more details for verification purposes with non-dynamic objects, refer to \cite{chatterjee00}, \cite{mueller02}.
%
% attributes modelling
For further discussion, refer to \cite{barnett04}, where verification conditions are formulae of \index{predicate logic} first-order predicate logic.
\index{field!attribute} Attributed fields are modelled separately and are allocated to dedicated memory regions.
Alternatively, an object's field is fully embedded into another object -- here, all fields and types must be known apriori.
However, this work's main contribution are two predicates dealing with predicate definitions with an intuitive meaning.
In analogy to \cite{barnett04}, \cite{chatterjee00}, dependent fields are highly recommended to be modelled as dependent objects instead.
Both papers suggest input of new assertions for object invariants, which are always valid as long as an individual object lives.
None of the approaches listed is related to \index{pointer} pointers (see sec.\ref{chapter:DynMemProblems}).\\

Class-based OCs are mentioned only briefly.
This family of calculi dominates recent \index{language!programming} PLs, e.g. Java \index{Java} or \index{C++} C(++) \cite{gcc15}, \cite{isocpp14}.
A class is used to instantiate \index{object instance} objects on runtime.
This section is not primarily interested in memory allocation yet, but \index{\texttt{new}} "\texttt{new}", \index{\texttt{malloc}} "\texttt{malloc}", and \index{\texttt{calloc}} "\texttt{calloc}" will be used.
After a short characterisation, we will decide in favour of the object-based OC.
Both calculi, however, allow type checking and enable program specification and verification.
\index{Abadi-Cardelli calculus} Abadi and Cardelli motivate their calculus \index{object calculus} with \index{subclass} \textit{subclasses}, \index{polymorphism} \textit{polymorphism} and \index{recursive type} \textit{closure/recursion through objects}.
Both authors refer to \index{typing} typed \index{$\lambda$-calculus} $\lambda$-calculi as an example of second-order calculus.
It is worth noting that \index{$\lambda$-term} $\lambda$-terms representing \index{object instance} objects may be used as quantified types (so, polymorphic, a type class set implicitly bound to variables by $\forall$, see. \cite{peirce10}, \cite{bruce02}) by using unbound symbols.
Next, recursive and recursion-free classes require a denotational structure.
In other words, \textit{types dependent on new parameters} --- represent an abstraction of $\lambda$-types of higher-order over \index{object type} object types in ACC.

For instance, w.l.o.g., the question of which region in \index{RAM} RAM is used to store a specific object cannot directly be answered since pointers are the only abstraction ACCs have.
The same holds for the question where is an object stored relatively to some top-level object.
$x$ denotes an object in fig.\ref{fig:ObjectPointsTo}, which has one inner object with two content cells, $abc$ and $def$, and some content $ghi$ again.
Apart from that, pointer $y$ points to some object containing $jkl$, where $x$ has some field (or any other subfields from a field), which points to the content of $y$, namely $jkl$.

\begin{figure}[h]
\begin{center}
\begin{minipage}{10cm}
\xymatrixrowsep{15pt}
\xymatrixcolsep{10pt}
\xymatrix{
 x \ar[r] & \txt{\fbox{\parbox{3.7cm}{\fbox{\parbox{2cm}{\fbox{abc} \quad \fbox{def}}} \quad \fbox{ghi} }}} \ar@{.} \ar@/^2pc/[rr]^{} && \txt{\fbox{\parbox{0.7cm}{jkl}}} & & \ar[ll] y
}
\end{minipage}
\end{center}
 \caption{$x$ points to an object which points to $y$}
 \label{fig:ObjectPointsTo}
\end{figure}
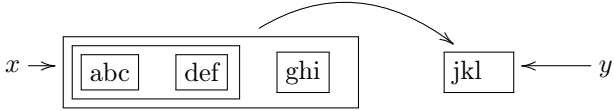

A class hierarchy by inheritance \index{poset} denotes a \textit{poset}, whose \index{infimum} infimum by default is an empty class, e.g.  \texttt{Object} or \texttt{[]}.
Thus, any two classes in the hierarchy are comparable.
Even two unrelated classes may be comparable, according to $\sqsubseteq$, $\not\sqsubseteq$.

ACL allows inheritance, as well as \index{delegation} \textit{delegation}, \textit{inner} fields \index{object!field} and \index{object instance} objects.
The use of keywords: \index{\texttt{self}} "\texttt{self}" and \index{\texttt{super}} "\texttt{super}" as defined in the \index{C++} C++ language \cite{isocpp14} are permitted.
However, their use is underlain by prior semantic checks in the context of the related class.
This additional security measure and further measured required for life-cycle management quickly bloat in \index{object calculus} OCs.
Cardelli notices that additional checks due to actual and pointed classes and just using the keyword \index{\texttt{self}} "\texttt{self}" leads to compulsory checks of the whole \index{inheritance hierarchy} \index{class} class hierarchy.
Extensive unneeded checks can easily be avoided if each variable is fully qualified at an early semantic analysis stage.
Subclassing and reference checks do imply (full) searches on the class hierarchy, which has an acceptable bound of $\Theta(n)=log(n)$.
However, a severe restriction lurks regards \textit{late binding} and virtual methods. 
Practical and theoretical implementation details on inheritance and class hierarchies for \index{C++} C++ can be found in Raman's paper \cite{raman11}.\\

Object methods contain code which by default, is visible only to the object itself.
For simplicity reasons, we will model methods accessible only after \index{\texttt{new}} \texttt{new} creates the object in contrast to \index{static} static methods.
There is no need to access the previously allocated memory region from a particular object in static classes because variables in the stack window are allocated independently from an object.
By convention, every object has a unique method, a constructor, for initialising fields.
Its code is generated to the \index{assembler} assembly file's section \index{constructor} \index{.ctor} \texttt{.ctor} in \index{C++} C++ \cite{isocpp14}.
For destruction, \index{.dtor} "\texttt{.dtor}" is used.
Thus, the following object phases may be defined (see fig.\ref{fig:ObjectLifeCycles}), which will serve us later in specification:

\begin{figure}[h]
 \begin{center}
  \begin{tabular}{rl}
    1. & Constructor (c-tor)\\
    2. & Destructor (d-tor)\\
    3. & Perpetual Presence\\
    4. & Precondition to a Class Method\\
    5. & Postcondition to a Class Method\\
    6. & Arbitrary Verification Condition
  \end{tabular}
 \end{center}
 \caption{Specifiable life-cycle of an object}
 \label{fig:ObjectLifeCycles}
\end{figure}

Points 1 and 2 have just been discussed.
Point 3 relates to any state of a non-dead object.
An object is non-dead iff it is life -- this is after construction but before \index{memory deallocation} destruction.
Points 4 and 5 imply specifications for procedures.
For the sake of application, a forced specification of the whole program shall be avoided.
That refers to the lightweight approach, which will not be considered further yet.
It is agreed upon that any arbitrary assertion may be inserted between statements as we go, which are optional (see sec.\ref{chapter:APs}), in order to increase flexibility and \index{localisation of errors} problem localisation by default.
In both calculi, the methods are not artificially bound apriori.
However, unbound methods may also include code modification, removal of existing and insertion of new methods.
Obviously, in unbound calculi, \index{specification} specifications grow in complexity, and changes are complex and error-prone.
Even worse, unbound calculi may become theoretically not decidable.
Thus, maximum opportunities here lead to too harsh verification limitations.

As mentioned earlier, unbound methods provide neither do they extend \index{Turing-computability} Turing-computability.
Hence, methods with mutable code are no more considered in this work by default.
Object-based calculi have the logical judgement
$$E \vdash obj : A :: B$$
where $E$ is an object environment containing all created objects.

The location $obj$ is an object represent (a variable in calculus 1, and a whole object in the case of calculus 2).
$A$ is of type $obj$ (so the class names $obj$).
$B$ is a specification of object $obj$.
The specification may contain, for example, \index{object!state} the object's internal state description.
If $obj.m()$ is called instead of $obj$, then $B$ has state $obj$ before and after $m$ is called, but $B$ may also include assertions from fig.\ref{fig:ObjectLifeCycles}.
The assertions may not contradict the actual runtime state.
Both the \index{register!special purpose} embedded auxiliary operators and the special-purpose register enable access to the results before and after statement execution.
The register may either contain the whole object (in case 2) or an object field.
To any considered OC, the judgement on objects, type and assertion, is essential.

Remarkably, on the one side, $B$ may alter over time, but on the other side typing of $A$ may not.
That becomes a problem when a \index{signature} remote method alters its signature, but its state does not.
For instance, a web-server or driver are loaded when its specification is fixed, but its behaviour is unknown.

Both \index{Abadi-Cardelli calculus} ACC and \index{Abadi-Leino calculus} ALC do not support pointers (see later).
Pointers avoid copying redundant memory regions, which leads to performance increase unless genuine copies are needed -- this question may not always be obvious to answer in practice (see later).

A \index{class type} class type $T$ is inductively defined as $T$ type, integer or another \index{class} class.
The class object includes fields $f_i$ and methods $m_j$, where all class components are unique, and each field is assigned some type $T$.
Methods have $$T_j \rightarrow T_{j+1} \rightarrow \cdots \rightarrow T_k$$ as a type signature.\\

Next, calculus 2 is considered (ALC).
The PL \index{Baby Modula 3} Baby Modula 3 \cite{abadi93} is introduced by Abadi and the company DEC as an experimental modular PL.
A minimal feature set adopted from \index{Modula-2} Modula-2 is very close to \index{Pascal} Pascal by its syntax and semantic, but with extensions as an OC. 
The language proposed is close to the languages suggested by Abadi, Cardelli and Leino.
It contains all is needed for \index{$\mu$-recursive schema} $\mu$-recursive schemas, particularly assignment and unbound recursion.
\cite{gunter94} contains a detailed discussion on object expressibility, even if some chapters seem old-fashioned at first glance, but it is still relevant today.
Originally calculus 2 was introduced as an alternative representation to the existing widely-known.
Both calculi are dual, which can formally be shown by \index{full abstraction} full abstraction \cite{mitchell96}, \cite{plotkin77}, \cite{cohn83}, \cite{honda05}.
Full abstraction requires to prove two different kinds of semantics given are equivalent, namely \index{denotational semantics} denotational and \index{operational semantics} operational \cite{allison89}, \cite{lavrov01}, \cite{abramsky94}, \cite{winskel93}, \cite{tennent76}, \cite{bird97}, regarding the observed (object) behaviour.
\cite{reus02}, \cite{reus05} is provided for full reference.
Although equivalent, in ACC soundness, it is worth noting that it is much easier to prove (especially the denotational part) and often is preferred in practice.
However, strong supporters of the calculus may dispute the latter remark.
The reason why object-based calculi with operational semantics are so much more difficult is explained by recursion.
If recursion is denoted symbolically, then it is much easier to detect and refer to such than guessing boundaries of object boundaries, especially when it comes to object specifications.
Baby Modula 3's syntax is based on term expressions of kind $a$, as shown in fig.\ref{fig:ObjectTermsAbadiLeino}, where $A$ denotes a typed expression.

\begin{figure}[h]
\begin{center}
\begin{tabular}{rll}
 a ::= & $x$ & .. variable\\
       & fun($x:A$)$b$ & .. sub declaration with body $b$\\
       & $b(a)$ & .. call of $b$ with parameter $a$\\
       & \texttt{nil} & .. nil type\\
       & a[f=b,m=c] & .. object extension\\
       & a.f & .. field access\\
       & a.m & .. method invocation\\
       & a.f:=b & .. field assignment\\
       & a.m:=c& .. method invocation\\
       & \texttt{wrong} & .. wrong type\\
\end{tabular}
\end{center}
  \caption{Definition of object terms after Abadi-Leino}
  \label{fig:ObjectTermsAbadiLeino}
\end{figure}

A term expression is very similar to Pascal's expressions, so it does not require further details.
$fun(x:A)b$ defines an \index{anonymous function} anonymous procedure equivalent to $\lambda$-abstractions with a single input $x$ of type $A$.
W.l.o.g. anonymous procedures accept zero or more arguments.
Variable $x$ is in the procedure's body $b$ and bound on call.
A procedure call, e.g. $b(a)$, implies actual and formal parameter types match.
An empty value \texttt{nil} \index{type compatibility} is by definition applicable with all pointer types.
For example, uninitialised object fields contain \texttt{nil} too.
An object extension adds initialised fields not clashing with existing fields.
Methods are extensible similar to field extension as just described.
Methods may be updated.
\index{wrong} "\texttt{wrong}" is a reserved keyword that indicates an incorrect calculation encountered.
An object can always be defined using a vector that consists of a \index{poset} poset in lexicographical order, and its elements are pairs of field/method name and corresponding value.
The content may again be an object, which is defined in analogy to fig.\ref{fig:ObjectTermsAbadiLeino}.

Its authors provide the semantics of \index{Baby Modula 3} Baby Modula 3 as denotational semantics, and top-level has two stages: \index{type checking} type checking (a fixed rule set) and one for verification (another flexible rule set).
Object type-checking must be component-wise and requires some subclass relationships $A<:B$, which holds iff $B$ contains all fields and methods of $A$.
Since no object identification possible in the object-based calculus as with the class-based calculus, recursively-defined objects must be simulated by \index{fixpoint combinator} fixpoint combinators as in the untyped $\lambda$-calculus having then similar restrictions.
$\mu(X)$ may be used to achieve just the said for any object $X$ according to fig.\ref{fig:ObjectTermsAbadiLeino}.
It is assigned after $\mu(X) \ \tau$ is applied the same way as for $\lambda$-terms, e.g. in $\lambda x.x \tau$.
So, object checking includes variable environments, \index{subtype} subtypes (including sub-classes) and common structural rules.
They are used to perform type checking and object verification (see later, cf. (minimalistic) \textit{LCF} \cite{plotkin77}).
For the sake of simplicity (see \cite{leino98}), structural operational semantics are often used \cite{plotkin81}, although, as mentioned earlier, a transformation into/from denotational semantics remains a quite difficult act.

However, $\mu$-objects are a considerable problem because, for the same initial values, reasoning may differ, so \index{soundness} unsoundness becomes evident.
Let us consider the rules from fig.\ref{UnsoundRules1} as an example.

\begin{figure}[h]
\begin{center}
\begin{tabular}{rcl}
 \inference{\vdash \texttt{nil}:\mu(X)\texttt{Root}\\ \texttt{Self}=\mu(X)\texttt{Root}}{\vdash \texttt{nil}:\mu(X)\texttt{Root}} & & \inference{\vdash \texttt{nil:Root}}{\vdash \texttt{nil}:\mu(X)\texttt{Root}}
\end{tabular}
\end{center}
 \caption{Example of unsound rule set}
 \label{UnsoundRules1}
\end{figure}

So, rules containing recursion may be violated or not provable due to non-confluent proofs (cf. fig.\ref{fig:CRTonHoareTriples}).
Here, only the trivial solution may be considered as the only reliable.
This situation is, unfortunately, not satisfactory.
That is why an object in ALCs, in general, may not be \index{typing} typed in first-order and second-order calculi.
This hinder may be overcome by enriching verification and typing rules with additional context-sensitive rules.
However, this dramatically wor\-sens sim\-pli\-ci\-ty.
In ALC, classes as such do not exist.
Objects may have an arbitrary number of fields with other objects as content.
So, a class may depend on other classes, including \index{typing!higher-order} itself.
This additional dependency motivates the introduction of the third-order $\lambda$-calculus $\lambda^{\rightarrow}_3$ (see def.\ref{def:TypeChecking}).
Thus, verification confluency and type closure significantly impact static analyses with class-types \cite{abadi93}.
If confluency is decidable in a finite amount of time and verification is countable, all named types may be joined together in a union-typed \index{record} struct.
Thus, the newly obtained struct unions all \index{depedent type} dependent types.
The authors of \cite{abadi93} mention united types as described are all enclosed, and all gained (sub-)types may again be used for object-composition.
The gained object type is an algebraic \index{type ideal} ideal based on object comparison \index{type relationship} "$<:$", see further.\\

\begin{figure}[h]
\begin{center}
\begin{tabular}{lll}
 $a,b$ & $::=$ $x$  & .. variable\\
       &     $| \ \texttt{false} \ | \ \texttt{true}$ & .. logical predicate\\
       &      $| \ \texttt{if }x\texttt{ then }a_0\texttt{ else }a_1$ & .. branch\\
       &      $| \ \texttt{let }x=a\texttt{ in }b$ & .. symbolic assignment\\
       &      $| \ [f_i=x_i^{i=1..n},$ &\\
       &      $  \ \ m_j=\psi(y_j)b_j^{j=1..m}]$ & .. c-tor\\
       &      $| \ x.f$ & .. field accessor\\
       &      $| \ x.m$ & .. method invocation\\
       &      $| \ x.f:=y$ & .. field assignment
\end{tabular}\\[0.3cm]
where $x,y,z,w$ are variables, $f,g$ are field objects of object $x$, and $m$ denotes a method.
\end{center}
  \caption{Simplified Abadi-Leino object definition}
  \label{fig:ObjectTermsAbadiLeinoSimplified}
\end{figure}

The PL presented in \cite{abadi97} (see fig.\ref{fig:ObjectTermsAbadiLeinoSimplified}) simplifies the language from fig.\ref{fig:ObjectTermsAbadiLeino}.
However, assignments do not exist.
Instead, there is a symbolic assignment (one-time only) and a \index{conditional branch} conditional branch, which is factually additional \index{syntactic sugar} syntactic sugar.
Furthermore, dynamic methods are prohibited.
The \index{syntax} syntax is defined as in fig.\ref{fig:ObjectTermsAbadiLeinoSimplified}.
\index{$\alpha$-conversion} \index{renaming} Re\-na\-ming on runtime is allowed, but only if the caller and callee obey calling conventions.
\index{variable!local} Locals are prohibited in procedures.
Method parameters are prohibited too if they alter in a method with no accurate restriction towards \index{expressibility} expressibility since additional fields can always be simulated.
Assignment \index{assignment} returns an object as a result.
An essential difference, though, from \cite{abadi93} is ALC excludes recursively defined objects.
The subtype relationship is defined in fig.\ref{DefinitionClassSubtype}.

\begin{figure}[h]
\begin{center}
\begin{tabular}{lll}
$\vdash A <:A'$ & $\Leftrightarrow$ &  $[ f_i:A_i^{i=1..n+p},m_j:B_j^{j=1..m+q} ]=A \ \wedge$\\
                                   \multicolumn{3}{c}{$\wedge \ [ f_i:A_i^{i=1..n},m_j:B_j^{j=1..m} ]=A', \quad p,q \in \mathbb{N}_0,$}
\end{tabular}
\end{center}
 \caption{A possible definition of subclassing}
 \label{DefinitionClassSubtype}
\end{figure}

\index{type relationship} where $A'$ is some class, and $A$ is its subclass.
The relation $<:$ facilitates \index{type checking} type checking, and most importantly, the object \index{rule!object constructor} construction rule (CONS).

\begin{center}
\begin{tabular}{c}
\inference[(CONS)]{A \equiv [f_i::A_i^{i=1..n},m_j::B_j^{j=1..m}]\\ E \vdash \diamond \quad E \vdash x_i::A_i^{i=1..n}\\ E,y_j::A \vdash b_j::B_j^{j=1..m}}{E \vdash [f_i=x_i^{i=1..n},m_j=\psi (y_j)b_j^{j=1..m}]:A}
\end{tabular}
\end{center}

This \index{type relationship} relationship is defined over $A$ being some \index{object type} object type (non-class here), $E$ an object environment, $x_i$ some field value, and $b_j$ some method value (procedure body).
Verification follows in ALC this judgement

$$\sigma,S \vdash b \rightsquigarrow r,\sigma'$$

where $\sigma$ is an \index{initial stack} initial stack, $S$ denotes \index{stack} stack state, $b$ given statement sequence ("\textit{the program}"), $r$ denotes the result which accommodates in the \index{virtual register} \textit{virtual register} of potentially unlimited size, and $\sigma'$ denotes the final \index{memory state} memory state.
\index{specification} Specification of \index{object instance} objects is recursively defined component-wise as

$$[f_i:A_i^{i=1..n},m_j:\psi (y_j)B_j::T_j^{j=1..m}] ,$$

where $A_i$, $B_j$ are specifications, $y_j$ is some parameter \index{anonymous function} of anonymous procedures $\psi$ used in $B_j,T_j$, and where $T_j$ denotes the specification transition from one memory state in the following.

\index{specification} \index{judgement} Narrowing specifications is defined among $E \Vdash a:A::T$, where $E$ is an object environment, $a$ denotes a program, $A$ is some type, $T$ is the transition description.
So, type checking $A$  and transition states in $B$ are applied simultaneously, even if spread around different stages.
It is not hard to see that $B$ is too bloated, and even tiny changes may lead to enormous consequences.
The \index{type relationship} type relationship $<:$ \cite{abadi97} proposed is accordingly to specifications.
Analogous problems affect specifications.
For example, let us verify the program:
$\emptyset \Vdash ([f=\texttt{false}].f:=\texttt{true}).f:\texttt{Bool}::(r=\texttt{true})$.
The corresponding annotated \index{proof!tree} proof tree can be found in fig.\ref{ExampleProofDerivation}.

\begin{figure}[h]
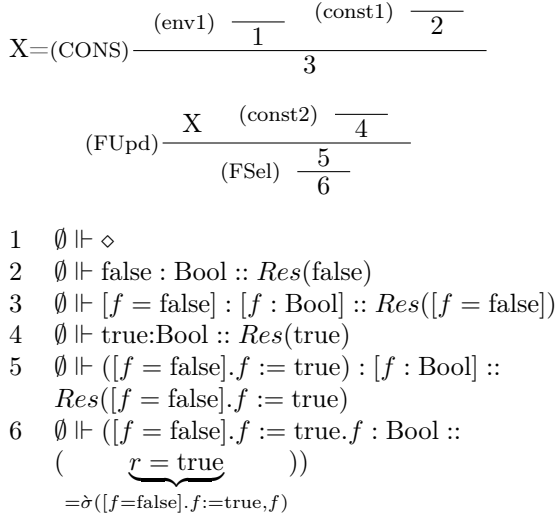

\begin{center}
\begin{tabular}{l}
\begin{tabular}{c}
X=\inference[(CONS)]{\inference[(env1)]{}{1} \quad \inference[(const1)]{}{2}}{3}\\\\
\inference[(FUpd)]{X \quad \inference[(const2)]{}{4}}{\inference[(FSel)]{5}{6}}
\end{tabular}\\\\

\begin{tabular}{ll}
 1 & $\emptyset \Vdash \diamond$\\
 2 & $\emptyset \Vdash \texttt{false}:\texttt{Bool}::Res(\texttt{false})$\\
 3 & $\emptyset \Vdash [f=\texttt{false}]:[f:\texttt{Bool}]::Res([f=\texttt{false}])$\\
 4 & $\emptyset \Vdash \texttt{true:Bool}::Res(\texttt{true})$\\
 5 & $\emptyset \Vdash ([f=\texttt{false}].f:=\texttt{true}):[f:\texttt{Bool}]::$\\
   & $Res([f=\texttt{\texttt{false}}].f:=\texttt{true})$\\
 6 & $\emptyset \Vdash ([f=\texttt{false}].f:=\texttt{true}.f:\texttt{Bool}::$\\
   & $(\underbrace{r=\texttt{true}}_{= \grave{\sigma}([f=\texttt{false}].f:=\texttt{true},f)}))$
\end{tabular}
\end{tabular}
\end{center}
 \caption{Proof tree example for an object-based calculus}
 \label{ExampleProofDerivation}
\end{figure}

The definitions of used rules are in fig.\ref{ExampleRulesetObjectCalculus}.

\begin{figure}[h]
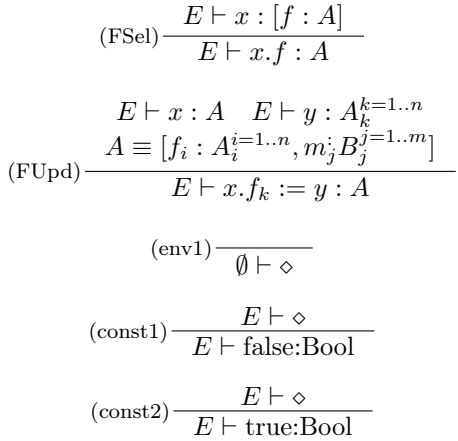

\begin{center}
\begin{tabular}{c}
\inference[(FSel)]{E \vdash x:[f:A]}{E \vdash x.f:A}\\\\

\inference[(FUpd)]{E \vdash x:A \quad E \vdash y:A_k^{k=1..n}\\ A \equiv [f_i:A_i^{i=1..n}, m_j^:B_j^{j=1..m}]}{E \vdash x.f_k :=y:A}\\\\

\inference[(env1)]{}{\emptyset \vdash \diamond}\\\\

\inference[(const1)]{E \vdash \diamond}{E \vdash \texttt{false:Bool}}\\\\

\inference[(const2)]{E \vdash \diamond}{E \vdash \texttt{true:Bool}}
\end{tabular}
\end{center}
 \caption{Example rule set for an object calculus}
 \label{ExampleRulesetObjectCalculus}
\end{figure}

Assertion $\diamond$ denotes true.
More on the chosen rule set's \index{soundness} soundness property can be found in Abadi's paper \cite{abadi93} and its accompanying technical report with the same title.
It is worth noting \index{Abadi-Leino calculus} ALCs lack \index{abstraction} abstraction techniques, as introduced in the previous sections.
The parameter visibility only goes from outside in, but not vice versa.
For instance, in $b_1 \equiv \texttt{let }x=(\texttt{let }y=\texttt{true in }[m=\psi (z)y]) \texttt{ in }x.m$, the inner $y$ cannot be checked outside the procedure.
According to the rule set from \cite{abadi93} and the difference between predicates (whichever order and restrictions it may finally be) and \index{assertion} assertions, this comes as a surprise.
Abadi and Leino suggest future work improvement of specification abstraction, further research on mutable parameters, especially w.r.t. completeness, pointers over objects and object fields, and research on recursively defined objects.
None of the selected calculi, neither ACC nor ALC, pay attention to pointers or \index{alias} aliases (also see later).
In \cite{cardelli96}, essential extensions are named but not yet been proposed: multi-threaded execution, addressing \index{field} fields, use of abstraction in specifications as initially was proposed by Hoare \cite{hoare69}.

%%%%%%%%%%%%%%%%%%%%%%%%%%%%%%%%%%%%%%%%%%%%%%%%%%%%%%%%%%%%%%%%%%%%%%%%%%%%%%%%%%%%%%%%%%%%%%%%%%%%%% 

Banerjee \cite{banerjee08} presents a language that supports objects accommodated in the stack \index{stack} within the same memory region (cf.\cite{tofte97}).
Banerjee's approach moves all \index{variable!local} local variables into the stack.
However, \index{dangling pointer} dangling pointers are forbidden by default.

Recursive predicates over objects are also forbidden.
Unique about \index{global invariant} global invariants is the dependencies between objects that remain invariant.
He warns about a strongly emerging problem of abstraction and supports the initiative of an object-centric predicate view.\\

%%%%%%%%%%%%%%%%%%%%%%%%%%%%%%%%%%%%%%%%%%%%%%%%%%%%%%%%%%%%%%%%%%%%%%

\textbf{\underline{Programming Threads}}. In \cite{abadi97}, \index{object} objects qualify as \index{ADT} ADT but not as a usual record (struct) \cite{ehrig80}.
The difference is that records are defined and identified not necessarily by their structural components, not primarily by their fields.

\index{graph transformation} \textit{Graph transformation} can be considered as \index{reduction over records} \textit{reductions over records}.
Ehrig and Rosen suggest a subcategory of procedures known to be safe, so it does not mix dependent object fields with analysis using Category Theory \cite{pierce91}.
Their approach may be considered safe towards single-threading application if access to a subset of object fields does not depend on each other.
That is how the authors grant sufficient field access without \index{thread blocking} blocking an entire thread.
Thus, the inherited procedure and \index{GC} GC may simultaneously run if they do not influence the program's behaviour.
Similar approaches in accessing object are in \cite{jones11}.

Huisman \cite{huisman07} also investigates multi-threaded access to an object, as do Hurlin \cite{hurlin09} and \cite{ehrig80}.
By \index{design by contract} design-by-contract, Huisman tries to grant \index{method} method access, but he also attempts \index{temporal assertion} temporary access control.
Huisman suggests tracking a specification's history of accessed objects and \index{object invariant} methods.
It is also suggested to enrich object descriptions with class \index{stereotype} stereotypes of collaborating object commands, due to a lack of \index{abstraction} abstraction level, especially in multi-threaded method calls.\\

\textbf{\underline{Frameworks}}.
Alternatives or modified \index{formal method} formal methods are needed (see def.\ref{def:FormalProof}) to achieve a specification of whole programs, one of which Meyer's approach about Separation of Concerns \index{separation of concerns} \cite{meyer98}.
Meyer \cite{meyer92} suggests, as many did before him, a methodology \index{Liskov's principle} based on Liskov's modularity approach, in analogy to \index{Parnas' principle} Parnas' hiding principle.
In tight algorithms, \index{actor} actors are assigned \index{stereotype} stereotypes according to which they act.
The definition of a role in a given \index{ontology} \textit{ontology} of many participants and collusions allows a better understanding of the given object structure.
Meyer's theorem may be summarised as:

\begin{center}
 The responsibility is divided [among ontology units], the more such derived objects may be trusted.
\end{center}

The intention is: the smaller a program gets, the more its \index{complexity} complexity is reduced.
There, the initial program is not touched, except for the performed reduction.
A metric, e.g. \index{metrics} McCabe's metric, is a real number that measures the complexity of a given program w.r.t. its \index{CFG} CFG.
The more a program is split to easy to control and logically related units, the simpler the CFG becomes.
Metrics correspond to programs, which becomes easier to compare.
We notice that an object that may occur simply at first glance may have many dependencies with an arbitrary number of other objects, which may not be known before whether the program is rather complicated.
Such \textit{unprecedented}, unexpected dependencies often are covered and hard to grasp totally at first glance.
Besides, they often appear as sources of errors due to poor design and are a clue in serious mismatch to specifications.
Often remote objects being not related to a scenario are involved, which to simple models and algorithms are unlikely and lead to situations where two related objects actually may not be connected.

Riehle \cite{riehle08} suggests tests be introduced covering \index{specification} specifications, s.t. localising specification mismatches on object-level is simplified within (object) \index{collaborations} \textit{collaborators}.
From his role-based inheritance models, variability can be observed as the single most important property besides extensibility, which is granted by zero or more aggregations.
Variability and extensibility together can be summarised as essential platform features (API) for (re-)designing, and modelling frameworks --- the heap verifier in this work, is a framework.

Interconnection between object modules are checked against specifications, e.g. using \index{modular algebra} Modular Algebra \cite{bergstra90} or \index{component algebra} Component Algebra \cite{feijs02}, \cite{sifakis05}, \cite{tonella01}.
These two approaches are well-founded and formalise interfaces between class objects \index{class instance} and other objects.
The interface description, however, is bound by input and output method parameters.
\cite{fischer00} describes a navigational attempt over objects according to a given specification.
In \cite{bergstra90}, \cite{feijs02} and \cite{assman06}, the method calls are recommended to track and be callable in constraints, so, e.g. a method may not be called until a list of all previously called methods is completed.
This approach seems very similar to Huisman's approach.
Since program verification often insists on a full \index{full specification} specification, it makes the overall verification and specification a nightmare since only a single gap would invalidate the whole specification.
When it comes to bigger programs, all that becomes close to absurd.
However, nearly all the approaches presented rely on a full specification \index{language!programming} and do reason about \index{pointer} dynamic memory.
A method call sequence may be useful not only inside an object as proposed but also as an \index{temporal assertion} extension that might be beneficial for the web of objects, e.g. exclusion of method calls to $m_2$ on any object first unless method $m_1$ is called on the third object.
Currently, the programmer must \textit{know} somehow a certain method may not be called unless a certain condition is met -- the same holds for the initialisation of field objects or arbitrary other calls and order.
This learning experience is crucial for any \index{framework} framework.
Instead, a static approach in analysing the object calling order would be convenient.

The \index{OCL} \index{Object Constraint Language} \textit{Object Constraint Language (OCL)} \cite{oclspec} is an extension of the \index{UML} \textit{Unified Modeling Language (UML)}.
OCL is graphical and textual and is the de-facto standard when modelling static and dynamic aspects of the desired software.
OCL allows insertion of formulae describing object interdependencies, e.g. it may restrict types in object relations.
Formulae describe mainly \index{$\lambda$-term} second-order predicate logic terms (see def.\ref{def:TypedLambda2ndOrder}), such as atomic types.

Apart from the earlier mentioned restrictions from previous sections, the lack of pointer specification solely remains.
Safonov \cite{safonov10} introduces and analyses many examples in compiler construction and related field areas.
His analysis accumulates in postulated criteria upon \index{trustworthy compilation} trustworthy compilation.
The author chooses compilers due to their very high complexity comparing their vast state transitions with others.
The complexity can easily be validated, for example, by a thorough code inspection generated, especially regarding branch instructions, or by analysing metrics for some PL processor, e.g. for \index{C++} C++.
The complexity analysis might be quite challenging when showing any given transformation is sound and optimal (also complete, reminding the \index{quality ladder} quality ladder from fig.\ref{fig:QALadder}).
In order to better achieve this objective, Safonov refers to Meyer's methodology.
However, as a concrete counter-argument against Meyer's approach, Leroy's CompCert \cite{bertot06} shall be named, which researches soundness, but only partially does the research performance.
Although the methodology proposed seems promising, unfortunately, it is tough to derive a significant practical meaning out of it because of the harsh practical restrictions.
Hence, Safonov considers the following criteria reasonable in practice instead:

\begin{itemize}
 \item Correct and plausible error diagnosis, including \index{counter-example} counter-example generation.
 \item \index{model transformation} Model transformation, including objects, is preferred due to small specification efforts.
 \item The option to extend and adjust \index{object calculus} object-based calculi.
 \item Use of formal descriptions if usability does not suffer.
 Particularly, Safonov considers complete coverage of a formal model and \index{locality} locality of a specification's most critical single criteria.
 \item \index{ADT} ADTs describe object invariants.
 \item IRs must be flexible enough and abstracted, s.t. they might be used on different verification stages.
\end{itemize}

%%%%%%%%%%%%%%%%%%%%%%%%%%%%%%%%%%%%%%%%%%%%%%%%%%%%%%%%%%%%%%%%%%%%%%%%%%%%%%%%%%%%%%%%%%%%%%%%%%%%%%%%%%%%%%%%%%%%%%%%%%%%%%%

\subsection{Dynamic Memory Models}
\label{sect:HeapModels}

\subsubsection{Transformation into Stack}
\label{sect:StackAlignment}

In sec.\ref{sect:TheoryOfObjects}, a discussion on OCs and transformation into stack took place.
\cite{tofte97}, \cite{tofte94}, \cite{regionmem10} research the \index{paradigm!functional} functional paradigm for verification where all variables are moved to the \index{stack!window} stack windows (in \cite{calcagno01}, a formal proof of soundness is demonstrated by Tofte and Talpin \cite{tofte97}).
Grossman's C-dialect called Cyclone is presented in \cite{grossman02} and implemented in ML.
It reasons about so-called regions (in the so-called memory Region Calculus) and is defined as a set of corresponding heap elements, which may be utilised at once in a release phase.
Particularly the transformation into stack must now pay special attention w.r.t. variable visibility, s.t. accidental destruction is avoided by any means.
In \cite{fluet06}, variable visibilities beyond a current stack window are discussed in more detail, and an extension towards Cyclone is proposed.
If the \index{visibility scope} variable visibilities are not carefully transformed, then derived \index{alias} aliases, particularly the "\textit{maybe-aliases}", may become unsound.
The functional approach excludes \index{dynamic parameters} dynamic parameters, recursively defined structures and procedure which return \index{functional} functionals.
All function types must be defined at compilation time.
In Meyer's \cite{meyer1-03}, \cite{meyer2-03} opinion \index{GC} GC (see \cite{larson77}) is most important beside soundness.
Hence, he proposes to fully avoid dynamic memory and re-arrange algorithms to work only with the stack instead.
As mentioned in \cite{tofte97}, the push into the current \index{stack!push} stack window does not occupy more than through a minimal size.
He insists object abstraction is dedicated to more attention and suggests letting dedicated variables be in a specification.
However, the full proof for \index{soundness} soundness is missing.
Even a simple \index{assignment} assignment may be unsound -- as not mentioned in the paper, it may alter the program code whilst program runtime.
However, this sporadic case can, under usual circumstances, be ignored.
Access to stack content may be faster than to dynamic memory but always requires overhead for pushing in and popping out data, which depending on the considered algorithms, may significantly reduce the overall performance, s.t. the use of dynamic memory may again be faster (cf.\cite{appel87}).

\subsubsection{Shape Analysis}

The main objective of \index{alias analysis} AA \cite{sagiv02, nielson99, pavlu10, distefano06, shapeanal10} is to find \index{invariant} invariants in a \index{heap} heap on program runtime, e.g. for the localisation of \index{alias} aliases.
Shapes are corresponding invariant and flexible partitions of the dynamic memory \cite{shapeanal10}, e.g. for a simply-linked list \cite{distefano06} or a doubly-linked list \cite{cherem07}.
A \index{dependency graph} dependency graph is always fully described by its \index{transfer function} transfer functions like empty shape, \index{field} field or pointer assignment, new \index{dynamic memory} dynamic allocation.
Sagiv et al. \cite{sagiv02} introduce relationships between two \index{pointer} pointers: \index{alias} "\textit{is alias}", "\textit{is not alias}" and "\textit{maybe alias}".

Pavlu \cite{pavlu10} noticed first that \cite{sagiv02} and \cite{nielson99} might lead to an imprecise (meaning unsound) reasoning depending on a given program.
Namely, if for a given if-then-statement in one case "\textit{maybe alias}" and in the alternative case "\textit{alias}" is inferred.
Since the overall "alias" defined by Sagiv's calculi would not be sound, the correct answer would be "\textit{maybe alias}". 
Moreover, \cite{pavlu10} contains a detailed comparison between \cite{sagiv02} and \cite{nielson99}.
\cite{pavlu10} assesses \cite{nielson99}.
The finding of \cite{nielson99} is a more precise method and allows a path optimisation of \index{shape graph} shape graphs with identical beginnings and aliases.
These shape graphs may be reused as more effective and shorter.
The optimisation proposed has an upper bound of $\Theta(n^2)$ and reduces the calculation on average by approximately 90\%, which is quite impressive.
Pavlu suggests complex AAs is simulated extra-procedural, s.t. procedure calls and \index{variable!global} globals are turned into local \index{intra-procedural alias analysis} intra-procedural elements as much as possible by renaming.
His proposition has already been applied to the current \textit{GCC} implementation.
Both, Pavlu and the author of this work consider context-independent attempts in exact AAs have no future (cf.\cite{hind01}) since characteristics seem evident and devastating.
Furthermore, he suggests an optimised separation of merged vertices, which build up sub-graphs from the vertices, which strictly do not contain aliases.

\cite{parduhn08} represents a tool for visualising heap shapes for fast analysis and localisation of \index{invariant} invariants and rarely used bridge-edges between loosely coupled shapes for navigation in a graph \index{graph abstraction} abstraction and \index{graph concretisation} concretisation.
Here, concretisation allows to fold \index{fold} and unfold \index{unfold} subgraphs (all manually).
Visualisation of \index{transfer function} transfer functions leaves much to be desired.
A transfer function implies a \textit{folds} or \textit{unfolds} of or into \index{subgraph} subgraphs.
However, a principal improvement is not expected soon due to fundamental restrictions here.

\cite{calcagno09} presents an attempt based on \index{heap separation} heap separation to shapes.
The approach compares potentially matching beginnings of rules by length (longest first).
Special is the \index{approximation} approximation of both sides of considered rules for a rule selection based on \index{abduction} bound abduction.

\subsubsection{Pointer Rotation}

\cite{suzuki82} suggests \index{pointer rotation} "\textit{pointer rotation}", which is safe under these conditions: (1) the content of heaps remains invariant, (2) all elements in memory remain valid before and after rotation, (3) the total number of variables remains invariant.
The advantage of a safe rotation is, first, the absence of GC -- this can easily be defined, and, second, effective list operations, e.g. copying.
Although pointer rotations do not insist on a direct specification, a minimal parameter modification can still lead to a hard-to-predict situation, especially when pointers are \index{alias} aliases.
\cite{suzuki82} suggests a base asset of rotations.
However, this may be insufficient in practice, as more flexibility and a plausible framework upon base rotation composition are crucial.
Furthermore, rotations are very sensitive to minimal changes on calls.
For example, changing the order of compatible parameters may look safe at first glance -- but it may fatally destroy or even free the existing heap totally as a surprise and, in fact, as an undetected corner-case.
Surprise also occurs with harmless-looking input data -- e.g. when the first parameter is an empty list -- or when rotations are sequenced.
Any rotation may be represented as an element of a group permutation, where its components $x_j$ denotes some pointer location -- for instance:

\[
  \begin{pmatrix}
   x_1 & x_2 & x_3 & \cdots x_{n-1} & x_n\\
   x_2 & x_3 & x_4 & \cdots x_n & x_1
  \end{pmatrix}
\]

Naturally, rotations may be composed.
For example, let us consider a modified sample from \cite{suzuki82} which is illustrated in fig.\ref{ExamplePointerRotationSuzuki}.

\begin{figure}[h]
\begin{center}
\begin{minipage}{7cm}
\begin{tabular}{ccc}
\begin{minipage}{3cm}
\begin{verbatim}
var temp: T;
y:=NIL;
while x!=NIL do begin
  temp:=x^.next;
  x^.next:=y;
  y:=x;
  x:=temp;
end
\end{verbatim}
\end{minipage}
 &
 $\Leftrightarrow$
 &
\begin{minipage}{5cm}
\begin{verbatim}
y:=NIL;
while x!=NIL do
  rotate(x,x^.next,y);
\end{verbatim}
\end{minipage}
\end{tabular}
\end{minipage}
\end{center}
 \caption{Code example for a pointer rotation after Suzuki (in Pascal)}
 \label{ExamplePointerRotationSuzuki}
\end{figure}
 
The code illustrated operates with a single-linked list with three difficult elements to transform and consists of steps 1 to 4 (see fig.\ref{HeapWatchOnProgrammExecutionPointerRotation}).

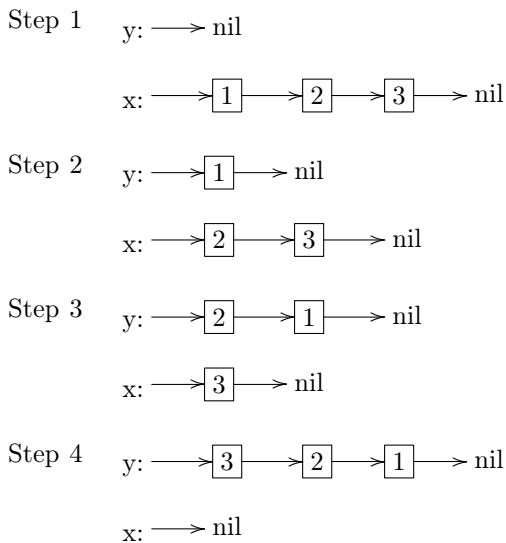
\begin{figure}[h]
\begin{center}
\begin{tabular}{ll}
 Step 1 &
 \xymatrix@C=2em@R=1em{
  \txt{y:} \ar[r] & \texttt{nil}\\
  \txt{x:} \ar[r] & *+[F] \txt{1} \ar[r] & *+[F] \txt{2} \ar[r] & *+[F] \txt{3} \ar[r] & \texttt{nil}
 }\\[1.5cm]
Step 2 &
 \xymatrix@C=2em@R=1em{
  \txt{y:} \ar[r] & *+[F] \txt{1} \ar[r] & \texttt{nil} \\
  \txt{x:} \ar[r] & *+[F] \txt{2} \ar[r] & *+[F] \txt{3} \ar[r] & \texttt{nil}
 }\\[1.5cm]
Step 3 &
 \xymatrix@C=2em@R=1em{
  \txt{y:} \ar[r] & *+[F] \txt{2} \ar[r] & *+[F] \txt{1} \ar[r] & \texttt{nil} \\
  \txt{x:} \ar[r] & *+[F] \txt{3} \ar[r] & \texttt{nil}
 }\\[1.5cm]
Step 4 &
 \xymatrix@C=2em@R=1em{
  \txt{y:} \ar[r] & *+[F] \txt{3} \ar[r] & *+[F] \txt{2} \ar[r] & *+[F] \txt{1} \ar[r] & \texttt{nil}\\
  \txt{x:} \ar[r] & \texttt{nil}
 }\\
\end{tabular}
\end{center}
 \caption{Dynamic Memory whilst rotating pointers}
 \label{HeapWatchOnProgrammExecutionPointerRotation}
\end{figure}

Depending on input data, one or more properties are observed, some of which are quite hard to generalise, others are not.
For instance, \texttt{rotate(x,x,y)} is a congruential mapping, s.t. 
$$\texttt{tmp:=x; x:=y; y:=x; x:=temp;}$$
Even a tiny modification may invalidate previously considered properties.
So, naturally appearing symmetries suddenly may not hold, which would occur alienating and spooky to a verifier, e.g. $rotate(x,y,x^{\wedge}.next)$, since $x$ and $x^{\wedge}.next$ are connected, and its intersection is not empty apriori.

%%%%%%%%%%%%%%%%%%%%%%%%%%%

\subsubsection{File System}

As we understand right now, a heap graph can also be simulated by a file system, s.t. vertices are files, and links are edges (e.g. "\texttt{ln -s}" creates a so-called \textit{symbolic link} under \index{linux} Linux).
Directories may be considered as objects, data containers literally.
The content of a vertex is in its file so that content may be huge.
Thus, a file system imitates a heap graph.
Now a schema may describe a given file system without tracing symbolic links.
The proposed structure is a tree and maybe encoded as semi-structured, such as an \index{XML} XML-document \cite{haberland07-1}, \cite{haberland07-2}.
In Linux, all operational resources are files.
Operations over files now correspond to operations over heaps.
The objective of verification may be defined as a check of a given file (system).
From another perspective, a file operation may be checked too by a soundness check (e.g. by an XML-schema), which may immediately reduce errors and defects with big and many files.
For the sake of simplicity, "\texttt{sha1sum}" or "\texttt{md5sum}" shall be used, especially for big and many files.

For check operations \cite{haberland08-1} and knowledge representation \cite{haberland08-2}, a declarative language might be used, ideally Prolog.
The outlined example \cite{haberland08-1} calculates two files from the top-level directory $a$ and matches both cases.
The top-level directory may be prefixed.

%\begin{center}
\begin{minipage}{12cm}
\begin{verbatim}
template(element(top,_,[A,A]),[text(T)]):-
     A=element(a,_,_), transform(A//p#1,T).
\end{verbatim}
\end{minipage}\\
%\end{center}

A maximal approximation in \index{XSL-T} XSL-T would be:\\

\begin{minipage}{12cm}
\begin{verbatim}
<xsl:template match="top[count( child::*)=2] and 
   a[1] and a[2]">
  <xsl:text>
    <xsl:value-of select="//a//p"/>
  </xsl:text>
</xsl:template>
\end{verbatim}
\end{minipage}

%%%%%%%%%%%%%%%%%%%%%%%%%%%%%%%%%%%%%%%%%%%%%%%%%%%%%%%%%%%%%%%%%%%%%%%%%%%%%%%%%%%%%%%%%%%%%%%%%%%%%%%%%%%%%%%%%%%%%%%%%%%%%%%%

\subsection{Further Areas}

Discussion and basic definitions on aliases and GC can be found in sec.\ref{sect:HeapModels}.
Furthermore, related areas are presented, which may be added to the framework suggested from sec.\ref{sect:ArchitectureVerificationSystem} and can be characterised as pointer analysis.\\

\subsubsection*{\textbf{\underline{Alias Analysis}}}
\label{sect:AliasAnalysis}

Weihl \cite{weihl80} gives an illustrative and still up-to-date review on this topic regardless of its age.
According to Weihl, AA is an emerging \index{pointer approximation} approximation process of pointers to the same memory cell, which is mutable, and as such, the cardinality of a dependent set of aliases rises exponentially.
AA beyond a procedure's boundaries is much more difficult in general because all possible calls and continuations of incoming pointers may be altered anywhere else.
Pointed content may alter in continuations, and this makes the overall analysis inequally more complex.
Weihl stratifies aliases to simplify, s.t. its complexity weights each program statement.

\cite{muchnick07} contains a full review of past and present \index{alias analysis} AAs, including his well-established too.
The papers from Muchnick and Weihl can, without exaggeration, be considered as the single most influential and most important in AA.
He parts analyses into dependent from a given \index{CFG} CFG and independent.
He also divides analyses by their result into "\textit{alias}" and "\textit{maybe alias}", and into closed subprograms and subprogram calls.
Naturally, although not done yet, Muchnick's AAs may be distinguished even further.
For example, Landi \cite{landi91} focuses on pointer visibility scopes.

Opposite to Muchnick, Cooper \cite{cooper89} suggests first an independent AA.
The primary motivation behind it is a substantial simplification because this is indispensable to track pointers and operations.
Imprecision is a significant drawback here.

It is worth noting that \index{GCC} GCC \cite{gcc15} allows the \index{C} C-programmer to tweak directives and pragmas manually, s.t. (strict) AA is skipped for selected functions.
Thus, generated code may become much more performant.

The approach in \cite{horwitz89} can be considered as Muchnick's extension.
By replacing the connectivity data structure with \index{bitfield} bitfields, this approach becomes more effective.
A general improvement may be derived from the analysis framework proposed by Khedker \cite{khedker09}.
Moreover, the presented framework is dedicated to assignments too.
It is extensible, might be generalised and is expected to incorporate even \index{SSA-form} \textit{SSA-form} \cite{cytron91}, \cite{ssabook15}, \cite{vert13}, \cite{sethi75}.
Naeem \cite{naeem09} asserts that access is improved by hashing (all) subsets of calculated aliases.
Notably, Naeem's approach contains a hidden proposition in rewriting \index{alias} AAs as a problem towards \textit{SSA}.
Notably, pointers are not just the only variables, which are hard to locate whilst static analysis \cite{srivastava92} beyond a procedure's boundaries --- they are one of the hardest problems ever.
Srivastava's approach \cite{srivastava92} is useful, apart from all other discussed work, because global optimisations during linking, e.g. whilst dead-code elimination, often may not be fully triggered on earlier compilation stages.

Pavlu \cite{pavlu10} divides \index{alias analysis} AAs into two groups:
a) unification-based \index{unification} \cite{steensgaard96}, which finds more \index{alias} aliases,
b) \index{floating approach} floating-based approaches that are much faster for the cost of higher imprecision.

Both \cite{landi92} and \cite{ramalingam94} reasonably question why \index{alias} alias-analysis is not yet resolved satisfactorily.
The answer is found in its complexity in locating aliases beyond procedure boundaries and its current limitation for the case "\textit{alias}".
In compressed dependency subgraphs, naming remains an unexpected but hard problem \cite{sagiv02}, \cite{nielson99}.
Moreover, these approaches are restricted in numerous ways: pointer access by index only allows constant expressions; pointers with objects are forbidden; C arrays with a dynamic length are forbidden (as so-called variable-length arrays introduced by C99 and available by default in Clang/LLVM); interleaving memory regions are forbidden (as \index{union} union-structures in C \cite{isocpp14}).

Clint \cite{clint73} deals with program verification and the complexity of soundness proofs related to \index{alias} aliases dealing with co-procedures.
Clint's contribution can be considered as first in the area.
A generalisation of Clint's work is the verification, including functionals, where \cite{mehta05} must be mentioned at least towards pointers in higher-order logic, as well as Paulson's technical report \cite{paulson93}, which tries to apply abstraction over procedures as substitution decision for a more effective pointer-based verification rule.

Gotsman \cite{gotsman06} presents an approximation approach towards \index{alias} inter-aliases based on SL.
Although the paper presents runtime characteristics and compares basic operations, the relationship is missing towards Muchnick's and Cooper's results but would be very helpful.\\

\subsubsection*{\textbf{\underline{Garbage Collection}}}

Jones et al. \cite{jones11} present an overcomplete review on \index{GC} GC focusing on parallel algorithms.
In contrast to the first edition, the second edition of \cite{jones11} is a total rewriting and almost exclusively parallel approaches.
Withington \cite{withington91} and Doligez \cite{doligez93} also give a review and can be considered as a genuine subset of the 2nd edition from \cite{jones11}.
Blackburn \cite{blackburn04} summarises the most important recent GCs and a whole chapter in \cite{jones11}.

Appel \cite{appel87} states that there are stack access limitations on code level in practice, especially w.r.t. copying onto and from the stack. 
The von-Neumann architecture and \index{ABI} \textit{ABI} primarily cause them.
Appel further suggests in that case, the use of dynamic memory might be better under the condition no additional overhead is needed or at least only a few resources are used for the \index{stack} stack housekeeping (see \cite{tofte97}, \cite{meyer2-03}).
Appel's multi-threaded contribution \index{GC} \index{copying GC} \cite{jones11} acts only on writing data, which may need seven times more free memory for the heap.

Regardless of its age, Larson's paper \cite{larson77} proposes a division between a fast \index{caching} cache and slower but massive memory is critical still valid up-to-date, not only with dynamic memory but with memory management in general.
Larson's \index{GC} \index{compacting GC} Compacting GC as a function of the number of memory regions to be freed ($R$), the amount of operational data ($A$) and the amount of available fast memory ($H$).
He postulates two optimal \index{proof!strategy} allocation/deallocation strategies: strategy no.1) maximise $R$, if $A\ll H$ does not hold, and strategy no.2) set $R=H$ when $A \ll H$.

Apart from the mentioned restrictions from \cite{appel87}, GC is also restricted due to addressing.
When \index{xor} XOR-calculated edges link memory regions, then the \index{garbage} garbage cannot be localised the traditional way because addresses belonging to objects are not absolute but are relative.
For example, this may affect object fields or pointers to some next element in a list.
Thus, the next and previous elements in a \index{doubly-linked list} doubly-linked list may be calculated, if any, by just one relative \index{offset} offset instead of two absolute addresses (namely pointers).

\cite{sun06} suggests a Generational GC, so allocated regions are assigned an age value that encodes a period that object may remain in memory until being \index{GC} collected.
If an object is rarely used, then after some iterations, that object may be moved into a bigger memory space which may be slower to access.
However, if an object is often referenced, then that object is moved to a faster \index{memory region} memory region -- this is in analogy how a program is managed by OSs and other memory managed facilities \cite{tristan09}, \cite{muchnick07}, \cite{tristan08}, \cite{reif80}.
Runtime performance is approximated close to the optimal solution based on extensive practical heuristics tracked and Sun's experiments.

\cite{hu10} assesses current \textit{SSD}-drives regarding different access operations.
The characteristics of access operations on \textit{HDD}-hardware drive behaviour is much closer to \index{heap} heap access than those of \textit{SSD}-drives – if implemented as such.
Essentially, heap writing is, on average, ten times slower than reading.
GC on SSD-based drives using a greedy-strategy is worst at a waiting rate of 45\%.
The worst significant \index{strategy} performance experience is achieved by writing through all-new blocks.
Writing blocks of consecutive data chunks achieve the best performance with GC turned off.

Calcagno et al. \cite{calcagno03} notice programs with a high abstraction level do not always have an advantage over programs with a low abstraction level altering heap.
For example, it is found that GC often is ineffective in functional PLs. 
Functional programs often cause small but dangerous programs.
On the one side, the danger is due to inaccurate uses of pointers.
On the other side, it is due to side-effects that cripple down the performance.

Waite-Schoor \cite{schorr67} suggest the first and most popular wide-spread approach.
A counter characterises the approach for each heap element.
The counter changes when elements are either freed (holding no further references) or when new elements are allocated.
The allocation is due to resources that the OS manages.
The same goes obviously for deallocation.

About 50 different well-defined GC techniques in total can be counted \cite{jones11}.
The more recent approaches are all focused on parallel execution.
The chances of further research on GC brings a break-through soon are nearly nought and not expected.
The meaning of such research is even more diminished by recent significant hardware improvements and available analysis tools.
After all, as also discussed, GC's overall role due to plenty of memory available changed over the years and is no more such a hot topic as it used to be in the previous 2-3 decades.\\

\subsubsection*{\textbf{\underline{Code Introspection}}}

Introspection means on runtime static program units are available in a dynamic environment so that code may be loaded, modified and added on runtime, e.g. for an object method.
Naturally, introspection makes basic reasoning much more complicated due to the Halting problem and, finally, the program verification in total.
All CI issues are theoretically bound by decidability and depend on PL and assertion feature and what is allowed during runtime.
In any case, a dependency mechanism must be implemented which controls communication between caller and callee.
For instance, in \index{C++} C++, CI is based on RTTI \cite{isocpp14} and in Java \cite{flanagan02}, each object has a unique identifier of a complex struct, which is then read on runtime.

Forman \cite{forman04} defines \textit{CI} in Java as a possibility to read and modify data structures and the loaded program itself during runtime.
Thus, code not only needs to be loaded but also compiled beforehand or interpreter on runtime.
The loading of objects is performed within a particular container providing all life-cycle operations needed.
Statements with reserved meaning perform alteration and access to classes, objects, and components.

Cheon \cite{cheon04} gives a short review of recent CI issues with Java and provides tractable solutions. 
Questions are hardly up-to-date, and almost all questions are resolved by now.
As executing unknown code is not only a security threat, it may also degrade soundness and completeness.
Here only a concise specification could take countermeasures.
It would be unfortunate to have a method with a specification that looks complete but would eventually not run, and the reason for that would be that specification is still missing some essential parts, especially when the specification is already bloated.
Imagine that there was just one case that would not work or work wrong among many working cases.
In order to tackle this problem, Cheon formulates critical circumstances that always must be taken into consideration:
\begin{enumerate}
 \item Loading must be possible at any time, where name clashings must be excluded in a container.
 \item On runtime, object accessors must obey typing.
 \item It is tough to track the specifications of inherited objects altogether. It is hard to understand huge hierarchies --- therefore, keep hierarchies flat.
 \item By its nature, specifications usually refer to objects in stack or heap --- how then to behave exactly w.r.t. introspection and stack windows, since stack unrolling may invalidate and roll back to a previous stable state.
\end{enumerate}

Nevertheless, what exactly is then considered "stable"? That would need to be considered from the user's perspective as well.
Major application fields on introspection are application loadings and web-servers, loading and access facilities for dynamic runtime libraries, e.g. \cite{zakharov15}.

%%%%%%%%%%%%%%%%%%%%%%%%%%%%%%%%%%%%%%%%%%%%%%%%%%%%%%%%%%%%%%%%%%%%%%%%%%%%%%%%%%%%%%%%%%%%%%%%%%%%%%%%%%%%%%%%%%%%%%%%%%%%%%%%

\subsection{Existing Tools}
 
\cite{reynolds02}, \cite{berdine05-2}, \cite{reynolds09} present an excellent introduction to \index{SL} SL, accompanying tools.
SL is a \index{substructural logic} Substructural Logic \cite{restall00}, \cite{restall94}, \cite{dosen93}, which is free of constants, e.g. \index{boolean denotation} boolean values and uses symbols as structural replacements, so as constants.
In SL as structural rules serve according to \cite{restall00} rules of \index{rule!thinning} thinning, \index{rule!contraction} contraction, \index{rule!substitution} substitution.
Constants are \index{dynamic memory} dynamic memory cells (see sec.\ref{chapter:expression}).
In the rules "\textbf{,}" is replaces by $\star$, which separated two non-interleaving and differing from each other heaps, except said differently.
Dynamic memory heaps are inductively defined.
The operator $\rightarrow$, e.g. in the expression $a\rightarrow b$, defines a relation between symbol variables on the left-hand side and some value expression to the right (e.g. an object).
Examples on SL imply the conventions lists, as well as other SL-specifics, hold.
The \index{frame} \textit{frame rule} states that if a subprocedure call does not alter parts of the heap, so there exists some frame denoted by $F$ catching this invariant behaviour, then for the \index{antecedent} antecedent, it is sufficient to prove that the Hoare triple without $F$ also holds.
Let us consider an example of the frame-rule on heaps (definition later):

\begin{center}
\begin{tabular}{c}
\inference[(FRAME)]{\{P\}C\{Q\}}{\{P\star R\}C\{Q\star R\}}
\end{tabular}
\end{center}

Here, $P$ and $Q$ are pre-and postconditions, $R$ denotes the heap frame.
This principle is a basic modularisation rule, and we will refer to it when needed, except when stated differently.
If recursive specifications are allowed together with functionals, then the core frame rule might be violated because the context may be modified by any calling instance (can be implied from \cite{birkedal04} and \cite{pottier08}).
Thus, recursive specifications must be generalised.

\cite{berdine05-2} performs verification examples based on SL for an unbound \index{arithmetic!pointer} arithmetics over \index{offset} offsets to \index{pointer} pointers with increasing \index{array} arrays and recursive procedures.
For example, \texttt{p+x}, where \texttt{p} is a pointer and \texttt{x} an integer variable, is not decidable in general.
It attempts to define the dynamic memory recursively by a closed built-in rule set.
Berdine et al. \cite{berdine05-2} question whether \index{typing} typing is a self-esteemed verification (cf. obs.\ref{obs:NonReducibilityToTypeChecking}).
The answer is that verification is not typing, as discussed earlier.
The paper implies non-decidability of unbound pointer use leads to fuzzy trigger events in \index{GC} GC, even for very trivially looking expressions as offsets.
Furthermore, it leads to an ambiguous rule set selection for verification based on greedy \index{heuristics} heuristics.

Bornat \cite{bornat00} suggests a similar model to \index{SL} SL, which he calls "\textit{remote separation}", and transforms data objects into an array.
Thus, any object becomes a list.
Naming convention, the object identifiers, and field functions slightly differ.
Nevertheless, all object fields are interchangeable between both models.
For a general uncurbed definition of heaps in this model, Bornat postulates that first-order predicates are required.

Hurlin's \cite{hurlin09} main contribution to SL are newly introduced \index{pattern} patterns for multi-threading.
He suggests unblocking functions that improve performance for the sake of a partial blockage of complex cells.

Not all heaps insist on a full specification, which might be shortened by using the anonymous operator "\_".

Parkinson \cite{parkinson05}, \cite{parkinson06}, \cite{parkinson05-2} presents an object-oriented extension of SL \cite{reynolds02} (before Hurlin), using \index{Java} Java as incoming language.
Modularity and inheritance are modelled using the \index{inversion of control} "\textit{inversion of control} and 
\index{abstract predicate family} "\textit{abstract predicate}" (AP) family.
Bornat's \cite{bornat00} object accessors can easily be implemented in terms of preexisting rules from an existing Hoare calculus supporting arrays -- there are no needs to break (accidentally) rules in this case.
For instance, the \index{frame} \index{rule!frame} frame rules on objects remain unchanged since the actual difference may be handled in additional rules critical for objects only.
He mentions, the dependency between predicates stratifies predicate calls.
His paper implies the built-in predicates he suggests cover most of the stack and heap, but they cannot be used to define user-defined predicates.
Assertion predicates, which heavily differ by syntax and semantics from the input language, are not restricted in types.
However, the use of symbols has a series of restrictions due to the non-symbolic use of "\textit{ordinary}" variables decorated with additional constraints to mimic symbolic variables used in class logics.
Factually, predicate symbols are used as \index{variable!local} locals in imperative PLs, and they cannot appear arbitrarily --- this, however, is a hard limitation.
As future work, he pleas to research in more detail: method calls from inherited classes,  \index{static field} static fields, \index{introspection} CI of objects, \index{inner class} inner classes and quantification of those in predicates.

The theorem prover \index{Smallfoot} "\textit{Smallfoot}" \cite{berdine05-2} is the first to implement \index{verification assistant} SL by reference.
It currently works experimental and has a minimal set of \index{predicate!built-in} built-in predicates for heap definitions.
Apart from that, no further predicates are allowed, including user-defined.
The current implementation permanently crashes, hangs, and tons of related errors are present.
Arbitrary uses of some tactics eventually lead to non-termination or other system errors.
After all, it is still a great tool to get familiar with SL, and some examples make its help to understand the theory behind much better.
The verifier system \index{SpaceInvader} "\textit{SpaceInvader}" \cite{calcagno09} is Smallfoot's successor and includes simple \index{abduction} abductive reasoning elements.
\index{jStar} "\textit{jStar}" \cite{distefano08}, \cite{parkinson06} can be considered as an object-oriented extension of Smallfoot.
jStar supports class invariants and some heavily restricted APs.
It turns all program statements into \index{JIMPLE} JIMPLE (IR), which is close to \index{IR} IR from \textit{GCC} (GCC's original IR was first and called \index{GIMPLE} GIMPLE \cite{merrill03}) and is used internally by a corresponding Java-object to perform the verification, which is all Java-implemented.
For industrial use, this makes perfect sense.
\index{Verifast} Verifast \cite{jacobs11} incorporates an input \index{C} C-dialect.
It allows user-defined definitions with nearly the same restrictions as SpaceInvader in \index{abstract predicate} APs.
In case of a proof refutation, the user must manually modify proof commands and provide hints to proceed with the proof.
Hurlin's prover \cite{hurlin09} is experimental also and very closely related to \cite{distefano08}.

\index{Cyclone} Cyclone \cite{grossman02} is a \index{dynamic memory} heap verifier based on \index{RC} RC (see sec.\ref{sect:StackAlignment}).
\index{SATlrE} SATlrE \cite{pavlu10} is a verifier based on \index{shape analysis} Shape-Analysis \cite{sagiv02} and is implemented in Prolog, which is very unusual compared to all others.
"\textit{SATlrE}" differs from the other tools introduced by shapes together with their dependent shapes are not calculated again from scratch, instead only altering edges are.
\index{Y-not} "\textit{Y-not}" \cite{nanevski08-2} is an \index{SMT-solver} SMT-solver implemented fully in \index{OCaml} OCaml.

Finally, the approaches (1) to (4) and program utilities shall be honourably mentioned.
(1) \index{KeY} \index{VDM++} KeY/VDM++ \cite{barnett04}, \cite{weissenbacher01} allows an industrial use of object-oriented specifications and an integration with \index{UML} \textit{UML}, which initially does not support proofs of \index{dynamic memory} dynamic memory.
(2) debugging and profiling tools, most prominently \index{Valgrind} \textit{Valgrind} \cite{valgrind} or \index{ElectricFence} \textit{ElectricFence} \cite{sun06}, \cite{electricfence} \cite{kirsch03}.
(3) integrated into programming packages and libraries, such as within compiler projects, such as \index{LLVM} \textit{LLVM}, \index{SAFECode} \textit{SAFECode} \cite{safecode14}.
(4) transforming programs and assertions into \index{semi-structured data} semi-structured data \cite{badros00} with optional post-mortem transformations \cite{johann03}, \cite{dodds08} (see sec.\ref{sect:HeapModels}).
Sec.\ref{sect:LogicalReasoningAutomation} contains a detailed review of abstract interpretation.
\textit{Valgrind} loads an executable program containing debugging symbols with calling arguments into memory and replaces all dynamic memory access with special instructions.
Thus, each invocation records each invalid access to memory as well as memory leaks.

\textit{GCC} \cite{gcc15} and \textit{LLVM} \cite{llvm15} are compiler frameworks.
Both contain modules on static analysis, such as \cite{lattner03}, \cite{safecode14}.
Khedker's \cite{khedker09} proposition is a static framework based on transitive closures upon graphs to analyse data dependencies.
The CompCert-project \cite{blazy09}, \cite{leroy12} proposes a purely academic framework for soundness analyses based on transformations for the PowerPC architecture only using the general-purpose verifier platform Coq.
Other general-purpose verifiers that are relatively robust include PVS, Proof General, Isabelle and others.
A concise comparison of verifiers, including those mentioned, may be found in \cite{emanuell08}.
Verifiers specialised in dynamic memory may be found in \cite{haberland16-4}.
Appel \cite{appel12} lists recent issues that are resolved successfully.
Articles from the \textit{iX}-journal \cite{kirsch95}, \cite{kirsch94} consider it most comfortable to use dynamic verifiers (all of which mentioned earlier for the problems described in sec.\ref{chapter:DynMemProblems}).
The static analyser \textit{ESC Java} (stands for "\textit{Extended Static Checking for Java}") \cite{flanagan02} introduces a modular specification fully in Java (see \cite{giorgetti10}, \cite{berg01}).
Unfortunately, pointers are absent, and the memory models assume that all data is in a heap.
ESC does not consider CI.
Java-ML \cite{badros00} proposes yet another Java-based specification.

"\textit{Jahob}" \cite{zee08} proposes a prototypical verification of linked lists for Java-programs based on functionals and is implemented in a LISP-dialect.
A heuristic is used, which imposes on rules a strengthened filter.
This project's objective is automation and runtime improvement due to specification abstraction using higher-order functions (cf. with Paulson's approach \cite{paulson93}).

Verifiers that fulfil industrial use include, e.g. KeY \cite{mueller02}.
The popular but mainly academic-only verifier Baby Modula 3 \cite{abadi93}, \cite{cardelli97} may be applied to some more complex examples but suffers from the scalability restrictions described in sec.\ref{sect:TheoryOfObjects}.

%%%%%%%%%%%%%%%%%%%%%%%%%%%%%%

\section{Issues related to Dynamic Memory}
\label{chapter:DynMemProblems}
% (<15p.)

%%%%%%%%%%%%%%%%%%%%%%%%%%%%%%%%%%

This section formulates typical problems with dynamic memory verification by carefully analysing the essentials from the previous section.
Afterwards, a short motivation is given for each problem on why research on that particular problem is actually from a theoretical and practical perspective.
Under \index{dynamic memory} dynamic memory (see fig.\ref{fig:ProcessSectionLoader}), that part of \index{RAM} RAM is meant the OS loads for a \index{process}process \cite{levine99}, \cite{love10}.

A \index{heap} heap is an \index{unorganised memory} unorganised memory portion (see fig.\ref{fig:ExampleHeap}) compared to \index{organised memory} organised memory, e.g. \index{stack} stack.
Organisation implies automatic allocation and deallocation of \index{variable!local} local variables w.r.t. \index{visibility scope} scope visibility (see fig.\ref{fig:ExampleStack}), with an agreed implicit linkage between elements.
In the case of a stack, the \textit{ABI} ABI forces an implicit and strict ordering of element accommodation into the local variables' current stack windows.
For example, declared locals are visible on a \index{conditional branch} conditional branch and maybe overlayed by other local variables with the same name.

While program execution, local variables are not visible outside the block in which they are defined.
The same can be noted with procedures and subprograms.
\index{stack!window} Visible locals, as well as \index{call by value} call-by-value parameters, are \index{stack!pushing} pushed onto the stack when entering a procedure.
All local variables from the current window are \index{stack!popping} popped when leaving a procedure.
Additional deallocations may be required invoked by destructors due to objects' life-cycle \cite{gcc15}.
In fig.\ref{fig:ExampleStack}, a stack window, basically a memory interval, is illustrated with some local $o$, an integer "$1$" in this example.
There is also an \index{array} array (local, not determined) of pointers to "$o$" and the upper element.
All elements are compactly pushed to the stack, so there are no gaps between stack elements.

In contrast to this, structs do not have to be compact, except the keyword "\texttt{packed}" is used.
However, in practice, often, this is not the case.
Hence, \index{performance} optimisation for runtime is the best default behaviour than optimisation for \index{code size} size.

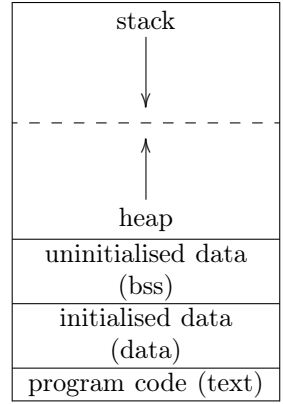
\begin{figure}[h]
\begin{center}
\begin{tabular}{|c|}
 \hline
 \parbox{10cm}{\xymatrix{\textnormal{stack} \ar[d]\\ \txt{}}}\\
 \hdashline
 \parbox{10cm}{\xymatrix{\txt{} \\\textnormal{heap} \ar[u]}}\\
 \hline
 uninitialised data\\
 (\texttt{bss})\\
 \hline
 initialised data\\
 (\texttt{data})\\
 \hline
 program code
 (\texttt{text})\\
 \hline
\end{tabular}
\end{center}
 \caption{A typical memory layout for a Linux process}
 \label{fig:ProcessSectionLoader}
\end{figure}

In contrast to this, \index{heap!cell} \textit{heap cells} of a process may be spread all over the heap memory segment.
The visibility of heap cells does not depend on syntax blocks.
It depends only on statements with a special meaning for allocation and \index{memory deallocation} deallocation.
If after a program statement or at the end of a block statement, a heap is not freed, then by default, the entire heap is freed with the most global \index{destructor} destructor available to the process before passing control back to the OS.
In C++, \texttt{.dtor} does this task when all (global) objects' destructors are called in declaration order.
The OS defines concrete stack and heap layouts.
Both are in \index{virtual memory} virtual memory space.
At this stage, we are not further interested in allocation/deallocation techniques.
Neither \index{stack} stack nor \index{heap} heap does have an apriori fixed boundary on process initialisation.
The boundary is depicted by a dashed line in fig.\ref{fig:ProcessSectionLoader} and is not fixed.
If the stack size may increase vastly, then the available heap size may be balanced.
\index{.bss} BSS (block start symbol) denotes a memory segment, which contains all \index{variable!global} global variables without an initial value assigned (occasionally, other non-local variables may be included depending on the compiler implementation \cite{isocpp14}).
All other \index{variable!local} non-locals and \index{variable!dynamic} non-dynamic variables are assigned to the "\texttt{.data}" segment.

Program code is loaded from the "\texttt{.text}" segment and has a starting point called the starting process memory address.
Absolute addresses substitute all relative addresses representing program variables after process loading.
In practice, the program code must have different stages of program execution, e.g. destructors.

Any cell from any segment from fig.\ref{fig:ProcessSectionLoader} may be \index{addressation} addressed linearly.
It implies that access to memory content may be obtained in consecutive increasing order.
Any memory cell has a successor at address $+\underline{1}$, where "$\underline{1}$" is a symbolic value here denoting one type's actual width to be iterated over once.
However, according to fig.\ref{fig:ProcessSectionLoader}, there is an imaginary floating boundary, which in practice is never reached, which is good for a process' stability.
For example, if some integer is of type \texttt{uint16_t}, then "$\underline{1}$" may become "$2$".

\begin{figure}[h]
  \begin{tabular}{c}
 %\fbox
 {
  \begin{minipage}{3.5cm}
\xymatrix @W=3pc @H=1pc @R=0pc @*[F-] {%
1 \save+<-4pc,1pc>*\hbox{\it o}
\ar[]
\restore \\
{\bullet}
\save*{}
\ar`r[dd]+/r4pc/`[dd][dd]
\restore \\
{\bullet}
\save*{}
\ar`r[d]+/r3pc/`[d]+/d2pc/
`[uu]+/l3pc/`[uu][uu]
\restore \\
456 }
    \caption{Example for a stack}
    \label{fig:ExampleStack}
  \end{minipage}}
  \end{tabular}
  \begin{tabular}{c}
%\fbox
  {
  \begin{minipage}{3.6cm}
  \begin{tabular}{c}
\xymatrix{
 o \ar[r] & *++[Fo]\txt{B} \ar[dd] \ar[dr] & &\\
 *++[Fo]\txt{A} \ar[rr] \ar[ur] & & *++[Fo]\txt{D}\\
 o2 \ar[r] & *++[Fo]\txt{C} \ar[ur] &
}
  \end{tabular}
   \caption{Example for a heap}
   \label{fig:ExampleHeap}
  \end{minipage}}
  \end{tabular}
\end{figure}

Let us pay attention, that \index{exclusive or} "\textit{exclusive-or}" (XOR $\oplus$) allows us to express two \index{pointer} pointers in a doubly-linked list just by one \index{jump field} "\textit{jump field}".
The content of the field stores $a \ xor \ b$, where $a$ denotes the start address of the field, and "$b$" denotes the next field, s.t. $a \oplus (a \oplus b) \equiv b$ holds, as well as $(a \oplus (a \oplus b)) \oplus (a \oplus b) \equiv b \oplus (a \oplus b) \equiv a$ (see \cite{parlante01}, \cite{sinha04}, \cite{haberland16-6}).
Thus, one pointer may be saved, although GC requires a different algorithm.
Assume a doubly-linked list is enqueued, then reachability of all list addresses needs to be calculated first.

The memory models of \index{C} C(++) \cite{gcc15}, \cite{llvm15} and \index{Java} Java \cite{sun06} differ significantly.
There is even a different memory organisation and translation between C and C++, although the latter's syntax may be considered an approximate superset.
So, for instance, PLs compatible with \index{ISO-C++} ISO C(++) do not consider a  \index{GC} GC.
In Java, GC often leads to overhead.
Fortunately, garbage collectors operate in parallel \index{Java} Java threads in a \index{virtual machine} VM environment.
The Java VM, as a compromise solution, integrates \index{generational GC} Generational GC, which, for a practical need, is an optimal compromise solution.
In Java, all variables and data are stored in a \index{heap} heap compared to C or C++ (but not C\#).
C-dialects generally allow \index{weakly typed} weak typing \cite{isocpp14}.
Incompatible types upon varying bit interpretations may be transformed without a cast in \index{C} C, e.g. a \index{byte} byte into a real by a \index{\texttt{union}} \texttt{union}.

\subsection{Motivation}

The technical report \cite{miller90} contains a research study over several decades of traced back and committed errors in numerous open-source and commercial software development projects.
The characteristics found in \cite{miller90} hardly change from year to year.
This reluctance can be confirmed by several newer research studies, e.g. \cite{miller95}, \cite{hind01}, \cite{abiteboul05}.
The results in \cite{miller90} imply commercial \index{Unix} Unix distributives contain over 23\% of incorrect code, wherein open-source platforms, the rate is significantly lower for \index{GNU} GNU-Linux it is 7\%.
The authors note that the rates are stable for years.
Moreover, their conservative estimation is that commercial software has even more errors than open-source.
According to Miller, these are the most frequent errors committed during programming:

\begin{enumerate}
 \item Errors related to pointers and fields.
In commercial applications, the authors emphasise bugs are either not located at all or are located too late, e.g., porting to a different \index{platform} platform.
The former happens during software development and at early testing phases due to a lack of \index{test coverage} test coverage.
The authors are very concerned about the errors are not obvious to detect and delay later projects, including portation projects due to unforeseen issues that need to be fixed first and affect several products at once.
Here invalid memory access is the most often problem -- an indicator of vague and contradictory design decisions, which precise specifications may catch.
 \index{variable!uninitialised} Uninitialised data cells or \index{memory leak} memory leak also occur and lead sooner or later to \index{invalid access} unreachable and lost data, or even program crashes.
 \item \index{invalid access} Either invalid memory access to \index{array} arrays beyond valid boundaries or modified data structures even if a \index{pointer} pointer remains unmodified.
 \item Insufficient checks when working with files.
 Often reading is done not thoroughly enough.
 Implicit assumptions are made that do not match with invalid files.
\end{enumerate}

Dynamic memory verification is problematic because the \index{heap state} heap is described in one way, but the program is described differently and still corresponds \cite{jones75} in an indirect manner.
If interpreting a heap as a graph (see next sections) and a program as an instruction sequence of that graph, then the description of just the same heap may be done in different ways.
While the heap is checked, both sides must correspond.
Hind \cite{hind01} asks the fair question considering the huge efforts already put into over several decades, why \index{alias} AA is not solved yet.
He calls for research on \index{soundness} soundness and \index{performance} performance (cf. fig.\ref{fig:QALadder}) since \index{alias} aliases not yet found may degrade compiled code.
If during compilation it was safe to assume some two local variables must alias, then \index{register allocation} register allocation \cite{muchnick07}, \cite{kennedy02}, \cite{hack11}, \cite{sethi75}, \cite{rideau10} could reuse one CPU register instead of spilling several registers, and it would consequently reduce the need of synchronising two or more registers (see \cite{ssabook15}, \cite{lattner03}, \cite{cytron91}, \cite{muchnick07}).
Even with a conservative estimate, it would mean a great runtime improvement and a significant code simplification.

Hind \cite{hind01} is certain about effective alias localisation requires a thorough analysis of the incoming program.
Landi \cite{landi91} proposes a grouping of \index{alias} AAs.
Analyses are \index{NP-hard problem} NP-hard problems.
He suggests the following groups: (1) aliases that \index{must-alias} "\textit{must-alias}" or \index{may-alias} "\textit{may-alias}" and (2) \index{alias} intra-procedural \index{alias analysis} AAs as well as inter-procedural AAs.
Intra-procedural AAs are effective \cite{hind99}, where interprocedural analyses are ineffective and still require improvement.
Hind \cite{hind01} proposes, in general, to use only very poor visibility scopes of variables when programming.
It can be observed that Hind's proposition is also beneficial from a performance viewpoint since long-living locals may, in general, degrade register allocation and scheduling effectiveness.
Hind also proposes to minimalise visibility towards dynamic allocation and \index{alias} aliases.

In his paper \cite{hu10}, Hu considers the technical correlation between writing cycles and reading from recent \index{flash memory} flash memories.
Hu's approach is of interest here because of the permanent storage characteristics in large, giving theoretical upper bounds in the small.
The OS allocates heap randomly when needed, and several patterns should be considered when designing and reasoning about heap-managed memory regions.
Reading is approximately ten times faster than writing on average, and reading and writing handy memory blocks is dramatically faster than addressing individual blocks due to cache coherency and other architectural facilities.
Hu's main contribution in dynamic memory is that \index{GC} GC is not needed for almost all considered applications.
Significantly, turning GC off extends the lifespan of flash devices remarkably.
Hu equitably questions whether GC nowadays is still an issue for non-critical applications.
Except for being used as \index{permanent memory} permanent memory, flash memory can also be used as RAM disks \index{virtual memory} in \index{embedded device} embedded devices.
It may crucially influence the runtime speedup of programs.

In order to improve the unsoundness of results in several verifiers, Bessey \cite{bessey10} suggests: (1) getting rid of programs with \index{introspection} CI, (2) always stay with the principle: "\textit{simple errors shall be found simply}", (3) accept finally that numerous iterations do not necessarily improve code quality, after all, this is because verification is a qualitative technique and not quantitative.
As already mentioned, point (1) is not essential, after all.
It only changes the time when a particular program code is defined.
On one side, CI increases flexibility, but on the other side, the semantic analysis may be cribbled, for instance, type checking or other phases -- or those that need to be implemented on a new as-you-execute base, so moving it towards a full-scaled interpreter with all its consequences.
Naturally, verification is theoretically bound towards static checks mainly due to the Halting problem.
Due to the restrictions mentioned (see \cite{reus06-2}, \cite{birkedal08}, \cite{schwinghammer09}, \cite{honda05}) hot-code updates (code that is updated during runtime) is not considered any further.
In general, when speaking about erroneous access to dynamic memory, these immediate consequences are implied:

\begin{itemize}
 \item access to \index{invalid access} inaccessible memory
 \item access to uninitialised memory
 \item lack of dynamic memory \index{memory leak} by the need
 \item decrease of runtime characteristics
 \item memory leak (see following).
\end{itemize}

Each of those is discussed next.
Errors may lead to very unpredictable situations.
In the worst case, a program may continue execution with corrupt data.
In the best case, the program may \index{termination} terminate immediately when the first corrupt data shows up with an appropriate hint.
Termination may indeed be considered the best of all errors since a defined behaviour is still better than any undefined.
Termination signals a fault, which complies with soundness and completeness.

Bessey \cite{bessey10} summarises his analyses on successful verifiers in general and explains why they are successful.
Most essentially, Bessey considers in favour of a verifier:

\begin{itemize}
 \item Theorem \index{verification} \index{theorem} proving is an exact science.
  Therefore, regardless of heuristics and tactics applied, the key to success lies in soundness and how quick a valid proof can be derived.
  Quantitative results shall always be considered very critical.
  A verifier that finds fewer known errors than another is undoubtedly weaker.
  Bessey remarks heuristics towards rule selection such as "\textit{take the longest rule first}" in practice regularly succeed, and there is reason to believe this is some meta-pattern.
Although not explained by Bessey, his experience might be explainable by this observation: Long rules are very concrete cases that are often only applied to prepare a shorter but more aggressive rule.
  It is recommended to choose representative examples as suggested by \cite{zakharov15}, such as \index{OS} OSs as input program verified in practice or open-source projects.
 \item Non-trivial proofs shall be standardised.
 If possible, sub-processes shall divide detailed verifications and join them by handing over IRs in a most normalised and most simplified form before continuing the primary verification process.
 Concrete rules must be separated from structural rules.
 The main motto shall be: "\textit{simple errors shall be detected simply}".
 If, however, the motto is not obeyed, then any verification may be considered useless.
 \item Exclude as early and as often as possible \index{language extension} extensions to input languages, which may be excluded w.l.o.g. of expressibility, e.g. the "\textit{hot update}" of code.
 Full specifications of program systems shall be forbidden.
Same counts for code fragments that are not subject to specification (see def.\ref{def:SpecificationLanguage}) at all or can only be specified by tremendous efforts, which would not be in balance with the expected specification size.
 Everything should be simple and plausible.
 Surprises in proofs are always nasty.
\end{itemize}

As mentioned earlier, \index{dynamic memory} heap programming may slow down \cite{larson77}.
There are cases when an algorithm implemented using the heap rather than stack may be faster \cite{appel87}, especially when \index{GC} GC is turned off.
However, this observation may not be lifted in general (see sec.\ref{chapter:intro}) since each algorithm would require a separate investigation.
Efficiency using heap also significantly depends on freed and active lists heavily managed by the \index{OS} OS.
A stack-based decrease in performance is often due to superfluous instructions \index{stack!push} pushed into the stack and \index{stack!pop} popped from the stack as an ABI-guaranteed convention.

This problem worsens with big objects that do not fit in ABI-reserved registers.
Since then, everything must go to the stack and requires additional copying routines that otherwise would not be necessary.
Every copying of data structures may be potentially superfluous, either it is a word or a complex class type, especially when talking about pointers.
If its value entirely passes an object(s), any duplicate is redundant and wasted.
Unfortunately, redundancy detection is intentionally kept too conservative in \index{GCC} GCC and \index{LLVM} LLVM \cite{llvm15} due to some scarce situations where it might be essential to avoid unsound removal.
It can be found that any algorithm implemented with heap may be rewritten solely by using stack, as suggested by Meyer (see \cite{meyer1-03}, \cite{meyer2-03}).
Addresses of arbitrary objects may freely be moved from heap (to stack) by assigning them to a fixed stack window under the condition that every stack element remains uniquely addressable.
This conversion requires a different convention on consecutive type ordering.
The maximum threshold for pushing objects to the stack must be obeyed --- the threshold increases when the heap decreases (see fig.\ref{fig:ProcessSectionLoader}).
Array elements, e.g. uniform objects on the heap, may be pushed to the stack by copying elements consecutively.
It is essential to obey the stack type sequence.
Thus, the only unanswered questions relate to:
(1) Is it indispensable for the sake of an algorithm's readability?
(2) Is this variant always valid or not?
The first question is highly disputable, especially when comparing tree elegancy in heaps (see \cite{parlante01}).
The second question is not always decidable, mainly when limits are known during runtime.
Of course, a sizeable pre-allocated array may fit all needs.
However, a lavish memory allocation is not a reliable approach and not helpful in practice.
Everything that may be pre-calculated statically is only of use in terms of the stack.
Unfortunately, often there are situations when this is unacceptable, e.g. for technical reasons.
It must be stated that both stack and heap are finite.

Let us consider some concrete \index{dynamic memory} dynamic memory problems in a not yet defined \index{C} C-dialect.

\subsection{Issues related to correctness}

\textbf{Example no.1-- Memory leak.}\\
\index{memory leak}
A leak occurs whenever memory is allocated to dynamic memory but never freed.
Often this does not cause an (immediate) crash.
However, this is far from satisfactory.
A program may crash any time after start without further notice.
Often this happens at an unexpected moment.
Whenever the OS finds a shortage of available memory or the \index{OS} OS is not granted access to a memory region requested.
It may also happen when the OS runs out of free memory.
However, this may happen not too often nowadays since the memory available is much more than it used to be one or two decades ago.
In a typical scenario, some flaky program runs a couple of minutes or weeks without any incident whatsoever. 
Then suddenly, it may be caused by some external event registered by the OS, that flaky program crashes without any prior notice.
This kind of problem might be a nightmare to reproduce, and these problems happen in real life.

Often, debugging, if possible at all, cannot localise the origin of the error, but random code places every time executed.
So, the genuine problem may come from somewhere else actually -- what is observed on a crash is just the symptom, but not necessarily the cause for the symptom(s).
Occasionally to diagnose the genuine problem might be quite challenging sometimes.

\begin{figure}[h]
\begin{center}
\begin{minipage}{6cm}
\begin{verbatim}
 MyClass object1 = new MyClass();
 ...
 object1 = new MyClass();
 ...
 (end of program)
\end{verbatim}
\end{minipage}
\end{center}
 \caption{Example for classically instantiated objects}
 \label{ExampleObjectInstantiation1}
\end{figure}

For example, the program's problem from fig.\ref{ExampleObjectInstantiation1} is the content of $object1$ never \index{memory deallocation} utilises after the next assignment.
Assume according to ISO-C++ \cite{isocpp14}, the OS would start to process "\texttt{.dtor}", and all globals call their destructors where applicable, then by lousy coincidence, that phase may crash and could be extraordinarily difficult to debug.
This situation occurs because, on every run, the thresholds responsible for crash or non-crash could change dynamically.
So, the problem is about the analysis of "\textit{danger-zones}", which may leave occupied cells reserved.\\

\textbf{Example no.2 -- Invalid memory access.}\\
From all examples in this section, \index{invalid access} invalid memory access to dynamic memory occurs relatively often.
This problem occurs because some object is uninitialised or has a temporary value that semantically becomes invalid at a specific code position.
It is worth noting that objects are occasionally handled differently among different OS if they are in \index{.bss} "\texttt{.bss}".
For example, OS \index{Windows} \textit{Windows} often does not initialise locals.
That leads to cells having \index{uninitialised value} uninitialised values on runtime.
If this is not handled safely by the algorithm, the behaviour gets erratic.
It may also mean a \index{security breach} security breach to the process.
The reason is incorrectly assigned or uninitialised variables for a given algorithm (see fig.\ref{ExampleObjectInstantiation2}).

\begin{figure}[h]
\begin{center}
\begin{minipage}{6cm}
\begin{verbatim}
...
// object1.ref equals NULL 
// (no assignment, or left NULL initially)
...
value = (object1.ref).attribute1;
\end{verbatim}
\end{minipage}
\end{center}
 \caption{Example for an incorrect assignment of an object pointer}
 \label{ExampleObjectInstantiation2}
\end{figure}

Naturally, debugging should focus on all possible assignments after the variable declaration of \texttt{object1}.
\index{invalid access} Invalid memory access may include access to objects that may have been \index{memory deallocation} freed before access.
Invalid memory access may also refer to an incorrect data representation of address or type, e.g., a processor \index{(CPU) word} word regarding an incorrect \index{typing} type (cast) with \index{pointer} pointers.
E.g. typecasting (down- and upcasts) in terms of the C-language \index{type cast} between \texttt{(int*)} and \texttt{void*}, alternatively \texttt{(char*)} or \texttt{(int8_t*)} may go utterly wrong if specific platforms do not consider those are equivalent.
Moreover, depending on the platform
  $$\texttt{sizeof(struct(int a, char b))}, \ \texttt{sizeof(int)+sizeof(char)}$$
does not have to be identical as described.
Struct sizes may also vary depending on \index{alignment} bits to represent.

Assuming $a$ occupies 2 bytes, and $b$ occupies 1 byte (depending on a platform may differ again if platform-dependent types are being used), the content may become one of the variants in fig.\ref{ExamplesBitmask}.

\begin{figure}[h]
\begin{center}
\scalebox{0.6}{
\begin{tabular}{l}
\begin{tabular}{|l|l||l|l|}
\hline
$a_0 a_1 a_2 a_3 \ a_4 a_5 a_6 a_7$ & $a_8 a_9 a_{10} a_{11} \ a_{12} a_{13} a_{14} a_{15}$ & $0 \ 0 \ 0 \ 0 \ \ \ 0 \ 0 \ 0 \ 0$ & $b_0 b_1 b_2 b_3 \ b_4 b_5 b_6 b_7$\\
$a_0 a_1 a_2 a_3 \ a_4 a_5 a_6 a_7$ & $a_8 a_9 a_{10} a_{11} \ a_{12} a_{13} a_{14} a_{15}$ & $b_0 b_1 b_2 b_3 \ b_4 b_5 b_6 b_7$ & $0 \ 0 \ 0 \ 0 \ \ \ 0 \ 0 \ 0 \ 0$\\
\hline
\end{tabular}\\\\

\begin{tabular}{|l|l||l|l|}
\hline
$0 \ 0 \ 0 \ 0 \ \ \ 0 \ 0 \ 0 \ 0$ & $b_0 b_1 b_2 b_3 \ b_4 b_5 b_6 b_7$ & $a_0 a_1 a_2 a_3 \ a_4 a_5 a_6 a_7$ & $a_8 a_9 a_{10} a_{11} \ a_{12} a_{13} a_{14} a_{15}$\\
$b_0 b_1 b_2 b_3 \ b_4 b_5 b_6 b_7$ & $0 \ 0 \ 0 \ 0 \ \ \ 0 \ 0 \ 0 \ 0$ & $a_0 a_1 a_2 a_3 \ a_4 a_5 a_6 a_7$ & $a_8 a_9 a_{10} a_{11} \ a_{12} a_{13} a_{14} a_{15}$\\
\hline
\end{tabular}\\\\

\begin{tabular}{|l|l||l|l|}
\hline
$a_{15} a_{14} a_{13} a_{12} \ a_{11} a_{10} a_9 a_8$ & $a_7 a_6 a_5 a_4 \ a_3 a_2 a_1 a_0$ & $0 \ 0 \ 0 \ 0 \ \ \ 0 \ 0 \ 0 \ 0$ & $b_7 b_6 b_5 b_4 \ b_3 b_2 b_1 b_0$\\
$a_{15} a_{14} a_{13} a_{12} \ a_{11} a_{10} a_9 a_8$ & $a_7 a_6 a_5 a_4 \ a_3 a_2 a_1 a_0$ & $b_7 b_6 b_5 b_4 \ b_3 b_2 b_1 b_0$ & $0 \ 0 \ 0 \ 0 \ \ \ 0 \ 0 \ 0 \ 0$\\
\hline
\end{tabular}
\end{tabular}}
\end{center}
 \caption{Examples of bit/byte ordering}
 \label{ExamplesBitmask}
\end{figure}

That is why making arbitrary assumptions on a specific platform is wrong.
When \index{Intel} Intel is assumed, but \index{PowerPC} Pow\-er\-PC or \index{ARM} ARM is found with a 64 or 32bit processor word, assumptions will eventually fail.
If a field is not initialised for some reason, then a harmless looking program such as in fig.\ref{ExampleUninitialisedFails} may not \index{termination} terminate or terminate not correctly.

\begin{figure}[h]
\begin{center}
\begin{minipage}{6cm}
\begin{verbatim}
object1.next = object;
...
root=object1;
while (root.next!=NULL){
  printf("%d", object.data);
  root=root.next;
}
\end{verbatim}
\end{minipage}
\end{center}
 \caption{Example of an incorrectly linked pointer and non-terminating program}
 \label{ExampleUninitialisedFails}
\end{figure}

\textbf{Example no.3 -- Dangling pointers and aliases.}\\
\index{dangling pointer} A dangling pointer \cite{afek07} appears when initially several pointers (possibly \index{alias} aliases) share an object, and some pointer manipulations lead to one of the \index{heap!connected} pointers is mistakenly separated from the remaining.
Although intuitively clear, this may be \index{alias} challenging and imprecise in practice.
An analysis requires not only to check the procedure where a pointer is declared but also to check the whole \textit{fan-in} and \textit{fan-out}.
Containing objects of dangling pointers shall be disposed of since those are garbage by definition.

Pointers that become dangling may be turned into non-dangling during program execution and vice versa.
However, heap elements becoming garbage are lost forever and must not become addressable again.
Dangling pointers may be the result of the situations described.
Furthermore, a pointer may become dangling by different memory cell interpretation using \index{\texttt{union}} "\texttt{union}", as described earlier.\\

\textbf{Example no.4 -- Side-effects.}\\
Instead of copying unneeded data structures, often pointers may be more effective.
However, this approach has the risk unaffected data regions may (accidentally) be corrupted.
This problem is a generalisation of example 2 and 3: Pointers do not alter on runtime, but the content suddenly gets changed.
It can be generalised: modification suddenly alters other variables.
A modification of a heap shall not alter other heaps.

\subsection{Issues related to completeness}

\textbf{Example no.5 -- Problems with expressibility.}\\
Some issues with \index{expressibility} expressibility have already been presented in this section.
Essential expressibility issues are:
 (1) Can all valid heaps be \index{specification} specified under the condition that all valid heaps are syntactical correct for a given rule set?
 (2) Do invalid heaps always verify false?
 (3) Which predicates may be used?
 (4) Which constraints must predicates obey?
 (5) Which dependencies between heaps exist, and how adequate are formulae it?
 (6) Which restrictions are there w.r.t. symbols in heap definitions?
 (7) How to define single or multiple heaps best?
 (8) How to express dependency between \index{alias} aliases?
 (9) Are existing heap definitions ambiguous, and if so, how to exclude ambiguity effectively?
 (10) Which \index{abstraction} abstraction level is required for a precise and handy heap definition, particularly for verification?
 (11) How much further can a heap be abstracted, s.t. is well understood by specifier, verifier and developer?
 (12) Can heap formulae be transformed, s.t. there is no need to re-specify on local program code updates?
 (13) Can more familiar concepts from programming be \index{transducer} reused, s.t. specification and programming become more familiar to a programmer?

\textbf{Example no.6 -- Issues related to a complete representation.}\\
According to fig.\ref{fig:QALadder} in \cite{suzuki82}, Suzuki proposes \index{safe pointer operations} "\textit{safe}" pointer operations.
As described earlier, its use must be handled with care since the \index{pointer rotation} clockwise \textit{rotation} of one data structure may quickly lead to corruption or \index{memory deallocation} disposal of another data structure, for instance.
Apart from that, implicit assumptions are too often not immediately apparent.
Suzuki's \index{Suzuki's memory model} approach and similar approaches introduced almost always suffer from incomplete rule sets and incomplete heap definitions. 
Too often, the situation is this: Given a rule set of 25 rules.
Questions: (1) Are those 25 rules enough, or are any further needed?
Do even existing rules need to be shrunk?
(2) What is the overall policy with overlapping or even \index{duplicate} redundant rules?
As any heap state depends on a given \index{program statement} program statement, heap subsets need to be formulated and then compared with the expected.
(3) Does it make sense to limit expressibility, s.t. termination is decidable but with acceptable limitations?
(4) Is there a possibility to specify and verify only parts of a heap?
(5) May a simple(r) heap model be proposed, s.t. a proof using a simple data structure and incoming program shall be simple?
For instance, \index{list inversion} list inversion \index{linear list} should be simple (cf.\cite{reynolds09}).

Preferably, a verification should not insist on a complete specification since a full specification and verification is not economic in practice and is in harsh disbalance to the benefit ever gained on average.
It is often sufficient to specify only focus on procedures and class objects one may be interested in the practice.
That would mean leaving unrelated program fragments unspecified.
An engineering requirement is to add additional \index{predicate!auxilary} constraints at random code places \index{error localisation} wherever the engineer wants to perform a more detailed analysis (see the algo.\ref{algo:AlgorithmProblemReduction}).
A developer may also be interested in adding checks at random places in code in addition to pred- and postconditions.\\

\textbf{Example no.7 -- Issues related to the level of automation.}\\
According to the \index{quality ladder} quality ladder from fig.\ref{fig:QALadder}, these problems belong to the second category.
Main automation problems from earlier sections are related to defining dynamic memory rules and separating them from other non-related rules.
If rules related only to heap are put into a separate \index{formal theory} formal theory, then the whole verification may be better logically structured, and the overhead might be diminished and more reliable towards future modifications.
Such a separation attempt would also diminish the amount and complexity of rules.
There is a need in comparing specifications with heap states.
Unfortunately, currently, all existing approaches make comparison manually (cf. following sections).
It is hard to predict when a \index{theorem} theorem is used and when to be applied with the needed symbols.
The same problem exists while transforming one heap state into another.
A locally optimal proof may still lead to non-decision.

While transforming a heap \index{deduction} deductively, confluency shall be improved (cf. fig.\ref{fig:CRTonHoareTriples}) -- this can converge by abduction.
The necessary theoretic restriction might be expressibility.
Terms may be partially defined during static analysis, or they are non-decidable in general.
Offset limitations in arithmetic expressions leads to expressibility limitation, but also a higher automation level.
So, the question arises how generous are algebraic heaps then?
How proper is strict typing for an input PL?

\subsection{Issues related to optimisation}

\textbf{Example no.8 -- Issues related to runtime performance.}\\
Runtime performance problems directly compete with \index{soundness} soundness (see fig.\ref{fig:QALadder}).
Assertions about \index{must-alias} "\textit{pointers alias}" or \index{can-alias} "\textit{pointer do not alias}" are more important than \index{may-alias} "\textit{may alias}", but those are more difficult to find.
The second and third assertion has a more substantial effect on \index{code generation} generated code.
The faster the generated code is, the more effective it becomes because a lack of need to store \index{register} CPU registers to stack means improvement.
Unfortunately, \index{data dependency} data dependency analysis over pointers is much more complicated than with locals because \index{pointer content} a pointer's $p$ content may alter at potentially any place in code.
Therefore, \index{alias analysis} AA is hard.
Often it can be observed from \index{framework} frameworks with pointers that some function on pointers works only for a very \index{generalised assertion} restricted domain since pointer have a very high expressibility, and it is challenging to guarantee the most general assertion upon pointers.

Remarkably, if a data structure is used only once and there are no duplicates, far less effort is needed for immediate modifications than in cases with duplicates.
In implementations, two main approaches can be observed: either complex and straightforward variables are passed manually (though including auxiliary remarks, e.g. in C by using the keyword \index{\texttt{register}} \texttt{register} \cite{gcc15}), or in exceptional cases, \index{alias analysis} AA is done (mainly intra-procedural only \cite{llvm15}).
This case is not excluded from \index{C} C \cite{gcc15} and can be achieved using \index{const} \texttt{const}.
Let us consider a small motivating example: \index{list inversion} list reverses with alteration \cite{reynolds09}, \cite{parlante01}.
Here, the initial list vanishes.
Depending on the concrete context, this might be desired behaviour.
\index{Reynolds' algorithm} \index{linear list} Reynolds' algorithm traverses the list only once in the forward direction.
There is no need for next or previous pointers.

Thus, a list does not need to be copied.
All operations are performed on the given list, which heavily speeds up.
If only the data structure changed and the original is safely not needed afterwards, then the old list may be dropped, and the algorithm's overall behaviour would speed up.
GCC and LLVM \cite{gcc15}, \cite{llvm15} do not perform such a test yet.
May the \index{ABI} ABI be modified, s.t. non-referenced objects are popped from the stack by the \index{caller} caller \cite{isocpp14}, \cite{gcc15} unconditionally?
May existing objects in \index{dynamic memory} dynamic memory not just be used immediately, mostly when objects are known for sure they are not changed on calls?
It must be mentioned that the proposition made in \cite{sinha04} is non-standardised.

However, the average efforts needed on linear lists seem to be strongly reduced, and the GC changed drastically.

When the number of iterations \index{limit} is known apriori to entering a cycle after static analysis (see \cite{appel87}), then the code portion may be optimised responsible for memory allocation and alignment.
All problems related to optimal GC, beginning with \index{White-Schorr's algorithm} White-Schorr's algorithm \cite{schorr67} and continuing until today, are more than extensively researched \cite{jones11}.
Micro-operations related to dynamic memory are done by the OS, which traces memory resources, including stack and stack for each process.\\

\textbf{Example no.9 -- Issues related to program security and access protection.}\\
In analogy to \index{stack!overflow} \textit{stack overflow} \cite{Kim15}, when the stack is flooded with undesired dangerous data in order to compromise a program by moving the actual instruction pointer beyond the actual stack window on procedure call or return from the procedure --- there is also an attempt of pushing a return address to malicious code during GC the heap \cite{kaempf06}, \cite{afek07}.
Obviously, in contrast to stack, the heap does not contain program locations nor addresses of (modified) program code.
Consecutively, the attack based on malicious code injection may not work on-the-fly as with the stack.
Stack overflow may lead to security risks for an individual program and comprise the whole OS.

When at stake, the problem worsens even more when remote servers and services and specifications are interfaces when access via device driver is acquired (cf.\cite{corbet05}).
Since several instances may load drivers simultaneously, the risks of not being able to restore due to a heap issue are susceptible and delicate.
In the worst-case, non-reliability may crash the core OS as it is an existential thread to monolithically designed OSs \index{GNU Linux} \index{linux} GNU-Linux.

\section{Expressibility of Heap Formulae}
\label{chapter:expression}

This section researches heap models, namely heap graph and heap predicates, which will later be reasoned in logical programming.
First, the differences between heap and stack are elaborated, and crucial properties are analysed.
Then, occurring phenomena are observed based on the fundamental problems from the previous section.
Finally, based on a pre-existing model, an (extended) heap term is introduced, and properties are investigated.
The heap term definitions are implemented in the Prolog-dialect (see sec.\ref{sect:Implementation}).\\

Assume some \index{paradigm} \index{paradigm!imperative} imperative program is given, which has the \index{CFG} CFG \cite{khedker09} as shown in fig.\ref{fig:ControlFlowGraph}.
Let us explain another example from \cite{cytron91} the difference between \index{variable!automated} automatically allocated and \index{variable!dynamic} dynamically allocated variables.
Automated variables are automatically allocated and deallocated.
These are bound to a \index{stack!window} stack window.
The stack window contains all \index{variable!local} locals and parameters passed in before and out after a procedure call.
These stack-local variables shall not be mixed up with "\textit{fan-in}" or "\textit{fan-out}" parameters.

\begin{figure}[h]
\begin{center}
\begin{tabular}{c}
 \xymatrix@C=3em@R=1em{
  *+[F]\txt{Entry} \ar[rr]  &   & *+[F-:<8pt>]\txt{1} \ar[d]\\
                                & & *+[F-:<8pt>]\txt{2} \ar[dr] \ar[dd] \\
                                & &  & *+[F-:<8pt>]\txt{3} \ar[d] \ar@/^1pc/[dr] &\\
                                & & *+[F-:<8pt>]\txt{7} \ar[dd]  & *+[F-:<8pt>]\txt{5} \ar[d] & *+[F-:<8pt>]\txt{4} \ar@/^1pc/[dl]\\
                                & & & *+[F-:<8pt>]\txt{6} \ar@/^1pc/[dl] &\\
                                & & *+[F-:<8pt>]\txt{8} \ar[d] & \\
                                & & *+[F-:<8pt>]\txt{12} \ar[rr] \ar@/^5pc/[uuuuu] & & *+[F]\txt{Exit}
 }
\end{tabular}
\end{center}
 \caption{Example CFG.} \index{CFG}
 \label{fig:ControlFlowGraph}
\end{figure}
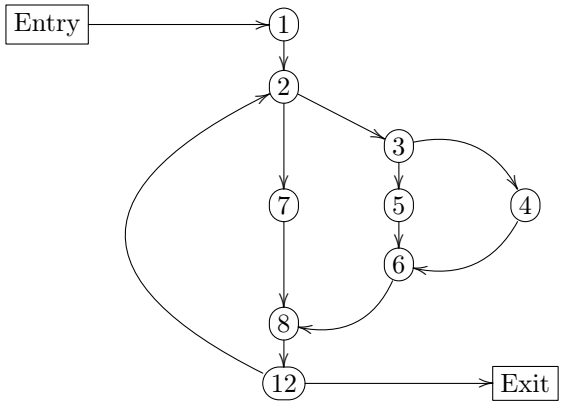

According to fig.\ref{fig:ControlFlowGraph}, \index{visibility scope} visibility scopes can be established.
A scope interval of the whole procedure [\textit{Entry,Exit}] denotes a \index{lattice} \textit{simple directed graph} in which passed variables are visible to all vertices between Entry and the ending vertex Exit.

A \index{block} "\textit{block}" is defined as a joining sequence of non-branching program statements.
Next, in order to avoid possible naming clashes with existing definitions, the \index{SSA-form} \textit{SSA-form} as defined by \cite{cytron91}, \cite{hack11}, \cite{ssabook15} shall be used.
A branching might be any \index{conditional branch} branching or unconditional jump.
However, immediate, unconditional jumps shall not be allowed for the sake of expressibility considerations in the following.
Assume some \index{block} blocks have \index{variable!local} locals represented in \index{SSA-form} SSA-form as in fig.\ref{ExampleSSABlocks}.

\begin{figure}[h]
\begin{center}
\begin{tabular}[t]{l}
block 2: \texttt{b$_0$=4; a$_0$=b$_0$+c; d$_0$=a$_0$-b$_0$;}\\
block 7: \texttt{b$_1$=a$_0$-c;}\\
block 3: \texttt{c$_0$=b$_0$+c;}\\
block 5: \texttt{c$_1$=a$_0$*b$_0$; f(a$_0$);}\\
block 4: \texttt{a$_1$=a$_0$+b$_0$;}
\end{tabular}
\end{center}
 \caption{Example of SSA-placements by basic blocks}
 \label{ExampleSSABlocks}
\end{figure}

Let us assume further, all other blocks are either artificial, empty or do not contain relevant just introduced variables.
Block 5: $f$ is a block that accepts one variable.
At first glance, there is nothing unusual.
However, \index{call by value} if inside $f$ access to a parameter's content is performed, which in C can be done using the \textit{de-reference operator} "\texttt{\&}", variables may be altered even beyond $f$'s boundaries.
Fortunately, this scenario can be caught with the earlier analysis of $f$'s incoming and outgoing variables.
Even if the scenario described may not frequently occur in practice, this scenario is the reason why a vast majority of code optimisation remains slow for the sake of soundness in the most general case.
If in block six $a$ was redefined, then $a_2= \phi(a_1,a_0)$.
Function $\phi$ is auxiliary.
It means a union of different definitions that can be made over a variable.
\index{$\phi$-function} $\phi$ is an implicit function (hence its origin: \textit{\underline{ph}oney}).
The whole concept of implicit functions seems intriguing on the first look, as we may see later when it comes to heap definitions.
What the phoney functions do is express \index{data dependency} data dependencies of variables compactly.
It may be applied to arbitrary automatic data in \index{paradigm} \index{paradigm!imperative} imperative programs.
The data dependency structure, in general, cannot be a tree.
\index{$\phi$-function} $\phi$-functions forbid this when there are dependencies, which refer to some previous block, which might be outside a loop.

Regardless, we are most interested in existing algorithms on scope visibility using $\phi$-functions \cite{ssabook15} of locals.
We rely on the remark that upper \index{limit} bounds exist regarding visibility scopes.
These bound may be defined recursively over the stack.
Bounds are calculated using the CFG and $\phi$-functions of each local (dominator algorithms may be found in \cite{ssabook15}, and the classic survey is \cite{cytron91}).
If variables are redefined in two different branches, then the content after a join may differ and its index increases.

When attempting to apply a given \index{SSA-form} SSA-form to dynamic variables, it must fail due to the following reasons: (1) there are program statements that \index{memory allocation} allocate and \index{memory deallocation} deallocate dynamic memory explicitly.
These statements must not necessarily be in the bounds of a procedure.
Heap cells may unconditionally exist after some procedure is called and even after the stacked pointer is freed.
So, the definition place in code (same goes for redefinitions and deallocation), and its place(s) of use differ from those for automatic variables.
So, the two places often do not match.
Pointer's content may alter even when \index{pointer} pointers are not referenced immediately.
(2) Size, frequency and the context of heap memory allocation are, in general, not always defined (see fig.\ref{fig:ProcessSectionLoader}).

In analogy to CFG, not only the content of a variable may be annotated in graph vertices \cite{floyd67}, but also \index{predicate!assertion} assertions about dynamic memory may be annotated.
An example may be \index{framework} Khedker's framework \cite{khedker09} which analyses for each pointer using a very conservative algorithm -- assuming a \textit{may-alias} for all not found pointers based on Floyd's transitive closure algorithm.
Thus, after every program statement, not only is the considered pointer re-calculated, but potentially everything related to it.
Khecker's algorithm is based on a long \index{bitfield} bitfield. 
Some optimisations may be done towards the bitfield.
However, it is not done because the estimated improvement would not necessarily be significant and, most importantly, beneficial and a base to many flow analysis techniques described.
Mandatory checks of all dependencies can hardly be simplified further due to transitive comparisons and connectivity checks.
The comparison of \index{alias} aliases chosen is based on Horwitz' \cite{horwitz89} and Muchnick's \cite{muchnick07} algorithms.

\begin{observation}[Organised Memory and Fresh Context]
\label{obs:OrganizedMemoryFreshContext}
The \index{stack} stack is organised memory.
The order of elements in it is well-defined.
Every time a procedure is called allocation and deallocation are performed (may also be the case for compound statements).
Every time in a new context, a new allocation means containing all elements from the surrounding procedure and is initially not linked.
\end{observation}

Next, an analogy might be shown in which every state of \index{dynamic memory} dynamic memory may be recorded as a single state.
Any state transitions may exist, but there must always be an entry and an exit, even when the \index{heap} heap is empty.
On exit auxiliary \index{heap!state} states may be modelled, s.t. all finite states are merged with the final exit state.

\begin{observation}[Unorganised Memory and Unique Context]
\label{obs:UnorganizedMemoryUniqueContext}
The \index{heap} heap (as the unit of dynamic memory) is unorganised.
On one side, elements may be in arbitrary order and allocated anyhow.
On the other side, \index{heap graph} the heap graph (yet to be introduced in sec.\ref{sect:HeapGraph}) alters after each program statement, and a new context linked to the current state is absent.
There is only one unique context. 
\end{observation}

For example, \index{paradigm!functional} functional PLs assume \index{Principle Data Independence} the \textit{principle of independent data} as their central concept.
It guarantees a fresh context permanently.
It contains initial parameters without dependencies.
Afterwards, the result is assigned and returned to the caller.
\index{continuation} A continuation (see \cite{joel70}, \cite{wiki23-02-2010}, \cite{dargaye07}, \cite{thompson97}) is a calculation state which is passed to another instance by need.
The transiting state is not interrupted since the context does not change.
Equal states mean the whole state of calculation, which contains stack and heap (see fig.\ref{fig:ProcessSectionLoader}).
In the case of \index{predicate!higher-order} higher-order functions, a continuation can nicely be described by \index{denotational semantics} denotational semantics.
Denotational semantics, in general, describes a program as a function, so it is assumed the correspondence to any environment is superfluous or not existent.
Continuations do not change, but process stack windows do.
Jo\"{e}l \cite{joel70} considers \index{continuation} continuations in functional-logical contexts by the example of \index{LISP} common LISP.
Introductory remarks are: (1) continuations lead to copying, pushing and popping of memory regions, transition addresses and returns to the stack, (2) readability and algorithms' modelling may significantly be improved.
Remark (1) causes concerns about performance (cf.\cite{appel87}).
Remark (2) has priority and is in contrast to the first remark.
Therefore, a meaningful compromise must be found in practice.

\begin{observation}[Phenomenon of Remote Manipulation]
Due to the separability of memory cells and due to dynamic memory in general, the phenomenon of \index{remote manipulation} "\textit{remote manipulation}" may be encountered, as was described earlier.
The phenomenon is that the evaluation of some \index{expression} (sub-)expression in some \index{program statement} program statement may "\textit{magically}" alter some unrelated content being addressed by some \index{pointer} pointer which is not part of that expression and hence may come as a total surprise to the caller of that procedure.
This phenomenon is neither locally bound to a specific BB only nor a procedure but can modify variables even beyond a procedure's boundaries.
 \label{observation:RemoteAlternation}
\end{observation}

\begin{observation}[Variable Visibility Scope]
\label{obs:VariablesScope}

A pointer's visibility scope, in contrast to \index{variable!static} static and \index{variable!logical} logical variables, may from the place of \index{memory allocation} allocation till \index{memory deallocation} deallocation be interrupted one or more times, even if the pointer does not change (cf. fig.\ref{fig:VisibilityScopes}, where the axis denotes the line number for some given program. The solid line indicates stacked variables may overlay visibility) because of progressing stack alteration over BBs.
It may happen accidentally, intentionally, or even temporary whilst the invocation of some program statements only.
\end{observation}

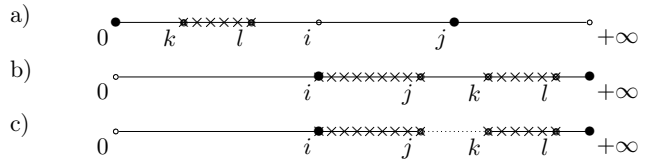
\begin{figure}[h]
\begin{center}
\scalebox{0.9}{
\begin{tabular}{lcr}
a) & & \begin{xy}
 (0,0) ="minZero" *\cir<1pt>{} *+!UR{0}*+{\bullet},
 (70,0) ="maxInf" *\cir<1pt>{} *+!UL{+\infty},
 (10,0) ="k" *\cir<1pt>{} *+!UR{k},
 (20,0) ="l" *\cir<1pt>{} *+!UR{l},
 (30,0) ="i" *\cir<1pt>{} *+!UR{i},
 (50,0) ="j" *\cir<1pt>{} *+!UR{j}*+{\bullet},
 "minZero";"maxInf"**{} -/1.5pt/;+/1.5pt/ **@{-},
 {\ar@{{o}{x}{o}} "k"*{};"l"*{}}
\end{xy}\\
b) && \begin{xy}
 (0,0) ="minZero" *\cir<1pt>{} *+!UR{0},
 (70,0) ="maxInf" *\cir<1pt>{} *+!UL{+\infty}*+{\bullet},
 (55,0) ="k" *\cir<1pt>{} *+!UR{k},
 (65,0) ="l" *\cir<1pt>{} *+!UR{l},
 (30,0) ="i" *\cir<1pt>{} *+!UR{i}*+{\bullet},
 (45,0) ="j" *\cir<1pt>{} *+!UR{j},
 "minZero";"maxInf"**{} -/1.5pt/;+/1.5pt/ **@{-},
 {\ar@{{o}{x}{o}} "k"*{};"l"*{}},
 {\ar@{{o}{x}{o}} "i"*{};"j"*{}}
\end{xy}\\
c) && \begin{xy}
 (0,0) ="minZero" *\cir<1pt>{} *+!UR{0},
 (70,0) ="maxInf" *\cir<1pt>{} *+!UL{+\infty}*+{\bullet},
 (55,0) ="k" *\cir<1pt>{} *+!UR{k},
 (65,0) ="l" *\cir<1pt>{} *+!UR{l},
 (30,0) ="i" *\cir<1pt>{} *+!UR{i}*+{\bullet},
 (45,0) ="j" *\cir<1pt>{} *+!UR{j},
 "minZero";"j"**{} -/1.5pt/;+/1.5pt/ **@{-},
 "k";"maxInf"**{} -/1.5pt/;+/1.5pt/ **@{-},
 "j";"k"**{} -/1.5pt/;+/1.5pt/ **@{.},
 {\ar@{{o}{x}{o}} "k"*{};"l"*{}},
 {\ar@{{o}{x}{o}} "i"*{};"j"*{}}
\end{xy}
\end{tabular}}
\end{center}
 \caption{Visibility of locals in a),b) and dynamically allocated variable in c)}
 \label{fig:VisibilityScopes}
\end{figure}

In the example in fig.\ref{fig:ExampleHeap}, if \texttt{free(o.B)} frees object $C$ from the \index{dynamic memory} heap, access to $C$ cannot be provided via $o2$.
If $o.B$ redefines a freshly allocated object, then $C$ is not notified about it, except if the new object is assigned the same address the previous old object used to be.
For the sake of formalism, it can be agreed upon that \index{dangling pointer} dangling pointers are assigned "\texttt{nil}" and a universal model, even if, in reality, the pointer contains the old address.
So, all other pointers need to be analysed too.

Other pointers may be \index{alias} aliases, which are vertices in a corresponding \index{heap graph} heap graph.
Those may be $o$ or $o.B$ (see issues from sec.\ref{chapter:DynMemProblems}).
Heap cell \index{heap!interpretation} interpretation depends on a pointer's type, which is checked accordingly during \index{semantic analysis} semantic analysis.
The type does not change anything except for subclassing (see sec.\ref{chapter:logical}).

\subsection{Heap Graph}
\label{sect:HeapGraph}

Now the fundamental aspects of dynamic memory are discussed in detail (including obs.\ref{obs:OrganizedMemoryFreshContext}, obs.\ref{obs:UnorganizedMemoryUniqueContext}), it is time to think about a heap graph representation.
Early but still essential to agree upon detailed \index{dynamic memory} heap models are introduced in sec.\ref{chapter:stricter} and sec.\ref{chapter:APs}.
The essential pair of heap allocating operations are "\texttt{malloc}" and "\texttt{free}".
Essentially, compactness and flexibility determine the success of a heap graph model.\\

\textbf{\underline{Heap graph as regular expression}.}
\index{heap graph} \index{regular expression} In the beginning, let us consider the simple graph $\cal A$$_1$, which obeys the constraints discussed earlier.
This graph could be described by a regular language and can be recognised by a simple finite automaton, as shown in fig.\ref{ExampleNFA1}, and was suggested by the author, e.g. by \cite{haberland08-2}, \cite{haberland07-2} in the context of semi-structured data as Prolog terms.
Regular expressions can also be justified because they describe a heap graph consisting of states and explicit state transitions of some calculation.

\begin{figure}[h]
\begin{center}
\begin{tabular}{c}
\xymatrix{
  \txt{s} \ar[r]       & *++[o][F-]\txt{$q_0$} \ar@(l,d)[]^b \ar[r]^a  & *++[o][F-]\txt{$q_1$} \ar@(l,d)[]^a \ar[r]^b  & *++[o][F-]\txt{$q_2$} \ar@/^1pc/[l]^a  \ar[r]^b  &*++[o][F=]{\{\forall q_F\}}\\
   *\txt{} & *\txt{} & *\txt{} & *\txt{} & *\txt{}\\
   *\txt{} & *\txt{} &  & *++[o][F-]\txt{D} \ar[uur] & *++[o][F-]\txt{E} \ar[uu]
}
\end{tabular}
\end{center}
 \caption{Example finite automaton $\cal A$$_1$}
 \label{ExampleNFA1}
\end{figure}
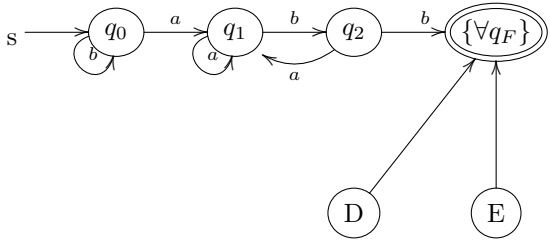

\index{Adren's lemma} Arden's lemma \cite{davis94} turns the determined \index{automaton!finite} finite automaton into this regular expression: $b^{*}(a^{+}b)^{+}b$.
So, what do this heap graph and the corresponding expression denote?
The graph's vertices represent a set of non-interleaving heap cells.
The graph's united final state is $\{q_F | \forall q_F \in F \subseteq Q \}$.
Naturally, all finite states may be merged into one finite state.
For example, in the graph, vertices $D$, $E \in Q$ denote some states that merge to $\{\forall q_F\}$.
The graph's \index{graph edge} edges are directed links encoded as addresses of a processor word's width in an edge's source.
So, the first question that arises: How are heap edge labellings represented?
These may be \index{pointer} pointers or \index{object field} object fields, so locations in both cases.

Both cases are not satisfactory because pointers always have to be dealt with individually from all other memory \index{memory cell} cells since those are in the \index{stack} stack.
Second, labelling all edges must be distinguishing.
Assume this is the case, then \index{regular expression} regular expressions cannot be described by the proposed compact representation (see following).
However, this means the initial very compact notation suffers dramatically.
Third, it is not yet clear what exactly are "\textit{initial}" and "\textit{final}" states?
At least, a pointer considered at one moment to be initial may be final at another moment or neither nor.
However, it could serve if certain strict constraints were invariant, e.g. by default, the last function parameter always denotes a final state.
An initial state may always be denoted as any stack variable transition - this is always valid to do.
Final states are more challenging to describe, e.g. whether labelling is the last one for some given calculation or not.
If $q_F$ were dropped, then the expression would be undefined.
A new state might be introduced which accepts those transitions that signal a final calculation of the data structure.

Now assume, all mentioned problems are resolved by now acceptably, and we continue with the question on \index{dynamic memory} heap alteration by using program statements after finding all (dis-)advantages of \index{regular expression} regular expressions.
If that notation were compact, then it would be required to check whether this is a stable solution regarding alteration.
If a tiny modification was made, then only one edge should change, and the overall regular expression should not change much -- if so, then this notation would be tractable indeed.
If this notation does not meet the requirements mentioned, then it shall be reasoned, and perhaps another model should be chosen.

For example, the last graph is added a new edge $b$.
The new graph becomes $\cal A$$_2$ and is illustrated in fig.\ref{ExampleNFA2}.
This graph is semantically equivalent to the regular expression $b^{*}a^{+}b((a^{+}b)^{*}+bba^{*}b(a^{+}b)^{*})^{*}b$.
Now we remove \index{graph edge} edge $a$ and obtain $\cal A$$_3$ (see fig.\ref{ExampleNFA3}).

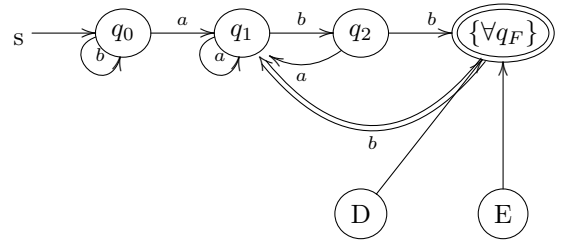
\begin{figure}[h] 
\begin{center}
\begin{tabular}{c}
\xymatrix{
  \txt{s} \ar[r]       & *++[o][F-]\txt{$q_0$} \ar@(l,d)[]^b \ar[r]^a  & *++[o][F-]\txt{$q_1$} \ar@(l,d)[]^a \ar[r]^b  & *++[o][F-]\txt{$q_2$} \ar@/^1pc/[l]^a  \ar[r]^b  & *++[o][F=]{\txt{$\{\forall q_F\}$}} \ar@2{->}@/^3pc/[ll]^b\\
   *\txt{} & *\txt{} & *\txt{} & *\txt{} & *\txt{}\\
   *\txt{} & *\txt{} &  & *++[o][F-]\txt{D} \ar[uur] & *++[o][F-]\txt{E} \ar[uu]
}
\end{tabular}
\end{center}
 \caption{Example finite automaton $\cal A$$_2$}
 \label{ExampleNFA2}
\end{figure}

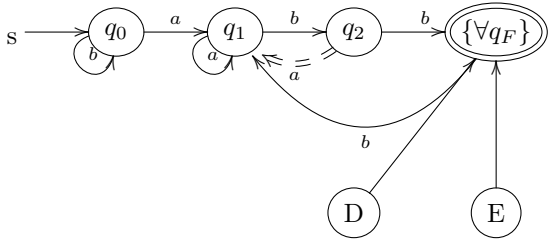
\begin{figure}[h]
\begin{center}
\begin{tabular}{c}
\xymatrix{
  \txt{s} \ar[r]       & *++[o][F-]\txt{$q_0$} \ar@(l,d)[]^b \ar[r]^a  & *++[o][F-]\txt{$q_1$} \ar@(l,d)[]^a \ar[r]^b  & *++[o][F-]\txt{$q_2$} \ar[r]^b \ar@2{-->}@/^1pc/[l]^a  &*++[o][F=]{\{\forall q_F\}} \ar@/^3pc/[ll]^b\\
   *\txt{} & *\txt{} & *\txt{} & *\txt{} & *\txt{}\\
   *\txt{} & *\txt{} & & *++[o][F-]\txt{D} \ar[uur] & *++[o][F-]\txt{E} \ar[uu]
}
\end{tabular}
\end{center}
 \caption{Example finite automaton $\cal A$$_3$}
 \label{ExampleNFA3}
\end{figure}

We get the regular expression $b^{*}a(a^{*}bb(\varepsilon + ba^{*}bb))^{+}$.
All three expressions may be rewritten and further simplified -- here, however, this is not the critical factor.
The real problem is if an edge is inserted to any place in the graph (and this could happen indeed in practice), then the \index{regular expression} regular expression given may heavily change -- in both cases, on average and on worst case quite severely.
The bigger the heap graph gets, the more complex the regular expression gets. It would be optimal and similar to the previous expression.
The linear equalities system based on \index{Arden's lemma} Arden's lemma does not change much on the alteration of a single edge in general.
For example, the first graph can be described as shown in fig.\ref{ExampleArdenLemma}.

\begin{figure}[h]
\begin{center}
\begin{tabular}{lcl}
\begin{tabular}{lcl}
 $Q_0$ & = & $aQ_q + bQ_0$\\
 $Q_1$ & = & $aQ_1 + bQ_2$\\
 $Q_2$ & = & $aQ_1 + bQ_F$\\
 $Q_F$ & = & $\varepsilon + \underline{bQ_1}$
\end{tabular}&
\parbox[b]{4cm}{$\underline{bQ_1}$ denotes insertion of an outgoing edge from $q_F$.}
\end{tabular}
\end{center}
 \caption{Equalities describing the automaton}
 \label{ExampleArdenLemma}
\end{figure}

It is not difficult to see, the system of linear equalities has a solution, but it heavily differs from the previous, although almost all equalities remain the same.
However, the corresponding regular grammar heavily changes.\\

\textbf{\underline{Heap graph upon alteration}.}
First, it is required to research the sequence of joint \index{program statement} program statements to discuss some "adequate" representation of the \index{dynamic memory} dynamic memory and the situation with pointer and transition descriptions. 
Let us consider \index{list inversion} \index{simply-linked list} list inversion as proposed in \cite{reynolds09}.
Both Reynolds \cite{reynolds09} and especially Parlante \cite{parlante01} bring numerous examples, and we decide on one.
Given the program from fig.\ref{ExampleListInversion} in a \index{C} C-dialect with a special syntax for \index{pointer} pointers.
In the program listing, $i$, $j$ and $k$ denote pointers.
$*(i+1)$ can grant access to the memory cell following $i$, which depends on $i$'s type.
Naturally, elements are linked to each other.
They are supposed to be a unit and monolithic and a contiguous \index{memory region} memory region in a \index{dynamic memory} heap, even if a similar \index{syntax} syntax reminds us about it.
The example may clarify the semantics from fig.\ref{ExampleListInversionSteps}, where a number denotes the iteration step until the loop is visited.

\begin{figure}[h]
\begin{center}
\begin{minipage}{10cm}
\begin{verbatim}
 j:=nil;
 while (i!=nil){
   k=*(i+1);       // access to the element following i
   *(i+1)=j;       // the pointer following i alters
   j=i;
   i=k;
 }
\end{verbatim}
\end{minipage}
\end{center}
 \caption{C code example for list inversion}
 \label{ExampleListInversion}
\end{figure}

\begin{figure}[h]
\begin{center}
\begin{tabular}{c}
\begin{tabular}{lr}
  1: \xymatrix{
    \txt{i} \ar[r] & *+=[o]+[F]{1} \ar[r] & *+=[o]+[F]{2} \ar[r] & *+=[o]+[F]{3} \ar[r] & nil\\
    \txt{j} \ar[r] & nil
  }\\\\
  2: \xymatrix{
    \txt{j} \ar[r] & *+=[o]+[F]{1} \ar[d] & *+=[o]+[F]{2} \ar[r] & *+=[o]+[F]{3} \ar[r] & nil\\
                   & nil                  & \txt{i=k} \ar[u]
  }\\\\
  3: \xymatrix{
                   & *+=[o]+[F]{1}        & *+=[o]+[F]{2} \ar@/^/[l]        & *+=[o]+[F]{3} \ar[r] & nil\\
                   & & \txt{j} \ar[u] & \txt{k=i} \ar[u]
  }\\\\
  4: \xymatrix{
     *+=[o]+[F]{1}  & *+=[o]+[F]{2} \ar[l] & *+=[o]+[F]{3} \ar@/^/[l]       & nil\\
                    & & \txt{j} \ar[u] & \txt{k=i} \ar[u]
  }\\\\
  \multicolumn{2}{c}{
  5: \xymatrix{
     \txt{j} \ar[r] & *+=[o]+[F]{3} \ar[r] & *+=[o]+[F]{2} \ar[r] & *+=[o]+[F]{1} & \txt{i=k} \ar[r] & nil
  }}
\end{tabular}
\end{tabular}
\end{center}
 \caption{Example heap trace for list inversion}
 \label{ExampleListInversionSteps}
\end{figure}
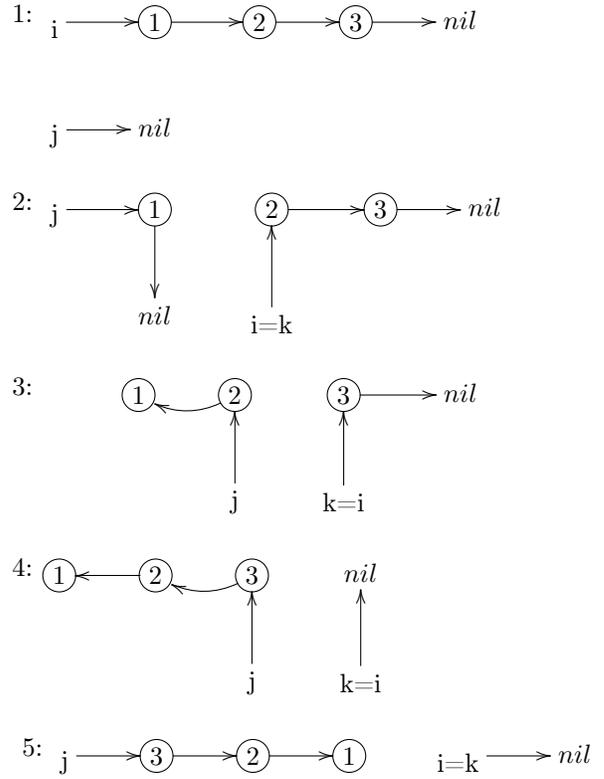

$i$ denotes a \index{linear list} list until entering the loop (cf. fig.\ref{RulesLoopReplacements}).
On leaving the loop, $i$ is empty, and $j$ contains the initial list of $i$ in reverse ordering.
No copies are created.
The input list is iterated precisely once.
We notice edges are not labelled in the example due to the selected implicit definition statement over pointers.
Of course, in general, \index{heap edge} heap edges may be labelled.
However, pointers exist, which are \index{variable!local} local variables.
Those are (not) assigned to heap vertices (when \index{pointer} a pointer is \index{uninitialised} uninitialised).
$k$ is an uninitialised pointer until entering the loop.
It is not hard to notice that compact notation is barely used, for example, from (1) to (2).
The main reasons are \index{expressibility} expressibility and adequate graph representation.
Although the presented regular notation is bare of any use, compact notation still played an essential role in formulating an automated approach (see sec.\ref{chapter:APs}).

It may be obvious why other global approaches, like \index{RC} RC, \index{SA} SA and others introduced in sec.\ref{chapter:intro}, failed.
The reason is the same as demonstrated by the last two examples.
Dodds' approach \cite{dodds08} is also not further considered for the same reason, even if here, instead of the \index{graph transformation} heap graph transition is described.
Still, his approach is not specific to pointers and does not consider imperative PLs (see sec.\ref{chapter:intro}).
Regarding graph transitions as specifications, there are further harsh reasons against (see next section).

\begin{observation}[Graph-based Representation of Dynamic Memory]
\label{obs:GraphIRHeap}

The graph in dynamic memory (heap) can be defined as triple $(V,E,L\times V\cup \{nil\})$, where $V$ are the heap graph vertices, $E$ is the edges set, $L$ is the set of pointer labellings.
A pointer points either to some $v\in V$ or to $nil$, e.g. when it is uninitialised.
\end{observation}

We observe the \index{heap graph} graph represented: (1) separately by each component, or (2) in a mixed form.
Obviously, (1) is not straightforward because it describes the vertices separately (in this case, independent from pointer labelling), and as such, the final specification might become too bloated and thus too hard to read.
A heap's specification must be easy to read, short and adequate.
Approach (1) requires denoting each vertex separately and insert all of them into the \index{specification} specification.
As mentioned in the introductory by numerous examples, the disadvantages outweigh all benefits; hence this approach is dismissed.
Approach (2) assumes naturally either (2a) \index{grap vertex} describe by \textbf{vertices} or (2b) \index{graph edge} by \textbf{edges}.
The problem with (2a) is that there are duplicates in \index{specification} specifications since edges fully identify vertices, and because of two vertices, there is always at most a single \index{graph edge} edge.
Each edge related to a given vertex must be annotated, identical to having each vertex a list of all neighbouring vertices.
Unfortunately, such an approach is very cumbersome.
For example, one vertex is connected to two or more vertices.
So, that all neighbouring vertices also need to store the source vertex into their \index{linear list} lists.
Moreover, if a heap is altered, then \index{specification} specification massively changes -- this is something that needs to be avoided at any cost.

In contrast, approach (2b) only describes \index{graph edge} edges, and \index{graph!vertex} vertices are inserted into these.
This approach contains fewer \index{duplicate} duplicates than with approach (2a). 
If only directed edges were allowed, then the total amount of checks would be halved.
Several outgoing edges per vertex are prohibited because one pointer may not point to more than one cell simultaneously.
Also, on a closer view, \index{programs statement} program statements more frequently alter the calculation state than vertices do.
Moving data, \index{memory deallocation} GC and \index{memory allocation} allocating heap memory are quite expensive, which may invoke system calls to the \index{OS} OS.
Conversely, pointer manipulations are cheap since no additional conditions need to be met, and no resource management is involved.
In the worst case, the approach mentioned \index{memory allocation} allocates and utilises garbage, but pointers hardly change.
Then it is essential to recall whether this approach can be corrected.

Fig.\ref{fig:GraphIsomorphisms} (a) illustrates a \index{regular graph} regular graph whose maximal vertex \index{graph!vertex degree} degree is three.
Here, we assume, each vertex represents an object with three fields that are \index{pointer} pointers.
For approach (2a), 11 vertices need to be specified, each of which has three edges, where the number of incoming and outgoing edges may differ, and this needs to be considered separately.
In total, there are 18 edges.
In approach (2b), only 18 \index{graph edge} edges need to be specified.
Here, vertices connected with more than one edge may be just symbols.
The more a given graph differs from the \index{graph!full} full graph, the fewer edges need to be specified.
If an edge alters, then a \index{specification} specification needs to change less, namely the altered edge only.
If \index{graph!vertex} a vertex changes, then all connected edges need to be checked.
A directed graph may be traversed in $\Theta(n)=n$ to scan both sides of an edge.
According to approach (2b), given some rule set using predicates, proof might be found faster and friendlier to handle.
\index{pointer} Pointers and \index{path accessor} access expressions to \index{object field} fields denote graph \index{graph!vertex} vertices (see sec.\ref{chapter:logical}).
Inaccessible fields during \index{specification} specification are beyond our interest because, by definition, those \index{memory cell} memory cells are \index{garbage} garbage (see sec.\ref{chapter:intro})
and lost forever.
Nevertheless, in \index{specification} specifications, our concern is on locating ambiguous places.\\

Jones \cite{jones11} defines a \index{heap} "\textit{heap}" as a contiguous leap of \index{RAM} RAM with size $2^k$, where $k\ge 0$.
The leap may be interpreted as any \index{data structure} data structure depending on the application.
Alternatively, it is a sequence of non-contiguous blocks of non-contiguous words.
For example, a tree may have fillings between its vertices in the memory layout, but individual vertices  may not be interpreted other than \index{(CPU) word} (CPU) words.
According to Jones, an \index{object instance} object is a set of memory cells that are not necessarily interconnected but whose \index{object field} fields are \index{addressation} addressable. 
Each allocated memory cell has a \index{pointer} pointer.
The question may arise whether an object is alive or not between \index{memory allocation} allocation and deallocation?
It is hard to disagree with Jones that fragmentation is a problem.
However, intra-object fragmentation is excluded.

Cormen claims, on page 151 in \cite{cormen09} as Burstall \cite{burstall72} does, any \index{data structure} data structure in \index{dynamic memory} dynamic memory is automatically a \index{binary tree} tree.
Here, a tree does not only implies some relation "$\le$" to its child vertices $V_j$, but it also obeys the \index{poset} ordering $f(V_{parent}) \le f(kid(V_{parent},j)), \forall j$.
Atallah \cite{atallah98} considers a \index{heap} heap as an array interpreted as a tree with wide memory usage but with high flexibility.
In Atallah's monography, several definitions of \index{heap} heap might be found.
First, a heap is defined as a \index{priority queue} priority queue (see page 79).
\index{Fibonacci heap} Fibonacci heaps (cf.\cite{fredman87}) are presented as specialised and effective queues for insertion and deletion.
Next, on page 105, a heap is defined as a \index{binary tree} binary tree preserving all queue elements with priority.
In contrast to \index{RAM} RAM, Atallah emphasises on page 111 the importance to reference on pointers only.
Thus, any random \index{addressation} access is excluded.
The responsibility of \index{heap} heap separation is assigned to connect heaps (in an implicit way) rarely and shall be supervised by a programmer and software designer.
Rarely connected heaps may effectively be separated and can effectively be processed by different methods.
It is hard to disagree with Cormen on random addressing because of the previously mentioned regulations regarding expressibility.
However, the programmer's responsibility in creating a data structure can hardly be disputed, even if occasionally data structures shall be trees and not graphs.
Sleator \cite{sleator86} suggests for the sake of minimalistic GC and faster heap access to balance trees and their sibling vertices, which leads to an overall fast search below $\Theta_{min}(n)=1$ and deletion below $\Theta(n)=log(n)$.
Even if his proposition is intriguing, still this approach is not further considered because it is currently not needed enough w.r.t. a small number of conjuncts and the possibility of linear search by locations.

Reynolds defines a heap set as a union of mappings from address space onto non-empty \index{memory cell} memory cells' values.
Following this definition, a heap is some address set that refers to some well-defined \index{data structure} data structure (without further notice).
Reynolds' definition is \index{structuralism} \textit{structuralistic} (the word origins to philosophical structuralism) because, strictly speaking, a single heap as an individual and an independent unit does not exist in Reynolds' terms (cf. sec.\ref{chapter:stricter}).
Pavlu \cite{pavlu10} defines a heap as an arbitrary \index{heap graph} graph.

%%%%%%%%%%%%%%%%%%%%%%%%%%%%%%%%%%%%%%%%%%%%%%%%%%%%%%%%%%%%%%%%%%%%%%%%%%%%%%%%%%%%%%%%%%%%%%%
  %%%%%%%%%%%%%%%%%%%%%%%%%%%%%%%%%%%%%%%%%%%%%%%%%%%%%%%%%%%%%%%%%%%%%%%%%%%%%%%%%%%%%%%%%%%%%%%
%%%%%%%%%%%%%%%%%%%%%%%%%%%%%%%%%%%%%%%%%%%%%%%%%%%%%%%%%%%%%%%%%%%%%%%%%%%%%%%%%%%%%%%%%%%%%%%

\subsection{Predicates}
\label{sect:ExprPredicates}

Apart from the previous remarks, the most critical requirements towards predicates are:
(1) after each \index{program statement} statement touches the \index{dynamic memory} heap implies the corresponding graph alters too minimally
(2) heap specification also allows minimal alteration only.

When obeying these requirements, then the next question relates to adequacy representation.

In ancient Greece, \index{Plato} Plato's philosophical school taught the concept of epistemological definition and research.
The object is to describe some object of interest by following a well-known allegory in a strictly defined order of questions and answers in dialogue form between two sides: a human who has some object in front of himself but is not able to show it to others due to isolation, and another human in freedom.
As often is the case with Greek legends, the earlier human is massively hindered from his luck and captured in a cave.
This man is keen on the understanding that object but has minimal possibilities.
He communicates with his outside world only by his voice and some eternal torch.
The torch's light hits the object and casts only the shadows on the cave's wall visible to the outside world.
However, the shadows are not exact, but the captured man's voice helps understand that object.
This allegory is better known as the \index{allegory of the cave} "\textit{allegory of the cave}".
It proposes to describe an object as a two-sided dialogue whose primary goal is to find properties of that object which may not be evident in the context of minimal communication bandwidth.
There are several philosophical concepts very close to the allegory, so-called "\textit{Gedankenexperimente}" (from German "\textit{thought experiments}"), for instance, a dialogue that is speechless and written only and better known as Searle's "\textit{Chinese Room}".
For simplicity, we restrict ourselves to Plato's classic cave myth, which already contains all we need to understand the phenomenon of heap predicates here.
Intuition may narrate to us what a heap abstracted predicate may tell us, but comparing previous definitions of it will lecture us the opposite.
An abstract explanation of a heap in practice is a concrete placement of incoming terms and symbols, but it hardly tells us about incoming properties (so-called  \index{reification} reification --- a concept of abstraction for a concrete realisation).
These philosophical concepts will help us bring some light to heap representation and processing in this section and sec.\ref{chapter:stricter}, sec.\ref{chapter:APs}.
A modern philosophical discourse in idealism can be found in Hegel's classic (philosophical) method.
First, a thesis is established founded on numerous observations made.
Second, more sophisticated analyses are performed and checked against the established thesis -- here, additions may be added to the theory build-up.
However, counter-positions are formed, s.t. a counter-thesis is established.
Third, thesis and counter-thesis are examined thoroughly, and \index{thesis-antithesis-synthesis} synthesis is formulated -- an assertion based on previous theses.
One can notice all heap properties are thoroughly analysed at the current stage so that a first proposition thesis might be formulated here.
An in-depth dialectic analysis (after Hegel) is required to localise the potential imprecision -- all very close obviously to the latter methodology mentioned.

Regarding \index{reification} reification, \index{heap} heaps must have both \index{syntax} syntax and at least one \index{semantics} semantics -- both will be defined soon, generalising the heap term as a (heap) \index{predicate} predicate.
Therefore, a summary is needed to find the essence starting with \index{Aristotle} Aristotle's predicates called \index{syllogism} syllogisms.
A syllogism is a logical rule containing three assertions (also compare with Hoare triples): if $A$ and $B$, then $C$ follows.
This oversimplification is just fair for our objective.
One remark must be made.
Most deductive systems nowadays do not use the over 2000 years old syllogisms but on slightly modified variants (see later), as may be of help in sec.\ref{chapter:APs}.

As was seen, the previous attempt in using \index{regular expression} regular expressions failed because modifications affected a heap graph's transitions.
A minimal requirement refers to a rule set but not to a given assertion to be proven itself because the representation is not too stable for the given problem.
In general, reasoning about and transforming abstract notations (cf. sec.\ref{sect:LanguageCompatibility}) is widely discussed and under constant research.
For example, Conway's cell automata \cite{wolfram02} attempt to establish invariant expressions --- to mention one example.

In sec.\ref{chapter:APs}, the inverse method is observed: patterns are observed from which properties of rules are derived.
Predicates connect heap graph vertices regardless of whether in the context of a \index{language!formal} formal or \index{language!natural} natural language first of all.
According to \index{semiotics} semiotics, a predicate may have two meanings:
(i) an intuitive meaning, this relates to the question: \index{heap} "\textit{What exactly is a heap supposed to mean}"?
(ii) a \index{connotation} connotative meaning --- \textit{What is a heap associated with exactly}?
Next, the question related to heap representation arises: may/must a heap use \index{symbol} symbols and relations and what exactly are they denote? May perhaps a heap be defined \index{full definition} partially?

By assessing (i), it is found that a heap represents a set of \index{pointer} pointers, which are somehow connected and point to \index{dynamic memory} heap objects.
Since previously issues were found related to \index{ambiguity} ambiguities and massive potential restrictions, it is wise to keep definitions as atomic as possible and reduce everything to a minimum.
By assessing (ii), the input program determines the connection between all components, and the heap is not organised.
That implies \index{graph!vertex} graph vertices may occur in any order, at any place in specifications, and the \index{memory region} memory layout they fit in must not be contiguous. 
\index{heap} Heaps may be linked to others.

It is important to note that a \index{predicate} predicate must provide a possibility to express a connection between two heap graph vertices but must also express that a link is forbidden.
If there is no such possibility, then \index{heap!separation} separation automatically becomes an implicit result of an analysis of all heaps, which is unfortunate due to effectiveness concerns.
The definition of \index{heap!connected} connectivity has different levels of graininess: "\textit{connected}", "\textit{maybe connected}", and "\textit{not connected}", "\textit{maybe not connected}".
With time, a heap's \index{heap!modality} mode is not of immense significance: first, because "\textit{connected}" and "\textit{not connected}" may be checked in linear time.
Second, the search efforts are discrete, especially before and after \index{program statement} each program statement.
In heap predicates, due to \index{logical operator} logical operators, varieties must all be considered.
So, for instance, logical \index{disjunction} disjunction must be expressible regardless of the \index{non-repetitiveness} non-repetitiveness principle and \index{predicate!negation} predicate negation.
The solution to this question is solved in sec.\ref{chapter:logical}, where a proof over heaps will be based upon \index{language!logical} logical programming.
We shall not lose track that pointers may be arbitrary, including \index{variable!local} locals and \index{variable!dynamic} dynamic variables and \index{object field} \index{object instance} object fields.

In the very first work on \index{SL} SL, Reynolds \cite{reynolds02} introduces the $\star$-operator over \index{heap} heaps and introduces the following properties:

\begin{theorem}[Reynolds' Spatial Properties]
For propositions about heaps $q$, $p$, $p_1$, $p_2$, the set of \index{variable!free} free \index{symbol} symbols $FV(.)$ and a binary spatial disjunction $\star$ the following rules (1-6) hold:

\begin{tabular}[t]{ll}
 (1) & Non-repetitiveness:\\
     & $p \not \Rightarrow p \star p$, $p \star q \not \Rightarrow p$, if $\exists q,q \not \equiv \texttt{emp}$\\
 (2) & Commutatitivity:\\
     & $p_1 \star p_2 \Leftrightarrow p_2 \star p_1$\\
 (3) & Associativity:\\
     & $(p_1 \star p_2) \star p_3 \Leftrightarrow p_1 \star (p_2 \star p_3)$\\
 (4) & Neutral element:\\
     & $p \star \texttt{emp} \Leftrightarrow \texttt{emp}  \star p \Leftrightarrow p$\\
 (5) & Distributivity:\\
     & $(p_1 \vee p_2) \star q \Leftrightarrow (p_1 \star q) \vee (p_2 \star q)$\\
     & $(p_1 \wedge p_2) \star q \Leftrightarrow (p_1 \star q) \wedge (p_2 \star q)$\\
 (6) & Quantification:\\
     & $(\exists x.p) \star q \Leftrightarrow \exists x.(p \star q)$, if $x \not \in FV(q)$\\
     & $(\forall x.p) \star q \Leftrightarrow \forall x.(p \star q)$, if $x \not \in FV(q)$
\end{tabular}
\label{theo:ReynoldsHeapProperties}
\end{theorem}

\begin{proof}
\index{non-repetitiveness} Non-repetitiveness implies that $p$ does not occur twice -- this matches the non-repetitiveness \index{compactness} of heap \textit{descriptions} \cite{restall94}.
However, in contrast to \index{assertion} classic logical assertions, $p \star q \not \Rightarrow p$ does not imply strict $p$ as a separate assertion.
The reason is not $p$ is mandatorily connected with $q$.
As supposed, descriptions and intuition about the separating operator might be referred to as introduced on separability.
Though this a solely scenario only.
The problem is in describing the element search on the \index{dynamic memory} heap.
However, this does not imply that one heap may suddenly vanish from an assertion about two arbitrary heaps.
This difference must be considered.

Rules (2-4) are understood so that the \index{heap graph} heap graph definition may be well-founded.
The soundness of rule (5) may be shown by applying \index{induction} induction to the remaining rules.
Here two cases may be distinguished.
First, when $p_1$ and $p_2$ clash and second when there is no clash with $q$.

Rule (6) may seem obvious.
Hints may be possible naming clashes that may be avoided by \index{$\alpha$-conversion} variable renaming, which in terms of \index{$\lambda$-calculus} $\lambda$-calculus is an \index{$\alpha$-conversion} $\alpha$-conversion.

Neither rule (1), nor (6), nor any other rule do not exclude \index{variable!free} free variables as in verifiers like \index{Smallfoot} Smallfoot or jStar.
One variable may be used, e.g. as a \index{pointer} pointer in one \index{heap} heap and as a referenced pointer as subexpressions in another heap.
\end{proof}

Let us consider the heap from fig.\ref{ExampleHeap1}.

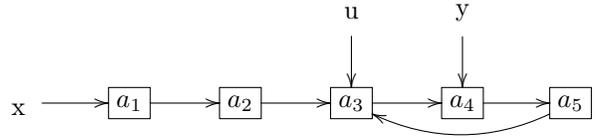
\begin{figure}[h]
\begin{center}
\begin{displaymath}
\xymatrix @W=1pc @H=1pc @R=0pc{
  &&& \txt{u} \ar[dd] & \txt{y} \ar[dd]\\\\
  \txt{x} \ar[r] &   *+[F]\txt{$a_1$} \ar[r] & *+[F]\txt{$a_2$} \ar[r] & *+[F]\txt{$a_3$} \ar[r] & *+[F]\txt{$a_4$} \ar[r] & *+[F]\txt{$a_5$} \ar@/^1pc/[ll]
}
\end{displaymath}
\end{center}
 \caption{Example heap with pointers \texttt{x,u,y}}
 \label{ExampleHeap1}
\end{figure}

The heap has references to the following pointers $x$, $u$, $y$.
The rectangles frame its content which is $a_1$, $a_3$ and $a_4$.
\index{memory model!Burstall} Burstall suggests not to reference the content directly but by an address.
He also suggests noting all intermediate pointers and labellings in some minimal model, s.t. a big heap may be equivalent to a single compact expression \xymatrix{ x \ar[r] ^{a_1,a_2,a_3} & y}.
This expression has overlayed meanings: pointers $x$ and $y$ exist, a path between these is guaranteed, and naturally, the whole path describes part of the overall \index{heap graph} heap graph.
In linear lists and trees pointers, "\texttt{next}", if any, are skipped by default.
One expression after Burstall denotes a whole \index{linear list} linear list.
The heap graph describes the union of such expressions.
Let us consider the following data structure, a "\textit{cactus}" (see fig.\ref{ExampleCactusExample}).

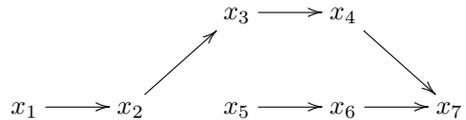
\begin{figure}[h]
\begin{center}
\begin{displaymath}
 \xymatrix{
              &             & x_3 \ar[r] & x_4 \ar[dr] \\
   x_1 \ar[r] & x_2 \ar[ur] & x_5 \ar[r] & x_6 \ar[r] & x_7
 }
\end{displaymath}
\end{center}
 \caption{Example \index{cactus} cactus in dynamic memory}
 \label{ExampleCactusExample}
\end{figure}

For its description, only two expressions are required.
$x_1$ is a \index{pointer} pointer, which is on the \index{stack} stack.
Assume, all other cells are not related to any pointer locations.
Consecutively, it is all in \index{dynamic memory} dynamic memory.
A handy, though compact, notation destroys the previously established minimality criteria.
Hence, Burstall's notation is no further consideration.
However, the \index{abstraction} abstraction of predicates is considered towards a more flexible description based on an \index{graph edge} edge-based graph.\\

For full pointer and \index{expressibility} expressibility analyses, \index{addressation} random access to dynamic memory is not required.
All accessory operations may be fully substituted, as mentioned earlier, by either read or write upon valid non-random locations only.
If, however, random heap access is needed for performance, most often access to consecutive memory cells, this would be a typical case of an algorithm using the \index{stack} stack.

Allocation of arbitrary memory sizes is not dedicated to heap only.
It may be implemented by both stack and heap.
Checks are going to be performed neither for dangling pointers nor for garbage nor invalid cells.
It is not hard to imagine \index{predicate!auxiliary} auxiliary functions over memory regions that may perform any computable transformation.
\index{.-operator} \index{path accessor} Pointers can only provide path accessors and be well-defined.
C(++) Objects are special \index{data structure} struct, where visibility is private by default rather than public \cite{isocpp14}.
Objects can either be fully on the stack or heap, where the pointer location itself is most likely in the \index{stack} stack, except when pointers of pointers are taken into consideration for a full heap accommodation.
On occasions, when stacked objects are tiny, they may even be put into \index{(CPU) word} CPU registers.
The latter variant fundamentally does not exclude the use of dynamic memory.
It must be stated, $e_1.f_1 \mapsto val_1$ may mean object $e_1$ may have two or more \index{object field} fields, and in generated code, $e_1.f_1$  may be different by use and memory layout then, for instance, $e_1.f_2$ due to \index{code generation} optimisation heuristics \cite{kennedy02}, \cite{muchnick07}, \cite{gcc15} regardless of previous assumptions.
However, this does not imply a required division on the \index{language!input} input PL or \index{language!specification} specification language, where any \index{object instance} object must be modelled as a contiguous \index{memory region} memory region whose \index{object field} fields may have references to \index{pointer} further objects.\\

As introduced to \index{SL} SL there were made numerous contributions and some propositions regarding improvement \cite{reynolds02}, \cite{reynolds09}, \cite{burstall72}, \cite{hurlin09}, \cite{parkinson05-2}, \cite{parkinson05}, \cite{parkinson06}, \cite{bornat00}, \cite{yang02}, \cite{berdine05-2}, \cite{berdine05}, \cite{ohearn04}, \cite{scholz99} (see sec.\ref{chapter:intro}).
Reynolds \cite{reynolds02} introduces \index{sequentialiser} \index{,-operator} the "\textit{sequentialiser}-operator" "\textbf{,}" for an implicit \index{linear list} linear list construction.
"\textit{Implicit}" means no assumption is made on how the concrete memory layout looks.
It only insists addressing has to be consecutively obeying the order of the "\textit{,}"-operator.
Thus, from the earlier example, a \index{cactus} cactus may be defined as $x_1 \mapsto x_2,x_3,x_4,x_7 \wedge x_5 \mapsto x_6,x_7$.
Alternatively to the implicit operator, some explicit pointers may be used, e.g. by pointer containing objects or Bozga's \cite{bozga08} length-delimited definition.

As discussed earlier, higher expressibility may be achieved by parameterisation.
Parameterisation may fundamentally make existing \index{predicate!parameterisation} heap predicates more flexible.
In order to achieve this, \index{variable!symbolic} symbolic variables must be introduced to \index{assertion} assertions.
An assertion may be valid or invalid for a given \index{heap} heap.
\index{symbol} Symbols are not typed apriori, as terms are in \index{type calculus!Church} Church's type calculus.
Initially, heap predicate-argument types seem closer to \index{type calculus!Curry} Curry's type calculus.
The introduction of (symbolic) variables, predicate definition, is closely related to parameterised terms in the $\lambda$-calculus.
So, a predicate is abstracted and must be applied to other \index{predicate} predicates.
In contrast to the functional paradigm, a predicate may return more than just true or false.
It may unify partially defined unground terms (see sec.\ref{chapter:logical}).

Let us consider a recursive example \index{binary tree} of a binary tree taken from \cite{reynolds02}:

$$tree(l)::=\texttt{nil} \ | \ \exists x.\exists y:\ l \mapsto x,y \ \star \ tree(x) \ \star \ tree(y)$$

According to Reynolds' \index{memory model!Reynolds} definition, the operator \index{$\star$} $\star$ defines a \index{heap} heap, consisting of two \index{disjunction} separated heaps.
It shall be mentioned, as stated earlier, that the initial definition of the $\star$-operator may still contain connecting elements.
Thus, starting with one tree $x$, there may be a path leading to some adjacent tree $y$ regardless of $tree(x) \star tree(y)$.
It is assumed the whole region pointed by $x$ indeed does not overlap with the analogous region pointed by $y$.
It is essential that without any change of the given predicate, this was made impossible.
However, during further parameterisation and predicate alteration (cf. predicate \texttt{tree} in sec.\ref{chapter:logical}), the problem does not change.
More illustratively, this is depicted in fig.\ref{ExampleSchemaHeapSeparation}.
Abstraction of predicates will be formalised and a new approach introduced in sec.\ref{chapter:APs}.

\begin{figure}[h]
\begin{center}
\begin{displaymath}
  \xymatrix{
    l \ar[rr] \ar@{.>}[drr] && *+[F] \txt{x} \ar[d]_{``,''} \ar[r] &  *+[F] \txt{...}  \\
    && *+[F] \txt{y} \ar[r] & *+[F] \txt{...} \ar@{|<<.>>|}[u]_{\txt{\Huge{\Leftscissors}}}
  }
\end{displaymath}
\end{center}
 \caption{Example of schematic separability of heap}
 \label{ExampleSchemaHeapSeparation}
\end{figure}
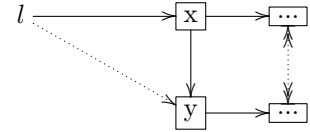

Based on Reynolds' definition, Berdine \cite{berdine05} introduces the heap satisfaction relation as shown in def.\ref{def:HeapSatisfactionRelation}.

\begin{definition}[Heap Satisfaction Relation]
\label{def:HeapSatisfactionRelation}

"$\models$" denotes the model relationship over heap formulae $s \in \Sigma$, where \index{stack} a pointer is defined on the stack and pointing to content in the heap is denoted by $h \in \Pi$.
$r(t_i)$ is the $t_i$-th component of structure $r$ (see fig.\ref{FormalDefinitionStackHeap}).
$E$, $F$ denote arbitrary assertions.
$\llbracket . \rrbracket$ is the functional denotation for some assertion mapping onto the boolean set.

\begin{figure}[h]
\begin{center}
\begin{tabular}[t]{lcl}
 $s\models E=F$ & if & $\llbracket E\rrbracket s = \llbracket F \rrbracket s$\\
 $s\models E\neq F$ & & $\llbracket E\rrbracket s \neq \llbracket F \rrbracket s$\\
 $s\models \Pi_0 \wedge \Pi_1$ & & $s\models \Pi_0$ and $s\models \Pi_1$\\
 $s,h\models E_0 \mapsto t_1: E_1,$  && $h=[\llbracket E_0 \rrbracket s\rightarrow r]$ \\
   $\cdots , t_k:E_k$ & &\\
 $s,h\models \texttt{emp}$ & & $h=\emptyset$\\
 $s,h\models \Sigma_0 \star \Sigma_1$ & & $\exists h_0,h_1.h=h_0\star h_1$,\\
    && $s,h_0 \models \Sigma_0$, $s,h_1 \models \Sigma_1$\\
 $s,h\models \Pi \wedge \Sigma$ & & $s\models \Pi$ and $s,h\models \Sigma.$
\end{tabular}
\end{center}
 \caption{Formal heap definition according to SL}
 \label{FormalDefinitionStackHeap}
\end{figure}
\end{definition}

A \index{denotational semantics} denotational semantics is a function $\llbracket . \rrbracket$ of \index{type} type $\Phi \times \Sigma \rightarrow Bool$, where $\Phi$ is a set of assertions on a heap, and $Bool$ \index{boolean denotation} is the boolean set.
For a given \index{linear list} linear list $s,h \models E_0 \mapsto t_1,\cdots , t_k$ with corresponding types $\forall i,j \in \mathbb{N}_0.E_j$, $r(t_i)=\llbracket E_i \rrbracket s$, $1 \le i \le k$.
Left of $\models$ is the calculation state, which has type $\Pi \times \Sigma$, right of it is some arbitrary \index{boolean denotation} boolean assertion.

In \cite{berdine05-2}, Berdine refers to problems, particularly the \index{frame} frame may remotely alter (see obs.\ref{observation:RemoteAlternation}).
However, deleting a memory cell's content referenced by some \index{variable!global} global pointer is a separate issue (see sec.\ref{chapter:intro}).

For the same reason, \index{subprocedure} subprocedures are neither considered further.
It must be noted that embedded procedures may only increase the complexity of a \index{specification} specification and verification w.r.t. \index{expressibility} expressibility.
They may, however, not increase computability.
Berdine fairly notices that the application of the frame-rule may lead to \index{non-determinism} non-determinism of \index{symbol} symbols to be substituted.
However, when heaps are under investigation, non-determinism may be excluded by renaming and further constraints (cf. sec.\ref{chapter:stricter}, sec.\ref{chapter:APs}).
Moreover, logical conjuncts fit into the \index{language!logical} logical PL (see sec.\ref{chapter:logical}).\\

Further, the \index{heap graph} heap graph is refined according to \cite{haberland16-2}, and afterwards, syntax and semantics are introduced, obeying the previous complex analysis.

\begin{definition}[Finite Heap Graph]
A finite \index{heap graph} heap graph is directed and connected.
The graph may contain cycles and must be simple.
It is transferable to dynamic memory by compilation, but also vice versa.
Each vertex contains content, has some \index{type} type and a \index{memory cell} memory address and occupies some contiguous \index{memory region} memory region.
In the case of object types, the target address (of any field) can be calculated from some offset base address, and the size an object occupies -- both are derived from the object's type.
Graph vertices do not overlap.
For the sake of simplicity, but w.l.o.g. each edge's destination refers to some absolute address in \index{dynamic memory} dynamic memory.
Locations may be assigned to vertices.
Vertices are finite.

In contrast to sec.\ref{sect:HeapGraph}, it is not intended to distinguish yet initial state and final states explicitly.
\label{def:finiteHeapGraphDefinition}
\end{definition}

Implicit definitions, for example, in \index{SSA-form} SSA-form, often succeed as they fully cover, though implicitly, \index{data dependency} data dependencies.
The initial definitions contain implicit definitions, which are investigated in this work regarding \index{heap} heaps.
Relations, spatial operators, and partially implicit definition refer to graph vertices, \index{typing} typing and \index{symbol} symbolic variables.
In sec.\ref{chapter:stricter}, \index{spatial operator} spatial heap operators \index{strengthening} are strengthened, \index{constant function} constant functions and other conventions are introduced for \index{object instance} objects.\\
From theo.\ref{theo:ReynoldsHeapProperties}, def.\ref{def:finiteHeapGraphDefinition}, and previous conventions introduced (see \cite{haberland16-2}), a "\textit{heap term}" may be defined.

\begin{definition}[Heap Term]
Heap term $T$ is inductively defined as follows and denotes \index{heap graph} a heap graph.

\begin{center}
\begin{tabular}{lll}
 $T::=$ & $loc \mapsto val$ \qquad \qquad & .. \textit{simplex (simple heap)} \\
        & $| \ T \star T$ \qquad \qquad & .. \textit{conjunction} \\
        & $| \ \underline{true} \ | \ \underline{false} \ | \ \underline{emp}$ & .. \textit{constant predicates}\\
        & $| \ ( \ T \ )$ \qquad \qquad &
\end{tabular}
\end{center}

where $loc$ denotes some location.
A complicated expression or some \index{symbol} symbol representing a heap may serve as a \index{location} location.
$val$ is a type, a compatible graph vertex with some value.
$\star$ denotes the spatial operator proposed by Reynolds (cf. theo.\ref{theo:ReynoldsHeapProperties}, but replaced later (see sec.\ref{chapter:stricter}, sec.\ref{chapter:APs}))
\label{def:HeapTermDefinition}
\end{definition}

$\underline{true}$ denotes \index{tautology} a tautology independent from some concrete heap.
The same holds for $\underline{false}$.
\index{predicate} Predicate $\underline{emp}$ is valid only when a given heap is empty.
Otherwise, it is false.
In sec.\ref{chapter:stricter}, \index{conjunction} conjunction is strengthened and resolves in two operations.
For logical assertions, logical conjunctions are introduced into the recursive definition of $T$.

\begin{definition}[Heap Term Extension]
The heap term extension \index{term} $ET$ extends $T$ from def.\ref{def:HeapTermDefinition} by \index{conjunction} logical conjunctions and is defined as follows.

\begin{center}
\begin{tabular}{lll}
 $ET::=$ & $T$ \qquad \qquad & .. \textit{heap term}\\
         & $| \ p(\alpha)$ \qquad \qquad & .. \textit{abstract predicate call}\\
         & $| \neg ET$ & .. logical negation\\
         & $| \ ET \wedge ET$ & .. logical conjunction\\
         & $| \ ET \vee ET$ & .. logical disjunction
\end{tabular}
\end{center}
\label{def:HeapTermExtendedDefinition}
\end{definition}

Logical conjunctions "$\wedge, \vee, \neg$" do not require further explanation.
A predicate call requires it has previously been defined in $\Gamma$ (considering def.\ref{def:PredicateRuleSetDefinition} and cor.\ref{corollary:PredicateEnv}).
Whilst a predicate call actual and expected \index{heap} heaps are compared.
Free \index{symbol} symbols from the predicate declaration need to be \index{term unification} unified.
Unified terms only contain bound variables.
Otherwise, the call is undetermined (cf. sec.\ref{chapter:logical}).\\

If a \index{language!logical} logical PL is used to reason about \index{symbol} symbolic variables, then many restrictions, such as one-sided assignments, impossibility to refer to symbols instead of values, and many others (cf.\cite{berdine05}, \cite{berdine05-2}, \cite{parkinson05}, \cite{parkinson05-2}) may be lifted finally.
When unifying terms, then comparisons become simplistic, and gaps get filled by meaningful content.
Otherwise, all subterms need to be compared manually and placed into the right places, but this is tractable, though it requires further differentiated conditions.
It may result in new mistakes and even further limitations, but in practice, limitations are often observed when for the sake of the non-logical language, whole comparisons are dropped, or symbolic use degrades, for instance, when call-by-value.
Usually, this can be observed mostly in \index{language!imperative} imperative PLs, but also some \index{language!functional} functional ones.

Let us consider the predicate call (more details in sec.\ref{chapter:logical})

\begin{center}
\begin{tabular}{c}
"\texttt{?-pred1(s(s(zero)),\_)}"
\end{tabular}
\end{center}

where for the sake of simplicity, the first \index{term!incoming} term is incoming, but the second \index{term!outgoing} is outgoing.
Conclusively the question arises: Does some anonymous variable "\texttt{\_}" exist, s.t. predicate \texttt{pred1} holds for incoming term \texttt{f(a)}?
If the answer is true, then the goal succeeds, and the result is dropped.
If not, then predicate \texttt{pred1} does not hold.
\texttt{s(s(zero))} represents integer "$2$" in \index{arithmetic!Church} \index{Church arithmetics} Church's arithmetic (cf. fig.\ref{code:NaturalNumbers}).
If, for instance, instead of \texttt{s(s(zero))}, say, \texttt{s(s(\_))} is chosen, then the result may be derived, for instance, to \\
\texttt{s(s(s(zero)))}, then the query without alteration \index{Prolog!query} (see def.\ref{def:QueryToProlog}) might be replaced by\\
"\texttt{?-pred1(s(s(\_)),} \texttt{s(s(s(zero))))}", assuming the predicate has double-sided definition.
In general, this does not only affect two but several directions.
Prolog has \index{evaluation ordering} a strict LR evaluation ordering.
So, symbols may replace concrete heaps.
If some goal requires a concrete heap and that heap is assigned in some later subgoal only, either the subgoal order may be changed, or the overall calculation does not terminate (see sec.\ref{chapter:logical}).\\

Now, it is required to investigate the properties of the mapping from a heap definition onto some heap graph and properties of separate pointers.

\textbf{Property 1 -- Soundness.}

Whenever from def.\ref{def:HeapTermDefinition} a strict difference is encountered between \index{heap!connected} connected and non-connected heap (see sec.\ref{chapter:stricter}), then syntactic description covers any heap graph.
If normalising according to theo.\ref{theo:ReynoldsHeapProperties}, including, for instance, \index{prenex-normal-form} prenex-normal-form, then obtained formulae commute.
\index{constant function} Constant predicates may be exclusions, and therefore may be single-sided enrichments: the set of satisfiable heaps \index{mapping} maps onto one set representative.
It is not hard to find out what a \index{inverse mapping} non-inverse mapping denotes in this context.
Furthermore, when excluding \index{anonymous function} anonymous symbols, then the syntactical description fully corresponds to a \index{heap graph} heap graph.
By excluding the only sources of \index{non-determinism} non-deterministic representations, it is easy to find the mapping is now \index{isomorphic mapping} isomorphic, and consecutively, two differing heap graphs may not be expressed by one description, and vice versa.

\textbf{Property 2 -- Completeness.}
\index{completeness} A heap graph can be fully described by \index{atomic assertion} simple heaps and enriched by pointers.
Knowledge representation of heap vertices may lead to \index{paradox} syntactic antinomies if addresses of objects to be referenced are not distinguished.
For example, if there co-exist $a \mapsto 3$ and  $b \mapsto 3$, then it does not mean that both memory cells containing "$3$" are identical in practice.
For modelling purposes, this means precisely it.
Namely, if a new object is not allocated and no pointer points to that cell, then the pointer becomes an \index{alias} alias.
Annotations may avoid this situation in the term-based object.

Pointers to pointers have an \index{integers} integer as address and therefore do not differ from other pointers.
The difference between some integer and a given address is only up to its interpretation, namely a variable's type.
On mappings from graph onto formula, the evaluation order and the predicate structure naturally may not be derived uniquely.
The mapping might be generated, but it may differ.
Assume a specific set of \index{abstract predicate} abstract predicate definitions is given, then the mapping described would, in general, not be decidable because of the \index{Halting-problem} Halting-problem.
The mapping from a heap graph to $T$ is fully defined, obeying mentioned remarks.
The inverse mapping is complete.
If in $loc\mapsto val$, the left-hand side is no pointer, then the annotation may, in general, be done by an additional field $f$: $loc\mapsto_{f} val$.
For a generalised structure of a heap graph, the additional annotation has no significance.
Hence, additional fields are not announced by default.

\textbf{Property 3 -- Equivalence relation.}
\index{equivalency relationship} For comparison might be used that "$\mapsto$" is a binary \index{functor} functor, and $a\mapsto b$ is a heap.
Thus, it might be confirmed that the equivalence relation "$\sim$" maybe well-founded by showing properties (i-iii).
It is agreed upon "$\mapsto$" has a higher priority than "$\sim$".
(i) Reflexivity: $a\mapsto b \sim a \mapsto b$.
(ii) Symmetry: $a \mapsto b \sim c \mapsto d$, then $c \mapsto d \sim a \mapsto b$.
(iii) Transitivity: $a \mapsto b \sim c \mapsto d$ and $c \mapsto d \sim e \mapsto f$, then $a \mapsto b \sim e \mapsto f$.

\textbf{Property 4 -- Locality.}
The heap graph, its corresponding term IR and its content hardly change on \index{locality} deletion, alteration or insertion of a \index{pointer} pointer.
Edge removal leads to a reduction to some simple heap.
Alternatively, it leads to \index{predicate!parameterisation} \index{abstract predicate} abstract predicate parameterisation or alteration of used abstract predicates because the relation between formula and graph is undecidable in general (see earlier), and therefore, predicates need to be considered separately.
A maximal alteration may lead to \index{graph!vertex} vertex removal if no predicates because it means a single vertex removal and all its \index{graph edge} connected edges.
An analogous insertion is tractable because edge insertion is done stepwise.
In the worst-case scenario, predicates may lead to a formula, which is very different from the previous step.
Hence, \index{recursive type} recursive structures are naturally and effectively described by recursive schemes, and one minor modification in a fragment does not affect all other elements, except neighbouring elements.
This estimate is just \index{heuristics} heuristics and depends on a concrete algorithm and graph separability of most independent subgraphs.
The same holds for the heuristic: "\textit{Some heap graph can be described better by abstract predicates than by compact descriptions}".
The separability problem is too general and beyond this work, particularly the integration of \index{SMT-solver} SMT-solvers (see sec.\ref{chapter:stricter}, \ref{chapter:APs}, cf. \cite{distefano06}).\\

In conclusion to this section, some notes shall be made on \index{predicate!higher-order} higher-order predicates.
As found in the previous section, they do not have a theoretical meaning for comparing heap.
Eventually, predicates having predicates as parameters seem to be promising regarding the \index{expressibility} expressibility of concise descriptions.
However, considered \index{heap} heap predicates are \index{inductive structure} inductively-defined.
So, arbitrary higher-order predicates may destroy these properties if not imposing further constraints.
Higher-order predicates over heaps may absorb recursion from rules, so recursion can no longer be noticed immediately within the rules --- this conflicts with \index{LALR-recogniser} LALR, LL(1) and some SLR-parsers.
However, theoretically, it does not break CF-mighty recognisers (cf. sec.\ref{chapter:APs}).
Thus, it is considered not targeted when using anything different than \index{grammar!formal} formal grammars (see sec.\ref{chapter:APs}).
It must be noted by introducing \index{predicate!higher-order} higher-order predicates essential properties break, such as the CFG and calls with alternating evaluation order and \index{subgoal} subgoal processing, as well as static \index{typing} typing.
However, from a heap expressibility's view, nothing changes.
According to the \index{quality ladder} quality ladder from fig.\ref{fig:QALadder}, properties 3 and 4 belong to property 1, and the last note on higher-order predicate recognition and the second half of property 4 belong to optimal stability.

\section{Logical Programming and Proof}
\label{chapter:logical}

This section researches how a heap (see sec.\ref{chapter:expression}) shall be represented in Prolog and obeyed criteria.
Afterwards, general theorems on heaps using Prolog are outlined.
The syntax of terms and rules is defined.
The cuts and applicability of recursive definitions are discussed.
All this enables logical reasoning, which will eventually solve verification as designed in more detail in sec.\ref{sect:Implementation}.
The abstract predicates' declarative characteristics are analysed in more detail reused later in sec.\ref{chapter:stricter} and sec.\ref{chapter:APs} to unify specification and verification.
The languages are represented in terms of Prolog based on relations and accompanied by software metrics.
Finally, a Prolog-based verification system is proposed.
The representation and integration of objects w.r.t. Prolog is discussed.

\subsection{Prolog as Logical Reasoner}
\label{sect:PrologAsReasoningSystem}
This section does not introduce \index{Prolog} Prolog whatsoever.
Instead, the core Prolog features are probed to reach this section's objective in designing a heap verification architecture based on Prolog.
This section refers to \cite{sterling94} and \cite{bratko01} as the first sources on Prolog, highly recommended.
Prolog has already been found as a convenient general-purpose logic PL for academic and practical needs. 
For example, Feigenbaum's theorem was proven with a Prolog verifier \cite{koch96}.
Furthermore, a rare but very severe bug with \index{real numbers} real numbers critical for Intel Pentium could be proven incorrect using the Prolog-based theorem prover \index{ACL2} ACL2 \cite{kaufmann00}.
Though considered by many purely academic, Prolog found application in verifier HOL Light \cite{price95}, validating and transforming terms and \index{semi-structured data} semi-structured data in general \cite{haberland08-1}.

A Prolog program is a \index{knowledge base} knowledge base, a rule set of \index{Horn-rule} Horn rules.
One or more \index{subgoal} subgoals might request an inquiry to the knowledge base.
In order to define rules and subgoals, Prolog expression terms need to be defined.

\begin{definition}[Generic Prolog-term]
 \index{term} A \textit{term} $T$ in Prolog is defined as:
 
 $T ::= 
 \left\{
 	\begin{array}{ll}
 		x & \mbox{symbol } x \in \overline{X} \mbox{ = alphabet}\\

 		X & \mbox{symbolic variable } X \in \overline{X}\\
 		
 		[] & \mbox{empty list}\\
 		
 		[ \ T \ | \ Ts \ ] & \mbox{head of list term } T \in \overline{X},\\
 		
 		 &   \mbox{ where } Ts \mbox{ is a list}\\
 		
 		[ \ T_0 \ , \ \dots \ , \ T_n ] & \mbox{list with with } T_j \mbox{, } \forall j \in [0,n]\\

 		f(T_0, \dots, T_n) & f \mbox{ functor, } T_j \mbox{, } \forall j \in [0,n]\\
 		
 		p(T_0, \dots, T_n) & p \mbox{ predicate, } T_j \mbox{, } \forall j \in [0,n]
 	\end{array}
 \right.
 $
 \label{def:PrologTerm}
\end{definition}

A \index{symbol} symbol denotes some logical object, e.g. "\textit{me}", "\textit{Santa Clause}", the integer "$33$", or just some \index{variable!local} local variable.
In \index{Prolog} Prolog, a symbol starts with a lowercase letter followed by an arbitrary sequence of letters and digits.

In contrast to symbols, \index{variable!symbolic} a symbolic variable always starts with a capital letter or is the unique character \index{\_} "\_".
For example, the variables $X$ is assigned \index{term unification} (or \textit{unified}) some value, e.g. "$33$", afterwards $X$ might be used in (complex) terms or subgoals (see def.\ref{def:PrologRule}).
Whenever "\_" is used, then referencing it will not be an issue at all.
Thus, "\_" is used exclusively in cases where precisely one term is applied, and further use of that value is not needed -- as it usually is the case with \index{pattern matching} \textit{pattern matching} over terms.
The visibility of \index{scope visibility} variables is bound to a rule.

\index{linear list} Lists defined over binary operations "\texttt{,}" or "\texttt{|}" must have at least two components.
The check of whether a list is linear is performed by the \index{predicate} predicates itself, which accept terms.
Example lists are \texttt{[12|[]]} and \texttt{[1,2,[4|5]}.
So, the base type of any list is a struct.

A functor is a structural operator, which links terms and denotes a new complex value.
An object or struct may be (partially) symbolic.
For example, the list constructor functor \index{.-operator} "\textbf{.}" applied to head $H$ and \index{linear list} list $Hs$ works the same as $[H|Hs]$.
Another example is the successor \index{natural numbers} of a natural number $succ$.
Its arity is one and accepts either one \index{arithmetic!Church} \index{Church arithmetics} Church-term to $succ$, or which is $zero$ with \index{arity} arity 0 (cf. fig.\ref{code:NaturalNumbers}).
The latter case denotes a \index{constant function} constant (see obs.\ref{obs:SimplificationByGeneralisation}).

The predicate returns "\textit{yes/no}" depending on if some given terms coincide according to the predicate or not.
If the \index{predicate} predicate is defined for any terms combination, then the predicate is \index{totality} total (see next section).
For example, \index{assertion} the assertion $older(plato,aristotle)$ denotes \index{Plato} \index{Aristotle} the proposition "\textit{Plato is older than Aristotle}".

\begin{definition}[Horn Rule]

\index{Hoare-rule} A \textit{Prolog-rule} consists of \index{predicate!head} head $p$ and \index{predicate!body} body $q_0, q_1, \dots, q_n$ of \index{subgoal} subgoals $q_j$ and $j,n \in \mathbb{N}_0$, $j\le n$.
 Subgoals are evaluated consecutively with increasing $j \ge 0$.
 Head $p$ may contain an arbitrary number of \index{term} terms (term \index{term vector} vector), which may be used in the body.
 A rule with empty body, where $j=n=0$, is called a \index{fact} fact.
 \label{def:PrologRule}
\end{definition}

The syntax of some rule $pred$ in \index{EBNF} EBNF is defined by

\begin{center}
\begin{tabular}{l}
  \begin{minipage}[t]{7.3cm}
  \begin{grammar}
<head> ::= <ID> ‘(’ <term> \{ ‘,’ <term> \} ‘)’

<rel> ::= ‘=’ | ‘!=’

<call> ::= <ID> ‘(’ <term> \{ ‘,’ <term> \} ‘)’
  \end{grammar}
  \end{minipage}\\\\
  \begin{minipage}[t]{6.5cm}
  \begin{grammar}
 <goal> ::= <term> <rel> <term> | <call>

 <body> ::= \{ <goal> ‘.’ \} <goal>

 <pred> ::= <head> [ ‘:-’ <body> ] ‘.’
  \end{grammar}
  \end{minipage}
\end{tabular}
\end{center}

Here $ID$ is a symbol identifier -- but it is not a symbolic variable.
Binary operator "\texttt{=}" denotes unification, and "\texttt{!=}" impossibility of term unification.
"\textbf{.}" denotes the end of a predicate definition.
Symbolic variables are visible inside a rule definition only.
The separator "\textbf{:-}" defines $head$ to the left and $body$ of a rule to the right.

The example from fig.\ref{fig:PrologExample1} shall be considered for illustration.
Fig.\ref{fig:PrologExample1} a) contains two obvious facts from Ancient Greece: (1) Socrates is a human, and (2) Zeus is immortal.
In analogy, other facts might be defined: the designer's responsibility of a (closed) \index{knowledge base} knowledge base.
A fact by default is an indisputable \index{assertion} assertion from the considered context of discourse.
The third rule states: "\textit{every human is mortal}".
In more technical detail, this means: if some term $X$ has the predicate "\textit{human}", then term $X$ is unconditionally assigned the predicate "\textit{mortal}".
Unconditionally implies in an immediate sense the absence of further subgoals, except for the subgoal $human(X)$.

\begin{figure}[h]
 \begin{center}
  \begin{tabular}{ccc}
  \begin{minipage}[t]{4cm}
\begin{verbatim}
   human(socrates).
   noneternal(zeus).
   mortal(X):-human(X).
\end{verbatim}
  \end{minipage}\\
  (a)\\\\
  % \qquad &
  \begin{minipage}[t]{7cm}
\begin{verbatim}
   a2(0,M,Res):-Res is M+1.
   a2(N,0,Res):-N1 is N-1, a2(N1,1,Res).
   a2(N,M,Res):-N1 is N-1, M1 is M-1,
                    a2(N,M1,Res2),
                    a2(N1,Res2,Res).
\end{verbatim}
  \end{minipage}\\
  (b)
  \end{tabular}
 \end{center}
 \caption{Example on Prolog rules and facts}
 \label{fig:PrologExample1}
\end{figure}

Terms obeying def.\ref{def:PrologTerm} may be noted in \index{infix notation} infix-notation in the case of the "\texttt{rl}"-operator modifier. 
An even more generalised relation can be written by using predicate $p(T_0, \dots, T_n)$.
Next, fig.\ref{fig:PrologExample1} b) deals with \index{Ackermann function} \textit{Ackermann's function}, for example.
It demonstrates that any \index{recursion} \index{recursion!schema} recursion (e.g. left, right, primitive, \index{recursion!mutual} mutual) may be defined in Prolog, which shows full theoretical expressibility.
The operator \index{\texttt{is}} "\texttt{is}" is an arithmetic functor, a Carnap functor \cite{carnap68}, which denotes arithmetic term calculations.
They are not logical, after all.
Ackermann's function is defined and \index{totality} total.
However, due to memory limitation and physical device management, a calculation might be interrupted in practice.
Thus, a predicate made of a function may be helpful in termination check for inductively-defined \index{data structure} incoming data.
Examples might include \index{natural numbers} natural numbers and \index{linear list} lists.

\index{subgoal} Subgoals $goal$ as countable sub-constraints of a predicate can be defined as follows.

\begin{definition}[A Query in Prolog]
\label{def:QueryToProlog}
A \index{Prolog!query} \textit{query of subgoals in Prolog} is defined as a sequence of \index{term unification} unification \index{predicate!call} of defined (Horn) rule calls.
Symbolically bound \index{symbol} symbols \index{stack!pushing} are pushed to the local symbol environment and are unified for each \index{subgoal} subgoal call.
A call implies a matching \index{predicate} predicate is found with a fixed \index{arity} arity, and all passed terms obey def.\ref{def:PrologTerm}, though not necessarily fully bound, at the moment of the call.
 In case several predicates match, then w.l.o.g. the first matching predicate occurrence is selected first, the following matching predicates may be considered alternatives if the first matching subgoal fails (cf. fig.\ref{fig:BoxModelPredicateCall}).
The syntax is described by EBNF rule "\texttt{goal}", the semantic of a single subgoal can be described as a call to a Horn rule (see def.\ref{def:PrologRule}).
\end{definition}

Alternative subgoals might be excluded by \index{cut} cuts (see later).
Alternatives are considered in the order of appearance in the program.
Besides its use, the essential difference between a procedure and a predicate is the calling convention, the passed terms and the permanent and mandatory search for alternatives (see later, see obs.\ref{obs:ProofAsSearching}).

\index{introspection} CI in \index{Prolog} Prolog \cite{diaz12} is rather limited, but still, it allows the check and definition of user-specified \index{type} class types, so-called "kinds" for a given term, which may be: \index{\texttt{var}} \texttt{var}, \index{\texttt{atom}} \texttt{atom}, \index{\texttt{number}} \texttt{number}, \index{\texttt{compound}} \texttt{compound} and other less important \index{predicate!built-in} embedded \index{predicate} predicates.
The merge of common case terms into one is allowed, so predicates may considerably ease predicate definitions.
"\texttt{atom}" checks the property of a symbol, e.g. if \texttt{atom(a)} or \texttt{atom([])} are correct, but \texttt{atom([1])} or \texttt{atom([1,2])} would be not.
\index{\texttt{var}} "\texttt{var}" checks whether a given term is \index{variable!symbolic} an \index{variable!free} unbound symbolic variable.
Hence, \texttt{var(X)} is correct, but \texttt{X=1,var(X)} is not.
"\texttt{list}" checks whether a given term \index{weak typing} is a (weakly-typed) list (as it is by default in \index{Prolog} Prolog, cf.\cite{diaz12} with \cite{isocpp14}).
So, both, \index{\texttt{list}} \texttt{list([])} and \texttt{list([1,2,3])} are correct, where \texttt{list(a)} is not.
To define properties over a given list or for compound terms (check may be done using built-in predicate \index{\texttt{compound}} \texttt{compound}), in both cases, \index{=..} "\texttt{=..}" is used.
This operator splits a given term into the associated \index{functor} functor to the left-hand side and passed \index{term} terms as a list to the right-hand side.

Predicate calls are considered later.
However, \index{term unification} term unification by default is not done in \index{Prolog} Prolog due to a full evaluation of \index{expression} sub-expressions before processing them.
Thus, by default, \texttt{X=X} may be unified, where \texttt{X=f(g)} may not.
Depending on Prolog implementation to be used, recursion may or may not detect \index{recursion} recursion in rules.
That is why unification may occasionally only be performed on the top term level, alternatively not at all.
Nevertheless, for the example, \texttt{X=f(g(X,X))} unification may often lead to failure or termination in the best case, and in the worst case, may crash execution of \index{WAM} WAM (due to stack-overflow since version 1.3.0 "GNU Prolog" \cite{diaz12}).
Conservative estimates show that 95\% of all cases do not require \index{term unification} unification term checks.
However, these are the essential percentages for generalised assertions about term representations and transformations since often complex terms may contain symbolic variables that may (by accident) contain mutual recursion.
Often limits of given algorithms need to be checked.
Here \index{Ackermann function} Ackermann's predicate might be quite helpful in AT the incoming domain.
That is why now a simple algorithm is provided to exclude that critical 5\% entirely.
Those problems may indeed correspond to, e.g. object definition according to sec.\ref{sect:TheoryOfObjects}, see fig.\ref{CodeUnificationWithOccursCheck}.

\begin{figure}[h]
\begin{center}
\begin{minipage}[t]{12cm}
\begin{verbatim}
 unify_with_check(X,Y):-var(X),var(Y),X=Y.
 unify_with_check(X,Y):-var(X),nonvar(Y),
                          not_in(X,Y),X=Y.
 unify_with_check(X,Y):-nonvar(X),var(X),
                          not_in(Y,X),Y=X.
 unify_with_check(X,Y):-nonvar(X),nonvar(Y),
                          X=..[H|L1],Y=..[H|L2],
                          unify_list(L1,L2).
\end{verbatim}
\end{minipage}
\end{center}
 \caption{Example code on occurs-checks}
 \label{CodeUnificationWithOccursCheck}
\end{figure}

\index{term} Term structures and lists must be defined as in fig.\ref{CodeUnificationList}.
 
\begin{figure}[h]
\begin{center}
\begin{minipage}[t]{12cm}
\begin{verbatim}
 unify_list([],[]).
 unify_list([H1|L1],[H2|L2]):-
     unify_with_check(H1,H2),
     unify_list(L1,L2).

 not_in(X,Y):-var(Y),X\==Y.
 not_in(X,Y):-nonvar(Y),Y=..[_|L],not_in_list(X,L).
 
 not_in_list(X,[]).
 not_in_list(X,[H|L]):-
    not_in(X,H), not_in_list(X,L).
\end{verbatim}
\end{minipage}
\end{center}
 \caption{Example code on list unification}
 \label{CodeUnificationList}
\end{figure}

In sec.\ref{chapter:APs} and fig.\ref{fig:mapFunctionalExample}, an example is provided when a predicate is used as a passed term.
In analogy to \index{functor} functor analysis, a \index{predicate!call} predicate's call is also done within a given \index{stack!window} stack window's boundaries.
The difference is in the computation model and reference implementation of Prolog and its stack synchronisation (see next section) regarding \index{caller} caller and \index{callee} callee.

Next, let us consider the example from fig.\ref{fig:PrologExample1} b) with \index{subgoal} subgoals and their  \index{predicate!call} calls.
Given a program and subgoal "\texttt{?-a2(X,1,3)}" (see fig.\ref{fig:PrologExample2}), which is put in by interactive interpretation.
First, rule (1) is taken when unifying with $\sigma_1$.
The substitution $\sigma_1$ leads to contradiction \index{term unification} because of \texttt{Res=2}.
In fig.\ref{fig:PrologExample3}, a set of unification bindings is given for the considered example's derivation tree. 
The search is stopped for the given predicate because it failed, and therefore the next \index{rule!alternative} alternative is chosen (see obs.\ref{obs:ProofAsSearching}).
Selecting the third rule with unification $\sigma_2$, applying the second rule with $\sigma_5$, and finally applying the first rule with $\sigma_9$ finishes the search successfully.
Since there is no \index{cut} cut (see further), \index{rule!alternative} alternatives are checked, and finally, yet another solution for the given \index{subgoal} subgoal is found.
The path from the first subgoal over all-new intermediate subgoals till successful termination is called the path to a solution.
Every path to a solution in fig.\ref{fig:PrologExample3} is marked with a black square.
For the example given, there are two, and both solutions are identical.
Here it shall be noticed that even if paths to solutions vary, the search could be cut as soon as a first solution is found.
However, this is very example-specific and cannot be generalised.

\begin{figure}[h]
 \begin{center}
\scalebox{0.6}{
\begin{tabular}{l}
 \xymatrix{
  && \txt{? - a2(X,1,3)} \ar[dll]^{(1) \cdot \sigma_1} \ar[dddl]^{(3) \cdot \sigma_2} \ar[dddr]^{(3) \cdot \sigma_3} &&\\
  \txt{Res $\approx$ 2} \ar[d] &&&&\\
  \txt{\underline{fail}} &&&&\\
  & \txt{? - a2(X,0,Res2)} \ar[dl]^{(1) \cdot \sigma_4} \ar[d]^{(2) \cdot \sigma_5}  && \txt{? - a2(N1,Res2,3)} \ar[d]|{(1) \cdot\sigma_6} \ar[ddr]^{(2) \cdot \sigma_7, \ (3) \cdot \sigma_8}  &\\
  \txt{Res $\approx$ 1} \ar[d] & \txt{? - a2(N1,1,Res2)} \ar[d]^{(1) \cdot \sigma_9} \ar[dr]^{(2) \cdot \sigma_{10}} && \txt{Res $\approx$ 2,\\Res $\approx$ Res2+1} \ar[d] &\\
  \txt{\underline{fail}} & \txt{Res $\approx$ 2} \ar[d] & \txt{\underline{fail}} & \txt{\underline{X=1} $\blacksquare$} & \txt{\underline{fail}}\\
  & \txt{\underline{X=1} $\blacksquare$} &&& 
}
\end{tabular}}
 \end{center}
 \caption{\index{derivation tree} Derivation tree for the Ackermann-predicate}
 \label{fig:PrologExample2}
\end{figure}
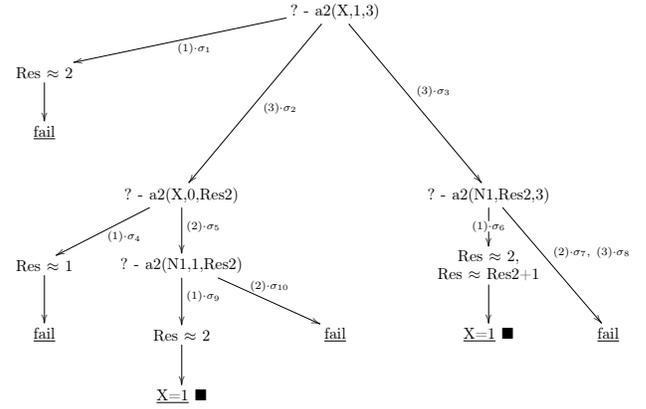

\begin{figure}[h]
 \begin{center}
\begin{tabular}{ll}
 $\sigma_1$ = & X $\approx$ 0, M $\approx$ 1, Res $\approx$ 3\\
 $\sigma_2$ = & N $\approx$ X, M $\approx$ 1, Res $\approx$ 3\\
 $\sigma_3$ = & Res $\approx$ 3, N $\approx$ X, M $\approx$ 1, N1 $\approx$ X-1, M1 $\approx$ 0\\
 $\sigma_4$ = & X $\approx$ 0, Res $\approx$ Res2, M $\approx$ 0\\
 $\sigma_5$ = & N $\approx$  X, M $\approx$ 0, Res $\approx$ Res2, N1 $\approx$ X-1\\
 $\sigma_6$ = & N1 $\approx$ 0, M $\approx$  Res2, Res $\approx$  3\\
 $\sigma_7$ = & N $\approx$  N1, Res2 $\approx$  0, Res $\approx$  3, N1 $\approx$  N1-1\\
 $\sigma_8$ = & N1 $\approx$ N, M $\approx$ Res2, Res $\approx$ 3, N1 $\approx$ N1-1\\
 $\sigma_9$ = & N1 $\approx$ 0, M $\approx$ 1, Res $\approx$ Res2\\
 $\sigma_{10}$ = & N $\approx$ N1, M $\approx$ 1, Res $\approx$ Res2, N1 $\approx$ N1-1
\end{tabular}
 \end{center}
 \caption{Unifications \index{derivation tree} for the proof tree from fig.\ref{fig:PrologExample2}}
 \label{fig:PrologExample3}
\end{figure}

Considering fig.\ref{fig:PrologExample3}, we notice that the evaluation \index{strategy} strategy may essentially change with \index{cut} cuts when \index{stratification} stratifying \index{subgoal} subgoals in a \index{paradigm! declarative} declarative paradigm.
For example, if cutting logical reasoning in fig.\ref{fig:PrologExample2} after subgoal "\texttt{?-a2(N1,1,Res2)}", we are confident the whole branch $(3) \cdot \sigma_3$ may be skipped.
For this, the \index{cut} cut needs to be defined.

\begin{definition}[Cutting Solutions]
\label{def:CuttingSolutions}
\index{cut} \textit{Cutting solutions} prevents \index{proof!search} a local search for alternatives at a given subgoal position in a given query according to def.\ref{def:QueryToProlog} after finding the \index{proof!solution} first solution till that subgoal position.
\end{definition}

\index{locality} Locality always applies to all \index{subgoal} subgoals inside a considered predicate until the \index{cut} cut operator "\texttt{!}".
However, cutting does not imply that any algorithm becomes more compelling or automatically better somehow since this depends on how the cut is applied.
Cut only provides a possibility depending on the given rules to skip further constraints if a solution found is obvious, and further constraints shall not be considered anymore.
It may be used to exclude \index{duplicate} duplicate solutions.
A cut is effectively a built-in predicate used as a \index{subgoal} subgoal and interprets the current calculation's state, namely the Prolog execution stack for non-deterministic solution branches.
Let us consider the cut's \index{predicate!semantics} semantics by the factorial \texttt{fact} example \index{factorial} from fig.\ref{CodeExampleFactorial}.

\begin{figure}[h]
\begin{center}
\begin{minipage}[t]{12cm}
\begin{verbatim}
fact(0,1).
fact(N,Res):-N1 is N-1,fact(N1,Res2),Res is N*Res2.

fact2(0,1):-!.
fact2(N,Res):-N1 is N-1,fact2(N1,Res2),Res is N*Res2. 
\end{verbatim}
\end{minipage}
\end{center}
 \caption{Example code for the factorials in Prolog}
 \label{CodeExampleFactorial}
\end{figure}

The first term represents an incoming \index{integers} integer, and the second \index{term} term represents the outgoing result term.
A predicate is a relationship between all \index{term!incoming} incoming and \index{term!outgoing} outgoing terms passed.
In the example, the order \textit{input-to-output} used is from left to right.
The predicate \texttt{fact2} almost equals the first predicate, except the base case contains a \index{cut} cut as the only \index{subgoal} subgoal.
Cut affects the first alternative of \texttt{fact2} is called only whenever \index{term!incoming} incoming and \index{term!outgoing} outgoing terms \index{term unification} unify for the \index{predicate!call} subgoal query according to def.\ref{def:QueryToProlog}, e.g. "\texttt{?-fact(5,R).}".
--- This happens only once for \index{inductive structure} recursion descent when the incoming number equals zero.
Whenever this situation happens, \index{cut} cut guarantees in all processed, then in  \index{subgoal} subgoals for no further alternatives will be searched.
On the \index{recursion!arise} recursive rise, alternatives are not considered.
If there were no matching subgoals on the recursive rise, then alternatives would be considered, except for cut alternatives, if any.
Therefore, \index{cut} the cut in \texttt{fact2} is essential.
It improves the understanding of the calculation of all its \index{subgoal} subgoals and its optimisation.
It is noted that the rules order heavily influences \index{soundness} soundness.
For example, heavy influence is reached by changing the ruling order or introducing a \index{cut} cut, see fig.\ref{CodeExampleFactorial2}.

\begin{figure}[h]
\begin{center}
\begin{minipage}[t]{12cm}
\begin{verbatim}
 wrong_fact2(N,Res):-!,N1 is N-1,
                      wrong_fact2(N1,Res2),Res is N*Res2. 
 wrong_fact2(0,1).
\end{verbatim}
\end{minipage}
\end{center}
 \caption{Counter-example code on factorial in Prolog}
 \label{CodeExampleFactorial2}
\end{figure}

In this example, \index{base case} the base case is unreachable because $0$ and $1$ both also \index{term unification} unify in the first alternative.
Cut leads here to ignoring the second alternative.
Next, \texttt{N1} is assigned -1 when calling "\texttt{?-wrong_fact2(0,1).}".
Thus, \index{recursion} left-recursion is present.
So, it does not \index{Halting-problem} terminate.
In this case, the \index{cut} cut is superfluous because \index{recursion!left} left-recursion does not allow to view any alternatives, although all logical \index{fact} facts and rules are defined.
The derivation tree contains one infinite branch since the corresponding calculation does not terminate.
In Prolog, this Halting-derived problem is better known as the \index{prioritisation issue} prioritisation issue.
Thus, for readability, rules must be defined in the order, s.t. base cases occur first.
It seems trivial, but it is not --- especially not when later Prolog is tuned (see sec.\ref{chapter:APs}).
So, particular and specialised rules, as well as facts, must be defined first.
Otherwise, they are suppressed by more generalised ones.

On the other side, non suppressed rules mean that rules corresponding to alternatives may alter according to their position.
Here, in addition to simplicity reasons, the preference shall always be given to simple and less constrained rules over more complicated rules.
In order to achieve this, it is mandatory, as much as possible, to differentiate incoming term vectors (see later on \index{heap!substraction} heap calculus in sec.\ref{chapter:stricter}).

It must be stated that \index{cut} cut is essentially \index{syntactic sugar} syntactic sugar.
It is sugar because it may always be \index{predicate} entirely rewritten (proof would be inductively over Prolog predicates), s.t. all possible continuations of the cut may be un-factored (contrary to rule factorisation).
Soundness may also be shown by excluding cut by introducing an additional argument to the predicate, which stores the temporary calculation result.
If the result is reachable, then there is a unique \index{base case} base case.
In all other cases, a strictly monotone raising \index{arising chain} chain exists.
The search may be generalised and expressed \index{fixpoint combinator} for a fixpoint in terms of $\mu$-expressions.
A full prove may be found in \cite{troelstra00}, also refer to discussion in sec.\ref{chapter:intro}.
Logics that allow cuts must implement and obey (CUT) formalisation (as a rule):

%\begin{center}
%\begin{tabular}{c}
$$\inference[(CUT)]{\Gamma \vdash \Phi,\Omega \quad\Phi,\Delta \vdash \Lambda}{\Gamma, \Delta \vdash \Omega,\Lambda}$$
%\end{tabular}
%\end{center}

\index{Prolog} Prolog upper \index{consequent} consequents of form $A \vdash B$ may be considered  $A$ is the subgoal of \index{rule body} a rule's body, and $B$ is the \index{predicate!head} predicate head, so the logical consequent.
Consequents are specific and abstract simultaneously and may cause unnecessary distraction and, therefore, are not mentioned later except explicitly needed.
\index{cut} (CUT) cuts \index{determinism} non-determination.
Cuts do not guarantee termination nor full determination whatsoever (see mentioned earlier example).
Alternatives may occur even with cuts, for instance, in following "\texttt{!}" subgoals.
Whenever an (intentional) cut does not threaten \index{duplicate} to drop the right solutions, the cut is \index{green cut} "\textit{green}".
Otherwise, that cut is \index{red cut} "\textit{red}".

Discussed solutions, constraints, rule modelling and subgoals will be applied later, particularly in sec.\ref{chapter:APs}.
\index{assertion negation} Negated assertions \index{assertion} shall be discussed in more detail next \index{expression negation} since they are an essential \index{literal} logical operation to (heap) predicates (cf. def.\ref{def:PrologTerm}).
Negation may be applied to predicates, and meaningful definitions may be found for relations, however symbolic variables (not propositions) may not have a meaningful negation in general.
Apart from the singleton case, it is trivial to note that negating a set (also infinite sets, but inductively-defined) is not decidable.
The method presented in sec.\ref{chapter:APs} is decidable, even if \index{inductive definition} inductive definitions are allowed due to the properties mentioned, mainly due to \textit{non-repetitiveness} and heap locality.
Naturally, a similar problem of undecidability encounters for predicates.
The situation is undefined for functors in general: 
if any function with any unspecified domain and any unspecified co-domain is negated -- what would this hypothetical function then denote?
-- This is more a rhetoric question, and it is obvious there is no obvious answer since it is undetermined.
That would change if the functor would act as a predicate and its co-domain would be \index{boolean set} boolean.
So, a meaningful \index{assertion negation} negation shall be found for \index{predicate} \index{Prolog} Prolog predicates.

A possible semantics of negation in Prolog is to \index{fail} fail since, in Prolog, a predicate searches for some solution and must find it in a negative sense only.
Prolog implements a \index{predicate!built-in} built-in predicate with arity 0 called \texttt{fail}, and from now on, this shall be used unless stated differently.

%%%%%%%%%%%%%%%%%%%%%%%%%%%%%%%%%%%%%%%%%%%%%%%%%%%%%%%%%%%%%%%%%%%%%%%%%%%%%%%%%%%%%%%%%%%%%%%%%%%%%%%%%%%%%%%%%%%%%%%

%\newpage
\subsection{Reasoning as Proof Search}
\label{sect:LogicalReasoningAsProof}

Prolog \index{Prolog} has successfully been used as a language processor to \index{language!formal} formal as well as \index{language!natural} natural languages \cite{pereira12}, \cite{matthews98}, \cite{kulik01}.
As shown later, some formal words must be generated and validated according to a given rule set.
Natural languages are characterised by mutations \cite{pereira12}, \cite{kallmeyer10} compared to formal languages due to cultural-sociological development.
Ambiguity is the primary factor for differences between both languages.
Even so, methods on natural language processing have been well studied over the last decades \cite{kallmeyer10}, \cite{pereira80}, which leads to intriguing results in \index{artiticial intelligence} Artificial Intelligence and \index{natural speech recognition} Natural Speech Recognition of human speech.
However, this work is focused on relatively easy \index{lexical analysis} lexical analyses.
This work's main objective is to find the most straightforward and adequate dynamic \index{dynamic memory} memory model and checks on it.
For a prototype's implementation but also for research of crucial topics, \index{Horn-rule} Horn rules will be used.
Warren's \index{WAM} proposition of the "\textit{use of predicate logic in order to solve logical tasks}" is put into life, even if he warns about his thoughts about Prolog's use may not be fully capable and tractable in real-life due to limitations of predicate logic's expressibility (cf.\cite{kowalski74}).
These limitations are due to \index{Turing-computability} Turing-computability.
A simple example demonstrates this complexity --- assume some function with \index{term!incoming} incoming and \index{term!outgoing} outgoing term vectors encoded as a predicate, as it may be done in Prolog.
Then, for a given incoming vector, the result may be hardly tractable, e.g. in the case of \index{one-time function} one-time functions for message encryption or functions that may be partially defined only, e.g. a 2D-mapping onto an ellipse.
The paper \cite{pereira80} researches similar \index{syntax analysis} parsing methods, but, unfortunately, the theoretical limitations of syntax analysis are not (entirely) considered. 
That is why parsing is heavily constrained, mainly due to recognisers of the \index{LL-analyser} LL(1) category or purely extended analysers based on it.
Unfortunately, it impairs expressibility in total.
\cite{haberland08-2} shows whenever generation and checks of data (in the contributions only terms are considered for common data) are of interest as two competing processes, then:

\begin{enumerate}
 \item \index{syntax analysis} Syntax analysis is bound by expressibility, first of all.
 If expressions were widened and artificial limitations were lifted in previously mentioned work, then  w.l.o.g. both syntax and \index{semantics} semantics of the whole checking process may aggressively be simplified.
 \item Simplification may also be achieved by a generalised \index{IR} IR.
 \item Particularly due to the \index{paradigm} \index{paradigm!functional} functional paradigm, a comparably bloated description is the result.
 In total, only five \index{predicate!higher-order} functionals are used, like \index{\texttt{concatMap}} \texttt{concatMap}, \index{\texttt{foldl}} \texttt{foldl}.
The probability is high that the check process semantics could have been formulated even closer in mathematical-logical terms similar to this work (see later).
\end{enumerate}

The papers \cite{haberland08-1} and \cite{haberland08-2} show a non-standard solution at first glance may lead to a friendly and straightforward but effective overall solution based on genuinely logical terms.
The reasons for that are diverse.
First, there is the representation of term transformations as rules and compactness.
Widefield analysis is made w.r.t other \index{language!functional} functional languages as \index{language!imperative} imperative languages worsen all characteristic metrics regarding quality and compactness.
Second, simplicity is due to the comparison made between a given program and its description, which may be caught by \index{regular expression} regular expressions in the most general form, even by some \index{context-sensitive} context-sensitive expressions in rare cases, as shown by examples on macro calls.
Qualitative analyses are taken out on each input operator and a quantitative comparison based on \index{metrics} software metrics.
The result is perhaps unexpected because nearly by all parameters, \index{Prolog} Prolog leads followed by \index{paradigm!functional} \index{paradigm} functionals \cite{haberland08-1}. 
Next, results are discussed in sec.\ref{sect:LanguageCompatibility}.

%%%%%%%%%%%%%%%%%%%%%%%%%%%%%%%%%%%%%%%%%%%%%%%

\begin{thesis}[Proof over Heaps is equivalent to Syntax Analysis]
\label{thes:ProvingEqualsParsing}
Logical programming on dynamic memory is the same as syntax analysis.
\end{thesis}

\begin{proof}
A heap verifier based on a \index{Prolog} Prolog-dialect is introduced, allowing to resolve the issues from sec.\ref{chapter:DynMemProblems} \index{automated verification} \index{theorem} by the theorems of this and the following sections, particularly on proof automation in \ref{chapter:APs} (cf. sec.\ref{chapter:intro}).
When talking about proofs over heaps, it is assumed that heaps are represented as Prolog terms, and those are automatically entailed \cite{haberland16-1} to compare equality (see following sections).
Naturally, core verification rules, e.g., subtraction, are not directly affected by automation since those still co-exist.
\end{proof}

A prerequisite to this thesis is that \index{syntax analysis} syntax analyses allow resolving issues related to \index{dynamic memory} dynamic memory proofs.
This section induces how to do it, which preconditions have to be obeyed and why.
Prolog does not only qualify as a theoretical apparatus, but it also qualifies as an instrument for solving a problem complex at once (see this and following sections).
Prolog \index{Prolog} is exceptional because only assertions derive true derived from a given Horn-rule set (see obs.\ref{obs:ModelOfComputationVerification}).
If a theorem was not derivable, then the reason may be in the rules (e.g. lack of \index{abstraction} abstraction), reasoning techniques, or the way \index{subgoal} subgoals are represented.
Naturally, there are practical bounds w.r.t. \index{syntax analysis} syntax analysis.

The idea to use a syntactical analyser as a \index{SMT-solver} solver of heaps (cf. sec.\ref{chapter:expression}) came into mind after several observations made (see sec.\ref{chapter:APs}):
\begin{enumerate}
 \item \index{Prolog rule}  Prolog rules are very compact and are unique among similar approaches, especially when it comes to logical tasks.
 \item A program in Prolog is a set of rules.
 Successful reasoning constructs, in fact, a \index{derivation tree} tree.
 The proof of some \index{theorem} theorem has the structure analogous to a program (cf. isomorphism  \index{isomorphic mapping} discussed earlier).
The same goes for queries that span a tree.
 \item Expressions may be \index{syntax} syntactically analysed if further needed.
The \index{locality} locality of pointers constraints memory models, namely its \index{expressibility} expressibility of corresponding generated \index{grammar!formal} grammars.
 So, some formal grammar here for a pointers theory may not supersede \index{grammar!context-free} \index{context-free} CF grammars in general.
 That is very useful to keep in mind.
 \item However, on one occurrence, a propagandist of \index{semantic web} Semantic Web attempted to advertise his and other persons contribution would be \textit{lingua franca} for at least the next 25 years.
The argument presented was used as self-defence that the Semantic Web solves almost all actual technical and scientific problems.
 As an argument of superiority, Semantic Web was described as higher than \index{abstraction} abstraction and, as he called it "legacy methods", as explicitly mentioned syntax recognition, he later added that it would have no future.
 As can be seen and understood today --- that presenter overemphasised, at least.
 Regardless of that, it can be found that syntax analysis may still be found useful today, not to forget about the base of any semantic web application since syntax still has to be recognised in a first instance.
 Apart from that, real problems, e.g. related to rule-based recognition, still emerge sometimes.
\end{enumerate}

Now it is up to note and compare the objectives of the discipline in more detail:

\begin{center}
\begin{tabular}{ll}
 1. Programming & $\longrightarrow$ find solution\\
 2. Logical Programming \index{logical programming} & $\longrightarrow$ find (logical) solution\\
 3. Syntax Analysis \index{syntax analysis} & $\longrightarrow$ find (syntactic) solution
\end{tabular}
\end{center}

When we try to type the mappings of 2 and 3, then characteristics emerge as shown in fig.\ref{CharacteristicsTyping23}.

\begin{figure}[h]
\begin{center}
\begin{tabular}{ll}
 \multicolumn{2}{c}{Logical Programming}\\
 \hline
 Input: & Prolog rules, query\\
 Output: & "\textit{yes/no}", \index{proof!tree} proof tree with all unifications\\\\
\end{tabular}

in comparison to that we have\\

\begin{center}
\begin{tabular}{ll}
 \multicolumn{2}{c}{Syntax Analysis}\\
 \hline
 Input: & formal grammar, input word\\
 Output: & "\textit{yes/no}", syntax tree (AST)\\
         &   (and symbols as attributes replaced)
\end{tabular}
\end{center}
\end{center}
 \caption{Characterisations of typings from points 2,3}
 \label{CharacteristicsTyping23}
\end{figure}

So, whenever the attempt in narrowing down input sets, the solution obtained by one approach may be replaced by a solution obtained by a different approach.

The approaches mentioned earlier for recognising \index{language!natural} natural languages differ from those for \index{language!formal} formal languages by some ambiguity and deviation level.
In the languages considered, a deviation is nearly absent.

\begin{lemma}[The Double-Semantics of Prolog Predicates]
\label{lem:DoubleSemanticsOfProlog}
A Prolog predicate has two semantics: (1) a procedural semantics based on calls, and (2) a relational semantics based on relational interpretations.
\end{lemma}

\begin{proof}
Since any Horn rule (see def.\ref{def:PrologRule}) can be represented almost identically by its syntax in Prolog, Prolog predicates are considered synonymous for Horn rules.
Full semantics and explanations of Prolog predicates can be found in \cite{sterling94}, \cite{warren83} (cf. def.\ref{def:QueryToProlog}).
\end{proof}

For the sake of this work's objective, only (2) will be reviewed and slightly modified later.
According to lem.\ref{lem:DoubleSemanticsOfProlog}, a (Prolog) semantics (1) based \index{predicate!call} on calls denotes a \index{paradigm!declarative} classic procedure (as in Pascal or C) with some execution \index{stack} stack. 
So, a \index{stack!window} stack frame is pushed to the top (see fig.\ref{fig:ExampleStack2}) when entering a predicate and removed on \index{GC} leaving.
In analogy to \index{parameters by reference} \textit{reference parameters}, \index{term unification} unified terms are used amongst stratified \index{stratification} predicate layers and passed to the calling predicates on leave.
Any violation of consecutive stack frames leads to term invalidation and even other severe errors.
In contrast to imperative PLs, a \index{symbol} symbolic variable may not just be passed immediately but may also be used as a genuine \index{term unification} logical part in a unifiable term.
So, any term can be passed in any direction across the stack.
This passing style allows any \index{object instance} class object to be used with high flexibility and compactness.
It allows the transformation to keep processes extraordinary complex and flexible if needed, but compact at the same time, which is remarkable.

The anonymous operator \index{\_} "\_", the built-in predicate \index{\texttt{concat}} "\texttt{concat}" as well as \index{predicate!higher-order} higher-order predicates in general, support practical \index{expressibility} expressibility on a high level, so recursion does not necessarily have to be defined through the predicates but allow flexible and shorter notations (as can be seen later and in the reference implementation of the prototype).
Expressibility increases whenever expressions get significantly simpler based on \index{functor} functors and shorter data models used.
The latter includes \index{linear list} linear lists.
The need to specify all list elements is dropped.
In contrast to \index{pattern matching} pattern matching, the remainder is not dropped but just unused, allowing simpler operations to be performed.
Let us remind \index{verification} verification of \index{heap} heaps is performed \index{stepwise verification} stepwise after each \index{program statement} program statement, so usually, only a small portion alters, often only the part referred.

\begin{figure}[h]
\begin{tabular}{c}
\begin{minipage}{5cm}
 \begin{center}
 \begin{tabular}{l}
 \begin{minipage}{5cm}
 \begin{center}
\xymatrix{
call \ar[r] & \ar[ld] \mbox{$\blacksquare$} \ar[r] & exit &\\
fail        &  & redo \ar[lu]
}
 \end{center}
 \end{minipage}
 \end{tabular}
 \end{center}
 \caption{Calling subgoals}
 \label{fig:BoxModelPredicateCall}
\end{minipage}\\\\\\
%%%
\begin{minipage}{8.3cm}
\begin{center}
\scalebox{0.7}{
\begin{tabular}{c}
\xymatrix @W=3pc @H=1pc @R=0pc @*[F-]{
A \save+<-4pc,1pc>*\hbox{\it Res4}
\ar[]
\restore\\ \cdots \\
{\bullet \save*{} \save+<-4pc,1pc>*\hbox{\it Res3}}
\ar[]
\restore
\ar`r[uu]+/r4pc/`[uu][uu]
\restore\\ \cdots \\
B \save+<-4pc,1pc>*\hbox{\it Res2}
\ar[]
\restore\\ \cdots \\
{\bullet \save*{}}
\ar`r[uu]+/r4pc/`[uu][uu]
\restore\\
{\bullet \save*{} \save+<-4pc,1pc>*\hbox{\it Res}}
\ar[]
\restore
\ar`r[ddd]+/r4pc/`[ddd][ddd]
\restore\\
\\\\ C\\
}
\end{tabular}}
\end{center}
    \caption{Example \index{stack} of stack windows and related subgoal references \index{predicate!call} with ingoing and outgoing \index{subgoal} terms}
    \label{fig:ExampleStack2}
\end{minipage}
\end{tabular}
\end{figure}

Prolog treats \index{predicate!call} predicate calls, as illustrated in fig.\ref{fig:BoxModelPredicateCall}, namely a black box with four traceable options.
A predicate call, so a subgoal, is generic and has these four options:
(a) $Call$ implies a \index{stack!window} stack window (frame) is created, which contains all passed terms.
(b) $Exit$ returns on success only when all subgoals succeed and the \index{predicate!head} predicate head fully unifies.
(c) $Fail$, when at least one \index{SAT-problem} subgoal fails or
(d) $redo$, a retry when (c) happens to one \index{subgoal} subgoal, then an alternative to the current predicate is searched, if it exists.
The danger with inaccurately formulated predicates is inefficiency due to a too broad range of alternatives.
Another extreme is when a \index{cut} cut is too restrictive, s.t. no (valid) result remains.
Both cases need to be avoided.
Predicates \texttt{call} and \texttt{fail} from fig.\ref{fig:BoxModelPredicateCall} indeed exist in \index{GNU} GNU \index{Prolog} Prolog \cite{diaz12}, although the fifth case, \index{exception} "\textit{exceptions}", is removed for the reasons discussed in sec.\ref{chapter:intro}.

Semantics (2) of predicate interpretations implies that any predicate is a generator of non-determinism.
A common \index{predicate} predicate (as discussed in sec.\ref{chapter:intro} and sec.\ref{chapter:APs}) may have a) different definitions for different combinations of \index{parameter!incoming} incoming and \index{parameter!outgoing} outgoing vectors, but also b) different derivable results for the same \index{subgoal} subgoal.

For example, if there was a \index{predicate} predicate "\texttt{mortal}" from fig.\ref{fig:PrologExample1} (a) and a few more \index{fact} facts are added, according to b) "\texttt{?-mortal(X)}" \index{derivable result} implies \texttt{X=socrates} ";" \texttt{X=plato}.
If for \index{Ackermann function} Ackermann's function "\texttt{?-acker(0,X,15)}" or "\texttt{?-acker(1,1,Res)}" was called, then depending on the concrete implementation of the WAM, the predicate is either defined (although it may take a very long time, or is decidable in theory, but not in practice for reasons of bound resources) or (partially) undefined as can be seen in fig.\ref{fig:PrologExample1} (b).
In the example for the given \index{parameter!incoming} outgoing term "\texttt{Res}" the predicate does not determine the incoming vector, even if the mathematical-logical definition certainly exists according to the given schema, which is in the rules factually without a change (see obs.\ref{obs:ProofAsSearching}, obs.\ref{obs:ModelOfComputationVerification}).
The discussion regarding this limitation continues later in this section, as the invertibility and representation of a \index{relation} relational model do.

\begin{observation}[Equality of Proof Elements]
\label{obs:EqualityOfProofElements}
The search for proof is a search for a subgoal solution.
Both search structures, the theorem-proof about a heap and the subgoal query shape a directed search tree.
\end{observation}

Proof elements are theorems, lemmas, (recursively enumerable) definitions, conclusions --- those need to be expressed in Prolog.
Theorems and lemmas are assertions, which are expressed by terms, and formulated as subgoals.
APs can express lemmas and recursively-enumerable definitions (see sec.\ref{chapter:APs}).
Conclusions are valid subgoal sequences with concrete symbols being unified, e.g. proven theorems.
\index{Prolog} Prolog is no general-purpose PL, apriori just by definition (cf. with \cite{warren83}, \cite{warren84}), even if it is nicely applicable in specific logical domains.
Not every kind of proof may unconditionally be solved with Prolog because of the fundamental limitations (see sec.\ref{chapter:intro}).
The critical problem is reducing a given problem, which might be better solved by narrowing down languages, as will be seen soon.
A subset of \index{language!logical} Prolog is used to solve complex \index{verification} verification problems based on \index{logical reasoning} logical reasoning based on WAM architectures.
Let us characterise subgoals and argument terms:

\begin{itemize}
 \item \index{predicate!built-in} 
 Built-in predicates, for instance, \index{\texttt{concat}} "\texttt{concat}", may be used anywhere in programs.
These (universal) predicates may be added to arbitrary \index{formal theory} formal theories written in \index{Prolog} Prolog, as successfully demonstrated by \cite{haberland08-2}, \cite{haberland07-2}.
 Reserved predicates and functors may be defined as \textit{multi-paradigmal}.
 \item Singleton and  \index{overloaded definition} overloaded definitions.
 Predicate definitions may appear as singleton or in a row but with an identical \index{arity} arity.
 This mechanism allows an introduction of new \index{Prolog theory} (Prolog) theories, including \index{multi-paradigmal} \textit{multi-paradigmal}.
 The effect might be \index{extensibility} extensibility, \index{variability} variability, flexibility and \index{polymorphism} polymorphism (the last two are on a different abstraction level compared to the first two criteria).
 Apart from that, overloaded predicates, as "\texttt{concat}", enable a flexible ordering of incoming and outgoing term substitutions.
 \item Operations over \index{heap} heaps are terms.
 Together with \index{heap instance} (class) objects, the relative position of a heap may be expressed by terms.
 Virtually, a heap graph can fully be represented by a heap term.
 Terms are arguments to subgoals, so are lemmas and theorem (sub-)proofs over heaps.
\end{itemize}

The universal call mechanism described in fig.\ref{fig:BoxModelPredicateCall} implies \index{backtracking} \textit{backtracking}.
Backtracking on predicates manifests a \index{recursion} \index{recursion!schema} (mutual) recursion scheme, and by doing so, exactly describes a \index{logical reasoning} proof tree (see CF imposed by Prolog rules in sec.\ref{chapter:APs}).
The limitations mentioned earlier must not be forgotten.
When a (sub-)search succeeds, then that (sub-)branch is considered proven.
When a branch \index{proof!refutal} is rejected, it is either not proven or just not provable in principle or may lead to an apparent contradiction.
Both variants, however, are sufficient to establish a refusal, as the universality of the formula must be shown.
If interpreting predicates \index{paradigm!declarative} declaratively, then predicates may easily be interpreted as procedures.

Consecutively, logical programming in Prolog may, in reality, be approximated to proving.
Even if a program's properties are specified and verified, the incoming programming is neither debugged nor executed.
As will be shown soon, Prolog is excellently qualified to be used as a \index{dynamic memory} heap verifier.
\index{expressibility} Expressibility is increased in this context by the heap IR, rule abstraction, inductively defined predicates, genuine logical \index{symbol} symbols and (quantified) formulae in terms of Horn-rules.
For example, in contrast to the lifted formulae in \index{predicate!first-order} first-order predicate logic, a compromise is sought, s.t. the heap graph construction is performed stepwise.
Here, the heap is described \index{bottom-up} bottom-up based on separate \index{heap!base} simplices, namely simple heaps.
A complete representation of logical expressions in formulae allows maximum flexibility towards heap expressibility and flexible strategies to be added and modified.

\begin{observation}[Deduction with Backtracking]
\label{obs:DeductionWithBacktracking}
The \index{backtracking} search with backtracking spans a formal proof tree in \index{Prolog} Prolog.
\end{observation}

From obs.\ref{obs:EqualityOfProofElements} follows, the (proof) structure spans \index{derivation tree} a tree in Prolog compatible with general heap searches (cf. sec.\ref{chapter:intro}).
If some open cycle occurs, this might be related to sub-proofs or symbols reapplied.
However, closed cycles shall not occur.
Otherwise, the proof does not terminate, and this would be an observed regression.
Regression is an indicator an applied rule did not succeed with a new calculation state, so there is no progress, so all consecutive regression states shall ideally be rejected.
Invalid states, in general, may only be implied from an unsound rule set.
If proof is found, then for a subgoal, all related results shall be bound accordingly (unbound \index{symbol} (symbolic) terms shall not further be considered since they result in questionable \index{predicate} predicate semantics and shall not be considered crucial for proof).

Prolog's search solely bases on rule \index{deduction} deduction (cf. sec.\ref{chapter:intro}).
Only what is derivable by Modus Ponens upon \index{fact} facts and rules is valid.
Other modes, as Modus Tollens, are invalid since not considered in ISO-Prolog.
In practice this means, that \index{falsification} predicate falsification may only be decidable in general if considered sets are finite (e.g. by introducing common symbols for class object definitions instead) and are defined.
Thus, self-contradictory cases are excluded.
Thus, proofs may become bloated, and some proofs cannot be taken out in classical deduction schemes (see the example of \index{Peirce's theorem} Peirce's third excluded in sec.\ref{chapter:intro}).
Predicates may be defined finitely, even if the incoming set is infinite.
Then \index{cut} cut-as-fail may represent a predicate cut or any other explicit definition upon it.

The motivation behind search with \index{backtracking} backtracking is that the search \index{continuation} continues with the first failing predicate relative to the parent's axis.
It works by guarding a minimal nominal value and implies permanent stack context switching.

In contrast to the method \index{return on fail} "\textit{return on fail}", the method "\textit{return on directed focusing}" may be considered generalised, s.t. in the most common cause, the focus is set to some extreme integer value.
That extreme value directs the choice of the next subgoal.
It leads to the insight that this kind of problem may be fully reduced by the strategy "\textit{branch and bound}" \cite{mitchell96}.
The generalisation of the extreme value is an (extreme) term, namely some bound term, s.t. it fits with the heads of the left-hand sides of rules by pattern-matching towards the most general unification.
Regardless of the bound looks, the stack-switching naturally involves backtracking already available to the Hoare calculus verification scheme (cf. fig.\ref{fig:PrologExample2}).

\begin{observation}[Stack-based Predicate Calls]
\label{obs:StackBasedCalls}
Prolog's underpinning \index{WAM} WAM is characterised on the top-level neatly by a declarative semantics, the Prolog-dialect itself.
For example, stacked terms may be referenced and partially unified among several stack windows -- in order to do that in an \index{language!imperative} imperative PL, this is much more complicated, especially since stack requires occasional state flips, as is the case with \index{backtracking} backtracking.
Prolog's syntax and semantics are relatively small compared to languages like C(++) to perform just what was stated.
However, backtracking and term processing are core language features that cannot merely be mimicked by C(++).
\end{observation}

\index{Hoare triple} Hoare triples in \index{Prolog} Prolog may be represented as a predicate with \index{arity} arity $3$, where the first two components are terms, and the third component is encoded as the last \index{term!incoming} outgoing predicate argument.
Of course, in general, any Prolog predicate may be implemented \index{invertible function} invertible, if whatever the concrete semantics might be and if desired.
However, we are not interested in that question yet (see sec.\ref{sect:LanguageCompatibility} and following).
The reasons for dependency will be emphasised in the following.\\

\textbf{Discussion about declarative approaches to the stack}\index{stack}

\begin{itemize}
 \item \textbf{Performance.} \index{performance} It is mainly about improving re-active systems' runtime behaviour (systems under another system's control).
 Although Warren \cite{warren83} emphasises performance is not the only factor, a slowly OS may be the effect of a slow algorithm (this is due to \underline{rule definitions}), e.g. \index{copying memory} frequently copying used memory regions.
 There are two reasons: First, the frequent use of \index{(sub-)expression} (sub-)expressions indicates a high \index{data dependency} data coupling, which may effectively be resolved by an algorithm redesign and thus rewriting the \underline{rule set}.
 Second,  \index{performance} performance improvement may come from outside (e.g. optimisations implemented \index{multi-paradigmal} multi-paradigmal) regardless of overhead costs due to an additional redirection because, in practice, it may be lower by a particular memory-managed technique like e.g. memory-mapped I/O \cite{love10}.
 Thus, the second problem is resolvable due to \index{variability} \underline{variability}.
 So, \index{predicate} a predicate may be replaced by an alternative behaviour but with the exact \index{specification} specification.
 Lingual \index{language extension} extensions and variability are both of permanent interest in this work, especially in sec.\ref{sect:LanguageCompatibility}.
 \item \textbf{Evaluation ordering.} \index{evaluation ordering} \index{infinite data structure} 
 Infinite data structures, like \index{\texttt{take}} "\texttt{take}" (see sec.\ref{chapter:intro}), shall be exempted from \index{operational semantics} operational semantics according to \cite{warren84} for core language features.
However, this does not necessarily exempt other semantics and tests as discussed in the introductory since this will be the essential application field.
 Since operational semantics has to fully accommodate any passed (possibly  \index{linear list} infinite) argument to the \index{stack} stack, hence the Prolog ABI would not terminate, thus would be incomplete.
 A partial push of \index{term} terms to \index{stack, pushing into} the \index{stack} stack would be thinkable.
 However, this would require severe modifications and a total redefinition of how \index{WAM} WAM deals with ABI.
 Since the disadvantages of massive expected changes would far outweigh any benefit gained,  this is not intended yet.
A slight modification of \index{WAM} WAM's original \index{operational semantics} operational semantics may seem possible, but changes would be too massive.
 It must be stated, term unification in Prolog implies naturally a difference check for the most common \index{term} terms that would hold even for \index{heap} heap terms -- and this expected (for further details, see sec.\ref{chapter:APs}).
 Only lazy-expression (term) evaluation outside-in (cf. sec.\ref{chapter:intro}) would w.l.o.g. allow the full support of \index{infinite data structure} infinite data structures natively.
 So the only critical phase remaining would be the invertibility of predicates.
 Thus, the only modifiable factor is \index{evaluation ordering} the evaluation order.
 In case it is not yet modifiable at all in practice, the order may need to be modified manually by rewriting the rule set.
 \item \textbf{Semantic context.} \index{predicate} Predicates may be overloaded or defined \index{multi-paradigmal} multi-paradigmal.
 \item \textbf{IR.} \index{IR} Term-IR \index{term} may evolve, so more and more not properly designed IR may become bloated.
 It may be required to redesign an IR, including more significant deletion and modifications.

 Thus, \index{variability} \underline{variability} and \index{extensibility} \underline{extensibility} (cf. API) are vital properties, which \index{Prolog} Prolog allows, which can easily be checked, e.g. against \cite{diaz12}.
If, for the sake of local optimisation, \index{performance} fast memory access and \index{B-tree} B-trees \cite{cormen09}, \cite{atallah98} shall be introduced instead, then this can quickly be indexed of built-in and hidden predicates, bridged by multi-paradigmal stubs.
\end{itemize}

The question of whether \index{stack} stack may be substituted by heap shall better be negated because past attempts suffered from tremendous performance bottlenecks and restrictions (see certain F\#-implementations, for example).
The reason was mainly \index{recursion!schema} recursion schemas used to be through the stack, suddenly had to be implemented in unorganised heap mimicking stack windows.
Even when some passionate fans may dispute it, that experiment was no real break-through and could be considered failed so far.

Furthermore, recent CPUs come almost always with optimisations that depend immensely on stack allocation and \index{stack!window} disposal, and breaking this rigid convention would require a good excuse since it affects several significant processor designs \cite{intel11}, \cite{raman12}.
\index{paradigm!declarative}
The declarative paradigm for the calculation state of heaps must be logical in order to tackle verification.\\

\textbf{Discussion on proof search using relations}\index{relation}\\

Predicates conjunct \index{term} terms, which consist of \index{symbol} symbols, variables and other terms.
It is not hard to notice that there exist $B^A$ different mappings between $A$ incoming and $B$ outgoing terms -- for the sake of simplicity and w.l.o.g., no mixes are allowed here. 
Several mappings may be bound by only one predicate.
Whether for a given input and output vectors, a predicate is decidable remains a question obviously (again without mixes).
This problem, however, can be reduced to Post's \index{(3-)SAT-problem} (3-)SATisfiability problem of some logical formula.
Naturally, from a practical view, it is the verifier's responsibility to write sound and complete \index{mapping} predicates which are overloaded, implementing a variety of singleton-predicate mappings so.
The more mappings a predicate includes, the higher is its functionality and potential reuse.
Overloaded predicates escalate the question regarding invertibility because non-termination of a particular aspect automatically leads to non-termination of a whole predicate, and so, to non-termination of a proof.
It will be considered later.

When considering some predicate as a \index{Cantor set} Cantor set, this is a discrete Cartesian product of domain (contains all incoming terms) and co-domain (contains all outgoing terms).
As domain and co-domain may contain more than one component, they may still be countable due to bi-directional \textit{diagonalisation}, hence be processed by recursive-schemes.
As domain and co-domain join, a predicate is just another indirection level.
Diagonalisation will do the job here too.
Hence recursive schemas are computable, and they are not another theoretic hinder from complexity' view here.
From a practical perspective, this is not less important but an essential universal consideration.
If we succeed in showing universal properties for all relations, then we might imitate it successfully.
As \index{algebra!relational} \index{relational algebra} \textit{Codd proved, relational algebra} was sound towards universal properties we are interested in most.
All that is required is to establish expressibility equality.
Most modern relational and object-relational query languages are known for their expressibility bounds.

\begin{thesis}[Expressibility of Relations in Prolog]
\label{the:ExpressibilityRelationsInProlog}
The expressibility of Prolog rules can be considered as pure \index{relation} \index{term unification} relational.
\end{thesis}

\begin{proof}
 A full proof can be found in the author's contribution \cite{haberland08-1}.
 The proof is direct and straightforward. All properties from Codd's algebra, e.g. \index{relation!projection} projection, must be replaced by equivalent \index{Prolog} Prolog rules one by one, avoiding naming clashes for two given relations, $R$ and $S$ (see fig.\ref{RelationalAlgebraOverProlog}).

\begin{figure}[h]
\begin{center}
\begin{tabular}{ll}
\index{relation!union} Union:&\\
$T = R \cup S$: & $t(x_1, \cdots, x_m):-r(x_1, \cdots, x_m).$\\
                & $t(y_1, \cdots, y_n):-s(y_1, \cdots, y_n).$\\
\index{relation!minus} Set minus:&\\
$T = R/S$: & $t(x_1, \cdots, x_n):-r(x_1, \cdots, x_n),$\\
  & $not\_s(x_1, \cdots, x_n).$\\
\index{relation!product} Cartesian product:&\\
$T = R\times S$: & $t(x_1, \cdots, x_m, y_1, \cdots, y_n):-$\\
 & $r(x_1, \cdots, x_m), s(y_1, \cdots, y_n).$\\
%\end{tabular}\\
%\begin{tabular}{ll}
\index{relation!projection} Projection: &\\
$T = \Pi_S(R)$:& $t(s_1, \cdots, s_n):-r(x_1, \cdots, x_m).$\\
               & $\forall s_i \in \{ x_1, \cdots, x_m \}$\\
\index{relation!selection} Selection: &\\
$T = \sigma_S(R)$: & $t(x_1, \cdots, x_n):-r(x_1, \cdots, x_n),$\\
    &  $s(x_1, \cdots, x_n).$\\
\index{relation!renaming} $\alpha$-conversion: &\\
$T = \rho_S(R)$: & $t(x_1, \cdots, x_n):-r(x_1, \cdots, x_n).$
\end{tabular}
\end{center}
 \caption{Relational model over Prolog predicates}
 \label{RelationalAlgebraOverProlog}
\end{figure}

In analogy, the transformation from Prolog back into Codd's algebra can be done.
Further steps, such as \index{canonisation} canonisation and \index{cut} cut (removal), may be required, though, but which can always be achieved -- as was mentioned earlier, e.g. in def.\ref{def:CuttingSolutions}.
\end{proof}

Here, \cite{laemmer02} shall be mentioned, where an optimal representation of incoming and outgoing data types is questioned.
His contribution is that data types in term notation may be more beneficial for most logical puzzles than any other paradigm, mostly not functional.
\cite{pitts96} researches fundamental properties \index{domain as relations} of domains as relations.
It proves a relation invariant \index{invariant} solely always exists.
Even so, the question regarding invariants remains exciting research and partially still open.
One good example deals with adaptations of \index{caching} caching \index{heap} heaps in \index{specification} specifications (see sec.\ref{chapter:stricter}).

After we just introduced relations and defined their basic properties, the following pleasant properties may be derived:

\begin{itemize}
 \item \textbf{Declarativeness.} When logical reasoning, arithmetic calculations may not bother too much, but objects are essential, so are \index{heap} heaps and their interconnections.
 This selection is not just a question of favour but is essential to the whole deductive system.
 Prolog is already a platform for term and knowledge representation as well as processing, and heaps are atoms or composed terms.
 Rule predicates are heaps.
 \index{alias} Aliases are \index{symbol} symbolic variables, which are used in different places.
All this is sufficient in order to define a theory on heaps.
 \item \textbf{Term-tree.} According to \index{theorem!Birkhoff} \textit{Birkhoff's theorem} (about term products) \cite{pierce91}, \cite{davis94} an abstract algebra implies that each term corresponds to a \index{term!tree} tree representation.
 The transformation, unfortunately, must not necessarily be unique if \index{term!normalisation} canonisation is skipped.
 The inverse transformation can be uniquely defined if canonised earlier.
These properties raise the need to either define \index{canonisation} tree canonisation or to trim.
So, e.g. \index{associativity} associativity helps to make tree canonisation implicit (see sec.\ref{chapter:stricter}).
 Initially, we do not want to restrict ourselves to trees but would like to have unconstrained heaps (cf. sec.\ref{chapter:expression}).
Even a heap graph consisting of a single vertex would be valid.
By a single vertex, we imply either a \index{heap!base} simple, a complex, e.g. \index{object instance} an object, vertex or a placeholder for a sub-graph.
 \item \textbf{Term Generation.}
 Research on this topic \cite{haberland08-1} shows Prolog is brilliantly for processing terms compared to other functional and imperative programming or transformation languages.
 That research includes a field study and so is based on numerous examples.
 Also, a quantitative analysis is included and based on \index{metrics} metrics, such as \index{expressibility} expressibility level and \index{intellectual level} intellectual language level and many more.
For various examples, nearly without exceptions, \index{Prolog} Prolog always supersedes functionals, occasionally even significantly.
 \item \textbf{Term validity.} 
The thematic research \cite{haberland08-2} shows that if both processes, generation and checking of terms, are narrowed down, then the critical limitation becomes the \index{language!validation} validation language's expressibility.
 W.l.o.g. the research considers a generalised regular language.
The validation process of semi-structured data may intuitively be understood to compare terms with \index{gaps} gaps filled on runtime.
 A term is a generalised IR and may widely be used, e.g. to \index{verification} verification.
 Operators describe not the program but the structure of the document to be generated.
 Statically, a loop not always may be bound.
 Hence, some loop containing vertex $a$ in a tree representing an \index{XML} XML-document only denotes $a^{*}$.
 %
%% This implies the condition of some loop may be represented and validated only by the most common kind.

 It implies the loop condition may be validated only against the most common assertion.

 Naturally, \index{automaton!finite} a finite automaton may not recognise $a^nb^n$, but that is dispensable.
 The major hinder is always the \index{expressibility} expressibility of the assertion language, namely its expressions.
Most curiously, among all approaches investigated \cite{haberland08-1}, relations come out to fit best in representing knowledge, e.g. for assertions, logical rules, \index{model transformation} transformations.
 For example, functionals must calculate each gap individually, whereas logical relations can be more elegantly and compact due to its term expressions and rule notation.
 In \cite{haberland08-1}, rule-based transformations mainly imply $\tau \rightarrow \sigma \rightarrow \tau$, where $\tau$ represents some term \index{IR} IR, and $\sigma$ denotes some environment with symbolic bindings.
 When interpreting this transformation of a \index{relation} relationship that denotes a \index{predicate} predicate, then the transformation may be extended as $\tau \rightarrow \sigma \rightarrow \tau \rightarrow \mathbb{B}$, where $\mathbb{B}$ is the boolean set.
 Now heap verification may be represented the same way as a family of transformations.
 Its main difference is the use of \index{abstract predicate} (abstracted) predicates and automated reasoning.
 On transforming needed terms, the heap models, however, will vary.
 In contrast to \index{regular expression} regular expressions and appropriate recognisers \index{recogniser} \index{partial derivatives} \cite{brzozowski64}, the \index{specification} specification schemas are mainly CF.
The prediction of the next element is a special operation in both methods.
Due to its context, both may significantly differ after all.
\end{itemize}

\begin{thesis}[Simplification using Term IR]
The use of (Prolog) terms eases the specification \index{heap verification} and verification of heaps.
 \label{thes:PrologMakesHeapSpecSimpler}
\end{thesis}

\begin{proof}
The solution to this thesis is presented in this and the following sections.
Simplification is mainly achieved by a direct heap term IR, which matches patterns straightforward in rules, and by a unified heap data structure to be used for both heap specification and verification (see sec.\ref{sect:LanguageCompatibility}, cor.\ref{cor:MinDiffGenCheckLanguages}).  
\end{proof}

Apart from \textit{terms}, other \index{IR} IRs may be used, for instance, \index{tetrads} tetrads, \index{polish-inverse notation} Polish-inverse notation, \index{triads} triads or others.
The advantage of terms here is the description of a heap state, mainly its simplicity and maximum \index{expressibility} expressibility.
An immediate term IR for \index{program statement} program statements may also be used since both statements and term are tree-like.
Hence, other IRs are not being considered (cf.\cite{opaleva05}, \cite{muchnick07}).
However, term-IR in \index{Prolog} Prolog has the advantage (sub-)terms do not require additional context definitions apriori, and its explanation of the link is straightforward and quite simple after all.
A heap term can "\textit{just be written}", where other IRs cannot, since they would require many additional conventions, constraints holding, be explicitly added.
Not only does this helps in lecturing students and rapid prototyping, but it also simplifies even the whole definition and transformation of terms (cf.\cite{baader98}, \cite{dodds08}).
Analogous implementation on a comparable practical level includes \index{LLVM-bitcode} \textit{LLVM-bitcode} \cite{llvm15}, GCC \index{GIMPLE} GIMPLE \cite{merrill03} and further annotated objects as IRs in terms of the ROSE \cite{rose17}, for instance.
In all these cases that use \index{tetrads} tetrads, the incoming program's temporary IR must first be transformed into an \index{AST} AST.

\index{syntax analysis} Syntax parsing in \index{LLVM} LLVM/CLANG is performed more flexible than in \index{GCC} \textit{GCC} due to a supplementary phase \index{CLANG} CLANG provides \cite{clang17}.
Even if the AST may theoretically be derived without \index{CLANG} CLANG, but in practice, that would fail due to the very \index{IR} bloated LLVM-IR as a human-readable format.
It should be remarked that the binary format is compact, but the manual restriction of \index{code transformation} IR transformations is heavily ineffective, and apart from that, is subject to permanent evolution.
Regardless of bloated namings and \index{duplicate} superficial-looking syntax definitions, variability and \index{extensibility} extensibility are raised in contrast to \index{GIMPLE} GIMPLE.

For the sake of closure and simplicity, and in contrast to \index{LLVM-bitcode} bitcode, the \index{Prolog} Prolog implementation excludes \index{unconditional branch} unconditional branches to arbitrary \index{program statement} program statements.
The implementation does not object to maximising program statements' coverage, even if expected future extensions may introduce new program statements.

Essential questions on the \index{extensibility} extensibility and \index{variability} variability of heap models include:
May a heap be easily modified, s.t. any stepwise heap modification may be mimicked? 
May new phases be added to the existing reasoning framework?

The chosen \index{IR} IR naturally corresponds to \index{Prolog} Prolog, for example, from \index{LLVM-bitcode} bitcode or \index{GIMPLE} GIMPLE.
Nevertheless, none of those is practically considered in this work.
\index{term} Prolog-terms represent \index{program statement} program statements (specification and rule set too), and as earlier in this section mentioned, may be represented as a tree.
Instruction \index{tetrades} tetrads may indeed be chosen (e.g. tetrads of some hypothetical \index{assembly} \index{abstract machine} abstract machine for the moment).
Regardless of the specific IR to be chosen, the \index{syntax analysis} \index{AST} AST is always the first IR directly derived from a given listing (see fig.\ref{fig:PhasesOfCodeGeneration}).
Choosing an IR different from AST or terms could provoke the drop of one or more phases depicted in fig.\ref{fig:PhasesOfCodeGeneration}.
However, this simple remark, was and often still is violated by several projects, including \index{GCC} GCC, \index{LLVM} LLVM, and by \index{jStar} jStar \cite{parkinson05-2}, where verification is performed on optimised \index{triads} triads just shortly before and during \index{code generation} code generation.
Many other projects implementing \index{static analysis} static analysis violate the remark, too, though for different reasons.
Although the remark seems simple, in consequence, more exact feature definitions are required, to which \index{alias} aliases must be added.
The following features seem essential in achieving the goal:
(1) having the possibility to extend and vary existing phases of \index{dynamic memory} heap analysis,
(2) to represent adequately, to avoid bloated semantic contexts as well as assets of external data containers.
The absence of possibility (2) is an indicator of a model with raising imprecision and complexity.
The simplicity of a model depends on a complete representation of all required data.
If needed, representation primarily affects \index{IR} term IRs of incoming programs and \index{semantic field} semantic fields, containing data such as depicted in fig.\ref{fig:PhasesOfCodeGeneration}.
Additional data is not mentioned in fig.\ref{fig:PhasesOfCodeGeneration}, even if it still may be needed in local phases.
Thus, they do not imply changes to the default\index{computability} computation model for \index{dynamic memory} heaps (see fig.\ref{fig:PipelineArchitecture}, see obs.\ref{obs:TypeCheckingPhases}).
Since this work's main objective is to perform verification, its critical IR is bound for static analyses.
Hence, terms ideally cover any hierarchical, tree-structured IRs --- also w.r.t. future extensibility, variability, and possible later implications with GIMPLE or bitcode.

\begin{figure}[h]
\begin{center}
\begin{tabular}{c}
\xymatrix{
  & \phi_1 \ar@{=>}[r] & \phi_1' \cdots \phi_1^{(n-1)} \ar@{=>}[r] & \phi_1^{(n)} \ar@{=>}[dl]\\
 \phi_0 \ar@{=>}[rr] && \phi_1 \ar@{=>}[lu] \ar@{=>}[rr] && \phi_2
}
\end{tabular}
\end{center}
 \caption{Pipeline \index{assembly line} architecture for the prover \index{static analysis} and syntax processors}
 \label{fig:PipelineArchitecture}
\end{figure}

The advantage of terms is their direct notation and direct access to the whole expression and sub-expressions.
If direct notation were absent, it would have to be composed individually, which would be a laborious task.
-- It is disadvantageous but essential what is mimicked by functionals.
Finally, the full description would have been bloated compared to direct term-IR, even if auxiliary functions and \index{semantic field} semantic fields would be introduced (e.g., \index{approximation} approximating limits).

\subsection{Language Compatibility}
\label{sect:LanguageCompatibility}

According to the introduction, \index{verification} verification is often considered over-engineered and too academic.
The reasons are diverse: too many (too rigid) conventions imposed on models and formal theories.
All descriptions involved are too bloated, cluttered and hard to understand before doing anything reasonable at all.
They are too many diverse formalisms to learn.
From thes.\ref{thes:PrologMakesHeapSpecSimpler} and the previous section follow generalisations regarding \index{language!specification} specification and \index{language!verification} verification languages for \index{dynamic memory} heaps (see sec.\ref{chapter:DynMemProblems}, see obs.\ref{obs:ComparisonDeclarativeParadigms}):

\begin{observation}[Comparison of Declarative Paradigms]
\label{obs:ComparisonDeclarativeParadigms}
The \index{paradigm!declarative} declarative paradigm, logical first of all,  fits heap verification better \index{heap} than \index{paradigm} the functional or imperative paradigms due to formula representation and \index{logical reasoning} logical reasoning.
\end{observation}

This observation may occur too insignificant on the first look.
It means that verification of \index{Hoare triple} Hoare triples is more comfortable in terms of a logical PL, the specification and verification of triples -- this may appear first evident.
However, only a few verifiers were built based on the logical paradigm, and even those often are too restricted or closed for variability or extensibility (see sec.\ref{chapter:intro}).

\index{logical predicate} Logical predicates describe \index{calculation state} calculation states.
Predicates first do not insist on \index{side-effect} side effects to express \index{calculation state} a state nor on any \index{variable!global} global state. 
\index{visibility scope} Visibility depends on predicate variable \index{variable!symbolic} symbols, not on external memory containers.
In an axiomatic system as Prolog is, rules may be altered individually.
By using (logical) programming, we allow solving problems related to verifying heaps.
From \cite{haberland15-2} follows functionals are less suitable for describing hierarchical elements --- as heap terms are (see def.\ref{def:HeapTermExtendedDefinition}, def.\ref{def:HeapAssertionDefinition}, cf. sec.\ref{sect:Implementation}).
Mixed input and output terms make the representation and query very compact and powerful simultaneously (cf. def.\ref{def:QueryToProlog}) and handy for proofs (cf. obs.\ref{obs:EqualityOfProofElements}).
Hence the Prolog-dialect is more appropriate than the functional paradigm for verification and specification over heap terms (cf. sec.\ref{chapter:APs}).

L\"{a}mmer \cite{laemmer02} suggests a proof framework focused on confluency smoothly tying up program components \cite{feijs02}.
Although it is not related to \index{dynamic memory} dynamic memory, it makes propositions of narrowing down further quasi-minimalistic \index{specification} specification and \index{logical reasoning} verification languages w.r.t. class objects (cf. sec.\ref{chapter:intro}).

\begin{observation}[Simplification of Inductive Proof Primitives by Generalisation]
\label{obs:SimplificationByGeneralisation}
If simple definitions and assertions may be formulated easily, proof searches are often generalised and straightforward.
\end{observation}

The easier some (recursively-enumerable) definition gets from the introductory examples, the less often corner cases exist, the less often exceptional cases need to be implemented and specified, the higher the reuse level.
It is common sense that the fewer paths a CFG spanned by lemmas, theorems, and inductive data structures have, the simpler it gets.
For example, replacing some existentially quantified \index{symbol} symbol in a formula with its polymorphic all-quantified symbol may look simple, but it may impose restrictions that indicate unreliable code, imprecise specification or both.
Restricting some symbol $a$, at least one predicate call and symbol binding is needed.
It implies all-quantified symbols may, by default, significantly increase expressibility.

Replacing atoms by symbols lifts the hosting formula.
Symbols are naturally in Prolog and may be bound by term \index{unification} unification, namely, the finite \index{normalisation} normalised resolution form \cite{diaz12}.
As discussed in the previous section, formal verification checks whether some given program (or the resulting structure of its calculation) matches a specified structure.
As also discussed previously, both processes are heavily influenced when being generalised.
Naturally, both processes are different from the beginning.
The most remarkable difference is due to \index{expressibility} expressibility.

\begin{corollary}[Minimisation of Differences between Generational and Checking Languages]
Simplified and lifted heap terms may be reused in expressions of the same language for describing
 \begin{enumerate}
   \item[(1)] \index{specification} specification,
   \item[(2)] \index{verification} verification and
   \item[(3)] the \index{language!input} incoming PL.
 \end{enumerate}
 \label{cor:MinDiffGenCheckLanguages}
\end{corollary}

\begin{proof}
%%%
Using one language for all three cases is extreme for the task of minimising the difference between languages.
This extreme position is to be used unless a plausible counter-example rejects it.
Whenever the difference between (1) and (3) is resolved, either by letting the \index{language!input} incoming PL be logical or letting the obtained \index{IR} IR be terms of an \index{input program} incoming (imperative) program, then \index{assertion} assertions are noted as \index{subgoal} subgoals (def.\ref{def:QueryToProlog}) and the incoming program as terms (def.\ref{def:HeapTermExtendedDefinition}).
Contribution \cite{haberland07-2} can be considered a broader field study on the proposed minimisation of a language and with validation being a regular-mighty typing issue (cf. sec.\ref{obs:TypeCheckingPhases}, sec.\ref{sect:HeapGraph}, fig.\ref{ExampleNFA1}).
Assertions about the program refer to the incoming program's term (see sec.\ref{chapter:APs}, sec.\ref{chapter:stricter}).
The other way round is not permitted.
This congruency is sufficient to resolve the difference between (1) and (2).
Formulae and other referenced inductive definitions may all be written as Prolog \index{fact} facts and  \index{Horn-rule} rules.
Those are sufficient as a specification in order to perform a verification.
\index{tactics} Tactics, specification rules, \index{auxiliary predicate} auxiliary predicates, and new user-defined predicates may be added to an empty or existing Prolog \index{formal theory} theory.

Tasks (2) and (3) rival: (2) sets how the \index{process} generation processes some \index{heap graph} heap graph, and specifically (3) attempts to verify based only on provided assertions.
The description must be heap graph focused, where the verification process has the "\textit{understanding}" (recognition) and analysis needed in that role.
In sec.\ref{chapter:APs}, the analysis bases on \index{syntax} syntactic definitions.
The precursors are the \index{term unification} unification of syntactic, semantic and pragmatic heap definitions and their checks.
%%%%%
\end{proof}

Lee's fair remark \cite{lee96} is that \index{higher-order logic} higher-order logics are of utmost importance to applicability.
As initially mentioned, \index{Prolog} Prolog supports it (cf. fig.\ref{fig:mapFunctionalExample}).
Rules may evolve, but not whilst \index{interpretation} runtime (see sec.\ref{chapter:intro}).
Hence, defined rules are accessible in both \index{process} processes (2) and (3).

\begin{observation}[Language Confluency upon Verification]
Given a PL $P$, specification language $S$, and verification language $V$.
$P$'s listings manipulate the dynamic memory.
The relationship, as illustrated in fig.\ref{MVCMetaPattern}, is inspired by Reenskaug's meta-pattern \textit{Model-View-Controller} (MVC) (see fig.\ref{MVCMetaPattern}).

\begin{figure}[h]
 \begin{center}
 \begin{tabular}{c}
\xymatrixrowsep{15pt}
\xymatrixcolsep{15pt}
 \xymatrix{
  Model \approx P & & View \approx S \ar[ll]\\
  & Controller \approx V \ar[ul] \ar[ur]
 }
 \end{tabular}
 \end{center}
 \caption{Roles in the Model-View-Controller meta-pattern applied to heap verification}
 \label{MVCMetaPattern}
\end{figure}
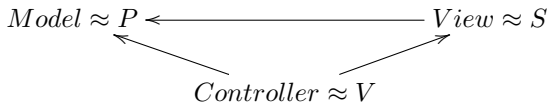
\end{observation}

The heap model is described by $S$ and processed by $V$ (when discussing MVC's components, we refer to particular instances, e.g. a specific program).
Some user may use $V$ to alter the state and observes results by $S$ (see def.\ref{def:SpecificationLanguage}).
However, in Reenskaug's classic pattern, $V$ immediately manipulates $P$, which varies here.
$P$ may have different graph representations of $S$ with the properties mentioned earlier.
As a side-effect, a narrowing down also leads to the interfaces that communicate in one language.

\subsection{Knowledge Representation}

This section considers core Prolog listings as a knowledge base, a well-founded and well-defined formal logical PL.
It is founded on a quantitative and qualitative comparison of terms, rules and subgoal queries.

Initially, rules may be represented in any \index{paradigm!declarative} declarative paradigm: 
\index{paradigm!functional} functional or logical.
For \index{XSL-T} XSLT and \index{Prolog} Prolog, a quantitative analysis was made \cite{haberland08-1} for \index{semi-structured data} semi-structured data (including terms, see the previous section), as well as a qualitative comparison of over 80 selected typical examples (cf. sec.\ref{sect:LogicalReasoningAsProof}, process signatures apply for heap verification).
The latter analysis bases on carefully assembled examples from monographies, tutorials and other relevant online resources available.
Based on H\r{a}lstead metrics (see fig.\ref{fig:HalsteadMetrics}), the comparison showed:

\begin{enumerate}
 \item \index{Prolog} Prolog leads in all, except one, examples on average of approximately slightly about 30\% compared to XSLT, which may be derived from the ratios  $N_T : N$ and $\eta_1 : \eta_2$.
 \item On average, the equivalent Prolog algorithm is 50\% shorter.
 Often it is even shorter.
 It bases on $N$, $\lambda$ and $\Delta_N = \| N_T - N \|$.
 \item Furthermore, functional implementations suffer from closure, so only built-in operators may be used.
 User-defined are not allowed by default but could be user-defined with significant efforts needed, but still tractable as in Kiselyov's SXSL-T implementation \cite{haberland08-1}.
 Prolog may introduce user-defined \index{term} terms,  \index{IR} IR and arbitrary user-defined operators and rules, although a rich list specification covers most spatial elements.
\end{enumerate}

The qualitative analysis explains why expressions and rules may be represented easier based on lingual and expressibility considerations and non-metric measures.
Mainly this is due to the genuinely logical representation of terms and \index{relation} relations.
In contrast to functional languages (having a compact \index{denotational semantics} denotation), logical languages are based on atoms, terms and prioritised rules (see. theo.\ref{the:ExpressibilityRelationsInProlog}).
As denotational semantics are \index{paradigm!declarative} declarative, a complete representation can be reduced to near to minimalistic expression.
Results are calculated on \index{parameter!incoming} input vectors, and no further constraints are needed for this.
Variables are not symbols.
Symbolic variables are variables with dedicated semantics.
In a \index{$\lambda$-term} $\lambda$-abstraction, they are only used as bound and one-way, namely a unique binding.
Their scope visibility is similar to those of program variables.
A general substitution of parameterised function as a replacement for relations does not solve the problem, as found in \cite{birkedal07}.
An \index{axiomatic semantics} axiomatic semantics should be preferred over denotational semantics if the logical inference and its logical units are in focus, as with heap verification.
Prolog listings are a particular way to define axiomatic semantics.
Logical units include \index{atomism} atoms, logical expressions, terms in general and relations over those, and predicates.
Knowledge representation, transformation and comparison succeed whenever the chosen paradigm fits, as for one case, it is demonstrated in \cite{haberland08-1}.

\begin{figure}
\begin{center}
\begin{tabular}{lcl}
$N$ & .. & program length \index{program length}\\
$N_T$ & .. & theoretical program length\\
$L$ & .. &  \index{intellectual language level} intellectual language level\\
$\lambda$ & .. & \index{language abstraction} language abstraction level\\
$\eta_1$ & .. & total amount of operators\\
$\eta_2$ & .. & total amount of operands
\end{tabular}\\[0.3cm]
\begin{tabular}{ll}
 $\eta = \eta_1 + \eta_2$ & $V = N ld(\eta)$\\
 $N_T = \eta_1 ld(\eta_1) + \eta_2 ld(\eta_2)$ & $\lambda = V L$
\end{tabular}
\end{center}
 \caption{Quantitative analysis using \index{metrics} H\r{a}lstead metrics \cite{halstead77}}
 \label{fig:HalsteadMetrics}
\end{figure}

Apart from these, there are further remarks to take into consideration on \index{dynamic memory} heap verification:

\begin{itemize}
  \item \textbf{Prolog Rules are compact} and may be called at any analysis phase.
  Goals and subgoals may be arbitrary.
  Non-termination, in general, is unavoidable and are not artificially bound.
  However, for the sake of efficiency, to the approach presented in sec.\ref{chapter:APs}, a harmless restriction may be introduced.
  \item \index{eager evaluation} \textbf{Eager evaluation} of terms is used. %
  Lazy evaluation is excluded.
For instance, this leads to recursive data that the \index{Ackermann function} Ackermann function (see fig.\ref{fig:PrologExample1}) cannot be evaluated until all passed arguments are thoroughly evaluated first.
  This restriction is powerless since programs may constantly be rewritten, s.t. an infinite data stream is replaced by \index{rule body} finite, except the infinite case itself, but that is excluded by practical means.
  Moreover, \index{term unification} term unification already leads to undefined term portions, so \index{symbol} symbolic variables are assigned.%
  There is no need to recalculate everything again in recursive definitions, but only those \index{term!sub-} subterms affected by the change are in the obedience of the principle from fig.\ref{fig:ExampleStack2}.
  \item The \textbf{base type} of verification is a \index{predicate!parameterisation} term.
  Next, the base type may be considered parameterised as in the \index{type} \index{$\lambda$-calculus} $\lambda$-calculus of third-order, so a type is of $\Lambda_{T_{\lambda3}}$, which allows \index{variable!quantified} quantified variables.
  A type built up from other types is a \index{dependent type} dependent type (see def.\ref{def:TypedLambda2ndOrder}).
  A composition is according to def.\ref{def:PrologTerm} using some \index{functor} functor.
  Thus, heap comparison can be formalised as term comparison.
  However, terms may theoretically in Prolog contain self-application \index{term!self-applicable} by referring to \index{symbol} symbolic variables.
  Self-application may be checked by the predicate \index{\texttt{unify_with_check}} \texttt{unify_with_check} as introduced at the beginning of this section.
  The higher the level of \index{$\lambda$-calculus} $\lambda$-calculus, the higher constraints impact, so the fewer \index{paradox} paradoxes matter.
  \index{functor} Functors may be used for modelling arbitrary \index{data structure} data structures, including lists and tree and \index{object instance} (class instantiated) objects.
  \item In case a \textbf{subgoal succeeds}, the set of all unifications is issued.
  For example: $$\texttt{?-H=pointsto(a,2),VC=pointsto(a,X),H=VC.}$$ issues the result set:
  \texttt{H = pointsto(a,2)}, 
  \texttt{VC = pointsto(a,2)} and \texttt{X = 2}. 
  \index{counter-example} \textbf{Counter-example generation} may be obtained for free from \index{term unification} term unification, e.g. by \index{\texttt{unify_with_check}} "\texttt{unify_with_check}".
  For improved traceability, it should be enriched by calls to \texttt{write} as a subgoal with all relevant \index{symbol} symbols \cite{sterling94}.
  The nearest term shall be evaluated top-down in order to locate the first differing \index{term} sub-expression.
  In general, ISO-Prolog allows several tracings, including built-in predicates for just that purpose \cite{diaz12}.
  If verification continues stepwise, then recording each step into a \index{DOT} \textit{DOT}-file is very recommended for understanding and troubleshooting purposes.
  \index{counter-example} Counter-examples have to reproduce at least one valid case that makes the verification fail at a specified position.
  \index{term unification} Term unification leads to the most common substitution.
Hence, a non-matching \index{functor} functor, \index{arity} arity or not unifiable variables are a perfect shot as counter-example --- this is so universal that no further explanation is needed except the information about the current position in the proof tree traceable via DOT, for instance.
 In general, a particular example reflects a term that makes the verification refute intentionally -- this may be when some atom does not unify with a functor, or \index{atom} atoms differ from functors.
  \item \textbf{Relief of symbolic restrictions} according to the heap definition from \cite{berdine05-2} on runtime (see sec.\ref{chapter:DynMemProblems}), including rule set analysis for deductive reasoning (\textit{(bi-)abductive}) \cite{calcagno09}, \cite{pottier08}, which may be better modelled (also cf. sec.\ref{chapter:stricter}, sec.\ref{chapter:APs}), is Horn-rules.
  Secondly, analysing rules and alternatives may succeed during a search without additional efforts.
  However, due to applicability, the amount of \index{rule!alternative} alternatives needs to be bound.
  If a rule is selected, then \index{abduction} \textit{abduction} may be mimicked by a \index{Horn-rule} Horn-rule of kind "\texttt{b:-a1,a2,...,an.}".
  If any of $\texttt{a}_j$ remains undetermined, so some remain symbolic, then arbitrary different $b$ may be chosen if any alternatives were available.
  That requires preceding terms, and \index{predicate!head} head arguments of \texttt{b} are not excluded.
  It must be noticed, in order to fit Prolog for heap verification, no further (auxiliary) predicates are required to interpret a term since it represents straightforwardly just itself.
  The \index{lexicographical ordering} lexicographical order on \index{pointer} pointer names contributes to normalisation.
The same goes for simple heaps regarding the left-hand location and calls to built-in predicates and lists, for instance, when dealing with \index{\texttt{concat}} "\texttt{concat}".
  \texttt{concat} and further one-pass scanning predicates are effective due to non-repetitiveness.
  Hence, numerous transforming and evaluating predicates may not be needed.
  W.l.o.g. reasoning rules are invertible, iff between all sub-terms of the precondition and the postcondition either an \index{isomorphic mapping} isomorphism exists, or \index{function invertibility} non-invertible built-in predicates are safe not to be called in the inverse case.
  In case isomorphism is violated, either the mapping from the incoming vector to the outcoming is extended, or the inverse case is extended (see \index{abstract interpretation} abstract interpretation in sec.\ref{chapter:intro}).
An extension is problematic in inverse mappings since \index{co-domain} co-domains thin the \index{domain} domain, so the mapping becomes non-continuous.
  \item \textbf{Rules overloading} \index{rules overloading} and prioritisation may be performed and enriched by different bodies.
Auxilary predicates may be written, e.g. in \index{Java} Java.
  Argument terms to those predicates may be \index{term!incoming} incoming, \index{term!outgoing} outgoing or a mixture of both.
  It is not sufficient to use \index{DCG} DCGs for this work's objectives (cf. sec.\ref{chapter:intro}) because the control flow and \index{tactics} tactics are subject to change.
  This work's objective also includes heavy heap \index{Prolog rule} term interpretation, as shown in the following sections.
  The best example for this will be discussed in due course in sec.\ref{chapter:APs}, referring to strategies of top-down or bottom-up \index{syntax analysis} syntactical analysers.
\end{itemize}
   
%%%%%%%%%%%%%%%%%%%%%%%%%%%%%%%%%%%%%%%%%%%%%%%%%%%%%%%%%%%%%%%%%%%%%%%%%%%%%%%%%%%%%%%%%%%%%%%%%%%%%%%%%%%%%%%%%%%%%%%

\subsection{Verification System Architecture}
\label{sect:ArchitectureVerificationSystem}

In \cite{haberland14-2} and \cite{haberland14-1}, a verification architecture is suggested for \index{dynamic memory} dynamic memory reasoning based on \index{Prolog} Prolog.
The architecture is illustrated in fig.\ref{fig:HeapVerificationArchitecture}.
The architecture follows the principles discussed in sec.\ref{chapter:intro}:
\begin{enumerate}
 \item[(1)] automation,
 \item[(2)] \index{openness} openness,
 \item[(3)] \index{extensibility} extensibility and
 \item[(4)] \index{plausibility} plausibility.
\end{enumerate}
(1) states that proof should be found \index{automated proof} automatically without further interaction.
In Prolog, a solution will be found depending on the given rules if there is no infinite recursion and all needed rules are provided.
Otherwise, the proof does not \index{Halting-problem} terminate or does quit prematurely.
Next, thanks to the approach from sec.\ref{chapter:APs}, which contains design and implementation details, automation is provided.
(2) means there shall be no artificial limitations introduced to terms nor rules.
On one side, \index{Prolog} Prolog is open for extensibility and variability so that new rules may be added and existing ones be modified.
On the other side, Prolog is closed to what is \index{deduction} derivable from the rules only.
(2) relates to the architecture and used memory models discussed in sec.\ref{chapter:expression} till sec.\ref{chapter:APs}.
(3) means a memory model is \index{extensibility} extensible and, if needed, variable too in its rules.
Modifications are always possible due to allowed Horn-rule \index{Horn-rule} redefinitions in Prolog.
The extensibility of terms and rules was discussed in full detail in the previous main section.
(4) means any \index{verification} verification step may be traced and checked for plausibility.
DOT-files' \index{DOT} generation visualises each step and serves as evidence in case of proof refutation right with a \index{counter-example} counter-example beside the last corresponding formula.
The possibility to debug \index{plausibility} verification immediately without any additional preparation and overhead tremendously simplifies debugging.
It is an essential contribution to Extreme Programming.
Essential parts include the rule set (acting as "programming listing"), the subgoal sequence to be investigated and a specified memory state, including the stack and heap.
No other conventions need to be obeyed.
No more massively overcomplicated set-up is needed.\\

In fig.\ref{fig:HeapVerificationArchitecture}, the incoming program is at the entry.
\index{C} C or any other Pascal-like imperative PL should be recognised in --- at least immediate IR injection works.
The program is turned into a Prolog term IR which is then together with assertions about the program turned into yet another Prolog \index{term} terms and rules.
First, the incoming listing is checked syntactically, then semantically, including \index{typing} type checking and further IR-related constraint checks.
Every term may at any time be visualised in the \index{DOT} DOT-format.
\index{assertion} Assertions may refer to \index{lemma} lemmas and \index{theorem} theorems, which might be noted immediately as Prolog term and, if needed, may be used then during \index{verification} verification.
During verification, several rules get activated, namely, those provided in Prolog theories and are loaded \index{formula interpretation} while interpreting those, e.g. with tuProlog's environment \cite{denti01}, \cite{denti05}.
Either external SMT-solvers or Prolog itself may be used as \index{formal theory} formal theory solvers.
Let us recapitulate that Prolog is widely prevalent among \index{constraint programming} constraint solvers \cite{sterling94}, \cite{bratko01}.

New phases to the heap verification may be inserted \index{dynamic memory} according to fig.\ref{fig:PipelineArchitecture}.
Minor and significant phases are bridged by handing over the \index{calculation state} calculation state explicitly.
This \index{CFG} data-driven approach (see fig.\ref{fig:HeapVerificationArchitecture}) is by design very closely incorporating with the \index{assembly line} assembly line architecture from \cite{kennedy02}, \cite{gcc15}, but differs in the visibility scope of heap variables (cf. sec.\ref{chapter:expression}).
However, the proposed conveyor may be used under the condition \index{alias} alias-analysis establishes individual \index{visibility scope} visibility scopes, see obs.\ref{obs:VariablesScope}.

\begin{figure}[t]
%\begin{center}
\scalebox{0.8}{
\begin{minipage}{12cm}
%\begin{tabular}{l}
\xymatrix{
  \fbox{\txt{\textit{C-like language}}}  & \ar@{=>}[dd] & \fbox{\txt{\textit{VC}}}\\
  \txt{\parbox{3.3cm}{\index{IR}\textbf{"IR"}}} & & \txt{\index{semantic analysis}Semantic\\ Analysis} \ar[ld]\\
  \ar[r]^<(0){\txt{\index{AST}AST}} & \fbox{\txt{\index{term}\textit{Terms}}} \ar[r] \ar@{=>}[dd] & \txt{\index{DOT}(DOT)} &\\
  \txt{\parbox{3.3cm}{\textbf{\index{proof!automation}"Automation"}}} &&\\
  \fbox{\textit{\txt{\index{rule normalisation}normalisation rules\\/SMT-solver}}} & \fbox{\txt{"yes/\\no"}} \ar[d] \ar[dr] & \fbox{\textit{\txt{subtraction\\ rules}}}\\
  \txt{\parbox{2.8cm}{\textbf{\index{code generation}(Code generation)}}} & \txt{\index{GC}Garbage Collection} & \txt{\index{alias analysis}\index{alias}Alias Analysis}
}
%\end{tabular}
\end{minipage}}
%\end{center}
 \caption{CFG on \index{dynamic memory} heap verification}
 \label{fig:HeapVerificationArchitecture}
\end{figure}

Furthermore, in \cite{haberland15-1} are proposed and discussed the required criteria for an \index{extensibility} extensible and variable heap verifier architecture.
The criteria include a minimal \index{IR} IR for statements of the incoming program (PCF-based \index{PCF} \cite{mitchell96}) with object extensions (cf sec.\ref{sect:TheoryOfObjects}).
The presented computation model, namely the \index{class-based calculus} class-based calculus, is \index{typing} typable, so properties from the second and third-order \index{$\lambda$-calculus} $\lambda$-calculus (see def.\ref{def:TypedLambda2ndOrder} and def.\ref{def:Lambda2ndOrderTermTypes}) apply.
However, as discussed earlier, now the hot-code update is excluded.
Furthermore, arguments are passed either arbitrarily,  \index{call by value} by-value or by-call.
In the latter case, an indicator is required, which is implemented by a reserved \index{functor} functor.
\index{extensibility} Extensibility applies to the \index{input program} incoming PL, \index{static analysis} static analysis phases, and rules (including pre-existing ones).

In \cite{haberland15-2}, \index{frame} frames are introduced and reviewed, as well as heap definitions and \index{heap!interpretation} interpretations.
For all considered concepts, the use of Prolog is suggested to address the spotted current restrictions effectively.
As heaps are term interpretations that may be bound not before in subgoals calls, in case of unification succeeds.
Otherwise, unification fails, and non-unified terms illustrate a valid counter-example.
It must be noticed that verification and its type signature are similar to the comparison process of template-instantiation and validation \cite{haberland08-2}.
The theoretic foundation behind is a \index{pattern matching} pattern-matcher, a finite \index{comparison} comparison tree-automaton, a tree graph matcher (cf.\cite{comon07}).
Although the initial model is similar, the comparison differs from \index{predicate!call} predicate calls, \index{abstract predicate} abstract predicates, \index{precondition} pre-and \index{postcondition} postconditions, and the heap graph model.
The previous analogy implies that if an assertion is a schema or type, and a given program builds up the heap graph stepwise for each program statement, and the cell's content is compatible (cf. fig.\ref{fig:HoareCalcVSTypeSystem}), then verification is as robust as type checking.

\begin{corollary}[Minimisation of Incoming Program]
 The result of the inhabitant check is some minimal \index{minimalistic program} \index{input program} incoming program, which alters \index{dynamic memory} dynamic memory.
\label{cor:PrologMinIncomProg}
\end{corollary}

\begin{proof}
%%%
First, it needs to be shown the claimed program always exists.
Considering fig.\ref{fig:HoareCalcVSTypeSystem}, the inhabitant problem here states: "\textit{Given a heap graph, which program will construct this graph (with a minimalistic number of steps)?}".
The heap graph can inductively be constructed by inserting incrementally (but bound by an upper limit of) edges for well-defined labelling.
Second, it needs to be shown that the program is minimal.
Cor.\ref{cor:PrologMinIncomProg} implies a (minimal) statement sequence may always be derived for a given pre-and a postcondition (corresponding to type).
The calculation steps are essential before comparing with arbitrary heap states.
However, the relation between \index{type} "\textit{type}" and "\textit{expression}" cannot always be unique.
So, the question arises: Which program is finally generated (cf. fig.\ref{fig:HoareCalcVSTypeSystem})?
The program is generated \index{incremental approximation} stepwise when needed according to the assertion provided.
The program is build up in a minimalistic manner because almost every \index{graph edge} edge of the heap graph corresponds with new program statements -- this is under the condition of a minimalistic program only.
So, a cycle in the heap graph does not necessarily have to correspond to a \index{program statement} cycle (a loop) in a compound statement -- it may be the case actually, e.g. it may correspond to two statements.
\index{infinite data structure} Infinitely many edges may be excluded from the graph.
Therefore only similar chains correspond to a cycle as a \index{program statement} program statement whose condition is determined by the \index{heap graph} heap graph.
Whilst construction by the \index{input program} incoming program \index{minimality} minimality means only a set of unrelated statements to the final graph are removed.
The correspondences are as follows: the "\texttt{new}"-operator adds heap vertices, vertex content is altered by assignment statements, which may link two heap vertices in case of pointers.

Only those statements remain that are essential for the final's heap graph.
Thus, the inhabitant check for a given type may generate a \index{minimalistic program} minimal program, which may be compared with a given program.
Program minimisation is (implicitly) subsumed by reasoning Hoare triples, and there it is not considered separately.
On a domain level of Prolog, program minimisation may indeed be subsumed by ongoing triple verification if predicates are invertible (see sec.\ref{obs:StackBasedCalls}).
%%%
\end{proof}

The proposed Prolog-based architecture resolves the following issues without any extra costs:

\begin{itemize}
 \item \textbf{Any incoming PL} \index{language!input} is allowed as long as its \index{IR} \index{term} term IR is sound.
 Soundness definitions and new PLs may be added.
 Rule sets may be extended and varied.
\index{input program} An incoming program may even be left empty, and \index{syntax analysis} syntax analysis is dropped totally if some program's IR and accompanying \index{verification} verifications \index{term} terms are manually inserted.

 The architecture illustrated in fig.\ref{fig:HeapVerificationArchitecture} excludes syntax and semantic mismatches because of the most flexible phases (cf. fig.\ref{fig:PipelineArchitecture}).
 \item \textbf{Modest specifications} allow avoiding \index{specification!full} fully-fledged specifications since only those \index{specification!modularity} modules are specified that later must be verified.
 Except for full specification and a very few exceptions \index{footprint} ("\textit{footprint}"), there were almost alternatives that would have avoided tremendous efforts and hard to read assertions in the past, such as partially Smallfoot in SL.
In contrast, Prolog's principle is straightforward:
if some \index{assertion} assertion is found correct, the proof succeeds and terminates unless non-termination within rules is encountered.
Otherwise, the verification continues until a contradiction is found or it succeeds finally.
 Verification is applied only to those modules that contain assertions.
Partial auxiliary predicates, such as \underline{$true$} (see sec.\ref{chapter:stricter}), shall avoid full specification and increase readability. 
\end{itemize}

\index{polymorphism} Polymorphism over classes is excluded \index{program statement} here for the sake of simplicity (see discussion in sec.\ref{chapter:intro}, \ref{chapter:expression} and followings).
The \index{dependency graph} dependency graph of given \index{Prolog} Prolog rules and \index{lemma} lemmas \index{assertion} are analysed by the Prolog interpreter whilst runtime.
It is also checked during \index{static analysis} syntax analysis and \index{compilation} compilation of abstract predicates (see sec.\ref{chapter:APs}).

\subsection{Class Objects}
\label{sect:ClassObjects}

In sec.\ref{sect:TheoryOfObjects}, both two prominent OCs were introduced, namely after \index{Abadi-Cardelli calculus} Abadi-Cardelli and after \index{Abadi-Leino calculus} Abadi-Leino.
Since simplicity of \index{specification} specifications has higher importance for this work's objective, the class-based calculi choice is evident (see sec.\ref{sect:TheoryOfObjects}) \cite{leino98}, \cite{reus02}.
For this work's sake, there is no need to prove the soundness of adapted objects separately, nor a full prove w.r.t. polymorphism, inheritance and encapsulation again --- it can be found, e.g. in \cite{cardelli96}, \cite{bruce02}.

Remarkably, pointers to object theories do not invalidate existing constructs.%
They extend variable references with new semantics that used to be missing.

Consequently, an \index{object instance} object is, first of all, an instance of some class.
For the sake of simplicity, \index{polymorphism} polymorphism is excluded since it does not directly affect core functionality (generic polymorphism, according to Cardelli's taxonomy, is not considered in the C(++)-dialects).
Even "\textit{ad-hoc}"-polymorphism is allowed only in boundaries of subclassing \cite{cardelli96-2}, namely by deriving \index{subclass} subclasses \cite{bruce02} and only for the sake of improved comfortability and modularity (see sec.\ref{sect:TheoryOfObjects}).
In Hoare calculi, polymorphism is tractable and expressible (cf.\ref{chapter:intro}, \cite{nanevski06}).
For the sake of this work, its only contribution does not affect expressibility.
Hence, an object represented as a \index{heap} heap is enclosed a mapped onto some \index{memory region} memory region without gaps containing just \index{object field} (object) fields.
Methods are not stored in \index{dynamic memory} dynamic memory since it is considered the code is static and since the hot-code update is prohibited in the first place.
Code updates on runtime are also prohibited because of the significant problems introduced regarding \index{completeness} completeness and \index{soundness} soundness (see discussion in sec.\ref{sect:TheoryOfObjects}).
As every object is of \index{type} some class type, its methods are entirely determined and typed before execution.
Inherited fields and methods are fully defined too.
W.l.o.g. it is agreed upon objects, their tuple pairs (all fields as \index{tuple} Cartesian product $\times$ its value domains) are always sorted in \index{lexicographical ordering} lexicographical order for simplicity and comparability.

Thus, object definition and transformation (down-cast in C++) into a \index{subclass} subclass may be implemented in two different manners: (1) each field is checked according to the relation \index{$>:$} "$>:$" (see def.\ref{def:TypeChecking}), here the fields in upper and lower subclass \index{object instance} instances diverge or (2) all fields are grouped according to some identifier \index{inheritance hierarchy} inheritance.
Thus, there must be an effective method to compare objects of (sub-)classes.
In order to demonstrate (2), we choose, for example: "\texttt{SubClass1 s1; SuperClass o1=(SuperClass1)s1;}".
Assume, \texttt{s1} contains \texttt{[o1,o2,o3]}, then the calculation of \texttt{o1}, where the upper class from \texttt{SuperClass1} inherits only the fields \texttt{o1} and \texttt{o2}, may be generated copying the first two fields, so by the initial contiguous range of \texttt{s1}.

Object fields \index{object instance} are noted as a list of \index{tuple} tuples of a kind: \textit{(identifier, value)}.
\index{symbol} Symbolic variables may express pointers to objects (including cycling). 
\texttt{A=object(A,A)} is prohibited and may be recognised by \texttt{unify_with_check}.\\
\texttt{A=object([(a,A),(b,A)])} is allowed.
The simplified notation \texttt{A=object((a,A),(b,A))} is assumed because the \index{functor} functor implemented by core \index{Prolog} Prolog already builds the list head containing an \texttt{object}.
Since the tuple contains precisely two elements, the composed term is always well-defined and unique w.r.t. the object web spanned.
Object fields are accessible by the \index{.-operator} "\textbf{.}"-operator and may be used within \index{program statement} program statements and \index{assertion} assertions.
Unsound access paths are recognised during \index{semantic analysis} semantic analysis.
Fields for some objects are (all) specified on the \index{abstract predicate} abstract predicates level.
These may also be defined, particularly particular (see sec.\ref{chapter:stricter} and following).\\\\
Conventions from sec.\ref{chapter:intro} rising precision in this section and conv.\ref{conv:RestrictedObjects} and conv.\ref{conv:HeapAlignment} are introduced to avoid paradoxes and narrow down towards a tractable UML extension.

%%%%%%%%%%%%%%%%%%%%%%%%%%%%%%%%%%%%%%%%%%%%%%%%%%%%%%%%%%%%%%%%%%%%%%%%%%%%%%%%%%%%%%%%%%

\section{Strengthening Heap Expressibility}
\label{chapter:stricter}
% (40p.)

In this section, \index{SL} \textit{SL} is introduced and analysed w.r.t. the problems generated by the \index{spatial operator} spatial conjunction operator.
One operator may be used to separate \index{heap graph} and link heaps, dependent on pointers and their \index{pointer content} content \cite{haberland16-3}.
Ambiguity allows a comfortable notation but induces several disadvantages.
A major flaw is \index{context-sensitive} \textbf{context-sensitivity}.
In this context, dependency, first of all, means the need to analyse a whole formula.
The extra steps needed may become cumbersome and end up in analysing all possible and unrelated subexpressions.
The dependent notation of the same \index{heap graph} heap graph needs to be probed first (a syntactically context-sensitive) to define an independent heap graph (a semantically CF).
Although this is not a \index{paradox} paradox, it, however, desires better.
Next, w.l.o.g. heap graphs are syntactically defined in a CF manner by excluding ambiguity and improving the \index{proof!automation} verification automation simultaneously from the automation perspective \index{syntax analysis} syntax analysis over \index{formula interpretation} \index{heap} heap formulae means overhead. 
In general, a heap relation may state two heaps are connected, or they are not.
Rewriting an overloaded (ambiguous) formula into a singular \index{non-ambiguous operator} (non-ambiguous) might not be trivial at all since all possible transitions have to be taken into consideration.
However, it is always decidable because any sound heap formula corresponds to some heap graph and vice versa.
An (un-)related heap to some other remains untouched regardless of which context or if it is a hierarchical heap.

Using a CF formula to describe a semantically context-independent model of a heap graph merges both concepts of what a heap is.
\index{object instance} Objects are considered complex SL units, which obey the same rules as simple units and (simple) pointers.
Later are discussed applications of formal properties towards a stricter memory model, the possibilities and limitations of the proposed model.

\subsection{Motivation}

Let us look at the following very generic syntactic definition of some \index{term!expression} term expression $E$ over integers in some classic \index{arithmetic} integer arithmetics as ambiguity problem:

\begin{center}
\begin{minipage}{5cm}
 \begin{grammar}
<E> ::= <k> | <E> ‘$\otimes$’ <E>
 \end{grammar}
\end{minipage}
\end{center}
 
It is relatively easy to conclude that the syntax in \index{EBNF} (E)BNF is inductively defined, and the base case denotes any arbitrary but fixed \index{integers} integer $k$.
Assume $\otimes$ is some \index{binary operator} binary operator that is \index{totality} total and defined over integer, like addition.
If we have the situation when for expressions $e_1,e_2,e_3$: $E_0 \cdots \otimes e_1 \otimes e_2 \otimes \cdots E_n$ and $n \in \mathbb{N}_0$ calculates $e_{1,2}$, where $E'_0 \cdots \otimes e_1 \otimes e_2 \otimes \cdots E'_n$ calculates $e'_{1,2}$, where $e_{1,2} \ne e'_{1,2}$, then either the rules of this calculation are unsound (possibly by design, but not later than that), or calculation is context-sensitive, so depends on $E_0$ and $E_n$, or $E'_0$ and $E'_n$.
It must be noticed that if $E_0 \equiv E'_0$ holds and so forth, then $E_n \equiv E'_n$.
So, the problem of differentiation matches with the problem of \index{soundness} (un-)proper calculation.
Beginning with the base case, namely, $e_{1,2} \ne e'_{1,2}$ where $E_0 \ne E'_0$, $E_n \ne E'_n$, it may be stated that both $E_0$ and $E_n$ are not empty at the same time.
Hence, dependency means that with $(e_1 \otimes e_2) \otimes e_3$, that $e_3$ contains syntactic information, which influences $e_1 \otimes e_2$.
Thus, for each $j$ multiplication, $\otimes_{\forall 0\le j}^{n} e_j$ means a complete analysis of all remaining factors in the worst case.
The \index{polynomial rank} polynomial’s rank $n \choose 2$ binds complexity.
So, which relation does this bound have towards heaps?

\begin{observation}[Operator Overloading]
\label{obs:OverloadedSpatialOp}
The operation $\star$ is \index{ambiguous operation} ambiguous (see def.\ref{def:ReynoldsHeapDefinition}).
It might be used for connecting and separating heaps and impacts the overall logical analysis of \index{heap} heaps.
\end{observation}

There is a structural analogy between the definition of $\star$ in SL and the previous definition of $\otimes$ (cf. def.\ref{def:ReynoldsHeapDefinition}).
Each heap needs to be thoroughly analysed during heap interpretations (it denotes $E$ in the previous example).

\begin{thesis}[Strengthening Expressibility of Ambiguous Spatial Operation]
\label{thes:StricterOpsExpressibility}
 If the \index{expressibility} expressibility of SL's spatial operator $\star$ is strengthened, semantic ambiguity can be diminished.
The exclusion of context-dependency eventually allows automation and eases heap analysis.
\end{thesis}

\begin{proof}
%%%%%%%%%%%%%%%%%%%%%%% THES.V.1
%Proofs will follow later with the next theses in this section.

The proof consists of two parts.
First, it needs to show that operation strengthening implies less semantic ambiguity.
Second, it has to be shown that the absence of context-sensitivity implies automation and enables a more accessible heap analysis.

The first part is common sense.
If some expression has many (semantic) meanings, then the context determines its denotation.
The fewer meanings (ideally one) some expression has, the less ambiguous it is.
In the case of a single meaning, the expression is non-ambiguous.
W.l.o.g. this naturally also applies to $\star$ (proofs on properties will follow with the following theses in this section).

The second implication follows from the new heap calculus to be established in this section, which is based primarily on the core properties of $\circ$ and $||$,  as well as the concept that objects are also heaps (see sec.\ref{sect:ClassObjectAsHeap}) and validated later practically by sec.\ref{sect:Implementation}.
The relevant properties of $\circ$ are def.\ref{def:HeapConjunctionDefinition}, theo.\ref{theo:GeneralizedHeapConjunctionTheorem}, and the theorems relevant for heap analysis lem.\ref{lem:HeapConjunctionMonoid}, theo.\ref{lem:HeapConjunctionGroupProperty}.
Due to remarks made in this section w.r.t. duality def.\ref{def:HeapDisjunction} and lem.\ref{lem:monoidOverDisj} suffice for $||$.
For the new heap calculus, it is essential to perform improvement calculations based on spatial distributivity, founded on theo.\ref{lem:DistributivityForConjDisj} and lem.\ref{lem:HeapInversionHomomorphism}.
%%%%%%%%%%%%%%%%%%%%%%%
\end{proof}

\begin{thesis}[Simplification by Differentiating Heaps]
\label{thes:SimplificationByDiffingHeaps}
When the syntactic and semantic unity is obeyed, an adequate representation simplifies the comparison and specification of given and expected heaps.
\textit{Diff(-erantiat)-ing heaps} may conduct the comparison.
\end{thesis}

\begin{proof}
%%%%%%%%%%%%%%%%%%%%%%% THES.V.2
Here it is necessary to show that heap differentiation leads to simplification.
If two and more heap (terms) may be compared, then the difference may be calculated effectively as a term.
The calculation is founded on term unification, the anonymous "\_" operator, heap normalisation in theo.\ref{lem:DistributivityForConjDisj}, inversion (conv.\ref{conv:EmptyHeapInversion}), thes.\ref{thes:StricterOpsExpressibility}, the properties established in sec.\ref{sect:APsProperties} and reference implemented as described in sec.\ref{sect:Implementation}.
That may be applied to a rule set to check completeness (cf. sec.\ref{chapter:stricter}, particularly def.\ref{def:IncompletePredicates}).
Missing and overlapping heaps in rules may be calculated.
Partial heaps simplify specifications due to a shorter notation (see sec.\ref{sect:PartialSpec}).
%%%%%%%%%%%%%%%%%%%%%%%
\end{proof}

This thesis implies that \index{context-free} CF-ness allows defining \index{formal theory} formal theories about equalities and inequalities of heaps.
The integration of \index{SMT-solver} SMT-solvers may further automate those theories.

\begin{thesis}[Complete Heap Rules by Incomplete Specifiers]
\label{thes:IncompletenessForCompleteness}
The notation of \index{heap!incomplete} incomplete heaps allow solving the problem of \index{completeness} complete heap rules.
\end{thesis}

\begin{proof}
%%%%%%%%%%%%%%%%%%%%%%% THES.V.3
This thesis is a specialisation of thes.\ref{thes:SimplificationByDiffingHeaps}.
%%%%%%%%%%%%%%%%%%%%%%%
\end{proof}

\begin{corollary}[Strengthening Modelling]
 \label{cor:StrengtheningOCL}
 \index{strengthening} The strengthening of the spatial operator does not violate the locality principle of (object) heaps.
 The strengthening is a proposition for \index{language!extension} extending the modelling language(s) \index{UML} \index{OCL} \textit{UML/OCL} by pointers.
\end{corollary}

\begin{proof}
The idea behind this is introducing strengthened spatial operators, which may either connect (cf. theo.\ref{theo:GeneralizedHeapConjunctionTheorem}, lem.\ref{lem:HeapConjunctionMonoid}) or separate heaps (refer to def.\ref{def:HeapDisjunction} and lem.\ref{lem:monoidOverDisj}).\\
 
It must be noticed that \textit{UML/OCL}'s expressibility is equivalent to third-order $\lambda$-calculus (cf. \cite{oclspec} for a direct proof), so it can be defined as in def.\ref{def:TypedLambda2ndOrder}.
Hence, it can be represented in Prolog terms described in sec.\ref{chapter:logical} and extended by APs described in sec.\ref{chapter:APs}.
\end{proof}

These observations made, and those follow from previous analyses sections and remarks.
The following might be implied:

\begin{enumerate}
 \item A simple model must be represented in simple terms.
 \cite{suzuki82} may serve here as a negative example.
It must be stated that Suzuki's model at the end is only manageable at first glance since the objective complexity is hidden in short but very hard to predict proper pointer semantics.
An incomplete set at first glance represents a complete rule set, which may commit considerably complex \index{pointer} pointer operations. 
 Even very trivial modifications may easily lead to unexpected behaviour.
 So, in a highly dynamical system, minimal edits shall not affect the overall behaviour, mostly not \textit{remote} spatial parts of a heap specification.
 \item  Various previously mentioned memory models, particularly their conventions, are not that important after all.
 It was shown that introducing new conventions does not necessarily extend but restrict expressibility.
 Potential new features from new or improved PLs too often are too specific, s.t. rules will often still need to be modified.
 Hence, variability and extensibility are that important (see sec.\ref{sect:TheoryOfObjects}, sec.\ref{chapter:logical}).
It is more important to describe and provide a given memory model adequately with constraints rather than covering full language features from the practical perspective.
 The \index{heap graph} \textit{heap graph} is agreed to be the primary domain of all possible heap denotations since it is the minimal consent that all possible heap verifiers may agree on unconditionally.
 This decision is motivated by a \index{utilitarian approach} \textit{utilitarian approach}.
 For details, please refer to sec.\ref{chapter:expression} for an epistemological definition of the term "\textit{heap}".
\end{enumerate}

The proposed strengthened model stands for a more effective verification, where the heap theory may be separated from Hoare logical non-heap-related rules.
Verification rules are represented as Horn-rules.
Their interpretation terminates due to Prolog's term interpretation and finite transducers \cite{haberland14-1} (see thes.\ref{thes:PrologMakesHeapSpecSimpler}, cor.\ref{cor:TranducersTerminate}).
In the next step, the strengthened operator $\circ$ replaces $\star$, s.t. \index{abstract predicate} APs may be automatically recognised during \index{syntax analysis} syntax analysis later (see sec.\ref{chapter:APs}).

\subsection{Ambiguity of Spatial Operators}

Hoare initially suggested using mathematic formulae as a most precise \index{language!specification} specification and verification language.
Later, \index{predicate logic} predicate logics was preferred, as well as multiple logics derived.
In the case of \index{heap!specification} heap specification and \index{verification} verification, practice showed that specialised logics might be applied even more successfully (see sec.\ref{sect:HoareCalc}, sec.\ref{sect:HeapModels}).
However, unbound formulae may be more comfortable when it comes to automation, as will be seen later.
One of these more acceptable conditions of automation may be non-ambiguous spatial operators.
Strengthening conditions in a formula leads to restrictions.

The problem with precision and \index{expressibility} expressibility is fundamental and is not only due to logics.
It covers different areas, starting with recognition, passing-over notation, obeying constraints and ending with expressions.
The reason for emerging ambiguity shall be researched (sec.\ref{chapter:expression}, sec.\ref{chapter:stricter}).
For example, it is no surprise when some given expression language differs from its declarative specification, and as such, there are apparent gaps in their semiotics (see sec.\ref{chapter:logical}).

A fundamental question concerns the equality of two heap representations.
It may significantly complicate equality if having different heap normal-forms or \index{canonisation} canonisations.

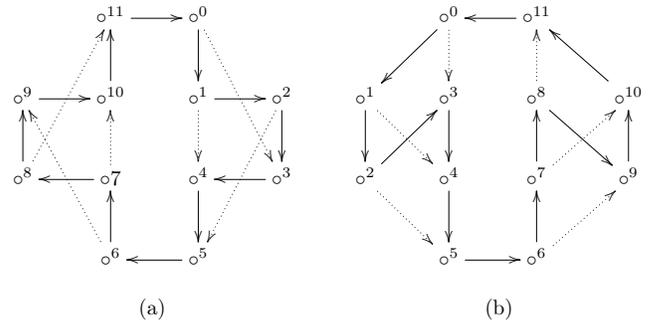
\begin{figure}[h]
\begin{center}
\scalebox{0.8}{
 \begin{tabular}{ccc}
 $\xymatrix{
  & \circ^{11} \ar[r] & \circ^{0} \ar[d] \ar@{..>}[ddr] &\\
  \circ^{9} \ar[r] & \circ^{10} \ar[u] & \circ^{1} \ar[r] \ar@{..>}[d] & \circ^{2} \ar[d] \ar@{..>}[ddl]\\
  \circ^{8} \ar[u] \ar@{..>}[uur] & \circ{7} \ar[l] \ar@{..>}[u] & \circ^{4} \ar[d] & \circ^{3} \ar[l]\\
  & \circ^{6} \ar[u] \ar@{..>}[uul] & \circ^{5} \ar[l]
 }$ &  \qquad &
 $\xymatrix{
  & \circ^{0} \ar@{..>}[d] \ar[dl] & \circ^{11} \ar[l] &\\
  \circ^{1} \ar@{..>}[dr] \ar[d] & \circ^{3} \ar[d] & \circ^{8} \ar[dr] \ar@{..>}[u] & \circ^{10} \ar[lu]\\
  \circ^{2} \ar@{..>}[dr] \ar[ru] & \circ^{4} \ar[d] & \circ^{7} \ar[u] \ar@{..>}[ur] & \circ^{9} \ar[u]\\
  & \circ^{5} \ar[r] & \circ^{6} \ar[u] \ar@{..>}[ur] &
 }$\\\\
  (a) & & (b)
 \end{tabular}}
\end{center}
 \caption{Isomorphic 3-regular heaps with objects}
 \label{fig:GraphIsomorphisms}
\end{figure}

So, the question regarding an \index{heap!isomorphism} isomorphism of two lifted heap graphs can be estimated as essential and further be discussed, referring to fig.\ref{fig:GraphIsomorphisms}.
In general, graph isomorphism is bound by exponential complexity, even for worse predictions.
In practice, exponential algorithms exist, which can be approximated by a third-order polynomial for a small number of graph vertices.
If the heap from fig.\ref{fig:GraphIsomorphisms} (a) contains only solid lines, later during analysis, there might be a connection between vertices $0$ and $5$ and both vertices were specified.
Deciding if isomorphism is present might become difficult.
However, complexity is reduced with types and labellings, which remain invariant.
The isomorphism complexity remains if pointers to a graph are renamed, and apart from that, this heap graph remains unchanged.
For fig.\ref{fig:GraphIsomorphisms}, this might be the case when applying the permutation (0 11)(1 8 2 10 3 9)(4 7)(5 6), starting with the graph from fig.\ref{fig:GraphIsomorphisms} b).
Notice the relatedness of permutations of heap locations w.r.t. arbitrary pointer rotation (see sec.\ref{sect:HeapModels}).
From a practical perspective, the isomorphism problem is only actual when it is necessary to check whether, in principle, a given abstract heap graph may be transformed into another abstract heap graph (cf. def.\ref{def:PredicateFolding}).
This transformation would be relevant only if naming may be parameters -- this is when names may change.

Sec.\ref{chapter:APs} considers APs in more detail.
For example, let us consider fig.\ref{ExampleConnectedHeapGraph}.

\begin{figure}[h]
\begin{center}
\scalebox{0.7}{
\begin{tabular}{c}
\xymatrix{
  && v_2 \ar[dr] &&& v_6 && v_7 \ar[dl] \\
  \ar[r] & v_0 \ar[rr] \ar[ur] && v_1 \ar[r] & v_3 \ar[ur] \ar[rr] && v_4 \ar[ul] \ar[rr] && v_5 \ar[ul]
}
\end{tabular}}
\end{center}
 \caption{Example of a connected heap graph}
 \label{ExampleConnectedHeapGraph}
\end{figure}
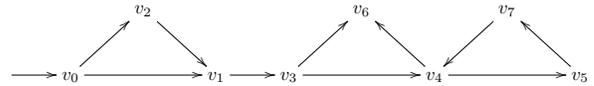

According to a minimal coupling, this graph may be bisected into sub-graphs among the bridge $v_1 \mapsto v_3$, see fig.\ref{ExampleSplitHeapGraph}.

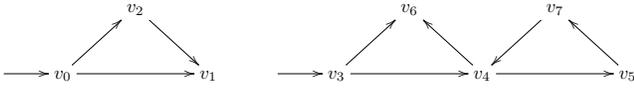
\begin{figure}[h]
\begin{center}
\scalebox{0.7}{
\begin{tabular}{lcr}
\xymatrix{
  && v_2 \ar[dr]\\
  \ar[r] & v_0 \ar[rr] \ar[ur] && v_1
} &  \quad &
\xymatrix{
  && v_6 && v_7 \ar[dl] \\
  \ar[r] & v_3 \ar[ur] \ar[rr] && v_4 \ar[ul] \ar[rr] && v_5 \ar[ul]
}
\end{tabular}}
\end{center}
 \caption{Example of a disconnected heap graph}
 \label{ExampleSplitHeapGraph}
\end{figure}

The graph may be described by separate predicates $\pi_0(v_0,v_1), \pi_1(v_3,v_4)$ and by predicate $\pi_2(v_4,v_5)$.
Alternatively, it may be abstracted to $\pi_0(v_0,v_1), \pi_{1,2}(v_3,v_5)$, where the graphs represented by predicates are connected, and visible vertices appear as arguments by the implicit definition $\pi_j$, see fig.\ref{ExampleSplitVariantsHeapGraphs}.

\begin{figure}[h]
\begin{center}
\begin{tabular}{c}
\xymatrix{
  \ar[r] & v_{0,1} \ar[r] & v_{3,4} \ar[r] & v_{4,5} \ar@/^1pc/[l]
}\\\\
alternatively as:\\\\
\xymatrix{
  \ar[r] & v_{0,1} \ar[r] & v_{3,5}
}
\end{tabular}
\end{center}
 \caption{Example of a possible heap dissection}
 \label{ExampleSplitVariantsHeapGraphs}
\end{figure}
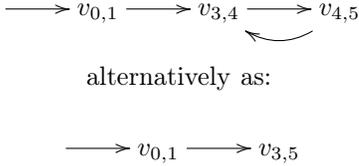

Unfolding, according to the definition of $\pi_j$, leads to the inverse.\\

Now, more generalised questions regarding an adequate heap representation may be asked:

\begin{enumerate}
 \item How to specify non-ambiguous formulae and perform the utmost determined verification?
 \item How to resolve isomorphism, locality and graph abstraction in simple means (w.r.t. heap representation)?
 \item How to restrict mandatory checks of objects in program statements (object representation)?
 \item How to solve equalities effectively of heaps if not all vertices (and edges) of the heap graph are defined (\textit{partial specification})?
 \item How to avoid redundant heap analyses (as well as redundant checks at each heap analysis phase; stepwise verification)?
\end{enumerate}

Sec.\ref{chapter:expression} introduced to \index{memory model!Reynolds} Reynolds' and \index{memory model!Burstall} Burstall's heap models.
This section reviews Reynolds's conclusions and effects and discusses derivable graphs and graphs generated after parameter modifications.
Properties and expressibility are observed.

\begin{definition}[Observable Heap by Reynolds]
 Heap is defined as $\bigcup_{A \subseteq Addr} A \mapsto Val^n$ with $n\ge 1$, where $A$ is some \index{address space} address space, and $Val$ denotes some \index{domain} values domain (e.g. \index{integers} integers or composite).
The observed heap operations behaviour (cf. theo.\ref{theo:ReynoldsHeapProperties}, def.\ref{def:HeapSatisfactionRelation}) may imply the following properties:
Given two heaps $H_1$ and $H_2$, then $H_1 \star H_2$, where $H_1$ denotes some heap assertion $H_1=(V_1,E_1)$ ($H_2=(V_2,E_2)$ in the analogy), where directed heap edges $E=V \times V$ are s.t. $\forall v_1 \in V_1$, $v_2 \in V_2$ with $v1 \ne v2$ and $V_1,V_2 \subseteq V$ hold in the following cases:

\begin{itemize}
 \item \textbf{First case} \index{heap!separation} (\textit{Separation}): $(v_1,v_2) \notin E_1$, and $(v_1,v_2) \notin E_2$.
 \item \textbf{Second case} \index{heap!merge} (\textit{Merge}): $\exists s \in V_1,\exists t \in V_2: (s,t) \in E_1$ or $(s,t) \in E_2$, then $H_1$ or $H_2$ contains $\star$-separated $s \mapsto t$.
\end{itemize}
\label{def:ReynoldsHeapDefinition}
\end{definition}

Both variables and \index{pointer} pointers are stored in the stack.
A pointer's content is stored in a \index{dynamic memory} heap.
The following domain equation holds according to \cite{berdine05}: $Stack=Values \cup Locals$.
As \index{program statement} program statements run, the internal memory state evolves and is checked against the explicit assertion annotations provided to the program.
Heap assertions are either \textit{true} or \textit{false}, depending on a specific heap to be interpreted.
The syntax of assertions is defined by def.\ref{def:HeapTermExtendedDefinition}.
From def.\ref{def:ReynoldsHeapDefinition} implies the \index{binary operator} binary operator $\star$ may be used with two intentions: either to specify two heaps do not overlap, or two heaps do overlap and share at least one common heap vertex.
Operator $\star$ is used as logical conjunction.
It is a spatial operator, which expresses the (relative) memory location and connectivity.
Conjuncting heap formulae define how two heaps are located and related to each other.
Heaps imply dynamic memory is occupied at a variable but fixed address.
The heap's types define their size.
If defining connectivity between two heaps as a \index{bipartite graph} bipartite graph \index{bigraph} (bigraph), then on its left-hand side, there are pointer locations, and on its right-hand side, the assigned content set.
A \index{bigraph!maximal matching} maximal matching is required to reduce the number of \index{conjunction} conjunctions for a given heap formula to describe a connected heap graph fully.
This approach is cumbersome in practice since there is no definite desire from the developer's side to describe a \textit{most compact heap graph} whatsoever (see sec.\ref{sect:HeapGraph}).
Nevertheless, if the obtained heap graph strongly differs from the expected, then unexpected \index{garbage} superfluous regions of the heap are indicators for an incorrectly working incoming program.
Other criteria may be critical too, for example, a good relationship between syntactical representations and a graph to locate vertices and edges and navigation through edges.
A compact does not benefit here, in fact, the opposite is the case, and it would be hard to read.
Concrete dynamic memory cells must be compared.
Hence, \textit{regular expressions} as the proposed comfortable variant are excluded.
Regular expressions, as mentioned, suffer from \index{locality} \textit{non-locality}: as soon as the heap graph changes locally at one place, then often the complete regular expression changes.
Desired behaviour is different.
Ideally, a single edge insertion has to change the formula hardly, and if, then only the affected part.
\index{heap!sub-} \index{heap} (Sub-)heaps not affected shall not impose changes to unrelated parts of the formula.

When comparing def.\ref{def:ReynoldsHeapDefinition} with the definitions from sec.\ref{chapter:expression}, quite considerable complexity may be implied, and definitions are ambiguous as long as only fragments are analysed.
In order to fully resolve dependencies, it is always required to analyse all \index{conjunct} conjuncts comprehensively.
The given definition may be considered an attempt to define a single heap because Reynolds did not define it as a singleton.
He defines a heap set instead.
Unfortunately, other authors (see sec.\ref{chapter:intro}) also only define a heap in fact as plural, although many author's informal remarks and sympathies tend to support more a singular definition even if it is plural.
It must be considered heap definition in def.\ref{def:ReynoldsHeapDefinition} as derived from observations is a direct singular definition and does not contradict its origin plural definition.
That is why the motivation sparks off  to \index{strict operator} strengthen $\star$ and replace it with a direct operation.
Reynolds' $\star$ is a classic example of \index{structuralism} \textit{structuralism}.
As soon as the new two operators are defined ($\circ$, $||$), properties and equalities are researched, and a new \index{term!algebra} \textit{term algebra} is defined.
It will lead to a new quality level regarding progressive \index{confluency} confluency since the heap calculus to be obtained may be decided on enclosed terms.
\index{term!algebra} Term algebras define new \index{formal theory} formal theories over \index{heap} heaps, which may immediately be used as Horn-rules in Prolog (see thes.\ref{thes:PrologMakesHeapSpecSimpler}).

%%%%%%%%%%%%%%%%%%%%%%%%%%%%%%%%%%%%%%%%%%%%%%%%%%%%%%%%%%%%%%%%%%%%%%%
\subsection{Strengthening Spatial Operators}
\label{sect:ConjunctionAndDisjunction}

Due to ambiguity, $\star$ may be used for heap conjunction or disjunction.
Apart from strict differences (thes.\ref{thes:StricterOpsExpressibility}) in different \index{SL} SL implementations, $\star$ is often used as another logical conjunction.
The solution to this problem is presented in this section.
Heap conjunction is formally introduced as well as its properties.
Afterwards, disjunction is introduced.
W.l.o.g. expressibility is not restricted by the new operations (cf. theo.\ref{theo:ReynoldsHeapProperties}).

\begin{definition}[Heap Conjunction]
\textit{Heap conjunction} $H \circ \alpha \mapsto \beta$ is defined as heap graph, where $G=(V,E)$ is the corresponding heap graph for given heap $H$, and where $\alpha \mapsto \beta$ denotes a base (simple) heap:

$\left\{
\begin{array}{lll}
  (V \cup \{\alpha,\beta\} \cup \beta', && \mbox{if } isFreeIn(\alpha,H) \\
   E \cup \{(\alpha,\beta)\} \cup \{ (\beta,b) | b \in \beta'\})  && \mbox{if } H=\underline{emp}\\
								  && \ \ \ (V=E=\emptyset)\\
  \mbox{\underline{false}}  && \mbox{otherwise}
\end{array}
\right.
$
\label{def:HeapConjunctionDefinition}
\end{definition}

Here $\beta' = vertices(\beta) \subseteq V$ defines all the heap graph vertices, which are pointing from $\beta$.
In case $\beta$ is an \index{object instance} object, then also all \index{object!field} object fields are considered that are pointers.
As $\alpha$ may point to just one unique heap graph vertex (e.g. \index{path accessor} \textit{path accessor} to some object), no more than one matching vertex exists in $isFreeIn$ for some given heap $H$.
The general assumption of the definition above is that a \index{stepwise graph construction} stepwise heap graph construction using conjunction always exists. There always exists one matching vertex.
Otherwise, two heaps may not be joined.

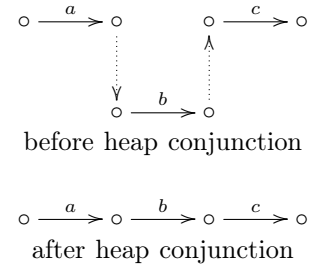
\begin{figure}[h]
\begin{center}
\begin{tabular}{c}
 $\xymatrix{
    \circ \ar[r]^a & \circ \ar@{..>}[d] & \circ \ar[r]^c    & \circ \\
                   & \circ \ar[r]^b     & \circ \ar@{..>}[u]
 }$\\
 before heap conjunction\\\\
 $\xymatrix{
    \circ \ar[r]^a & \circ \ar[r]^b & \circ \ar[r]^c    & \circ
 }$\\
 after heap conjunction
\end{tabular}
 \caption{Heap graph before and after conjunction}
 \label{fig:HeapGraphConjBeforeAfter}
 \end{center}
\end{figure}

For example, three pairs of \index{pointer} pointers need to be joined (pointer points to some content) $a,b,c$ (see fig.\ref{fig:HeapGraphConjBeforeAfter}).
First, heap $a$ has to be expressed, which may either be arbitrary or $\underline{emp} \circ a$.
Only if $a$ exists, $a$ points to some content that equals the initial heap $b$ and is not connected (for the moment being, we consider this to be the initial state), only then both heaps are connectible.
If assuming the opposite, so when in the heap graph there are two same content cells, then by definition, this is not possible, since only one unique content (see later) is allowed though with an arbitrary number of aliases.
Assuming there are the exact contents, but not the same as mentioned, then the accessor path must vary.
Otherwise, a contradiction is found.
In both cases, the possibility of incorrect conjunction is excluded.
So, the result is a connected heap $a \circ b$.
Now conjunction may continue under just mentioned conditions for $c$.
In case of success, we get the heap graph as illustrated in fig.\ref{fig:HeapGraphConjBeforeAfter}.
As we are interested in the conjunction of arbitrary graphs, e.g. binary trees, we allow conjunction in any parts of the connected graph.
For example, $a \mapsto 5$ is valid conjunction.
However, $a \mapsto 5 \circ b \mapsto 5$ is not.
Thus, we may express \index{alias} aliases, for example, heap graph $\xymatrix{ x \ \circ \ar[r] & \circ^z &  \ar[l] \circ \ y }$ as heap term expression  $x \mapsto z \circ y \mapsto z$.

Heaps may be connected in different ways when a graph vertex is an object.
For example, it may be agreed upon object assignment only changes a pointer or just one single qualified field.
The assignment may also be subject to different terms.
For example, when all fields are simultaneously assigned either to some input vector (e.g. array or string with specified delimiters), only a single field is \index{field!assignment} assigned, and all others are dropped (see sec.\ref{chapter:expression}).
For simplicity, but w.l.o.g., for a C(++)-dialect further, only the assignment of kind one-to-one is considered (see sec.\ref{sect:ClassObjectAsHeap}).\\
\textbf{Remark:}
Let $\varPhi_0$ be some heap, then $\varPhi = \varPhi_0 \circ a_0 \mapsto b_0 \Leftrightarrow \exists (a_m \mapsto b_m) \in \varPhi_0 \wedge (a_m=a_0, b_m\neq b_0 \vee a_m \neq a_0, b_m=b_0)$ follows.

\begin{theorem}[Conjunction of Generalised Heaps]
If given two heaps $H_1$ and $H_2$, the conjunction $\circ$ links those together to $H_1 \circ H_2$, if at least one common vertex exists in both \index{heap graph} heap graph representations.
By default, it is agreed upon $H_1 \circ \underline{emp} \equiv \underline{emp} \circ H_1 \equiv H_1$ holds.

In contrast to def.\ref{def:HeapConjunctionDefinition}, the first part of the $\circ$-term is scanned as potentially matching vertex -- this is chosen arbitrarily and does not impose any obligation.
Otherwise, it may freely be changed.
Further, it is agreed upon that $H_1 \circ H_2$ denotes a united heap graph.
\label{theo:GeneralizedHeapConjunctionTheorem}
\end{theorem}

\begin{proof}
This theorem is a generalisation of def.\ref{def:HeapConjunctionDefinition}.
Both $H_1$ and $H_2$ may be simple heaps $a_1 \mapsto b_1 \circ a_2 \mapsto b_2 \circ \cdots \circ a_n \mapsto b_n$.
If a shared element does not exist in both graphs, then according to def.\ref{def:HeapConjunctionDefinition}, this yields $\underline{false}$ -- if this is shown first, then soundness holds.
This implication, however, matches the expected \index{conjunction} conjunction.
Otherwise, if at least one common element exists, then according to \index{induction} induction, we select one element, and then both heaps connect.
It shall be noted that conjunction only assumes graphs are connectible, and we are not interested in more than one shared elements to be connected.
All vertices except one may hypothetically be used for \index{graph!merge} joining.
The only exception is naturally the one to be used for joining graphs.
In this case, the resulting graph remains \index{graph!simple} simple.
Otherwise, the joined heap graph's beginning or end is either in $H_1$ or in $H_2$.

Furthermore, it must also be in both graphs $H_1$ and $H_2$ simultaneously.
The latter, however, can be excluded.
Due to this contradiction, the validity of the theorem follows.

By definition, $a\circ a$ equals \textit{\underline{false}}, which needs to be filtered for all conjuncts.
The discussion on \index{SMT-solver} $a\circ a$-solver will be later.
It may be implemented using some "\textit{active set}", which contains statically all successfully processed simple heaps.
\end{proof}

\begin{conventions}[Locations]
In \index{abstract predicate} APs, \index{location} \textit{locations} may be \index{symbol} symbols.
For the sake of usability, it is agreed upon APs and locations (mainly w.r.t. object fields and locations), the \index{associativity} left-associative binary operator grants that field access \index{.-operator} ".".
Field access may be inductively recursed.
\label{obs:ObjectPathAccessor}
\end{conventions}

Left-associativity denotes some term $object1.field1.field2.field3$ by default to equal:
$$(((object1).field1).field2).field3.$$
Thus, symbolic variables may replace parts of the access expressions to a field.
That increases modularity, but particularly also expression flexibility.

\begin{lemma}[Monoid over Conjunction]
%%%
$G=(\Omega, \circ)$ is a \index{monoid} monoid, where $\Omega$ denotes the set of well-defined \index{heap graph} heap graphs and the binary operator $\circ$ denotes \index{conjunction} heap conjunction.
%%%
\label{lem:HeapConjunctionMonoid}
\end{lemma}

\begin{proof}
In order to show that $G$ is a \index{monoid} \textit{monoid}, it is necessary to show: (i) $\Omega$ is closed \index{closure} under $\circ$, (ii) $\circ$ is associative, and (iii) $\exists \varepsilon \in \Omega. \forall m \in \Omega: m \circ \varepsilon = \varepsilon \circ m = m$.

Let $\omega \in \Omega$ be some connected heap graph obtained by the binary functor $\mapsto$ according to def.\ref{def:HeapTermDefinition}.
According to def.\ref{def:HeapConjunctionDefinition} $\forall \omega \in \Omega: \omega \circ \omega = \underline{false}$ holds.
Otherwise, for $\omega_1, \omega_2 \in \Omega$ there may be only two cases.
If $\omega_1$ and $\omega_2$ share at least one vertex, then according to theo.\ref{theo:GeneralizedHeapConjunctionTheorem}, the corresponding heap graph is defined, otherwise the result is $\underline{false}$ (denoting $\omega_1$ and $\omega_2$ do not intersect).
Thus, we showed that $\Omega$ is enclosed over $\circ$ and that a heap graph is the result of conjunction.
In that case, a link is successfully established.

Next, associativity must be shown.
Namely, $m_1 \circ (m_2 \circ m_3) = (m_1 \circ m_2) \circ m_3$ holds.
Fig.\ref{fig:HeapGraphConjBeforeAfter} immediately implies the validity of the equation from both sides.
The equation holds because of two possible options.
First, it does not matter if the contents of $a$ and $b$ are connected.
Second, it does not matter whether $a$ is linked to $b\circ c$.
The latter might be because the connecting vertex $b$ remains \index{invariant} invariant, where the ordering of \index{conjunction} conjunction changes.

$G$ builds up a \index{group!sub-} subgroup.
It remains to show the existence of a \index{group!neutral element} neutral element $\varepsilon$ s.t. (iii) holds.
The equations follow from a lifted theo.\ref{theo:GeneralizedHeapConjunctionTheorem}.
\end{proof}

\textbf{Remark:} (i) implies $c \not \in b \wedge c \neq a$: $a \mapsto b \ \circ \ c \mapsto d \equiv \underline{false}$, and $a\mapsto b \ \circ \ a \mapsto d \equiv \underline{false}$ holds.
It is evident if there is a choice, then it does not matter which vertex to connect first -- the result will be the same due to \index{confluency} \textit{confluency} due to (ii) property shown later in lem.\ref{lem:HeapConjunctionGroupProperty} (cf. fig.\ref{fig:CRTonHoareTriples}).

\textbf{Remark:} Closure (i) indicates the \index{non-repetitiveness} \textit{non-repetitiveness} \index{substructural logic} of a substructural logic (SL as such, see sec.\ref{sect:ExprPredicates}), which still holds and will be shown later.

%%%%%%%%%%%%%%%%%%%%%%%%%%%%%%%%%%%%%%%%%%%%%%%%%%%%%%%%%%%%%%%%
\begin{theorem}[Abelian Group over Conjunction]
$G=(\Omega, \circ)$ is an \index{group!Abelian} Abelian group.
\label{lem:HeapConjunctionGroupProperty}
\end{theorem}

\begin{proof}
Due to lem.\ref{lem:HeapConjunctionMonoid}, $G$ denotes a \index{monoid} monoid.
Hence, it is still necessary to show, (i) the existence of an inverse for any element from the \index{carrier set} carrier set over heap graphs, s.t.:
\begin{eqnarray}
\forall \omega \in \Omega. \exists \omega^{-1} \in \Omega: \omega \circ \omega^{-1} = \omega^{-1} \circ \omega = \varepsilon
\label{eqn:InverseExists}
\end{eqnarray}
hold and (ii) $\circ$ is a commuting operator.

Let us show first (ii) holds that in the base case, $loc_1 \mapsto var_1 \circ loc_2 \mapsto var_2 = loc_2 \mapsto var_2 \circ loc_1 \mapsto var_1$ holds due to the inductive def.\ref{def:HeapTermDefinition}.
Also, the inductive case holds as long as conditions hold of $\circ$.
The induction condition may be obtained from fig.\ref{fig:HeapGraphConjBeforeAfter} the $\circ$-operator commutes for two arbitrary connected heaps. 
However, as soon as APs are considered, the denotation of $\circ$-connected terms may be limited by a predicate's boundaries, and those terms may then not be dropped, as there still might be a subgoal in APs.
From a practical perspective, this does not bother.
On the contrary, it requires the developer to obey stricter software modularity, positively affecting the specification.

The proof continues by showing property (i).
An \index{group!inverse element} inverse element, if existing, neutralises any given element.
So, we better answer the following questions:
What specifically hinders us from defining an \index{heap!inverse} inverse heap?
When discussing natural numbers, then the inverse of addition is subtraction enclosed to the same \index{carrier set} carrier set (for simplicity, we skip corner cases).
The same happens to a field over \index{complex numbers} complex numbers, which is an extension of \index{field!real numbers} real numbers.
Moreover, even if real or complex numbers are not countable, in practice, arithmetic field \index{field!arithmetic} extension still leads to a significant simplification of calculation.
Even if the answer to numerous questions related to real numbers might seem vague, such as: "\textit{What does the imaginary number "$i$" stand for in reality}"?
Still, it showed that $i$ is useful for calculations in numerous applications due to the identity $i^2=e^{i\pi}$.
So, it would only be fair then to ask: Why not postulate some inverse heap in eqn.\ref{eqn:InverseExists}, at least so we can be sure to calculate the right thing until shown otherwise?

Thus, what can \textit{intuitively} be understood by some \index{heap!inverse} \index{heap!inversion} \textit{inverse heap}?
If natural or real numbers are meant, then a new imaginary axis is introduced: numbers increase or decrease relative to zero depending on the axis' direction.
However, how to deal with heaps, for example, with ordinary ones of form $a\mapsto b$?
May perhaps an invertible heap be assigned the inverse of an accompanying predicate?
It might not be accurate.
May a heap inversion perhaps be assigned the nil denotation?
However, this may not be correct since any \index{heap!empty} non-empty heap will, by definition, then be empty.
Furthermore, what about \index{heap!inversion} inversion of the empty heap?
Such a na\"{\i}ve definition is counter-productive since it does not seem to cover all heaps comprehensively.
What if sides of heap graph edges flip, e.g. $a\mapsto b$ turns into $b\mapsto a$?
This variant would be just an interesting idea but not tractable yet because \index{conjunction} conjunction does not contradict such a definition. 
What is needed instead is conjunction yields the empty heap.
In such an approach, the left-hand side would be left undefined in the case of an object.
It is easy to imagine that joining heaps removes some positive heap if connected with its negative supplement.
Thus, inversion means some transcendental removal operation of a heap, and the operator may be applied to any heap, including a complex.
$(a\mapsto b)^{-1}$ may denote "$a$ not pointing to $b$" or better as "$a$ points inversely to $b$".
The first explanation fails, because \index{pointer} "\textit{does not point}" by mistake might be considered as $a\mapsto c$, where $\exists c\in \Omega, c\ne b$ -- this implies vertex $a$ exists and between $a$ and some element whose content differs from $b$ --- all this is false because there is no such evidence to trust.
Hence, the hypothetical case "\textit{inversely points to}", namely "\textit{inversely has a reference to}", requires some deletion regardless of the strange value, which may have allowed a clean deletion of all superfluous elements, which require a further check.
At first glance, this looks odd.
Until now, only heap parts are specified existing.
Inversion changes the rules of the play.

\index{group!inverse element} Inversion shall not be mixed up with a heap that may not be in a heap.
It would be predicate negation, but heap inversion is an operation --- currently, our focus is on eqn.\ref{eqn:InverseExists}.
The given equation means that for "\textit{the negated points-to}", $a \not \mapsto b$ holds.
--- It is a predicate relation between $a$ and $b$.
In generalised form \index{heap!negative} a negated heap $H^{-1}$, which by default holds: $a \mapsto b \circ a \not \mapsto b = \underline{emp}$ and $a \not \mapsto b \circ a \mapsto b = \underline{emp}$, and lifted $H \circ H^{-1} = H^{-1} \circ H = \underline{emp}$.
It implies that $\omega \circ \omega^{-1}$ removes a heap.
If needed, the superfluous edge and vertex are removed from the \index{heap graph} heap graph if no more incoming from or outcoming to edges remain w.r.t. still existing graph vertices.
Thus, the given equality about inverse pointers becomes clear.
It is not difficult to check the equality $H \circ H^{-1} \circ H \equiv H$.
In the example from fig.\ref{fig:HeapGraphInversion}, the heap state before \index{group!inverse element} inversion is $d\mapsto a \ \circ \ a\mapsto b \ \circ \ c \mapsto b$, but when inversion $\circ (a\mapsto b)^{-1}$ is applied, this results in $d\mapsto a \ \circ \ a\mapsto b \ \circ \ c \mapsto b \ \circ \ (a\mapsto b)^{-1}$, which equals $d\mapsto a \ \circ \ a\mapsto b\ \circ \ (a\mapsto b)^{-1} \ \circ \ c \mapsto b$ equals $d\mapsto a \ \circ \ c \mapsto b$, which is not evident since both \index{pointer} pointers do not overlap.
So, this calculation is not fully cleaned up and requires further steps to correct superfluous elements.
Thus, two more normalisation steps are needed.

\textbf{Normalisation -- First step}: This step is generic.
If a \index{graph!bridge} bridge exists between the considered heap graphs as the only connection between those, then the operator shall be replaced by \index{disjunction} disjunction (see later).

As a \index{graph!bridge} \textit{bridge} between $a$ and $b$ is found, $\circ$ is replaced by $\|$ in the remaining term.
The result may be considered plausible.
For completeness, it may become needed vertices must be entirely removed from the corresponding \index{heap graph} heap graph.
This need occurs on \index{object field} object field \index{location} locations.

\textbf{Normalisation -- Second step}: Deleting vertex $a$ can be accomplished when there are no more pointers left to $a$ in the remaining heap.

We apply those two \index{heap!normalisation} normalisation steps to avoid the mentioned exclusions (cf. generalised heaps with obs.\ref{obs:SimplificationByGeneralisation}).
\end{proof}

\textbf{Remark:} \index{heap!generalisation} Generalised heaps were not discussed.
For the soundness proof, it is necessary to show $H \circ H^{-1} \equiv \underline{emp}$.
The proof can be taken out inductively over $\circ$ by using $(g_1 \circ g_2)^{-1} \equiv g_1^{-1} \circ g_2^{-1}$, s.t. a \index{homomorphism} homomorphism exists for "$.^{-1}$" relating $\circ$ (see lem.\ref{lem:HeapInversionHomomorphism}).

\begin{conventions}[Heap Invertibility]
\label{conv:EmptyHeapInversion}
Condition (i) implies a partial case $\underline{emp} \circ \underline{emp}^{-1} \equiv \underline{emp}$ because we agree upon $\underline{emp}^{-1} \equiv \underline{emp}$.
\end{conventions}

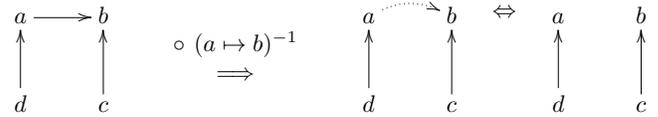
\begin{figure}[h]
\begin{center}
\scalebox{0.9}{
\begin{tabular}{lllll}
 \xymatrix{
   a \ar[r] &  b\\
   d \ar[u] & c \ar[u]
 }
& \parbox[t]{2.7cm}{\begin{center}$\circ \ (a \mapsto b)^{-1}$\\ $\Longrightarrow$\end{center}} &
 \xymatrix{
   a \ar@{..>}@/^/[r] &  b\\
   d \ar[u] & c \ar[u]
 }
 &
  $\Leftrightarrow$
 &
 \xymatrix{
   a  &  b\\
   d \ar[u] & c \ar[u]
 }\\
\end{tabular}}
\end{center}
 \caption{Heap graph before and after inversion}
 \label{fig:HeapGraphInversion}
\end{figure}

\textbf{Explanation:} $H \circ a\mapsto b \circ (a\mapsto b)^{-1}$ denotes:

\begin{enumerate}
 \item remove edge between $a$ and $b$
 \item remove vertex $a$, if there are no references to it in graph $H$
 \item also remove vertex $b$, if there is no reference to it in graph $H$
\end{enumerate}

Group properties \index{group} allow us to establish equalities for heap terms, so improving \index{confluency} proof confluency or transformation into normal-form (see thes.\ref{thes:SimplificationByDiffingHeaps}).
Thus, bloated heap term IR may be reduced.
Partial specifications allow a reduced rule notation (see sec.\ref{sect:PartialSpec}).
Further work includes \index{SMT-solver} SMT-solver integration for practical heap term simplification.
It is required to consider when the content is a pointer again, e.g. $a\mapsto o.f \circ  (a\mapsto o.f)^{-1}$, where $o$ denotes an object containing field $f$.
The field $o.f$ is not removed from dynamic memory.
Otherwise, object encapsulation will be under threat --- further research here is highly appreciated but is beyond this work's scope (see discussion in sec.\ref{sect:ClassObjects}).
That is why by default, it is agreed upon that any field remains in memory as a unit of some contiguous object memory layout and may point to \texttt{nil}.
Thus, all the mentioned steps shall remain without change as is.

\begin{lemma}[Homomorphism over Simple Heaps]
 $(g_1 \circ g_2)^{-1} \equiv g_1^{-1} \circ g_2^{-1}$ holds for any \textit{non-simple heaps} $g_1$ and $g_2$.
 \label{lem:HeapInversionHomomorphism}
\end{lemma}

\begin{proof}
 If we succeed in showing the generalised form $G=g_1\circ g_2\circ \cdots \circ g_n$, then the problem would be solved.
 For doing so, $G \circ G^{-1} = \underline{emp}$ must be shown.
 This equation may be proven inductively using $n$.
 In the base case ($n=1$), $g_1\circ g_1^{-1} \equiv \underline{emp}$.
 It holds because an inverse element must exist.
 For the inductive case, we assume
 $$G=\underbrace{(g_1\circ g_2 \circ \cdots \circ g_k)}_{G_k} \circ g_{k+1}$$
Then for $$G\circ G^{-1} = (G_k \circ g_{k+1}) \circ (G_k \circ g_{k+1})^{-1}$$, the inverse part of the heap is a natural extension of the heap.
 The first part of equality holds

$$\underbrace{G_k\circ G_k^{-1}}_{\underline{emp}}  \circ \underbrace{g_{k+1} \circ g_{k+1}^{-1}}_{\underline{emp}} \ \equiv \ \underline{emp}$$
because of the inductive property of \index{group!inverse element} inversion, obeying conv.\ref{conv:EmptyHeapInversion}.
\end{proof}

%%%%%%%%%%%%%%%%%%%%%%%%%%%%%%%%%%%%%%%%%%%%%

\begin{definition}[Heap Disjunction]
\label{def:HeapDisjunction}
\textit{Heap disjunction} $H \ \| \ a \mapsto b$ denotes a heap $H$, and some simple heap $a \mapsto b$, which does not overlap.
If $G_H$ is some \index{heap graph} heap graph $H$, then $G_H=(V,E)$, for all edges $(\_,a) \not \in E$ and there is no path from any location of $b$ to $H$, and there is no inverse path from $H$ back to $a$.
"\_" denotes some variable term placeholder.
\end{definition}

Thus, $x.b \ \| \ x.c$ does not hold for any $x$ with fields $b$ and $c$ if at least one common vertex exists for any path from $x.b$ or $x.c$.

Assume, $\Sigma = X_0 \| X_1 \| \cdots \| X_n$ where $n>0$ and $X_j$ has form $x_j \mapsto y_j$, then $\Sigma = \Sigma_0 \ \| \ a_0 \mapsto b_0 \Leftrightarrow \forall (a_j \mapsto b_j) \in \Sigma_0: a_j \neq a_0 \wedge b_j \neq b_0$.\\

\begin{lemma}[Monoid over Disjunction]
\label{lem:monoidOverDisj}
$G=(\Omega, \|)$ is a \index{monoid} monoid and \index{group} group, if $\Omega$ is a set of heap graphs and $\|$ denotes \index{disjunction} \index{heap} heap disjunction.
\end{lemma}
\begin{proof}
In analogy to lem.\ref{lem:HeapConjunctionMonoid}, $\forall m_1,m_2 \in \Omega: m_1 \| m_2$ holds when $m_1$ and $m_2$ do not share a \index{shared vertex} common vertex because when there is no path from $m_1$ to $m_2$, and there is no graph surrounding both $m_1$ and $m_2$.
If $m_1$ and $m_2$ differ, then $m_1 \| m_2$ again is a valid heap in $\Omega$, because $m_1$ is from another part of the heap graph than $m_2$ is, and vice versa.
Hence closure follows.
\index{associativity} Associativity follows for apparent reasons.
$\underline{emp}$ may serve as a neutral element, then $\underline{emp} \| m_1 = m_1 \| \underline{emp} = m_1$.
Let by default $\underline{emp} \| \underline{emp} = \underline{emp}$ hold.
Finally, it is agreed upon $s \| s^{-1} = s^{-1} \| s = \underline{emp}$, which is similar to $\circ$.
In general, heaps follow this convention.
\end{proof}

%%%%%%%%%%%%%%%%%%%%%%%%%%%%%%%%%%%%%%%%%%
\index{conjunction} Conjunction and \index{disjunction} disjunction of heap fragments may be expressed by rules and rules that are dual to those in the following manner:

\begin{center}
\begin{tabular}{lcl}
 \inference[$\circ_{[B,C]}$]{U \circ B \ \| \ C}{U\circ B\circ C} & \qquad \qquad &
 \inference[$\|_{[B,C]}$]{U\circ B\circ C}{U \circ B \ \| \ C}
\end{tabular}
\end{center}

\begin{eqnarray}
\|_{[B,C]} ; \circ_{[B,C]} ; \|_{[B,C]} \equiv \|_{[B,C]}
\label{eqn:HeapDisjunctionInvariant}\\
\circ_{[B,C]} ; \|_{[B,C]} ; \circ_{[B,C]} \equiv \circ_{[B,C]}
\label{eqn:HeapConjunctionInvariant}
\end{eqnarray}

\index{dual operation} Operations $\|$ and $\circ$ are dual.
They may be transformed into each other using the dual operation (see eqn.\ref{eqn:HeapDisjunctionInvariant} and eqn.\ref{eqn:HeapConjunctionInvariant}), where "\textbf{;}" denotes the sequential operator.
Equalities hold because the operations are \index{self-inverse operation} self-inverse and because of the provided assertion about the existence of both heap vertices $B$ and $C$.

\begin{theorem}[Distributivity]
\index{distributivity} Distributivity holds for $\forall a,b,c \in \Omega$ for $\circ$ and $\|$:

\begin{tabular}{ll}
 (i) & $a \circ (b \| c) = (a \circ b) \| (a \circ c)$\\
 (ii) & $(b \| c) \circ a = (b \circ a) \| (c \circ a)$ 
\end{tabular}
\label{lem:DistributivityForConjDisj}
\end{theorem}

\begin{proof}
  Only (i) has to be shown because (ii) immediately follows from the duality of $\circ$ and $||$.
  Two implications show equality (i).
  "$\Rightarrow$" means that from $a\circ (b||c)$ follows $(a\circ b)||(a\circ c)$.
  "$\Leftarrow$" denotes implication for the opposite direction.\\
  "$\Rightarrow$": $b||c$ implies no connections between $b$ and $c$ exist.
  Thus, $a\circ (b||c)$ means either between $a$ and $b||c$ there is also no connection, which eventually leads to \textit{\underline{false}}, or one connecting vertex exists, and by default, it must be in $b$ or $c$.
  W.l.o.g. a connecting vertex in $b$ is assumed, then $a\circ b$ holds.
  The analogous may follow from $c$.
  Regardless of which, but one graph is to be extended: $b$ or $c$.
  Other graphs do not extend.
  Thus, the second graph is always \textit{\underline{false}}.\\
  "$\Leftarrow$":  Either $a\circ b$ or $a\circ c$, both equal to \textit{\underline{false}} due to def.\ref{def:HeapConjunctionDefinition}, because $a$ may not connect to $b$ and $c$ at the same time, otherwise it would contradict def.\ref{def:HeapDisjunction}.
  W.l.o.g. it may be stated that if $a\circ b$ equals \textit{\underline{true}}, then $a\circ (b||c)$ would be too.
\end{proof}

\textbf{Remark:} Since the \index{group!neutral element} neutral element for both operations $\circ$ and $\|$ is $\underline{emp}$, there is no way to define an \index{algebraic field} \textit{(algebraic) field} (e.g. \index{field!Galois} a Galois field) over these operations since lem.\ref{lem:HeapConjunctionMonoid}, lem.\ref{lem:HeapConjunctionGroupProperty} and lem.\ref{lem:DistributivityForConjDisj} hold and regardless of the fact the \index{carrier set} carrier set $\Omega$ is finite.
Hence, any \index{heap!finite} finite heap and all operations defined over it are also generating a finite heap.

\textbf{Remark:} In analogy to logical conjuncts $\wedge$ and $\vee$, a normalised form using $||$ always exists.
Previously mentioned equalities and lem.\ref{lem:HeapInversionHomomorphism} can achieve this in order to \index{heap!inversion} invert a generalised heap.\\

It is necessary to hoist $\|$ as far out as possible by either applying \index{distributivity} distributivity rules or reordering non-complex heaps to shrink a heap graph.
So, the left-hand sides of pointers are sorted ascending by their location in lexicographical order. 
The idea is to cut redundant searches, e.g. for \index{stepwise verification} \textit{stepwise verification}, s.t. only changing heap parts are re-calculated.

\begin{figure}[h]
\begin{displaymath}
  \xymatrix@C=2em@R=1em{
       &&& \top \ar@{-}[dl] \ar@{-}[dd]\\
       &  & \parbox[t]{12mm}{$(\Delta\_\Delta\Delta)\\\ \ G_1''$} & \\
       & \parbox[t]{7mm}{$G_1'\\(\Delta\_)$} \ar@{-}[ru] &  & \parbox[t]{7mm}{$G_2'\\(\Delta\Delta)$} \ar@{-}[lu] \\
      \parbox[t]{7mm}{$G_1\\(\Delta)$}  \ar@{-}[ru] &  & \parbox[t]{7mm}{$G_2\\(\_)$} \ar@{-}[lu] \ar@{-}[ru]\\
      && &\bot \ar@{-}[lllu] \ar@{-}[lu] \ar@{-}[uu]
  }
\end{displaymath}
 \index{poset}
 \caption{Poset over heap graph inclusion}
 \label{fig:HeapGraphPoset}
\end{figure}
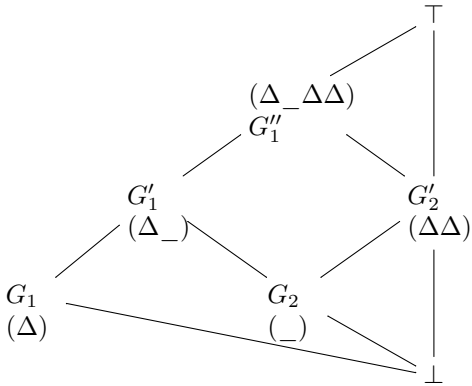

A poset may be established over heap graphs and set-inclusion as \index{join} lattice join operator $\circ$.
The infimum ($\bot$) is the \index{heap!empty} empty heap.
It is the least bound.
The least (upper) bound is the \index{graph!full} full graph ($\top$).
For example, \index{absorption} \textit{absorption} does not hold for \index{heap} heaps.
So, the poset may not be a complete lattice.
Fig.\ref{fig:HeapGraphPoset} the \textit{poset} $G$ contains $\{G_1,G_2,G_1',G_2',G_1''\}$, which hold the following inequalities in ascending ordering $G_1 \sqsubseteq G_1' \sqsubseteq G_1''$, $G_2 \sqsubseteq G_1'$ and $G_2 \sqsubseteq G_2' \sqsubseteq G_1''$.
The heaps $\bot$ and $\top$ are always there and may always be added if not present.
Hence, these may sometimes not be explicitly mentioned.
$G_1''$ is the supremum, and $inf(G)=\underline{emp}$ is infimum, where $\sqsubseteq$ defines the subset relation over \index{graph!relation} graphs.
It is not difficult to realise two not connected heaps having no common vertex under conjunction (so are equal $\underline{emp}$ due to def.\ref{def:HeapConjunctionDefinition}) always remain unconnected according to the \index{Hasse diagram} Hasse-diagram.
Thus, a merge is always $true$ because either (first contradiction) $a\mapsto b \ \circ \ a\mapsto d$ cannot be implied from the first \index{conjunction} conjunction or any other \index{heap!complex} complex heap.
Because of (the second contradiction), $a\mapsto b \| b\mapsto d$ contradicts the definition of $\|$.
However, it is essential to notice a poset's ordering may be violated after \index{heap!inversion} introducing heap inversion (see earlier) if using inversion unbound.
Currently, \index{group!inverse element} inversion satisfies our needs w.r.t. a more comfortable heap comparison, particularly for the specification to be checked, therefore the \index{locality} locality property from sec.\ref{chapter:expression} remains unchanged.

\subsection{Class Object as Heap}
\label{sect:ClassObjectAsHeap}

\begin{figure}
\def\gA#1{\save
[].[]!C="gA#1"*[F]\frm{}\restore}

\def\gB#1{\save
[].[dd]!C="gB#1"*[F]\frm{}\restore}

\def\gC#1{\save
[].[ddr]!C="gC#1"*[F]\frm{}\restore}
\begin{center}
\begin{tabular}{cc}
%\fbox
{
\scalebox{0.8}{
\xymatrix{
  \\
  \circ \ar[r] & \circ \ar[r] & \gA1 \circ \circ \circ
}}}
&
%%%
%\fbox
{\scalebox{0.8}{
\xymatrix@C=2em@R=1em{
 & \gB1 \circ \ar[r] & \\
 \circ \ar[r] & \circ \ar[r] & \circ\\
 & \circ \ar[r] &
}}}\\
a) simple &  b) single field\\[0.3cm]
%%%
%\fbox
{\scalebox{0.8}{
\xymatrix@C=2em@R=1em{
 & \gB1 \circ \ar[r] & \circ\\
 \circ \ar[r] & \circ \ar[r] & \circ\\
 & \circ \ar[r] & \circ
}}}
&
%%%
%\fbox
{\scalebox{0.8}{
\xymatrix@C=2em@R=1em{
 &  \gC1 & \ar[r] & \circ\\
 \circ \ar[r] && \ar[r] & \circ\\
 & & \ar[r] & \circ
}}}\\
c) multi-assignment & d) schematic multiple
\end{tabular}
\end{center}
 \caption{Shapes of object assignments}
 \label{fig:ObjectAttributesScheme}
\end{figure}
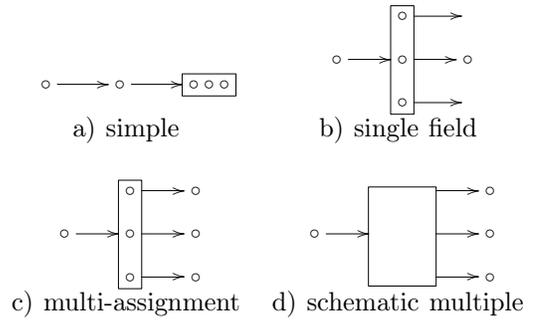

Let us consider how class objects are modelled in the heap graph according to fig.\ref{fig:ObjectAttributesScheme}.
In fig.\ref{fig:ObjectAttributesScheme} a), there is a separate assignment presented for an existing graph (without details).
The object in the rectangle refers to a particular \index{class instance} object instance.
In fig.\ref{fig:ObjectAttributesScheme} b), only the second \index{object!field} field of the object instance is assigned to some \index{type compatibility} compatible value.
All remaining fields are assigned \texttt{nil}.
Fig.\ref{fig:ObjectAttributesScheme} c) and fig.\ref{fig:ObjectAttributesScheme} d) have the same meaning as in b), except assignment is performed on all fields.
Following up, we introduce additional restriction w.l.o.g. for better reuse of operators still to be defined:

\begin{conventions}[Modelling of Object Instances] Object instances are modelled w.l.o.g. as
\begin{itemize}
 \item[a)] There are no \index{inner object} inner objects.
 Inner objects may always be modelled equivalently to separate objects associated with an externalised object \cite{bruce02}. 
References are only allowed to those objects that are replaced by non-overlapping.
 Hence, w.l.o.g. \index{\texttt{union}} \texttt{union} are excluded from modelling objects since they do not increase expressibility.
 \item[b)] Object fields differ and may be used by further objects and be defined in \index{inherited class} inherited classes.
 In the case of name clashing amongst the \index{inheritance hierarchy} class inheritance hierarchy, clashing classes closer to a common superclass may be renamed, s.t. all references to the problematic reference may be renamed.
 \item[c)] Due to closure during their life-time, objects do not grow, except cases where the \index{pointer} pointer casts to a new object of a differing (sub-)class.
 The problem of \index{memory fragmentation} \textit{memory fragmentation} may only arise where memory is limited.
It may be worth transforming dynamically allocated objects into \index{stack} automatically managed ones to avoid fragmentation with reasonable efforts in this context since a stack frame dynamically grows when an additional block guards objects.

 Due to inheritance, objects may only monotonously grow according to the $<:$ relationship from fig.\ref{DefinitionClassSubtype} (when intentionally  disobeying enhanced class visibility modes, as it was defined in C++, for instance).
 Consequently, regions in memory layouts may only grow.
 An increase may lead to severe problems, whose solution may involve utilisation and renewed memory allocation but with different addresses.
 In general, this problem is not decidable because of the Halting problem.
 However, if the problem on maximum limits may statically be resolved just in particular cases, an improved variant may be chosen.
 In reality, object fields may also vanish due to visibility and class inheritance.
Therefore a monotony of object fields is only a heuristic, where a fixed class size may be a reasonable estimate when \textit{late object binding} is prohibited.
 Assigning arbitrary objects of a (sub-)class is prohibited.
 Since Burstall's separating heaps \cite{burstall72} do not obey compactness, object fields may not occur more than once in one conjunction in a heap expression as (left-hand side) location.
 \item[d)] \index{array} Arrays as base type is prohibited.
 The heap graph must remain simple.
 However, this does not prohibit an arbitrary number of pointers pointing to a particular address.
 \item[e)] Common heap graph vertices may be complex.
W.l.o.g.\index{attribute} attributes of APs may be generalised.
\end{itemize}
\label{conv:RestrictedObjects}
\end{conventions}
In order to avoid contradictions in the following definitions, two steps are needed. First, to be able to check relations between two vertices real quick.
A given non-empty complex heap contains one or more \index{heap!simple} simple heaps.
To check whether two heap graph vertices are connected, it is required to check all left sides of simple \index{conjunct} conjuncts in the worst case, so a full edge-scan is required.

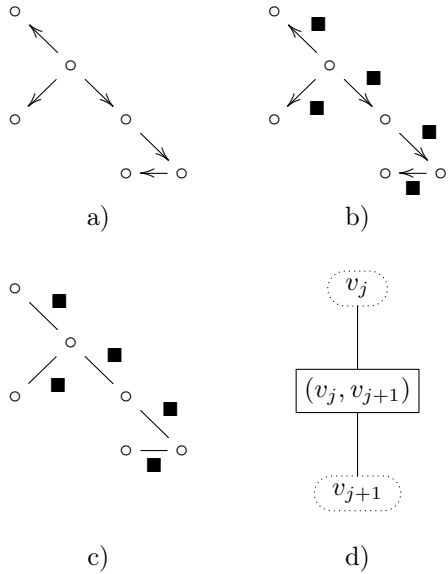
\begin{figure}[h]
\begin{center}
\begin{tabular}{ccccccc}
\xymatrix@C=1em@R=1em{
 \circ \ar@{<-}[dr]\\
 & \circ \ar@{->}[dl] \ar[dr] \\
 \circ && \circ \ar[dr]\\
 && \circ & \circ \ar[l]
}
& \quad &
\xymatrix@C=1em@R=1em{
 \circ \ar@{<-}[dr]^-{\blacksquare}\\
 & \circ \ar@{->}[dl]^-{\blacksquare} \ar[dr]^-{\blacksquare} \\
 \circ && \circ \ar[dr]^-{\blacksquare}\\
 && \circ & \circ \ar[l]^-{\blacksquare}
}\\
a) && b)\\\\

\xymatrix@C=1em@R=1em{
 \circ \ar@{-}[dr]^-{\blacksquare} \\
 & \circ \ar@{-}[dl]^-{\blacksquare} \ar@{-}[dr]^-{\blacksquare} \\
 \circ && \circ \ar@{-}[dr]^-{\blacksquare}\\
 && \circ & \circ \ar@{-}[l]^-{\blacksquare}
}
& \quad &
\xymatrix{
  *+[F.:<20pt>]\txt{ $v_j$ \ } \ar@{-}[d]\\
  *+[F]\txt{($v_j,v_{j+1}$)}  \ar@{-}[d]\\
  *+[F.:<20pt>]\txt{ $v_{j+1}$ \ }
}\\\\
c) && d)
\end{tabular}
\end{center}
 \caption{Transformation of schematic heap graph}
 \label{fig:InvertedVertexGraph}
\end{figure}

Second, to obtain all neighbouring vertices for a given vertex, a model is beneficial that is quickly traversed by vertices and not by edges.
For a quick but accurate definition of neighbours, e.g. \texttt{reach\-es(x,y)}, \texttt{reach\-es(x,Y)}, \texttt{reach\-es(X,y)}, \texttt{reach\-es(X,Y)} is proposed, where \texttt{x} denotes a well-defined vertex, and \texttt{X} denotes a vertex (sub-)set (to \texttt{y} and \texttt{Y} in the analogy).
The model based on vertices is illustrated in fig.\ref{fig:InvertedVertexGraph}.
Small opaque squares accordingly allow gaining a more effective iteration as described.
Both models, where a) is vertex-based and b) is edge-based, are \index{duality} dual, so both w.l.o.g. are interchangeable.
Vertex pairs, including source and destination, are encoded over the middle of an edge (see (d)).
Then this is interpreted as one vertex with a unified name and is connected to neighbouring vertices (see fig.\ref{fig:InvertedVertexGraph} (a) and (b)).
Turning one model to another and vice versa according to the scheme from fig.\ref{fig:InvertedVertexGraph} c) is possible without losses, so the mapping between both structures is bijective regardless of the edge directions (see (c)).\\

Hence, transforming the graph with conjunctions or disjunctions is allowed.
It is more convenient for normal-forms and optimisations.

\begin{conventions}[Object Fields]
\textit{Object fields} do not overlap, and field addresses are different (cf. conv.\ref{obs:ObjectPathAccessor}).
However, \index{pointer} pointers may have \index{alias} aliases.
An object's content may be expressed as $x \mapsto object(fld_1,fld_2, ...)$.
W.l.o.g. by default, object field pointers may not be used as an L-Val \cite{isocpp14} in arbitrary arithmetic expressions.
It may only be used by locations according to obs.\ref{obs:ObjectPathAccessor}.
\index{late binding} \textit{Late binding} is not considered.
The lack of late binding means, polymorphism mimicked by subclassing during runtime is absent, implying an additional restriction.
However, all accessor fields, including field types, maybe thoroughly analysed \index{static analyis} statically.

Consequently, whilst static analysis, only the most common class amongst the inheritance hierarchy may be used.
The question of which class to use during object instance generate is undecidable in general.
The program statement responsible for generation may be accessible or not.
To predict it earlier is, in general, not possible.
So, only the most common \index{heap!generalised} heap can be built whilst verifying.
\label{conv:HeapAlignment}
\end{conventions}

\subsection{Partial Heap Specification}
\label{sect:PartialSpec}

As mentioned earlier, according to def.\ref{def:HeapTermDefinition}, object instances may be interpreted as data containers for each field $obj.f_1 \mapsto .. \circ obj.f_2 \mapsto .. \circ obj.f_n \mapsto .. $.
All fields generate some heap slices.
In contrast to abstraction, this transformation may be characterised as \index{concretisation} concretisation.
Class object fields have restrictions on naming and types, which both have to match those of the corresponding class.
All fields of one object co-exist.
Initially, they relate to simple heaps only, which later may be extended by $\circ$-conjunctions.
Object fields may not always be freed one by one (see obs.\ref{obs:GraphIRHeap} and the last discussion).
Local object fields must also have the possibility to specify parts, so subclass objects, as do have ordinary local variables in SL-extensions.
In analogy to non-object pointers, constant functions over heaps may be defined.
For example, $\underline{true}(obj)$ or $\underline{false}(obj)$ from def.\ref{def:HeapTermExtendedDefinition} and def.\ref{def:HeapTermDefinition}.
In contrast to non-object pointers, \textit{\underline{true}}(obj) introduces some object term as an additional parameter.
This way, APs may be used to modularise further a specification, which syntactically and semantically insignificantly differs from the non-object case.

In synchronisation with the proposed $a \circ a$-solver based on a stack and a stratified AP (see sec.\ref{chapter:APs}), incoming and outgoing terms of APs may be traced and skipped for redundant or superfluous calculations.
The comparison may be made either on the stack for the same levels or by labels between predicates' arbitrary levels using translating rules (cf. sec.\ref{sect:TranslatingHornRules}).

\begin{definition}[Incomplete Predicate]
 \label{def:IncompletePredicates}
An incomplete predicate is an AP, where the term argument(s) passed if unfolded denote only a partial heap graph, although all remaining fragments are automatically accepted upon heap interpretation.
For simplicity, the AP $\underline{true}(obj)$ is considered next.
 $\underline{true}(obj)$ defines some heap of \index{conjunct} $\circ$-conjuncts of all \index{field} fields of some object $obj$ (includes the \index{heap!empty} empty heap, which represents an empty object).
\end{definition}

 All directly specified fields are calculated first from the set of remaining fields.
 Whilst verification, all fields are meant when referring to \index{constant function} constant functions, which were not directly specified in a considered AP.
More details will follow in sec.\ref{chapter:APs}.
It is not hard to agree on the correctness of $\wedge,\vee,\neg$-conjunctions for \textit{\underline{true}}(obj) and \textit{\underline{false}}(obj).
These constant functions may accept any number of object fields (see thes.\ref{thes:SimplificationByDiffingHeaps}).
Its boolean value does not depend on concrete fields.
However, compound heaps combined with these may cause unexpected behaviour, which can be excluded during context analysis.
Partial specifications, which use constant functions, partially describe fields, but they may cover much more.
Although a specification is partial, its denotation includes more heaps (see thes.\ref{thes:IncompletenessForCompleteness}).
Moreover, a calculation may compare given heaps with expected ones, and it can localise heaps that are missing.
This shortcut allows determining heaps in the rule set fully.

\begin{example}[Incomplete Predicate no.1]
Given an object $a$, which has three fields, $f_1$, $g_1$ and $g_2$. \comp{.}{} denotes \index{semantic function} a semantic (implicit) function over heap terms of type $ET \rightarrow ET \rightarrow \mathbb{B}$, where $ET$ from def.\ref{def:HeapTermExtendedDefinition}, $\mathbb{B}$ is the boolean set, where the first extended heap is the expected and the second is the obtained heap, then:
 \begin{eqnarray*}
   & \comp{$a.f_1 \mapsto x\circ \underline{true}(a)$}{} = & \comp{$a.f_1 \circ a.g_1 \circ a.g_2$}{}\\   
   = & \comp{$\underline{true}(a) \circ a.f_1 \mapsto x$}{} \neq & \comp{$p(a) \circ a.f_1 \mapsto x$}{}
 \end{eqnarray*} 
where $p$ is an AP and denotes $\underline{true}(a)$ (see def.\ref{def:IncompletePredicates}).

However, \comp{$a.f_1 \mapsto x \circ p(a)$}{} would mean equality because recognition, based on the actual \index{stack!window} stack window, finds all remaining fields even if several layers of abstracted calls are hidden beneath.
  \comp{.}{} denotes a \index{homomorphism} homomorphism w.r.t. the discussed constant functions and conjunction $\circ$.
\end{example}

\begin{example}[Incomplete Predicate no.2]
 \comp{$\underline{true}(a) \circ \underline{true}(a)$}{} = \comp{$a.f_1 \circ a.g_1 \circ a.g_2$}{} $\circ$ \comp{$\underline{true}(a)$}{} = \comp{$a.f_1 \circ a.g_1 \circ a.g_2$}{} $\circ \ \underline{emp}(a)$.
\end{example}

\subsection{Discussion}
\label{sect:StricterDiscussions}
Two stricter operations replaced one spatial operation.
The SL did not change initial properties, except for unbound heap inversion, which was discussed.
If replacing $\star$ with $\circ$ in theo.\ref{theo:ReynoldsHeapProperties} (for disjunction replace with $||$), correctness follows immediately due to the exclusion of axiom no.5 for conjunction because of strengthening $\star$.

Given a form allowing normalisation of heap terms, then specifiable heaps may now be analysed linearly (cf. sec.\ref{sect:PartialSpec}).
Due to the \index{non-repetitiveness} non-repetitiveness of simple heaps, complicated heaps may effectively be excluded by applying \index{memoisation} memoisation.
In practice, repeating locations must be localised and excluded (potentially with differing syntactic but semantically equal expressions), as was mentioned in sec.\ref{sect:ConjunctionAndDisjunction}.
Otherwise, heaps and heap graph contradicts as well as SL-properties may be infringed.
The essential principle of non-repetitiveness must be obeyed.
A simple heap specifies one time only.
If principles are not correctly obeyed, then applying rules may lead to incorrectness.

Consequently, from a wrong precondition, any result may be derived, including the correct one.
Hence, repeating heaps must be excluded.
However, this problem may be resolved without running the program (just the verifier is sufficient).
APs may be interpreted procedurally (see sec.\ref{chapter:APs}).
Consequently, a \index{stack!architecture} \textit{stack architecture} is used for its analysis (cf. sec.\ref{chapter:logical}).
Its semantics is very close to \index{WAM} Warren's \cite{warren83} \index{operational semantics} operational semantics, except for \index{parameter by call} parameters by call.
The minor modification of \index{WAM} WAM may be applied to AP interpretation using strengthened \index{conjunction} conjunction and \index{disjunction} disjunction.
For instance, this allows recognising $\forall a \in \Omega.a \circ a$ to check whether non-repetitiveness is obeyed. 
The \index{memoisation} memoiser can only \index{caching} cache those calls of \index{abstract predicate} APs that do not alter the global calculation state.
The memoiser can recall (i) \index{term!incoming} incoming or (ii) \index{term!outgoing} outcoming or (iii) \index{term!incoming/outgoing} both incoming and outcoming variable terms at the same time.
If further unbound symbols inside an AP by some later subgoal are restricted, this must be considered in the order of subgoals (definition must follow \index{evaluation ordering} the LR order, see sec.\ref{chapter:APs}).

Prolog \cite{sterling94} as a general-purpose logic PL may be used (see thes.\ref{thes:PrologMakesHeapSpecSimpler}) as a platform based on \index{recursive scheme} recursively-defined rules and terms for \index{logical reasoning} logical reasoning (cf.sec.\ref{chapter:logical}).
It is based on theorems about heap term extension and APs.
In Prolog, \index{Peano arithmetics} the well-defined  generalised induction scheme after Peano may always be defined as "\texttt{p(0).}" for base cases, and as "\texttt{p(n):-n1 is n-1,p(n1).}" for the inductive cases using \index{predicate!auxiliary} auxiliary predicate \texttt{is} for determining \texttt{n1}.
It is not hard to see in Prolog any generalised \index{$\mu$-recursive scheme} $\mu$-recursive predicate may be expressed.
The proof is straightforward, constructive and direct and is not special at all.
Hence it is skipped here.
In general, predicates do not need to be determined (e.g. when a procedural call does not terminate).

The advantage of \index{Horn-rule} Horn-rules is \index{Prolog} Prolog, in contrast to classic \index{one-way function} one-way functions as used by \cite{parkinson05-2}, can interpret terms as (i), (ii) or even (iii).
By doing so, it effectively unions exponential many different classic one-way functions.
Not every aspect may be defined.
Therefore, \index{function inversion} function inversion requires further attention.
\index{arithmetic} Arithmetical calculations and \index{green cut} \textit{green} and \index{red cut} \textit{red cuts} \cite{sterling94} are possible reasons why a predicate may not be invertible \cite{haberland08-2}.
Invertibility of arithmetic expressions may partially be compensated by replacing \index{natural numbers} natural numbers with Church-terms (see \cite{haberland08-2}, cf. fig.\ref{code:NaturalNumbers}).
Basic operations over naturals are invertible if only using constants, \index{arity} \index{monad} monad-styled unitary term-constructors and \index{unification} unification is used as its main basic operation.
Generally speaking, there must be a strong correlation between input and output terms, becoming an isomorphic mapping between both domains.
Moreover, Prolog's cuts may be replaced with w.l.o.g. and expressibility since they are essentially only \index{syntactic sugar} syntactic sugar (see sec.\ref{chapter:expression}).\\

Expressions in \textit{OCL} \cite{oclspec} are formalised and kept intentionally close to predicate logic.
It supports \index{variable!quantified} quantified variables, arrays and ADTs, particularly classes and "\textit{ad-hoc}" \index{polymorphism} polymorphism by subclassing.
OCL allows specifying the life-cycles of objects and their methods.
However, OCL apriori does not know about \index{pointer} pointers nor \index{alias} aliases.
Expressions about pointers are missing.
Therefore, UML/OCL may currently not be used to model pointer involvement, e.g. in \index{prototyping} \index{rapid prototyping} rapid prototypes, as described in sec.\ref{chapter:intro} (see \cite{haberland14-1}, cor.\ref{cor:StrengtheningOCL}).

Hence, it is proposed that pointers are introduced to OCs, namely into UML/OCL directly.
An object state matches how a stack is executed, and heaps are processed, so it fits in the context of the pre-existing object model in UML/OCL.
Complex objects have fields of simple or complex type.
APs are proposed to be introduced to OCL's logical predicates (see sec.\ref{chapter:logical}).
There is no obligation why one predicate describes only a single object.
Nevertheless, it is strongly recommended to follow this recommendation.
Exceptions may be made when intricate behavioural patterns are described, such as with the Observer design pattern.
Spatial relations between objects is described by $\circ$ and $||$.
Relations do not insist on spatial components if clear from the context it is about OCL.
Future work may involve the integration of \index{abstract predicate} APs into OCL \cite{haberland16-1}.
The begun proposition so far shall be researched further, especially considering conv.\ref{conv:RestrictedObjects} and def.\ref{def:HeapTermExtendedDefinition}.
An expressibility in modularity, increase, and abstraction are expected.
A useful review in calculus extensions with pointers may be found in sec.\ref{chapter:intro} and \cite{parkinson05-2}.

\newpage
\section{Automated Verification with Predicates}
% (40p.)
\label{chapter:APs}

Let us consider a \index{doubly-linked list} doubly-linked list as an example from \index{geometry, computational} Computational Geometry \cite{deberg08}.
This \index{doubly-connected edge list} list will have for each \index{graph edge} edge two adjacents for a polyhedron in 2D or 3D space representing a \index{mesh} \index{polyhedron mesh} wired mesh.
After \index{triangulation} triangulation, this mesh will consist of \index{triangle} triangles.
So, each \index{graph!vertex} vertex will be connected to precisely two neighbouring vertices.
Each edge starts and ends with well-defined vertices.
A \index{normal vector} normal vector is required for visibility and shading calculations of geometric objects.
It can be determined by any two edges from any polygon, e.g. a triangle.
In a \index{heap!base} simple heap, the location \index{location} points to its content.
For the example given, this may be vertices or edges.
According to the doubly-linked list, each element has pointers forward to the next edge and a pointer back to the previous edge.
A simple heap contains edges that are interconnected.
Unconnected edges are not specified by default.
It is agreed upon that a mandatory absence of heaps is either excluded (see sec.\ref{chapter:stricter}) or is not considered separately due to Abelian group properties and previously mentioned restrictions.
Each time a \index{heap copy} pointer is dereferenced, copies of all three vertices are allocated to the memory.
It is not hard to see this attempt would be pretty inefficient.
When working with \index{pointer} pointers, this problem is dropped, especially when additional \index{abstraction} abstraction is introduced.
These abstractions, namely APs, allow expressing \index{heap!complex} complicated heaps intuitively.
For example, the Prolog subgoal \texttt{face(p1,p2,p3)} may denote three vertices $p1,p2,p3$, which are connected by a triangle instead of specifying again and again fully $\exists v1.v2.v3$, where \texttt{p1.data $\mapsto$ $v1$} $\star$ \texttt{p2.data $\mapsto$ $v2$} $\star$ \texttt{p3.data $\mapsto$ $v3$} $\star$ \texttt{p1.next $\mapsto$ p2} $\star$ \texttt{p2.next $\mapsto$ p3} $\star$ \texttt{p3.next $\mapsto$ p1} $\star$ \texttt{p1.prev $\mapsto$ p3} $\star$ \texttt{p3.prev $\mapsto$ p2} $\star$ \texttt{p2.prev $\mapsto$ p1}.

First of all, abstract means generalisation by the introduction of additional parameters.
AP here implies a Horn-rule with an arbitrary number of parameters.
Some authors tend to use instead \index{abstract predicate} "\textit{Abstract Predicate}" as a new term as by \cite{parkinson05}.
\index{abstraction} Abstraction does neither require a new definition (see sec.\ref{chapter:expression}) nor the introduction of a new concept because, in both cases, it is already there --- the same counts for predicates.
Well understood concepts as abstraction and predicates do not require new definitions, mostly since \index{predicate!parameterisation} predicates' parameterisation is already covered by \index{predicate logic} predicate logic.
That is why "\textit{abstract}" is an adjective to (genuine logical) predicates.
Essentially, an abstract predicate does not differ from a predicate since there is not a single semantic justification for distinguishing from a predicate, especially not after applying Prolog-dialect to reasoning over heaps.
An AP has an arbitrary number (including nil) of term parameters and may contain an arbitrary number (including nil) of subgoals to built-in or previously declared predicates.
In the given example \texttt{face(p1,p2,p3)} equals the unfold \index{$\star$} $\star$-formula of subgoals mentioned earlier.
Depending on which state of calculation the subgoal \texttt{face} is, it may either fold or unfold.
The predicate \texttt{face} may depend on further predicates.
However, currently, insufficient information is available on when to fold or unfold the AP definition. 
If unfolding fails, then this might be because it is impossible to do so in general or because unfold and fold were performed at the wrong time or in a wrong order.
Next, a new approach is considered, which allows us to resolve this issue with heaps automatically.

Warren \cite{warren83} uses the term \index{programming by proving} "\textit{programming by proving}" to emphasise Prolog as a PL that might be used to find a solution for a problem formulated as \index{Horn-rule} Horn-rules.
The isomorphism \index{Curry-Howard isomorphism} by \textit{Curry-Howard} \cite{mitchell96} states there is a dependency between proof and programming.
W.r.t. heaps, the philosophical motto behind might be characterised as "\textit{proving is a syntax problem}", which means that \index{syntax analysis} syntactic analysis may prove the soundness of a specifiable heap and heap representation is close to programming models -- in fact to \index{Horn-rule} Prolog's rules.
The major clue in this section is that APs describe a \index{language!formal} formal language (see thes.\ref{thes:PrologMakesHeapSpecSimpler}).
Later we will see that the formal language is a logical dialect.
Consequently, problems with the heap, which by its nature are mostly semantical, may be resolved by syntactic methods and recognition.

\subsection{Fold and Unfold}

The approach presented in this section heavily deviates from existing traditionalistic ones.

According to Reynolds \index{memory model!Reynolds} \cite{reynolds02,parkinson05-2}, the spatial operator $\star$ \index{$\star$} conjuncts two separated heaps.
By doing so, the \index{substructural logic} properties of the underpinning \textit{substructural logic} remain, and \index{rule!thinning} the \textit{thinning rule} does not hold.
Sec.\ref{chapter:expression} shows the classic operator $\star$ is ambiguous.
That is why in the following, only the stricter conjunction operator $\circ$ is meant to replace its spatial predecessor predicate.

APs \index{abstract predicate} are used by the user having the specification role the same way as proposed by Verifast \index{Verifast} \cite{jacobs11}.
The \index{conclusion!logical} conclusion may be achieved faster by \index{tactics} tactics, which may be defined inductively in the verifier \index{Coq} Coq \cite{bertot04} and is derived in a \index{conclusion!semi-automated} \textit{semi-automated manner}.
Systems are based on the fold/unfold \index{fold} principle, which first was proposed by Hutton \cite{hutton98}.
Hutton's approach would fully automate if a complete proof hint list were provided.
Proof hints \index{proof!hint} indicate which \index{redex} \textit{redex} shall be used next to get a proof finished and finished quickly.

Prolog \index{Prolog} is used as \index{language!assertion} assertion language.
\cite{kowalski74} illustratively demonstrates the essence use of \index{Horn-rule} Horn-rules as a concrete instance of formulae in predicate logic.
\cite{kowalski74} proposed Prolog to be used as a \index{language!programming} PL, which, unfortunately, is not always possible due to computability (see sec.\ref{chapter:logical}).
\cite{warren83} contains a logical reasoning implementation based on \index{semantics!operational} operational semantics.
Kallmeyer \cite{kallmeyer10} demonstrates Prolog as a recogniser of \index{morpheme} \textit{morphemes} and a so-called \index{grammar!mutation} grammar mutation by the example of some \index{language!natural} natural languages.
Notably, \index{adjoint-trees} adjoint trees are proposed as a processing technique in natural languages based on explicitly defined $\lambda$-terms, which do not occur in \index{language!formal} formal languages, particularly in PLs avoiding default \index{ambiguity} ambiguities.
Examples show overloading and having difficulties in the analysis is due to the exponential growth of required checks of side-effects.
\cite{matthews98} shows in examples of how Prolog might help to resolve ambiguity problems using syntactic analysis.
Matthews widely uses \index{recogniser!recursive} recursively-descendent recognisers, which work on \index{syntax tree} trees and are implemented very effectively and simple in Prolog.
The implementations base on stack automata recognising \index{recogniser!top-down} LL(k)-grammars with modifications: finite states \index{automaton!finite} are explicitly defined, and the stack mimics recursive Prolog rules.
Matthews uses \index{difference lists} difference lists to implement \index{language!regular} regular languages based on an automaton of \index{automaton!partial-derivatives} partially derived regular expressions \cite{brzozowski64}.
Brzozowski proposes Prolog's formal grammar utility DCG \index{grammar!formal} \index{DCG} \index{Prolog!Definite Clause Grammars} and built-in subgoals to Prolog \index{built-in command} for updating knowledge bases \index{knowledge base} he believes shall be extended.
His approach is bound to regular expressibility (cf. sec.\ref{chapter:expression}).
Pereira's contribution \cite{pereira12} may be assessed as a classic on Prolog and natural language processing.
However, there are no apparent limitations.
These include certain unsupported pattern sets and diverging examples, but those can easily be bypassed.
For example, this refers to the inability to effectively resolve \index{recursion!left} left-recursion  in given Horn-rules since LL(1)-recognisers cannot recognise all \index{token!predecessor} predecessor tokens to determine which rule was applied.
Furthermore, Pereira proposes \textit{$\lambda$}-calculi over trees in order to recognise natural languages.
Thus, \index{morpheme} morphemes and \index{lexeme} lexemes are linked and get a dependent denotation by introduced parameters.
Next, we introduce the first prototypical heap definition obeying def.\ref{def:FirstOrderPredicateLogicFormula} and sec.\ref{chapter:expression}, which can immediately be defined in Prolog.

\subsection{Predicate Extension}
\label{sect:PointsToHeaplets}

\begin{definition}[Heap Assertion] A \textit{heap assertion} $H$ is inductively defined as
\begin{center}
\begin{tabular}{ll}
H ::= & $\textbf{emp} \ | \ \textbf{true} \ | \ \textbf{false} \ | \ x \mapsto E \ | \ H \star H$\\
      & $ | \ H \wedge H \ | \ H \vee H \ | \ \neg H \ | \ \exists x.H \ | \ a(\vec{\alpha})$
\end{tabular}
\end{center}
\label{def:HeapAssertionDefinition}
\end{definition}

Assertion \underline{$emp$} denotes an \index{heap!empty} empty heap, which is true if the passed heap is empty (cf. sec.\ref{sect:ExprPredicates}).
Empty heap behaves neutral towards spatial conjunctions (see sec.\ref{chapter:expression}, sec.\ref{chapter:stricter}).
The operator $\star$ separates two heaps into two \index{heap!independent} independent heaps.
For the sake of simplicity, this section does not stress the specific strictness w.r.t. $\circ$ or $\circ$ (see sec.\ref{chapter:stricter}), although it shall be taken into account. 
However, stricter operations shall always be preferred over $\star$.
Assertion \underline{$true$} \index{heap!assertion} denotes any heap (including empty) is accepted, where \textbf{false} always fails regardless of which heap is about to be passed for interpretation.
These definitions are close to Reynolds' \cite{reynolds02}.
Reynolds' foundation to all definitions is the \index{heap!base} \textit{ordinary (simple) heap}: $x \mapsto E$, where $x$ denotes some \index{location} \textit{location} (e.g. object field accessor $o1.field1$), and $E$ is some valid expression that is assigned to a memory cell at location $x$.

Type checking is done at an earlier stage (see \cite{haberland14-1}).
Currently, it does not matter whether an immediate value or a pointer to some memory cell is in memory (see \cite{burstall72}).
Let us consider two arbitrary \index{heap!complex} compound heaps in Prolog as depicted in fig.\ref{ExampleComplexHeaps1}.

\begin{figure}[h]
\begin{center}
\begin{minipage}{10cm}
\begin{verbatim}
p2(X,Y):-pointsto(loc2,X),pointsto(loc3,Y).
p1(X,Y):-pointsto(loc1,val1),p2(X,Y).
\end{verbatim}
\end{minipage}
\end{center}
 \caption{Example of compound heaps}
 \label{ExampleComplexHeaps1}
\end{figure}

Here $p2$ denotes some predicate with two \index{symbol} symbols $X$ and $Y$, representing some value pointed by locations $loc2$ and $loc3$.
In contrast to this, $p1$ is defined by predicate $p2$.
As we call $p2$ with two syntactically correct arguments, we obtain the $a(\vec{\alpha})$ form.

An interpretation of a \index{heap!interpretation} heap formula $H$ for a given heap denotes a mapping of two heaps, namely the calculated and the expected heap, onto the \index{boolean set} boolean co-domain.
It means, if two given heaps are equal, then interpretation succeeds.
Otherwise, it fails.
Only deductively obtained conclusions are considered for a given (heap) interpretation (see sec.\ref{Intro:HoareTriple} and sec.\ref{Intro:LogicalReasoning}).
A search for a logical conclusion succeeds when a query succeeds.
In all other cases, it fails.
Without any doubt, this is what is expected to obtain from the expected behaviour (see thes.\ref{thes:ProvingEqualsParsing}).\\
For simplicity, it is agreed upon that the heap formulae must be \index{heap!canonisation} canonised by lexicographical ordering to
$$a_0 \star a_1 \star \cdots \star a_n   \equiv \prod^n_{\forall j} a_j,  n \ge 0$$
The $\wedge$ and $\vee$-connected heap graphs are noted in \index{Prolog!query} Prolog as queries in the form $s_j, s_{j+1}, \cdots , s_{j+k}$.
Alternatively, \index{Prolog!disjunction} disjunctions of Prolog subgoals may be split into alternatives whose heads differ from each other using the operator "\textbf{;}".
Assertion negation \index{assertion negation} is considered \index{predicate negation} predicate negation.
In general, sequence negation does not mean predicate negation.
Double negation, in general, does not hold w.r.t. subgoals, especially not when a predicate is not defined totally.
Existential variables may be introduced at any place in Prolog rules.
However, it is expected that all introduced variables are bound and used.

The \index{constant function} constant functions \underline{$true$} and \underline{$false$} are  syntactic sugar since these may be rewritten by other rules specifying concrete heaps that would qualify.
However, constant functions simplify specifications towards any matching heaps.
\underline{$true$} may be replaced by \index{true} boolean truth.
\underline{$false$} may be replaced by \index{false} boolean \index{predicate!contradiction} contradiction.

\begin{corollary}[Heap Soundness]
Any syntactically correct heap formula describes a corresponding \index{heap graph} heap graph.
Moreover, any heap graph may be represented by a corresponding \index{heap formula} heap formula.
In general, both sides hold, except infinite heaps, which intentionally are not considered.
\end{corollary}

\begin{proof}
Soundness holds, except APs (see later).
Def.\ref{def:HeapAssertionDefinition} may be applied to theo.\ref{theo:ReynoldsHeapProperties}, def.\ref{def:HeapSatisfactionRelation} and in consequence to def.\ref{def:HeapTermExtendedDefinition}.
For the sake of simplicity and for historic reasons of publication ordering, here we are assuming $\star$ instead of $\circ$ --- as was already mentioned.
\end{proof}

\begin{definition}[Informal Heap Graph sketch]
\label{def:NonformalHeapGraph}
 A heap graph denotes a \index{connected graph} connected graph that is directed, simple and resides \index{dynamic memory} in the dynamic memory segment.
Dynamic regions are allocated and freed by the \index{OS} OS.
They can be altered in a program by wrapped special-purpose function calls to some virtual C-library (which is not of interest in this work and therefore not further considered).
Graph vertices are pointed by at least one \index{variable!local} local variable (as a pointer) or are addressable by available locations.
Each graph vertex is corresponding to some address inside the dynamic memory segment.
The width of each vertex is variable but fix and depends on the \index{type variable} \index{pointer} pointer type. 
When one vertex points to another, then both vertices are neighbouring in the corresponding heap graph.
If a vertex is associated with two pointers, then the first pointer \index{alias} \textit{aliases} the second pointer.
\end{definition}

Def.\ref{def:NonformalHeapGraph} is a specialisation of obs.\ref{obs:GraphIRHeap}.

\subsection{Predicates as Logical Rules}

\index{abstract predicate} APs allow abstraction from simple heaps to complex ones, making them more intuitive.
For example, \cite{parkinson05-2} introduces APs, which annotate a given incoming program and are translated together with program statements till assembly output is generated.
The presented approach automatically is an attempt to overcome the gap between specification and logical reasoning.
Prolog in this section is used as a PL, in which assertions and APs over heaps are specified.
However, a semantic mismatch exists between program and specification language: two formalisms exist simultaneously and often two not related implementations using different IRs.
Diverging notations and representations are the consequence.
\index{language!imperative} PLs may and sometimes even must differ from the exemplary language, e.g. chosen in this work.
Following \index{logical formula} Hoare calculi, logical formulae are described in predicate logic, so in a logical paradigm.
Unfortunately, sometimes this is violated (see sec.\ref{chapter:intro}), and further restrictions are unavoidable.
Thus, for example, variables (objects) are used as \index{variable!typed} locals instead of terms, resulting in many restrictions and workarounds to achieve real-term behaviour. 
Thus, a \index{language!specification} specification language, namely its units, are degrading to a statement sequence, and nothing else is left in common with the initial idea of a Hoare calculus. 
Units would not have anything in common with the declarative-logical paradigm (see obs.\ref{obs:EqualityOfProofElements}, obs.\ref{obs:StackBasedCalls}): the calculation state needs to be described referring to symbols and predicates.
Symbol and its \index{scope visibility} visibility scope differ from local variables.
If symbols are suddenly redefined, then those are not symbols anymore.
By definition, a symbol's value is bound exactly once compared to conventional non-constant variables as known from imperative PLs.
The difference \index{paradigm} \index{paradigm!imperative} \index{paradigm!declarative} may not look that big, but this is already a severe unpassable barrier for symbol and predicate definitions.
Logical algorithms often (but not always) rely purely on symbols and not on variables.
Terms \index{term} and \index{predicate} predicates are vital concepts in logical reasoning.
In both models, heap specification and verification, a minimal Turing-computable subset (cf. PCF \cite{mitchell96}) of Prolog (cf. LCF \cite{plotkin77}, \cite{cohn83}) may be chosen.
Verification requires comparing calculation states to decide the proof finished adequately.
The heap comparison does not exclude \index{arithmetic expression} arithmetic expressions, but Hoare's main characteristics remain declarative initially intended (cf. sec.\ref{chapter:intro}).

The approach in \cite{berdine05-2} introduces \index{heap!with symbols} symbols on heaps but with problematic limitations, e.g. the possibility to describe whole heaps as $X \star Y$ is not foreseen (full restrictions are listed in sec.\ref{chapter:intro}).
Here, we allow symbols to remain arbitrarily in a logical sense, so as allowed by Prolog (see sec.\ref{chapter:logical}).

When it comes to non-determinism with rule selection, Berdine \cite{berdine05-2} suggests selecting the longest as a heuristic.
We do the same, which is naturally achieved by prioritising rules by their position in a Prolog listing.
In general, the search strategy of Prolog is preferred for proves.
In cases when a specified depth shall not exceed, \index{branch and bound} branch-and-bound shall be taken into consideration \cite{sterling94}, \cite{bratko01} (cf. obs.\ref{obs:DeductionWithBacktracking}).
\index{WAM} WAM \cite{warren83} is the closest approach to functional or imperative in comparison.
When we describe heap state checks, it is easier to refer to \index{fact} facts and \index{logical rule} rules than program statements or operational semantics (cf. sec.\ref{chapter:expression}).
The program developer is responsible for a compact programmer, the specifier for facts and rules.
The following formalisms help in transforming APs into \index{grammar!formal} formal grammar used for further representation.
All encountered advantages are characteristic for Prolog and are used further in the method proposed.

\begin{definition}[Predicate Rule]
\label{def:AbstractPredicateRule}
  A \textit{predicate rule} is defined as $\forall n.a:-q_{k \times n}.$ $\Leftrightarrow$ $a:-q_{k,0}, q_{k,1}, \ldots, q_{k,n}.$ for $k \in \mathbb{N}_0$, where $a$ denotes the predicate head, $q_{i,j}$ denotes some subgoal $\forall i,j$.
\end{definition}

It is intentional def.\ref{def:AbstractPredicateRule} is close to def.\ref{def:PrologRule}.
The first definition is used for descriptions and calls to heap predicates.
By default, it is agreed upon $a$ is valid iff all \index{subgoal} \textit{subgoals} $q_{k,j}$ in $a$ hold for $0 \leq j \leq n$.
The EBNF from fig.\ref{fig:PrologRulesEBNF} defines the syntax for the predicate rules.
\texttt{<number>} denotes any valid Prolog number, \texttt{<atom> '(' <arguments> ')'} denotes some \index{functor} functor with some simple name \texttt{<atom>} and an arbitrary number of arguments.
\texttt{<var>} denotes some symbol variable, which starts with a capital letter, e.g. $X$.
Fig.\ref{fig:mapFunctionalExample} in sec.\ref{chapter:expression}, and sec.\ref{chapter:logical} both present an example of the useful functional "\texttt{map}" in Prolog.
Using this kind of functionals, together with "\texttt{fold}" and "\texttt{concat}", often allows very compact notations, including recursion-free notations (cf.\cite{haberland08-2}).

\begin{figure}[t]
\begin{center}
\scalebox{0.77}{
\begin{tabular}{c}
\begin{minipage}{11cm}
 \begin{grammar}
 <predicate> ::= <head> [ ':-' <body> ] '.'
 
 <head> ::= <atom> [ '(' <arguments> ')' ]
 
 <body> ::= <sub_goal> \{ ',' <sub_goal> \}*
 
 <sub_goal> ::= '!' | 'fail' | <functor_term> | <term> <rel> <term>
 
 <rel> ::= '=' | '\textbackslash =' | '\textless' | '\textless=' | '\textgreater' | '\textgreater='
  
 <functor_term> ::= <atom> '(' [ <arguments> ] ')'
 
 <arguments> ::= <term> \{ ',' <term> \}*

 <term> ::= <atom> | <var> | <list> | <number> | <functor_term>
 
 <list> ::= '[' [ <term> '|' ] <arguments> ']'
\end{grammar}
\end{minipage}
\end{tabular}}
\end{center}
 \caption{\index{EBNF} EBNF for logical rules}
 \label{fig:PrologRulesEBNF}
\end{figure}

Let $a$ denote some predicate, then its subgoals $q_{k,j}$ are evaluated LR.
The symbolic environment $\sigma$ \index{symbols environment} is updated after each subgoal call inside a \index{predicate!body} predicate's body, according to fig.\ref{fig:BoxModelPredicateCall}.
Non-assigned symbols remain but may be assigned after a subgoal.
The resulting terms of previous subgoals do not require updates as symbols are assigned compatible value.
The semantics of a \index{predicate!semantics} predicate call is defined as:

$$C(a)\llbracket a(\vec{y}):-q(\vec{x}_{k,n})_{k \times n} \rrbracket \sigma =$$
$$D \llbracket q_{k,n} \rrbracket \sigma(\vec{x}_{k,n}) \circ \cdots \circ D \llbracket q_{k,1} \rrbracket \sigma(\vec{x}_{k,1})$$

By default, the \index{term vector} term vector $\vec{y}$ may contain common elements with vector $\vec{x}_{k,n}$, $\forall k,n$, and $C\llbracket . \rrbracket$ of type $atom \rightarrow predicate \rightarrow \ \sigma \ \rightarrow \sigma$, $D\llbracket . \rrbracket$ has type $subgoal \rightarrow \sigma \rightarrow \sigma$, and $\sigma$ is of type $term^\star$ $\rightarrow term$, where $\star$ denotes \textit{Kleene's star} \index{Kleene's star} \cite{davis94}.

The subgoal $q_{k,j'}$ does not necessarily have to define apriori the \index{connected graph} connected graph.
However, if so, then this indicates the heap (sub-)set is fully defined.
At least, it is complete regarding the intuitive understand of a data structure.
\index{modularity} Modularity and \index{separation of concerns} "\textit{separation of concerns}" may always be considered a good virtue of program development.
Thus, an AP forces a designer sooner or later to describe a data structure's intention.

Consequently, it may imply that one AP shall correspond to one heap.
Next, adding more and more $\star$-conjuncts, the corresponding heap graph grows continuously.
$\star$-conjuncts build up heaps, possibly connected ones corresponding to APs.
When reviewing fold and unfold \index{fold} \index{unfold} of APs (in analogy to procedures), it can be noted it has parameters, heap graph vertices to be precise, which may be on both sides, on the callee's side as on the caller's side of a predicate call.
Also, vertices may exist, visible only inside a predicate, which may not be used outside (at least not directly).

W.l.o.g. it is agreed upon object access is \index{field accessor} granted by the functor "\textbf{.}", e.g. \texttt{a.b} (see fig.\ref{fig:PrologRulesEBNF}) or \texttt{oa(object5, fld123)} \cite{haberland14-1}.
For the sake of demonstration and modularity, it is further agreed upon an object with its fields can be used to predicate parameters as well as part of their subexpressions.
The difference between objects and ordinary \index{variable!local} stack-local parameters is nil.
A detailed explanation is given later.

\begin{definition}[Predicate Rule Set]
\label{def:PredicateRuleSetDefinition}
A \textit{predicate rule set} (also called "partition" or "family") $\Gamma_a \subseteq \Gamma$ for some predicate named $a \in T$ and $\forall i.j. q_{i,j} \in (T\cup NT)$, where $T$ are terminals, and $NT$ denotes non-terminals (NT), is defined as:
\begin{center}
\begin{tabular}{ll}
$\Gamma_a ::=$ & $a:-q_{m \times n}$\\
 & $\Leftrightarrow
\begin{array}{llllllll}
 a:-    & q_{0,0} & , & q_{0,1} & , & \ldots & , & q_{0,n}\\
 \vdots & \vdots  &  & \vdots  &  & \ddots &  & \vdots\\
 a:-    & q_{m,0} & , & q_{m,1} & , & \ldots & , & q_{m,n}
\end{array}
$
\end{tabular}
\end{center}

If $n = 0$, then $a$ is a \textit{fact}.
Terms may be associated with $a$ (containing symbols, for instance, when $m=0$, $n>0$).
If $t \in T$, then $t$ is of kind $loc \mapsto val$, otherwise $t \in NT$ denotes a predicate named $t$ in $\Gamma$.
\end{definition}

By default, it is agreed upon the sequence $q_{k,0}, q_{k,1}, \ldots , q_{k,m}$ from $q_{m \times n}$ any line is in \index{rule canonisation} canonised form, s.t. for $\exists s\leq m$ non-trivial elements the first $s$ subgoals are positioned, and for all remaining $m-s$ subgoals it is a \index{tautology} tautology as subgoal each, whose \index{domain} domain is fully defined as truth ($\top$).
Furthermore, it is agreed upon that $$\exists k.a:-q_{k} \preceq a:-q_{k+1}$$ holds, meaning that the predicate listed earlier in $\Gamma_a$ has higher priority than the priority listed later.

\begin{corollary}[Predicates Environment]
\label{corollary:PredicateEnv}
For the predicates environment $\Gamma$ for a given Prolog program, $\Gamma = \bigcup_{t \in T} \Gamma_t$ holds.
All predicates $\Gamma_t$, which depend on each other, must be located in a closed \textit{section of a predicate} $\overline{\Gamma_t}$. $\overline{\Gamma_t} \subseteq \Gamma$ holds.
\end{corollary}
\begin{proof}
The main idea is to show: all dependent $\forall t.\Gamma_t$ are in one partition of the predicate, and all independent partitions are independent of the dependent \index{predicate environment} predicate environments.
All predicate environments are located in $\overline{\Gamma}$, including both dependent and independent ones.
Predicates $\Gamma_a$ and $\Gamma_b$ from non-neighbouring partitions from $\overline{\Gamma}$ can never depend on each other.
\end{proof}
\textbf{Remark:}
Obviously, due to the \index{Halting-problem} Halting problem, the predicate call from a partition is, in general, not decidable.
Next, the expressibility of predicates is considered.

\textbf{Remark:}
Naming clashes in $\Gamma$ may be dispatched by choosing a proper encoding.
The predicate \index{location} location shall be encoded into the name, for example, the class for which the predicate is foreseen is prefixed, so it becomes possible to distinguish predicates.
By default, predicates with equal \index{name clashing} location are part of the same predicate partition, so there is no more conflict.

\begin{lemma}[Predicate Completeness regarding Expressibility]
\label{lem:APsAreFOPs}
APs cover all first-order predicates for describing a heap.
\end{lemma}

\begin{proof}
A complete proof on expressibility of first-order predicates in Prolog is already done.
It can be found in \cite{kowalski74}.
\end{proof}

\begin{lemma}[Higher-order Predicate Expressibility]
\label{lem:APsAreSOPs}
APs may express \index{higher-order logic} second-order and even higher-order.
\end{lemma}

\begin{proof}
%%%%
First, let us briefly stick to the essence of \index{expressibility} expressibility of first-order Prolog predicates.
Next, we increase the expressibility level.
The higher-order differs from the first-order by further predicates abstraction.
So, e.g. a predicate may be used as a logical expression or variable.
In Prolog, this may happen during calls using some \index{predicate!built-in} built-in predicate \texttt{call}, which accepts \index{term!incoming} incoming and \index{term!outgoing} outgoing terms, e.g. \texttt{pred1(X):-call(pred2,X)}.

Given some predicate \texttt{P} with arguments: list of incoming terms \texttt{[X|Xs]}, list of outgoing terms \texttt{[Y|Ys]} and some terms list used for incoming and outgoing terms.
W.l.o.g. the last mixed-use term list may be empty for the sake of induction here.
The predicate "\texttt{map}" may then be defined as shown in fig.\ref{fig:mapFunctionalExample}, which accepts a predicate to each incoming term consecutively \index{calculation ordering} in the LR order.
Higher-order predicates may be useful in modules and for single compact data structure accessors.
For example, class-based objects alter while execution, according to the CFG.
Here, even the control may be bidirectionally passed between caller and callee instances, as is the case with most \index{pattern} \textit{behavioural patterns}, e.g. the \index{pattern} Observer pattern \cite{kerievsky05}, \cite{krishna09}.
The pattern behaviours are very important use-cases that decide usability of proposed heap definitions, as was investigated by \cite{krishna09}, \cite{krishna08}, \cite{pottier08}.

\begin{figure}[t]
\begin{center}
\begin{tabular}{l}
\begin{minipage}{7cm}
\begin{verbatim}
map([],P,[]).
map([X|Xs],P,[Y|Ys]) :- 
    Goal =.. [P,X,Y],
    call(Goal), map(Xs,P,Ys).
\end{verbatim}
\end{minipage}
\end{tabular}
\end{center}
  \caption{Functional map/3}
  \label{fig:mapFunctionalExample}
\end{figure}

The predicate type of \texttt{map/3} is $$list_a \rightarrow (list_a \rightarrow list_b) \rightarrow list_b$$, so the second incoming type makes a single predicate $a$ \index{abstract predicate family} a predicate family, which accepts a list of some base type as input, and base type $b$ as output.
There is no hint, whether $a=b$ and $a \ne b$.
By $list_a$, we denote a dependent Prolog definition of a list whose elements are homogeneous of type $a$ --- third-order $\lambda$-terms (cf. def.\ref{def:TypedLambda2ndOrder}).
Thus, recursion may be replaced by the third- and higher-order predicates, for example, by using \index{left folding} \textit{left folding} \index{fold} (\textit{foldl}), which applies some \index{predicate!higher-order} predicate $\oplus$ to a given list of incoming terms to the current intermediate result of the calculation.
When starting with the given neutral element,

 $$\texttt{foldl($\oplus::a\rightarrow b \rightarrow a, \ \varepsilon::a, \ X::list_b$)::a}$$

right-folding works in analogy, starting with the right-hand side and continuing calculation in right-to-left order. 
"\texttt{foldl}" determines \index{algebra!initial} \index{initial algebra} the \textit{initial algebra} with some initial value $\varepsilon$ and \index{carrier set} carrier set $X$ and binary operation $\oplus$, which is defined over the same type as $\varepsilon$ is and is applied element-wise for each element from $X$ consecutively.
The operation $\oplus$ determines the result, which is of the same type as $\varepsilon$.
Assume $a$ equals $b$, and both variable integers and $X=[1,2,3]$ are of the kind integer list.
Assume further, the initial value $\varepsilon$ equals $7$, then \texttt{foldl} calculates $((\varepsilon+1)+2)+3)$, which yields $13$, which is an integer.
As we see later, higher-order predicates are not that useful for heap verification as they are for user-defined limits and predominantly arithmetic lists.
It shall be mentioned that higher-order (except parameterised Prolog terms) will explicitly not be considered since they otherwise may not be expressible by some CF-grammar to be parsed.
%%%%
\end{proof}

For the sake of completeness of syntactical definitions in fig.\ref{fig:PrologRulesEBNF} and the translation (next section), it shall be considered whether transforming Prolog subgoals stepwise using the \index{sequence operator} sequence operator adequately "\textbf{;}" and cut operator "\textbf{!}" is needed.
If a predicate's body contains "\textbf{;}", then the whole sequence after "\textbf{;}" shall be turned into a new predicate with the same left-hand side.
For instance,
 $$b :- a_0, a_1, ..., a_m; a_{m+1}, ... , a_n$$ 
is $\exists m.0\le m \le n$ divided into
 $$b :- a_0, a_1, ..., a_m. \quad \ b :- a_{m+1}, ... , a_n.$$
If in analogy to that "\textbf{!}" occurs in $$b :- a_0, a_1, ..., a_m,!, a_{m+1}, ... , a_n$$, then $a_0$, $a_1, ...,  a_m$ may contain alternatives, which will be considered in case of a fail.
"\textbf{!}" asserts that if only one \index{subgoal} subgoal ranging from $a_{m+1}$ to $a_n$ fails, then $b$ completely fails without any further search for \index{non-terminal} alternatives.
All alternatives may be \index{rule factorisation} left-factorised starting at "\textbf{!}" s.t. other alternatives are excluded.
In short, this is the reason why "\textbf{;}" and "\textbf{!}" may be excluded from Prolog w.l.o.g. and expressibility (see obs.\ref{obs:SimplificationByGeneralisation}).
The question of further predicate generalisation is without any doubt exciting, but w.r.t. this work's main objectives will be skipped.
Paulson \cite{paulson93} researches generalisation by introducing functionals to abstraction and logical rules (cf. with \cite{plaistead79}, \cite{menzies96}, \cite{giunchiglia89}, \cite{degtyarev01} and the resolution method \cite{kalinina01}, \cite{bratchikov98}, \cite{gast08} --- which for this work's main objectives are also not further considered).
The tactics modules in Coq \cite{bertot04} (cf. sec.\ref{chapter:intro}) base on the same principle as \cite{paulson93}.

The essence of lem.\ref{lem:APsAreFOPs} and lem.\ref{lem:APsAreSOPs} is arbitrary equality predicates in Prolog that may be expressed without the significant limitations considered in the introductory.
Due to higher-order predicates, even explicitly stated (mutual) recursion might be replaced.
We do not restrict ourselves and allow \index{$\mu$-recursive predicate} $\mu$-recursive predicates for the cost of partial correctness and the uncertainty if predicate interpretation will eventually halt.

\begin{definition}[Predicate Folding]
\label{def:PredicateFolding}
The \textit{unfold and fold} of some predicate definition $a(\vec{\alpha})$ for/into some predicate $a$ for the given predicates $\Gamma_a$ with the actual term denotation $\vec{\alpha}$ and subgoals $q_k$ is defined as follows.

From lem.\ref{lem:APsAreSOPs}, we allow $\Gamma_a$ w.l.o.g. equals $a(\vec{y}):-q_k$ with $q_k = q_{k,0}(\vec{x}_{k,0}),q_{k,1}(\vec{x}_{k,1}), ... ,q_{k,m}(\vec{x}_{k,m})$.
If $\vec{\alpha} = (\alpha_0,\alpha_1,..,\alpha_A)$ and $\vec{y} = (y_0,y_1,..,y_A)$, then
%%%
$$a(\vec{\alpha}) \Leftrightarrow q_{k,0}(\vec{x}_{k,0}),q_{k,1}(\vec{x}_{k,1}), ... ,q_{k,m}(\vec{x}_{k,m})$$
with:
$$\alpha_0 \approx y_0, \alpha_1 \approx y_1, ... , \alpha_A \approx y_A.$$
\end{definition}

In the "$\Rightarrow$"-case for the equality from above $a(\vec{\alpha})$, the predicate is unfolded.
In the "$\Leftarrow$"-case, the right-hand side of the predicate definition is folded into a \index{predicate call} predicate call, namely as a subgoal.
$\approx$ denotes \index{term unification} \textit{term unification}.
It is required predicates are well-defined in their environment according to cor.\ref{corollary:PredicateEnv}.

\subsection{Interpretation of Heap Predicates}

The approach proposed in this section is universal though not classic.

\begin{enumerate}
 \item Transforming the input program and annotated assertions into Prolog terms, which then are processed by the heap verification system within Prolog (see fig.\ref{fig:HeapVerificationArchitecture}, \cite{haberland14-1}).
 \item Definition of APs is lodged into a reserved section of a resulting Prolog listing and IR of the imperative input program in Prolog.
 The rules are part of the heap interpretation and therefore require syntactical correctness.
 As point 1, this is analysed in more detail in sec.\ref{sect:ArchitectureVerificationSystem}.
 
 \item Definition of formal grammar for APs.
The grammar is passed to a language processor, which generates a concrete syntax analyser on success.
 \item The use and integration of previously transformed APs into a corresponding syntax analyser enriched by subgoals as translating attributes during the verification phase.
\end{enumerate}

\begin{observation}[Narrowing down with Formal Languages]
\label{obs:FormalLanguageObservation}
 Looking at the heap structure, it reminds the definition of a formal language.
\end{observation}

A simple heap assertion of kind $\mapsto$ becomes \index{terminal} \index{non-terminal} a terminal (see def.\ref{def:PredicateRuleSetDefinition}).
An AP becomes an NT, namely a call.
Terminals may be linked consecutively using the binary operator \index{$\star$} $\star$, which commutes (see theo.\ref{theo:ReynoldsHeapProperties}).
A terminal denotes a unit of some sub-grammar, which is equivalent to some edge of a heap.
Assertions of kind $\mapsto$ may effectively be linked together when left locations are sorted in \index{lexicographical ordering} lexicographical order.
Whenever a naming clashes, $\alpha$-conversion resolves that issue depending on which module the clash occurs, e.g., introducing naming \index{prefix} prefixes.

\begin{thesis}[Syntax Recognition is Proving]
 \label{thes:ReasoningAsProving}
 The recognition of APs is interpreted as proof.
\end{thesis}

\begin{proof}
 Particularly, def.\ref{def:FirstSetParsers} and def.\ref{def:FollowSetParsers} allow defining a corresponding syntactic analyser, which is total and terminates according to cor.\ref{cor:TranducersTerminate}.
 The derived "\textit{word}" (in terms of a formal grammar) contains heap terminals and NTs defined in def.\ref{def:AbstractSentence}.
 (see later this section for the constructive part of the proof).
\end{proof}

\begin{corollary}[Context-Freeness of Heap Expressions]
The family of APs describing a set of rules is \index{grammar!formal} CF.
\label{cor:CFOfHeapExpressions}
%%%
\end{corollary}

\begin{proof} 
The left-hand side of the predicate cannot contain by default more than one NT.
Terminals are neither allowed.
Consequently, only one NT is allowed on the left-hand side.
The lack of the requirement, where the right-hand side must strictly be  \index{recursion!right} right-recursive, so the absence of a regular grammar with rules of a kind $S\rightarrow aA$, leads to CF-ness.
Grammar recognition is allowed, which is equivalent to some bracket grammar \cite{grune90} (CF-grammars whose rules may be of a kind $S\rightarrow ( S )$), namely when $n \in \mathbb{N}_0$ for some terminals $a$,$b$,$x_0$,$x_1$ and $x_2$, which might be $\mapsto$-assertions, s.t. $x_0 a^n x_1 b^n x_2$ and $x_0 \ne a$, $x_0 \ne x_1$ and $x_1 \ne b$, $b \ne x_2$.
If a \index{predicate!head} predicate's head contains arguments, then this situation still does not change the \index{data dependency} predicates' static dependency.
Each predicate partition may be assigned an initial \index{non-terminal} NT.
\end{proof}

\begin{observation}[Heap Reduction]
\label{obs:HeapReduction}
 The generation of heaps by APs raises the question about abstraction level and checks since a heap is a generated element subject to checks --- as being generalised from my research \cite{haberland08-2}.
\end{observation}

It implies both the derived and expected heap may contain folded predicate definitions, which may need to unfold for a final decision on equality.
After a closer look, one can find this is a bi-directional process.
Folded heaps may occur in both heaps.
It is essential to notice, this kind of problem can be reduced to \index{Post's problem} \textit{Post's Correspondence Problem} (PCP) \cite{davis94}, which from its theoretical foundation would not be decidable in general, but in some cases, only it is.
Hence, this general problem w.r.t. PCP upon heaps is not being considered further.\\\\

Obs.\ref{obs:FormalLanguageObservation}, together with obs.\ref{obs:HeapReduction} may, however, be considered as a precursor to this formulation: \textit{Given a  $\star$-connected heap.
Question: Does it match or not the given heap specification}?
It may also be valid to ask: \textit{Which heap is the closest to the correct one, s.t. specification holds}?
It must be noted that the answer to this question would eventually answer the problem of \index{counter-example} \textit{generic counter-example generation} and minimal correction.

\begin{lemma}[Heap as Word]
\label{lem:HeapAsWord}
The \textit{word} problem for a partition of APs $P$ is defined as follows.
 Given $\alpha_1, \alpha_2 \in L(G(P))$, is then $\alpha_1 \equiv \alpha_2$ decidable
whenever $\alpha_1$ derives $\exists \alpha_3$, and $\alpha_2$ derives $\alpha_3$ too (cf. fig.\ref{fig:CRTonHoareTriples}), where
$G(P)$ denotes a formal \index{context-free} \index{CF} CF grammar obtained from predicate partition $P$.
\end{lemma}

\begin{proof}
Here $\alpha = (a+A)^{*}$, $a \in T$, $A \in NT$, $\alpha$ is a \textit{sentence} (cf. def.\ref{def:AbstractSentence}).
$T$ denotes the set of terminals, which are parameterised and noted as units of pairs with both sides' content of corresponding $\mapsto$-assertions.
$NT$ denotes the set of NTs, which contains all predicates and syntactically correct input terms, which may be parameterised.
The predicate partition $P$ establishes a formal grammar $G(P)$.
From thes.\ref{thes:ReasoningAsProving} and cor.\ref{cor:CFOfHeapExpressions} follow CF-grammars.
This and sentences imply the theorem of some CF-generated language $L(G(P))$ regarding the generated language is not context-sensitive, as can formally be seen referring to the pumping lemma.
The start NT denotes a predicate call from $\alpha_1$ or $\alpha_2$.
Initially, it is necessary to calculate and define while analysing the \index{follow-set} follow-set of terminals $\sigma(\alpha)$, as well as \index{first-terminals set} the \textit{set of NT beginnings} (so-called first-terminal set) $\pi(\alpha)$ must be calculated (see def.\ref{def:FirstSetParsers}, def.\ref{def:FollowSetParsers}).
$\sigma(\alpha)$ and $\pi(\alpha)$ are calculated only once for any fixed $P$.
Here it must be noticed that for a given $G(P)$, the initial NT may differ --- if the corresponding analyser is flexible enough to support a start change.
It essentially depends on the subgoals to be chosen from $\alpha_1$ and $\alpha_2$.

Moreover, one path is searched starting with $\alpha_1 \vdash^{*} \alpha_2$ and $\alpha_2 \vdash^{*} \alpha_1$, where $\vdash$ denotes a one-time-only rule application.
Only when no path is found from both sides can it be stated that $\alpha_1$ does not match with $\alpha_2$, and vice versa does not hold.
It is required to build one path between predicates and accept both starting and intermediate terminals of a kind $\mapsto$ to check whether two (formal) sentences over heaps match.
Parameters in terminals and NTs in semantically correct assertions are connected, so there are no unrelated (non-)terminals.
Therefore, a parameter is guaranteed to exist, which assigns some value in both $\alpha_1$ and $\alpha_2$.
The proposed platform (see sec.\ref{chapter:logical}) promotes additional checks on predicates, if needed, maybe toggled user-defined, for example, checking the presence of specific terminals inside a predicate body.
It is better to assume for both $\alpha_1$ and $\alpha_2$ some $\alpha_3$ to which both sentences may derive to, in order to lift terminals and allow abstract sentences, which may contain NTs (def.\ref{def:AbstractSentence}).
Here, $\alpha_3$ may be abstract again, containing NTs --- so, CRT holds for heap words (cf. fig.\ref{fig:CRTonHoareTriples}).
\end{proof}

\subsection{Translating Horn Rules}
\label{sect:TranslatingHornRules}
In this section, the translation of APs into rules of an attributed CF-grammar is considered.
APs are given as \index{Prolog rule} Prolog rules.
Before this, it is required to turn \index{assertion!base} ordinary heap assertions of kind $loc \mapsto val$ into \index{token} tokens.
It is a mandatory lexico-analytical step in interpretation taken by Prolog.
The use of \index{multi-paradigmal programming} \textit{multi-paradigmal programming} \cite{denti05} allows interpreting Prolog rules during execution by different \index{language!processor} language processors.
Thus, the verification process may be initiated, controlled by and terminated by an incoming program.
The translation process of Prolog rules into some formal grammar is unexpectedly simple.
Prolog is well-known and has a simple syntax and semantics.
However, Prolog rules may contain arguments on both sides of "\textbf{:-}" in the rule definition.
Parameters and arguments of the grammar to be generated may be modelled w.l.o.g. as \index{grammar!formal} formal grammar \index{grammar!attributed} attributes.
Consequently, \index{transducer} the translation process $C\llbracket \rrbracket$ may be characterised by the following definition.

\begin{theorem}[Transduction into Attributed Grammar]
\label{theo:TransductionIntoAttributedGrammar}
$C\llbracket . \rrbracket$ denotes a semantic transducing function taking APs as input and turning them into an attributed grammar.
\end{theorem}

\begin{proof}
%%%
The proof is constructive by providing the actual transducer:

\begin{center}
\begin{tabular}{l}
 $C\llbracket \rrbracket = \emptyset$\\
 $C\llbracket C_1 . C_2 \rrbracket = C\llbracket C_1 \rrbracket \ \dot{\cup} \ C\llbracket C_2 \rrbracket$\\
 $C\llbracket a(\vec{x}):-q^{0}(\vec{x}), ..., q^{n}{(\vec{x})} \rrbracket$ = $\{ a_{\vec{x}} \rightarrow q^0_{\vec{x}} ... q^n_{\vec{x}}  \}$
\end{tabular}
\end{center}

In contrast to the notations earlier introduced in the following \index{subgoal} subgoals are used.
Now the input vector $\vec{x}$ contains all symbolic variables for a comfortable notation inside each predicate.
If some subgoal $q_j$ for $j\ge 0$ does not require all components of $\vec{x}$, then the subgoal does not need it either.
$\dot{\cup}$ is the set union, preserving sequences and duplicates.
Both notations are dual and very similar to each other.
An AP describes a heap.
Thus, $C\llbracket \rrbracket$ transduces a heap, so it interprets some associative heap (see lem.\ref{lem:HeapAsWord}).
The transducer \index{transducer} $C^{-1}\llbracket . \rrbracket$ turns an attributed grammar back to \index{Prolog} Prolog.
%%%
\end{proof}

\begin{theorem}[Inverted Transduction from Attributed Grammar]
\label{theo:InvertedTransductionIntoAttributedGrammar}
$C^{-1}\llbracket . \rrbracket$ denotes a transducer of some attributed grammar as input and a set of APs as output.
\end{theorem}

\begin{proof}
%%%
Again the proof is given by a direct transducer definition.
\begin{center}
\begin{tabular}{l}
 $C^{-1}\llbracket \rrbracket = \emptyset$\\
 $C^{-1}\llbracket C_1 \ C_2 \rrbracket = C^{-1}\llbracket C_1 \rrbracket \ . \ C^{-1}\llbracket C_2 \rrbracket$\\
 $C^{-1}\llbracket a_{\vec{x}} \rightarrow q^0_{\vec{x}} ... q^n_{\vec{x}} \rrbracket$ = $\{ a(\vec{x}):-q_0(\vec{x}), ..., q_n(\vec{x})\}$
\end{tabular}
\end{center}
%%%
\end{proof}

\begin{corollary}[Transducer Termination]
\label{cor:TranducersTerminate}
Both $C \llbracket . \rrbracket$ and $C^{-1} \llbracket . \rrbracket$ permanently terminate for any incoming well-defined finite vector.
\end{corollary}

\begin{proof}
 The proof is relatively simple as infinite cycles are excluded.
 Both transducers $C \llbracket \rrbracket$ and $C^{-1} \llbracket \rrbracket$, according to theo.\ref{theo:TransductionIntoAttributedGrammar} and theo.\ref{theo:InvertedTransductionIntoAttributedGrammar}, scan incoming rules stepwise in LR ordering.
Assuming an infinite cycle was in the given rules, e.g. mutual recursion through CF rules was hypothetically allowed, transducers would still terminate because that cycle would only affect the analysis.
 The start of the predicate partition corresponds 1:1 with the start NT of the corresponding \index{grammar!sub-} \textit{sub-grammar}.
As APs may have more than one entry point w.r.t. subgoal call, this behaviour may be mimicked to the corresponding formal grammar by selecting an arbitrary defined NT as starting NT.
Naturally, rules do not change.
A fitting syntax analyser may be used if allowing all NTs from the caller's side, which is the subgoal call level.
The selection of the starting NT without the need to rebuild the recogniser from scratch can be observed in specific parsers, e.g. ANTLR.
It is observed that many other parsers do not allow to change the NT even when the static rule set does not change.
Recognition of different NTs means the recognition of different (sub-)expressions for a given grammar.
\end{proof}

Soundness and completeness still need to be researched $C \llbracket \rrbracket$ and $C^{-1} \llbracket \rrbracket$.

\begin{corollary}[Soundness and Completeness of Transductions]
$C \llbracket \rrbracket$ and $C^{-1} \llbracket \rrbracket$ are complete and sound.
\end{corollary}

\begin{proof}
Trivially, that $C \circ C^{-1} \circ C \equiv C$ and $C^{-1} \circ C \circ C^{-1} \equiv C^{-1}$ hold by comparing correct definitions.
The punchline of sec.\ref{sect:PointsToHeaplets} is neither "\textbf{!}" nor "\textbf{;}" as subgoals do not affect expressibility after all.
If for any element $C\llbracket \rrbracket$ does not terminate, then the \index{co-domain} \index{co-domain} co-domain, too, is not defined.
The same affects $C^{-1}\llbracket \rrbracket$ due to cor.\ref{cor:TranducersTerminate}.
\end{proof}

It is not hard to notice that Prolog rules are not the only representation, but it is an IR that is very close to formal grammar.
In general imperative PLs supporting procedures may be included based on automated stack execution.

\subsection{Syntax Parsing as Heap Verification}

For the sake of a simple and intuitive algorithm understanding, constants from def.\ref{def:HeapAssertionDefinition} are still considered further.
Regarding class objects, the predicate \underline{$true$} may mean accepting all $\mapsto$-assertions until some marker, requiring different conventions and depends on a rule.
Such a marker would technically represent a \index{synchronisation point} safe synchronisation point used in \index{error production rules} error-production rules in certain \index{syntax analysis} parser classes (cf.\cite{grune90}).
The marker would allow proceeding syntax analysis for a derived stream of \index{token stream} tokens until it gets stuck due to an invalid or unexpected token sequence.
It is assumed, the input word representing the heap is always finite since we are not interested in infinite verification conditions (cf. discussion in sec.\ref{Intro:LogicalReasoning}).
This assumption is considered fair because memory is finite and linear in address space.
Purely hypothetical, the amount of committed unfolds and folds together approximates infinity.
Later we show that for $\exists j$ functions, $\pi_j$ and $\sigma_j$ are bound by polynomial complexity.

This section further proposes and discusses the basic conventions needed to implement a syntax analyser for heap verification (further considered by \index{LL-analyser} LL(k)-analysers by example).
A LL(k)-parser is a syntactic analyser that can look up ahead $k$ tokens to resolve the ambiguity.

Parsers other than LL(k)-based fulfilling required interfaces are more than welcome to be used but are unfortunately not considered in this work for simplicity.
Those alternative syntax analysers may include, for instance, \index{LALR-recogniser} LALR, \index{SLR-recogniser} SLR, \index{Earley recogniser} Earley analysers.
Jourdan \cite{jourdan12} proves the soundness of his LR(1) analyser using \textit{Coq}.
The question of analyser soundness is undoubtfully essential too, but not considered further, because its overall impact here is too small.
Next, a \textit{sentence} is defined as a composition of separate $\mapsto$-assertions and subgoals as NTs.
Second and third, in analogy to LL(k)-analysers, where separate terminals are given as ordinary assertions, the operations "$first$" and "$follow$" are introduced.
Fourth, both operations \index{SHIFT} (SHIFT) and \index{REDUCE} (REDUCE) are introduced, s.t. generalised analysers may be defined, presumably with minor modification upon existing wide-spread parsers.

\begin{definition}[Abstract Sentence]
\label{def:AbstractSentence}
An \textit{abstract sentence} $\alpha$ is a $\star$-conjunction of heaps.
Simple heaps are noted as $a \mapsto b$, where $a$ is a location and $b$ denotes a value, e.g. an integer or some value, which contains either meaningful (composite) value or \textbf{nil}.
\end{definition}

For example, $$\alpha::= \texttt{[ pointsto(x,nil),} \texttt{pointsto(y,1), member(x,[y])]}$$ describes the actual heap state during verification for a given imperative program.
$\star$ is replaced in the previous list by a comma.
The rule specification may depend on
$$\texttt{[pointsto(Y,1),member(X,[Y|_]),pointsto(X,_)]}.$$
Hence, in order to check an abstract sentence (specification) for a given program (generating the heap), it is needed to compare whether one of both sides is derivable from the other side or not.

An abstract sentence may also contain \index{term unification} term unification, e.g. \texttt{pointsto(X,5),X=Y}.
Term unification must be thoroughly checked and separated from both $\mapsto$-assertions and predicate calls -- so, NTs.
\index{recursion!unbound} Unbound term recursion must be limited, so self-application is excluded by definition.
Therefore, as stated, \index{term!self-containing} \textit{self-containing terms} must be \index{occurs-check} bound, as they often may occur in definitions in Prolog (see sec.\ref{sect:PrologAsReasoningSystem}), where a check for self-containment is absent by default.
Analysers that are forced to process unbound terms tend to be incomplete whenever they run into non-termination.
It is agreed upon that by default, checks on cycles are done (presumably) during semantic analysis or not at all in order to avoid running into a previously mentioned problem.

Let us consider now the na\"{\i}ve approach for comparing equality of two abstract sentences.
Let us consider algo.\ref{AlgorithmEqualityCheck}.
$\pi$ denotes the function of starting terminals for a given rule set and a given sentence, as defined earlier.
The problem with this approach is uncertainty.
Uncertainty arises from when and how often shall \index{unfold} "$unfold$" or \index{fold} "$fold$" be applied.
There is always the risk of running into infinity on unfolding or folding.
It remains uncertain, especially after finding one (sub-)solution candidate since it is not clear the found solution is optimal indeed -- which depends undoubtedly on the given rule set.
Let us assume

$$\alpha_1 = [\underbrace{a \mapsto b}_{i_0},i_1,\cdots, i_{m_1}, \underbrace{q_1(x)}_{p_0}],$$ 
$$\alpha_2 = [\cdots, \underbrace{a \mapsto b}_{j_3},\cdots, \underbrace{q_1(x)}_{q_7}]$$

These two heaps shall be compared.
The shift operations on terms (SHIFT-TERMs) leads to unifications $i_0$ and $j_3$.
The comparison would continue with the remaining terms.
First, the reduce operation (REDUCE-PREDs) checks whether predicates are matching and applicable for unfolding.
Thus, the first terminal of the predicate may be requested first.
The unfolding of $expand(p_k,\alpha_1)$ replaces the subgoal with the body of predicate $p_k$ (see def.\ref{def:PredicateFolding} and fig.\ref{fig:PrologRulesEBNF}) and generates the new abstract sentence $\alpha'$, which may be described in Prolog as \texttt{concat($\alpha$,$[i_7,i_8,i_9],\alpha'$)} if $q_1(x)$ unfolds to list $[i_7,i_8,i_9]$.

\textbf{Remark:} Notice the remarkable similarity between folds and unfolds, expansion and reduction, template instantiation and schema validation, as well as heap construction and verification.

\begin{algorithm}[t]
\caption{Na\"{\i}ve algorithm for equality check of abstract sentences. \textbf{Input:} $\alpha_1$ = $[i_0,...,i_{m_1},$ $p_0,...,p_{n_1}]$, $\alpha_2$ = $[j_0,...,j_{m_2},q_0,...,q_{n_2}]$ with $i$ and $j$ as \texttt{points-to}-terminals and subgoal $p$ and $q$ as NTs. $\Gamma$ contains all defined predicates. \textbf{Result:} \underline{$True$} if $\alpha_1$ equals $\alpha_2$, \underline{$False$} otherwise.}
\begin{algorithmic}[1]
\Procedure {SHIFT-TERMs}{$\Gamma$, $\alpha_1$, $\alpha_2$}
 \ForAll {$k \in \{0, ... , m_1\}$}
   \If {$\exists l. l \in \{0, ..., m_2\} \wedge i_k \approx j_l$}
     \State $\alpha_1 \leftarrow \alpha_1 \setminus i_k$
     \State $\alpha_2 \leftarrow \alpha_2 \setminus j_l$
   \EndIf
  \EndFor
\EndProcedure
\Statex
\Procedure {REDUCE-PREDs}{$\Gamma$, $\alpha_1$, $\alpha_2$}
 \ForAll {$k \in [0 .. n_1]$}
   \Comment{compare terminals}
   \If {$\exists l. l \in [0..m_2] \wedge j_l \in \pi(p_k)$}
     \State $expand(p_k, \alpha_1)$
   \Else
   \Comment{reduce in $\alpha_2$, if possible}
     \If {$\exists l. l \in [0..n_2] \wedge (\pi(q_l) \cap \pi(p_k) \ne \emptyset) $}
       \State $\alpha_1 \leftarrow expand(p_k, \alpha_1)$
       \State $\alpha_2 \leftarrow expand(q_l, \alpha_2)$
     \Else
       \Comment{Match}
       \If {$m_1 = m_2 = n_1 = n2 = 0$}
         \State $true \rightarrow Halt!$
       \Else
         \Comment{No match}
         \State $false \rightarrow Halt!$
       \EndIf
     \EndIf
   \EndIf
  \EndFor
\EndProcedure
\end{algorithmic}
 \label{AlgorithmEqualityCheck}
\end{algorithm}

\begin{definition}[First-Set of Non-Terminals]
\label{def:FirstSetParsers}
The \textit{set of starts of NTs} \index{first set} (first-set) is defined as a co-domain of the total mapping $\pi$ of kind $(T \cup NT) \rightarrow 2^T$ for $m \in \mathbb{N}$, s.t. holds\\
 
$\pi(a) ::=
\left\{
	\begin{array}{ll}
		a  & \mbox{if } a \mbox{ is } X \mapsto Y \mbox{ or } \Gamma_a ::= a.\\\\
		\bigcup_{0 \leq j \leq n} \pi(q_{j,0}) & \mbox{if } \Gamma_a ::= a:- q_{m \times n}, n \in \mathbb{N}.
	\end{array}
\right.
$
%%%
\end{definition}

Independent from concrete arguments, $\pi$ defines all terminals, either \index{$\mapsto$-assertion} $\mapsto$-assertions or the first terminal(s) of some predicate \index{subgoal} subgoal.
It is implied that both subgoal unification and calls to built-in subgoals are filtered and are no further considered during $\pi$ and $\sigma$.

\begin{definition}[Set of Following Terminals]
\label{def:FollowSetParsers}
The \index{follow-terminals set} \textit{set of following terminals} (follow-set) $\sigma(t) \subseteq T$ for $t \in (T \cup NT)$ is defined as:\\

$\sigma(t) ::= 
\left\{
	\begin{array}{lll}
		\bigcup_{i,j} \pi(q_{i,j+1})  & \mbox{if } t \mbox{ at } (i,j<n) \mbox{ in } q_{m \times n}\\
		                              &            \wedge \ 0 \leq i \leq m\\
		                              &            \wedge \ q_{i,j+1} \neq \top\\
		                              &            \wedge \ \exists a.\Gamma_a ::= a:-q_{m \times n}\\\\
		\bigcup_{a} \sigma(a)                & \mbox{if }  t \mbox{ at } (i,n) \mbox{ in } q_{m \times n}\\
		                              &            \wedge \ \Gamma_a ::= a:-q_{m \times n}\\
		                              &            \wedge \ \exists b.\Gamma_b ::= b:-q_{m_b \times n_b}\\
		                              &            \wedge \ a \mbox{ at } (i_b,j_b) \mbox{ in } q_{m_b \times n_b}\\\\
		\emptyset                     & \mbox{otherwise}
	\end{array}
\right.
$
\end{definition}

An informal remark may illustrate the key characteristics better.
The set of following \index{terminal} terminals defines all terminals that may follow a given actual $\mapsto$ assertion or a given subgoal if some terminal may be the last of a sentence in a rule.
According to \cite{grune90}, \index{LL(k)-recogniser} LL(k)-recognisers may be fully defined based on just $\pi$ and $\sigma$.

\begin{example}[Rules Ambiguity]
Given the \index{production rules} production rules $ q_1 \rightarrow a, \  q_2 \rightarrow a q_2 \ | \ q_3 b, \ q_3 \rightarrow \varepsilon \ | \ q_3 a$.
As can easily be found these rules are \index{ambiguity} ambiguous, because
$$\pi(q_2)=\{a\}, \sigma(a)=\{ \varepsilon\} \cup \pi(q_2) \cup \pi(q_3) \cup \sigma(q_3).$$
\end{example}

\begin{example}[Example for Word Correctness]
Given the following finite specification
$$[ (loc1,v1), p1(loc1,loc2), (loc2,v2) ]$$ 
as well as the conditional chain fitting to it:
$[ (loc1,v1), (loc2,v2)  ]$.
The chain is a correct word for the given specification only in case $p1$ generates the \index{heap!empty} empty heap.
\end{example}

If the word that is given is not just a sequence of terminals, but it is some \index{abstract sentence} abstracted sentence, so it contains an arbitrary sequence of terminals and one or more NTs, then the comparison problem may arise at different places of the sequence again (cf. validation of semi-structured data \cite{haberland08-2}).
One of these problems described may be the re-calculation of repeating (sub-)heaps and abstract parts of predicates.
Thus, it is advantageous when calls to APs are remembered.
Since this is about temporary associativity, a \index{memoiser} memoiser as a software-based \index{caching} cache-system should be a perfect match.
In Prolog, seldom can be found in \cite{warren99} --- in it, memoisation is called "\textit{tabling}".
Applying memoisers subgoals may be compared much faster if a matching instance as compared previously.
A Prolog-based memoiser would benefit that its size can be diminished significantly since symbols unground to subgoals act as parameters.
The tremendous reduction is for the cost of additional unifications per matching subgoal predicate, which in this context pays off rather than reiterating whole sub-proofs.
Thus, predicates may be compared as procedure calls but without inflicting \index{side-effect} side-effects and without changing the calculation state several times in a row.
All term-arguments (possibly unground) need to be stored, presumably in a table, to memoise a predicate.

Predicate negation does not require further discussions for the reasons provided in sec.\ref{sect:ConjunctionAndDisjunction}, sec.\ref{sect:PartialSpec}, sec.\ref{sect:StricterDiscussions}, sec.\ref{sect:PrologAsReasoningSystem}, and sec.\ref{sect:LogicalReasoningAsProof}.

\subsection{Properties}
\label{sect:APsProperties}
In fig.\ref{fig:heapConfiguration}, an example of a heap configuration is shown.
Every vertex $v_j$ with $j \ne 2, j \ne 5$ has outgoing edges for class instances having more than one field of type pointer.
The configuration consists of \index{triangle} 8 triangles, each drawn with solid, dotted and tortuous lines.
Between $v_0$ and $v_3$, a line starts in the mid.
It ends in $M_1$ and spans a triangle $\Delta(v_0,M_1,v_1)$. 
Thus, fig.\ref{fig:heapConfiguration} demonstrates a variety of the same structure in dynamic memory using various descriptions.
It is assumed that two and more outgoing pointers imply a corresponding number of different fields in the united instance (see sec.\ref{chapter:expression} and fig.\ref{fig:GraphIsomorphisms}).
So, vertices may cover the same graphs but with different \index{abstract predicate} APs, e.g. triangles, dotted or tortuous lines.
All vertices have to be taken into account in different representations --- this is a mandatory criterion, but it is not sufficient.
The question about equality of two (syntactically) different directed graphs can be reduced to \index{graph isomorphism} \textit{graph isomorphism}.

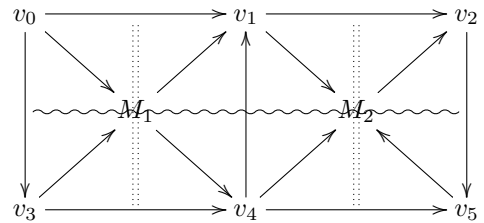
\begin{figure}[b]
\begin{center}
\begin{tabular}{l}
\begin{minipage}{7cm} 
\xymatrix{
 v_0 \ar[rr] \ar[rd] \ar[dd] &  \ar@{:}[dd]        & v_1 \ar[rr]        \ar[rd] & \ar@{:}[dd] & v_2 \ar[dd]\\
        \ar@{~}[rrrr]        & M_1 \ar[ru] \ar[rd] &                            & M_2 \ar[ru] & \\
 v_3 \ar[ur] \ar[rr]         &                     & v_4 \ar[uu]\ar[ru] \ar[rr] &             & v_5 \ar[lu]                         
}
\end{minipage}
\end{tabular}
 \caption{Example of a heap configuration}
 \label{fig:heapConfiguration}
\end{center}
\end{figure}

When the input token stream is analysed, the next state of the analyser may be determined using the current token and following tokens as far as it may be necessary to resolve the ambiguity.
In the worst case, all or at least many input tokens need to be checked --- a decision will eventually be made.
This case is when the heap describes \index{object instance} an object instance, and bounds need to be determined.
If an object were either defined only within the bounds of one predicate, or object fields which are pointers were described by predicates, then this convention would avoid the described problem and would allow generic heap canonisation by predicates.
A syntactical analysis obeying this convention is bound by polynomial complexity \cite{grune90}.
In the following, however, a convention-free \index{syntax analysis} syntax analysis will be traced.
If object boundaries are not considered, then the analysis may become unsound.
Namely, it may overlap with other heaps belonging to other object instances.
This overlap would allow, in practice, \textit{partial or parallel analyses} of different heaps at the same time -- which is believed not to be advantageous after all on the big picture since there may still be many dependencies in between sub-proofs.
First, as Wos mentioned in the introductory, separability was a hard hinder in any parallelisation attempts --- this all is due to tight dependencies, auxiliary definitions and predicate hierarchies and inductive definition dependencies in general as well as local frames.
Second, a realistic estimate shall be questioned, as may seem after all doubtful regarding the benefits, if any.
Thus, parallelisation is from both theoretical and practical consideration not appropriate, at least not now.
In practice, parallel heaps would partially be of interest when separable, and very elementary specifications may only be considered.

Trieber's stack seems to be the most promising approach in an SL-environment \cite{bornat00}, even if it currently fights unsoundness issues.
If all predicate partitions may be parsed under the convention mentioned, then analysing $\mapsto$-assertions are tractable and terminate either with a positive answer or a refusal.
On refusal, parsing fails, and either the invalid token or the concrete sentence for a given NT are yielded.
The error shows the heap state and explains the difference between the currently calculated and expected heap.
Moreover, it is worth mentioning that the confluency of calculations is held while parsing a LL/LR-grammar because the calculation state can be recorded as a tuple $(state, word)$, and for each state, the transition is determined.
This final result is that the main-memory model does not change.
However, it is extended by APs instead.

\subsection{Implementation}
\label{sect:Implementation}

The system is based on \index{GNU Prolog} GNU Prolog \cite{diaz12} and implemented using ANTLR version 4 \cite{parr12}.
The work with APs is foreseen \cite{haberland14-1}, \cite{haberland15-2}.
Initially, both contributions are based on Prolog.
They were selected for the sake of simplicity and lecturing.
For maximum support, implemented packages are based on GNU Prolog-dialect for maximum portability and generic definitions only.
For extension and the possibility of supporting numerous software libraries even in different PLs, the \index{multi-target paradigm} multi-target paradigm \cite{denti05} is applied so that code may be loaded dynamically. 
Therefore, Denti's \cite{denti05} library is chosen, Java-libraries are then loaded as-you-verify.
For the sake of simplicity, it is agreed upon that further references to Java are by example.
Whenever Java is mentioned in Denti's library, this implicitly implies any other loadable modules in arbitrary languages that match the interfaces required.
Procedures in Java communicate via a Proxy-interface which comes as a Java-class file and can be loaded on a Prolog subgoal.
The caller layer "\textit{tuProlog}" acts as \index{pattern} Mediator in the Mediator-pattern as well as Adapter (cf.\cite{kerievsky05}).
Here, call and results may be passed from both sides by implementing the \index{pattern} Proxy-pattern, as the remote-procedure call was performed locally in Prolog.
This technique also allows defining new built-in and auxiliary predicates in Java, but visible to the user as Prolog subgoal.
Immediate Prolog-defined subgoals are, of course, allowed too.
Due to its dynamic context, newly created predicates may call existing predicates and vice versa.
So, a subgoal in Prolog remains in full compliance with the principle from fig.\ref{fig:BoxModelPredicateCall}.

The implementation turns an incoming program into an \index{IR} IR, which is a pure Prolog term.
Afterwards, assertions are copied one by one to the Prolog theory.
APs are turned in a formal ANTLR-grammar as described earlier.
Then, syntax analysis is triggered.
The generated recogniser is activated after the full ANTLR-grammar is available.
The recogniser is generated Java-code, but this is not of primary interest in this work.
Afterwards, the process quits, and control is passed to the calling instance back again.
If needed, APs may dynamically be checked fully automated to preempt the Prolog core to fail due to syntax and semantic errors, which will eventually be checked.
In case of a syntax mismatch, the generated parsing error is issued accordingly.

%%%%%%%%
\index{ANTLR} ANTLR uses two technologies to fight ambiguities in its definitions, "\textit{syntactic predicates}" and "\textit{semantic predicates}" \cite{parr12}.
Both terms sound very similar and may accidentally be mixed up with those terms in this work, but this should thoroughly be avoided.
Except for these two techniques, as an alternative, there is always the possibility to rewrite Prolog rules, s.t. ambiguity is not emerging, e.g. by \textit{factorising rules}.

Unfortunately, in practice, \index{ANTLR} ANTLR does not fully cover the class of all LL(k)-recognisers.
Predicates are limited and often require a \textit{rewriting of (term) rules} in case of a conflict.
Often ANTLR is not able to distinguish automatically complex expressions, especially not regular expressions.

Consequently, this leads to a full code generation for the needed recogniser.
The same problem with limitations of LL(k) is observed in many other compiler-generators (CC), for instance, in \index{YACC} YACC or \index{BISON} BISON \cite{levine09}.
From a practical standpoint, this means a given \index{grammar!formal} grammar is formal and correct, but its recogniser just cannot be built due to some definite limitations that are specific.
Luckily, rewriting that grammar often resolves this practical flaw.
Involved top-down recognisers often can overcome this problem solely by huge additional memory reserves just for all transitions possible, as \index{BISON} BISON does it, for instance, working on the \index{shift-reduce} shift-reduce principle, which is very close to the approach proposed earlier in this section.
Detailed surveys on syntax recognisers may be found in \cite{opaleva05}, \cite{grune90}.

An example with required transformations at a proper level shall be shown for processing lexemes and tokens.
First, $bar\mapsto foo$ is transformed into $pt\_3bar\_3foo$, where the number "$3$" is either the number of names or some complex expression denoting an integer location required to distinguish consecutive namings.
If a location is intricate, e.g. $b.f.g$, then a location together with length determines the assertion expression in \index{name mangling} \textit{mangled} notation (cf.\cite{levine99}).
For example, \texttt{pointsto(X,2)} is transformed into the Prolog \index{atom} atom $p\_X\_2$.
The transformation's objective is to obtain a single atomic symbol, representing a simple heap assertion, so a terminal.
A terminal may be fully restored to a heap assertion whenever needed.
So, all annotated information is available and makes every Prolog atom subject to this encoding unique.
These straightforward and backward processes may be implemented as a built-in predicate.
Afterwards, they are used either in the essential parts of the verification or other built-in predicates in Java (see prerequisites, see \cite{denti05}).

Furthermore, the left-hand side of \index{canonisation} (de-)canonising (in Prolog)

$$\texttt{p1(X,[X|Y]):-... .}$$

is transformed into

$$\texttt{p1(X1,X2):-X1=X,X2=[X|Y],... .}.$$

The Prolog rule "\texttt{p(X,Y):-$\alpha$.}" may be transformed into the following ANTLR-grammar (rule)

\begin{center}
\begin{tabular}{c}
\texttt{p[String X,String Y]:} $\alpha$.
\end{tabular}
\end{center}

In this style, all \index{synthesised attribute} \textit{synthesised attributes} may be passed top-down.
Some predicate \texttt{p} may be decorated with \textit{inherited attributes}.
In \textit{ANTLR}, this can be achieved by adding the keyword "\texttt{returns}", followed by the attribute names before the colon.
Thus, once it is clarified which attribute is synthesised and inherited, all attributes need to be included in the corresponding ANTLR-rule accordingly.
Rules, conflicts and limitations of Prolog on this topic are almost identical to those from ANTLR.

When APs are turned into some concrete grammar, e.g. an \index{ANTLR} ANTLR-grammar, there are significant problems encountered. 
Unified terms, \index{$\mapsto$-assertion} $\mapsto$-assertions (terminals) and NTs are supposed to be transformed according to the provided syntax into grammar together with annotations.
ANTLR also allows \index{translating rule} translating rules, which define a program's semantics and are noted in an ANTLR file definition in curly brackets. 
The negation of a (sub-)sentence (remember, an arbitrary concatenation of terminals and NTs) may be specified by "$\sim$" and parentheses, which denote a given \index{regular expression} regular expression of $\mapsto$"-terminals must or must not follow.
Next, elements of the translation grammar are not discussed here because attributes and translating rules may imitate all other expressions and sentences, including a negative match of a predicate (cf.\cite{opaleva05}, \cite{lavrov01}, \cite{grune90}, \cite{parr12}, \cite{gcc15}, \cite{mitchell96}).

\subsubsection*{\underline{User Interface}}

\begin{figure}[h]
\begin{center}
\includegraphics[width=9cm]{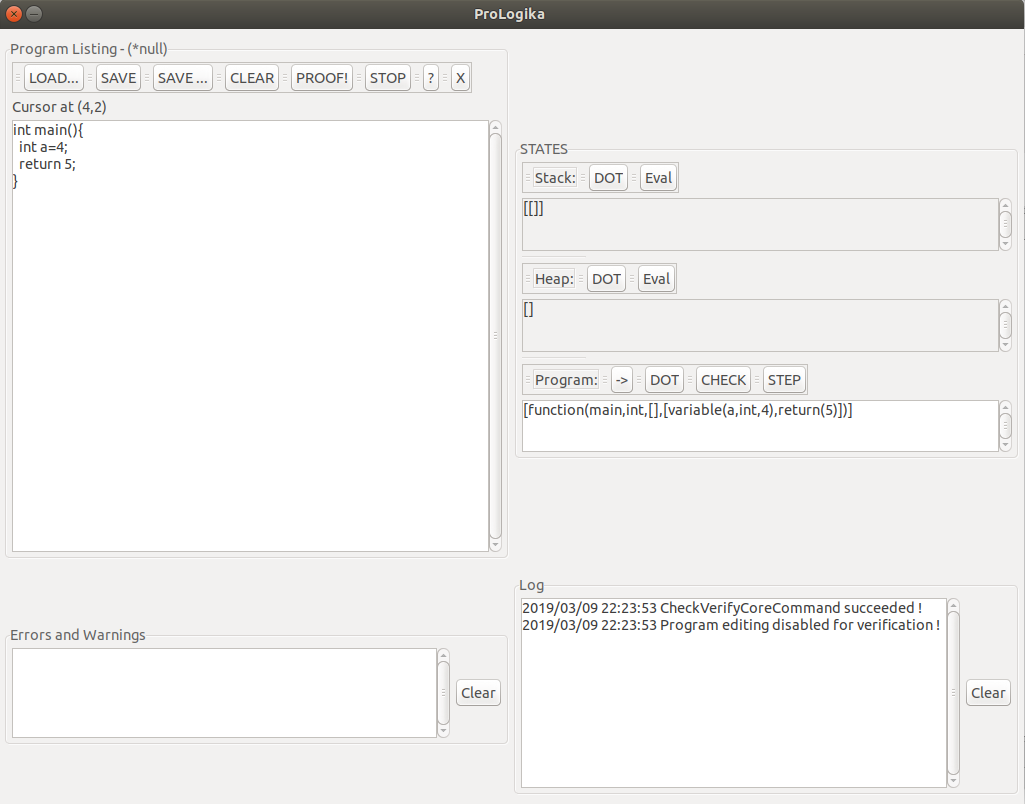}
\end{center}
 \caption{Graphical user interface}
 \label{ScreenshotGUI}
\end{figure}

The graphical user interface (see fig.\ref{ScreenshotGUI}) consists of three major parts: (i) the left part, which contains the program listing in a C-dialect, (ii) the right upper part, which displays the state of the current verification, including the dynamic memory and (iii) right lower part which contains IRs of the program in Prolog terms.
The memory content and the program IR may be turned into a visual IR (and screened as \index{DOT} DOT-file).
Apart from that, the white section informs about errors and warnings and logs major verification news.

The GUI comes in handy when verifying and running demos.
So, automation of transformations, tracing proofs and heap IR may be probed and verification rules, the imperative input program.
It appears more effortless if all it needs to reason about a program's properties are nearby.

\subsubsection*{\underline{Use Cases}}

A programmer can edit a given input program in C-dialect (see fig.\ref{UMLUseCaseProgrammer}).
Numerous checks are implemented, including the syntactic analysis and semantic checks on program IRs, which may also be transformed into \index{DOT} DOT (for visualisation only).

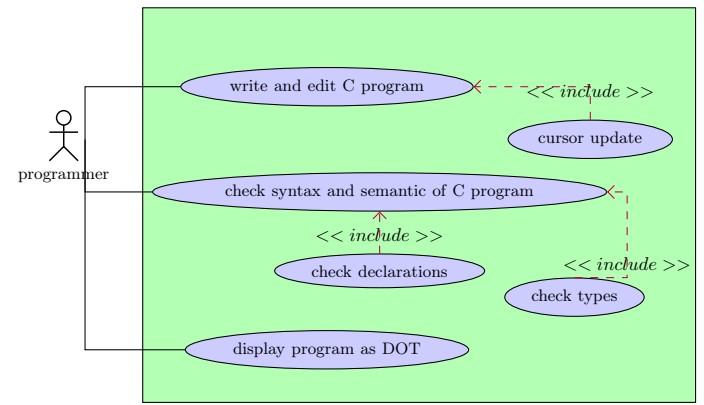
\begin{figure}[h]
\begin{center}
\begin{turn}{0}
\begin{minipage}{\linewidth}
\begin{tikzpicture}[scale=0.7, every node/.style={transform shape}]
 %% \begin{umlsystem}{Editing C program} %% tikz's bug: does not scale, so drop it
 \draw[fill=green!30] (-5.5,1.5) rectangle (5,-6);  % workaround for tikz's bug on scaling
 \umlusecase[name=CASE1,x=-2]{write and edit C program}
 \umlusecase[name=CASE2,x=-1,y=-2]{check syntax and semantic of C program}
 \umlusecase[name=CASE3,x=-2,y=-5]{display program as \index{DOT} DOT}
 \umlusecase[name=CASE4,x=3,y=-1]{cursor update}
 \umlusecase[name=CASE5,x=2.7,y=-4]{check types}
 \umlusecase[name=CASE6,x=-1,y=-3.5]{check declarations}
 \umlactor[x=-7,y=-1]{programmer}
 % use case inclusions and extensions
 \draw[umlcd style dashed line,->] (CASE4.north) -- ++(0,0.3) node[above,sloped,black]{$<<include>>$} |- (CASE1.east);
 \draw[umlcd style dashed line,->] (CASE6.north) -- ++(0,0.1) node[above,sloped,black]{$<<include>>$} |- (CASE2.south);
\draw[umlcd style dashed line,->] (CASE5.north) -- ++(1,0) node[above,sloped,black]{$<<include>>$} |- (CASE2.east);
 %
 % associations from actor
 \draw[-] (programmer.east) |- (CASE1.west);
 \draw[-] (programmer.east) |- (CASE2.west);
 \draw[-] (programmer.east) |- (CASE3.west);
\end{tikzpicture}
\end{minipage}
\end{turn}
\end{center}
 \caption{Programmer's view (UML use case)}
 \label{UMLUseCaseProgrammer}
\end{figure}

The role of a typical user who specifies differs from the traditional role a programmer might have (see fig.\ref{UMLUseCaseSpecifier}).
\textit{Specifications} are very formal logical formulae that catch a program's behaviour.
A specification may be added to the system, edited before and after program execution.
When running a specified program, the program is not run.
Instead, specifications are checked.
In case a  problem encounters, it will be placed to the right lower part.
Specification checks include pre-and postconditions of procedures, APs, class invariants but may also include a check for rules completeness for some predicate partition.
All specification and verification of relevant entities are entirely in Prolog.
The current heap state is checked against an explicitly annotated assertion to the program.

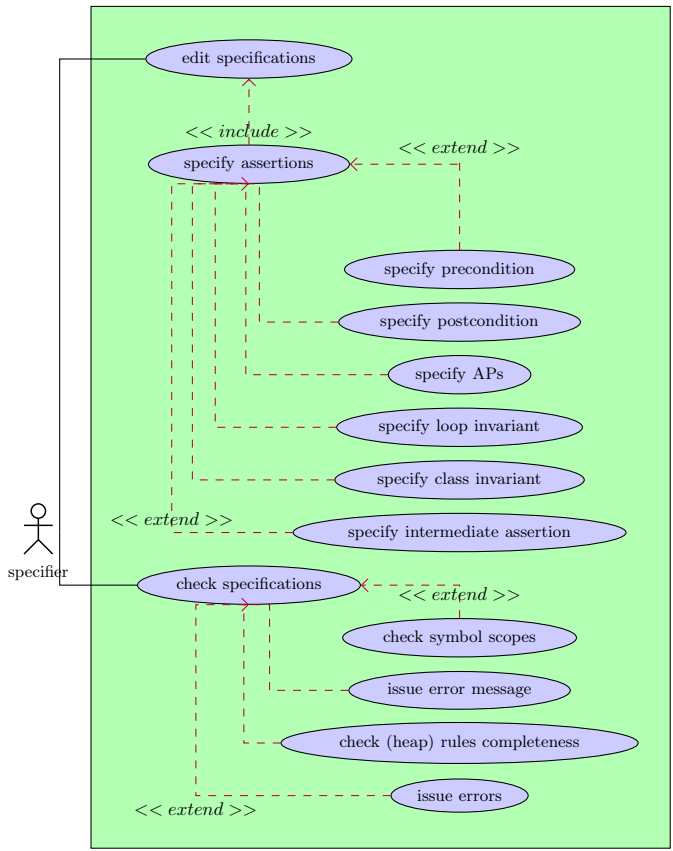
\begin{figure}[h]
\begin{center}
\begin{tikzpicture}[scale=0.7, every node/.style={transform shape}]
 \draw[fill=green!30] (5,5) rectangle (16,-11);  % workaround for tikz's bug on scaling
 \umlusecase[name=CASE1,x=8,y=+4]{edit specifications}
 \umlusecase[name=CASE2,x=8,y=2]{specify assertions}
 \umlusecase[name=CASE3,x=12,y=0]{specify precondition}
 \umlusecase[name=CASE4,x=12,y=-1]{specify postcondition}
 \umlusecase[name=CASE5,x=12,y=-2]{specify APs}
 \umlusecase[name=CASE6,x=12,y=-3]{specify loop invariant}
 \umlusecase[name=CASE7,x=12,y=-4]{specify class invariant}
 \umlusecase[name=CASE8,x=12,y=-5]{specify intermediate assertion}
 
 \umlusecase[name=CASE20,x=8,y=-6]{check specifications}
 \umlusecase[name=CASE21,x=12,y=-7]{check symbol scopes}
 \umlusecase[name=CASE22,x=12,y=-8]{issue error message}
 \umlusecase[name=CASE23,x=12,y=-9]{check (heap) rules completeness}
 \umlusecase[name=CASE24,x=12,y=-10]{issue errors}
 
 \umlactor[x=4,y=-5]{specifier}
 
 % use case inclusions and extensions
\draw[umlcd style dashed line,->] (CASE2.north) node[above,sloped,black]{$<<include>>$} -| (CASE1.south);

 \draw[umlcd style dashed line,->] (CASE3.north) -- ++(0,1.7) node[above,sloped,black]{$<<extend>>$} |- (CASE2.east);
 \draw[umlcd style dashed line,-] (CASE4.west) -- ++(-1.5,0) node[above,sloped,black]{} |- (CASE2.south);
 \draw[umlcd style dashed line,->] (CASE5.west) -- ++(-2.7,0) node[above,sloped,black]{} |- (CASE2.south);
 \draw[umlcd style dashed line,-] (CASE6.west) -- ++(-2.3,0) node[above,sloped,black]{} |- (CASE2.south);
 \draw[umlcd style dashed line,-] (CASE7.west) -- ++(-2.7,0) node[above,sloped,black]{} |- (CASE2.south);
 \draw[umlcd style dashed line,-] (CASE8.west) -- ++(-2.3,0) node[above,sloped,black]{$<<extend>>$} |- (CASE2.south);

 \draw[umlcd style dashed line,->] (CASE21.north) -- ++(0,0.2) node[above,sloped,black]{$<<extend>>$} |- (CASE20.east);
 \draw[umlcd style dashed line,-] (CASE22.west) -- ++(-1.5,0) node[above,sloped,black]{} |- (CASE20.south);
 \draw[umlcd style dashed line,->] (CASE23.west) -- ++(-0.7,0) node[above,sloped,black]{} |- (CASE20.south);
 \draw[umlcd style dashed line,-] (CASE24.west) -- ++(-3.7,0) node[below,sloped,black]{$<<extend>>$} |- (CASE20.south);
 
 % associations from actor
 \draw[-] (specifier.east) |- (CASE1.west);
 \draw[-] (specifier.east) |- (CASE20.west);
\end{tikzpicture}
\end{center}
 \caption{Specifier's view (UML use case)}
 \label{UMLUseCaseSpecifier}
\end{figure}

The essence of heap verification (see fig.\ref{UMLUseCaseVerifier}) is the calculation of heap state and checking against explicitly defined assertions by referring to heap specification being annotated to the program.
Therefore a verification requires some heap to be normalised and checked if needed.
Each of this step can be traced in the GUI parts described.
For simplicity and plausibility of the current verification, additional proof tree visualisation elements, generation (derivation) of counter-example, and memoisation of previously run (sub-)proofs for an increased proof are used.

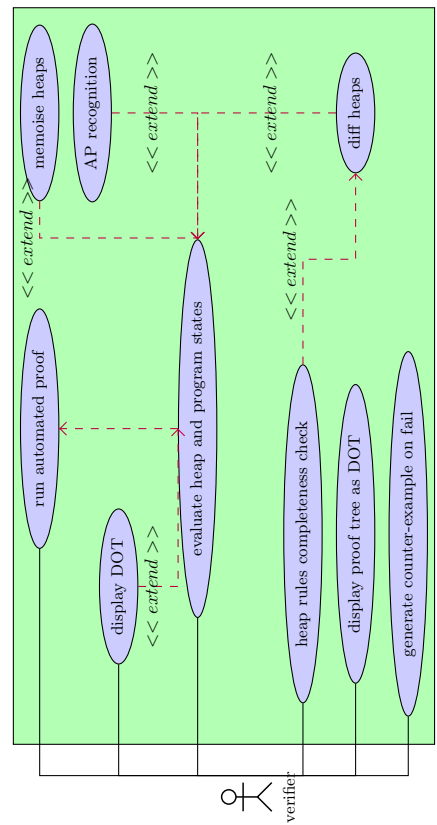
\begin{figure}[h]
\begin{center}
\begin{turn}{90}
\parbox[t]{11cm}{
\begin{minipage}{\linewidth}
\begin{tikzpicture}[scale=0.7, every node/.style={transform shape}]
 \draw[fill=green!30] (4,0.5) rectangle (18,-7.5);  % workaround for tikz's bug on scaling
 \umlusecase[name=CASE1,x=10,y=0]{run automated proof}
 \umlusecase[name=CASE2,x=7,y=-1.5]{display \index{DOT} DOT}
 \umlusecase[name=CASE3,x=10,y=-3]{evaluate heap and program states}
 \umlusecase[name=CASE4,x=16,y=0]{memoise heaps}
 \umlusecase[name=CASE5,x=16,y=-1]{AP recognition}
 \umlusecase[name=CASE6,x=16,y=-6]{diff heaps}
 \umlusecase[name=CASE7,x=8,y=-5]{heap rules completeness check}
 \umlusecase[name=CASE8,x=8,y=-6]{display proof tree as DOT}
 \umlusecase[name=CASE9,x=8,y=-7]{generate counter-example on fail}
 \umlactor[x=3,y=-4]{verifier}
 % use case inclusions and extensions
 \draw[umlcd style dashed line,->] (CASE7.east) -- ++(2,0) node[above,sloped,black]{$<<extend>>$} |- (CASE6.west);
 
 \draw[umlcd style dashed line,->] (CASE2.south) node[below,sloped,black]{$<<extend>>$} |- (CASE3.north);
 \draw[umlcd style dashed line,->] (CASE4.west) -- ++(-0.7,0) node[above,sloped,black]{$<<extend>>$} |- (CASE3.east);
 \draw[umlcd style dashed line,->] (CASE5.south) -- ++(0,-1) node[above,sloped,black]{$<<extend>>$} |- (CASE3.east);
 \draw[umlcd style dashed line,->] (CASE6.north) -- ++(0,1) node[above,sloped,black]{$<<extend>>$} |- (CASE3.east);
 
 \draw[umlcd style dashed line,->] (CASE3.north) |- (CASE1.south);
 
 % associations from actor
 \draw[-] (verifier.east) |- (CASE1.west);
 \draw[-] (verifier.east) |- (CASE2.west);
 \draw[-] (verifier.east) |- (CASE3.west);
 \draw[-] (verifier.east) |- (CASE7.west);
 \draw[-] (verifier.east) |- (CASE8.west);
 \draw[-] (verifier.east) |- (CASE9.west);
\end{tikzpicture}
\end{minipage}}
\end{turn}
\end{center}
 \caption{Verifier's view (UML use case)}
 \label{UMLUseCaseVerifier}
\end{figure}

Example 1) Reverse of a single-linked list\\
Let us assume we are given the list \{1,2,3\}, and pointer $y$ pointing to the last element of a linear list.
Then list inversion may recursively be defined by splitting off the last element and adding it to the beginning of the remaining list.
This algorithm results in the list illustrated in fig.\ref{HeapExampleAP1}.

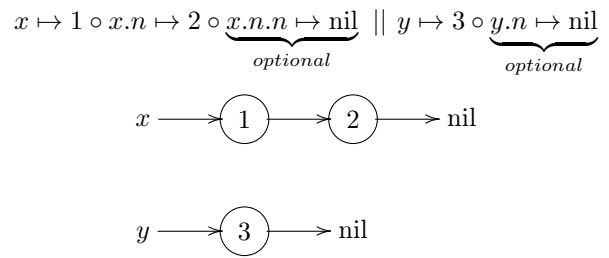
\begin{figure}[h]
$$x\mapsto 1 \circ x.n\mapsto 2 \circ  \underbrace{x.n.n\mapsto \texttt{nil}}_{optional} \ || \ y\mapsto 3 \circ \underbrace{y.n\mapsto \texttt{nil}}_{optional}$$

\begin{center}
\begin{tabular}{l}
\xymatrix{
  x \ar[r] & *++[Fo]\txt{1} \ar[r] & *++[Fo]\txt{2} \ar[r] & \texttt{nil}\\
  y \ar[r] & *++[Fo]\txt{3} \ar[r] & \texttt{nil}
}
\end{tabular}
\end{center}
 \caption{Example of a simply-linked list no.1}
 \label{HeapExampleAP1}
\end{figure}

Next, its list is retrieved as depicted in fig.\ref{HeapExampleAP2}.

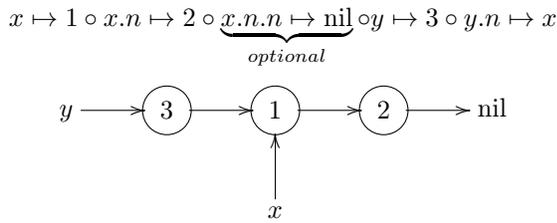
\begin{figure}[h]
$$x\mapsto 1 \circ x.n\mapsto 2 \circ \underbrace{x.n.n\mapsto \texttt{nil}}_{optional} \circ y\mapsto 3 \circ y.n\mapsto x$$

\begin{center}
\begin{tabular}{l}
\xymatrix{
  y \ar[r] & *++[Fo]\txt{3} \ar[r] & *++[Fo]\txt{1} \ar[r] & *++[Fo]\txt{2} \ar[r] & \texttt{nil}\\
           &&  x \ar[u] &&
}
\end{tabular}
\end{center}
 \caption{Example of a simply-linked list no.2}
 \label{HeapExampleAP2}
\end{figure}

%%%%%%%%%%%%%%%
\subsubsection*{\underline{Class Definitions}}
Class instances refer mainly to classes in the sense they are used in C++ or Java, except they only have very rudimentary functionality.
Classes have a one-time bound identifier, fields and methods.
Fields and methods are defined at most once in arbitrary order, except forwarding declarations which are allowed.
Visibility is not mandatory, although allowed, for the universality of the model.
Nested methods are prohibited.
"\texttt{this}" grants access to fields of a class's instance.

\begin{center}
%\fbox{
\begin{minipage}{8cm}
\begin{grammar}
<class_definition> ::=  [ <class_modifier> ] 'class' <ID> '\{' <class_fields_methods> '\}'

<class_fields_methods> ::= \{ <class_field> | <class_method> \}*

<class_field> ::= [ <class_modifier> ] <variable_declaration> ';'

<class_method> ::=
  [ <class_modifier> ] <function_definition>
  \alt [ <class_modifier> ] <function_declaration> ';'

<class_modifier> ::= ( 'private' | 'public' | 'protected' ) ':'
\end{grammar}
\end{minipage}
%}
\end{center}

Type declaration mainly follows the ISO-C++ standard recommendation, which allows uninitialised and initialised variables.

\begin{center}
%\fbox{
\begin{minipage}{8cm}
\begin{grammar}
<variable_declaration> ::=
  <type> <ID> \{ ',' <ID> \}*
  \alt <type> <ID> \{ ',' <ID> \}* '=' <expression>
\end{grammar}
\end{minipage}
%}
\end{center}

Function definitions are straightforward, as known from before C99, so variadic parameters are disallowed.
The calling convention follows "\texttt{cdecl}".
A return type for some method must be stated as either "\texttt{void}" or any other base type (complex types are currently not supported in the prototype).
Object types are scheduled for support in the next future, but all complex modifications are recommended to be passed in function arguments until then.

\begin{center}
%\fbox{
\begin{minipage}{8cm}
\begin{grammar}
<function_definition> ::= <function_declaration>
   <block>

<function_declaration> ::=
  <type_or_void> <ID> '(' ( <formal_parameters> | ) ')'

<type> ::= 'int'

<type_or_void> ::= <type> | 'void'
\end{grammar}
\end{minipage}
%}
\end{center}

%%%%%%%%%%%%%%%
\subsubsection*{\underline{Built-in Rules for Automated Logical Reason}}
%%%%%%%%%%%
Heap nor\-ma\-li\-sa\-tion is made by built-in predicates defined in Prolog or Java.\\
\texttt{simplify/2}.
This predicate accepts some heap term and returns a heap term in which $||$ is utmost outside.
It is assumed that for an optimised and further processing, some heap is in $||$-normal-form, so of kind $q_0 || q_1 || ... || q_k$, where $\forall j.q_i$ is of kind $X_i\circ X_{i+1}\circ ... \circ X_{k}$.

The "\texttt{subtract}" predicate does \texttt{subtraction} (set comparison).
This built-in predicate accepts a list of Prolog terms and creates a list of missing heap terms.
The predicate may be used to check a given predicate family is complete or if there are any missing heap terms (e.g. this is important to simplify, transform and process with Hoare triples in general).

\textbf{Comparison} is performed component-wise by selecting two terms: the expected heap term and the calculated.

%%%%%%%%%%%
%%%%%%%%%%%%%%%
\subsubsection*{\underline{Layered Verification Architecture}}
\label{sect:Layers}
The proposed architecture of this work's verifier called "ProLogika" \cite{haberland19-1} has six main layers (see fig.\ref{LayeredArchitectureVerifier}), which on its peripheral cooperates with libraries, for instance, the logical reasoning library ``\textit{tu\-Pro\-log}'' (its short form is "\textit{2p}"), as well as the CC ANTLR version 4.
The system is based on Java.
Verification, specification and programs are all in (GNU) Prolog and may be defined and edited in the GUI.
ANTLR currently generates Java code, but it may also create code in a different PL, as already mentioned.
It shall also be noted that ANTLR may be replaced since the well-defined IR of an incoming input language may also be generated by other CCs.
Thus, extensibility and variability principles hold by ``\textit{ProLogika}''.
Any replacement for ANTLR is fair as long as its recognition capabilities are compatible with the requirements from earlier in this section.
Most notably is the ability to parse LL(k) or LR-grammars, select different NTs for a given formal grammar, annotate synthesised and inherited attributes, and process incoming and outcoming attributes as termed string or object expression.

\begin{figure}[h]
%\fbox
{\parbox[t]{9cm}{\includegraphics[width=18cm, angle=90]{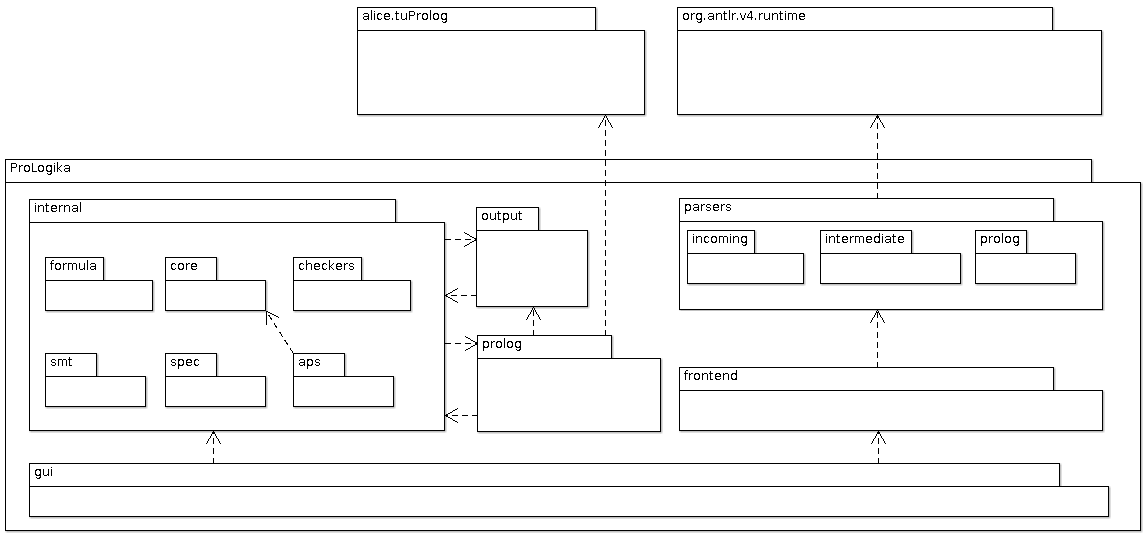}}}\\[-0.7cm]
 \caption{Layered architecture of verifier "ProLogika"}
 \label{LayeredArchitectureVerifier}
\end{figure}

The heap verifier ProLogika consists of \texttt{internal}, \texttt{output}, \texttt{parsers}, \texttt{prolog}, \texttt{frontend} and \texttt{gui}.
\texttt{gui} represents the graphical interface and all controls needed.
Results, IR, and visualisation are (interactively) communicating with the GUI.

The packages \texttt{frontend} and \texttt{parsers} are tightly coupled.
The package \texttt{frontend} contains interfaces and modules of different abstraction level for lexicographical and syntactical analyses.
The parameterisation of syntax analysis is one of the basics of verifying APs.
The package \texttt{parsers} contain apriori generated lexical and syntactic analysers and analysers to be created whilst verification by need.
The interpretation of APs is contained in package \texttt{ProLogika.inter\-nals.aps}.

The \texttt{prolog} package enriches the integrated logical reasoning core by additional Prolog predicates, such as heap term processing, failure location and analysis, and implementing intended side-effects, like I/O-operations.
The implementation of additional built-in predicates is covered by this design, flexible and easy to extend.
Prolog rules may always be added and extended during verification runs.

The package \texttt{output} provides auxiliary functions and predicates for term IR, e.g. for serialisation into \index{DOT} DOT-format or formatted console output.
Currently, the generation of OCL-code is not implemented yet.

The package "\texttt{internal}" implements specification and verification of related functions.
Package "\texttt{core}" contains different modules, which are related to heap term IR processing.
Package "\texttt{aps}" relates to APs and contains all compulsory modules

Consequently, it is also related to the package \texttt{core}.
Package \texttt{checkers} perform heap checks before, after and during a compound program statement is processed.
All remaining packages are subject to permanent corrections and experimental modifications.

%%%%%%%%%%%%%%%
\subsubsection*{\underline{Components}}

\begin{figure}[h]
\begin{center}
 %\fbox
 {\parbox[t]{8cm}{\includegraphics[width=12cm,angle=90]{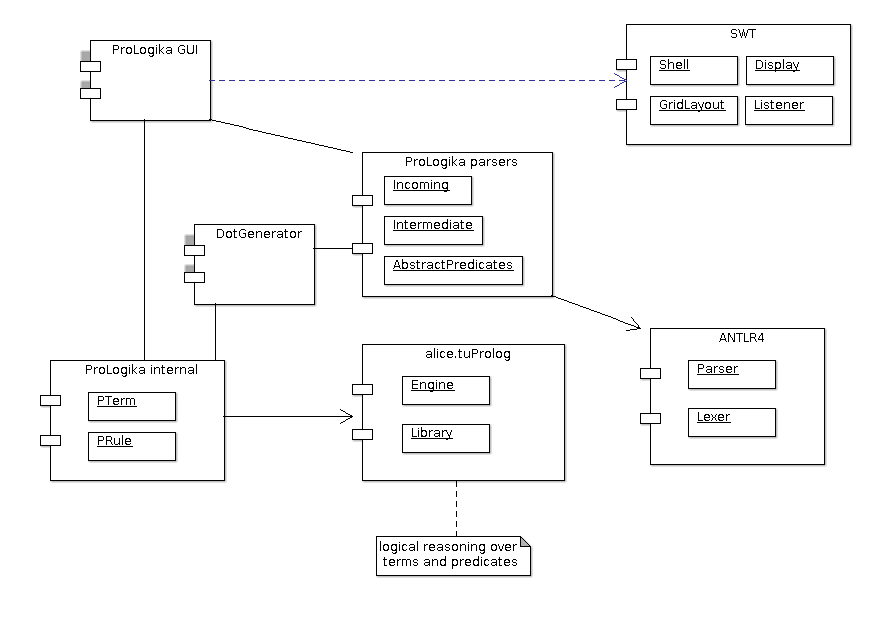}}}\\%[+1cm]
 \caption{ProLogika components}
 \label{VerifierComponents}
\end{center}
\end{figure}

Following sec.\ref{sect:Layers}, six main components can be derived (see fig.\ref{VerifierComponents}):
(i) \texttt{ProLogika.GUI} provides all functionalities with associated visualisation model, view and control,
(ii) \texttt{Pro\-Log\-ika.par\-sers} provide all modules that perform syntax analyses, including constraints checkings on syntax and semantic of generated IR, as well as static checks on APs (all  related to packages \texttt{parsers}, \texttt{incoming}, \texttt{intermediate} and \texttt{AbstractPredicates}).
(iii) the Java \texttt{SWT} library, which is stubbed into dependent components.
It is closely related to \texttt{ProLogika.GUI} and is introduced to dock OS-dependencies.
(iv) Package \texttt{ANTLR} already establishes a component for building language processors.
Its interfaces are already on a very high abstraction level.
Hence, it is taken as is, and no further gapping is required here.
(v) Package \texttt{alice.tuProlog} represents a library.
Its primary functions are encapsulated on ProLogika's side, especially what is not intuitive with the logical engine initialisation and state transitions.
(vi) Package \texttt{ProLogika.internal} represents a component representing Prolog terms (\texttt{PTerm}) and rules (\texttt{PRule}, \texttt{PSubgoal}), which are eagerly used by Prolog's other components.

%%%%%%%%%%%%%%%
\subsubsection*{\underline{Packages}}
%%%%%%%%%%%%%%%%

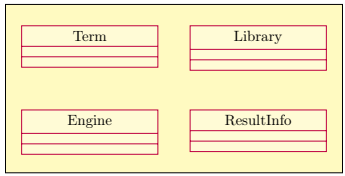
\begin{figure}%[h]
\begin{center}
\scalebox{0.8}{
 \begin{tikzpicture}[scale=0.7, every node/.style={transform shape}]
     \draw[fill=yellow!30] (-2,0.5) rectangle (6,-3.5);
     \begin{class}[text width=3cm]{Term}{0,0}
     \end{class}
     \begin{class}[text width=3cm]{Engine}{0,-2}
     \end{class}
     \begin{class}[text width=3cm]{Library}{4,0}
     \end{class}
     \begin{class}[text width=3cm]{ResultInfo}{4,-2}
     \end{class}
 \end{tikzpicture}}
\end{center}
 \caption{Package alice.tuProlog}
 \label{PackageAlice2p}
\end{figure}

The package \texttt{alice.tuProlog} (see fig.\ref{PackageAlice2p}) contains: (i) class "\texttt{Term}", which represents a common Prolog term that is not abstract, (ii) "\texttt{Library}" is a Prolog-based library, which may contain a whole formal theory build of Prolog rules.
"\texttt{BuiltinLibrary}" is a concrete implementation of tuProlog's "\texttt{Library}".
"\texttt{Engine}" is a container and controller for invoking logical reasoning based on Horn-based rules and a given formal theory.
"\texttt{ResultInfo}" represents the result set as a table that may be accessed in a navigating way (by evaluating subgoals only).\\

 %%%%%%%%%%%%%%%%%%
\texttt{Package parsers:}
 
The package \texttt{parsers} (see fig.\ref{PackageParsers}) contains "\texttt{incoming}", "\texttt{intermediate}", "\texttt{pro\-log}", and the two classes "\texttt{MetaType}" and "\texttt{MetaTypeManager}".
"\texttt{MetaType}" represents the definition of some type in the input program, e.g. class, integer.
"\texttt{Me\-ta\-Type\-Ma\-na\-ger}" is the watchdog instance for all "\texttt{MetaType}" instances derived and referred to during verification.

\begin{figure}%[h]
\begin{center}
\begin{tabular}{cc}
\texttt{parsers.incoming} & \texttt{parsers.intermediate}\\

\scalebox{0.8}{
\begin{tikzpicture}[scale=0.7, every node/.style={transform shape}]
    \draw[fill=yellow!30] (-3,0.5) rectangle (3,-6);
    \begin{class}{ClassBasedLexer}{0,0}
    \end{class}
    \begin{class}{ClassBasedParser}{0,-2}
    \end{class}
    \begin{class}{ClassBasedListener}{0,-4}
    \end{class}
\end{tikzpicture}}
&
\scalebox{0.8}{
\begin{tikzpicture}[scale=0.7, every node/.style={transform shape}]
    \draw[fill=yellow!30] (-3,0.5) rectangle (3,-6);
    \begin{class}{IRPrologLexer}{0,0}
    \end{class}
    \begin{class}{IRPrologParser}{0,-2}
    \end{class}
    \begin{class}{IRPrologListener}{0,-4}
    \end{class}
\end{tikzpicture}}\\

\multicolumn{2}{c}{\texttt{parsers.prolog}}\\
\multicolumn{2}{c}{
\scalebox{0.8}{
\begin{tikzpicture}[scale=0.7, every node/.style={transform shape}]
    \draw[fill=yellow!30] (-3,0.5) rectangle (3,-6);
    \begin{class}{AbstractPredicatesLexer}{0,0}
    \end{class}
    \begin{class}{AbstractPredicatesParser}{0,-2}
    \end{class}
    \begin{class}{AbstractPredicatesListener}{0,-4}
    \end{class}
\end{tikzpicture}}}
\end{tabular}%\\[0.7cm]
\end{center}
 \caption{Package parsers}
 \label{PackageParsers}
\end{figure}
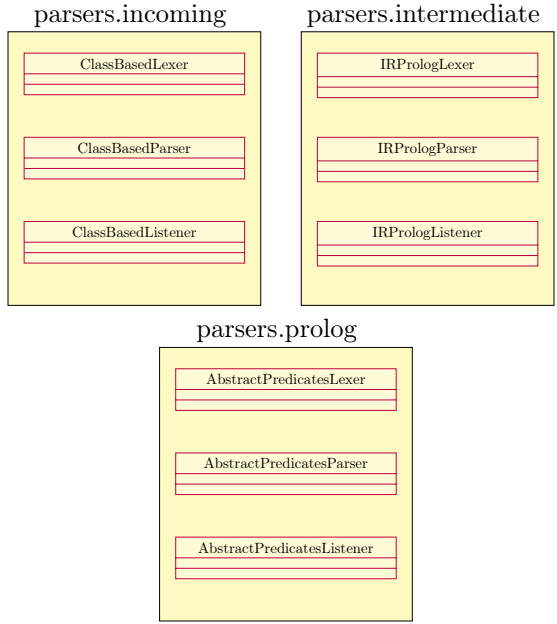

\texttt{Package internal:}
The package contains these sub-packages: "\texttt{core}", "\texttt{checkers}", "\texttt{aps}", "\texttt{spec}".\\\\

\texttt{Package internal.core:}

The package \texttt{core} (see fig.\ref{PackageCore}) contains classes representing the model core for logical reasoning, e.g. terms, rules, symbolic variables, subgoals.
Prefix "\texttt{P}" in the class names indicates each class is a Prolog representation.
Each of the classes has a Prolog-equivalent, which is strictly ISO-conform.

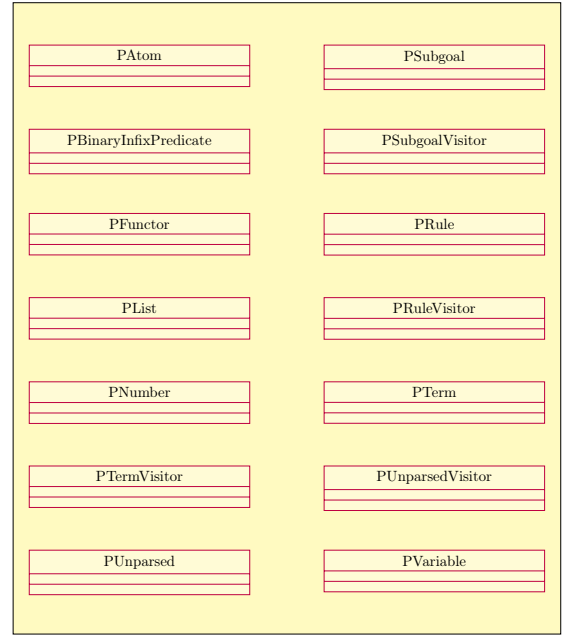
\begin{figure}[h]
\begin{center}
\scalebox{0.8}{
\begin{tikzpicture}[scale=0.7, every node/.style={transform shape}]
    \draw[fill=yellow!30] (-3,1) rectangle (10,-14);
    \begin{class}{PAtom}{0,0}
    \end{class}
    \begin{class}{PBinaryInfixPredicate}{0,-2}
    \end{class}
    \begin{class}{PFunctor}{0,-4}
    \end{class}
    \begin{class}{PList}{0,-6}
    \end{class}
    \begin{class}{PNumber}{0,-8}
    \end{class}
    %%%
    \begin{class}{PSubgoal}{7,0}
    \end{class}
    \begin{class}{PSubgoalVisitor}{7,-2}
    \end{class}
    \begin{class}{PRule}{7,-4}
    \end{class}
    \begin{class}{PRuleVisitor}{7,-6}
    \end{class}
    \begin{class}{PTerm}{7,-8}
    \end{class}
    %%%
    \begin{class}{PTermVisitor}{0,-10}
    \end{class}
    \begin{class}{PUnparsed}{0,-12}
    \end{class}
    \begin{class}{PUnparsedVisitor}{7,-10}
    \end{class}
    \begin{class}{PVariable}{7,-12}
    \end{class}
\end{tikzpicture}}%\\[0.7cm]
\end{center}
 \caption{Package core}
 \label{PackageCore}
\end{figure}

The main base elements are "\texttt{PTerm}", "\texttt{PRule}" and "\texttt{PSubgoal}".
All remaining P-classes correspond to Prolog, except for class "\texttt{Unparsed}", which is a partially calculated Prolog-term.
Partially means it contains gaps, filled in later (cf. gaps-concept for logical terms from \cite{haberland08-1}).\\

\texttt{Package internal.checkers:}

Checks are based on normalisation and subtraction rules, e.g. on APs (\texttt{APCom\-plete\-ness\-Checker}, \texttt{Heap\-Com\-plete\-ness\-Chec\-ker}, \texttt{Pro\-log\-APs\-Chec\-ker}) (see fig.\ref{PackageInternalCheckers}).
Syntactic analyses are performed for some input program and verification conditions (e.g. by \texttt{ANTLR\-Gram\-mar\-Checker}).
All checks need to implement \texttt{Checker\-In\-ter\-face}.\\

\begin{figure}[h]
\begin{center}
\scalebox{0.8}{
\begin{tikzpicture}[scale=0.7, every node/.style={transform shape}]
    \draw[fill=yellow!30] (-3,0.5) rectangle (10,-5.5);
    \begin{interface}{CheckerInterface}{0,0}
    \end{interface}
    \begin{class}{ANTLRGrammarChecker}{0,-2}
    \end{class}
    \begin{class}{APCompletenessChecker}{0,-4}
    \end{class}
    \begin{class}{HeapCompletenessChecker}{7,0}
    \end{class}
    \begin{class}{PrologAPsChecker}{7,-2}
    \end{class}
\end{tikzpicture}}%\\[0.7cm]
\end{center}
 \caption{Packages internal.checkers}
 \label{PackageInternalCheckers}
\end{figure}
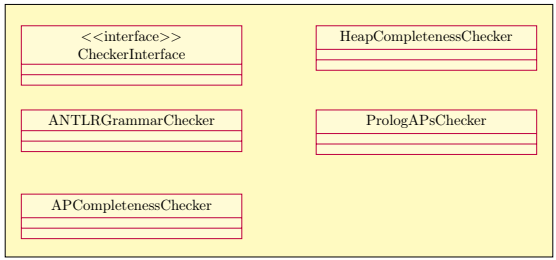

\texttt{Package internal.aps:}

The package "\texttt{aps}" (see fig.\ref{PackageInternalAPs}) contains \texttt{ANTLRBuil\-der}, \texttt{ANTLRLauncher}, \texttt{Grammar\-Buil\-der}, and all other AP-processing classes towards a formal attributed grammar (in a sense Knuth and Wagner introduced).
\texttt{ANTLBuilder} is a specialisation of the class \texttt{GrammarBuilder}.
ANTLR is supposed to represent only as an example a whole class of syntax analysers matching the conditions previously mentioned.
Any other CC and hand-written analysers would be acceptable as long as they implement the standard interface needed.

In analogy to what was said, \texttt{ANTLRGrammar} is a specialisation of class \texttt{Formal\-Grammar}.
\texttt{ANTLRBuilder} is a concrete example for a CC meeting the interface needed.
Therefore it can be considered as a promotion for different analysers to be following.

\texttt{ANTLRLauncher} is a driver for the ANTLR-generator (for a fully-fledged ANT\-LR processor).

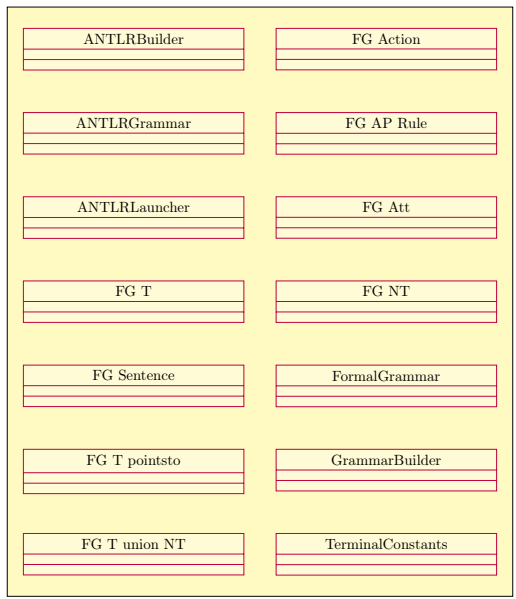
\begin{figure}[h]
\begin{center}
\scalebox{0.8}{
\begin{tikzpicture}[scale=0.7, every node/.style={transform shape}]
    \draw[fill=yellow!30] (-3,0.5) rectangle (9,-13.5);
    \begin{class}{ANTLRBuilder}{0,0}
    \end{class}
    \begin{class}{ANTLRGrammar}{0,-2}
    \end{class}
    \begin{class}{ANTLRLauncher}{0,-4}
    \end{class}
    \begin{class}{FG Action}{6,0}
    \end{class}
    \begin{class}{FG AP Rule}{6,-2}
    \end{class}
    \begin{class}{FG Att}{6,-4}
    \end{class}
    \begin{class}{FG NT}{6,-6}
    \end{class}
    \begin{class}{FG Sentence}{0,-8}
    \end{class}
    \begin{class}{FG T pointsto}{0,-10}
    \end{class}
    \begin{class}{FG T union NT}{0,-12}
    \end{class}
    \begin{class}{FG T}{0,-6}
    \end{class}
    \begin{class}{FormalGrammar}{6,-8}
    \end{class}
    \begin{class}{GrammarBuilder}{6,-10}
    \end{class}
    \begin{class}{TerminalConstants}{6,-12}
    \end{class}
\end{tikzpicture}}
\end{center}
 \caption{Packages internal.aps}
 \label{PackageInternalAPs}
\end{figure}

%%%%%%%%%%%%%%%%%%%%%%%
\subsubsection*{\underline{Syntax analysers}}
Class \texttt{Main} (see fig.\ref{PackagesMain} related to \texttt{Parser} from fig.\ref{PackagesParser} and \texttt{SyntaxAnalyzer} from fig.\ref{PackagesSyntaxAnalyzer}) controls three instances of syntax analysers according to the grammars: (1) ClassBased.g4, (2) IRProlog.g4, and (3) AbstractPredicates.g4.\\
 
\begin{figure}[h]
\begin{center}
\begin{turn}{90}
\parbox[t]{17cm}{
\begin{minipage}{\linewidth}
\begin{tikzpicture}[scale=0.7, every node/.style={transform shape}]
 \begin{class}[text width=11.5cm]{Main}{0,-2}
   \operation{parseClassBasedProgram(prog: String): ClassBasedDescription}
   \operation{parseIRProgram(prog: String): IRPrologDescription}
   \operation{parseAbstractPredicates(s: String): AbstractPredicatesDescription}
 \end{class}
 \begin{abstractclass}[text width=6cm]{SyntaxAnalyzerCommand}{0,-5}
  \attribute{- program: String}
  \operation{+ ClassIdentifier(): String}
  \operation{+ execute(): SyntaxAnalyzer}
  \operation{+ logError(): Logger}
 \end{abstractclass}
 \begin{abstractclass}[text width=6.5cm]{SyntaxAnalyzer}{12,-5}
  \attribute{- tree: ParserRuleContext}
  \operation{+ getLexer(): Lexer}
  \operation{+ getParser(): Parser}
  \operation{+ generateAST(): ParserRuleContext}
  \operation{+ getErrors(): List$<$String$>$}
  \operation{+ launchMainNT(): void}
  \operation{+ isGeneratedASTEmpty(): boolean}
 \end{abstractclass}
 \begin{class}[text width=5cm]{ClassBasedCommand}{-3,-10}
  \inherit{SyntaxAnalyzerCommand}
 \end{class}
 \begin{class}[text width=4cm]{IRPrologCommand}{0,-12}
  \inherit{SyntaxAnalyzerCommand}
 \end{class}
 \begin{class}[text width=6.5cm]{AbstractPredicatesCommand}{4,-10}
  \inherit{SyntaxAnalyzerCommand}
 \end{class}
 \begin{class}[text width=5cm]{ClassBasedDescription}{15,-12}
  \inherit{SyntaxAnalyzer}
 \end{class}
 \begin{class}[text width=6.5cm]{AbstractPredicatesDescription}{8,-12}
  \inherit{SyntaxAnalyzer}
  \attribute{- termTree: ParserRuleContext}
  \operation{+ generateTermAST(): ParserRuleContext}
 \end{class}
 \begin{class}[text width=4cm]{IRPrologDescription}{13,-2}
  \inherit{SyntaxAnalyzer}
 \end{class}
 %%%
 \begin{class}[text width=10cm]{org.antlr4.v4.runtime.ParserRuleContext}{-1,-13.5}
 \end{class}
 %%%
 \aggregation{Main}{commands}{3}{SyntaxAnalyzerCommand}
 \aggregation{SyntaxAnalyzerCommand}{state}{}{SyntaxAnalyzer}
\end{tikzpicture}%\\[0.7cm]
\end{minipage}}
\end{turn}
\end{center}
 \caption{Class Main}
 \label{PackagesMain}
\end{figure}
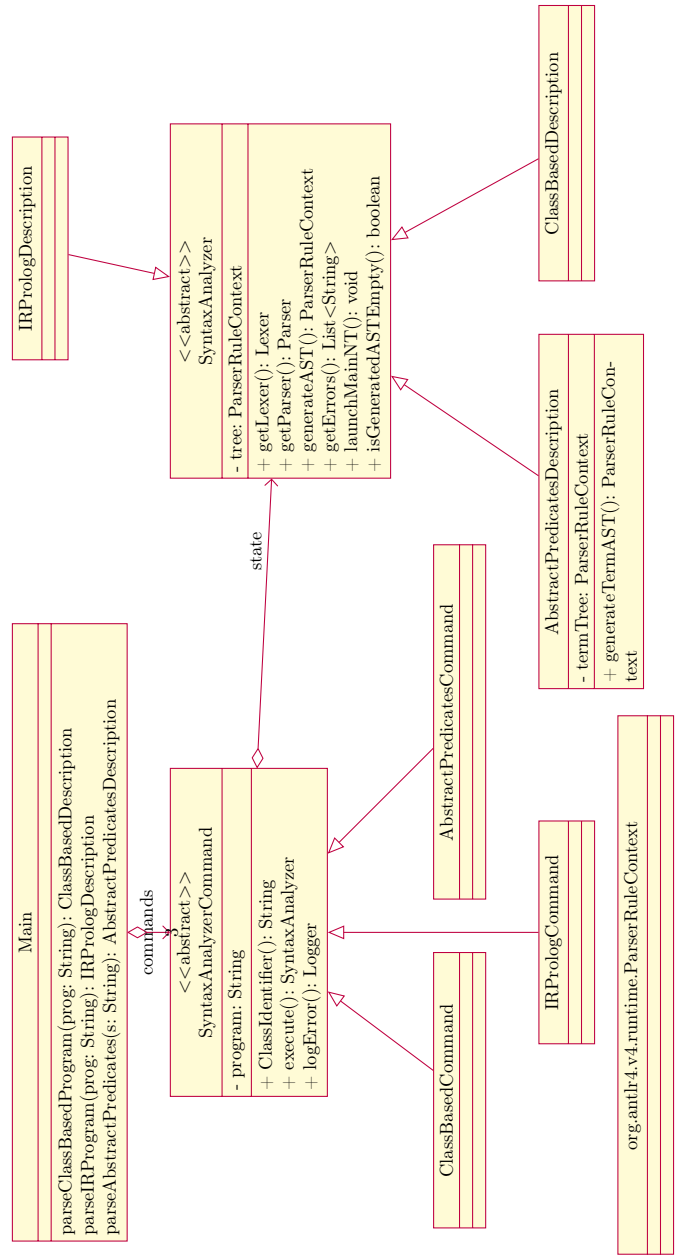

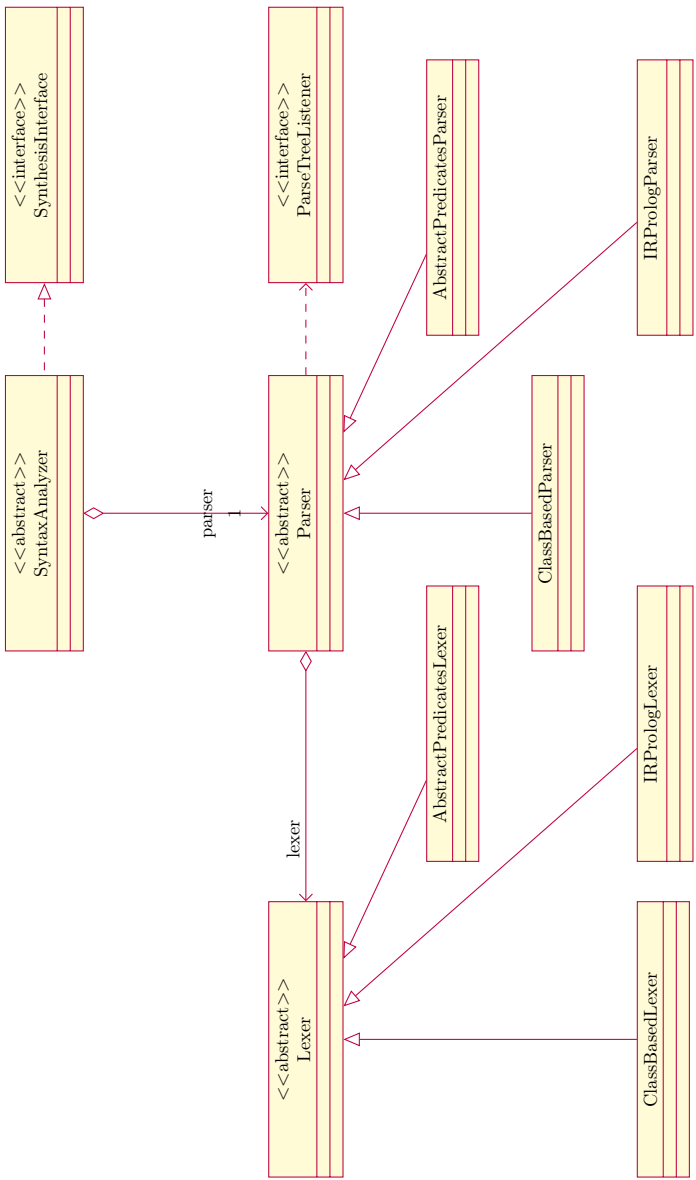
\begin{figure}[h]
\begin{center}
\begin{turn}{90}
\parbox[t]{16cm}{
\begin{minipage}{\linewidth}
\begin{tikzpicture}[scale=0.7, every node/.style={transform shape}]
 \begin{interface}[text width=5cm]{SynthesisInterface}{7,0}
 \end{interface}
 \begin{abstractclass}[text width=5cm]{SyntaxAnalyzer}{0,0}
  \implement{SynthesisInterface}
 \end{abstractclass}
 \begin{abstractclass}[text width=5cm]{Lexer}{-10,-5}
 \end{abstractclass}
 \begin{abstractclass}[text width=5cm]{Parser}{0,-5}
 \end{abstractclass}
 \begin{interface}[text width=5cm]{ParseTreeListener}{7,-5}
 \end{interface}
 \begin{class}[text width=5cm]{ClassBasedLexer}{-10,-12}
   \inherit{Lexer}
 \end{class}
 \begin{class}[text width=5cm]{IRPrologLexer}{-4,-12}
   \inherit{Lexer}
 \end{class}
 \begin{class}[text width=5cm]{AbstractPredicatesLexer}{-4,-8}
   \inherit{Lexer}
 \end{class}
 \begin{class}[text width=5cm]{ClassBasedParser}{0,-10}
   \inherit{Parser}
 \end{class}
 \begin{class}[text width=5cm]{IRPrologParser}{6,-12}
   \inherit{Parser}
 \end{class}
 \begin{class}[text width=5cm]{AbstractPredicatesParser}{6,-8}
   \inherit{Parser}
 \end{class}
 %%%
 \aggregation{SyntaxAnalyzer}{parser}{1}{Parser}
 \aggregation{Parser}{lexer}{}{Lexer}
 \draw[umlcd style dashed line,->] (Parser) |- (ParseTreeListener);
\end{tikzpicture}%\\[0.7cm]
\end{minipage}}
\end{turn}
 \end{center}
 \caption{Class Parser}
 \label{PackagesParser}
\end{figure}

\begin{figure}%[h]
\begin{center}
\begin{turn}{90}
\parbox[t]{17.5cm}{
\begin{minipage}{\linewidth}
\begin{tikzpicture}[scale=0.7, every node/.style={transform shape}]
 \begin{interface}[text width=5cm]{SynthesisInterface}{15,0}
 \end{interface}
 \begin{interface}[text width=6cm]{IRPrologSynthesisInterface}{7,-2}
   \inherit{SynthesisInterface}
   \operation{+ getErrors(): List$<$String$>$}
   \operation{+ getPTerm(): PTerm}
 \end{interface}
 \begin{interface}[text width=8cm]{ClassBasedSynthesisInterface}{8,-6}
   \inherit{SynthesisInterface}
   \operation{+ getVcm(): VerificationConditionManager}
 \end{interface}
 \begin{interface}[text width=8cm]{AbstractPredicatesSynthesisInterface}{17,-4}
   \inherit{SynthesisInterface}
   \operation{+ getPrologRules(): List$<$PRule$>$}
 \end{interface}
 \begin{abstractclass}[text width=6.5cm]{SyntaxAnalyzer}{0,0}
  \attribute{- errorListener: VerboseErrorListener}
  \implement{SynthesisInterface}
 \end{abstractclass}
 \begin{class}[text width=5cm]{BaseErrorListener}{0,-7}
 \end{class}
 \begin{class}[text width=5cm]{VerboseErrorListener}{0,-5}
  \inherit{BaseErrorListener}
 \end{class}
 \aggregation{SyntaxAnalyzer}{errorListener}{}{VerboseErrorListener}
\end{tikzpicture}%\\[0.7cm]
\end{minipage}}
\end{turn}
\end{center}
 \caption{Class SyntaxAnalyzer}
 \label{PackagesSyntaxAnalyzer}
\end{figure}
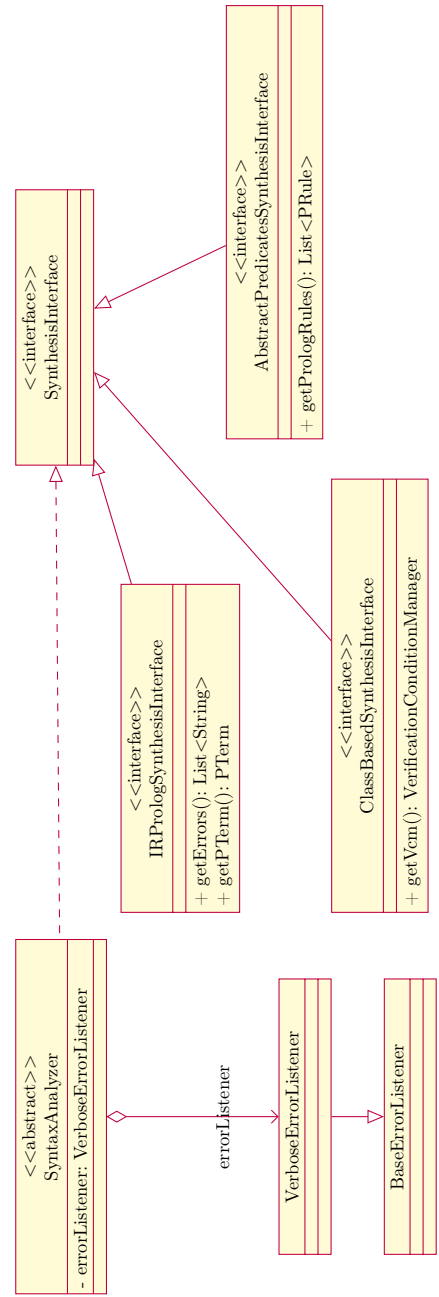

\texttt{Package frontend:} see fig.\ref{PackageFrontend1}.

\begin{figure}%[h]
\begin{center}
\begin{turn}{90}
\parbox[t]{14cm}{
\begin{minipage}{\linewidth}
\begin{tikzpicture}[scale=0.7, every node/.style={transform shape}]
 \begin{abstractclass}[text width = 5cm]{SyntaxAnalyzerCommand}{0,0}
   \operation{execute(): SyntaxAnalyzer}
 \end{abstractclass}
 \begin{class}[text width = 5cm]{ClassBasedCommand}{-7,-5}
   \inherit{SyntaxAnalyzerCommand}
   \operation{execute(): SyntaxAnalyzer}
 \end{class}
 \begin{class}[text width = 5cm]{IRPrologCommand}{0,-5}
   \inherit{SyntaxAnalyzerCommand}
   \operation{execute(): SyntaxAnalyzer}
 \end{class}
\begin{class}[text width = 6cm]{AbstractPredicatesCommand}{7,-5}
   \inherit{SyntaxAnalyzerCommand}
   \operation{execute(): SyntaxAnalyzer}
 \end{class}
\end{tikzpicture}
\end{minipage}}
\end{turn}
\end{center}
 \caption{Package frontend}
 \label{PackageFrontend1}
\end{figure}

Remark:
\texttt{IRPrologCommand.execute()} returns a \texttt{IRPrologDescription object}.
Given the design pattern roles:

\begin{verbatim}
ADAPTER::(CLASS) ADAPTER:
  SyntaxAnalyzerCommand

ADAPTER::ADAPTEE:
  ClassBasedCommand, IRPrologCommand,
             AbstractPredicatesCommand
\end{verbatim}

\texttt{Package parsers.incoming:}

The following structure is applicable, not just to \texttt{parsers.incoming} (see fig.\ref{PackagesParsersIncoming}) but also to \texttt{par\-sers.in\-ter\-me\-di\-ate} and \texttt{par\-sers.pro\-log}.
(\textit{reduce}) classes with \texttt{Listener}-suffixes were introduced to modify syntax analysis when entering and leaving an NT's rule.\\

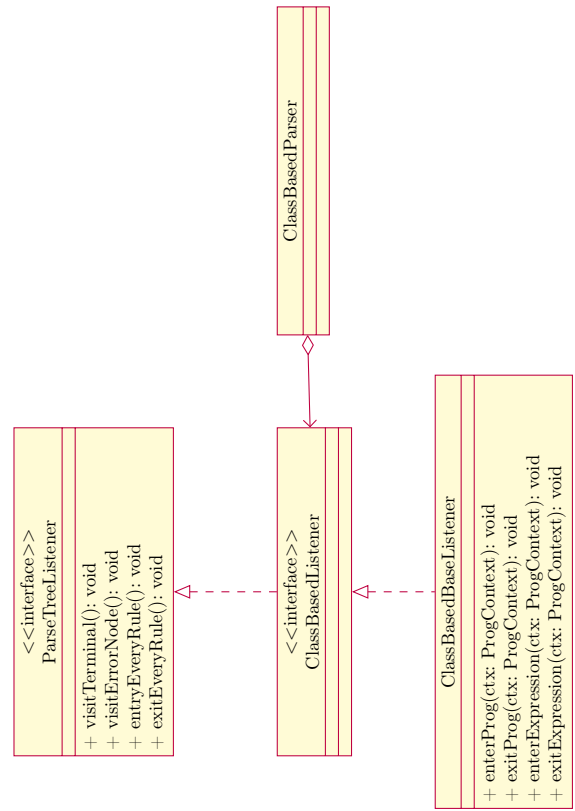
\begin{figure}[h]
\begin{center}
\begin{turn}{90}
\parbox[t]{11cm}{
\begin{minipage}{\linewidth}
\begin{tabular}{c}
\begin{tikzpicture}[scale=0.7, every node/.style={transform shape}]
 \begin{interface}[text width = 6cm]{ParseTreeListener}{0,5}
   \operation{+ visitTerminal(): void}
   \operation{+ visitErrorNode(): void}
   \operation{+ entryEveryRule(): void}
   \operation{+ exitEveryRule(): void}
 \end{interface}
 \begin{interface}[text width = 6cm]{ClassBasedListener}{0,0}
   \implement{ParseTreeListener}
 \end{interface}
 \begin{class}[text width = 6cm]{ClassBasedParser}{8,0}
 \end{class}
 \begin{class}[text width = 8cm]{ClassBasedBaseListener}{0,-3}
   \implement{ClassBasedListener}
   \operation{+ enterProg(ctx: ProgContext): void}
   \operation{+ exitProg(ctx: ProgContext): void}
   \operation{+ enterExpression(ctx: ProgContext): void}
   \operation{+ exitExpression(ctx: ProgContext): void}
 \end{class}
 %
 %% assoc
 \aggregation{ClassBasedParser}{}{}{ClassBasedListener}
\end{tikzpicture}
\end{tabular}%\\[0.7cm]
\end{minipage}}
\end{turn}
\end{center}
 \caption{Package parsers.incoming}
 \label{PackagesParsersIncoming}
\end{figure}

%%%%%%%%%%%%%%%%%%%%%

\texttt{Package parsers:} see fig.\ref{PackagesParsers}.

\begin{figure}[h]
\begin{center}
\begin{turn}{90}
\parbox[t]{17cm}{
\begin{minipage}{\linewidth}
\begin{tikzpicture}[scale=0.7, every node/.style={transform shape}]
 \begin{class}[text width = 12cm]{MetaType}{-13,0}
   \attribute{+ \underline{UNDEFINED_TYPE_ID: int = -1}}
   \attribute{+ \underline{INT_TYPE_ID: int = 0}}
   \attribute{+ \underline{VOID_TYPE_ID: int = 1}}
   \attribute{+ TypeID: int}
   \attribute{+ TypeName: String}
   \attribute{+ CoercionType: CoercionType}
   \attribute{+ PSignature: List$<$MetaType$>$}
   \attribute{+ CFields: HashMap$<$String,MetaType$>$}
   \attribute{+ CMethods: HashMap$<$String,LinkedList$<$MetaType$>>$}
   \operation{- MetaType()}
   \operation{- \underline{createIntMetaType(): MetaType}}
   \operation{- \underline{createVoidMetaType(): MetaType}}
   \operation{+ equals(type: MetaType): boolean}
   \operation{+ \underline{equals(left: List$<$MetaType$>$, right: List$<$MetaType$>$): boolean}}
   \operation{+ allButLast(l: List$<$MetaType$>$): List$<$MetaType$>$}
   \operation{+ toString(): String}
   \attribute{+ \underline{INTTYPE: MetaType}}
   \attribute{+ \underline{VOIDTYPE: MetaType}}
 \end{class}
 \begin{class}[text width = 10cm]{MetaTypeManager}{0,0}
   \attribute{+ Types_Table: LinkedList$<$MetaType$>$}
   \attribute{+ Errors: LinkedList$<$String$>$}
   \operation{+ existsVariableType(typeName: String): MetaType}
   \operation{+ existsClassType(typeName: String): MetaType}
   \operation{+ existsProcedureSignature(typeName: String): MetaType}
   \operation{+ lookup(newType: MetaType): MetaType}
   \operation{+ mangle(className: String, idName: String): String}
   \operation{+ demangleClassName(mangledName): String}
   \operation{+ demangleProcedureName(mangledName): String}
   \operation{+ checkUniqueMain(): boolean}
 \end{class}
 \begin{class}{CoercionType}{0,-9}
   \attribute{+ \underline{INT: int = 0}}
   \attribute{+ \underline{VOID: int = 1}}
   \attribute{+ \underline{PROCEDURE: int = 2}}
   \attribute{+ \underline{CLASS: int = 3}}
 \end{class}
 % associations
 \composition{MetaTypeManager}{Types\_Table}{1}{MetaType}
 \aggregation{MetaType}{CoercionType}{1}{CoercionType}
 \draw[umlcd style dashed line,->] (MetaTypeManager.south) -- (CoercionType.north);
\end{tikzpicture}
\end{minipage}}
\end{turn}
\end{center}
 \caption{Package parsers}
 \label{PackagesParsers}
\end{figure}

The class \texttt{MetaType} represents the built-in and user-defined types found before type checking in an input program.
Each type, though, is identified by some field, "\texttt{typeID}".
However, during type construction, "\texttt{typeID}" is not assigned initially (just the name, but without a full signature, since recursive types are allowed).
Statically defined built-in types are "\texttt{INTTYPE}" and "\texttt{VOIDTYPE}".

Elements of package \texttt{parsers.intermediate} act as DIRECTOR for the class "\texttt{Me\-ta\-Type\-Ma\-na\-ger}".\\[0.7cm]

\texttt{Package frontend} see fig.\ref{PackageFrontend2}.

\begin{figure}[h]
\begin{center}
\begin{tikzpicture}[scale=0.6, every node/.style={transform shape}]
 \begin{abstractclass}[text width = 6cm]{SyntaxAnalyzer}{0,0}
   \attribute{tree: ParserContextTree}
   \operation{generateAST(): ParserRuleContext}
   \operation{launchMainNT()}
   \operation{prelaunch()}
   \operation{postlaunch()}
 \end{abstractclass}
 \begin{class}[text width=7cm]{ClassBasedDescription}{-4,-5.5}
   \inherit{SyntaxAnalyzer}
   \operation{generateAST(): ParserRuleContext}
   \operation{getVcm(): VerificationConditionManager}
   \operation{getResult(): String}
 \end{class}
 \begin{class}[text width=6cm]{IRPrologDescription}{0,-9}
   \inherit{SyntaxAnalyzer}
   \operation{generateAST(): ParserRuleContext}
   \operation{getPTerm(): PTerm}
   \operation{getErrors(): List$<$String$>$}
 \end{class}
 \begin{class}[text width=6cm]{AbstractPredicatesDescription}{4,-5.5}
   \inherit{SyntaxAnalyzer}
   \operation{generateAST(): ParserRuleContext}
   \operation{getPrologRules(): List$<$PRule$>$}
 \end{class}
\end{tikzpicture}
\end{center}
 \caption{Package frontend}
 \label{PackageFrontend2}
\end{figure}

Remarks:

\begin{itemize}

\item %
\begin{verbatim}
STRATEGY::TEMPLATE:
  SyntaxAnalyzer.launchMainNT()

STRATEGY::CONCRETE STRATEGY:
  ClassBasedDescription.generateAST()
  IRPrologDescription.generateAST()
  AbstractPredicates.generateAST() 
\end{verbatim}

 \item The method \texttt{SyntaxAnalyzer.launchMainNT} is:
\begin{verbatim}
prelaunch();
this.tree=generateAST();
postlaunch();
\end{verbatim}

 \item In class "\texttt{IRPrologDescription}", \texttt{get}-methods mean access to synthesised attributes of the corresponding formal grammar.

 \item Class "\texttt{ParserRuleContext}" contains all that is needed by inherited and synthesised attributes.
\end{itemize}

%%%%%%%%%%%%%%%%%%%%%
\subsubsection*{\underline{Prolog Terms}}

\texttt{Package internal.core} see fig.\ref{PackageInternalCore}.

This package contains all IR with terms, subgoals and rules.

\begin{figure}[h]
\begin{center}
\begin{turn}{90}
\parbox[t]{13cm}{
\begin{minipage}{\linewidth}
\begin{tikzpicture}[scale=0.5, every node/.style={transform shape}]
  \begin{abstractclass}[text width = 5.2cm]{PTerm}{0,-1}
    \attribute{children : List$<$PTerm$>$}
    \attribute{name: String}
    \attribute{visitors: List$<$PTermVisitor$>$}
    \operation{addVisitor(v: PTermVisitor)}
    \operation{acceptVisitors()}
    \operation{setName(s: String)}
    \operation{getName(): String}
    \operation{getChildren(): List$<$PTerm$>$}
    \operation{setChildren(l: List$<$PTerm$>$)}
  \end{abstractclass}
  %%%%%%%%%%%%%%%%%%%%%
  \begin{class}[text width=4cm]{PList}{-5, -10}
    \inherit{PTerm}
    \operation{addChild(kid: PTerm)}
    \operation{toString(): String}
  \end{class}
  \begin{class}[text width=4cm]{PFunctor}{5, -10}
    \inherit{PTerm}
    \operation{setName(s: String)}
    \operation{addChild(kid: PTerm)}
    \operation{toString(): String}
  \end{class}
  \begin{class}[text width=3cm]{PVariable}{11, -3}
    \inherit{PTerm}
  \end{class}
  \begin{class}[text width=6cm]{PUnparsedVisitor}{11, -5}
    \operation{+ visitUnparsed(u: PUnparsed)}
  \end{class}
  \begin{class}[text width=5cm]{PNumber}{0, -8}
    \inherit{PTerm}
    \operation{getValue(): String}
    \operation{setValue(s: String)}
    \operation{toString(): String}
    \operation{isInteger(): boolean}
    \operation{getIntegerValue(): int}
    \operation{getFloatingValue(): float}
    \operation{isCoercionSuccess(): boolean}
  \end{class}
  \begin{class}[text width=4cm]{PUnparsed}{11, -8}
    \inherit{PTerm}
    \operation{toString(): String}
  \end{class}
  \begin{class}[text width=5cm]{PAtom}{11, -10}
    \inherit{PTerm}
    \operation{setName(s: String)}
    \operation{getChildren(): List$<$PTerm$>$}
    \operation{setChildren(kids: List$<$PTerm$>$)}
  \end{class}
  \begin{interface}{PSubgoalVisitor}{-8,5}
   \operation{visitSubgoal(s: PSubgoal)}
  \end{interface}
  \begin{class}[text width=6cm]{PSubgoal}{-8,0}
    \implement{PSubgoalVisitor}
    \attribute{relationName: String}
    \attribute{terms: List$<$PTerm$>$}
    \operation{toString(): String}
    \operation{addVisitor(v: PSubgoalVisitor)}
    \operation{acceptVisitors()}
  \end{class}
  \begin{class}[text width=5cm]{PBinaryInfixPredicate}{-8,-7}
     \inherit{PSubgoal}
     \operation{addTerm(t: PTerm)}
     \operation{toString(): String}
  \end{class}
  \begin{class}[text width=6cm]{PRule}{0,5}
    \attribute{name: String}
    \attribute{arguments: List$<$PTerm$>$}
    \attribute{rhs: List$<$Subgoal$>$}
    \attribute{visitors: List$<$PRuleVisitor$>$}
    \operation{addRhs(rhs: PSubgoal)}
    \operation{addArgument(newArg: PTerm)}
    \operation{toString(): String}
    \operation{addVisitor(v: PRuleVisitor)}
    \operation{acceptVisitors()}
  \end{class}
  \begin{interface}{PRuleVisitor}{11,5}
    \operation{visitRule(rule: PRule)}
  \end{interface}
  \begin{interface}{PTermVisitor}{11,0}
    \operation{visitTerm(term: PTerm)}
  \end{interface}
  
  %% associations:
  \aggregation{PSubgoal}{terms}{0..*}{PTerm}
%%  \draw[umlcd style dashed line,->](PTerm.north) ++ (1,0) -- ++(0,0.3) -- ++(5,0) node[above,sloped,black]{children} --  ++(0,-3) -| (PTerm.east);
  \aggregation{PTerm}{visitors}{0..*}{PTermVisitor}
  \aggregation{PRule}{visitors}{0..*}{PRuleVisitor}
  \aggregation{PRule}{arguments}{0..*}{PTerm}
  \aggregation{PUnparsed}{visitors}{0..*}{PUnparsedVisitor}
\end{tikzpicture}
\end{minipage}}
\end{turn}
\end{center}
 \caption{Package internal.core}
 \label{PackageInternalCore}
\end{figure}
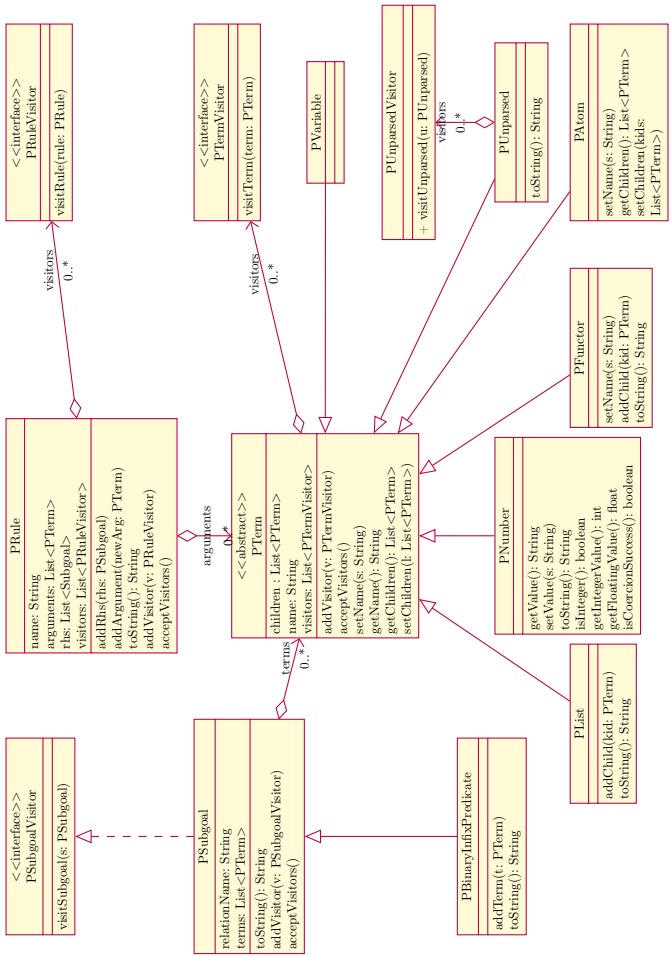

\begin{verbatim}
COMPOSITE::COMPONENT:
 PTerm
COMPOSITE::COMPOSITE:
 PTerm
COMPOSITE::CONCRETE COMPOSITE:
 PAtom, PVariable, PFunctor, PNumber, PList
\end{verbatim}

\texttt{PAtom.getChildren} returns \texttt{null}.
\texttt{PAtom.setChildren} ignores the set and, if possible, issues a warning to the console.
\texttt{PSubgoal} represents a subgoal, which may further contain \texttt{PTerm} as content (non-subgoals).

%%%%%%%%%%%%%%%%%%%%%

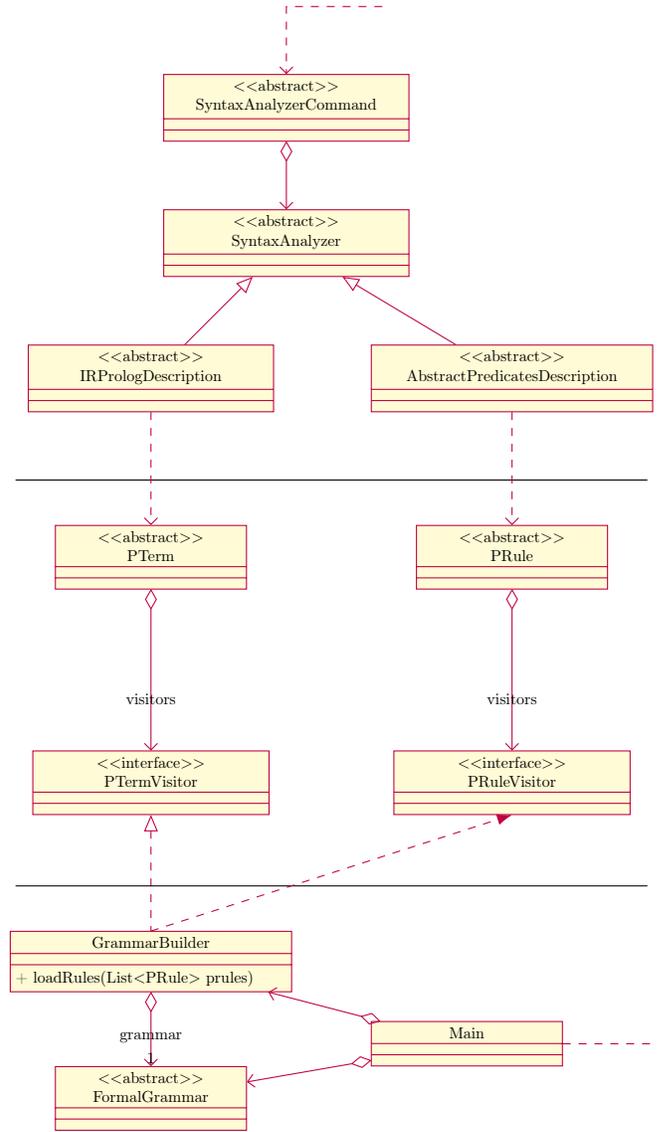
\begin{figure}[h]
\begin{tabular}{l}
%  \begin{minipage}[c]{3.5cm}
%  \parbox[t]{3.5cm}{
%  \texttt{frontend}\\[6cm]
%  \texttt{internal.core}\\[5cm]
%  \texttt{internal.aps}}
%  \end{minipage}
% &
 \begin{minipage}[c]{4cm}
\begin{tikzpicture}[scale=0.6, every node/.style={transform shape}]
  \begin{abstractclass}[text width = 5.2cm]{SyntaxAnalyzerCommand}{10,3}
  \end{abstractclass}
  \begin{abstractclass}[text width = 5.2cm]{SyntaxAnalyzer}{10,0}
  \end{abstractclass}
  \begin{abstractclass}[text width = 5.2cm]{IRPrologDescription}{7,-3}
    \inherit{SyntaxAnalyzer}
  \end{abstractclass}
  \begin{abstractclass}[text width = 6cm]{AbstractPredicatesDescription}{15,-3}
    \inherit{SyntaxAnalyzer}
  \end{abstractclass}
  \aggregation{SyntaxAnalyzerCommand}{}{}{SyntaxAnalyzer}

  %% layer separation
  \draw (4,-6) -- (18,-6);

  \begin{interface}{PTermVisitor}{7,-12}
  \end{interface}
  \begin{abstractclass}[text width=4cm]{PTerm}{7,-7}
  \end{abstractclass}
  \begin{abstractclass}[text width=4cm]{PRule}{15,-7}
  \end{abstractclass}
  \begin{interface}{PRuleVisitor}{15,-12}
  \end{interface}
  \aggregation{PTerm}{visitors}{}{PTermVisitor}
  \aggregation{PRule}{visitors}{}{PRuleVisitor}

  %% layer separation
  \draw (4,-15) -- (18,-15);

  \begin{abstractclass}[text width=4cm]{FormalGrammar}{7,-19}
  \end{abstractclass}
  \begin{class}[text width = 6cm]{GrammarBuilder}{7,-16}
     \implement{PTermVisitor}
%    \implement{PRuleVisitor}  % workaround for a rare tixz flaw, see below !
     \operation{+ loadRules(List$<$PRule$>$ prules)}
  \end{class}
  \begin{class}[text width = 4cm]{Main}{14,-18}
  \end{class}

  \aggregation{GrammarBuilder}{grammar}{1}{FormalGrammar}
  \aggregation{Main}{}{}{GrammarBuilder}
  \aggregation{Main}{}{}{FormalGrammar}
  \draw[umlcd style dashed line,-{Latex[length=2mm]}] (GrammarBuilder.north) -- (PRuleVisitor.south); % this is a workaround
  \draw[umlcd style dashed line,->] (IRPrologDescription.south) -- (PTerm.north);
  \draw[umlcd style dashed line,->] (AbstractPredicatesDescription.south) -- (PRule.north);
  \draw[umlcd style dashed line,->] (Main.east) --node[above,sloped,black]{} ++(2,0) ++(0,23) ++(-6,0) -| (SyntaxAnalyzerCommand.north);
\end{tikzpicture}
 \end{minipage}
\end{tabular}
 \caption{Package layers frontend, internal.core, internal.aps}
 \label{PackagesFrontendInternalcoreInternalaps}
\end{figure}

Remarks:

\begin{itemize}
 \item In package "\texttt{internal.core}" (see fig.\ref{PackagesFrontendInternalcoreInternalaps}), there is no need for some "\texttt{PSubgoal}", since it is already included in "\texttt{PRule}", and since subgoals are part of queries (which are not part of the grammar generated by APs).

 \item The method "\texttt{GrammarBuilder.loadRules}" reads all "\texttt{PRule}" and puts them into the formal grammar, where attributes remain the way they are in Prolog.
       Class "\texttt{GrammarBuilder}" implements both interfaces "\texttt{PTermVisitor}" and "\texttt{PRuleVisitor}".
       This distinction is required for rule interpretation and because one starting NT is required.

 \item "\texttt{PTerm}" and "\texttt{PRule}"-instances are generated in dependency of instance descriptions (Description-objects).
\end{itemize}

%%%%%%%%%%%%%%%%%%%%

\subsubsection*{\underline{Abstract Predicates}}

APs are needed for representation and logical processes.
Used analysers all have one starting NT symbol, although an analyser must have provided the possibility to switch starting NTs.
For demonstration purposes, ANTL-grammars are considered.\\

\texttt{Package internal.aps} see fig.\ref{PackagesInternalAPs}.

\begin{figure}[h]
\begin{center}
\begin{turn}{90}
\parbox[t]{13cm}{
\begin{minipage}{\linewidth}
\begin{tikzpicture}[scale=0.5, every node/.style={transform shape}]
 \begin{abstractclass}[text width = 6cm]{GrammarBuilder}{-10,0}
   \attribute{- grammar: FormalGrammar}
   \operation{+ loadRules(rules: List$<$PRule$>$)}
 \end{abstractclass}
 \begin{class}[text width = 6cm]{ANTLRBuilder}{0,0}
   \operation{+ reset()}
   \operation{+ build()}
   \operation{- buildGrammarName(): String}
   \operation{- buildGrammarOptions(): String}
   \operation{- buildGrammarHeader(): String}
   \operation{- buildGrammarTerminals(): String}
   \inherit{GrammarBuilder}
 \end{class}
 \begin{class}[text width = 4cm]{TerminalConstants}{8,-1}
   \attribute{\underline{UNIFY: int = 0}}
   \attribute{\underline{NONUNIFY: int = 1}}
   \attribute{\underline{LPAREN: int = 2}}
 \end{class}
 \begin{abstractclass}[text width = 5cm]{FG T unioin NT}{8,-7}
   \operation[0]{+ toString(): String}
 \end{abstractclass}
 \begin{abstractclass}[text width = 7.5cm]{FormalGrammar}{-10,-5}
   \attribute{- statNonterminal: FG_NT}
   \attribute{- nonterminals: HashMap$<$String,FG_NT$>$}
   \attribute{- terminals: HashMap$<$String,FG_T$>$}
   \attribute{- rules: List$<$FG_AP_Rule$>$}
   \attribute{- errors: List$<$String$>$}
 \end{abstractclass}
 \begin{class}[text width = 5cm]{FG NT}{0,-6}
   \attribute{- nonterminalName: String}
   \attribute{- attributes: List$<$PTerm$>$}
   \attribute{- compactNotation: boolean}
   \operation{+ toString(): String}
   \inherit{FG T unioin NT}
 \end{class}
 \begin{class}[text width = 5cm]{ANTLRGrammar}{-10,-12}
   \attribute{- grammarName: String}
   \attribute{- grammarOptions: String}
   \attribute{- grammarHeader: String}
   \attribute{- grammarText: String}
   \operation{+ toString(): String}
   \inherit{FormalGrammar}
 \end{class}
 \begin{class}[text width = 5cm]{FG T}{0,-10}
   \attribute{- characterName: String}
   \attribute{- definition: String}
   \attribute{- hasDefinition: boolean}
   \operation{+ toString(): String}
   \inherit{FG T unioin NT}
 \end{class}
 \begin{class}[text width = 10cm]{FG AP Rule}{-2,-15}
   \attribute{- ruleName: String}
   \attribute{- alternatives: HashMap$<$List$<$FG_Att$>$, FG_Sentence$>$}
 \end{class}
 \begin{class}[text width = 4cm]{FG Att}{5.7,-12}
   \attribute{- varName: String}
 \end{class}
 \begin{class}[text width = 6cm]{FG Sentence}{8.5,-15}
   \attribute{sentence: List$<$FG_T_union_NT$>$}
 \end{class}

 %% associations
 \aggregation{GrammarBuilder}{grammar}{1}{FormalGrammar}
 \aggregation{FormalGrammar}{nonterminals}{0..*}{FG NT}
 %\composition{FormalGrammar}{startNonterminal}{}{FG NT}   % yet unresolved flaw: overlaps
 \aggregation{FormalGrammar}{terminals}{0..*}{FG T}
 \aggregation{FormalGrammar}{rules}{0..*}{FG AP Rule}
 \aggregation{FG AP Rule}{alternatives}{0..*}{FG Att}
 \aggregation{FG Sentence}{sentence}{0..*}{FG T unioin NT}
 \aggregation{FG AP Rule}{alternatives}{0..*}{FG Sentence}
 \draw[umlcd style dashed line,->] (ANTLRBuilder.east) -- (TerminalConstants.west);
\end{tikzpicture}
\end{minipage}}
\end{turn}
\end{center}
 \caption{Package internal.aps}
 \label{PackagesInternalAPs}
\end{figure}

Remarks: All these classes are involved in a formal grammar's notation (further implementations shall be considered using ANTLR).
"\texttt{PRule}" may be presented in the form "\texttt{FG\_AP\_Rule}", (and vice versa) in the following sense:

\begin{center}
\begin{tabular}{c}
   $\texttt{p1(X,pointsto(X,value1)):- ... .}$\\
   $\Leftrightarrow$\\
   $p1_{x1,x2} \rightarrow \{X_1 \approx X\}\{X_2 \approx pointsto(X,value1)\}, ... .$
\end{tabular}
\end{center}

Next, \texttt{FG\_NT} has attributes \texttt{List<PTerm>}, except \texttt{List<FG\_Att>}, because it refers only to an existing NT with a concrete sequence of \texttt{PTerm}, by the ordering of attribute definitions.

%%%%%%%%%%%%%%%%%%%%

\subsubsection*{\underline{Heap Terms}}

Prolog terms are defined inductively over operators $\{ \circ, || \}$.
The term obtained is the result of previous syntax analysis.\\

\texttt{Package internal.ht}, see fig.\ref{PackagesInternalHT}.

Subgoals in APs may generate a connected (sub-)graph.
$a\circ a$ is currently done for the general case for simplicity purposes.
$a_1 \circ a_2 \circ a_3$ may directly be noted using $a_1,a_2,a_3$ (as part of some AP) or as APs $a_x :- a_1, a_2.$, wherein the following may act as some subgoal: $..., a_x, a_3, ...$.\\\\
Example 2) $a_1 \circ (a_2 || a_3)$ may be rewritten as AP $$a_x :- a_2,disj,a_3.$$, where there is some sequence as part of an AP $...,a_1,a_x,...$.

\begin{figure}[h]
\begin{center}
\begin{turn}{0}
\parbox[t]{9cm}{
\begin{minipage}{\linewidth}
\begin{tikzpicture}[scale=0.6, every node/.style={transform shape}]
 \begin{abstractclass}[text width = 8cm]{HeapTerm}{-2,0}
   \attribute{- LHS: PathExpression}
   \attribute{- RHS: Expression}
   \attribute{- children: List$<$HeapTerm$>$}
   \operation{+ simplify(): HeapTerm}
   \operation{+ foldl(initial: HeapTerm): HeapTerm}
   \operation{+ foldr(initial: HeapTerm): HeapTerm}
   \operation{+ conj(h: HeapTerm): HeapTerm}
   \operation{+ disj(h: HeapTerm): HeapTerm}
   \operation{+ equals(h: HeapTerm): Boolean}
   \operation{+ normalise()}
   \operation{+ substract(h: HeapTerm): HeapTerm}
   \operation{+ substract(hs: List$<$HeapTerm$>$): HeapTerm}
   \operation[0]{+ toProlog(): String}
   \operation[0]{+ fromProlog(term: HeapTerm)}
 \end{abstractclass}
 
 \begin{class}[text width = 5cm]{DotGenerator}{0,3}
 \end{class}
 
 \begin{interface}[text width = 5cm]{CheckerInterface}{6,-3}
   \operation{+ getNewErrors(): Collection$<$String$>$}
   \operation{+ check()}
 \end{interface}
 
 \begin{class}[text width = 5cm]{HeapSoundnessChecker}{6,0}
   \implement{CheckerInterface}
 \end{class}
 
 \begin{class}[text width = 5cm]{HeapCompletenessChecker}{6,-8}
   \implement{CheckerInterface}
 \end{class}
 
 \begin{class}[text width = 6cm]{AnyHeap}{-3,-12}
   \operation{+ toProlog(): String}
   \operation{+ fromProlog(term: HeapTerm)}
   \implement{HeapTerm}
 \end{class}

 \begin{class}[text width = 6cm]{Heaplet}{5,-12}
   \operation{+ toProlog(): String}
   \operation{+ fromProlog(term: HeapTerm)}
   \implement{HeapTerm}
 \end{class}

 \begin{class}[text width = 6cm]{HeapDisjunction}{-3,-16}
   \operation{+ toProlog(): String}
   \operation{+ fromProlog(term: HeapTerm)}
   \implement{HeapTerm}
 \end{class}

 \begin{class}[text width = 6cm]{HeapConjunction}{5,-16}
   \operation{+ toProlog(): String}
   \operation{+ fromProlog(term: HeapTerm)}
   \implement{HeapTerm}
 \end{class}

 %% associations
 \aggregation{DotGenerator}{}{1}{HeapTerm}
 \aggregation{HeapTerm}{}{1}{HeapSoundnessChecker}
 \aggregation{HeapTerm}{}{1}{HeapCompletenessChecker}
% \draw[umlcd style dashed line,->](HeapTerm.north) ++ (4,0) -- ++(0,1) -- ++(3,0) node[above,sloped,black]{children} -- ++(0,-5) -| (HeapTerm.east);
\end{tikzpicture}
\end{minipage}}
\end{turn}
\end{center}
 \caption{Package internal.ht}
 \label{PackagesInternalHT}
\end{figure}

Remarks:
\begin{itemize}
 \item \texttt{HeapTerm} is an internal IR for heaps based on SL.
       Heap terms over $\{\circ, ||\}$ may also contain partial specifications. 
       A heap term is build up stepwise during data analysis of APs.

 \item \texttt{HeapTerm.simplify} refers to \texttt{Built\-in\-Li\-bra\-ry\-.simp\-li\-fy_1()}.
       Heap term simplification using overloaded built-in Prolog predicates\\
       \texttt{simplify\_2}.

 \item A heap term represented by heap term may also be described partially using Prolog's anonymous placeholding operand ``\texttt{\_}''.
       The normalisation of given conjuncts is subject to lexicographical sorting of all left-hand sides on simple heaps.

 \item \texttt{DotGenerator} generates \index{DOT} DOT vector graphics as a single string expression that optionally may be serialised to an output file.
 
 \item Soundness proof is considered, whether (i) non-repetitiveness is obeyed, (ii) a heap is connected.\\
\end{itemize}

%%%%%%%%%%%%%%%%%%%
\texttt{Package prolog} (see fig.\ref{PackagesProlog}).

The library \texttt{BuiltinLibrary} contains all Prolog predicates, which are granted access for the logical analysis by subgoals.
If an AP is entirely accepted by the recogniser and terminates afterwards, the kind $a\circ a$'s condition is checked.

\begin{figure}[h]
\begin{center}
\begin{turn}{90}
\parbox[t]{14.5cm}{
\begin{minipage}{\linewidth}
\begin{tikzpicture}[scale=0.6, every node/.style={transform shape}]
 \begin{class}[text width = 5cm]{Library}{3,4}
 \end{class}
 \begin{class}[text width = 11cm]{BuiltinLibrary}{0,0}
   \implement{Library}
   \operation{+ pointsto\_2(l: Term, r: Term): boolean}
   \operation{+ serialize\_2(term: Term, fileOut: String): boolean}
   \operation{+ serialize\_1(term: Term): boolean}
   \operation{+ disj\_0(): boolean}
   \operation{+ unify\_v\_2(): boolean}
   \operation{+ unify\_nc\_2(): boolean}
   \operation{+ simplify\_1(t: Term): boolean}
   \operation{+ verifyH\_3(bef: HeapState, c: Term, aft: HeapState): boolean}
   \operation{+ verifyS\_3(bef: StackState, c: Term, aft: StackState): boolean}
 \end{class}
 \begin{class}[text width = 4cm]{HeapState}{-3,-8}
 \end{class}
 \begin{class}[text width = 4cm]{HeapTerm}{3,-8}
 \end{class}
 \begin{class}[text width = 6cm]{DynamicInvoker}{-3,2}
 \end{class}
 \begin{class}[text width = 6cm]{MemoryComparator}{-15,0}
 \end{class}
 \begin{class}[text width = 4cm]{StackState}{-8,-8}
 \end{class}
 \begin{class}[text width = 6cm]{PU Iterator}{-15,-7}
   \attribute{- PU: String}
   \operation{+ nextStatement(): Term}
 \end{class}
 \begin{class}[text width = 5cm]{APMemoizer}{-15,4}
 \end{class}
 \begin{class}[text width = 7cm]{AssertionNormaliser}{-8,4}
 \end{class}
 %
 % associations:
 \unidirectionalAssociation{BuiltinLibrary}{}{}{HeapState}
 \unidirectionalAssociation{DynamicInvoker}{}{}{BuiltinLibrary}
 \unidirectionalAssociation{MemoryComparator}{}{}{DynamicInvoker}
 \aggregation{HeapState}{terms}{0..*}{HeapTerm}
 \composition{MemoryComparator}{}{2}{StackState}
 \composition{MemoryComparator}{}{2}{HeapState}
 \aggregation{PU Iterator}{}{}{MemoryComparator}
 \aggregation{MemoryComparator}{}{}{APMemoizer}
 \aggregation{MemoryComparator}{}{}{AssertionNormaliser}
\end{tikzpicture}
\end{minipage}}
\end{turn}
\end{center}
 \caption{Package prolog}
 \label{PackagesProlog}
\end{figure}
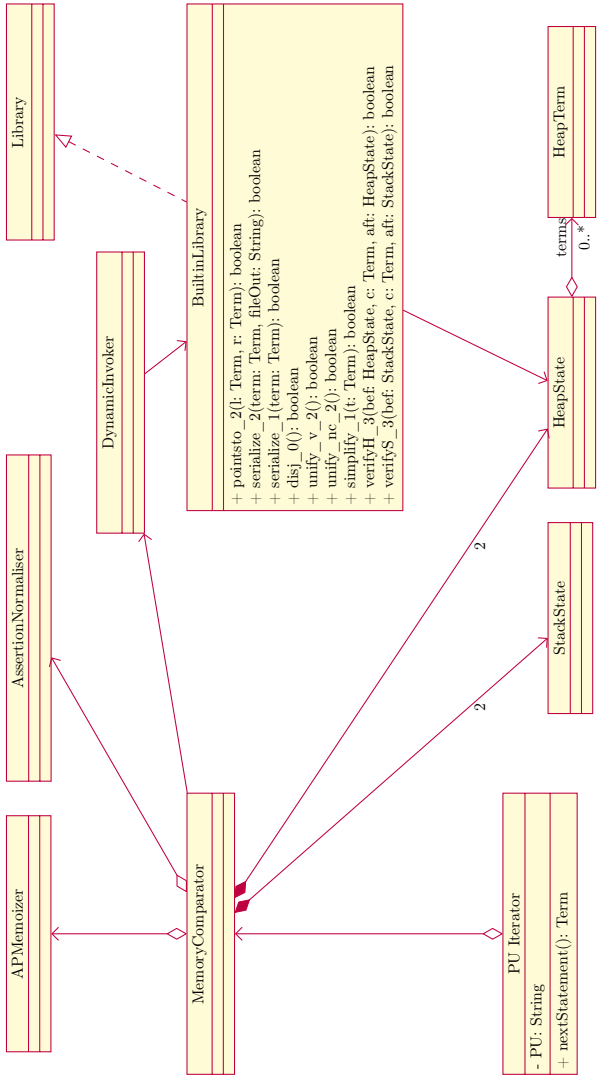

Remarks:

\begin{itemize}
 \item  "\texttt{BuiltinLibrary:disj\_0}" accomplishes one step in heap verification.
        Unique markers may be introduced in case of inaccessibility.

 \item "\texttt{BuiltinLibrary:unify}" is a replacement for the classic term unification (\texttt{=/2}) using as wells as by not using checks on repetition.

 \item Class "\texttt{HeapTerm}" refers to "\texttt{Built\-in\-Li\-bra\-ry::\-simp\-li\-fy\_1}".

 \item "\texttt{BuiltinLibrary::serialize\_1}" is almost equivalent to "\texttt{serialize\_2}", except the content is issued to the console.

 \item "\texttt{StackState}" is needed for checking visibility scopes.
       For example, it checks whether a symbol is in scope after all other semantic checks are accomplished.
       Its signature is of type $\sigma:: String \rightarrow Type \rightarrow Value$.
\end{itemize}

%%%%%%%%%%%%%%%%%%%

\subsection{Conclusions}

The approach introduced presents a new technique for \index{verification automation} automated verification of $\mapsto$-assertions in APs by syntactic analysis.
It avoids manual predicate \index{fold} "\textit{folds}" as well as \index{unfold} "\textit{unfolds}".
If some predicate partition may be represented as the proper rule set of a \index{syntax analysis} syntax analyser, then a solution will be found, if any.
Otherwise, the failing simple heap will be provided automatically by the syntactic analyser as a generic counter-example.
Remarkably, the latter remark is identical with non-AP-related heap comparisons.
After Knuth-Wagner, an attributed grammar is applied \cite{grune90} as a connecting model between Prolog rules on the one side and a sound rule set on the other side.
The \index{grammar!attributed} attributed grammar fully supports inherited and synthesised attributes, translated into head parameters of some Horn-rule.
Apart from the comfortable representation, Prolog allows top-down processing and its symbols are identical with the specification ordering and simple logical reasoning.

Without any modification of Prolog's WAM \index{WAM} core reference implementation (cf.\cite{warren83}, \cite{warren99}, \cite{denti05}), heap verification may be achieved.
Extensions, for example, demonstrated transducers on terms and heap modelling do not contradict fundamental properties, e.g. of \index{SL} SL (see sec.\ref{chapter:expression} and \ref{chapter:stricter}).
That way, properties do not change if the \index{locality} locality is obeyed  in class-object instances.
Non-pointer object fields must be specified in a predicate only once.
Remarks regarding a potential slow interpretation of queries in a \index{abstract machine} \index{WAM} WAM-based interpreting environment may not be taken too seriously.
In general, query interpretation may be replaced by compilation and applied to code optimisations that considerably may improve runtime.
It was proven wrong more than once in practice, including running Java on tiny embedded devices in real-time.
Even so, any attempt in making verification better has the highest priority.

Nevertheless, in the first place, there is still the priority on complexity and expressibility.
Thus, intentionally not the fastest verification is of interest.

In contrast to previous approaches, a new complex approach is proposed.
For example, \index{symbol} symbols are not just called "\textit{symbols}" anymore, but they genuinely are now.
Symbols may be used logically in arbitrary logical terms and predicates.
Assertion language \index{language!assertion} and verification language \index{language!verification} now are identical --- a Prolog-dialect.
Additional transformations are now no more needed, and limitations are lifted, known from previous attempts (cf.\cite{bertot04} and sec.\ref{chapter:intro}).

An open question remains how to integrate \index{rule!syntax error} error production rules to the \index{counter-example} \index{counter-example} counter-example generation effectively.
When a syntactic error is detected, then the default behaviour is to terminate, even if the complete error analysis is provided.
Sometimes, however, it might be worth knowing about the second and further errors to better analyse the \textit{real} problem behind it -- however, this will not incorporate with several compiler generators, especially not with LR-styled ones.
Thus, a \index{automaton!finite} finite automaton based approach is suggested for minimising the \index{edit distance} edit distance between calculated and expected heaps, in analogy to \index{Levenshtein's algorithm} \textit{Levenshtein's algorithm} \cite{levenshtein65}, \cite{wagner74}, \cite{andoni12}.
Its complexity is expected to be bound by $\Theta(n^3)$ for the heap model suggested in this work.
The expected edit trace would include positive and negative (sub-)heaps (cf. sec.\ref{chapter:stricter}) -- this is expected to be fully compatible with Levenshtein's insertion and deletion operations.

Regarding term IR \cite{america89}, \cite{zhang89}, \cite{rutten96} shall be mentioned as proof that Levenshtein's algorithm is tractable to terms.
"\textit{Panic mode}" as described in \cite{grune90} and in \cite{parr12} attempts to continue parsing till some well-known safe marker is put in artificially to a formal grammar in case of an emerging syntax error in the input word.
The effect on heaps shall be researched for a stable error restore and its usability towards proofs.

Furthermore, it is recommended to consider \index{partial specification} partial specification, as was introduced earlier by constant functions "\underline{$emp$}", "\underline{$true$}" or "\underline{$false$}", for the sake of proof simplification.
Notably, the research shall focus on whether partial specifications may locate inaccessible bridges in heap graphs, which connect independent heaps, for the sake of a better specification rewriting.

\section{Conclusions}

\underline{Proof automation:}
\begin{enumerate}
  \item Analysis of gaps between specification languages and verification languages for dynamic memory.
  Analysis of narrowing down both languages.
  \item Expressibility analysis and measurement of term transformations in the logical paradigm using a Prolog-dialect.
  Finding confirmation for logical assertions based on symbols, variables, terms, and relations is more convenient and closer to heap verification in Prolog than other existing approaches.
  \item Prolog terms as IR fit best for the whole static analysis assembly line.
  \begin{itemize}
    \item Foundation and use of a logical language for specification and verification heap.
    The main idea --- logical programming for proving theorems about the heap.
    \item The proposed platform, designed as open, extensible and variable.
    The platform allows the integration of additional external libraries and optionally logical solvers directly in Prolog.
    \item Research done showed numerous advantages of Prolog over competing techniques, including direct SMT-solver injection in Prolog.
   \end{itemize}
   \item Proposition on heap verification automation as a syntactic analysis of abstract predicates based on heap graph edges' stepwise processing.
\end{enumerate}

\underline{Expressibility of specification and verification languages:}
\begin{enumerate}
  \item Now, w.l.o.g. abstract predicates and subgoals to be proven may contain symbols and variables in Horn-rules, which extends its previous limitations.
  Predicates may now have any recursive definitions, where the chosen syntactic analysis method may bound recursion requiring a rule rewriting (e.g. a left-factorisation).
  \item A rise of expressibility is imminent due to a strengthening of heap operations.
  A strict relationship was established between heap expressions and their graph representation.
  Now spatial heap operators build up unique term assertions.
  Algebraic properties were proven, founding a new heap calculus, which obeys monoid and group properties.
  The proposition was made to extend the de-facto standard \textit{UML/OCL} by pointers enriched by a Prolog notation.
  \item Expressibility and completeness were lifted due to partial heap specifications of variables, and within bound of object-instances was observed.
  Comparison of heaps is now enabled, s.t. incoming heaps of one set may be compared to another, so "\textit{gaps}" and "\textit{intersections}" may effectively be spotted.
  Now, rules may be written potentially shorter and more readable and cover potentially more cases simultaneously.
  \item Theoretical possibility is now given, based on a stricter representation, to localise and check a given input program for minimality w.r.t. full heap graph coverage ("\textit{inhabitant problem}") while excluding abstract predicate definitions in general.
\end{enumerate}

\underline{Conclusions:}
\begin{enumerate}
  \item A unified specification and verification language.
  Narrowing down specification languages to verification languages, and vice versa, to a logical programming language, namely a Prolog-dialect, is considered an extreme solution to a target function with numerous unknown parameters, which is considered ambiguous. 
  The unified solution fights ambiguity and resolves several issues analysed w.r.t. expressibility and completeness.
  The fewer languages there are, the fewer exceptions and the less, in general, need to be considered. 
  That is the case when specification and verification are essentially the same w.r.t. syntax and semantics.
  The verification process is represented as a comparison of some expected and a calculated heap.

  During the research of existing heap verifiers and related projects, both tools, "\textit{shrinker}" and "\textit{builder}", were developed as helpful universal analysis tools.
  Research showed that despite claims made, other verifiers still are restricted to a few academic examples only, beyond which its practical meaning sharply drops.
  Tremendous efforts could be seen to be put by scientists worldwide to be put into infrastructural re-definitions of well-known verification concepts, again and again, hindering further break-through massively.

  \item "No" to new and duplicate (re-)definitions.
  Prolog use allows avoiding the (re-)introduction of conventions, specification, programming, and verification language features and software modules based on those.
  There is no need to introduce more and more model heap definitions.
  There is now one adequate heap definition.
  Important properties and criteria were derived, which are now covered by the platform proposed.
  Now heaps and parts of it may be denoted by (sub-)terms and symbols.
  Abstract predicates are now represented by Prolog rules and fully fit in the context of Prolog.
  The term IR is used as the input programming language for a C-dialect or any other imperative programming language, including no programming language.
  Instead, in this case, direct Prolog terms may be injected and interpreted directly by the verification core.
  Specification and verification are now defined in one language instead of two, namely a Prolog-dialect.
  Optionally the input listing may too be provided in the Prolog-dialect.
  \item Increased expressibility by direct and unrestricted use of Prolog symbolic variables.
  Heap predicates became genuinely logical and abstract (but not "\textit{Abstract}"), namely Prolog predicates.
  In predicates, symbols are now allowed at random places, and no more restrictions exist here.
  Logical reasoning is based on symbols.
  Due to unification and a WAM-based execution stack, logical reasoning has by its nature the same deductive structure for no extra costs.
  The discussion showed that the heap graph model presented completely covers all desired heaps.
  Additional emphasis was put on optional higher-order predicates.
  Class-objects are heaps too.
  Class methods as functions do not insist on pre-and postconditions.
  The frame rule is applied to procedures and among the object life-cycle.
  Symbolic parameters may pass objects to abstract predicates.
  Thus, modularity raises, and due to partial constant functions, objects may still be fully specified.

  \item The relational model proposed showed that Prolog implementations are complete and adequate, particularly towards abstract predicates.
  In contrast to competing, the relational model was chosen as the most promising for checking heaps' comfortability.
  References and comparisons to literature references were discussed.
Although quantitative analysis is complicated and occasionally not representative, still representative examples were chosen and analysed using metrics.
  So, the field expertise showed Prolog terms are nearly always better on average by 30\% than equivalent functional notations.
  "\textit{Better}" means especially compacter and more flexible.
  Typical examples were cherry-picked for the sake of demonstration that cannot be represented (at all or with immense difficulties only) in functionals or imperative paradigms.
  \item In Prolog, a new verification platform based on a conveyor/assembly line architecture.
  The platform is minimalistic and centric to terms derived from the incoming program and annotations and verification rules, all in Prolog.
  The platform is from a user's perspective altogether in Prolog from the specification and verification view.
  It is explicitly designed for heap reasoning and allows user-specified definitions and rules to be added dynamically.
  The platform follows the core principles: (1) automation, (2) openness, (3) extensibility, (4) plausibility.
  It can only be derived what follows by the given Prolog rules (or its multi-paradigmal built-in predicates) -- this is a quick but fair and intentional design decision.
  Verification rules, lemmas and inductively-defined abstract predicates are all noted in Prolog.
  A proof tree can be traced (e.g. by DOT-visualisation) and by logging messages.
  Any problematic case can be debugged individually by jump-starting with arbitrary memory state (stack and heap) and some input program statement or Prolog-equivalent term statement -- so debugging is possible at any phase.

  \item Simplification of specification and verification due to strengthening spatial operation.
  A stricter heap definition denoting single simple heaps was found instead of the previous ambiguous heap definition.
  Strengthening means obeying the atomism of heaps.
  More precise, more straightforward heap results, and as such, specification and verification rules become simpler.
  A mandatory comparison of a heap in the context of all other heaps is no more needed.
  Now heaps, including one-to-one as well as one-to-many, can be done more effectively.
  Due to established group properties, now heaps can be calculated.
  Heaps may be differentiated for a rule set, and comparison may indicate missing parts, so a rule set may effectively be checked for completeness.
  Due to partial specification, now the total amount of heap rules that need to be specified can be reduced and simplified.
  Exceptions in heap expressions remain excluded, but the rule set completeness can effectively be checked using partial (incomplete) specifiers.
  In practice, a strengthening means that heap subexpressions do no more need to be extensively analysed.

  As \textit{UML/OCL} does not support pointers, the models introduced are compatible, so it is proposed to enrich it with pointers, and the logical notation introduced.

  \item Possibility defining new heap theories.
  Shown algebraic properties allow the definition of equalities and inequalities over heaps.
  These formal theories may immediately be expressed either in Prolog or in SMT-solvers that can be integrated using multi-paradigmal programming as Prolog subgoals.
  Now rules and (heap) theories can be checked easier and without superfluous checks.

  \item Abstract predicates form an attributed grammar.
In analogy to schema validation and templates processing, semi-structured data bridge verification and specification processes by similar data, namely heap data terms.
  Here abstract predicates are found to be a heap abstraction technique containing "\textit{gaps}" as templates do.
  The gaps are filled during verification with concrete data.
  This meta-concept finally leads to syntax analysis.
  As a result, logical reasoning on abstract predicates may be interpreted as syntax analysis.
  Comparing an actual heap with the expected may be considered as a word problem for formal languages.
  This universal method has recognition limitations.
  However, the method is convenient, regardless of the well-founded theoretic apparatus underneath.

  \item Universal approach in automating assertion checks.
  Predicates and proof are based on a Prolog-dialect and may multi-paradigmal include arbitrary libraries.
  So, proving becomes programming essentially.
The need for redundant process definitions of translation and rule description becomes obsolete because it is implicitly already expressed by Prolog terms and rules, even when the language introduced is minimalistic. 
  There is no more manual fold and unfold needed for recursive definitions applied to abstract predicates.
If a heap is not definable according to given abstract predicates for an individual parser, they may either be rewritten or not definable in principle, e.g. if they are not context-free.
  Syntactic analysis is decidable for heaps, and termination is guaranteed, e.g. for LL-recognisers in polynomial (approximately linear) time.
  An analyser is generated by need and only if rules change.
  Its process is initiated once on-the-fly.
  Depending on the reachability of starting non-terminals, the recognition of subgrammars avoids new recogniser rebuilds.
  In practice, all non-terminals may be available after a recogniser is built once.
  Syntactic analysers do not require further research for this work's objective because they have been researched to a near full extent over the last decades.

  In case of an error, the automated reasoning, either with or without parsing, can immediately derive counter-examples from a simple term comparison of having an expected term founded on simple term unification and without any further efforts required.
\end{enumerate}

\underline{\textit{\textbf{Future Research:}}}
Future works shall include the following theoretic and practical work (not prioritised):

\begin{enumerate}
 \item \textit{Integration of further formal theory solvers}.
 After and during the definition of heap theories, \textit{SMT}-solvers shall be integrated for improved and even more simplified proofs.
 Solvers may naturally be written in Prolog.
 Prolog's constraint programming abilities is one primary reason for its widespread popularity, at least in academia, so this would be considered a bargain.
 \item \textit{I Integration) into existing environments}.
 Besides the mentioned limitations and for the sake of broader popularity, term-IRs shall be turned into \index{GIMPLE} GIMPLE and \index{LLVM} LLVM bitcode, and vice versa for full industrial compliance.
 \item \textit{II Integration) into existing environments}.
 A detailed code analysis, namely detecting "\textit{growing}" points (memory regions bloating), may lead to an improved allocation strategy w.r.t. fragmentation, memory segmentation and layout.
 Here the not initially compatible abstraction interpretation technique shall be further tracked towards and incorporated with the stricter memory model introduced in this work.
 Also, memory addresses with static offsets shall be further integrated into the existing heap model for a broader acceptance and easier integration towards \index{GIMPLE} GIMPLE.
 \item \textit{Heap isomorphism}.
 The check for graph isomorphism relates to how a heap is not fully specified, so containing gaps or symbols is still effectively compared?
 \item \textit{Finding (even minimal) deviation from a specification}.
 During syntax analysis, in case of an error, a minimum global search might be better instead of an exist on first error occurrence, e.g. by using Levenshtein's edit distance.
 Here further research is needed.
 \item \textit{Further Pointer dependency analysis}.
 In contrast to alias analyses, it is highly suggested to research some extended SSA-forms incorporating pointers.
 Nobody has even tried that before, presumably since it is feared it is too complicated, although both SSA and pointers are about analysing data dependencies.
 Initially, a deviation is expected, but if synergies could be found, this would not only ease pointer analysis but would for free provide a powerful optimisation technique with in-depth knowledge of how the heap graph is supposed to be used most efficiently.
 \item \textit{Verifying existing code repositories}.
 The approach proposed here remains academic, after all.
 Verification of more extensive source repositories always was and still is a separate research branch.
 So, eventually, such a field study would be an outstanding practical achievement.
 However, that would require first to cover \index{GIMPLE} GIMPLE entirely for a realistic estimate.
 \item \textit{Further research on transformation languages by practical means}.
 Discussed transformation languages may provide a compact and rule-based notation, usually including the calculation states before and after.
 However, a significant gap is expected between the introduced memory model, the expressibility model and content in dynamic memory, and the input language.
 This gap will need to be clarified first -- and after all, could be an alternative to this work's main contribution.
 \item \textit{Abstract reasoning}.
 Rare Prolog-dialects do support abductive reasoning in addition to Prolog main's reasoning: deduction.
 This issue becomes prominent, especially since O'Hearn's attempt in introducing abduction to separation logic.
In general, abstraction techniques are considered an additional mechanism in matching rules that otherwise would not be selected.
 The ordering of logical conclusions made may be further simplified if the relationship between the incoming and outgoing terms are invertible in terms of predicate inversion, as discussed in section \ref{chapter:logical}.
\end{enumerate}

\underline{\textit{\textbf{Recommendations for practitioners:}}}

\begin{itemize}
 \item A practical heap specification was simplified.
 A given program is annotated by verification conditions and user-defined (abstract) predicates, namely the specification.
 \item Abstract predicates need to be well-defined and valid Horn-rules.
 If those build up in Prolog a proper context-free grammar, then subgoals to these can be recognised.
 Alternatively, rule rewriting may resolve recognition problems.
 \item The use of stricter spatial operators allows extending \textit{UML\-/OCL} specifications by pointers and a well-founded calculus.
 \item The introduction of new programming languages is not mandatory, although it is open potentially to any language obeying requirements.
 A programming listing may be missing if it is injected by a Prolog term instead.
 The extensible and variable IR gains maximum flexibility.
 Verification architecture has the same properties.
 \item The heap comparison of a given expression with the expression from rules allows now to find differences in heaps.
 In more details, it may check all heaps for a predicate partition in order to find missing heaps in rules.
 \item The overloading of relations simplifies logical reasoning, e.g., abduction and different techniques shall be considered in the future.
 \item The use of SMT-solvers for given heap theories, either in Prolog or by a multi-paradigmal implementation, may increase transformation speed in normalised heap forms.
 \item The integration of memoisation (in Prolog called \textit{tabling}) to abstract predicates may lead to an improved performance throughput.
 \item When introducing new or modifying existing phases in heap verification, it is highly recommended to stick to Prolog term IR.
\end{itemize}

\underline{\textit{\textbf{Formal Concept $(G,M,I)$ on Achievements:}}}\\

Set $G$ --- Objects are the achievements:

\begin{enumerate}
  \item[(1)] Unified heap specification and verification language
  \item[(2)] Architecture principles for heap verification
  \item[(3)] Increase of expressibility
  \item[(4)] Relational model for heap description
  \item[(5)] Prototypical reference implementation
  \item[(6)] Interpretation and introduction of heap theories
  \item[(7)] A fresh definition of a singleton heap
  \item[(8)] Simplification of heap specifications based on $\circ$, $||$
  \item[(9)] Pointer-based object calculus proposed
  \item[(10)] APs represent attributed grammars\\
\end{enumerate}

Set $M$ --- Features are the central theses:

\begin{enumerate}
 \item[N1] Prolog-dialect as heap specifier and verifier
 \item[N2] Strengthening approach of ambiguous heaps
 \item[N3] AP verification as attributed grammar parsing
 \item[N4] Unit of utils for heap analyses\\
\end{enumerate}

$I \subseteq G\times M$ --- Incidence relation:

\begin{center}
\begin{tabular}{c|cccc}
  G/M & N1 & N2 & N3 & N4\\
  \hline
  (1) &x&&&x\\
  (2) &x&&&x\\
  (3) &&x&x&\\
  (4) &x&&x&\\
  (5) &&&&x\\
  \hline
  (6) &&x&x&x\\
  (7) &&x&x&\\
  (8) &&x&x&x\\
  (9) &x&x&&\\
  (10) &&&x&
\end{tabular}
\end{center}

The corresponding Hasse diagram for $I$ can effectively be constructed and visualises the dependencies:\\

%\fbox
{\parbox[t]{7.5cm}{
\resizebox{0.5\textwidth}{!}{
\begin{minipage}{8cm}
  \xymatrix@C=1.5em@R=1.5em{
    && \txt{\{N1,N2,N3,N4\}} &&&\\
    && \txt{\{N2,N3,N4\}\\\underline{(6),(8)}} \ar@{-}[u] &&&\\
    \txt{\{N1,N2\}\\\underline{(9)}} \ar@{-}[uurr] & \txt{\{N1,N3\}\\\underline{(4)}} \ar@{-}[ruu] & \txt{\{N2,N3\}\\\underline{(3),(7)}} \ar@{-}[u]  & \txt{\{N1,N4\}\\\underline{(1),(2)}} \ar@{-}[uul]\\
    \txt{\{N1\}} \ar@{-}[u] \ar@{-}[ur] \ar@{-}[urrr] & \txt{\{N2\}} \ar@{-}[ul] \ar@{-}[ur] & \txt{\{N3\}\\\underline{(10)}} \ar@{-}[ul]  \ar@{-}[u] & \txt{\{N4\}\\\underline{(5)}} \ar@{-}[u] \ar@{-}[uul] &&\\
    && \txt{$\bot$} \ar@{-}[ull] \ar@{-}[ul] \ar@{-}[u] \ar@{-}[ur]  &&&
  }
\end{minipage}
}}}

\section{Original}

\newenvironment{rudefinition}[1]{\begin{small}\underline{Определение} \textbf{\textit{(#1)}}. }{\end{small}}
\newenvironment{rucorollary}[1]{\begin{small}\underline{Заключение} \textbf{\textit{(#1)}}. }{\end{small}}
\newenvironment{rutheorem}[1]{\begin{small}\underline{Теорема} \textbf{\textit{(#1)}}. }{\end{small}}
\newenvironment{ruexample}[1]{\begin{small}\underline{Пример} \textbf{\textit{(#1)}}. }{\end{small}}

In the original translation, all figures, algorithms, tables and additional materials are stripped off.
Here, the original text refers to the English translated figures.
In most cases, this is believed to be accurate.

%%%%%%%%%%%%%%%%%%%%%%%%%%%%%%%%%%%%%%%%%%%%%%%%%%%%%%%%%%%%%%%%%%%%%%%%%%%%%%%%%%%%%%%%%%%%%%%%%%%%%%%%%%%%%%%%%%%%%%%%%%%%%%%%%%%%%%%%%%%%%%%%%%% 1 Intro
%%%%%%%%%%%%%%%%%%%%%%%%%%%%%%%%%%%%%%%%%%%%%%%%%%%%%%%%%%%%%%%%%%%%%%%%%%%%%%%%%%%%%%%%%%%%%%%%%%%%%%%%%%%%%%%%%%%%%%%%%%%%%%%%%%%%%%%%%%%%%%%%%%%

\section*{Введение}

Для введения в предметную область диссертации «\textit{Логический язык про\-грам\-ми\-ро\-ва\-ния как инструмент спецификации и верификации динамической памяти}» при\-во\-ди\-тся этот раздел. Первый раздел посвящается вычислению Хора, основным определениям и постановлению решения вычислением, альтернативным подходам и свойствам систем  вычислений Хора. 

Далее для сравнения различных динамических и статических подходов, можно использовать  «\textit{качественную лестницу}» из рисунка \ref{fig:QALadder} в любой момент.
Эта лестница является моим личным предложением для классификации качества, которое должно соблюдать любое программное обеспечение. Чем больше программа соблюдает  критерии качества, тем искреннее можно считать данную программу качественной. То есть, если программа некорректно вычисляет результат, то на самом деле можно считать неважным, насколько данная функция выполняет критерий входного домена, и тем более безразлично, насколько быстро это сделано. Самое главное - это кор\-рек\-тное вычисление программы. Если это соблюдается, то только тогда можно считать уровень  «\textit{базисно надёжным}».

С растущим спросом в надежности растет спрос на универсальность данной функции, т.е. тогда имеет смысл распространить качество на любые входные функции.
Если базисные операции работают правильно, то тогда можно требовать, что для всех входных данных функций всё работает правильно --- это второе требование, которое мы обозначим  «\textit{полностью надёжным}», более жёстким, чем первое. Если первые два критерия соблюдаются, то с вычислительной точки зрения, данная функция корректна и полна. Тогда имеет смысл, далее вводить оптимизации, которые мы обозначим  «\textit{оптимально надёжными}» и которые не меняют вычислительные свойства   корректности и полноты. Надёжными оп\-ти\-ми\-за\-ци\-я\-ми могут послужить, например, ускорение часто востребованных программных блоков, визуализация эффектов, работа с файлами и т.д. Визуальные эффекты и работа с файлами не будут рассматриваться.

Каждый из установленных критериев можно чётко обозначить и можно пред\-ло\-жить некую метрику для измерения выполнимости.
Не имеет смысла рассуждать о качестве программы, когда программа вычисляет корректно и быстро, но не полностью, потому, что универсальность применения подрывает качество су\-щес\-твен\-но, когда не имеется определение для входного вектора, даже, если все остальные случаи высчитываются очень быстро. Тогда можно, либо ограничить диапазон видимости входного вектора для полного покрытия функции, либо сузить покрытие ради не\-пол\-ной надёжной функции. Когда будут рассматриваться различные подходы, ка\-чес\-твен\-ная лестница будет использована как ориентир для принятия или отклонения пред\-ло\-же\-ний в дальнейшем.

В данной работе рассматриваются императивные и искусственно-гипотетические языки программирования. Для практической применимости необходимо рас\-сма\-три\-вать об\-ъе\-кт\-но-ор\-и\-ен\-ти\-ро\-ван\-ные модели, поэтому в следующем разделе вводятся основные понятия объектно-ориентированных вычислений, в частности,  \textit{Теория Об\-ъек\-тов} (ТО). Далее рассматриваются теоретические и практические проблемы ра\-бо\-ты с динамической памятью, а затем вводятся модели представления  ди\-на\-ми\-че\-ской памяти, как, например,  \textit{анализ образов} (АО),  \textit{вычисление регионов} (ВР), а также другие модели представляющие косвенно-динамическую память. \textit{Логика рас\-пред\-е\-лён\-ной па\-мя\-ти} (ЛРП) является основной теоретической моделью пред\-ста\-вле\-ния динамической памяти этой работы. Затем даётся обзор по автоматизации до\-ка\-за\-тельств и обсуждаются её факторы противостояния. С целью лучшего понимания проводимых далее доказательств с помощью модели Хора и улучшения схо\-ди\-мос\-ти деревьев доказательств, вводятся представления «\textit{абстракции}». С кратким обзором по тематике можно также ознакомиться в \cite{haberland16-5}.
Затем, с целью оз\-на\-ком\-лен\-ия, рассматриваются области применения и связанные с главным направлением работы в этой области, в частности ---  \textit{анализ псевдонимов} (см. раздел \ref{sect:AliasAnalysis}),  \textit{сбор мусора} и верификация кода с  \textit{интроспекцией}. В конце этой главы представляются существующие программные системы и среды.

%%%%%%%%%%%%%%%%%%%%%%%%%%%%%%%%%%%%%%%%%%%%%%%%%%%%%%%%%%%%%%%%%%%%%%%%%%%%%%%%%%%%%%%%%%%%%%%%%%%%%%%%%%%%%%%%%%%%%%

\subsection*{Вычисление Хора}

\textit{Вычисление Хора} (назван в честь информатика Чарлься Антони Ричард Хора, в литературе также встречается фамилия «Хоаре» вместо «Хор», однако известно, что сам автор пишет свою фамилию как «Хор» на русском, а также это способствует избегать неверные произношения других авторов) - это  формальный метод \textit{верификации}, который позволяет про\-ве\-рить верность данной программы, данной некоторой \textit{спецификации}, т.е. согласно описанию, характеризующему свойство поведения программы. Спецификации опи\-сы\-ва\-ют  \textit{со\-сто\-я\-ни\-е вычисления} с помощью математических формул и теорем, а вывод происходит согласно аксиомам и правилам рассматриваемой области дискуссии.

Цель верификации программы, c помощью вычисления Хора, - это проверка со\-блю\-де\-ния свойств программы, что является повышением качества и надёжности программы. Если программа соблюдает свойства утверждениями, то это позволяет нам характеризовать программу и сравнивать её с другими программами, где одно или иное свойство, возможно, не соблюдается.

Любое \textit{правило Хора} необходимо рассматривать, как  \textit{логическое суждение} $A \Rightarrow B$, которое читается: «\textit{если} суждение  \textit{антецедент} $A$ выполняется, \textit{тогда} выводимое суждение  \textit{консеквент} $B$ следует». Рассмотрим рисунок \ref{fig:HoareCalc}:

Правило Хора $(P1)$ является  \textit{аксиомой}, пусть $A$ будет пустым антецедентом. Главная идея аксиомной системы заключается в проверке верности следствия согласно при\-ме\-не\-нию конечной последовательности определённых правил. Отныне, мы не различаем между  правилами и аксиомами, т.к. первое является обобщением второго. Правило $(P2)$ является аксиомой, если \textit{антецедент} $B \vee \neg B$ выполняется всегда аналогично \textit{логике утверждений} и сопоставляется  «\textit{истинно}». Однако, изначально не имеются искусственные ограничения в вычислении Хора, например, касательно суждений. Утверждения могут быть сопоставлены предикатами, а  \textit{логический вывод} становится \textit{суждением} над предикатами.
Различные методы могут быть применены для вывода с \textit{предикатами первого порядка}, например: \textit{метод естественного вывода} \cite{troelstra00},   \textit{метод резолюции} или  \textit{метод семантических таблиц} («\textit{Tableaux method}»). Перечисленные методы являются  \textit{дедуктивными}. В отличие от дедукции, ещё имеются  \textit{абдукция} и  \textit{индукция}.
Индуктивный метод вывода предлагает ввод ещё не существующие или тяжело выводимые, явным образом, утверждения в набор рассматриваемых правил. Индуктивное правило введённое «кажется» искусственно, нельзя отклонить из-за от\-сут\-ствия  противоречивого примера. Классическим примером индукции служит  \textit{теория опровергаемости по Попперу}: данное утверждение «\textit{все лебеди белые}» на практике, из-за ограничений не может быть проверено целиком за конечный про\-ме\-жу\-ток времени. Поэтому, предлагается исходить из верного утверждения до тех пор, пока не будет замечен противоречивый экземпляр, например, не будет найден чёрный лебедь. Индукция нам предлагает ради преодоления необъяснимых и неотрицаемых феноменов в рамках рассматриваемого ракурса ввод нового правила, согласно которому феномен может быть обоснован. Свойства индуктивно оп\-ре\-де\-лен\-ных перечисляемых, возможно бесконечных структур, проверяются с помощью конечной формулы.
Когда наблюдается противоречивый индукцией экземпляр, мы вынуждены скорректировать наши утверждения о мировоззрении. В отличие от этого, абдукция ищет необходимые и допустимые предпосылки, чтобы следствие было выводимо из набора правил.

Вычисление Хора можно охарактеризовать, как  дедуктивное суждение, применив к императивному  программному оператору (как это было предложено изначально Хором), а каждое выводимое суждение описывается состоянием вычисления, до и после выполнения оператора и программным оператором (см. раздел \ref{Intro:HoareTriple}). Ради улучшения  \textit{сходимости доказательства}, часто можно добавлять индукцию и аб\-дук\-цию в качестве метода вывода. Вывод дедукции интуиционистен \cite{brouwer10} потому, что утверждение только тогда верно, когда даны условия и соответствующий факт предусловия. Когда мир расследуемых выводов производится строго согласно правилам и пересекаемые в антецеденте, пра\-ви\-ла исключены, то  \textit{космос выводимых следствий} является \textit{замкнутым}, к примеру, правила $(P3)$ и $(P4)$ из рисунка \ref{fig:HoareCalc}, интерпретации возможных логических утверждений также замкнуты.

\subsubsection*{Тройка Хора}

Для  \textit{формальной верификации} с помощью спецификации, Хор в \cite{hoare69} предлагает для \textit{императивных языков программирования} определить тройку.

 \textit{Декларативные языки прог\-рам\-ми\-ро\-ва\-ния}, например, \textit{функциональные}, отличаются от   императивных тем, что состояние вычисления меняется в зависимости от содержимого переменных.  Порядок выполнения программных операторов не строгий.  Декларативную парадигму программирования можно отделить от  им\-пе\-ра\-тив\-ного с помощью модели организации памяти касательно символов и переменных. Далее в этой работе рассматриваются исключительно императивные языки прог\-рам\-ми\-ро\-ва\-ния, в качестве входного языка программирования.
Изначально Хор пред\-ла\-гал, в качестве языка программирования простой императивный язык наподобие \textit{диалекта Паскаль} и нотацию $P\{C\}Q$. Скобочная нотация над оператором оказалась не популярной, поэтому скобки стали ставить вокруг $P$ и $Q$.
$P$ и $Q$ оба описывают  \textit{состояния вычисления} в качестве утверждений до и после $C$. $C$ может содержать любое количество императивных операторов. Из сказанного следует, что $C$ меняет пошагово состояние вычисления, если $C$ не делится далее, а следовательно, состояние памяти меняется пошагово. Под  памятью мы подразумеваем  \textit{список процессорных регистров},  \textit{стек} и  \textit{динамическую память} (см. рисунок \ref{fig:ProcessSectionLoader}).  Тройка Хора интерпретируется так: если имеется состояние вычисления $P$ и  программный оператор $C$ выполнен, то вычисление совершено и находится в состоянии $Q$. Другими словами, $P$ и $Q$ описывают состояния памяти до и после $C$. Если тройка Хора соблюдается, то переход состояний памяти верный, в противном случае, имеются следующие причины несоблюдения:

\begin{enumerate}
 \item Если $C$ не завершает работу, то данные правила Хора назовём «\textit{не\-пол\-ной определённой}» и состояние $Q$ не достижимо. Такое поведение $C$ опишем фор\-маль\-но как: $\vdash \{P\}C\{Q\} \rightarrow (\llbracket C \rrbracket \neq \bot) \wedge Q$, где $\llbracket . \rrbracket$ денотационное преобразование программного оператора \cite{allison89}, \cite{abramsky94}, \cite{winskel93}.
 \item Аксиоматические правила не полны. Выбираемое правило отсутствует для данного состояния $P$.
 \item Полученное актуальное состояние вычисления $Q$ не является ожидаемым со\-сто\-я\-ни\-ем $Q'$. 
\end{enumerate}

На рисунке \ref{fig:ProofRulesEx1} указаны наиболее важные аксиомы и правила для императивных языков программирования. Программу написанную императивной парадигмой мож\-но преобразовать в программу  декларативной парадигмы, а также обратно, бла\-го\-да\-ря симуляции  \textit{машиной Тьюринга} о вычислимости. Программы можно более абстрактно пред\-ста\-вить  блок-схемами. Вызовы подпрограмм завершаются «\textit{обычной}» стековой ар\-хи\-тек\-ту\-рой.
Изначальное вычисление Хора \cite{hoare69} не накладывает дополнительные ограничения на описания утверждений в математических формулах, также как и предложенные подходы из \cite{apt93} не накладывают дополнительные ограничения. Апт предлагает вычисление Хора для верификации одно- и многопоточных программ, а также ут\-вер\-жде\-ния классифицировать на входные и выходные переменные. 
Позже мы оце\-ним достоинства и недостатки различных подходов верификации программ.

Итак,  \textit{правило последования} (SEQ) из рисунка \ref{fig:ProofRulesEx1} служит примером и означает: если имеется оператор разделитель «;», то для доказательства  корректности $A_1;A_2$ с предусловием $P$ и постусловием $R$ необходимо доказать, что существует промежуточное утверждение $Q$, если доказать сначала первый оператор $A_1$ с  предусловием $P$, а затем $Q$ служит в качестве предусловия в доказательстве $A_2$. Корректность последования считается доказанным, как только $A_1$ доказано, а $A_2$ доказано с  постусловием $R$.

Вторым примером служит \textit{правило логического последствия} (CONSEQ), которое является обобщением правил (SP) и (WP).  Кон\-кре\-ти\-зи\-ро\-ван\-ные правила, либо ужесточают предусловие, либо обобщают постусловие. Если предикат верный в общем случае, тогда для некоторого множества $V$ предикат также верен в частном случае для любого $v \in V$ или для любого подмножества $V_1 \subseteq V$. $V$ является  \textit{обобщением} от $V_1$, а $V_1$ является \textit{конкретизацией} $V$. Нетрудно убедиться в том, что  импликация $V \Rightarrow V_1$ в силе, однако, $V_1 \Rightarrow V$ не для каждого предиката действительно. Исходя из общего случая, можно выводить верность предиката без ограничения общности для конкретного случая, поэтому в силе пра\-ви\-ло (SP), а также в силе (WP). Если (SP) и (WP) объединить, учитывая переименование переменных состояний вычисления, то как раз получается (CONSEQ).

Третьим примером служит  \textit{правило цикла} (LOOP). Это правило гласит, что если имеется цикл в качестве программного оператора и мы имеем предусловие $P$ и постусловие $Q$, то достаточно доказать корректность блока цикла $S$, как единый про\-грам\-мный оператор и добавить к $P$ условие цикла $B$. При выходе из блока цикла необходимо добавить в постусловие  отрицательное условие цикла.

Отмечается, что благодаря циклу и  \textit{правилу присвоения} возможно преобразовать любой другой \textit{вид цикла} в данный вид цикла (LOOP) как изложено в рисунке \ref{RulesLoopReplacements}.

Ради простоты, договоримся, что в \textit{правиле} (FOR) \textit{ранее известных количеств итераций} $n$ переменная $i$ является свежей, т.е. неиспользуемой в  предикате $P$ или $Q$.
Если переменная не свежая, то согласно \textit{диапазону годности} переменных, вводится свежая переменная. С проблемами \textit{индексации термовых выражений} и подхода избежания  коллизии наименований можно ознакомиться в \cite{pitts02}. То есть, правило (REPEAT) можно преобразовать в (LOOP) и обратно, однако, в общем (FOR) можно только в (LOOP) преобразовать, т.к. для обратного шага всегда необходимо заранее знать количество итераций, что не всегда известно. С вопросами преобразования \textit{примитивной рекурсии}, \textit{$\mu$-рекурсии} и общими видами взаимной формы рекурсии, а также с вопросами свойств терминации при  \textit{взаимной рекурсии} можно ознакомиться в \cite{bekic84}.

Правило для оператора  условного перехода здесь отдельно не вводится потому, что он без ограничения общности может быть сопоставлен  правилом цикла (LOOP). Для полного понимания преобразуемости, рекомендуется ознакомиться с ми\-ни\-маль\-ным языком с точки зрения   вычислимости «\textit{PCF}» \cite{plotkin77}, \cite{cohn83}. Основной идеей вычислимости является симуляция Тьюринг-вычислимых функций с помощью ми\-ни\-маль\-но\-го набора программных операторов.
Особенность циклов, в отличие от других \textit{линейных программных операторов}, заключается в том, что один и тот же блок повторяется множество раз, при этом, состояние переменных обычно меняется. Чтобы зафиксировать совокупность всех изменений блока цикла, необходимо про\-ан\-ал\-из\-и\-ро\-вать  \textit{зависимость данных} и все изменения переменных до выхода из цикла. Повторение цикла необходимо абстрагировать, обычно вручную, подставляя новые искусственные переменные, для определения наиболее обобщённого уравнения целевых переменных (ср. рисунок \ref{fig:CFGEx1}), которые формируют \textit{инвариант цикла}. На\-хож\-де\-ние  инварианта может часто оказаться трудным, по крайней мере не тривиально и чисто автоматически не может быть выявлен, поэтому требует построение спе\-ци\-фи\-ка\-ции вручную. Постусловие блока цикла представляет собой формулу, которая содержит инвариант. Аналогично к фиксированному отображению в  проективной геометрии, когда одна точка при трансформации остается фиксированной, тогда инвариант блока цикла это утверждение, которое остается верным при любом множестве раз итераций. Итак, данная инвариантная формула $\Phi$ должна соблюдать равенство $\Phi \circ Y = Y \circ \Phi \circ Y$, где $Y$ является некоторым \textit{комбинатором плавающей точки}, а $\circ$ является бинарным оператором применения функций. $Y$ является \textit{некоторым} синтаксическим методом симуляции повтора, либо его цель может определиться как «\textit{поисковиком минимума}», который широко применяется в  $\lambda$-вычислениях \cite{barendregt93}. Можно ис\-поль\-зо\-вать альтернативную нотацию вычислимости \textit{$\mu$-оператора Клини}, указав в качестве ми\-ни\-ми\-зи\-ру\-ю\-щих параметров меняющиеся переменные.

Рассмотрим пример из рисунка \ref{fig:ExampleSimpleGCD}.

Инвариантом здесь может быть $a \cdot y + b = x^b \ge 0$, т.к. равенство  остатка при  делении на $y$ не меняется циклом. $y$ является делителем целого числа $x\ge 0$, $a$ целое частное число, а $b$ это остаток при делении $x$ на $y$.
 Правило оператора присвоения (ASN) означает: если переменной $x$ присваивается годное значение $e$, то, до и после присвоения состояния $P$ остается без изменений. Состояние до присвоения среды присвоенных символов необходимо расширять, включив $x$. В случае возникновения коллизии с именем, необходимо сначала в спецификации провести переименование  конфликтующих переменных.\\

 \textit{Полнота правил} зависит от  \textit{полноты троек Хора} в виде $\{P\}C\{Q\}$, которая за\-ви\-сит от покрытия всех программных операторов $C$ вместе со всеми допускаемыми пред\-у\-сло\-ви\-я\-ми $P$ (см. опр.\ref{def:HoareCalculusCompleteness}). Постусловия $Q$ являются лишь логическими последствиями  тро\-ек, которые выводимы из $P$ и $C$.
Так как правила для любой годной программы могут применяться потенциально в любом порядке, возникает вопрос, а может ли одно правило нечаянно или преднамеренно исключить любое другое данное правило? Найти ответ наивным подходом может оказаться сложным делом.
Очевидно, что если имеется постусловие $Q$ и данная программа $C$ выводит различные $P_1$ и $P_2$, то это явно показывает на \textit{некорректность} правил Хора. Кроме полноты и корректности, согласно  лестнице качества из рисунка \ref{fig:QALadder}, оптимальность ресурсов также важна. Вопрос, насколько эффективны или «\textit{удобны}» правила Хора, затрагивает также вопрос о компактном, но понятном для пользователя представлении. Если идет вопрос об автоматизации, то это также затрагивает вопрос о вычислительной технике.
В этом разделе мы увидели, что задача постановления инварианта может оказаться сложной попыткой, т.к. для разумного и обобщённого вывода, необходимо включить все переменные блока цикла, включая все переменные памяти. Обратим внимание на то, что автоматически выделенные переменные имеют  диапазон видимости, а у динамически выделенных переменных диапазон отличается. Кроме того, необходимо заметить, что сравнение равенства между имеющейся и ожидаемой спецификацией может потребовать некоторый консенсус по представлению  состояний вычисления.

\subsubsection*{Логический вывод}

Правило $(P1)$ из рисунка \ref{fig:HoareCalc} представляет собой самую обобщённую форму логического правила. Верификация,  это проверка данной программы $C$ с условием, что при начальном предусловии $P$, следует постусловие $Q$. Верификация, это формализованный про\-цесс (см. рисунок \ref{obs:DeductionWithBacktracking}), который начинает проверку последовательности программных неделимых операторов с предусловием $P$ и с постусловием $Q$.  \textit{Ре\-зуль\-тат ве\-ри\-фи\-ка\-ции} либо верный, либо неопределённый, когда постусловие отсутствует или оператор не терминирует, либо отрицательный. Применяя правила, может воз\-ни\-кать необходимость  разветвления доказательства на под-доказательства, в итоге, \textit{струк\-ту\-ра доказательства является деревом}.

\begin{rudefinition}{Входной язык программирования}
 \textit{Язык программирования} является  \textit{формальным языком}, чьи  слова соответствуют программам, которые являются последовательностью про\-грам\-мных операторов. После запуска каждого из  программных операторов, ме\-ня\-ет\-ся  состояние памяти. В качестве входного доказуемого языка прог\-рам\-ми\-ро\-ва\-ния, рассматривается по умолчанию \textit{императивный язык}, близкий к под\-мно\-жес\-тву Си с объектным расширением.
\end{rudefinition}

Императивный язык программирования выбирается по целому ряду причин. Во-первых, императивные диалекты, как Си или Ява, довольно популярны, и необходимость введения всё новых формализмов к большой части отпадает. Во-вторых, в этой работе выбирается прототипно Си диалект, который имеет возможность удобно и просто обсуждать синтаксис и семантику операций над динамической памятью. Естественно, Си имеет моменты, которые зависят от одной или иной платформы, однако это обсуждается и важно понять, например, для дальнейшего применения предложенных подходов.

\begin{rudefinition}{Язык спецификации}
 \textit{Язык спецификации} является формальным языком, который ссылается на переменные и \textit{единицы данной входной программы},  \textit{символьные вы\-ра\-же\-ния}, \textit{кванторы} и \textit{вспомогательные единицы} для проведения до\-ка\-за\-тель\-ств. Язык спецификации подлежит некоторой согласованной \textit{формальной логике}. Спецификация, в отличие от входного языка программирования, несет  \textit{декларативный} характер, а не императивный.
\end{rudefinition}

Язык спецификации в каждом блоке графа потока управлений (см. рисунок \ref{fig:CFGEx1}) описывает состояние вычисления, опираясь на состояние памяти (см. рисунок \ref{fig:ProcessSectionLoader}).
Рассмотрим свободно выбранное  \textit{дерево вывода} из рисунка \ref{fig:ProofTreeEx1} со следующими утвер\-жде\-ни\-ями $\{A_1,A_2,A_3,B_1,B_2,B\}$. Изначально требуется доказать  тройку $B$, согласно опр.\ref{def:HoareTriple}. Для этого применяется данное правило, в антецеденте которого должно иметься $A_3$ и $B_2$, оба из которых следует отдельно доказать с соответствующими сопоставлениями так, чтобы имелось соответствие строго по правилу (например, пре\-об\-ра\-зо\-ва\-ние локальных символов) с \textit{консеквентом} $B$. Далее, применив некоторые име\-ющ\-ие\-ся правила, мы показываем, что $B_1$ является предусловием для $B_2$, а $B_1$, согласно правилу, всегда верное утверждение, например, $\{n=0 \wedge n \ge 0\}a=5;\{n=0\}$ если очевидно, что $a$  не связанная, т.е. \textit{свободная} переменная с $n$. Далее доказывается $A_3$, в результате чего, получается, что $A_2$ противоречит самому себе. Это условие достаточное, чтобы $A_3$ вычислялось как «\textit{ложь}», а следовательно и $B$. То есть, мы только что доказали, что тройка утверждения Хора $B$ неверна и причина тому $A_2$. Поэтому в данном примере нет необходимости дальше доказывать $A_1$, независимо, верно оно или нет.

\begin{rudefinition}{Логическое следствие}
\textit{Логическое следствие} $A\vdash B$ обозначается как утверждение $A$, к которому применяется некоторое данное правило один раз. Оно приводит к ут\-вер\-жде\-нию $B$ (по Фреге \cite{frege}). Если $B$ получаем после применения правил несколько раз подряд, включая ни разу, то следствие получается $A \vdash^{*} B$. Если мы хотим выразить, что некоторая тройка $A$, согласно данному \textit{набору правил Хора} $\Gamma$ всегда истина, то мы это обозначим как $\models A$, либо как $\models_{\Gamma} A$ если ударение поставить на выбранный набор правил из набора $\Gamma$.
\end{rudefinition}

\begin{rudefinition}{Доказательство как поиск}
Доказательство в вычислении Хора, при данном наборе правил $\Gamma$ и данного следствия $B$, является \textit{поиском аксиом}, т.е. $\models_{\Gamma} B$.
\end{rudefinition}

В данном определении мы сознательно допускаем неточность в связи с набором правил и вычислением Хора. Например, ссылаясь на $\Gamma$, мы ссылаемся на  \textit{формальную логику}, которая состоит из замкнутого по зависимости подмножества данных правил Хора, как \textit{множество носителя} и базисных логичных констант. 
По умолчанию мы согласуем, что для любого данного набора правил, для логического вывода, мы подразумеваем присутствие корректно определённого  вычисления Хора, согласно тройке Хора из опр.\ref{def:IncomingProgrammingLanguage}, опираясь на опр.\ref{def:LogicalJudgement}.

\begin{rudefinition}{Корректность вычисления Хора}
Вычисление Хора является \textit{корректным}, если исключен случай, когда синтаксически подлинная программа $C$ и данный набор правил $\Gamma$ выводят различные противоречивые результаты.
\end{rudefinition}

Если один логический вывод приводит к одному результату $B_1$, а второй также допущен согласно $\Gamma$ и вывод приводит к другому результату $B_2$, который не выводим из $B_1$ или наоборот. т.е. $\{P\}C\{Q\} \vdash^{*} \{P1\}C_1\{Q1\}$ и $\{P\}C\{Q\} \vdash^{*} \{P1\}C_2\{Q2\}$ но, при этом, $\{P1\}C1\{Q1\} \nvdash^{*} \{P2\}C_2\{Q2\}$ и $\{P2\}C2\{Q2\} \nvdash^{*} \{P1\}C_1\{Q1\}$ (см. рисунок \ref{fig:CRTonHoareTriples}), то правила являются \textit{не корректными}. Это означает, что свойство корректности данного вычисления Хора соблюдает  \textit{свойство диаманта/ромба}, т.е.  \textit{теорема Чёрча-Россера} применяется к тройкам Хора (см. \cite{peirce10}).
Если хотя бы одно утвер\-жде\-ние при выводе противоречит результату другого вывода, то правила Хора являются \textit{не корректными} и \textit{не полными} (см. опр.\ref{def:HoareCalculusCompleteness}).
 \textit{Сходимость} --- более жесткое требование, чем кор\-рект\-ность и не всегда соблюдается например, из-за \textit{незавершения вычисления} (более подробнее можно ознакомиться в \cite{steinbach94}). Однако, можно заметить, что корректность обязательное условие для сходимости. То есть, если вычисление не корректно, то имеется хотя бы один случай, когда выводятся два или более различных результата и это обязательно не корректно.
Согласно Штайнбаху \cite{steinbach94} и его рассматриваемой  \textit{системе переписки термов} (с англ. «\textit{term rewriting system}») \cite{baader98}, \textit{терминация программ} сильно определяет сходимость. Он представляет правила вывода в качестве правил переписки термов. С помощью \textit{ограниченности упорядоченных цепочек}, он может в частных случаях доказать тер\-ми\-на\-цию системы переписки. Аппроксимация вер\-хне\-го порога выводов системы в общем из-за теоретической нерешимости не рас\-прос\-тра\-ня\-ет\-ся на  самосодержащие термы или на неограниченные символы \cite{plaisted85}. Идея  \textit{нисходящей цепочки}   тесно переплетается с \textit{теорией доменов} \cite{scott76}. Штайнбах использует цепочки для решения  \textit{лимита} выводов, т.е. для решения вопроса терминации, которая является условием корректности.

Кук \cite{cook78} рассматривает корректность и полноту различных вычислений Хора. Отмечается, что оба свойства существенно могут меняться, например, применив лишь несколько модификаций к переменным в процедурах. Общее понятие полноты по Куку определяется, как \textit{тотальная функция} покрывая все входные программы, учитывая ранее упомянутые свойства в \cite{apt81}.
Корректность определяется, как эк\-ви\-ва\-лент\-ность между \textit{наблюдаемым и должным поведением}, как это предлагается в \cite{davis94}, \cite{nielson99} с помощью  \textit{операционной семантики} \cite{plotkin81} над тройками Хора. Вычисления интерпретируются с помощью  \textit{абстрактного автомата}. Для до\-сти\-же\-ния «\textit{удоб\-но\-го}» --- здесь подразумеваются, либо полные, либо корректные вычисления, вводятся до\-пол\-ни\-тель\-ные ог\-ра\-ни\-че\-ния в императивном языке программирования, такие как:

\begin{itemize}
 \item \textbf{Огран. №1} разрешается использовать  глобальные переменные в процедурах, однако, за\-пре\-ща\-ет\-ся их передавать в качестве актуальных параметров.
 \item \textbf{Огран. №2} запрещаются   параметры по вызову (или по ссылкам).
 \item \textbf{Огран. №3} рекурсивные процедуры и  \textit{функционалы} (\textit{функции высшего порядка})
\end{itemize}

Почему такие ограничения, как только что были показаны, могут привести к некорректности или к неполноте? 
Классический стек при передаче \textit{параметров по ссылкам} нуждается в обратной связи сквозь стек-окон вызовов. Как только, прекращается вызов, то так прекращают существовать и значения. Передача по ссылке разрешает манипуляции единого содержимого в других  \textit{стековых окнах}, но адрес  \textit{указателя} на объект может меняться в каждом стековом окне.
Если допускать \textit{параметры по вызову} или рекурсивные функции, то это может привести к из\-ме\-не\-ни\-ям вне соответствии \textit{вызову стека}, как например, \textit{глобальные переменные}, они практически могут меняться везде. В таких случаях «\textit{наивные}» спецификации могут быть просто неправильными в общих случаях. Этому можно противостоять, ограничив модус входных и выходных параметров процедур.  Иерархические спецификации \cite{schwinghammer09}, \cite{birkedal06} могут сильно раздуть спецификации, и из-за широкой периферии возможных мест в программе, где может поменяться один указатель, польза при этом, явно ог\-ра\-ни\-чи\-ва\-ет\-ся.
В итоге, они также могут оказаться мало эффективными, т.к. все встроенные процедуры специфицируются --- при этом, исходя из наиболее общего сценария. В частности, объектные экземпляры могут потерять свойства идентичности и целостности, т.к. по обобщённой схеме передача процедур в качестве параметра, может привести к не непредвиденным действиям, а следовательно, может кардинально поменять состояние вычисления.\\

Кларк \cite{clarke79} выделяет исключительно актуальные ог\-ра\-ни\-че\-ния вычисления Хора, которые остаются до сегодняшнего дня. Острые ог\-ра\-ни\-че\-ния касаются:

\begin{itemize}
 \item \textbf{Огран. №4}  выразимость и  неполнота \textit{языка утверждений}
 \item \textbf{Огран. №5} ограничения в связи с  инвариантами циклов (преобразование, и т.д.)
 \item \textbf{Огран. №6}  рекурсивные процедуры (не) использующие глобальные и \textit{статические пе\-ре\-мен\-ные}
 \item \textbf{Огран. №7}  со-процедуры в качестве входного и выходного параметра
 \item \textbf{Огран. №8}  динамическое выделение и освобождение ячеек памяти (см. \cite{raman12})
\end{itemize}

Также как и Кук, Кларк подразумевает под корректностью гарантии о том, что все синтаксически корректные теоремы выводятся верно, а все не корректные те\-о\-ре\-мы, как ложные.
Как любая формальная система, вычисления Хора тоже подлежат теоретическим ограничениям.
Например, если рассмотреть  \textit{теорию целых чисел}, то можно всегда придумать всё новые теоремы над целыми числами, которые явно корректны, но которые нельзя доказать данным замкнутым вычислением. Этот феномен лучше известен, как  \textit{теорема Гёделя о неполноте}.
Кларк также замечает, как это ранее до него заметил Кук, что если вычисление Хора содержит не завершающийся цикл, то причиной служит рекурсия в правилах, которая применяется сама к себе, поэтому верификация может не завершаться.
Он предлагает запретить  \textit{псевдонимы ссылок} (см. далее), а ради корректности исключить неограниченную рекурсию.
Про\-бле\-му Кларка можно частично разрядить, если правила сопоставления при\-ме\-ня\-ют\-ся сначала к крайне внешнему терму, вместо к крайне внутреннему терму.
Это необходимо учесть при реализации правил Хора в программных средах, где сначала вычисляются параметры процедур, а только затем, передаётся  содержимое параметров.
Языки программирования без  \textit{ленивого вычисления}, как например,  «\textit{OCaml}», в отличие от \textit{ленивого вычисления}, как например,  «\textit{Хаскель}», должны, если этого требует ленивый алгоритм, добавить ленивое вычисление искусственным добавлением дополнительных проверок на местах использования параметров.
По Кларку выразимость всегда касается языка утверждений, который может быть представлен в виде термов. Ограничение №6 касательно рекурсивных процедур можно исключить, если: (i) рекурсия завершается, т.е. некоторая нисходящая  цепочка вывода всегда существует (см. \cite{steinbach94}) и/или (ii) последовательность параметров и  ти\-пи\-за\-ция па\-ра\-мет\-ров \cite{cardelli96-2} при вызове процедур должна точно совпадать и исключать  переменные, выделенные в актуальном стековом окне, т.е. не глобальные, не ста\-ти\-чес\-кие пе\-ре\-мен\-ные, и т.д.

Кофмэн \cite{kaufmann04} выделяет проблемы выразимости и возможность при\-ме\-ни\-мос\-ти на практике, как наиболее важные, которые сильно противодействуют популяризации \textit{формальному методу} верификации --- языки спецификации и языки про\-грам\-ми\-ро\-ва\-ния.
Наиболее важными практическими проблемами выразимости он считает  ин\-дук\-цию и абстракцию. Для того, чтобы простым образом различать ошибки некорректного вычисления от не\-до\-сто\-вер\-ной спецификации, он также замечает необходимость, простого но обобщённого подхода для генерации   \textit{контр-примеров} --- пример, который доказывает неверность данной формулы с помощью конкретных входных символов.

Герхарт \cite{gerhart76} анализирует существующие к тому времени подходы для решения проблем полноты, которые к сегодняшнему дню решены не в полном объёме и решены не удовлетворительно: (i) наиболее обобщённый подход в решении  терминации прог\-рамм, (ii) решение открытых вопросов верификации переменных в различных об\-лас\-тях памяти, как например,  статические переменные (см. \cite{clarke79}), (iii) вопросы удобства наиболее обобщённого представления и использования систем верификации. Она отмечает, что  \textit{индуктивные определения} являются одним ключевым методом в ре\-ше\-нии проблем выразимости, и поэтому их рассматривает как многообещающий технический аппарат. 
Герхарт выдвигает требование: простые доказательства дол\-жны ссылаться на универсальные и глобальные требования. Из \cite{clarke79} также следует, что проблемы из (ii) можно считать трудными и глобальными, что вряд-ли ми\-ни\-маль\-ная модификация вычисления Хора позволит их решить.

В случае нехватки хотя бы одного правила до завершения верификации, вы\-чис\-ле\-ние считается \textit{неполным}, а сама верификация \textit{неопределённой}.
На практике, уже маленькая модификация программных операторов или свойств единиц спе\-ци\-фи\-ка\-ции, может привести к существенному изменению системы вычисления \cite{clarke79, cook78, cook71}. Это свидетельствует о большой сложности системы верификации.
%  EXAMPLE 1
Пример 1: согласно ограничению № 6 отсутствие статических переменных или рекурсивных процедур может всё равно привести к полной программе, в зависимости от того, какая реальная часть программы рассматривается, и какие комбинации допускаются.
% EXAMPLE 2
Пример 2: полнота верификации программ с  внутренними про\-це\-ду\-ра\-ми в общем недействительна, когда хотя бы одно из ограничений 6 или 7 не соблюдается. Если два правила приводят к различным результатам, т.е. из данного состояния $A$ вы\-во\-ди\-тся $A \vdash B_1$  и $A \vdash B_2$, при этом $B_1$ и $B_2$ синтаксически различны, но  свойство диаманта соблюдается, то $B_1$ и $B_2$ являются лишь промежуточными состояниями и оба состояния сходимы. Проблема проверки  сходимости методом конечного от\-сле\-жи\-ва\-ния может оказаться неэффективным, т.к. во время ве\-ри\-фи\-ка\-ции необходимо проводить экспоненциальное количество под-доказательств. Эту про\-бле\-му можно избежать, если детерминировать все правила.
Также Кларк при\-во\-дит следующее ограничение для устранения неполноты:
%%%%%%%%%%%%%%%%

\begin{itemize}
 \item \textbf{Огран. №9} частично-вычисляемые структуры данных.
\end{itemize}

У  частично-вычисляемых структур данных все  поля вычисляются в момент дос\-ту\-па. Типичный пример взят из \cite{thompson97} на языке  \textit{Хаскель} о потенциально бесконечных  линейных списках.

\begin{verbatim}
 take 10 [ (i,j) | i <- [1..], let k = i*i, j <- [1..k] ]
\end{verbatim}

вычисляет

\texttt{[(1,1),(2,1),(2,2),(2,3),(2,4),(3,1),(3,2),(3,3),(3,4),(3,5)]},

од\-на\-ко, определение  линейного списка как второй аргумент от функции \texttt{take} не имеет лимита. Данная структура определяет множество, которое имеет только нижний порог (это целое число 1), но не имеет верхнего порога. С помощью  частично-вычисляемых структур, можно проверить, насколько  \textit{строго вычисляет} данная процедура.
«\textit{Строго}» подразумевает, что данная процедура сначала вычисляет все  входные параметры, а затем их  записывает в память, а \textit{не строгие} процедуры означают, что  входные параметры вычисляются частично и только тогда, когда они требуются на данном этапе алгоритма (см. \cite{thompson91}, \cite{thompson97}).
Таким образом, из этого можно сделать вывод, на\-при\-мер, проверить  терминацию программы можно, используя  бесконечно определённую структуру данных в качестве быстрого теста.

Уэнд \cite{wand76} подразумевает под  \textit{полнотой функции} определения Кука, уточняя, что каждый верный параметр на входе соответствует верному параметру на выходе, а каждый не верный параметр соответствует ошибке. То есть, переходная функция должна быть  тотальной, и все  не терминирующие функции не определены по умолчанию (см. опр.\ref{def:HoareCalcCorrectness}).
Уэнд показывает, что спецификация программы, представленная  графом потока управления тогда не определена и не полная в общем случае, когда для описания графа используется  \textit{логика высшего порядка}.
По Уэнду  квантифицируемые предикаты должны давать «\textit{понятное оп\-ре\-де\-ле\-ние}» самым лучшим  образом так, чтобы человек прочитав лишь спе\-ци\-фи\-ка\-цию, мог бы сразу и интуитивно охватить полностью замысел предиката.

Кук \cite{cook71} сравнивает решение проблемы  \textit{выполнимости булевых формул} с тео\-ре\-ти\-чес\-ки\-ми оценками решения проблем верификации. Статья представляет тео\-ре\-ти\-чес\-кие пороги сложности. С практической точки зрения, статья не пригодна по двум причинам. Первая причина - пороги слишком грубые и поэтому для при\-ме\-не\-ния недостаточны. Вторая причина - основной характер статьи --- это философский дискурс познавательного характера.

Лэндин \cite{landin64} предлагает элементы формального вычисления преобразовать в вы\-ра\-же\-ния  $\lambda$-термов, как это было изначально  предложено Чёрчом. Под фор\-маль\-ным вычислением можно также подразумевать вычисление Хора. Лэндин на примерах  условных переходов и  рекурсии показывает, что термы полностью покрываются. Более того, он показывает, что программы на основе   функциональной парадигмы \cite{thompson97}, \cite{bird88} могут быть представлены также в программе на императивном языке про\-грам\-ми\-ро\-ва\-ния с помощью  \textit{замыканий} (с англ. «\textit{closure}») и на основе  операционной семантики. То есть, представляются обобщённые  \textit{модели вычислимости}, о которых важно знать при моделировании систем верификации.\\

Исходя из опыта последних декад, Апт \cite{apt81} определяет некоторый обобщённый диалект  Си с редуцированным множеством  программных операторов для анализа  полноты и  корректности. В дополнении ранее обсуждаемых работах, Апт видит общую рекурсию, как наиболее важную проблему, которую трудно прогнозировать в отношении следующего состояния программы. Поэтому, он предлагает ог\-ра\-ни\-чи\-ться  примитивной рекурсией. Также Апт предлагает допускать только те вызовы про\-це\-дур, у которых актуальные параметры совпадают с  \textit{формальными па\-ра\-мет\-ра\-ми}, например,  «\textit{инкорректные параметры}» исключать, т.к. они могут послужить це\-ло\-му ряду аномалий. Например, неверное разделение,  неинициализированные поля, а также полученный  функционал в связи с  отсечением параметров при вызове, могут полностью поменять семантику процедуры, но принципиально не решить ни одну проблему. С помощью урезания параметров можно сильно варьировать семантику и заметно изменять синтаксис.

Кук \cite{cook78} и другие упомянутые авторы рассматривают две основные проблемы полноты:

\begin{itemize}
 \item \textbf{Полнота №1}  Незавершение вычисления процедуры.
 \item \textbf{Полнота №2}  Ограничение выразимости  языка утверждений (например, предикаты и  ин\-ва\-ри\-ан\-ты).
\end{itemize}

В виде примера корректных и  полных правил Хора, Кук выделяет рисунок \ref{fig:SoundNCompleteExample} (Кук использует обратную запись скобок, как и Хор).
На рисунке $A_j$ представляют программные операторы, $D$ является блоком декларации пе\-ре\-мен\-ных, $\sigma$ является переменной средой, а $\star$ обозначает \textit{звезду Клини}. $\sigma$ имеет тип:
$$\textnormal{имя переменной} \rightarrow \textnormal{значение}$$

В дополнении к ограничению Хора 6, надо отметить принципиальное ограничение:

\begin{itemize}
 \item \textbf{Огран. №10} не автоматически выделенные переменные
\end{itemize}

Ранее уже упоминалось, что подключение глобальных, статических и динамически выделенных переменных, может аннулировать правила Хора.
Далее,  локальные переменные используемые в некоторых  нитях одновременно, могут также ан\-ну\-ли\-ро\-вать правила, т.к. аспекты параллельного вычисления в общем не могут быть покрыты, например, правилами Кука. Тема этой работы посвящается однопроцессорному вычислению, а не параллельному запуску программы.

Для полной индустриальной применимости, необходимо включить обработку  \textit{ис\-клю\-че\-ний}, которые просто так не вписываются в существующие правила и довольно трудно формально корректно и полностью охарактеризовать не только, но и в целом из-за возможного изменения  стека при каждом  программном операторе по-разному (см. \cite{dedinechin00}, \cite{goodenough75}). До сегодняшнего дня не было найдено предложение о  вычислении Хора, которое позволило бы (корректно и/или пол\-но\-стью)  \textit{отматывать ди\-на\-ми\-чес\-кую память}, аналогично \textit{стеку} \cite{goodenough75}.

Хор \cite{hoare69} считает  абстракцию спецификации самой главной преградой для ве\-ри\-фи\-ка\-ции мало тренированным инженерам, которые желали бы быстро научиться верифицировать простые примеры. По Хору тяжело формализовать утверждения для программных меток, безусловных переходов и передач параметров по имени.
Он не настаивает и не опровергает использование любой  логики, любого порядка, однако, он считает декларативную спецификацию утверждения решительным фак\-то\-ром успеха.

%%%%%%%%%%%%%%%%%%%%%%%%%%%%%%%%%%%%%%%%%%%%%%%%%%%%%%%%%%%%%%%%%%%%%%%%%%%%%%%%%%%%%%%%%%%%%%%%%%

\subsubsection*{Автоматизация логического вывода}
Цель этого раздела заключается в введении и демонстрации практических и тео\-ре\-ти\-чес\-ких проблем  автоматизированной верификации.
В этом разделе мы рас\-смо\-трим примеры в системе верификации  «\textit{Coq}». В Coq \cite{bertot04} можно доказывать заранее специфицированные утверждения с помощью теорем, различных ин\-дук\-тив\-но определённых структур и различных команд на основе типизированного $\lambda$-вы\-ра\-же\-ний второго порядка.
Утверждения задаются на функциональном языке  «\textit{Gallina}», а последовательность доказательств задается языком последовательных команд  «\textit{Vernacular}».
«\textit{Coq}» является ассистентом доказательств потому, что часто он сам не в состоянии полностью и самостоятельно, даже для простых примеров, находить доказательства. Вместо этого, в «\textit{Coq}» имеется множество узкого круга поддерживаемых  теорий, которые можно подключать в работающее  ядро ассистента.  Ассистент позволяет записывать и отслеживать доказательство и выявить про\-ме\-жу\-точ\-ные состояния доказательства.

В последовательность команд включается набор так называемых  «\textit{тактических команд}». Они пробуют текущее представление программного состояния упростить полу\--авто\-ма\-ти\-чески, используя рассматриваемую теорию.
Теория  ве\-щес\-твен\-ных чисел, например, используется для ускорения сходимости, результатом чего является, либо завершение доказательства, либо упрощение состояния вычисления.
«\textit{Coq}» --- довольно мощная и широко используемая  платформа для верификаций, что в области авто\-ма\-ти\-за\-ции доказательств теорем, можно встретить не так часто. Применение также включает в себя верификацию корректности фреймворков компиляции \cite{rideau08}, \cite{leroy09}, \cite{blazy06}, \cite{blazy05}, \cite{leroy06}, \cite{leroy09-2}, \cite{raman12}.

Формулы из  логики предикатов задаются языком «\textit{Vernacular}».

«\textit{Coq}» использует при редукции нормализованные формулы, которые могут быть не определены или  не полностью определены. «\textit{Coq}»-схемы  основаны на \textit{ленивой редукции} на вычислительной модели  $\lambda$-вычисления второго порядка (см. \cite{cardelli96-2}, \cite{peirce10}, \cite{mitchell96}). Разумеется, что из-за произвольного вида правил, естественно не может быть гарантии касательно  полноты.  Типизация $\lambda$-вычислений позволяет избегать целый ряд  \textit{парадоксов типизации}, в связи с \textit{рекурсивными определениями Кан\-тор\-ских множеств} (см. \cite{barendregt93}, \cite{bird97}), как  термы содержавшие сами себя. От\-сутс\-твие типизации может приводить, к  \textit{парадоксу Расселя о брадобрея}.

Примерами корректно определённых $T_{\lambda2}$-типов, являются например, $\forall a.a$ или $(\forall a_1.a_1\rightarrow a_1)\rightarrow (\forall a_2.a_2\rightarrow a_2)$.

Корректно определёнными $\Lambda_{T_{\lambda2}}$-типами являются, например, $\Lambda a.\lambda x: a.x$ или 
$$\lambda x:(\forall a.a\rightarrow a).x (\forall a.a\rightarrow a) x.$$
Редукция ($\beta$-редукция, см. \cite{barendregt93}) $\lambda$-термов проводится как применение возможно неопределённого терма к данной $\lambda$-абстракции, например: 
$$(\Lambda a.\lambda  x: a.x) \ Int \ 3$$
можно редуцировать к $(\lambda x: Int.x)\ 3$, далее редуцируется к «$3$», при этом $Int$ тип целых чисел, т.е. множество $V$.

Задача  редукции $T_{\lambda2}$-термов заключается в: (1) вычислении результата и (2)  про\-вер\-ке типов вычисления. Проверка типа отличается от задачи верификации, отсутствием состояния.

Ради простоты, в определении базовые правила были упущены, т.к. уни\-вер\-саль\-ность разрешает  их определить более обобщёнными не типизированными $\lambda$-вы\-чи\-сле\-ни\-я\-ми \cite{barendregt93}, например, аксиома $\overline{\Gamma \vdash x:t}$ или правила ($\exists$-Intro) и ($\exists$-Elim), которые определяются аналогично ($\forall$-Intro) и ($\forall$-Elem).

Так как Coq основан на $T_{\lambda2}$ и редукция  \textit{редексов} производится снаружи во внутрь, то проблема проверки типов решима. Сложность проверки линейная. Для \textit{стратегии редукции с внутренней стороны к внешней}, в общем случае, проверка не решима (см. \cite{peirce10} и \cite{bertot04}).\\

Рассмотрим относительно простой пример «\textit{исключённого третьего}» на рисунке \ref{code:CoqPeircesTautology}. Пример тавтологии $p \vee \neg p$ можно считать интуитивно понятным, но при отсутствии  ло\-ги\-чес\-ких таблиц и при использовании только правил импликации --- так оно положено в  \textit{интуиционистском суждении логических утверждений}, задача может оказаться го\-раз\-до труднее. В зависимости от набора правил, \textit{тавтология Пирса} может быть в принципе не выводима.

Рассмотрим пример из рисунка \ref{code:CoqPeircesTautology}.
Нам необходимо доказать, что первое определение \texttt{peirce}, которое обходится без дизъюнкции и отрицания, но содержит импликацию «\texttt{$\rightarrow$}» может быть преобразовано в определение \texttt{lem}. Обратим внимание, что оба определения содержат кван\-ти\-фи\-ци\-руе\-мые утверждения. Доказательство является последовательностью команд, начиная после  ключевой команды  \texttt{Proof} и заканчивая перед ключевым словом  \texttt{Qed} (перевод c латинского «\textit{quod erat demonstrandum}» означает, «\textit{что и требовалось доказать}»). Сначала имеется лишь теорема о равенстве обоих теорем, затем обе стороны равенства развёртываются.
Затем \texttt{firstorder} пробует сопоставить $\forall$-кван\-ти\-фи\-ци\-ро\-ван\-ные утверждения преобразованные в  \textit{нормальную форму Сколема}. Те\-перь необходимо доказать на правой стороне определения \texttt{lem}, что для любого предположительного верного утверждения $p$,  $p \vee \neg p$ также верно.
В этом случае, левая сторона является доказуемой  гипотезой $H = \forall p,q. ((p\rightarrow q)\rightarrow p)\rightarrow q$, которую необходимо доказать для любых утверждений $p$ и $q$. Сейчас $p \vee \neg p$ заменяется $p$, $q$ сопоставляется $p \vee \neg p$. Таким образом, мы получаем $(p \vee \neg p \rightarrow \neg (p \vee \neg p)) \rightarrow p \vee \neg p$ --- в качестве доказуемой текущей гипотезы, по-прежнему условию, что $p$ является верным утверждением. Теперь, чтобы доказать верность гипотезы, необходимо следствие \texttt{peirce} отделить от предусловия. Необходимо текущую гипотезу аб\-стра\-ги\-ро\-вать и затем опять преобразовать в нормальную форму Сколема, чтобы $q$ более не являлась квантифицируемой переменной. Затем мы получаем новую более простую гипотезу $H_0 = (p \rightarrow q) \rightarrow p$, которую нам удастся доказать, если предположить, что из $H = \forall p. p \vee \neg p$ правая часть дизъюнкции верная. То есть, мы вводим новую гипотезу $H_1 = \neg p$. Такой выбор произвольный и применив $\neg p$ к $H_0$, нас сразу приведет к тому, что $p$ верное, потому, что  импликация всегда верна, как только левая сторона импликации ложная. Так как мы предположили изначально, что $p$ является верным утверждением, мы доказали правоту теоремы.  Тактика  \texttt{tauto} обязательна для успешного завершения доказательства теоремы, которая из данных существующих гипотез и верных утверждений пытается простейшим образом, механически найти завершение доказательства без каких-нибудь дополнительных знаний о теореме или используемых лемм.

Обратим внимание на то, что, хотя пример простой и интуитивен, успешное доказательство все-таки требует довольно не малых ресурсов для  преобразования в нужную форму. Для применения одной или иной тактики требуется также применение, не очевидных  формул абстракций. Не трудно заметить, что автоматически такого рода \textit{нестандартные преобразования} будет очень тяжело выявить и рас\-поз\-нать.

Если  правила Хора  полны и удастся провести доказательство из аксиом до доказуемой теоремы, то доказательство может быть \textit{автоматизировано}. Если пра\-вила Хора полны, то любая теорема может автоматически доказываться -- ответ, либо положительный, либо отрицательный.
Чтобы выразить  \textit{формальную теорию}, необходимо определить:  \textit{семиотику},  \textit{синтаксис} (см. например опр.\ref{def:FirstOrderPredicateLogicFormula}),  \textit{семантику}. Когда речь идёт о языках, а формальная теория обозначается именно выражениями некоторого языка, целесообразно определить  \textit{прагматику}. Семантики, например, \textit{аксиоматичная семантика}, обозначают в формальной системе, например в вычислении Хора, какое выражение выводимо или нет.  Знаки и взаимосвязи формул являются абстракцией некоторых реальных предметных объектов или физики (классический термин был введён в аналогии реальных объектов, которые при\-над\-ле\-жат естественным правилам природы). Поэтому, описание элементов «\textit{физики}» исторически часто называется  \textit{метафизикой}, т.е. абстрагированной физикой, либо  логикой. Необходимо отметить, что любая  \textit{формальная логика} также является по определению специфической  \textit{формальной алгеброй}.

Аналогично  модульному программированию, доказательства могут строиться с целью упрощения и выявления основных мыслей доказательств. Более подробное обозначение находится в следующих главах. Для аналогии послужат \textit{единицы доказательств}  \textit{леммы} (вспомогательные теоремы),  \textit{теоремы} и  \textit{ин\-дук\-тив\-ные определения}. Самым простым примером индуктивного определения можно считать  естественные числа (см. рисунок \ref{code:NaturalNumbers}).

 Мы не будем отдельно описывать тактики, т.к. тактики лишь вспомогательные образцы последовательности команд при доказательствах. Они лишь имеют не\-ко\-то\-рый абстрактный характер для упрощения написания доказательств. Абстракция пре\-ди\-ка\-тов, для написания и восприятия человеком, имеет большое значение. Можно в общем считать: чем проще доказательство, тем оно лучше.
Символы и предикаты могут быть использованы и их определение необходимо  раз\-вёр\-ты\-вать, лишь при необходимости.
Исходя из текущего состояния вывода и со\-че\-та\-ем\-ых правил, желательно было бы автоматизировать принятие решения системой ве\-ри\-фи\-ка\-ции --- когда развёртывать определение и когда свёртывать часть данной формулы обратно к определению.
Количество и порядок развёртываний и свёртываний заранее не предсказуемо.
По ранее упомянутым причинам, лучше редуцировать термы  лениво и снаружи во внутрь, иначе текущая редукция может приостановиться, хотя имеются редексы.

Из рисунка \ref{code:CoqPeircesTautology} можно выявить: следующие проблемы касательно автоматизации доказательств.

\begin{enumerate}
 \item Описание проблемы соперничает с выразимостью. Это может привести к се\-рье\-зным ограничениям выразимости и к раздутым описаниям.
 \item Упрощение равенств косвенной теории раздувает объём и количество правил Хора, хотя, например,  арифметическая теория не связана на прямую с со\-сто\-я\-нием памяти или программными операторами. Это также препятствует автоматизации.
 \item Объяснения при отрицательном выводе, либо отсутствуют полностью, либо желают лучшего, например, интуитивный и обобщённый метод  генерации контр\-примера.
\end{enumerate}

Подробнее ознакомиться с проблемами автоматизации доказательства теорем мо\-жно в статьях Воса \cite{wos91}, \cite{wos88}, несмотря на возраст статей, до сих пор проблемы в основном остаются актуальными. Проблемы Воса можно разбить на три класса: (i) представление  модели утверждений, (ii) представление правил вывода, (iii) выбор оптимальной  стратегии и тактик логического вывода. Для повышения эффективности верификации, Вос предлагает вводить па\-рал\-ле\-ли\-за\-цию, индексацию баз данных и знаний для более быстрого доступа и обработки данных, а также анализировать модифицированные технологии вывода. В статьях можно обратить внимание на следующие, часто неявные, замечания Воса:

\begin{itemize}
 \item  Проверка типов может быть полезной с целью избежания ста\-ти\-чес\-ких ошибок, но на самом деле бесполезной для верификации.
 \item Нотация формул, возможно, не столь важна, но охватить  семантику (им\-пе\-ра\-ти\-вную) программы важно.
 \item Хотя  метод резолюции широко обсуждается в статье,  \textit{естественный вывод} все-таки рассматривается как более эффективный, особенно на практике, чем, например \textit{метод резолюции}. Причину этому, наверное можно искать только в близости к классическим математическим доказательствам.
 \item Существует соперничество между локальным и глобальным поиском оптимума в доказательствах.
\end{itemize}

Из всех проблем, которые Вос выявляет \cite{wos91}, наиболее важной можно считать избавление от избыточных формул при спецификации, а при верификации от из\-бы\-точ\-ных частей повторного вывода.
Главный лозунг Воса: «\textit{Лучше упрощать утверждения, чем перебирать различные варианты}». Эвристика Воса считает в худшем случае лучше с полиномиальной частотой проводить упрощения насколько возможно, чем рисковать экспоненциальный взлет поискового пространства про\-вер\-ки. Вос исходит из 90\% сбережений. Престон \cite{preston88}, ссылаясь на Воса, ожидает, что удовлетворительное исследование каждой из первоначальных проблем Воса \cite{wos88} займёт объём работы, как одна диссертационная рукопись.
Овербейк \cite{overbeek88}, ссылаясь на \cite{wos88}, даёт каждому из перечисленных проблем обзор практических примеров. Мекок \cite{macock75} является старой, но не устаревшей в целом, обзорной статьёй. Галлье \cite{gallier03} является более актуальной и подробной монографией по теме автоматического доказательства. Его проблемы совпадают большой частью с обсужденными. В \cite{leroy11} Леруа задаёт вопросы к общему процессу верификации и даёт критические оценки с практической точки зрения, его оценки совпадают в основном с упомянутыми. Он, Аппель и Докингс \cite{appel07} предлагают обществу разработчиков и исследователей соблюдать общие принципы при разработке верификаторов для простого сравнения между собой.\\\\

\textbf{Абстрактная интерпретация.} Кусо(-вые) \cite{cousot77} \cite{cousot92} предлагают пер\-выми  формальный метод (см. опр.\ref{def:FormalProof})  «\textit{абстрактной интерпретации}», как универсальный  ста\-ти\-чес\-кий метод, ос\-но\-ван\-ный на аппроксимации пошагово вычисляемых ин\-тер\-пре\-та\-ций данной прог\-рам\-мы. Метод анализирует граф потока управлений \cite{nielson99} данной программы.
Если граф не нормализован, т.е. либо нет ровно одного со\-сто\-я\-ния входа, либо одного состояния выхода, то граф преобразуется именно в такой, объединив входные и выходные состояния.  Программные операторы упорядочиваются не строго согласно  алгебраической решётке по порядку операторов.
С помощью абстракции, ин\-тер\-пре\-та\-ция пытается приблизить лимит \cite{abramsky94} ради введения новых зависимых параметров c помощью ограниченного случая (\textit{су\-же\-ния диапазона видимости}) и с помощью общего случая (\textit{расширения диапазона}). Таким образом, неинициализированные переменные имеют порог $[- \infty \ .. \ \infty]$.
Процесс  по\-ша\-го\-вой аппроксимации интерпретации считается завершенным, как только, две последовательные интерпретации равны. Ин\-тер\-пре\-та\-ции обновляются после каж\-до\-го программного оператора.
Из-за \textit{проблемы приостановки} интерпретации могут иметь не предсказуемые, т.е. арифметически очень большие, либо очень маленькие  лимиты, на\-при\-мер, в циклах с неопределёнными значениями переменных при входе в цикл.
% Применения:
\textit{Статическое вычисление допустимых порогов зна\-че\-ний} \cite{nielson99} является одним клас\-си\-чес\-ким применением.  Методы Кусо могут быть модифицированы и применены в самых различных областях, например, в \cite{gcc15} для улучшения быстродействия в связи с предсказыванием более оптимального ветвления запуска кода.
%ISOMORPHISM:
Чтобы сравнить равенство между двумя абстрактными интерпретациями, вводится условие для  изоморфизма изображений.
Две ин\-тер\-пре\-та\-ции $I_1, I_2$ считаются эк\-ви\-ва\-лент\-ными, когда существуют два  отображения $\alpha, \alpha^{-1}$, которые  \textit{инъективны} и  \textit{сюръ\-ек\-ти\-вны}, c которыми интерпретация $I_1$ преобразуется с помощью $\alpha$ в $I_2$, и $I_2$ обратно в $I_1$ с помощью $\alpha^{-1}$.
%End:
Хотя метод Кусо(-вых) универсален, его можно было бы автоматически с трудом преобразовать в правила Хора и наоборот из-за целого ряда причин. Во-первых, метод Кусо предполагает  аппроксимацию, а правила  Хора знают только точные результаты, либо символьные значения. То есть, можно было бы расценивать систему правил Хора, как частный случай метода Кусо(-вых), хотя цель метода направлен все-таки на арифметическую аппроксимацию интервалов. Во-вторых, для верификации вычисления Хора, ис\-поль\-зу\-ют\-ся термовые утверждения некоторого языка спецификации и т.д. 
Ком\-мен\-ти\-ро\-ван\-ный пример  абстрактной ин\-тер\-пре\-та\-ции находится на рисунке \ref{fig:CFGEx1}.

Штайнбах \cite{steinbach94} демонстрирует наглядно, что  терминация является жест\-ким ус\-ло\-ви\-ем,  предусловием  сходимости (например, применив  \textit{алгоритм Кнута-Бендикса}) и  полноты, а сам вопрос терминации в частных случаях является  решимым.

Кроме упомянутых ограничений, отдельно от этого, имеется теоретическое ограничение в связи с  \textit{арифметикой Пресбургера}.
Арифметика Пресбургера яв\-ля\-ет\-ся настоящим подмножеством более известной  \textit{арифметики Пиано} над  натуральными числами, которая знает только об операции «+» (или инверсная ею операция минус) \cite{presburger29}, и где любые другие операции отсутствуют.
Это приводит к тому, что «\textit{простые входные программы}» содержавшие плюс, можно в выражениях за конечное время решить в  предикатной логике первого порядка (на\-при\-мер, в решении офсета для ссылочных указателей, см. позже), в отличие от выражений по ари\-фме\-ти\-ке Пиано или сложнее.
В частности, это имеет значение для ре\-ши\-мос\-ти автоматического доказательства теорем \cite{cooper72}, \cite{cherniavsky76} (см. набл.\ref{obs:ModelOfComputationVerification}). Решимость выражений из арифметики Пресбургера ограничена $\Theta(m,n)=2^{2^{n \cdot m}}$, где $n$ ми\-ни\-маль\-ное количество пред\-ла\-га\-емых разветвлений поиска, а $m\geq 1$ некоторый константный фактор \cite{fisher74}. Сложность разрастает очень быстро для $n$. Важ\-но заметить, что несоблюдение арифметики Пресбургера, например, включение даль\-ней\-ших операторов, на практике не обязательно приводит к нерешимости частных случаев, однако, соблюдение ариф\-ме\-ти\-кой может сильно снизить  верхний порог решимости доказуемых теорем.

Поммерель \cite{pommerell94} оценивает теоретические критерии сложности и при\-хо\-дит к вы\-во\-ду, что неэффективные реализации прикладных  теорий являются са\-мым большим барьером эффективной верификации. К примеру, он приводит состояния моделирование  троек Хора в качестве огромных линейных систем уравнений, строки чьи, почти все пустые.
Ситуацию можно переломить, если прикладные теории изолировать от правил Хора.
Ре\-ко\-мен\-ду\-ет\-ся использовать   \textit{решатель}, который упростит анализ  полноты ос\-нов\-ных (\textit{логические}) и прикладных правил (\textit{структурные}). Из практического опыта, решатель резко повысит эффективность решения прикладных теорий \cite{nelson78}. Ре\-ша\-те\-ли запускаются во время верификации при необходимости.
Примерами ре\-ша\-те\-лей являются  функциональные языки  «\textit{Why}» \cite{why15} и  «\textit{Y-not}» \cite{nanevski08-2}.

\subsubsection*{Альтернативные подходы}

Как было упомянуто в начале главы,  верификация является  формальным методом с помощью которого, можно проверить, соблюдает ли данная программа спе\-ци\-фи\-ка\-цию или нет.
Верификация  троек Хора является статическим методом без запуска программы.
В зависимости от точности спецификации для достижения максимально качественного уровня ПО при разрастающейся сложности и функ\-ци\-о\-наль\-ности данной системы, на практике может потребоваться добавление, уточнение и сокращение специфицируемых свойств. Необходимо заметить, что сокращение правил происходит всегда тогда, когда правила обобщаются, т.е. отпадает не\-об\-хо\-ди\-мость специфицировать детально.
Кроме представленного подхода в прошлых разделах, упоминались ряд иных  статических подходов, как, например,  \textit{проверка моделей} и  локализации проблемы с помощью удаления синтаксического и окружающего кода. Также имеются динамические подходы, как, например: тестирование, безошибочная разработка программного обеспечения и расширенная  проверка типов.

\textbf{Подход №1 -- Тестирование.}  Тестирование является динамическим подходом проверки ранее подготовленных сценарий, в которых данная программная система или модуль (не) вычисляет ожидаемый результат. Сценарии это тесты, которые задаются разработчиком вручную и описывают свойства по отдельности. Когда тест успешен, следующий тест проводится до тех пор, пока все тесты выполнены.
Тест, это критерий качества, который (не) соблюдается при запуске программы.
Если хотя бы один тест не успешен, то тестирование провалилось. Когда сценарий теста известен, то, все, входные данные и выходные результаты, сравниваются согласно данному модулю. Если при проведении теста требуется ручное вмешательство, то тест не автоматизирован.  Автоматизация теста (АТ) \cite{beck15} дает существенные преимущества: высокая эффективность, т.к. за короткий промежуток времени можно проверить большое количество тестов. Далее, исключаются ошибки из-за неточных входных данных. Все внешние зависимости должны описываться тестом. Если тесты простые, то модули можно легко понять и вероятность выше, что ошибки отсутствуют.

Для данной программы $p$ и набора тестов $\forall i.t_i(p)$, каждый из тестов запускается для $p$ и наблюдается, либо успех, либо провал. Каждый из тестов можно рас\-сма\-три\-вать как утверждение одного аспекта функциональности ПО. Недостатками АТ можно считать ручную постановку, в зависимости от сложности графа потока управления $p$. Число $n$ необходимых тестов $p_1,p_2,...p_n$ для покрытия всех  ветвей может сильно возрастать. Чем сложнее получается  граф в частности, когда граф содержит обратные связи и циклы, тем сложнее ПО в общем.

\textbf{Подход №2 -- Автоматическая проверка моделей.}
Основной проблемой подхода №1 является экспоненциально возрастающий  объём тестов, который трудно авто\-ма\-ти\-зи\-ро\-вать для представительного множества тестов. Идея проверки модели заключается в описании данной программы с помощью утверждений и временных выражений, которые затем проверяются пошагово, изучив выполнимые входные. Проблема заключается в том, что решения дискретных проблем часто не три\-ви\-аль\-ные и могут быть выявлены, либо с большими затратами, либо вовсе не могут быть выявлены. Если уравнения сильно ограничиваются, то вероятность высокая, что результаты пропускаются. Если уравнения слишком обобщённые, то ре\-зуль\-та\-ты тоже могут теряться или быть не найденными из-за отсутствия быстрого нахождения, либо отсутствия решения в принципе. К уравнениям \textit{модели}, которые описывают  булевую формулу выполнимости и при правильных результатах всегда верна, при\-ме\-ня\-ет\-ся стратегия поиска для расширения найденных результатов \cite{clarke99}, сог\-лас\-но данной формуле.
Отсутствие эффективной стратегии поиска равномерно к не\-за\-вер\-ше\-нию или к очень медленному  завершению проверки \cite{kohli94}.
Пред\-ла\-га\-ет\-ся целый ряд оптимизаций, как например  \textit{символьное использование} \cite{clarke99} при работе с  \textit{моделями} или статистические методы над контролем расширения \cite{kwiatkowska07}.
Прог\-рам\-мные инструменты включают в себя, например  \textit{«VD\-M++»} \cite{weissenbacher01}, \cite{weissenbacher04}.

Альтернативно к аксиоматическим правилам можно назвать (автоматическое) тестирование \cite{marinov12} и проверку моделей, как это было впервые введено Пе\-ле\-дом и Кларком \cite{clarke99} (\cite{kroening09} \cite{weissenbacher04} представляют собой обзоры современных инструментов).
Преимущество тестирования, в отличие от проверки моделей, является простота проверки поведения программы без дополнительных формул, но, проблема полного перебора всех возможных тестов, остается без из\-ме\-не\-ния в обоих случаях.

\textbf{Подход №3 -- Безошибочная разработка программного обеспечения.} Аль\-тер\-на\-ти\-вой ранее представленных подходов можно считать \textit{«прог\-рам\-ми\-ро\-ва\-ние вообще без ошибок»}, при котором не требуется статическая или динамическая фаза проверки, точнее фаза проверки проводится одновременно с моделированием ПО, упираясь на знакомые и  «\textit{хорошие}» паттерны \cite{kerievsky05}. Если ошибок мало или ошибки незначимые с точки зрения быстродействия в не критическом месте, то естественно количество ошибок будет стремиться к нулю. Если количество программных строк приближается к нулю, то и количество возможно совершенных ошибок будет приближаться к нулю -- эти утверждения всегда в силе, независимо от конкретной программы и не требуют дополнительного объяснения: пустая программа содержит минимальное количество ошибок. 
Ошибки в программе \textit{создаются} человеком. --- Почему? При создании программы полностью исключить  «\textit{человеческий фактор}» по определению не возможно: ПО и спе\-ци\-фи\-ка\-ция задаются человеком. Ошибки в программах тоже создаются человеком. Отладка и обнаружение возможных ошибок проводится человеком.  Абстрактные типы данных, алгоритмы и интуиция --- всё это свойственно для человека. Разработка и описание модулей и интерфейсов ПО проводится человеком. Следовательно, не удастся исключить человека при написании ПО, но могут иметься подходы, которые пытаются исключить ошибки на более ранней стадии при создании ПО.

\textbf{Разработка ПО.} Вопрос об отсутствии ошибок в ПО на раннем этапе нужно искать при построении и моделировании, например с по\-мо\-щью   языка моделирования  «\textit{UML}» и  «\textit{OCL}». На практике много раз доказано, что с помощью тщательного моделирования можно избежать множество типичных ошибок, но увы, не всегда и не все виды ошибок это покрывает потому, что:
\begin{itemize}
 \item индивидуальности предполагаемого разработчика, создание ПО является твор\-чес\-ким трудом человека.
 \item трудности в предсказывании ошибок со ссылками и в том числе, из-за трудности локализации проблем.
\end{itemize}

\textbf{Преобразования.} Задача проверки корректности  \textit{транс\-фор\-ма\-ций струк\-тур данных} заключается в соблюдении корректного перехода от одного графа к другому \cite{dodds08}, соблюдая программные операторы. Подход похож на позже пред\-став\-лен\-ную модель динамической памяти в отрезках. Однако, внутри описывается и работает трансформация как   «\textit{система переписывания}», поэтому в даль\-ней\-шем она не будет рассматриваться.
Системой переписывания является, например  «\textit{Stratego XT}» \cite{johann03} или \cite{haberland08-1}, \cite{dodds08}. \cite{katz73} предлагает использовать эвристики систем переписывания для лучшей сходимости верификации программ.

\textbf{Подход №4 -- Проверка типов.}  Проверка типов исключает большинство оши\-бок, которые могут возникать из-за некорректного применения типов переменных и выражений. Например, присвоение  32-разрядного целого числа, 8-разрядному числу плавающей точки, в лучшем случае, может «\textit{проглатить}» или быть без\-обид\-ным, а в худшем случае, может привести к совершенно иному числу, если интерпретируются и копируются только все передние 8 от 32 битов.

Задача проверки типов заключается в выявлении проверки совместимости данных и выявленных типов в выражениях программных операторов.
Однако, если про\-вер\-ка типов соблюдается, то можно сказать, что семантическая сходимость типов соблюдается. Из совместимости типов не следует утверждать, что данный цикл «\textit{правильный}». Для того, чтобы можно было это утверждать, необходимо сначала определить свойство, а затем проверять. Увы, проверка типов является лишь предыдущим шагом перед верификацией. Проверка типов не со\-дер\-жит ин\-фор\-ма\-цию о  зависимостях данных и тем более не содержит никакой информации о том --- какие, на каком этапе и сколько объектов будет выделено в динамической памяти. Эта информация проверяется на этапе  «\textit{ста\-ти\-чес\-ко\-го ана\-ли\-за}» и является независимой от этапов построения программы:\\

Проверку типов можно задать как «$\Gamma \vdash e:t$?», где $\Gamma$ переменная среда содержавшая термы вместе с типами, $е$ является проверяемым термом и $t$ является проверяемым типом (см. \cite{barendregt93}). Для проверки, имеет ли данная последовательность программных операторов $e$ тип $t$, используя аксиомы и правила типизации $\Gamma$, необходимо проверить каждый из операторов в $e$. В отличие от троек Хора, термам пред\-ста\-вля\-ю\-щие программные операторы, присваивается тип, т.е. множество до\-пус\-ти\-мых значений.
В отличие от  вычисления Хора,  проверка типов только ус\-та\-на\-вли\-ва\-ет принадлежности к множеству возможных значений, т.е. к некоторому типу и проверяет сходимость типов. Частное значение некоторого объекта не специфицируется.
Имеются исключительно  программные операторы и типизированные выражения.
При про\-вер\-ке типов отсутствует также понятие как  «\textit{состояние вычисления}», при корректной типизации состояния, либо выводимы из \textit{набора правил} $\Gamma$  типизации, либо нет.

Аналогично  вычислению Хора,  проверка типов проводится снизу-вверх, т.е. про\-вер\-ка начинается с данного $e$ и успешно завершается в случае вывода.
Проверка типов является проверкой узкого круга свойств программ, т.е. является сильно ограниченной формой верификации троек Хора. Проверка типов программы яв\-ля\-ет\-ся обязательной семантической проверкой программы, которая должна про\-во\-дить\-ся до любых других семантических проверок, как например, верификации свойств  динамической памяти. Как мы увидим позже, представленные модели динамической памяти, также как и другие семантические фазы исходят из того, что выражения при присвоении годные, иначе нельзя считать входную программу корректной, как тот язык программирования интерпретируется \textit{строго} согласно спецификации, например \cite{isocpp14}.

Если далее сравнивать проверку  типов по Хинд\-лей и Милл\-не\-ру \cite{hindley69}, \cite{peirce10}, \cite{bruce02} с  трой\-ка\-ми Хора, то необходимо заметить, что при проверке типов, предусловие всегда самое универсальное, а значит верное. Далее, при проверке типов, имеются фундаментальные проблемы, которые являются следующими эк\-ви\-ва\-лент\-ны\-ми проб\-ле\-ма\-ми при проверке упрощённых троек Хора (см. глава \ref{chapter:logical} на рисунке \ref{fig:HoareCalcVSTypeSystem}).

Были предприняты расширения, которые позволяют: проверять сов\-мес\-ти\-мость  записей, ужесточение типов, свойства неприсвоенности и статические лимиты объёмов базисных структур данных \cite{flanagan02}. Но для проверки свойств  ди\-на\-ми\-чес\-кой памяти, увы, даже те упомянутые расширения, совершенно здесь не достаточны. Например, приведем \textit{правило последовательности}:

\begin{center}
\begin{tabular}{c}
  \inference[($\rightarrow$-elim)]{\Gamma \vdash e_1:(s\rightarrow t)\ \Gamma \vdash e_2:s}{\Gamma \vdash (e_1 e_2):t}
\end{tabular}
\end{center}

Правило ($\rightarrow$-elim) наглядно демонстрирует, что указание типа (например, $s \rightarrow t$) может служить как элементарное утверждение процесса верификации (см. \cite{chatterjee00}, \cite{nanevski08-1}).
В \cite{barendregt93} широко представлены доказательства всех важ\-ных теорем и  лемм  $\lambda$-вычисления в связи с определением его синтаксиса и семантики, в том числе, кратко о проблемах и ограничениях не типизированного и типизированного $\lambda$-вычислений. Статья \cite{plaisted85} пос\-вя\-ща\-ет\-ся вопросам теоретической  нерешимости редукции $\lambda$-термов в общности без явной типизации (например, в  вы\-чис\-ле\-ниях Карри, в отличие от  вычислений Чёрча). Например, терм $\lambda x.x x$ не может быть решён в общем, но при типизации становится очевидно, что данный терм не принадлежит регулярной  типизации. Если бы $\lambda x.x x$ был типизирован, то допустимый второй терм $x$ имел бы некоторый тип $t_2$, а первый $x$ обязательно должен иметь некоторый тип  функционала $t_0 \rightarrow t_1$, что невозможно приравнять к $t_2$ в типизированном $\lambda$-вычислении первого порядка, в котором  типы далее не специфицируются. Следовательно, $\lambda x.x x$  не типизируемый терм. Однако, можно сконструировать в качестве  бесконечно приближенного типа лимит \cite{macqueen84}. Этот тип соответствует рекурсивно определённому  аб\-стракт\-ному типу данных, который имеет в общем случае зависимость (возможно, ре\-кур\-си\-вную) к другим типам данных и представляет собой тип данных записи.

\textbf{Подход №5 -- Автоматическая редукция проблемы:}
\textit{Автоматическая ре\-дук\-ция проблемы}    пытается редуцировать строки программы, чтобы данная ошибка, либо иное поведение программы, например \textit{«результат функции возвращает 5»}, сохранялось. То есть, если некоторая программа вычисляет некоторую таблицу, используя различные формулы, а нас интересует только одна ячейка в ней (на\-при\-мер, единственная с ошибкой), то можно весь код удалить, который не нужен для вы\-чис\-ле\-ния той самой \textit{критической} ячейки. Таким образом, получается ре\-ду\-ци\-ро\-ван\-ная и более простая программа, которую можно легче анализировать.
Очень часто бывает, что ошибка в программе связана лишь с маленькой частью изначальной программы. Увы, особенно в больших или незнакомых программах, локализация программ занимает очень много лишнего времени.

Допустим, мы строим проект одной командной строкой. Вместо  «\textit{GNU make}» \cite{make16} ради простоты, используется сильно  упрощённая программа для пос\-тро\-ен\-ия любого ПО   «\textit{builder}» \cite{haberland15-4}. Для автоматизации и редукции входной программы используется программа  «\textit{shrinker}» \cite{haberland15-3}. Про\-грам\-ма работает циклически: строит ПО, анализирует или запускает ПО, сравнивает  наблюдаемое поведение программы с ожидаемым результатом. Если программа после сокращения по-прежнему стро\-ит\-ся, запускается и сравнение успешное, то редукция продолжается, иначе редукция про\-ва\-ли\-ва\-ет\-ся. Если редукция про\-ва\-ли\-ва\-ет\-ся, то, либо берётся другая редукция, либо прежняя (успешная), как окончательная.
Программа при запуске выдаёт всё, что важно на  «\textit{stdout}» и  «\textit{stderr}». Уговаривается по умолчанию, что любое другое поведение можно записывать в «\textit{stdout}» или «\textit{stderr}», в том числе, интересующее нас поведение программы (так называемый  «\textit{симптом}»), без ог\-ра\-ни\-че\-ния общности.
Если симптом всё ещё на\-блю\-да\-ет\-ся при запуске программы, то поведение  инвариант и мы «\textit{пробуем}» сокращать, далее пока не останется возможности.
Предложенный подход описывается на рисунке алгоритм №\ref{algo:AlgorithmProblemReduction}.
Представленный алгоритм наивный и не оптимальный, т.к. выбор $linesseq$ не де\-тер\-ми\-ни\-ро\-ван\-ный и может содержать любое количество строк, равно хотя бы одной или более строк.
Нужно заметить, что объём выбранных строк может расти, но также уменьшаться со временем, когда варианты редукции отклоняются несколько раз подряд, в зависимости от стратегии редукции.

\subsection*{Объектные вычисления}
  Объектные вычисления представляют собой формализм для различных вычислений с объектами. Вычисления могут в себя включать, например, совместимость типов, проверку верности условий. Объектные вычисления можно разделить на  два вида, по Абади-Карделли и по Абади-Лейно (см. рисунок \ref{fig:ObjectCalculi}).

Итак, зачем на самом деле нужна формализация объектов? Почему нельзя со\-по\-ста\-вить объекты простыми переменными?
%
% MEANING OF OBJECT CALCULI:
Отвечая на второй вопрос первым: да, можно сопоставить на самом деле.
 Однако, в связи с введением объ\-ек\-тов поля, больше абстрагируются. Вследствие этого, алгоритмы упро\-ща\-ют\-ся и обобщаются.
Класс объектного экземпляра, это прежде всего,  «\textit{абстрактный тип данных}» (АТД), а не только хранитель данных. АТД также является носителем множества объектов, предусмотренные для решения комплексной проблемы (см. \cite{leino02}). Операции над объектом замкнуты, т.е. сам объект не может распадаться на некоторые другие объекты при запуске операций. Поэтому программист, при написании об\-я\-зу\-ет\-ся, определять в объектах лишь те операции, которые к нему подходят и которые действуют только на внутреннем состоянии всех полей самого объекта -- поля других объектов не доступны. Это подразумевает разделение данных аналогично записям.
Операции объектов на практике, это методы классов, а объект, это объектный экземпляр. 

% OOA + OO-stuff
  \textit{Объектно-ориентированная парадигма} популярная, и успешно использовалась в ин\-дус\-трии на протяжении последних десятилетий, опираясь на АТД. Одним из пре\-им\-ущ\-еств АТД является --- возможность строить надёжное, гибкое, интуитивное и быстрое ПО за счёт технологии моделирования, как например «\textit{объ\-ек\-тно-\-ор\-и\-ен\-ти\-ро\-ван\-ное моделирование}» (ООМ) \cite{dsouza98}, \cite{rumbaugh91} или «\textit{паттерны}» \cite{kerievsky05}. Перечисленные методы и идея широко применяются в индустрии до ныне.
% patterns and OO-M
Нельзя недооценивать важность ис\-поль\-зо\-ва\-ния паттернов на практике, которые способствуют к существенному снижению ошибок при написании кода, благодаря четкому распределению ролей вовлечённых объектов. Паттерны прог\-рам\-ми\-ро\-ва\-ния являются  \textit{стереотипами}, т.е. эпистемологическим механизмом -  идиомой, которая способствует человеку «\textit{быстрее}» воспринимать во времени вовлечённые объекты и связывать их, для решения поставленной задачи, не читая весь программный код полностью. Многие классические, философские идеи и течения вовлечены в представление объектов, как, например  атомизм,  кон\-нек\-ти\-визм и т.д.  Паттерны можно интерпретировать, как сравнение объектов ООМ с ролями  \textit{актёров} сценария, для решения определённого сценария: каждый актёр в определённой ситуации ведёт себя так, как этого требует «манускрипт», точнее спецификация.  Прагматизм паттерна повышается, если сценарий можно объяснить лучше и понятнее.

% performance issues with objects/records
 Объектное вычисление может быть использовано для  семантических анализов, например для проверки типов или для верификации, что будет представлено чуть позже. Кроме того, оно может послужить анализу зависимости  полей данного объекта, например для эффективного  \textit{распределения процессорных регистров}. В частных случаях может иметь смысл, все объектные поля преобразовать в  массив или отдельные  локальные переменные, если граница массива известна и не велика \cite{bornat00} и размер объекта не велик.
При распределении регистров, зависимость данных на прямую определит какие  ре\-гис\-тры будут использованы (например, регистры общего пользования) или стековые поля. Быстрые ресурсы, т.е. процессорные регистры, в разы быстрее, чем доступ к операционной памяти, однако, они сильно ограничены по объёму и ёмкости. С другой стороны, процессорные регистры, в зависимости от архитектуры ЭВМ, могут, согласно   \textit{двоичному интерфейсу приложений} «\textit{ABI}» (см.  \cite{gcc15}, \cite{llvm15}, \cite{leroy11}), временно \textit{хранить} только несколько \textit{слов}. Слова ограничены и тем не менее, имеющийся \textit{ABI} обычно не делит общие блоки памяти представляющие объ\-ект\-ные эк\-зем\-пля\-ры потому, что при компиляции трудно предугадать точно, когда и где, какие поля будут или не будут использованы. Однако, если возможно было бы это уточнить, стоило бы модуль предсказывания «\textit{GCC}» в этом плане улучшить, т.к. на данный момент всегда принимается отрицательный прогноз. Реальность такая, что объекты остаются неделимыми и если памяти не остаётся, то всё выкладывается на стек, что тормозит и желает иметь лучшее.

% diff to record calculi
Вычисление записей \cite{ehrig80} отличается в основном от объектов тем, что они являются чистыми хранителями данных. Запись можно рассматривать, как  кортеж раз\-лич\-ных типов, что может быть полезно для описания вычислений одновременных процессов. Скотт \cite{scott76} характеризует любой тип данных (т.е. в обобщённом виде, в качестве кортежа) как  \textit{алгебраическую решётку}, которая имеет  \textit{инфимум} как изначальное  (возможно не инициализированное) значение и окончательный ре\-зуль\-тат (возможно параметризованный), как  \textit{супремум} и все состояния между обо\-ими  \textit{экстремальными точками} объясняются, как состояния кортежей, которые соблюдают  \textit{порядок вычисления}.
Для более подробного ознакомления применения и обоснования верификации с объектами, особенно в связи с нединамическими пе\-ре\-мен\-ны\-ми, можно ознакомиться в \cite{chatterjee00} и \cite{mueller02}.
% attributes modelling
Для дальнейшей дискуссии в качестве примера мы при\-ве\-дем работу \cite{barnett04}: условия ве\-ри\-фи\-ка\-ции задаются формулами  логики предикатов первого порядка.
 Атрибутные поля моделируются, либо от\-дель\-но, т.е. атрибуты выделяются в отдельные регионы памяти, либо встроено в содержащий объект, т.е. все поля и типы известны. Наи\-более комфортабельным вкладом работы можно считать предикаты  \texttt{pack} и  \texttt{unpack}, которые свёртывают, либо развёртывают предикат вручную.
Аналогично к \cite{barnett04} и \cite{chatterjee00}, рекомендуется мо\-де\-ли\-ро\-вать зависимые поля, как от\-дель\-ные зависимые объекты. Обе работы пред\-ла\-га\-ют ввод новых утверждений для инвариантов объектов, которые всегда верны до тех пор, пока некоторый объект жив. Никакой из упомянутых подходов не имеет отношения к  \textit{указателям} (см. глава \ref{chapter:DynMemProblems}).\\

Вычисление объектов классных экземпляров упоминаются только коротко. Этот вид вычисления доминирует современные  языки программирования, как, например,  Ява или   Си(++) \cite{gcc15}, \cite{isocpp14}. Класс используется для генерации объектных  эк\-зем\-пля\-ров при запуске программы. В этом разделе нас пока не интересует область выделенной памяти. Используются командные операторы  \texttt{new},  \texttt{malloc} и  \texttt{calloc}. После короткой характеризации, мы переходим к обыкновенному виду вычисления объектов.

Оба вида вычислений представляют возможность проверки типов и спе\-ци\-фи\-ка\-ции/\-ве\-ри\-фи\-ка\-ции программных операторов.  Абади и Карделли мотивируют пред\-ло\-жен\-ное вычисление необходимостью формализовать  вычисление с объектами из-за различных эффектов, как  \textit{подклассы},  \textit{полиморфизм} и  \textit{замкнутость/рекурсия объектных типов}. Оба автора ссылаются на систему  типизации  лямбда-вычислений как одним примером второго порядка. Нужно заметить, что  $\lambda$-термы, пред\-став\-ляю\-щие  объекты, могут быть использованы, как квантифицируемые типы (т.е. типовое множество классов, которое выражается знаком $\forall$, см. \cite{peirce10}, \cite{bruce02}) с помощью неопределённых символов. Далее, ре\-кур\-сив\-ные и нерекурсивные классы требуют уточнения структуры. Другими словами, «\textit{зависимые от новых параметров типы}», которые представляются аб\-страк\-ци\-ей $\lambda$-типов высшего порядка дополнительную абстракцию над  объ\-ек\-тны\-ми типами по Абади-Карделли.

Без ограничения модели памяти, например, в каком регионе  опе\-ра\-ци\-он\-ной памяти содержится объект, или, где объект расположен в соответствии ма\-тер\-ин\-ско\-го объекта: внутри и снаружи, любой объект по модели Абади-Карделли имеет указатель. На рисунке \ref{fig:ObjectPointsTo} $x$ указывает на объект, который имеет один внутренний объект, а он опять же имеет два содержания $abc$ и $def$ и некоторое содержимое $ghi$. Кроме того, имеется указатель $y$, который ссылается на объект содержавший $jkl$, при этом, объект $x$ содержит некоторое поле (или любое из подполей одного из полей), которое ссылается на содержимое от $y$, т.е. $jkl$.

Классная иерархия наследования представляет собой   \textit{упорядоченное множество},  инфимум чей по конвенции пустой класс, например  \texttt{Object} или \texttt{[]}. То есть, любые два класса из иерархии можно сравнить  с самим собой, два класса, которые ничего не связывает и не связаны между собой, т.е. имеются следующие сравнения порядка наследованности: $\sqsubseteq$, $\not\sqsubseteq$.

Абади и Карделли лишь предоставляют возможность наследования в своей модели вычисления, также как и  \textit{делегация} или \textit{внутреннее} представление  полей и  объектов. Использование ключевых слов  \texttt{self} и  \texttt{super} в  языке Cи++ \cite{isocpp14} допускается, однако, использование подлежит к заранее проводимому семантическому ан\-ал\-изу кон\-тек\-ста ссылаемого класса. Эта предосторожность сильно увеличивается, вместе с целым рядом других способов вычисления, при использовании об\-ык\-но\-вен\-но\-го  объектного вычисления.
Карделли замечает, что дополнительные накладки в связи c проверками актуального и ссылочного класса, например, при использовании  \texttt{\texttt{self}}, приводят к принудительным проверкам всей иерархии классов, которые можно было бы избежать, если контекст достаточно однозначен.
Нужно отметить, что на практике это замечание не столь важное потому, что определение  класса по  иерархии не является критической проблемой при ком\-пи\-ля\-ции программы. Поиск наследованности по иерархии слияния фактически является поиском по дереву, что можно сделать за счёт $\Theta(n)=log(n)$ в худшем случае. Это вполне достаточно.
Серьезным ограничением наследованности, является угадывание, какой именно метод будет запускаться. С вопросами прак\-ти\-чес\-кой и теоретической реализации наследованности множества классов на  языке Cи++ можно ознакомиться в работе Рамана \cite{raman11}.\\

Методы объектов содержат код, который доступен только объекту. Ради простоты, мы смоделируем автоматически выделенные методы, которые в отличие от  ста\-ти\-чес\-ких методов, означают, что методы доступны исключительно тогда, когда соответствующий объект был ранее выделен оператором  \texttt{new}. В случае статичных классов, нет необходимости доступа к заранее выделенной памяти объекта, т.к. переменные в стековом окно выделяются глобально независимо от объекта.
Далее, обозначим по умолчанию, что любой объект, если не уговорено по другому, всегда имеет один метод  «\textit{конструктор}» для построения объекта и инициализации нейтральным значением всех полей (что записывается в отдел  \texttt{.ctor} на  ассемблере, в случае языка  Cи++ \cite{isocpp14} при переводе) и один аналогичный метод утилизации объекта, который записывается в  «\texttt{.dtor}». Таким образом, можем определить следующие фазы объектов (см. рисунок \ref{fig:ObjectLifeCycles}), которые нам послужат в качестве спецификации:

Только что были обсуждены пункты 1 и 2. Пункт 3 касается любого состояния «\textit{немёртвого объекта}». Объект немёртвый, когда он «\textit{жив}», это когда прог\-рам\-мный оператор вызывается после момента  конструкции, но до  утилизации объ\-ек\-та. Пункты 4 и 5 подразумевают ту самую спецификацию, как у процедур. Ради применимости пытаемся избежать принудительную спе\-ци\-фи\-ка\-цию всей программы. Обозначим это «\textit{лёгким}» подходом, но пока что не будем рассматривать такой вариант. В целях увеличения гибкости и  локализации проблем, уговаривается по умолчанию, что любые «\textit{дополнительные}» ут\-вер\-жде\-ния разрешаются, которые также могут быть упущены (пункт 6).

Методы в обоих вычислениях изначально не ограничиваются. Однако, не ог\-ра\-ни\-чен\-ные методы также могут включать в себя изменение методов в любой момент запуска программы, удаление существующих методов и добавление но\-вых методов. Очевидно, что при неограниченном вычислении  спецификация ста\-но\-вит\-ся всё сложнее, и возможные изменения трудно, или даже теоретически и прак\-ти\-чес\-ки не решимы. Следовательно, максимальная возможность приводит к жест\-кому ограничению верификации.

Как было упомянуто, неограниченные методы не дают и не расширяют с тео\-ре\-ти\-чес\-кой точки зрения  «\textit{Тьюринг-вычислимость}», изменения методов при запуске («\textit{динамические методы}») и умолчании далее не рассматриваются. Когда речь идёт о динамических методах, только тогда, они могут рассматриваться.

Системы объектных вычислений имеют следующее логическое суждение:

$$E \vdash obj : A :: B$$

где, $E$ объектная среда, содержавшая существующие объекты в момент актуального состояния вычисления программы, $obj$ является объектным представлением (т.е. переменной в вычислении №1, а целым объектом в случае вычисления №2). $A$ является типом $obj$ (т.е. наименование класса объектного экземпляра $obj$). $B$ является спецификацией объекта $obj$, которая содержит, например,  состояние объекта. Если вместо $obj$ проводится вызов метода $obj.m()$, то $B$ содержит состояние $obj$ до и после вызова $m$, но $B$ также может далее содержать утверждения из рисунка \ref{fig:ObjectLifeCycles}, которые не противоречат данной фазе программного запуска. На результаты до и после запуска программного оператора или метода, можно ссылаться с помощью встроенных вспомогательных операторов и  регистра специального назначения, ко\-торый содержит, либо весь объект (в случае №2), либо объектное поле. В центре объектного вычисления стоит объектное суж\-де\-ние вместе с типом и утверждением.
Отметим, что $B$ может меняться и при этом $A$, т.е. типизация, остаётся без изменения. Это становится актуальной проблемой, когда  сигнатура/тип дальнего метода знакома, но само состояние нет. Это тот случай, когда используется архитектура веб-сервера или драй\-ве\-ра, когда известно как метод должен выглядеть, но детали функциональности не (полностью) известны. 

Отметим, что обе модели,  Абади-Карделли и  Абади-Лейно, не поддерживают изначально ссылочные указатели на любые объекты, которые были бы выделены во время запуска программы. Указатели будут рассматриваться позже, однако, уже сейчас обсудим, что избегая лишние операции копирования, можно существенно улучшить быстродействие, если гарантированно, что некоторые объекты разрушаются и эти объекты не будут использоваться далее, или все копии объекта остаются без пользы.

 Классовые типы $T$ определяются рекурсивно. $Т$ Тип --- либо целое число, либо состоит из  классов. Классовый объект содержит поля $f_i$ и методы $m_j$ и все наименования отличаются друг от друга согласно конвенции. Все поля имеют некоторый тип $T$. Методы имеют $T_j \rightarrow T_{j+1} \rightarrow \cdots \rightarrow T_k$ в качестве типа сигнатуры.\\

Далее рассматривается вычисление №2, объектное вычисление по Абади-Лейно. 

Язык программирования  «\textit{Baby Modula 3}» \cite{abadi93} был предложен Абади и ком\-па\-ни\-ей «\textit{Di\-gi\-tal Equ\-ip\-ment Cor\-po\-ra\-ti\-on}» в качестве экспериментального модулярного языка. Он содержит минимальный набор операторов языка  \textit{Modula 2}, который очень близок к синтаксису и семантике  языка Паскаль, с дополнениями для про\-то\-ти\-пи\-за\-ции языка объектного вычисления обыкновенного вида. Язык похож на предложенные и далее расследованные языки Абади, Карделли, Лейно и содержит все необходимые единицы для  $\mu$-рекурсивных схем: присвоение и рекурсию. С более подробными дискуссиями на тему модели объектов и выразимость объектов можно ознакомиться в \cite{gunter94}, где имеются чуть устарелые, но значимые статьи по данному вопросу.
Целью ввода иного вычисления является альтернативное представление фор\-ма\-ли\-за\-ции. Однако, имеется эквивалентность обоих видов, которую можно доказать методом  «\textit{пол\-ной аб\-страк\-ции}» \cite{mitchell96}, \cite{plotkin77}, \cite{cohn83}, \cite{honda05}, при которой необходимо доказать эк\-ви\-ва\-лент\-ность  денотационной и  операционной семантик \cite{allison89}, \cite{lavrov01}, \cite{abramsky94}, \cite{winskel93}, \cite{tennent76}, \cite{bird97}, ка\-са\-тель\-но наблюдаемого поведения объектов. Доказательство изложено в \cite{reus02}, \cite{reus05} и следует отметить, что в целях простоты модель Абади-Карделли описывается в разы проще. Кроме того, попытка доказать полную абстракцию увенчается успехом, но за счёт очень сложной операционной семантики в объектном вычислении. Это можно обосновать тем, что возможные рекурсивные типы гораздо проще вывести, если рекурсия содержится лишь в наименованиях типов, т.е. классов, чем определяется рекурсией объектных эк\-зем\-пляр\-ов.

Синтаксис языка «\textit{Baby Modula 3}» определяется с помощью термовых выражений $a$ как определено на рисунке \ref{fig:ObjectTermsAbadiLeino}, где $A$ является типовым выражением.

Термовое выражение эквивалентно к выражениям языков  «\textit{Modula}» или  «\textit{Паскаль}» и не нуждается в дополнительных объяснениях, кроме: fun($x:A$)$b$ определяет  ано\-ним\-ную процедуру (эквивалентно к $\lambda$-абстракциям) с одним входным параметром $x$ с типом $A$. В общем разрешаются анонимные процедуры с нуля или более параметрами. Переменная $x$ является в теле процедуры $b$ связанной переменной. Вызов процедур, например $b(a)$, подразумевает, что типы последующих аргументов совпадают с типами декларации процедуры. Пустое значение \texttt{nil}  совместимо со всеми другими типами указателей. Например, не инициализированные объектные поля равны \texttt{nil}. Объектное расширение добавляет поля с указанным значением к предлагаемому объекту, таким образом, что объекту добавляется новое поле. Методы также расширяемы. Метод может обновляться, т.е. код метода может также как и переменная меняться во время запуска программы.  \texttt{wrong} является резервированным ключевым значением, который показывает на неверное вычисление. Данным определением объекта всегда является вектор, который состоит из лексикографически  упорядоченного множества пар из наименований полей/методов вместе с содержанием. Содержание опять же может быть объектом, которое определяется аналогично определению из рисунка \ref{fig:ObjectTermsAbadiLeino}.

Семантика  «\textit{Baby Modula 3}» определяется аксиоматически и проводится двумя шагами: сначала проводится  проверка типов (один набор правил), затем верификация с другим набором правил. Проверка типов объектов также требует определение на подкласс, которое может проводиться по-компонентно: $A<:B$ тогда и только тогда, когда $B$ содержит все поля и методы с наименованиями как в $A$. Из-за отсутствия классовых идентификаторов в объектном вычислении, рекурсивные объекты необходимо симулировать подобно случаю  «\textit{комбинаторам плавающей точки}» в $\lambda$-вычислениях, которые имеют похожее ограничение. Итак, смешано-рекурсивные объектные оп\-ре\-де\-ле\-ния выявляются с помощью искусственного комбинатора $\mu(X)$ для любого объекта $X$ к данному объектному выражению согласно рисунку \ref{fig:ObjectTermsAbadiLeino}, который приписывается после $\mu(X) \ \tau$, таким же образом, как аппликация $\lambda$-термов, например, $\lambda x.x \tau$.

Следовательно, для проверок объектов имеются: переменные среды,  подтипы (в том числе подклассы) и общие структурные правила. Ими производятся проверка типов и верификация объектов (см. позже, ср. минимальную логику вычислимости «\textit{LCF}» \cite{plotkin77}). Ради простоты (см. \cite{leino98}) часто ссылаются на структурные операционные семантики \cite{plotkin81}, хотя, как было упомянуто ранее, преобразование в/из денотационной семантики(-e) в общем случае остается тяжелой проблемой.

Однако, $\mu$-определения объектов --- существенная проблема, т.к. для одних и тех же начальных данных вывод может быть различным, т.е.  некорректность очевидна. Приведем в качестве примера правила из рисунка \ref{UnsoundRules1}.

То есть, правила с рекурсивными определениями могут быть подорваны и не\-до\-ка\-зу\-емы из-за возможных расходящихся доказательств, при этом, лишь тривиальное решение, т.е. неподвижная точка, может действовать единственно «\textit{надёжным}» путём. Такая обстановка, увы, не удовлетворительна. Это как раз и есть причина, почему объекты в вычислениях Абади-Лейно в общем не могут быть  типизированы в вычислении типов первого и второго порядка. Разрешить эту проблему возможно, если правила верификации и проверки типов обогатить дополнительными правилами, которые сильно зависят от контекста. Это резко ухудшит простоту правил.
По Абади-Лейно классы как таковы, не существуют. Объекты могут иметь произвольное число полей других объектов, следовательно, классы могут зависеть от других классов, включая от  собственного класса. Это поражает тип третьего порядка, а вычисление эквивалентно $\lambda^{\rightarrow}_3$ (см. опр.\ref{def:TypeChecking}). Поэтому, сходимость и замкнутость по типу имеет большое значение для решения при статическом анализе классового типа \cite{abadi93}. Если сходимость конечная и процесс верификации пе\-ре\-чис\-ля\-ем, то таким образом, можно сложить все перечисленные типы вместе, в одну  запись. Таким образом, новой записью является тип объединяющий все  зависимые типы.
Авторы \cite{abadi93} обращают внимание на то, что любые операции полученные объединением типа замкнуты и все полученные (под)-типы также могут ис\-поль\-зо\-вать\-ся как базисный объектный тип. То есть, полученный объектный тип является  идеалом, опираясь на сравнение объектов  «$<:$», см. позже.\\

Язык программирования представленный в \cite{abadi97} (см. рисунок \ref{fig:ObjectTermsAbadiLeinoSimplified}) является упрощением языкового представления рисунка \ref{fig:ObjectTermsAbadiLeino}, однако, присвоение отсутствует, а вместо этого имеется символьное присвоение и некоторый дополнительный  \textit{син\-так\-си\-чес\-кий сахар}  условного перехода. Далее, динамические методы за\-пре\-ща\-ют\-ся.  Синтаксис оп\-ре\-де\-ля\-ет\-ся следующим образом:

  Переименование во время запуска разрешается, если оно соблюдается вызывающей и вызванной сторонами.  Локальные переменные в процедурах за\-пре\-ща\-ют\-ся. Также запрещаются параметры методов, которые мен\-я\-ют\-ся внутри методов. Для  выразимости это не является ограничением, т.к. это всегда можно симулировать с вводом дополнительных полей. Особенность  присвоений за\-клю\-ча\-ет\-ся в том, что присвоения полей возвращают в качестве результата объект, в котором некоторому существующему полю было присвоено значение. Важным ограничением в отличие от \cite{abadi93} является по Абади-Лейно в частности, исключение рекурсивно-определённых объ\-ек\-тов.

Отношение подтипов определяется в рисунке \ref{DefinitionClassSubtype}.

 где $A'$ это некоторый класс, а $A$ соответствующий подкласс.

Соотношение $<:$ позволяет вместе с другими правилами проводить  проверку типов, например, самое главное --- это  \textit{правило объектного построения}:

\begin{center}
\begin{tabular}{c}
\inference[(CONS)]{A \equiv [f_i::A_i^{i=1..n},m_j::B_j^{j=1..m}]\\ E \vdash \diamond \quad E \vdash x_i::A_i^{i=1..n}\\ E,y_j::A \vdash b_j::B_j^{j=1..m}}{E \vdash [f_i=x_i^{i=1..n},m_j=\psi (y_j)b_j^{j=1..m}]:A}
\end{tabular}
\end{center}

 где $A$ некоторый  объектный тип (не класс), $E$ это объектная среда, $x_i$ некоторые значения полей, а $b_j$ некоторые значения методов, т.е. тела процедур.
Проверка верности программы по упрощенному Абади-Лейно проводится согласно следующему суждению:

\begin{center}
  $\sigma,S \vdash b \rightsquigarrow r,\sigma'$
\end{center}

где $\sigma$ является  начальным состоянием стека, $S$ описывает состояние  стека, $b$ данная последовательность программных операторов («\textit{программа}»), $r$ является результатом, который сохраняется в  \textit{виртуальном регистре} неограниченной ёмкости, и $\sigma'$ описывает финальное  состояние памяти.
 Спецификации  объектов определяются рекурсивно по компонентам:

$$[f_i:A_i^{i=1..n},m_j:\psi (y_j)B_j::T_j^{j=1..m}] ,$$

где $A_i,B_j$ это спецификации, $y_j$ параметры  анонимных процедур $\psi$ используемые в $B_j,T_j$, а $T_j$ определяет спецификации перехода памяти из одного состояния в следующее.

  Суждение спецификации определяется по $E \Vdash a:A::T$, где $E$ это объектная среда, $а$ программа, $A$ тип, $T$ переходное описание. То есть, проверка типа по $A$ и переходные состояния записаны в $B$ и применяются одновременно, если даже тех\-ни\-чес\-ки по разным этапам, всё равно не трудно увидеть, что $B$ окажется сильно раздутым и незначительное изменение памяти может привести к большим последствиям. Пред\-ло\-жен\-ное  соотношение $<:$ \cite{abadi97} аналогично применяется к спецификациям. Аналогичные проблемы распространяются на спецификации. Например, требуется доказать очевидную программу:
$\emptyset \Vdash ([f=\texttt{false}].f:=\texttt{true}).f:\texttt{Bool}::(r=\texttt{true})$.
 Дерево доказательства с аннотациями использованных правил находится в рисунке \ref{ExampleProofDerivation}.

Определения использованных правил находятся в рисунке \ref{ExampleRulesetObjectCalculus}.

Утверждение $\diamond$ обозначает истину.
 Корректность представленного набора правил можно прочитать подробнее в статье \cite{abadi93} и в её сопровождающей технической статье. В подходе  Абади-Лейно не\-об\-хо\-ди\-мо отметить, что ситуация с  абстракцией желает иметь лучшее. Па\-ра\-мет\-ры в методах видны только снаружи во внутрь, но не наоборот, например в $b_1 \equiv \texttt{let }x=(\texttt{let }y=\texttt{true in }[m=\psi (z)y]) \texttt{ in }x.m$ внутренний $y$ не может быть проверен снаружи, согласно правилам из \cite{abadi93} и разница между предикатами (какого именно порядка и ограничения) и  утверждениями всё-таки не достаточно ясна. Абади и Лейно предлагают в качестве дальнейшей работы следующее: улучшение абстракции спецификации, исследование меняющихся параметров, а в связи с этим, вопросы о полноте, использовании указателей на объ\-ек\-ты и на поля объектов, а также более детально разобраться с рекурсивно-определёнными объектами. На данный момент не вычисление по Абади-Карделли и по Абади-Лейно не обращают внимание на указатели или  \textit{псевдонимы} (см. позже).

Оба вида вычислений страдают от того, что отсутствует поддержка указателей. В \cite{cardelli96} перечислены важные и нужные расширения: поддержка па\-рал\-лель\-нос\-ти, возможность ссылаться на адресные  поля и использования абстракции в спецификациях, так как впервые это уже заметил Хор \cite{hoare69}.

%%%%%%%%%%%%%%%%%%%%%%%%%%%%%%%%%%%%%%%%%%%%%%%%%%%%%%%%%%%%%%%%%%%%%%%%%%%%%%%%%%%%%%%%%%%%%%%%%%%%%% 

Бан\'{е}ри\`{й} \cite{banerjee08} представляет язык, который поддерживает объекты, разместившиеся в  стеке и в том же регионе памяти (см. \cite{tofte97}). Подход в \cite{banerjee08} перемещает все  локальные переменные в стек, но  висячие указатели по определению языка не допускаются. Рекурсивные предикаты над объектами запрещены. Особенность  глобальных ин\-ва\-риан\-тов заключается в не меняющихся зависимостях между объек\-та\-ми. Он предупреждает о сильно возрастающей проблеме абстракции и поддерживает ин\-и\-циа\-ти\-ву по-объектного взгляда индивидуальных предикатов.

%%%%%%%%%%%%%%%%%%%%%%%%%%%%%%%%%%%%%%%%%%%%%%%%%%%%%%%%%%%%%%%%%%%%%%

\textbf{Программные нити}.\\
В \cite{abadi97}  объекты квалифицируются как  абстрактные типы данных, но только не как обыкновенные вычисления записей \cite{ehrig80}. Отличие определяется в первую очередь не по имеющимся полям, а по структуре данных. Записи имеют любые зависимости к иным записям.

Переходы из одного состояния записи в другое, можно получить с помощью при\-ме\-не\-ния правила редукции. Таким образом,  \textit{трансформации графа} можно рас\-сма\-три\-вать как  \textit{редукции с записями}. Эриг и Роузен предлагают подкласс процедур, который безопасен тем, что он не смешивает зависимые поля объектов вместе с анализом с помощью теории категории. Их подход можно считать «\textit{безопасным}» при запуске нескольких нитей одновременно, если доступ к подмножеству полей объекта не зависит друг от друга. Таким образом, авторы дают возможность к более эф\-фек\-тив\-но\-му доступу полей, без необходимости полной  блокировки при много-поточных программах. Поэтому, процедуры порождения и  утилизации, при соблюдении условия, могут запускаться одновременно без изменения по\-ве\-де\-ния и  корректно. Похожие подходы к доступу объектов представляются более подробнее  в \cite{jones11}.

Хойсман \cite{huisman07} задаётся также вопросом параллельного доступа к объектам, также как и Хурлин \cite{hurlin09} и \cite{ehrig80}. Хойсман пытается с помощью «\textit{кон\-трак\-тов между объ\-ект\-ным взаимодействием}»  дать возможность не только определить полный или никакой доступ к  методу классу, но также подключить  временные условия доступа. Хойсман предлагает подключать в спецификацию историю объектов  и вызовы методов. Также предлагается объединить описание различных объектов при отдельных  сте\-рео\-ти\-пах и совместных действиях объектов, т.к. предложенный вариант видимо страдает от необходимого уровня  абстракции, в первую очередь, для параллельного запуска методов объектов.

\textbf{Фреймворки}.\\
Для достижения цели проверки спецификаций целых программных комплексов, представляются иные  формальные подходы (см. опр.\ref{def:FormalProof}). Один из наиболее широко известных и популярных подходов, является подход Мейера  «\textit{о делении ответственности}»  \cite{meyer98}. Мейер \cite{meyer92} предлагает, как и многие до него, методологию  «\textit{модулярный подход}», аналогично принципу Лискова и принципу скрытия по  Парнасу. В алгоритмах связанных между собой,  актёрам ответственности приписываются  стереотипы, согласно которым, выполняются ранее определённые роли. Определение ролей в соответствующей  \textit{он\-то\-ло\-гии} участников и взаимосвязей позволит лучше охватить и по\-ни\-мать объ\-ект\-ную структуру.
Принцип Мейера можно вкратце охарактеризовать как:

\begin{center}
 Чем сильнее деление ответственности [каждого из единиц онтологии], тем более таким объектам, можно доверять.
\end{center}

Замысел заключается в том, что, чем меньше становится программа, тем сильнее редуцируется её  сложность, если изначальную программу не менять (кроме сокращения). Например,  метрика, в том числе по  Мак-Кейбу, вы\-чис\-ля\-ет показательное число, которое оценивает сложность  графа управления данной программы. Чем сильнее разбита программа на маленькие легко-контролируемые и логически связанные между собой единицы, тем проще получается граф потока управлений. Метрики можно сравнивать, т.е. и сложности двух программ.
Заметим, что объект, который кажется на первый взгляд простым, может в реальности иметь зависимости c любыми другими объектами, которые заранее неизвестны. Такие «\textit{неожиданные зависимости}» часто на первый взгляд скрыты и их трудно обнаружить. Неожиданные связи уже являются признаком тому, что в отличие от данной спецификации, данный объект вероятно содержит ошибку или серьезную не выявленную проблему. Неожиданные связи соединяют часто отдаленные объекты между собой, что для простых моделей и алгоритмов маловероятно и приводят к ситуации, когда два объекта связаны между собой, хотя на самом деле не должны быть связаны.

Рилэ \cite{riehle08} предлагает вводить тесты, покрывающие  спецификации для об\-на\-ру\-же\-ния отклонений объектов и взаимосвязей между объектами, так называемые  «\textit{коллаборации}».
Проверяются взаимосвязи между объектными модулями с по\-мо\-щью спецификаций, например,  «\textit{модульной алгеброй}» \cite{bergstra90} или  «\textit{компонентной алгеброй}» \cite{feijs02}, \cite{sifakis05}, \cite{tonella01}.
Эти подходы формализованы и описывают интерфейсы между классовым объектным  и иными экземплярами.
Описание интерфейсов, увы, ог\-ра\-ни\-чи\-ва\-ет\-ся лишь входными и выходными параметрами методов.
В \cite{fischer00} пред\-ла\-га\-ет\-ся подход навигации по объектам, согласно данной спецификации. В \cite{bergstra90}, \cite{feijs02} и \cite{assman06} предлагается добавлять последовательность вызовов методов в качестве утверждения, для избежания ошибок в связи с не правильным порядком вызовов методов. Это предложение похоже на подход Хойсманой.
Далее, программы, часто страдают от того, что их необходимо спе\-ци\-фи\-ци\-ро\-вать  целиком. Это очень неудобно, когда речь идет о средних и больших программах. Фактически, все представленные подходы исключают  указателей из  языка программирования, т.к. любые утверждения о программе могут просто ока\-за\-ться ложными.
Порядок вызовов методов может оказаться практически значимым не только внутри  одного объекта, но а также для целой сети объектов, например: исключение возможности вызова метода $m_2$ первым объектом, до тех пор, пока не будет вызван метод $m_1$ третьего объекта.
Программисту важно помнить в какой момент \textit{безопасно} вызывать метод, а когда требуются инициализации или иные вызовы и в каком порядке.
Это нужно рассматривать как постоянную «\textit{классическую}» проблему при изучении любого нового  фреймворка. Вместо де\-таль\-ной и длинной спецификации с прак\-ти\-чес\-кой точки зрения, хорошо бы иметь спе\-ци\-фи\-ка\-цию, которая статически обнаруживает соответствующие зависимости между объектами и порядком загрузки.

Язык объектных ограничений  «\textit{OCL}» \cite{oclspec}  является расширением языка  «\textit{UM\-L}». Он графический и текстовой де-факто стандарт по статическому и динамическому мо\-де\-ли\-ро\-ва\-нию ПО. «\textit{OCL}» позволяет добавлять формулы взаимосвязей, например, ограничить связь типов связанных между собой объектов. Формулы описывают в основном  лямбда-термы второго порядка предикатной логики (см. опр.\ref{def:TypedLambda2ndOrder}), т.е. с атомными типами. Кроме ранее упомянутых ограничений в разделах, остаётся отсутствие возможности спе\-ци\-фи\-ци\-ро\-вать указатели.

Сафонов \cite{safonov10} представляет и анализирует целый ряд примеров из области ком\-пи\-ля\-ции и связанные с ней проблемы. Цель анализа заключается в постановлении критерий  \textit{доверительных компиляторов}. Необходимо заметить, что автор выбирает компиляторы преднамеренно, т.к. они являются по сложности одними из наиболее сложными ПО. В этом можно убедиться, например, с помощью тщательного анализа кода, особенно переходы и ответвления, например используя метрики. Сложность проверки для данного ком\-пи\-ля\-то\-ра заключается в том, как определить, что любая преобразованная программа верна и оптимальна? Для получения этого, Сафонов ссылается на метод Мейера о \textit{делении ответственности}. В качестве контр-аргумента можно привести, например подход Леруа из проекта «\textit{CompCert}» \cite{bertot06}, который также исследует корректность, но также частично быстродействие. Увы, подход не имеет, большого прак\-ти\-чес\-ко\-го значения из-за слишком резких практических ограничений. --- По\-эт\-ому, по Сафонову, можно выявить следующие критерии значимые на практике и которые решают уровень приемлемости:

\begin{itemize}
 \item Верный и понятный диагноз ошибки, включая  генерацию контр-примеров.
 \item  Трансформация моделей, включая объекты, предпочитается переход между моделями из-за относительно малых затрат спецификации.
 \item Возможность расширять и модифицировать  объектные модели вычисления.
 \item Использовать формальные описания, если применимость от этого не страдает. В частности, Сафонов считает \textit{полноту покрытия} формальной модели и  \textit{ло\-каль\-ность} спецификации важными критериями.
 \item  Абстрактные типы данных описывают собственные инварианты объектов.
 \item Промежуточные представления должны быть достаточно гибкими и аб\-стра\-ги\-ро\-ван\-ны\-ми так, чтобы они могли быть использованы на различных этапах верификации.
\end{itemize}

%%%%%%%%%%%%%%%%%%%%%%%%%%%%%%%%%%%%%%%%%%%%%%%%%%%%%%%%%%%%%%%%%%%%%%%%%%%%%%%%%%%%%%%%%%%%%%%%%%%%%%%%%%%%%%%%%%%%%%%%%%%%%%%

\subsection*{Модели динамических куч}

\subsubsection*{Преобразование в стек}

В разделе \ref{sect:TheoryOfObjects} в связи с объектными вычислениями уже велась небольшая дискуссия о преобразовании в стек.

В \cite{tofte97}, \cite{tofte94}, \cite{regionmem10} расследуется  функциональный подход, в котором все переменные перемещаются в  стековое окно (в \cite{calcagno01} приводится формальное доказательство корректности подхода Тофти и Тальпина \cite{tofte97}). Реализован подход Гросманов диалектом Си «\textit{Cyclone}» при поддержке языка «\textit{ML}» (также см. \cite{grossman02}). Регоин определяется как совокупность связанных элементов динамической памяти, которые, например, при удалении полностью могут утилизироваться в один раз. Стековое преобразование, в частности, должно следить за диапазоном видимости, чтобы исключить нечаянного уничтожения, согласно автоматическому выделению и утилизации стека. В \cite{fluet06} обсуждаются диапазоны видимости за пределами стекового окна и демонстрируется расширение в системе «\textit{Cyclone}».
При перемещении переменных, необходимо следить за  интервалами годности, иначе полученные  «псевдонимы» и «возможно-псевдоним» могут стать ложными. Функциональный подход исключает  ди\-на\-ми\-чес\-кие списки как параметры, рекурсивные структуры данных и процедуры, которые возвращают  функционал. Все типы функций должны быть определены при компиляции. Мейер \cite{meyer1-03}, \cite{meyer2-03} считает, что кроме корректности программы, важ\-ным является быстрый  сбор мусора (см. \cite{larson77}). Поэтому он предлагает, по-возможности, цел\-и\-ком избавляться от динамической памяти и переоформить ал\-го\-рит\-мы для работы только над стеком. Как уже было предложено в \cite{tofte97},  вталкивание в стековое окно не занимает большие за\-тра\-ты при минимальных объёмах. Он требует, что к абстракции объектов необходимо уделять больше внимания и предлагает ис\-поль\-зо\-ва\-ние вспо\-мо\-га\-тель\-ных переменных для уменьшения спе\-ци\-фи\-ка\-ций.  Корректность присвоений в статье полностью не доказана, а только для полного класса данных, как утвер\-жда\-ет\-ся в статье. На самом деле,  присвоение может оказаться ложным, когда, например, допускаются изменения программного кода во время запуска программы.
Доступ к содержимому в стеке может быть быстрее, чем в динамической памяти, но всегда требует дополнительные затраты на вталкивание и выталкивание в стек и из стека, которые, в зависимости обсуждаемого алгоритма, могут резко снизить общее быстродействие так, что использование динамической памяти может быть быстрее \cite{appel87}.

\subsubsection*{Анализ образов}

Главной целью  анализа образов \cite{sagiv02, nielson99, pavlu10, distefano06, shapeanal10} является выявление  инвариантов в  кучах памяти при запуске программы, например, для об\-на\-ру\-же\-ния  псевдонимов. Выявляются фигуры описывающие инвариантную и гибкую часть динамической памяти \cite{shapeanal10}, например, для линейного списка \cite{distefano06} или двусвязного списка \cite{cherem07}.
 Граф зависимостей между образами всегда описывается полностью с помощью  \textit{трансферных функций}, таких как: пустой образ «\textit{пуст}», при\-сво\-ен\-ие  полю или указателю и выделение новой  динамической памяти. \cite{sagiv02} вводит основные соотношения между двумя  указателями:  «\textit{псевдонимы}»,  «\textit{не псевдонимы}» и «\textit{возможно псевдонимы}».

\cite{pavlu10} замечает, что \cite{sagiv02} и \cite{nielson99} могут привести в зависимости от данной прог\-рам\-мы, к не точному, а следовательно, к не правильному выводу, если для «\textit{если-тогда}»-команды в одном случае вычисляется «\textit{возможно псевдонимы}», а в аль\-тер\-на\-тив\-ном случае «\textit{псевдонимы}», тогда результатом вычисления послужит «\textit{псев\-до\-ни\-мы}», хотя правильный ответ «\textit{возможно псевдонимы}».

Более того, \cite{pavlu10} содержит подробное сравнение подходов \cite{sagiv02} и \cite{nielson99}. \cite{pavlu10} оценивает \cite{nielson99}, как более точный метод и предлагает оптимизацию путей по  графам образов с одинаковыми началами и псевдонимами, в качестве более эффективного описания, которое одновременно и короткое. Предложенная оптимизация имеет верхнюю сложность $\Theta(n^2)$ и сокращает вычисление в среднем, при\-мер\-но на 90\%. Павлу предлагает симулировать более сложный анализ псевдонимов за пределами процедур, преобразуя, насколько это возможно, вызовы процедур и  гло\-баль\-ные переменные в локально  интра-процедуральные элементы через пе\-ре\-и\-ме\-но\-ва\-ние. Подход в его работе был, на самом деле, уже ранее успешно применен в системе «\textit{GCC}» для решения комплекса иных проблем. Павлу, так же как и я, считает, что контекст-независимые подходы в точном анализе псевдонимов, мало\-пер\-спек\-тив\-ны (ср. \cite{hind01}), сравнив характеристики. Далее он предлагает оптимизацию отделения объе\-ди\-ня\-ю\-щих вер\-шин, которые образуют подграфы, от вершин, которые твердо не содержат псевдонимы.

\cite{parduhn08} представляет собой среду для визуализации образов куч для бы\-стро\-го ана\-ли\-за и обнаружения  инвариантов, редко используемых связей между об\-ра\-за\-ми. Для навигации по графу используются операции  «\textit{абстракция графа}» и  «\textit{конкретизация графа}» с помощью чего, удаётся подграфы  свёртывать и  развёртывать. Ви\-зу\-а\-ли\-за\-ция  трансферных функций, которая приводит к развёртыванию/\-свёртыванию более од\-но\-го  подграфа, желает быть лучше, однако, ограничение принципиальное и сле\-до\-ва\-те\-льно улучшение не ожидается.

\cite{calcagno09} представляет собой подход, основанный на  распределении куч по образцам. Подход сравнивает предположительно подходящие начала правил по длине и вы\-би\-ра\-ет наиболее длинное правило первым. Особенность подхода за\-клю\-ча\-ет\-ся в  аппроксимации обоих сторон рассматриваемых правил для выбора принимаемых правил с помощью ограниченной  абдукции.

\subsubsection*{Ротация указателей}

\cite{suzuki82} предлагает  «\textit{ротацию указателей}», которую можно считать безопасной, если: (1) содержание куч не меняется, (2) все элементы ротации годны до и после ротации (3) количество переменных не меняется. Пре\-иму\-щес\-твом безопасной ротации можно считать отсутствие нужды сбора мусора, который определяется просто, но также эффективные операции над списками, как, например, копирование. Хотя ротация указателей в явной спе\-ци\-фи\-ка\-ции не нуждается, всё-таки минимальное изменение параметров может привести к трудно-прогнозируемому поведению программы, особенно часто, если указатели нео\-жи\-дан\-но оказываются  псевдонимами. \cite{suzuki82} предлагает базисный объём ротаций, однако, на практике этого далеко не достаточно, в следствии чего, необходимо компонировать индивидуальные ротации из базисных.
Далее, ротации очень чувствительны к минимальным изучениям вызовов. Например, поменять местами аргументы, на первый взгляд --- безопасная операция, но она может привести к полной утилизации всех входных списков. Неожиданность также возникает часто с безобидными данными, например, когда один входной список пуст или ротации комбинируются.

Ротация может быть представлена как пермутация, где все её компоненты $x_j$ являются указателями:

\[
  \begin{pmatrix}
   x_1 & x_2 & x_3 & \cdots x_{n-1} & x_n\\
   x_2 & x_3 & x_4 & \cdots x_n & x_1
  \end{pmatrix}
\]

К примеру рассмотрим из рисунка \ref{ExamplePointerRotationSuzuki} модифицированный пример ротации из \cite{suzuki82}.

Это соответствует линейному списку с тремя элементами следующей сложной трансформации, которая состоит из шагов (№1--№4), см. рисунок \ref{HeapWatchOnProgrammExecutionPointerRotation}.

В зависимости от входных данных, можно выявить отдельные свойства, которые иногда тяжело обобщать, но иногда просто, как на примере \texttt{rotate(x,x,y)} является тождественным отображением, т.к. \texttt{tmp:=x; x:=y; y:=x; x:=temp;}. Очевидно, даже маленькая модификация может привести к негодности свойств, а следовательно свойства симметрии могут быть нарушены, а следовательно доказуемые свойства больше неверны, например, \texttt{rotate(x,y,x\^.next)}, т.к. $x$ и $x^{\wedge}.next$ связаны и пересечение между $y$ и $x^{\wedge}.next$ не полностью исключено априори.

%%%%%%%%%%%%%%%%%%%%%%%%%%%

\subsubsection*{Файловая система}

Граф кучи отображается в файловую систему так, что вершинами являются файлы, а зависимости между ними --- ярлыки (например с помощью команды \texttt{ln -s} под  Линуксом). Папки могут быть использованы как хранитель объектного экземпляра. Содержимое вершины находится внутри файла. Таким образом, файловая структура имитирует графовую кучу. Теперь можно описать схему файловой системы, не отслеживая символьные ссылки. Полученный перечень имеет структуру дерева и может быть записан в качестве слабоструктурированной единицы, например, в  XML-документ. Под Линуксом все операционные средства преобразованы в виде файла. Операции над файлами теперь соответствуют операциям над кучами. Задача ве\-ри\-фи\-ка\-ции теперь может быть задана как проверка данного состояния файловой системы. С другой стороны, операции над файловой системой можно проверить на корректность с помощью (например,  XML-) схемы, что явно может упростить ошибки или дефекты с большими и многими файлами. Во избежание траты времени на сравнение многих и больших файлов и папок, можно ис\-поль\-зо\-вать  команду \texttt{sha1sum} или \texttt{md5sum}, т.к. писание и чтение на носителе часто ограничивается физическим барьером.

Для проведения операций \cite{haberland08-1} и описания \cite{haberland08-2} может быть использован де\-кла\-ра\-ти\-вный язык, основанный на Прологе. В выделенном примере из \cite{haberland08-1} высчитываются из верхней папки ровно два файла \texttt{a} (например, файлы с указанным префиксом) и высчитывается содержимое, которое должно совпадать в обоих случях.

\begin{center}
\begin{minipage}{6cm}
\begin{verbatim}
template(element(top,_,[A,A]),[text(T)]):-
  A=element(a,_,_),transform(A//p#1,T).
\end{verbatim}
\end{minipage}
\end{center}

Приближенный аналог на  языке XSL-T является:

\begin{center}
\begin{minipage}{12cm}
\begin{verbatim}
<xsl:template match="top[count( child::*)=2] and 
   a[1] and a[2]">
  <xsl:text>
    <xsl:value-of select="//a//p"/>
  </xsl:text>
</xsl:template>
\end{verbatim}
\end{minipage}
\end{center}

%%%%%%%%%%%%%%%%%%%%%%%%%%%%%%%%%%%%%%%%%%%%%%%%%%%%%%%%%%%%%%%%%%%%%%%%%%%%%%%%%%%%%%%%%%%%%%%%%%%%%%%%%%%%%%%%%%%%%%%%%%%%%%%%

\subsection*{Побочные области}

Основные определения псевдонимов, мусора и обсуждения находятся в разделе \ref{sect:HeapModels}. Далее представлены побочные области, которые в принципе могут быть добавлены к предложенному фреймворку из раздела \ref{sect:ArchitectureVerificationSystem} и характеризуются как анализ указателей.

\subsubsection*{Анализ псевдонимов}

Вайль \cite{weihl80} представляет наглядный обзор по те\-ма\-ти\-ке, несмотря на возраст статьи. По Вайлю анализ псевдонимов является по\-буж\-да\-ющ\-им процессом  приближения указателей, ссылающихся на одну и ту же ячейку памяти, которая может быть изменена из-за того, что создаёт множество псевдонимов с экспоненциальным ростом вариантов. Анализ псевдонимов за пре\-де\-ла\-ми процедуры гораздо тяжелее, потому, что необходимо анализировать все возможные вызовы и продолжения, в котором далее будут существовать переменные. Естественно, что в продолжении переменные так же могут меняться, и это является причиной трудного анализа. Чтобы упростить анализ, Вайль вводит ярусы псев\-до\-нимов, по которым оценивается сложность каждой программной команды.

\cite{muchnick07} содержит полный обзор всех нынешних и прошлых  анализов псевдонимов, в том числе свои собственные. Статья Мучника вместе со статьёй Вайля представляются самыми важными в области. Он делит анализы на зависимые от  графа потока управлений и независимые, на «\textit{псевдонимы}» и «\textit{возможно псев\-до\-нимы}», а также на замкнутые подпрограммы и вызовы подпрограмм. Анализы псевдонимов по Мучнику можно уточнять, однако, классификация по сложности Лэнди \cite{landi91} тоже ориентируется по диапазону видимости указателей.

В отличие от Мучника, Купер \cite{cooper89} предлагает, одним из первых, подход анализа псевдонимов независимый от графа потока управления. Главной мотивацией тому служит резкое упрощение алгоритмов, т.к. необходимо следить за указателями и операциями использования. Недостатком является неточность.

Стоит отметить, что система  «\textit{GCC}» \cite{gcc15} разрешает программисту языков  Си ус\-та\-на\-вли\-вать директиву для компиляции вручную, по которой возможно твёрдо исключать псевдонимы для определённой функции. Таким образом, генерируемый код становится более эффективным.

\cite{horwitz89} является расширением подхода Мучника. С помощью  битовых век\-то\-ров, подход становится более эффективным. Общий подход улучшения анализа следует также искать в \cite{khedker09}. Более того, представленный подход анализа присвоений и использований может быть обобщён в подходе к  «\textit{SSA}» \cite{cytron91}, \cite{ssabook15}, \cite{vert13}, \cite{sethi75}. Наим \cite{naeem09} утверждает, что улучшение доступа можно достичь через хеширование подмножеств псевдонимов. Остаётся отметить, что подход Наима содержит предпосылки и намерения переписать проблему анализа  псевдонимов как проблему «\textit{SSA}». В общем следует отметить, что указатели не единственные переменные, которые трудно выявить при статическом анализе \cite{srivastava92} за пределами процедур. Они являются одними из самых трудных проблем. Подход Сривастава полезен, независимо от всех остальных обсужденных работ, по причине, что глобальная оптимизация на уровне связывания кода, например, передвижение мало-эффективного кода или устранение ненужного кода часто не может полностью быть исключено во время предыдущих фаз компиляции.

\cite{pavlu10} делит  анализ  псевдонимов на две методологии: подход с помощью  «\textit{уни\-фи\-ка\-ции}» \cite{steensgaard96}, который находит больше случаев  «\textit{псевдоним}» и на подход с помощью  «\textit{плавления}», который ради скорости слабее.

\cite{landi92} и \cite{ramalingam94} задаются вопросом, почему сегодня анализ  псевдонимов остаётся в общем не решённым. Ответ лежит в сложности выявления псев\-до\-ни\-мов за пределами процедур и в ограниченности решаемости для случая «\textit{псев\-до\-ни\-мы}». В сжатых подграфах зависимостей наименование вершин остаётся общей проблемой в \cite{sagiv02} и \cite{nielson99}. Далее, эти подходы имеют целый ряд ог\-ра\-ни\-че\-ний, как например: доступ к указателям с индексом, разрешает лишь не переменные выражения; исключаются объектные указатели; запрещаются массивы с гибкой шириной; исключаются пе\-ре\-се\-ка\-ем\-ые в памяти структуры данных (например,  «\texttt{union}»-структуры на языке Си \cite{isocpp14}).

Клинт \cite{clint73} занимается верификацией программ, и его центральная проблема заключается в доказательстве сложности  корректности  псевдонимов в связи с со-процедурами. Работу Клинта можно считать одной из первых в этой области. Обобщённой областью Клинта можно считать верификацию с функционалами, в которой необходимо упомянуть работу \cite{mehta05} для указателей с логикой высшего порядка, и работу Поулсона \cite{paulson93}, который пробует применить абстракцию процедур для принятия более эффективного правила при верификации с указателями.

Готсман \cite{gotsman06} представляет подход приближения  псевдонимов между процедурами с помощью модели ЛРП. Хотя статья представляет характеристики запуска и сравнения базисных оценок, для сравнения было бы интересно узнать о соотношении подходов Мучника и Купера.

\subsubsection*{Сбор мусора}
Джоунс и соавторы \cite{jones11} представляют, более чем полный, обзор на тему  \textit{сбора мусора} и детально обсуждают актуальные подходы с большим вни\-ма\-нием на параллельные алгоритмы, однако, это не цель данной статьи.
Помимо первого выпуска \cite{jones11}, второй выпуск одноимённой монографии не только сильно отредактирован и представляет целый ряд новых параллельных подходов, но практически одновременно является совершенно другой книгой. Стоит лишь упомянуть Уитингтона, \cite{withington91} в качестве обзора на тему, который сфокусирован в основном на многопоточные реализации сбора мусора, также как и Долигез \cite{doligez93}. Блэкбёрн \cite{blackburn04} приводит сравнения наиболее важных сборщиков мусора, что также упомянуто в \cite{jones11}.

На мой взгляд, Аппель \cite{appel87} прежде всего опасается на то, что на практике из-за ар\-хи\-тек\-ту\-ры фон-Неймана и \textit{программного интерфейса} «\textit{ABI}» на кодовом уровне имеются ограничения с обращением к стеку, связи с копированием данных в и из стека. Аппель предполагает, что, в таком случае, выгоднее использовать динамическую память, но при условии, что либо вообще ничего не потребуется, либо меньше ресурсов, чем для  стека (см. \cite{tofte97}, \cite{meyer2-03}). Подход Аппеля заключается в многопоточной реализации с   «\textit{ко\-пи\-ру\-ющ\-им сбор\-щи\-ком му\-со\-ра}» \cite{jones11}, который действует только, если было совершено изменение в данных, а зарезервированный превышает в семь раз объём свободных куч.

Несмотря на возраст \cite{larson77}, деление памяти на быстрый  кэш и медленную, большую память принципиально остаётся в силе. Ларсон рассматривает   «\textit{уплотняющий сборщик мусора}» в зависимости от объёма освобождающей области ($R$), объёма активных данных ($A$) и объёма быстрой памяти ($H$). Он постулирует две оптимальные  стра\-те\-гии для быстродействия и выделения/очистки динамической памяти: стратегия №1) Мак\-си\-ми\-зи\-ро\-вать $R$, если $A\ll H$ не соблюдается. Стратегия №2)  Приравнять $R=H$, если $A \ll H$.

Кроме упомянутых в \cite{appel87} ограничений, сбор мусора ещё имеет ог\-ра\-ни\-че\-ние в связи с адресацией. Если речь идет о данных   «\textit{XOR}»-связанных структурах, то  мусор не может быть локализован классическим методом потому, что адреса принадлежащих объектов (например, поля объектов или ссылки на по\-сле\-ду\-ющий элемент в списке) получаются не абсолютными, а относительными. То есть, для  дважды связанного списка с помощью одного  «\textit{офсета}», можно выявить сле\-ду\-ю\-щий и предыдущий элемент, если таков имеется, вместо двух отдельных указателей, которые содержали бы два абсолютных адреса.

\cite{sun06} предлагает  сбор мусора по возрасту выделяемых ячеек, т.е. в зависимости от того, как часто выделенный объект употребляется и как долго он остаётся. Если объект используется редко, то через несколько итераций он уже может оказаться в менее быстром сегменте памяти. Если объект используется часто, то он, наоборот, перемещается в более быстрый  регион памяти --- это касается не только сбора мусора, а также запуска программ в целом \cite{tristan09}, \cite{muchnick07}, \cite{tristan08}, \cite{reif80}. Быстродействие в среднем приближено к оптимуму, исходя из большого практического опыта и множества экспериментов.

\cite{hu10} даёт оценки к нынешним технологиям «\textit{SSD}»-дисководов. Операции до\-сту\-па к «\textit{HDD}»-дискам похожи на динамические  кучи. Писание превышает дли\-тель\-ность чтения в 10 раз. Сбор мусора на «\textit{SSD}»-дисках при «\textit{жадной}» стратегии, становится наихудшей при задержке на 45\%. Самая худшая  стратегия по производительности --- писание всё новых блоков. А лучшая стратегия --- писать по\-сле\-до\-ва\-тель\-ные данные без сбора мусора, по отдельности.

Кальканё \cite{calcagno03} замечает, что программы высокого абстрактного уровня не всегда имеют преимущества по сравнению с программами, которые манипулируют динамической памятью. Например, замечается, что сбор мусора во многом неэффективен в функциональных языках программирования. Выходит, что маленькие программы «\textit{опасны}» потому, что неверное использование с указателями может иметь побочные эффекты, но с другой стороны эффективны.

Подход Уэйт-Шора \cite{schorr67} классический и первозванный в данной области. Он характеризуется тем, что все элементы динамической памяти имеют счётчик активных указателей. Если, по каким бы то ни было причинам, ячейка вдруг более не жива, т.е. количество ссылающихся указателей становится ноль, тогда система управления (обычно ОС) активируется и принимает соответствующие шаги по утилизации ненужного объекта.

Ситуацию вокруг сбора мусора на данный момент лучше всего можно охарактеризовать: подходов имеется больше, чем достаточно, в общей сложности, как минимум 50 различных. Выигрыш через новооткрытие на сегодняшний день кажется мало вероятным и бесперспективным, т.к. на данный момент сборщики мусора в системах эксплуатации могут быть запрограммированы так, чтобы они не мешали запуску программы с помощью нитей.

\subsubsection*{Интроспекция кода}

Интроспекция означает, что при запуске программы единицы модуля определяются, например, код вставляется, меняется или добавляется во время запуска. Очевидно, что интроспекция фундаментально усложняет статический анализ, ради проблемы приостановки, а значит, и верификацию программы. Принцип интроспекции основан на решимости, которая приводит к загрузке, либо известного, либо экстерного неизвестного кода. В любом случае должен существовать механизм в заимосвязи, который определяет коммуникацию между вызывающей и вызванной сторонами. Например, в Си++ интроспекция основана на «\textit{RTTI}» \cite{isocpp14}, в Яве \cite{flanagan02} каждый объект содержит ещё, кроме полей и адресов, идентификатор в качестве смежной записи, которая читается при запуске.

Формэн \cite{forman04} определяет интроспекцию программы на Яве, как возможность вычитания и модификации данных и самой программы во время запуска. Это включает в себя загрузку кода, который может определиться даже только во время запуска. Следовательно, это потребует, кроме запуска, ещё компиляцию кода во время запуска. Запуск совершается объектами из определённого загрузочного контейнера. Манипуляция и доступ к классам, объектам и их компонентам осуществляется с помощью зарезервированных операторов.

Чеон \cite{cheon04} дает краткий обзор текущих и предложений вопросов интроспекции в области Ява. Вопросы частично ещё актуальны на данный момент. Так как запуск неизвестного кода представляет собой не только возможную опасность, но и сильно могут пострадать корректность и полнота. Печально было бы, если вдруг метод, именно для одного входного значения, вообще не работает, либо вычисляет неправильный результат. Для этого Чеон формулирует список предложений и важных  обстоятельств, как например: (1) загрузка и избегание конфликтующих имен в контейнере, (2) при запуске доступ к объектам должен соблюдать типизацию, (3) спецификацию наследованных объектов очень тяжело понять и отследить, если иерархия наследованности велика, (4) спецификации обычно ссылаются на объекты в стеке и куче, как в таком случае быть при интроспекции с окнами стека, которые просто разваливаются, т.к. объектная замкнутость может нарушаться при интроспекции?

Главными обсластями применения интроспекции являются загрузка приложений или методы веб-серверов, загрузка и доступ к динамическим библиотекам, например \cite{zakharov15}.

%%%%%%%%%%%%%%%%%%%%%%%%%%%%%%%%%%%%%%%%%%%%%%%%%%%%%%%%%%%%%%%%%%%%%%%%%%%%%%%%%%%%%%%%%%%%%%%%%%%%%%%%%%%%%%%%%%%%%%%%%%%%%%%%
\subsection*{Существующие среды} % Existing Tools
 
\cite{reynolds02}, \cite{berdine05-2}, \cite{reynolds09} представляют собой хорошее введение в  ЛРП, в том числе при поддержке программных средств. ЛРП является  подструктурной логикой \cite{restall00}, \cite{restall94}, \cite{dosen93}, которая избавляется от констант, например,  булевых значений и использует символы для структурных замещений, в качестве констант. В ЛРП структурными правилами служат, согласно \cite{restall00}, правила  \textit{сужения},  \textit{контракции} и  \textit{сопоставления}, а константами служат ячейки  ди\-на\-ми\-чес\-кой памяти (см. главу \ref{chapter:expression}). В правилах «\textbf{,}» заменяется на «$\star$», которая разделяет две не\-пе\-ре\-се\-ка\-ем\-ые и различные друг от друга кучи, кроме того случая, когда уговаривается что-нибудь иное. Кучи в динамической памяти определяются индуктивно. Оператор «$\rightarrow$», например в выражении «$a\rightarrow b$», определяет соотношение между переменным символом на левой стороне и вы\-ра\-же\-ни\-ем (например объект) на правой. Когда мы рассматриваем примеры ЛРП, то подразумеваем, что эти конвенции вместе с правилом фрейма соблюдаются.  \textit{Правило фреймов} означает, что если вызов подпроцедуры не меняет части куч, а именно раму обозначенной $F$, то в  антецеденте достаточно доказать, что тройка Хора без $F$ верна. Рассмотрим к примеру правило фрейма над кучами, которые определяются позже:

\begin{center}
\begin{tabular}{c}
\inference[(FRAME)]{\{P\}C\{Q\}}{\{P\star R\}C\{Q\star R\}}
\end{tabular}
\end{center}

Здесь $P$ и $Q$ являются пред- и постусловием, а $R$ является кучевым фреймом, т.е. сложное утверждение, которое не зависит от выполнения $C$, а также, на чью верность $C$ не имеет эффекта. Этот принцип является основой модуляризации и мы будем всегда исходить из того, что принцип в силе, кроме отдельно уговариваемых случаев.
Если допускаются рекурсивные спецификации вместе с функционалами, то основные свойства фреймового правила могут нарушаться в связи с контекстом, который может меняться любой вызывающей инстанцией \cite{birkedal04}, \cite{pottier08}. Следовательно, рекурсивные спецификации, в таком случае, должны обобщаться.

\cite{berdine05-2} вводит вычисление на основе ЛРП с безграничной арифметикой в  офсетах над  указателями, с возрастающими  массивами и рекурсивными про\-це\-ду\-ра\-ми. Например, \texttt{p+x}, где \texttt{p} указатель, а \texttt{x} переменное целое число, не решимо в общём. Оно представляет собой попытку определить области динамической памяти рекурсивно, с помощью замкнутого объёма встроенных правил. Бёрдайн и соавторы, лишь задают вопрос, не является ли  типизация достаточной верификацией (см. набл.\ref{obs:NonReducibilityToTypeChecking})? Из статьи например следует, что нерешаемость безграничного использования ука\-за\-те\-лей, приводит к не точному определению момента  cбора мусора,  даже для самых безобидных выражений в качестве офсетов памяти и к не точному выбору правил для верификации, которые управляются жадной  эвристикой.

Борна \cite{bornat00} предлагает модель очень похожую на  ЛРП, которую он на\-зы\-ва\-ет «\textit{даль\-нее разделение}» и которая преобразует объекты данных в массив. Таким образом, любое объектное поле получается индивидуальным указателем с конвенцией для наименования и идентичности. Для общего определения куч он постулирует, что необходимо использовать предикаты первого порядка.

Главной заслугой Хурлина \cite{hurlin09} в ЛРП, является нововведенный  паттерн для ве\-ри\-фи\-ка\-ции доступа к кучам с многими потоками. Он предлагает разблокировать функции, которые улучшают производительность за счёт частичной блокировки сложных ячеек. Все кучи не нуждаются в полной спецификации, которую можно сократить, используя анонимный оператор «\_».

Паркинсон \cite{parkinson05}, \cite{parkinson06}, \cite{parkinson05-2} представляет объектно-ориентированное расширение подхода ЛРП \cite{reynolds02} до Хурлина, используя  Ява в качестве входного языка. Модульность и наследование мо\-де\-ли\-ру\-ю\-тся с помощью  «\textit{передачи собственности над вызовами}» и  «\textit{семейством абстрактных предикатов}». Модель доступа Борна \cite{bornat00} при\-мен\-я\-ет\-ся, т.к. свойство непрерывности наблюдается у   правила фреймов для объ\-ек\-тов. Он замечает, что зависимость между предикатами создает порядок вызовов предикатов. Из ста\-тьи можно выявить, что его предикаты разрешают определить всю динамическую память и стек, но они не разрешают, например, определить любые предикаты первого порядка. Предикаты утверждения, которые сильно от\-ли\-ча\-ют\-ся синтаксисом и се\-ман\-ти\-кой от входного языка, не ограничены в типах, однако, использование сим\-во\-лов имеет целый ряд ограничений, которые связаны с несимвольной реализацией сопоставления переменных символов. Фактически, предикаты используются как  локальные переменные в императивных языках программирования --- это существенное ограничение. Для даль\-ней\-ше\-го расследования, он предлагает: вызовы методов из родительских клас\-сов,  статические поля,  интроспекцию объектов,  вну\-трен\-ние классы и кванторы над пре\-ди\-ка\-та\-ми.

Система верификации  «\textit{Smallfoot}» \cite{berdine05-2} является первым  верификатором на основе  ЛРП. Она работает чисто экспериментально и имеет маленький объём  встроенных предикатов для определения куч. Система сильно ограничена в представленных и предложенных к применению структур данных. Она очень часто выходит из строя из-за быстрой прототипизации и связанных с ней множеством ошибок. Часто наблюдается нетерминация всязи с неограниченной тактикой применения подходящих правил. Система  «\textit{SpaceInvader}» \cite{calcagno09} является наследником «\textit{Smallfoot}» и включает в себя простые элементы  абдукторного вывода.  «\textit{jStar}» \cite{distefano08}, \cite{parkinson06} можно рассматривать, как объ\-ек\-тно-\-ор\-и\-ен\-ти\-ро\-ван\-ное расширение «\textit{Smallfoot}», которое уже поддерживает классные ин\-ва\-ри\-ан\-ты и ограниченные \textit{(абстрактные) предикаты}. «\textit{jStar}» преобразует все про\-грамм\-ные команды в форму  «\textit{JIMPLE}», которая очень близка к  IR от «\textit{GCC}», язык который называется  «\textit{GIMPLE}» \cite{merrill03} и с помощью которого, проводится верификация поблочно-управляющим командам. Для индустриального применения, это решение является наи\-бо\-лее приемлемым.  «\textit{Verifast}» \cite{jacobs11} работает с входным  диалектом Си и разрешает про\-из\-воль\-ные определения  абстрактных предикатов. При не\-воз\-мож\-нос\-ти вывода, пользователю приходится вручную добавлять команды и подсказки к правилам для завершения доказательства. Программная среда Хурлина \cite{hurlin09} очень похожа на \cite{distefano08}.

 «\textit{Cyclone}» \cite{grossman02} является верификатором  динамической памяти на основе  ВР (см. раздел \ref{sect:StackAlignment}).  «\textit{SATlrE}» \cite{pavlu10} верификатор, который работает на основе \textit{анализа образов} \cite{sagiv02} и реализован в Прологе. «\textit{SATlrE}» отличается тем, что образцы и зависимости между ними вычисляются не каждый раз заново, а только меняющиеся пути.  «\textit{Ynot}» \cite{nanevski08-2} является  «\textit{SMT}»-решателем на основе  «\textit{OCaml}».

Далее, имеются следующие подходы и программные среды: (1)   «\textit{KeY/VDM++}» \cite{barnett04}, \cite{weissenbacher01} способствуют применению в индустрии объектно-ориентированные спе\-ци\-фи\-ка\-ции и интеграции с языком моделирования  «\textit{UML}», которое не допускает доказательства  динамической памяти, (2) отслежка ячеек памяти при запуске программы (динамический подход), например  «\textit{Valgrind}» \cite{valgrind} или  «\textit{ElectricFence}» \cite{sun06}, \cite{electricfence} \cite{kirsch03}, (3) программные среды интегрированные (статический анализ) в пакете компиляции  «\textit{LLVM}», как например,  «\textit{SAFECode}» \cite{safecode14}, (4) преобразование программы и ут\-верж\-де\-ний в вид  слабоструктурированных данных \cite{badros00} и возможной трансформацией ею \cite{johann03}, \cite{dodds08} (см. раздел \ref{sect:HeapModels}). Раздел \ref{sect:LogicalReasoningAutomation} содержит подробнее обзор по теме «\textit{Абстрактной интерпретации}». Пункт 2 это динамический подход, который можно коротко описать на примере «\textit{Valgrind}». Оно загружает данную программу с аргументами запуска в память и сопоставляет все доступы динамической памяти своими командами. Таким образом, каждый запуск записывает все неверные доступы к памяти, а также утечки памяти.

«\textit{GCC}» \cite{gcc15} и «\textit{LLVM}» \cite{llvm15} являются фреймворками компиляции. Обе содержат модули по статическому анализу \cite{lattner03}, \cite{safecode14}. Хедкер \cite{khedker09} предлагает статический фреймворк, основанный на вычислении транзитивного замыкания достижимости в графах для анализа потока данных. Проект «\textit{CompCert}» \cite{blazy09}, \cite{leroy12} предлагает чисто академический фреймворк для анализа корректности трансформаций кода для процессорной архитектуры «\textit{PowerPC}» при поддержке верификатора «\textit{Coq}».

В качестве верификаторов общего назначения можно образцово привести уже упомянутую систему «\textit{Coq}», но также имеется ряд иных верификаторов, например, «\textit{PVS}»,  «\textit{Proof General}», «\textit{Isabelle}» и многие другие. Сравнение выбранных верификаторов можно найти в \cite{emanuell08}, а для верификаторов динамической памяти в \cite{haberland16-4}. Аппель \cite{appel12} представляет перечень, какие проблемы за последние годы были хорошо решены и какое имеется сравнение. Статьи из немецкого популярного технического журнала «\textit{iX}» \cite{kirsch95}, \cite{kirsch94} посвящаются наиболее удобными динамическими программами, ранее представленными, для анализа проблем динамической памяти (см. главу \ref{chapter:DynMemProblems}).

Статический анализатор «\textit{ESC Java}» (с англ. «\textit{Extended Static Checking for Java}») \cite{flanagan02} вводит модульную спецификацию в Ява (см. \cite{giorgetti10}, \cite{berg01}), но увы, указатели исключаются и модель памяти исходит из того, что все элементы существуют в динамической памяти, а интроспекция не учитывается. Спецификация Ява-программ предлагается подходом «\textit{Java-ML}» в \cite{badros00}.

Система «\textit{Jahob}» представленная в \cite{zee08} предлагает прототипную верификацию линейных списков для Ява-программ на основе функционалов на функциональном языке программирования близок к языку «\textit{LISP}». Используется эвристика, которая выбирает правила с более жёсткими требованиями. Цель проекта заключается в поиске автоматизации или ускорения за счёт абстракции спецификации функциями высшего порядка (см. подход Поулсона \cite{paulson93}).

Система верификации объектно-ориентированных систем, как в этой главе уже было изложено, которые более близки к индустриально применимой верификации являются, например, также «\textit{KeY}» \cite{mueller02}. Система «\textit{Baby Modula 3}» \cite{abadi93}, \cite{cardelli97} может быть применима, но страдает от множества ограничений (см. раздел \ref{sect:TheoryOfObjects}).

%%%%%%%%%%%%%%%%%%%%%%%%%%%%%%

%%%%%%%%%%%%%%%%%%%%%%%%%%%%%%%%%%%%%%%%%%%%%%%%%%%%%%%%%%%%%%%%%%%%%%%%%%%%%%%%%%%%%%%%%%%%%%%%%%%%%%%%%%%%%%%%%%%%%%%%%%%%%%%%%%%%%%%%%%%%%%%%%%% 2 Problem
%%%%%%%%%%%%%%%%%%%%%%%%%%%%%%%%%%%%%%%%%%%%%%%%%%%%%%%%%%%%%%%%%%%%%%%%%%%%%%%%%%%%%%%%%%%%%%%%%%%%%%%%%%%%%%%%%%%%%%%%%%%%%%%%%%%%%%%%%%%%%%%%%%%
\section*{Проблемы динамической памяти}

%%%%%%%%%%%%%%%%%%%%%%%%%%%%%%%%%%

В этом разделе имеется короткое введение в типичные и актуальные проблемы в работе с динамической памятью. Затем дается короткая мотивация, почему изучение и исследование этих проблем актуальны с теоретической и практической точки зрения.

Под  \textit{динамической памятью} подразумевается (см. рисунок \ref{fig:ProcessSectionLoader}) та часть  операционной памяти  процесса, которая выделяется во время загрузки по запросу \cite{levine99}, \cite{love10}.  Куча является  неорганизованной частью памяти (см. рисунок \ref{fig:ExampleHeap}), в отличие от  организованной части памяти, т.е.  стека. Организованность подразумевает автоматическое выделение и утилизацию  локальных переменных по  диапазону видимости (см. рисунок \ref{fig:ExampleStack}), с чёткой пропиской связи между элементами. В случае стека, при компиляции «\textit{ABI}» вынуждает к строгому порядку помещения в текущее стековое окно локальных переменных.   
Например, задекларированные локальные переменные, при  условном переходе видны в нём и могут перекрываться другими локальными переменными с таким же наименованием.
При выполнении программы, локальные переменные не доступны за пределами блока видимости. Такое же наблюдается с процедурами и подпрограммами: локальные переменные, а также  параметры по значению вталкиваются  при запуске процедуры в видимое  окно на стеке, т.е. актуальное. При выходе из процедуры, стековые  переменные выталкиваются из стека, окно утилизируется (подробные реализации можно найти в \cite{gcc15}). На рисунке \ref{fig:ExampleStack} имеется стековое окно с локальной переменной $o$, которая в данном случае, --- целое число $1$. Также имеется на рисунке  массив, (локально, не обозначен) с указателями, которые ссылаются на $o$ и на высший элемент. Все элементы \textit{стекового окна} компактно помещены в стек, т.е. между ячейками обычно не имеются дыры. В отличие от этого, запись не обязательно должна быть компактной, точнее, если не упоминается ключевое слово «\texttt{packed}», то часто это и не происходит, благодаря генерации наиболее оптимального кода по  быстродействию или  размеру кода.

В отличие от этого,  ячейки \textit{кучи процесса} могут быть разбросаны как угодно. Видимость ячеек кучи не зависит от синтаксических блоков, а зависит исключительно от момента выделения динамической памяти до явной  утилизации памяти. Если куча не утилизируется программным оператором, то по умолчанию все кучи утилизируются глобальным  деструктором программы перед завершением процесса операционной системы. С этой целью кодовый сегмент \texttt{.dtor} в программах Си++ содержит все адреса деструкторов переменных объектных экземпляров, которые перед передачей контроля ОС вызываются. Стек и куча определяются операционной системой,  они находятся в  \textit{виртуальной памяти}. Подробные механизмы памяти выделения и утилизации в данный момент нас пока не интересуют. Для  стека и  кучи нет фиксированного порога при запуске программы, а имеется плавающий порог (штрихованный, см. рисунок \ref{fig:ProcessSectionLoader}), который может, либо расти, т.е. стек увеличивается за счёт уменьшения кучи, либо наоборот, в зависимости от объёма употребления памяти.  «\texttt{bss}» обозначает сегмент памяти, который содержит все  глобальные (часто могут включаться иные, не локальные, в зависимости от конкретных реализаций при компиляции \cite{isocpp14}) переменные, которым не было присвоено изначальное значение. Все остальные, не  локальные и не  динамические переменные выделяются в сегмент «\texttt{data}». 
Загружается программный код. В нём содержится объектный код с адресом точки запуска. Все относительные адреса заменяются абсолютными адресами. На практике программный код может и должен содержать различные фазы запуска программы, например, деструкторы.

Любая ячейка любого сегмента из рисунка \ref{fig:ProcessSectionLoader} может быть  адресована линейно. Это означает, что доступ к содержимому памяти можно получить по последовательному увеличивающемуся порядку. Любая ячейка памяти имеет последовательную по адресу $+\underline{1}$, где $\underline{1}$ является символическим числом, которое представляет собой размер одного элемента соответствующего типа. Однако, согласно рисунку \ref{fig:ProcessSectionLoader}, имеется плавающая граница, которая практически никогда не достижима. Например, если целое число имеет тип \texttt{uint16_t}, то размер $\underline{1}$ станет $2$. Обратим внимание, что с помощью  \textit{эксклюзивного бинарного оператора «\textit{ИЛИ}»} (XOR, $\oplus$) можно моделировать два  указателя в двусвязанном списке с помощью лишь одного  «\textit{поля прыжка}», если в качестве содержимого поля записывается $a \ xor \ b$, где $a$ является адресом начального поля, а $b$ следующее поле, т.к. действует $a \oplus (a \oplus b) \equiv b$, а также $(a \oplus (a \oplus b)) \oplus (a \oplus b) \equiv b \oplus (a \oplus b) \equiv a$ (см. \cite{parlante01}, \cite{sinha04}, \cite{haberland16-6}).

Таким образом, можно сэкономить один указатель, хотя, утилизация ненужных ячеек памяти не может быть решена стандартно. Допустим, производится манипуляция двусвязного списка, тогда необходимо проверять достижимость с помощью сумм всех адресов списка.

Памятные модели между языками  Cи(++) \cite{gcc15}, \cite{llvm15} и  Ява \cite{sun06} сильно отличаются, даже между Си++ и Си, если даже на первый взгляд это не очевидно.
Так например, у языков программирования, которые совместимы с  ISO Cи(++)  сбор мусора ненужных элементов часто не проводится. C одной стороны, это экономит дополнительные расходы, с другой стороны, это перекладывает большую ответственность на программиста. В языке Ява, сбор мусора очень часто приводит к дополнительным затратам, к счастью, сбор мусора работает часто в параллельной нити внутри  Явы в  виртуальной машине, что способствует к избеганию больших дополнительных расходов. Ява, как компромиссное решение, использует подход  «\textit{сбора мусора по генерациям}», который можно считать, приближенным к субоптимальным решениям со стороны практических порогов и объёмов мусора. В Яве все переменные и данные хранятся в  кучах. Диалекты Си допускают  \textit{слабую типизацию} \cite{isocpp14}. Не совместимые типы при различной интерпретации битов могут быть преобразованы, например,  байт в вещественное число с помощью  Си оператора  \texttt{union}.

\subsection*{Мотивация}
В техническом докладе \cite{miller90} на протяжении десятилетий анализируются совершённые и отслеженные ошибки при разработке открытых и коммерческих программных систем обеспечения. Результаты \cite{miller90} практически из года в год без изменения подтверждаются многочисленными новыми исследованиями, как например \cite{miller95}, \cite{hind01}, \cite{abiteboul05}. В \cite{miller90} демонстрируется, что в приложениях для коммерческих дистрибутивов операционной системы  «\textit{UNIX}», содержатся около 23\% ошибочного кода, а код открытых проектов содержит лишь около 7\% на платформах  «\textit{GNU}». Авторы осторожны и оптимистичны. По Миллеру наиболее часто встречаются следующие ошибки:

\begin{enumerate}
 \item Ошибки с указателями и полями. Авторы обращают внимание, что для коммерческих приложений, ошибки либо вообще не обнаруживаются во время разработки программ и начального тестирования из-за недостаточного  покрытия тестов, либо обнаруживается лишь при \textit{портировании} продукта на другую  платформу. Авторы более всего обеспокоены этим видом, т.к. их трудно обнаружить и исправить так, чтобы другие модули не пострадали. Чаще всего замечаются следующие трудности: неверный доступ к памяти, 
  неинициализированные ячейки данных или  утечки памяти, которые рано или поздно становятся  не доступными.
 \item  Неверный доступ к  массивам с преодолением доступных границ, либо в связи с изменением структуры данных, хотя  указатель остаётся без изменений.
 \item Недостаточная проверка работы с файлами. Часто при чтении, достижение конца файла вообще не проверяется.
\end{enumerate}

Верифицировать динамическую память тяжело потому, что  описывается состояние куч, а данная программа с описанием корреспондирует \cite{jones75} «\textit{неявным}» образом. Если интерпретировать кучу, как граф (см. следующие главы), а программа --- это последовательность инструкции построения того графа, то описание одной и той же кучи может быть сделано разными путями. При проверке, обе стороны должны корреспондировать.

Хинд \cite{hind01} задаётся вопросом, почему анализ  псевдонимов до сих пор не решён. Он призывает заняться расследованием  корректности и  быстродействия (см. рисунок \ref{fig:QALadder}), т.к. \textit{не выявленные псевдонимы}  означают возможную деградацию откомпилированного кода. Если во время компиляции известно, что некоторая пара локальных переменных обязательно ссылается друг на друга, то при  \textit{выделении регистров} \cite{muchnick07}, \cite{kennedy02}, \cite{hack11}, \cite{sethi75},  \cite{rideau10}, можно  гарантировано использовать один регистр, тогда отпадает необходимость синхронизировать два или более регистров (см. \cite{ssabook15}, \cite{lattner03}, \cite{cytron91}, \cite{muchnick07}).

Хинд \cite{landi91} убеждён в том, что для эффективного выявления  псевдонимов нужно анализировать входную программу. Лэнди \cite{landi91} вводит классификацию анализа  псевдонимов. Анализы являются  NP-твёрдыми. Они делятся на: (1) псевдонимы, которые  «\textit{должны ссылаться}» или  «\textit{могут ссылаться}» и на (2)  псевдонимы внутри  процедуры или за её пределами. Анализ внутри процедуры считается эффективным \cite{hind99}. Анализ за пределами процедуры считается неэффективным и нуждается в улучшении. Хинд \cite{hind01} предлагает использовать маленькие отрывки области видимости с локальными переменными при использовании эвристического подхода для решения вопроса  псевдонимов.

Ху \cite{hu10} представляет статью, которая посвящена техническим границам циклов писания и чтению актуальных  флэш-памятей. Флэш-память как постоянный накопитель, здесь упоминается потому, что виртуальная память ОС при необходимости использует её и это быстрее чем использование винчестерского диска. Писание в память в среднем занимает приблизительно в десять раз больше времени, чем чтение, а блочная запись имеет наибольший эффект. Результаты Ху можно растолковывать так:  сбор мусора не нужен вообще, в связи с продолжительностью жизни флэш-устройства, т.к. сбор мусора сильно задерживает процессы писания и чтения остальных процессов. Кроме  \textit{постоянного накопителя}, флэш-память также может быть использована как  виртуальная память во  встроенных системах. Оно может критически повлиять на быстродействие целого приложения.

Бэссей рекомендует из-за неточности сред верификации для решения различных проблем с динамической памятью: (1) избавляться от программ, которые при запуске  меняют программный код (2) соблюдать лозунг: «\textit{обыкновенные ошибки должны быть найдены обыкновенно}» (3) учесть, что провождение анализов программы не обязательно приводит к улучшению качества программы. Отметим, что пункт (1) лишь меняет момент определения программного кода. C одной стороны --- имеется увеличение гибкости, с другой стороны --- имеется возможность ограничения семантического анализа, например, с проверкой типов, но также дополнительные расходы на динамические проверки и генерации кода во время запроса.
Естественно, верификация не в состоянии угадывать поведение статически из-за проблемы приостановки.
Из-за упомянутых недостатков (см. \cite{reus06-2}, \cite{birkedal08}, \cite{schwinghammer09}, \cite{honda05}) динамически меняющегося программного кода далее динамические изменения не рассматриваются. В связи с ошибочным обращением к динамической памяти подразумеваются:

\begin{itemize}
 \item доступ к  недоступной памяти
 \item доступ к неинициализированной памяти
 \item нехватка динамической памяти  при востребовании
 \item снижение быстродействия
 \item утечка памяти  (см. далее).
\end{itemize}

Далее все эти последствия рассматриваются. Ошибки могут привести к самым непредсказуемым ситуациям, в худшем случае вплоть до неверных дальнейших ответов программы, а к  приостановке в лучшем случае. Приостановку можно без преувеличения считать «\textit{наилучшим}» вариантом, т.к. выход в определённой точке программы предпочтительно к неопределённому выходу. Факт приостановки программы говорит о неисправности, а далее коррумпированное состояние не является предпочтительным в различных смыслах корректности и целостности.

Бесси \cite{bessey10} приводит итоги анализа успешных верификаторов в общем, и объясняет актуальные причины и принципы. Наиболее важными пунктами Бесси считает:

\begin{itemize}
 \item Верификация   теорем всегда является точной наукой: эвристики могут применяться, но в конце интересует --- следует ли предполагаемый результат из  аксиом и
  правил или нет? Поэтому, давать квантифицированный результат трудно. Малое количество обнаруженных ошибок может означать, что верификатор плохой или хуже других. Бесси обращает внимание, что 
  эвристика «\textit{бери самое крупное правило первым}» на практике даёт хорошие результаты. Рекомендуется выбирать представительные примеры из \cite{zakharov15}, как например  операционные системы, либо программное обеспечение в публичном доступе.
 \item Нетривиальные доказательства должны быть стандартизированы насколько это возможно. Если имеется возможность, то посторонние подпроцессы должны независимо от правил доказательств приводить к состоянию вместе с программным представлением, в наиболее нормализованный и упрощенный вид, с которым можно продолжать верификацию. Конкретные специфические правила должны быть логично отделены от остальных структурных правил. Общим лейтмотивом должен послужить: «\textit{простые ошибки необходимо обнаруживать простым способом}». Если лейтмотив не соблюдается, то можно считать верификацию мало полезной.
 \item По возможности как можно чаще и раньше  исключать «\textit{дополнения}» входного языка, которые можно исключать без большой потери выразимости, как например, динамическое обновление кода. Специфицировать всё подряд или то, что не связано со спецификацией (см. опр.\ref{def:SpecificationLanguage}) или верифицируется с большими затратами в ущерб читаемости и при необоснованной потере ключевых свойств правил данного вычисления, не доступно.
 Всё должно быть простое и проверяемое, иначе доказательство может быстро оказаться неприемлемым, ради исключения редких случаев.
\end{itemize}

Как было ранее упомянуто, программирование с  динамической памятью может быть накладным со стороны быстродействия \cite{larson77}. Бывают случаи, когда использование куч быстрее \cite{appel87}, чем стека, в зависимости, как активно и эффективно работает  сбор мусора. Нельзя говорить обобщёно (см. главу \ref{chapter:intro}) нужно разбираться индивидуально, чтобы принять решение, какой алгоритм лучше. Эффективность куч также в значительной степени зависит от свободных/занятых списков куч, которые контролируются  операционной системой. Торможение из-за стека происходит в связи с генерацией ненужных инструкций  вталкивания и  выталкивания. Проблема увеличивается с большими объектами, которые не помещаются в регистры.     
Каждый процесс копирования структуры данных, будь-то слово процессора, либо комплексный тип, является потенциально лишним, как только речь идет о ссылках. Если значение передаётся, тогда любые дубликаты являются лишними. Увы, устранение по этому принципу не всегда наблюдается в  «\textit{GCC}» \cite{gcc15} или   «\textit{LLVM}» \cite{llvm15} из-за сильно консервативного подхода.
Замечаем, что часто любой алгоритм, который реализован с помощью куч, можно записать при использовании  стека последовательно, как было частично предложено в \cite{meyer1-03}, \cite{meyer2-03}.
Адреса любого объекта куч можно свободно передвигать, присваивая зафиксированный адрес при условие, что каждый элемент стека можно однозначно адресовать. Для этого потребуется дополнительная конвенция идентификации последовательных типов (т.е. различия размеров) элементов на стеке. Важно соблюдать максимальный порог вталкиваемых объектов в стек, который увеличивается при уменьшении кучи (см. рисунок \ref{fig:ProcessSectionLoader}). Элементы массивов, т.е. кучевые объекты одинарного типа, можно вталкивать, просто копируя элементы последовательно. Важно, сохранять на стеке информацию о том, что следует массиву. Таким образом, единственными открытыми вопросами остаются: (1) Действительно ли так нужно поступать, ради читаемости, эффективности алгоритма и т.д.? (2) Всегда ли применим такой вариант или нет? Первый вопрос довольно спорный, сравнив например, элегантность деревьев в кучах (см. \cite{parlante01}). Второй вопрос не всегда решим, в частности, когда лимиты устанавливаются только при запуске программ, прочитав например информацию из файла. Конечно, можно попытаться установить некоторый максимальный технический массив для покрытия большинства случаев, однако, таков алгоритм далеко не эффективен, а тем более неполон. Всё, что можно заранее вычислить и статическим образом установить, это очень полезно для вычисления только со стеком. Увы, часто бывают ситуации, когда это не приемлемо по техническим или другим причинам. Нельзя не заметить, что стек и куча конечны.

Далее рассмотрим некоторые конкретные проблемы  динамической памяти на образном диалекте  Си.

\subsection*{Проблемы в связи с корректностью}

\textbf{Пример 1-- Утечка динамической памяти.}
 Утечка памяти является одной из проблем, которая встречается очень часто. При утечке выделяется объект в динамической области памяти, которая, в конце концов, не освобождается. Часто это не приводит к краху загруженной системы, однако, нельзя сказать, что это удовлетворительно. Программа раньше выйдет из строя без предупреждений, либо в любой неопределённый момент.
Это может произойти тогда, когда ОС вдруг обнаруживает нехватку доступной памяти, либо ОС не может разрешить доступ к памяти
из-за специфических правил  ОС по безопасности, либо просто в ОС не осталось достаточно иных ресурсов.
Типичный сценарий такой, что программа работает пару минут или даже три недели подряд без событий, а затем, либо из-за внешнего события, либо сообщения, программа неожиданно выходит из строя или просто  терминирует без указания ошибки. Часто отслеживание, если оно возможно, может не привести к самой ошибке, т.к. выделение памяти (если обнаружить) может являться только \textit{симптомом}, но не настоящей причиной ошибки. Выявить настоящую ошибку крайне тяжело. Ошибкой программы из рисунка \ref{ExampleObjectInstantiation1}

является тот факт, что содержимое от \texttt{object1} не  утилизируется после повторного присвоения. Если предположить согласно ISO Си++ \cite{isocpp14}, что ОС при завершении программы освободит все выделенные ячейки, то уже до завершения программа может выйти из строя из-за нехватки памяти. Проследить такой тип ошибок будет крайне тяжело или практически не возможно, в частности потому, что при каждом запуске пороги ответственные за провал могут кардинально отличаться. Проблема заключается в анализе «\textit{опасных}» мест в программе, которые могут не освобождать ранее зарезервированные ячейки.\\

\textbf{Пример 2 -- Неверный доступ к памяти.}
 Неверный доступ к динамической памяти является ошибкой, которая встречается относительно часто. Причиной служит тот факт, что некоторый объект не инициализирован, либо имеет не правильное значение. Отметим, что ОС по-разному относится к объектам, которые находятся в  «\texttt{bss}». Например, ОС  «\textit{Windows}» часто не инициализирует локальные переменные. Это приводит к ячейкам с  неопределённым значением, тогда продолжение программы не определено, т.к. вычисляются альтернативные пути программы.
Это может  уязвить процессы.
Причиной является неправильное или неприсвоенное значение переменных в рассматриваемом алгоритме (см. рисунок \ref{ExampleObjectInstantiation2}).

Анализ ошибки должен сосредоточиться на все присвоения с момента первой декларации переменной \texttt{object1}, включая все присвоения, если значение было передано. Из всех присвоений необходимо выявить то присвоение, которое «\textit{неправильное}» согласно формальной/неформальной спецификации данного алгоритма.  Неверный доступ к памяти можно считать как попытку доступа к объектам, которые были нечаянно  утилизированы. Неверным доступом можно считать множество ошибок связанных с неправильной разрядностью или шириной  процессорного слова в связи с неправильным применением преобразования  типизации  указателей. Например, в связи с принудительной  \textit{конверсией типов} между \texttt{(int*)} или \texttt{void*}, либо \texttt{(char*)} и \texttt{(int8_t*)}, которые в зависимости от компилируемой платформы могут различаться. Также в зависимости от платформы:
  $$\texttt{sizeof(struct(int a, char b))}, \ \texttt{sizeof(int)+sizeof(char)}$$
не должны быть идентичны при типизированном доступе к адресу памяти. Размеры структур и её  разрядность могут тоже отличаться. Если предположить, что $a$ занимает 2 байта, а $b$ занимает 1 байт (что на разных платформах может быть не так), то содержимое, может быть одним из вариантов битовых масок из рисунка \ref{ExamplesBitmask}.

Поэтому, предпринимать какое-либо утверждение, зависимо от платформы, может быть даже ошибочным, если например речь идёт, об архитектуре  «\textit{Intel}»,  «\textit{Pow\-er\-PC}» 64-битных,  «\textit{ARM}» 32-битных или иных процессорах.

Если, по какой-то причине не инициализируется некоторое поле, то безобидная программа из рисунка \ref{ExampleUninitialisedFails} может  не терминировать или терминировать не правильно, а также может выдать неверную информацию -- любое из побочных эффектов является недопустимым.

\textbf{Пример 3 -- Висячие указатели и псевдонимы.}

 Висячие указатели (см. \cite{afek07}) получаются, когда несколько указателей ссылаются на один объект и операции над другими указателями, возможные  псевдонимы, приводят к тому, что хотя бы один указатель «\textit{нечаянно}» ссылается по ошибке на пустое место или операции над одним указателем меняют  связанную кучу. Хотя феномен интуитивно понятен, но на практике определение  псевдонимов может оказаться довольно трудным и неточным. Необходимо анализировать не только одну процедуру с указателями параметров, но все возможные вызовы --- т.е. одним вызовом множество содержимых указателей может поменяться, а в других случаях  может ничего не поменяться.
 Содержавшие объекты висячих указателей подлежат утилизации, т.к. по определению они являются мусором. Указатели, ставшие висячими, могут стать невисячими в ходе запуска программы, а также наоборот. Однако, содержимые, ставшие мусором, навсегда утеряны.
Феномен висячих указателей можно достичь именно таким способом, но также при использовании различных интерпретаций данных ячеек, например, предусмотренные структурой объединения c помощью  «\texttt{union}».\\

\textbf{Пример 4 -- Побочные эффекты.}

Вместо ненужного копирования структур данных, часто можно при вызовах процедур лучше использовать ссылки. Однако, этот подход содержит опасность, что нечаянно могут пострадать посторонние данные и переменные. Эта проблема является обобщением примеров №2 и №3: указатели не меняются в ходе запуска, но содержимое неожиданно меняется. Далее можно  обобщить: модификация неожиданно меняет другие переменные.
Модификация одной кучи не должна отражаться на другую.

\subsection*{Проблемы в связи с полнотой}

\textbf{Пример 5 -- Проблемы в связи c выразимостью.}

Несколько проблем с  выразимостью уже были представлены в этой главе. Основные проблемы выразимости заключаются в: (1) можно ли все допустимые кучи  специфицировать при условии, что все корректные кучи действительно синтаксически верны, исходя из данного набора правил? (2) Выводится ли, что все выбранные неверные кучи верифицируют как неверные, согласно правилам? (3) Какими можно использовать предикаты? (4) Каким условиям и свойствам должны придерживаться предикаты? (5) Каковы взаимосвязи между кучами и насколько адекватно это отражается в формулах? (6) Какие ограничения имеются в связи с использованием символами в описаниях куч? (7) Как лучше описать множество и отдельную кучу? (8) Как можно выразить зависимость между  псевдонимами? (9) Имеются ли многозначимые кучи или их описания, если да, то почему и можно ли их эффективно исключать? (10) Какой уровень  абстракции нужно вводить для удобного описания куч и для решения задачи верификации? (11) Имеется ли возможность абстрагировать настолько, чтобы, интуитивно стало ясно, о чём идёт речь, и пользователь, без особого труда, имел бы возможность лучше «\textit{понять}» спецификацию?  (12) Можно ли формулы для описания куч преобразовать так, чтобы не было необходимости специфицировать повторно? (13) Можно ли, если потребуется дополнительное  преобразование, использовать «\textit{что-то более знакомое}» программисту для спецификации и верификации чем искусственно определённые и ограниченные подвыражения логики предикатов (например, первого порядка), которые вычисляются и поддерживаются не полностью и часто не интуитивно?\\

\textbf{Пример 6 -- Проблемы с полными представлениями.}

Согласно рисунку \ref{fig:QALadder} в статье Сузуци \cite{suzuki82} были предложены операции над указателями, которые можно считать  «\textit{безопасными}». С использованием нужно быть крайне аккуратно потому, что  \textit{ротация} одной структуры данных по часовой может очень быстро привести к  уничтожению или фальсификации куч. Кроме этой проблемы корректности, часто неявные условия не очевидны.  Подход Сузуци, а также другие подходы страдают практически всегда от неполного набора правил и не полностью описанных куч. Часто имеется ситуация: дан набор 25 правил. Вопрос: (1) достаточно ли этих 25 правил или требуется ещё добавлять или даже необходимо удалять правила? 
(2) Как быть с  «\textit{правилами дубликатами}»?
Состояния куч связаны с  программными операторами. То есть, необходимо описать целые  подмножества куч и их сравнивать. (3) Имеет ли смысл ограничить выразимость куч так, чтобы приостановка была решимой с приемлемыми ограничениями? (4) Имеется ли возможность только часть кучи специфицировать и доказывать? (5) Можно ли предложить или использовать простую модель памяти так, чтобы доказательство данной простой структуры данных была простой, как например,  реверс  линейного списка (см. \cite{reynolds09})?

Желательно, чтобы для верификации не было необходимости постоянно все методы полностью специфицировать. Выработка корректной и полной спецификации на практике означает большие затраты рабочего времени инженера, в чем часто нет необходимости. На практике часто специфицировать достаточно лишь некоторые процедуры или объектные классы. Для этого необходимо отдельные фрагменты программы оставлять не специфицируемыми. Практическим требованием для инженера является возможность добавления  вспомогательных утверждений в те места программы, где инженер желает подробнее проанализировать некоторую нестыковку со спецификацией для  локализации и выявления ошибки (см. алгоритм №\ref{algo:AlgorithmProblemReduction}). Разработчик также может быть заинтересован в добавлении проверок в произвольных местах программы дополнительно к пред- и постусловиям.\\

\textbf{Пример 7 -- Проблемы в связи со степенью автоматизации.}

Согласно  лестнице качества из рисунка \ref{fig:QALadder}, эти проблемы входят во вторую категорию.
Главные проблемы автоматизации в ранее упомянутых разделах связаны с необходимостью определять аксиомы и правило динамической памяти от общих логических и иных правил, которые не связаны с преобразованием элементов динамической памяти.
Если установить  формальную теорию, основанную на равенствах и неравенствах куч, и эту теорию записать в отдельный набор правил, то набор уменьшается и верификация упрощается. Такая попытка значительно уменьшила бы численность и сложность данных правил. Необходимо улучшить сравнение спецификации с данным состоянием куч, которое к большому сожалению проводится в существующих подходах практически вручную (см. далее разделы и главы). Когда данную  теорему нужно использовать, а когда сопоставить с нужными символами --- трудно предсказать. Проблема также существует при преобразовании из одного состояния кучи в другое. Локальное оптимальное решение доказательства может всё равно привести к полной нерешимости.
При преобразовании куч с помощью  дедуктивного метода, также стоит задуматься об улучшении сходимости доказательства, подключив например абдукцию. Основным теоретическим ограничением может выступать выразимость формул. Термы могут быть, либо не полностью определены во время статического анализа, либо они нерешимы в принципе. 
Ограничение офсетов в арифметических выражениях приводит с одной стороны к ограничению выразимости, с другой стороны к повышению уровня автоматизации.
Возникает практический вопрос, насколько полезны определения алгебры куч? Насколько достаточна строгая типизация во входном языке программирования?

\subsection*{Проблемы в связи с оптимальностью}

\textbf{Пример 8 -- Проблемы в связи с быстродействием.}

Проблемы быстродействия конкурируют напрямую с  корректностью (см. рисунок \ref{fig:QALadder}).
Выявление утверждений  «\textit{указатели обязательно ссылаются}» или  «\textit{обязательно не ссылаются}» важнее, чем  «\textit{указатели могут ссылаться}», но их одновременно труднее выявить. Первое и второе утверждения имеют более сильный эффект на  генерацию кода.
Чем больше сокращается время запуска соответствующего кода, тем эффективнее он, т.к. отсутствие необходимости сохранения  процессорных регистров в стек, означает ускорение. Увы,  анализ зависимостей с указателями данных сложнее, чем с локальными переменными потому, что  содержимое указателя $p$ может меняться не только там, где имеются присвоения к $p$, но теоретически в любом другом программном операторе.

 Анализ псевдонимов является трудной частью, т.к. необходимо отслеживать все использования и вызовы процедур, что может привести к самым различным результатам. Часто, в  фреймворках с указателями наблюдается, либо наиболее  обобщённые утверждения, которые могут оказаться полностью без эффекта, либо утверждения вообще не рассматриваются.

Также важно заметить, что если структура данных используется только в одном месте и дубликаты отсутствуют, то требуется меньше расходов для непосредственной манипуляции.
В реализациях наблюдаются в основном два подхода: либо вручную передаются простые и объектные переменные (но возможно с вспомогательными замечаниями, например, в Си с помощью ключевого слова  \texttt{register} \cite{gcc15}), либо в исключительных случаях проводится  анализ псевдонимов (в основном исключительно внутри процедур \cite{llvm15}).
Этот случай не исключается в языках  Cи \cite{gcc15}, используя ключевое слово  \texttt{const}. В качестве мотивирующего алгоритма, например,  реверса списка с изменением существующей структуры \cite{reynolds09}, \cite{parlante01}, изначальный список исчезает, но это в зависимости от контекста может вполне устраивать.    Алгоритм Рейнольдса только один раз проходит через линейный список (без явных и обратных указателей). Таким образом, отпадает необходимость копировать список. Все операции производятся по данному списку, что сильно ускоряет. Если бы знать заранее, что структура данных будет меняться и оригинал будет не нужен, почему бы и не забросить старый список и таким образом резко ускорить алгоритм?
Данные компиляторы \cite{gcc15}, \cite{llvm15} пока не в состоянии проводить такой анализ, по крайней мере, не детально. Далее, можно ли поменять  «\textit{ABI}» при компиляции так, чтобы не используемые объекты удалялись из стека  вызывающей стороны \cite{isocpp14}, \cite{gcc15} безусловно, и уже существующие объекты в  динамической памяти были бы использованы непосредственно, если объекты строго не меняются при вызове? Надо отметить, что  модель указателей в \cite{sinha04} не стандартная, затраты для линейных списков сильно сокращаются,  сбор мусора изменен до неузнаваемости.

В частных случаях, когда заранее известно число  итераций циклов статическим анализом (см. \cite{appel87}), то можно оптимизировать место расположения ячеек памяти.

Проблемы оптимальности сбора мусора, начиная с  \textit{алгоритма Уэйт-Шора} \cite{schorr67} доныне, можно считать, более чем достаточно исследованы \cite{jones11}. 
Микрооперации в связи с динамической памятью производятся операционной системой, которая следует за свободными ресурсами, в том числе, куча и стек.\\

\textbf{Пример 9 -- Проблемы в связи с целостностью и безопасностью программы.}

В качестве оптимизации (см. рисунок \ref{fig:QALadder}) также рассматриваются проблемы в связи с анализируемой программой, где предполагаются корректность и полнота.

Аналогично к  \textit{переполнению стека} \cite{Kim15}, когда стек наполняется нежелаемыми данными, с целью передвижения актуального указателя за пределы актуального стекового окна при вызове или возврате с процедуры --- имеется такая же попытка вталкивания обратной метки, например, при сборе мусора в куче \cite{kaempf06}, \cite{afek07}.
Очевидно, что в отличие от стека, куча не содержит адреса программного кода, следовательно, атака включения вредного кода прямым образом не сможет сработать априори. «\textit{Переполнение куч}» может привести к потенциальной уязвимости процесса, либо целой системы.

Особенно критично стоит вопрос о безопасности с интерфейсами дальних серверов и служб, а также со спецификациями, когда вызывается некоторый системный доступ к драйверу \cite{corbet05}. Так как, драйвера могут запускаться несколькими инстанциями одновременно, то неисправность в связи с ошибкой в динамической памяти является особенно критической. 
В худшем случае нестабильность может привести к краху ядра ОС, как это имело место быть в случае с монолитной архитектурой ОС   «\textit{GNU Линукс}» с одним ядром.

%%%%%%%%%%%%%%%%%%%%%%%%%%%%%%%%%%%%%%%%%%%%%%%%%%%%%%%%%%%%%%%%%%%%%%%%%%%%%%%%%%%%%%%%%%%%%%%%%%%%%%%%%%%%%%%%%%%%%%%%%%%%%%%%%%%%%%%%%%%%%%%%%%% 3 Expression
%%%%%%%%%%%%%%%%%%%%%%%%%%%%%%%%%%%%%%%%%%%%%%%%%%%%%%%%%%%%%%%%%%%%%%%%%%%%%%%%%%%%%%%%%%%%%%%%%%%%%%%%%%%%%%%%%%%%%%%%%%%%%%%%%%%%%%%%%%%%%%%%%%%
\section*{Выразимость формул куч}

Предположим, что имеется некоторая   императивная программа, где  граф потока управления \cite{khedker09} выглядит, как представлено на рисунке \ref{fig:ControlFlowGraph}. Пример из \cite{cytron91} послужит нам образцом выявленных разниц между  автоматически и  динамически выделенными переменными. Особенность автоматически выделенных переменных заключается в том, что они выделяются и уничтожаются автоматически открытием и закрытием  стекового окна. Стековое окно содержит все  локальные переменные и параметры, которые поступают во внутрь и выходят наружу. Нельзя это путать с «\textit{fan-in}» и «\textit{fan-out}».

Согласно рисунку \ref{fig:ControlFlowGraph} определяются  \textit{интервалы видимости}. Интервал видимости всей процедуры, например [\textit{Вход,Выход}] означает, что  передаваемые переменные видны на всех вершинах между начальным блоком «\textit{Вход}» и конечным блоком «\textit{Выход}».  \textit{Блок} определяется как объединение последовательных неразветвляющихся программных операторов. Далее, во избежание  коллизий между различными переопределениями, предпочитается форма блоков в виде  «\textit{SSA}» \cite{cytron91}, \cite{hack11}, \cite{ssabook15}. Разветвлением может быть любой  условный или безусловный переход. Программные операторы безусловного перехода исключаются.

Допустим, в некоторых  блоках определены  локальные переменные по  «\textit{SSA}»-форме как имеется в рисунке \ref{ExampleSSABlocks}.

Предположим, все остальные блоки, либо искусственные, пустые, либо не содержат определения только что введенных переменных.
В блоке № 5: \texttt{f}, является процедурой, которая принимает одну переменную. На первый взгляд в этом нет ничего не обычного, однако,  если внутри \texttt{f} производится доступ к содержимому параметру, в Си это производится с помощью \texttt{\&}, тогда переменные могут меняться за пределами \texttt{f}. К счастью этот сценарий можно выявить с помощью предыдущего анализа всех входных и выходных переменных от \texttt{f}. Хотя описанный сценарий на практике может встречаться не часто, сценарий является причиной, почему множество оптимизаций не могут совершаться, спасая корректность в общности.
Если в блоке № 6 \texttt{a} определяется заново, то он выглядит так: $a_2= \phi(a_1,a_0)$. Функция $\phi$ является вспомогательной и означает объединение различных определений одной переменной.  $\phi$ неявная функция (отсюда и название с англ. «\textit{\underline{ph}ony}», что означает \textit{ненастоящая} или \textit{поддельная}). Концепция неявно определённой функции интересна, с ней мы встретимся позже при определении куч. Такого рода определения данных переменных программ является определением  зависимостей данных и может быть применено к любым автоматически выделенным данным в   императивных программах. Структура зависимостей в общем случае не может быть деревом, это запрещают  $\phi$-функции когда имеются зависимости, которые показывают на более ранний блок, который снаружи от циклов.

Как бы ни было, нас интересует, какие подходы имеются для вычисления интервалов видимостей с помощью $\phi$-функций \cite{ssabook15} локальных переменных. Мы обходимся тем замечанием, что имеются  максимальные границы, которые определяют интервалы. Эти границы могут определяться рекурсивно по стеку. Границы вычисляются с помощью графа потока данных и $\phi$-функций для каждого из локальных переменных (алгоритмы при поддержке вершин доминаторов представлены, например, в \cite{ssabook15}, классический подход смотри в \cite{cytron91}). Очевидно, что если переменные определяются в двух разных ветвях заново, то содержание после ветвления может отличаться, а индекс повышается.

Если попытаться применить  «\textit{SSA}»-форму к динамически выделенным переменным, тогда такая попытка не удастся из-за следующих причин: (1) существуют программные операторы, которые  выделяют и  уничтожают ячейку в динамической памяти явным образом. Эти операторы могут лежать за пределами блоков и процедур. Ячейки могут, безусловно, существовать за пределами процедур и после уничтожения переменной указателя. Это означает, что место определения (в том числе переопределений и уничтожения) и использования сильно отличаются от автоматически выделенных переменных, т.е. места не обязательно совпадают с местами употребления в программе. Даже могут меняться содержимые  указателей, когда на указатели, вообще, ничего не ссылается и указатели уже давно  утилизированы. (2) размер, частота и контекст выделения памяти в куче в общности не всегда определены (см. рисунок \ref{fig:ProcessSectionLoader}).

Аналогично к графу потока данных, можно приписывать не только содержимое переменных к вершинам графа аннотации о свойствах программы \cite{floyd67}, но также, например,  утверждения о динамической памяти. Примером тому, является  фреймворк представленный в \cite{khedker09}, который анализирует для каждого указателя при сильно консерваторском подходе -- возможность о псевдониме для всех остальных указателей (через расширенный алгоритм вычисления транзитивного замыкания). Следовательно, после каждого программного оператора высчитывается не только явно выявленный указатель, но также возможно всё связанное с ним. В качестве структуры данных, используется длинное  битовое поле, алгоритм подлежит улучшению, но это не делается, потому, что битовое поле является ключевой структурой данных. Принудительные проверки всех взаимосвязей почти нельзя упростить из-за транзитивности операций сравнения и выявления возможной связанности. Сравнение  псевдонимов основано на методе Хорвицы \cite{horwitz89} и Мучника \cite{muchnick07}.

Далее, можно привести аналогию, при которой каждое состояние  динамической памяти записывается как одно состояние. Очевидно, могут иметься любые переходы состояний, но всегда имеется начальная и конечная точка, когда все  кучи пусты. При выходе можно смоделировать вспомогательное  состояние такого рода, что все конечные состояния присоединяются к общему выходу.

Например,  функциональные языки программирования исходят из  \textit{принципа независимости данных} как главной концепции. Она гарантирует новый и всё свежий контекст. В нём находятся начальные параметры без взаимосвязей. Затем результат присваивается и передаётся высшей инстанции при возврате.
 Продолжение (с англ. «\textit{continuation}») \cite{joel70}, \cite{wiki23-02-2010}, \cite{dargaye07}, \cite{thompson97} может быть характеризовано как состояние вычисления, которое передается другой инстанции. При переходе состояние не прерывается, т.к. контекст не меняется. Когда речь идет об одинаковом состоянии, то подразумевается всё состояние вычисления, т.е. прежде всего, нас интересует стек и динамическая память (см. рисунок \ref{fig:ProcessSectionLoader}). Продолжение хорошо характеризуется  денотационной семантикой в случае  функций высшего порядка, но денотационная семантика в общем означает вычисление как функции, т.е. без взаимосвязей с окружающей средой. Отмечается, что продолжения сохраняют и обрабатывают стековые окна.
Джоэль \cite{joel70} рассматривает  \textit{продолжения} для функционально-логически смешанного  языка «\textit{LISP}». Основные итоги таковы: (1) продолжения повлекут за собой копирование, вталкивание и выталкивание регионов памяти, адреса переходов и возврата из стека, (2) читаемость и моделирование алгоритмов может быть существенно улучшено. Пункт (1) не может не вызывать озабоченность в связи со скоростью \cite{appel87}. Пункт (2) приоритетный и противоположен к первому пункту, поэтому на практике необходимо находить разумный компромисс.

В примере на рисунке \ref{fig:ExampleHeap}, если \texttt{free(o.B)} приводит к тому, что объект $C$ утилизируется из  динамической памяти, то доступ к $C$ через $o2$ недоступен. Если $o.B$ снова присвоить новый выделенный объект, то $C$ об этом не заметит, кроме, если новый объект находится по тому же адресу где находился старый объект. Ради формализма можно согласовать, что  висячие указатели присваиваются ради простоты и универсальной модели \texttt{nil}, если даже в реальности указатель содержит устаревший адрес. Это означает, что при анализе зависимых указателей необходимо рассматривать также все другие указатели, которые могут быть  псевдонимом  вершин пути доступа, т.е. $o$ или $o.B$ (см. проблемы из главы \ref{chapter:DynMemProblems}).
 Интерпретация ячеек динамической памяти определяется в зависимости от типа указателя, который проверяется согласно определению во время  семантического анализа. Тип не меняется, ради исключения подклассов (см. главу \ref{chapter:logical}).

\subsection*{Граф над кучами}

Теперь, когда основные аспекты динамической памяти были подробно обсуждены (в том числе набл.\ref{obs:OrganizedMemoryFreshContext} и набл.\ref{obs:UnorganizedMemoryUniqueContext}), пора задуматься о представлении графа. Заранее оговаривается, что уточнённые модели  динамической памяти будут вводиться в главах \ref{chapter:stricter} и \ref{chapter:APs}. Основными операциями манипуляции динамической памяти являются \texttt{malloc}, \texttt{free} и манипуляция с указателями.\\
\textbf{Граф динамической памяти как регулярное выражение.}    Для начала рассмотрим простой граф $A_1$ с упомянутыми ранее условиями, который мог бы быть описан регулярным языком и, следовательно, распознан простым конечным автоматом (см. рисунок \ref{ExampleNFA1}).

С помощью  леммы Ардена \cite{davis94} этот детерминированный  автомат может быть представлен следующим регулярным выражением: $b^{*}(a^{+}b)^{+}b$. Что теперь означают этот граф и соответствующее выражение? Граф содержит вершины и грани. Вершины представляют собой не пересекаемые ячейки памяти. Имеется некоторое объединённое финальное состояние $\{q_F | \forall q_F \in F \subseteq Q \}$. Естественно, все конечные состояния можно обозначить одним состоянием. В графе вершины $D$ и $E$ означают, например, некоторые состояния, которые объединяются в $\{\forall q_F\}$.  Грани означают ссылки, которые записаны в качестве адреса в размере процессорного слова в источнике грани. Возникает первый вопрос: что представляют собой наименования над  гранями? Это могут быть  указатели или  поля объектов, т.е. локации. Оба случая не очень удобны: во-первых, указатели должны выделяться отдельно от  ячеек, т.к. они расположены в  стеке. Во-вторых, наименования всех граней должны различаться друг от друга. Допустим это так, тогда  регулярное выражение уже никак не вписывается в данное компактное представление (см. далее). Это означает, что высокая изначальная компактность сильно страдает. В-третьих, не совсем ясно, что всё-таки означают «\textit{начальные}» и «\textit{конечные}» состояния? Начальное состояние можно всегда обозначить переходом одной стековой переменной, это всегда допустимо. Конечные состояния обозначить гораздо тяжелее и неординарно: является ли это обозначение конечным состоянием для вычисления? Если опустить $q_F$, то выражение просто не определено. Можно ввести новое состояние, которое принимает от всех состояний те переходы, которые сигнализируют окончательное вычисление структуры данных.

Допустим, все обозначенные проблемы в данный момент соблюдаются с достаточно удовлетворительным способом и мы продолжим вопрос о внесении изменения  динамической памяти с помощью программных операторов после выявления (не-)удобств  регулярных выражений. Если запись окажется компактной, то надо проследить, насколько она стабильная при манипуляции. Если имеется маленькая манипуляция, например, меняется только одна грань, то выражение не должно сильно меняться, и тогда можно было бы нотацию считать практичной. Если вдруг запись не устраивает, то необходимо выявить причину и искать другую модель представления памяти. К примеру, для последнего графа добавляется новая грань $b$, после ввода граф выглядит как $A_2$ в рисунке \ref{ExampleNFA2}.

Это эквивалентно выражению $b^{*}a^{+}b((a^{+}b)^{*}+bba^{*}b(a^{+}b)^{*})^{*}b$. Теперь мы удаляем  грань $a$ и получаем $A_3$ (см. рисунок \ref{ExampleNFA3}).

и получаем регулярное выражение $b^{*}a(a^{*}bb(\varepsilon + ba^{*}bb))^{+}$.
Допускается, что все три выражения могут быть переписаны и далее упрощены, но здесь это не решающий фактор. Проблема заключается в том, что если грань вставляется в любое место графа, а из этого надо исходить, то данное  регулярное выражение может очень сильно поменяться, в реальном и худшем случаях практически полностью. Чем больше граф, тем труднее записывать регулярное выражение, которое было бы оптимальным и как можно больше похоже на предыдущее выражение. Причина лежит в том, что система линейных уравнений  по Ардену при манипуляции подвергается малым изменениям одной грани. Например, первый граф описывается как в рисунке \ref{ExampleArdenLemma}.

Нетрудно убедиться в том, что система линейных уравнений имеет решение, но оно сильно отличается от предыдущего. При этом, уравнения почти не меняются. Однако, соответствующая регулярная грамматика сильно отличается.\\
\textbf{Граф при манипуляции.} Сначала необходимо рассмотреть последовательность типичных  программных операторов, чтобы обсудить некоторое «\textit{адекватное}» представление  динамической памяти, а также ситуацию с описаниями указателей и переходов.
Рассмотрим  инверсию  списка из \cite{reynolds09}. \cite{reynolds09} и особенно \cite{parlante01} содержат многие примеры, но остановимся на выбранном примере, который можно вполне считать представительным.
Дана следующая программа из рисунка \ref{ExampleListInversion} на  диалекте Си со специфическим синтаксисом касательно  указателей.

В программе $i$, $j$, $j$ и $k$ обозначают указатели. Доступ к последующему элементу в данной программе реализуется к примеру с помощью неявного оператора $*(i+1)$. Подразумевается, что элементы связаны между собой, а не только являются единым, монолитным, непрерывным  регионом в  динамической памяти, даже если об этом напоминает похожий  синтаксис. Семантику программу можно пояснить на примере рисунка \ref{ExampleListInversionSteps}, где номер, означает шаг итерации до посещения цикла.

До входа в цикл, $i$ содержит  список, при выходе $i$ пуст, а обратный список содержится в $j$. Копии не создаются. Входной список итерируется ровно один раз. Замечаем, что грани в примере не подписаны, но это в связи с выбранным неявным определением оператора над указателями. Конечно, в общем  грани могут быть подписаны. Однако, существуют указатели, которые являются  локальными переменными. Они присваиваются к вершинам графа, либо не присваиваются, тогда, когда  указатель  неинициализирован. Не инициализированным указателем является $k$ до вступления в цикл.

Нетрудно заметить, что компактная нотация даже при простых манипуляциях, например от (1) к (2), совершенно не приспособлена --- главная причина лежит в  выразимости и адекватном представлении графа. Хотя представленная регулярная запись для данной проблемы не пригодна, всё равно вопрос о компактном представлении также сыграла роль при формулировке автоматизированного подхода в главе \ref{chapter:APs}.

Также нетрудно понять, почему другие глобальные подходы, которые были введены в главе \ref{chapter:intro}, например  ВР,  АО и т.д. не являются настолько успешными. Причина та же самая, которая была продемонстрирована в последних двух примерах. Подход Доддса \cite{dodds08} далее не рассматривается, хотя он описывает графы до и после  трансформации и сосредоточен на описание трансформаций графов, но подход не специфичен для указателей и не рассматривает входной императивный язык программирования (см. главу \ref{chapter:intro}).

Исходя из наблюдения,  граф можно описывать различными способами: (1) или каждую компоненту по отдельности, (2) или предпочесть смешанную форму. Очевидно, что подход (1) не целесообразен потому, что описывать вершины по отдельности (при этом независимо от наименований указателей) может быстро оказаться нечитаемым и не удобным, а спецификация к графу должна быть удобной, короткой и адекватной. 
Подход (1) требует все вершины графа обозначить по отдельности и затем включить их в  спецификацию. Как было обнаружено в примерах введения, такого рода подходы подвергаются целому ряду недостатков и поэтому, в наших целях это не допустимо.
Подход (2) подразумевает, либо (2а)  \textbf{описать вершины}, либо (2b)  \textbf{грани}, а другую компоненту выявить из данной. Проблемой при описании только вершин (2a), являются дубликаты в  спецификации потому, что вершина полностью идентифицируется  гранями. Каждая грань имеющая отношение к данной вершине должна приписываться. Это равнозначно тому, что к данной вершине надо иметь список всех соседних вершин. Увы, такой подход очень неудобен. Например, одна вершина связана с двумя или более того вершинами. Это означает, что все соседние вершины также должны вписывать вершину в  свои списки. Кроме того, если происходит манипуляция кучи, то,  спецификация подвергается сильному изменению, а этого нужно обязательно избегать.

Подход (2b) наоборот описывает только  грани, а  вершины в них включаются. Этот подход содержит меньше  дубликатов и ближе к программе, т.к.  ссылки проводятся над реально существующими указателями программы. Если только направленные грани разрешаются, то проверять нужно на половину меньше. Несколько исходящих граней от одной и той же вершины запрещается потому, что один указатель не может ссылаться одновременно на несколько ячеек. Кроме того,  программные операторы при более внимательном наблюдении меняют состояние вычисления чаще, чем вершины. --- Перемещение,  утилизация и  выделение являются дорогими операциями, которые включают в себя вызовы  операционной системы, а манипуляция указателями является недорогой. В худшем случае, упомянутый подход только  выделяет и  утилизирует элементы, а указатели мало или вообще не меняются. Тогда принципиально важно задаться вопросом, не подлежит ли подход исправлению?

На рисунке \ref{fig:GraphIsomorphisms} (a) указан  регулярный граф, вершины которого имеют  степень «$3$», при этом подразумевается, что каждая вершина представляет собой объект, который содержит ровно три  указателя. Для подхода (2a) необходимо специфицировать 11 вершин, каждая из которых имеет три грани, при этом, количество входящих и выходящих граней может различаться и это необходимо рассматривать отдельно. В общем, имеется 18 граней. При подходе (2b) необходимо специфицировать только 18  граней. При этом вершины связанные больше чем с одной гранью могут обозначаться символами. Чем больше данный граф отличается от  полного графа, тем меньше граней необходимо специфицировать. Если меняется грань, то меньше нужно менять в  спецификации, точнее, только меняющуюся грань. Если  меняется вершина, то необходимо проверить все связанные грани. Направленный граф можно обыскать за линейное время все левые и правые стороны граней. Согласно подходу (2b) и данному набору вершин с помощью предикатов, проводить доказательства будет удобнее. Для описания  вершин используются  указатели или выражения  доступа к  полям (см. главу \ref{chapter:logical}). Не доступные поля при  спецификации нас не интересуют, т.к. такие  ячейки по определению являются  мусором (см. главу \ref{chapter:intro}), навсегда потеряны и не подлежат восстановлению, однако, при  спецификации мы заинтересованы выявить такие места, если таковы имеются.\\

Джоунс \cite{jones11} определяет  \textit{кучу} как непрерывный сегмент  операционной памяти размером $2^k$ с $k\ge 0$, который представляет некоторую  структуру данных --- альтернативно, как по\-сле\-до\-ва\-тель\-ность прерывных блоков непрерывных слов. Например, дерево может иметь между вершинами свободные элементы, но отдельные вершины не интерпретируются иначе, чем как данным(и)  процессорным(и) словом(-ами). По Джоунсу  объект является множеством  ячеек памяти, которые не обязательно связаны между собой, но чьи  поля  адресуемые. Каждая выделенная ячейка памяти имеет  указатель и с момента  выделения до момента  утилизации возникает вопрос, а жив ли содержимый объект или нет? Трудно не согласиться с Джоунсом в том, что фрагментация является проблемой, однако фрагментация внутри объекта исключается.

Коурмен \cite{cormen09} на стр. 151 определяет, также как и Бурстолл \cite{burstall72} любую  структуру данных в  динамической памяти, обязательно как  деревом. Дерево не только подразумевает некоторую связь «$\le$» к дочериным вершинам $V_j$, а также обязуется к соблюдению  упорядоченности $f(V_{parent}) \le f(kid(V_{parent},j)), \forall j$.
Аталлах \cite{atallah98} рассматривает  кучу как массив, интерпретируемый как дерево с более широким использованием памяти, но с более высокой гибкостью. В работе Атталаха можно заметить некоторые различающиеся определения  куч. Сначала куча определяется как реализация  «\textit{очередь с приоритетом}» (на стр. 79). Приводятся и обсуждаются  «\textit{кучи Фибоначчи}» \cite{fredman87} как специализированные и эффективные очереди для добавления и удаления. Далее, на стр. 105 куча определяется как  двоичное дерево, сохраняемое все элементы очереди с приоритетом. В отличие от свободного доступа к  операционной памяти, Атталах подчёркивает на стр. 111 важность доступа исключительно через имеющийся указатель. Таким образом, произвольная  адресация исключается. Ответственность деления  куч ложится на мало связанные кучи неявным образом, исключительно, на программиста и на моделирование ПО. Мало связанные кучи можно хорошо делить и обрабатывать эффективными методами. С Коурменом можно не соглашаться касательно произвольной адресации, из-за ранее упомянутых постановлений выразимости. Однако, с ответственностью программиста за создаваемую структуру данных нельзя не согласиться, если даже имеются строгие предпосылки в том, что все структуры данных должны являться деревьями, а не графами --- если даже на практике это часто так. Слитор \cite{sleator86} предлагает, для ускорения доступа к куче и минимизации операций сбора мусора, балансировать деревья равняя соседние вершины в отличие от других сбалансированных деревьев, что позволяет произвести быстрый поиск за $\Theta(n)_{min}=1$ и удаление за $\Theta(n)=log(n)$. Хотя предложение интересное, всё равно далее подход не будет рассматриваться в данный момент из-за отсутствия острой необходимости реализации, в связи с относительно маленьким объёмом конъюнктов и возможности линейного поиска по локации.

Рейнольдс определяет множество куч, как объединение всех отображений от адрес\-ного множества на не пустое значение  ячеек памяти. Следуя этому определению, одна куча --- это некоторое множество адресов, которые ссылаются на некоторую определённую  структуру данных (без дальнейшего уточнения). Рейнольдское опре\-де\-ление  структуралистское, т.к. куча, как отдельная, отличающаяся и независимая единица просто не существует (см. главу \ref{chapter:stricter}). Павлу \cite{pavlu10} определяет кучу как любой  граф, который не обязательно связан --- с этим трудно не согласиться.

%%%%%%%%%%%%%%%%%%%%%%%%%%%%%%%%%%%%%%%%%%%%%%%%%%%%%%%%%%%%%%%%%%%%%%%%%%%%%%%%%%%%%%%%%%%%%%%
  %%%%%%%%%%%%%%%%%%%%%%%%%%%%%%%%%%%%%%%%%%%%%%%%%%%%%%%%%%%%%%%%%%%%%%%%%%%%%%%%%%%%%%%%%%%%%%%
%%%%%%%%%%%%%%%%%%%%%%%%%%%%%%%%%%%%%%%%%%%%%%%%%%%%%%%%%%%%%%%%%%%%%%%%%%%%%%%%%%%%%%%%%%%%%%%

\subsection*{Предикаты}

Кроме ряда замечаний и конвенций, наиболее важными требованиями остаются: (1) после выполнения каждого  программного оператора имеющего отношение к  динамической памяти граф кучи меняется наименьшим образом, (2) спецификация соответствующая куче должна также меняться минимальным образом.

Для описания состояния, соблюдая все требования  куч, возникает вопрос, как это лучше описать, если очевидно, что сделать это не так просто?

В древней Греции в философской школе  Платона эпистемология некоторого описываемого объекта предлагалось известной аллегорией при чётко урегулированном порядке задачи вопросов и ответов между двумя сторонами: человеку, который имеет объект и находится на свободе и человеку, который заперт в пещере и желает понять сущность того объекта с ограниченными возможностями. Запертая персона коммуницирует исключительно речью, а также имеется факел, который горит не бесконечно. Свет факела попадает на обсуждаемый объект и оставляет за собой на пещерной стене тень и силуэты --- это неточное изображение и речь, всё, что воспринимает персона в пещере. Эта аллегория лучше известная как  «\textit{миф о пещере}». Она предлагает описание объекта через двустороннюю коммуникацию с целью последовательного выявления непосредственных свойств при имеющихся внешних преградах. Существует множество философских «\textit{Геданкеншпилей}», например, лишённый естественной речи диалог, эксперимент искусственного интеллекта «\textit{Китайская комната}» философа Серл. Мы ограничимся мифом о пещере ради классического и древнего характера, который содержит всё, что необходимо для понятия предикатов и куч.
Казалось бы, интуитивно понятно, но абстрактное объяснение, получает конкретное присвоение входящих и не входящих свойств (так называемая  «\textit{ре-ификация}» --- концепция от абстрактной мысли к конкретной реализации). Более современный философический дискурс по течению идеализма наблюдается в классическом подходе  Гегеля: тезис, который устанавливается из отмеченных наблюдений, затем выводятся для более тщательного анализа и определения свойств противоположных тезисов, что часто из-за не правильного или неточного определения тезиса приводит к конфликту, который затем разъясняется и выводится общий вывод ---  синтезированное утверждение. На данном этапе можно считать, что выявленные свойства и требования касательно куч проводились достаточно тщательным образом, чтобы предложить первые предложения. Для обнаружения неточностей и выявления более тщательного определения куч применяется подход, близок к этой концепции.

Что касается  реификации, то  кучи должны иметь  синтаксическое и  семантическое обоснование, и они будут основываться на  предикатах. Предикаты уже были предложены  Аристотелем, которые у него назывались  «\textit{силлогизмом}». Силлогизм, это единство трех компонентов логического правила формой: если «$A$» и «$B$», то следует «$C$». Здесь нечего добавить, кроме того, что большинство логических систем основаны на слегка модифицированной модели.

Предыдущая попытка представить кучи с помощью  регулярных выражений не увенчалась большим успехом потому, что изменения происходят в переходах графа. Минимальное требование применяется к правилам, но не к выражению, т.к. представление не слишком стабильное для данной проблемы. Но, в общем, выявление и преобразование представления из одной формы в другую очень широко обсуждается и расследуется. Довольно наглядно на примерах клеточных автоматов \cite{wolfram02} можно наблюдать за установлением инвариантов выражений. В главе \ref{chapter:APs} наблюдается обратный подход: наблюдаются образцы, из которых выводятся свойства о правилах.

Предикаты как связывающие вершины графа единицы языков, будь-то   формальных/естественных, по  семиотике имеют два значения: (i) интуитивное значение, это касается вопроса ---  \textit{что собой представляет на самом деле «куча»?}, (ii)  коннотативное значение --- \textit{с чем «куча» ассоциируется}? Далее, задаётся вопрос о представлении кучи: должна/может ли куча использовать  символы и реляции и что они означают для ее представления? Может ли куча определяться не  полностью?

Оценив (i) нужно заметить, что куча представляет собой множество  указателей, которые «\textit{как-то}» связаны между собой и указывают на объекты, которые находятся в  динамической памяти. Так как в прошлом большие трудности могли быть выявлены в связи с  многозначностью и резкими ограничениями, поэтому необходимо будущие определения редуцировать к минимуму. Оценив (ii) можно выявить, что связь всех компонентов задается данной программой, и куча не организована, это означает:  вершины графа могут поступать в любом порядке, в любом месте и непрерывность  сегмента памяти естественно не должна соблюдаться.  Кучи могут быть связаны с другими кучами.

Важно заметить, что  предикат должен иметь не только возможность выразить связь между вершинами графа кучи, но также должен существовать эффективный способ выразить, что две кучи не связаны. Если такой возможности нет, то  разделение автоматически получается неявным результатом анализа всех куч, что плохо из-за эффективности. Также у определения  связанности имеются различные модусы: «\textit{связан}», «\textit{возможно связан}», «\textit{не связан}», «\textit{возможно не связан}».  Модальность кучи, также как указание времени, не имеет большого значения: во-первых, «\textit{связан}» и «\textit{не связан}» можно проверить за линейное время, из графа кучи. Во-вторых, время отслежки дискретно и до/после каждого  программного оператора.
В предикатах вариации куч с помощью  логических операторов необходимо учесть, что логическая  дизъюнкция, несмотря на принцип  неповторимости и  отрицания предиката, должна быть выразительной --- этот вопрос решается в главе \ref{chapter:logical}, где доказательство куч будет основываться на  логическое программирование. Нам нельзя упустить, что в качестве указателей могут действовать любые допустимые указатели, это  локальные и  динамические переменные, а также  поля  объектных экземпляров.

В первоначальной работе по  «\textit{логике распределённой памяти}» \cite{reynolds02} вводится оператор $\star$ над  кучами и устанавливает законы следующим образом:

 Несжимаемость подразумевает, что из $p$ нельзя следовать $p$ дважды -- это совпадает с запретом  \textit{повторимости описаний} \cite{restall94}  кучи. Однако, $p \star q \not \Rightarrow p$, в отличие от  классической логики утверждений, не подразумевает строго $p$ как отдельное утверждение. При\-чи\-на не в том, что $p$ обязательно связана с $q$ -- как это можно предполагать, исходя из описаний и интуиции о делящем операторе, как было введено и про\-де\-мон\-стри\-ровано. Это лишь как частный случай. Проблема заключается в описании нахождения элементов в  динамической памяти. Это означает, что из утверждения о двух кучах, как бы они не были связаны между собой, не следует, что одна куча вдруг может исчезнуть. Такую разницу обстоятельств надо иметь ввиду.

Правила (2-4) понятные и нетрудно обосновать  графом кучи. В верности правила (5) можно убедиться, используя  индукцию к остальным правилам, при этом, различив два случая, когда $p_1$ и $p_2$ конфликтуют и когда не конфликтуют с $q$.

Правила (6) кажутся понятными. Эквивалентности с обеих сторон могут казаться очевидными. Предпосылками являются возможные конфликты термов утверждений согласно конфликтным ситуациям, разрешив заранее проводимое  переименование, что согласно  $\lambda$-вычислению конгруэнтно  $\beta$-конверсии.

Однако, ни правило (1), ни (6) и ни другие, не исключают возникновения  свободной переменной в двух кучах, которые связаны с помощью оператора  $\star$, что и наблюдается в разработках,  «\textit{smallfoot}» или  «\textit{jStar}». То есть, одна переменная может быть использована, например, в качестве  указателя в одной  куче, а в качестве используемой ссылки в подтерме, в другой куче.\\

Рассмотрим кучу из рисунка \ref{ExampleHeap1}.

В куче имеются указатели $x$,$u$,$y$. В прямоугольниках находятся содержимые $a_1$, $a_3$ и $a_4$.  Бурстолл предлагает обозначить содержимое не напрямую, а лишь на адрес. Также он предлагает записывать в минимальную модель все промежуточные ссылки или наименования, таким образом, можно большую часть кучи описать лишь одним выражением \xymatrix{ x \ar[r] ^{a_1,a_2,a_3} & y}. Это выражение означает многое: существуют указатели $x$ и $y$, существует гарантированный путь между ними и естественно весь путь описывает отрывок  графа кучи. В линейных списках, а также в деревьях указатели \texttt{next}, если таковы имеются, упускаются по умолчанию. Фактически одно выражение по Бурстоллу описывает  весь линейный список. Совокупность таких выражений описывает граф кучи. Рассмотрим следующую структуру данных в ввиде ``кактуса'' (см. рисунок \ref{ExampleCactusExample}):

Для описания понадобятся лишь два выражения. $x_1$ является  указателем, который выделен на  стеке. Допустим, все остальные ячейки не имеют указателей, следовательно, она расположена в  динамической памяти. Удобная, хотя и компактная, запись ломает обозначенные ранее критерии минимальности. Поэтому, запись Бурстолла далее не используется, но, рассматривается возможность  абстракции с помощью предикатов для более гибкого описания, которое основано на  гранях графа.\\

Для полноценного анализа указателей и  динамической памяти со стороны  выразимости нет необходимости разрешать доступ к  динамической памяти с любыми выражениями вычисления  адреса. Все операции доступа целиком могут быть выявлены писанием и чтением динамических ячеек. Если по причинам быстродействия нужен произвольный, последовательный доступ к памяти, то это является типичным сценарием применения  стека. Однако, выделение любого доступного размера, допускается как для стека, так и для динамической памяти. Надо понимать, что память выделяется как один непрерывный сегмент и проверка на «\textit{неверные}», «\textit{висячие}» ячейки и на свойство «\textit{мусора}» там, естественно, не проводится. Нетрудно себе представить  вспомогательные функции, которые могут произвести любые преобразования типов с данным сегментом памяти, поэтому этот вопрос далее не рассматривается. Обоснования и   путь доступа к объектам имеют только  указатели. Объектные экземпляры рассматриваются как  структуры, выделенные в динамической памяти, в отличие от классических структур, например в  Си \cite{isocpp14}, которые, либо расположены на  стеке, либо в  процессорных регистрах. Второй вариант принципиально не исключает подключение динамической памяти. Нужно  отметить, что $e_1.f_1 \mapsto val_1$ может означать, что объект $e_1$ может иметь ровно два или более  полей и в  откомпилированном коде $e_1.f_1$ совершенно иным путем будет реально обрабатываться и располагаться, чем $e_1.f_2$ в процессе различных оптимизаций кода \cite{kennedy02}, \cite{muchnick07}, \cite{gcc15} и, несмотря на введенные конвенции, всё равно это так. Это не означает раздробление на уровне  входного языка и  спецификации, где  объект должен моделироваться как единый не делимый  регион памяти, где  поля могут иметь  ссылки иных экземпляров.\\

Как во введении уже было изложено, в  ЛРП были внесены целый ряд предложений \cite{reynolds02}, \cite{reynolds09}, \cite{burstall72}, \cite{hurlin09}, \cite{parkinson05-2}, \cite{parkinson05}, \cite{parkinson06}, \cite{bornat00}, \cite{yang02}, \cite{berdine05-2}, \cite{berdine05}, \cite{ohearn04}, \cite{scholz99} (см. главу \ref{chapter:intro}).
Рейнольдс \cite{reynolds02} вводит   \textit{оператор последовательности} «\textbf{,}» для неявного определения  линейного списка. «\textit{Неявно}» означает, что не уточняется, каким образом последовательные содержимые связаны между собой, лишь уговаривается, что элементы связаны. Таким образом, из раннего примера  кактус можно определить $x_1 \mapsto x_2,x_3,x_4,x_7 \wedge x_5 \mapsto x_6,x_7$. Альтернативы неявному определению Рейнольдса можно считать, либо явное определение, которые всегда допустимое и полностью покрывает неявный оператор, либо по определению ограниченные по длине списки Бозга \cite{bozga08}.  

Если мы хотим  параметризовать кучу, необходимо вводить  символьные переменные в  утверждениях. Утверждение верное или неверное для данной  кучи.  Символы априори не типизируются, как термы  по Чёрчу, а информация о типе поступает из окружения, таким образом, оно более похоже на  типизацию по Карри. Вводя (символьные) переменные, определение предиката относится как параметризованный терм в $\lambda$-вычислении, т.е. предикат абстрагирован и подлежит к применению с другими  предикатами. Однако, предикат не возвращает иного результата как «\textit{истина}» или «\textit{ложь}», а присутствие входных и выходных данных отличается от классических функций (см. главу \ref{chapter:logical}). Рассмотрим рекурсивный пример  двоичного дерева из \cite{reynolds02}:

$$tree(l)::=\texttt{nil} \ | \ \exists x.\exists y:\ l \mapsto x,y \ \star \ tree(x) \ \star \ tree(y)$$

Согласно определению  Рейнольдса, оператор  $\star$ определит  кучу, которая состоит из  двух разделяющихся куч. Нужно отметить, как было упомянуто ранее, что изначальное определение $\star$-дизъюнкции может иметь соединяющие элементы. Таким образом, от одного дерева $x$ всё-таки можно попасть в соседнее дерево $y$, несмотря на то, что имеется $tree(x) \star tree(y)$ и предполагается, что весь регион под $x$ действительно не пересекается с регионом под $y$.
Нужно, чтобы без изменения данного предиката это могло оказаться невозможным. Однако, при дальнейшей параметризации и при манипуляции предиката (ср. предикат \texttt{tree} в главе \ref{chapter:logical}) проблема принципиально остается. Более наглядно это можно увидеть в рисунке \ref{ExampleSchemaHeapSeparation}.

На основе определения по Рейнольдсу, Бердайн \cite{berdine05} вводит соотношения выполнимости куч в опр.\ref{def:HeapSatisfactionRelation}.

 Денотационная функция $\llbracket . \rrbracket$ имеет  тип $\Phi \times \Sigma \rightarrow Bool$, где $\Phi$ множество утверждений о куче, а $Bool$  булевое множество. Для  линейного списка $s,h \models E_0 \mapsto t_1,\cdots , t_k$ с соответствующими типами $\forall i,j \in \mathbb{N}_0.E_j$, $r(t_i)=\llbracket E_i \rrbracket s$, $1 \le i \le k$. Слева от $\models$ записывается состояние вычисления, которое имеет тип $\Pi \times \Sigma$, справа имеется любое  булевое утверждение.

В \cite{berdine05-2} Бердайн указывает на проблемы, что  фрейм может дистанционно поменяться (см. набл.\ref{observation:RemoteAlternation}). Это является проблемой, которую предикатам необходимо учесть. Однако, удаление содержимого ячейки памяти, на которой  ссылается глобальный указатель, является отдельной проблемой (см. главу \ref{chapter:intro}). Аналогичное касается  подпроцедур, которые далее здесь не рассматриваются.  Необходимо отметить, что встроенные процедуры с точки зрения  выразимости в общем случае могут лишь усложнить  спецификацию и верификацию, но вычислимость они не увеличивают.
Бердайн справедливо обращает внимание на то, что применение правила фрейма в общем случае может привести к определённым трудностям в связи с  недетерминированностью сопоставлений  символов. Однако, когда речь идёт о кучах, недетерминированность можно ограничить наименованиями и дополнительными конвенциями (см. главы \ref{chapter:stricter},\ref{chapter:APs}). Более того, логические конъюнкты вписываются прямо в  язык логического программирования (см. главу \ref{chapter:logical}).\\

Далее уточняем  определение графа кучи согласно \cite{haberland16-2} и затем вводим  синтаксическое и семантическое обозначение согласно предыдущему анализу.

Неявные определения, например, в  «\textit{SSA}»-форме, пользуются успехом, несмотря на то, что чёткого определения, например  «\textit{зависимости данных}» нет и не нужно. В изначальных определениях также имеются неявные определения, которые касательно  куч уточняются в этой работе. Соотношения, пространственные операторы и частично неявное определение касаются вершины графа,  типизацию  символьных переменных и т.д. В главе \ref{chapter:stricter}  пространственные операторы куч  ужесточаются и для  константных функций вводится дополнительное обозначение для  объектов.\\

Из трм.\ref{theo:ReynoldsHeapProperties}, опр.\ref{def:finiteHeapGraphDefinition} и предыдущих конвенций (см. \cite{haberland16-2}) терм кучи может быть определен следующим минимальным образом:

$\underline{true}$ обозначает  тавтологию независимо от того, как выглядит данная куча. Обратное действительно для $\underline{false}$.  Предикат $\underline{emp}$ верный только тогда, когда данная куча пуста, во всех остальных случаях ложна. В главе \ref{chapter:stricter}  конъюнкция ужесточается и распадается на две операции. Для логических утверждений вводятся логические конъюнкции в рекурсивное определение $T$.

Логические конъюнкции «$\wedge, \vee, \neg$» не нуждаются в объяснении. Вывоз предиката подразумевает, что соответствующий предикат определен в $\Gamma$ (при опр.\ref{def:PredicateRuleSetDefinition} и закл.\ref{corollary:PredicateEnv}). При запуске предиката с целью сравнения с актуальной  кучей, все  свободные символы должны быть  унифицированы термами не содержащие свободные переменные, иначе данный запуск не определён (см. главу \ref{chapter:logical}).\\

Если в связи с  символьными переменными использовать  логический язык программирования, то ограничения как одностороннее присвоение и невозможность использования символа вместо значения и многие другие ограничения \cite{berdine05}, \cite{berdine05-2}, \cite{parkinson05}, \cite{parkinson05-2} можно будет снять. Если унифицировать термы, то сравнение простое и «\textit{дыры}» наполняются нужным содержанием, иначе нужно вручную все подтермы сравнивать и вставлять необходимые термы в нужные места подтермов. Это возможно, но необходимы дополнительные условия, вследствие чего, вводятся всё новые ошибки и ограничения. На практике ограничения наблюдаются в основном тогда, когда ради используемого нелогического языка упускаются полные сравнения или деградируется символьное использование полностью, как вызов по значению. Обычно это наблюдается в  императивных и большом количестве  функциональных языках программирования.

Возьмем к примеру вызов предиката (подробно о логическом представлении в главе \ref{chapter:logical})

\begin{center}
\begin{tabular}{c}
«\texttt{?-pred1(s(s(zero)),\_)}»
\end{tabular}
\end{center}

где ради простоты первый  терм входной, а второй  выходной. Запрашивается, существует ли таким образом, анонимная переменная «\texttt{\_}», чтобы предикат \texttt{pred1} был выполним для входящего терма \texttt{f(a)}? Если ответ верный, то подцель успешная и результат забрасывается. Если нет, то предикат \texttt{pred1} не соблюдается. \texttt{s(s(zero))} представляет собой целое число Чёрча «$2$». Если например, вместо \texttt{s(s(zero))} представить \texttt{s(s(\_))} и возможно предъявить результат, например\\
\texttt{s(s(s(zero)))}, то без изменений  запрос (см. опр.\ref{def:QueryToProlog}) можно поменять на\\
«\texttt{?-pred1(s(s(\_)),} \texttt{s(s(s(zero))))}», подразумевая, что предикат определён двунаправлено. Принципиально это касается не только двух, а нескольких направлений. Пролог имеет  строгий порядок присвоения и вычисления термов слева направо. Это означает, что символы могут замещать конкретные кучи, но если подцель потребует конкретную кучу, а куча присваивается только в одном из следующих подцелях, то можно, либо порядок подцелей поменять, либо вычисление не завершается успехом (см. главу \ref{chapter:logical}).\\

Теперь необходимо рассмотреть свойства отображения между определением куч и графом куч, а также свойства отдельных ссылок.

\textbf{Свойство 1 -- Корректность.}
 Если из синтаксического опр.\ref{def:HeapTermDefinition} следует строгое различие между связанной и  несвязанной кучей (см. главу \ref{chapter:stricter}), то  синтаксическое описание охватывает любой граф кучи, а также любой граф кучи может быть представлен данным синтаксическим определением. Если нормализовать согласно правилам трм.\ref{theo:ReynoldsHeapProperties}, например, по \textit{пренекс-нормальной форме}, то таким образом все, полученные формулы коммутируют.  Константные предикаты являются исключением, и поэтому являются односторонним укрупнением: множество выполнимых куч  отображается на один представитель множества. Не трудно убедиться в том, что такое отображение не  обратимое. Если ещё исключить  анонимные символы, то синтаксическое описание полностью соответствует  графу куч. Исключив единственные очаги  недетерминированного представления, легко убедиться в том, что отображение теперь  изоморфное, а, следовательно,  не могут быть выражены два различных графа куч из одного описания и наоборот.

\textbf{Свойство 2 -- Полнота.}
 Очевидно, что граф кучи полностью описывается  базисными кучами и аннотируется ссылками. Представление значения вершин графа может привести к  синтаксическому парадоксу, если не различать адреса ссылаемого объекта. Например, если $a \mapsto 3$ а также имеется $b \mapsto 3$, то на практике это вовсе не означает, что обе ячейки памяти содержавшие «$3$» идентичны.  Для моделирования это именно то и означает. Если именно новый объект не выделяется и указатель на эту же ячейку не ссылается, то указатель становится  псевдонимом. Избежать этой ситуации можно с помощью аннотации в объектном виде терма.

Указатели на указатели содержат  целое число как адрес, и поэтому не отличаются от других указателей. Различие между целым числом и адресом производится за счёт типа переменной. При отображении от графа к формуле порядок вычисления и структура предикатов естественно не могут быть выведены однозначно. Отображение может быть сгенерировано, но оно может различаться. Если допустить, что дан набор определений  абстрактных предикатов, то отображение в общем случае не решимо из-за  проблемы приостановки. Отображение от графа кучи к $T$ полностью определено, соблюдая упомянутые особенности.
Обратное отображение очевидно полное. Если в $loc\mapsto val$ левая сторона не указатель, то аннотацию можно принципиально решить дополнительным полем $f$: $loc\mapsto_{f} val$. Для общей структуры графа кучи дополнительная аннотация не имеет большого значения, поэтому дополнительные поля умалчиваются по определению.

\textbf{Свойство 3 -- Отношение эквивалентности.}
 Для сравнения может быть использовано, что «$\mapsto$» бинарный  функтор, а $a\mapsto b$ куча. Таким образом, можно убедиться в том, что отношение эквивалентности «$\sim$» может быть обосновано, показав свойства (i-iii). Оговаривается, что «$\mapsto$» имеет выше приоритет присваивания, чем «$\sim$».  (i) Рефлексивность: $a\mapsto b \sim a \mapsto b$. (ii) Симметричность: $a \mapsto b \sim c \mapsto d$, то $c \mapsto d \sim a \mapsto b$. (iii) Транзитивность: $a \mapsto b \sim c \mapsto d$ и $c \mapsto d \sim e \mapsto f$, то $a \mapsto b \sim e \mapsto f$.

\textbf{Свойство 4 -- Локальность.}
 При удалении, изменении, добавлении  указателя или его содержимого, граф кучи и соответствующий терм меняются минимально. Удаление грани приводит к редукции на одну базисную кучу, либо к  параметризации или изменению используемых  абстрактных предикатов, т.к. соотношение между обоими, формулой и графом, в общем случае не решимо (см. ранее), поэтому, необходимо предикаты рассматривать отдельно. В свободном от предикатов случае максимальное изменение может повлечь за собой удаление  вершины, потому, что это означает удаление вершины и всех  граней, которые с этой вершиной связаны. Аналогичное добавление не является «\textit{сложной}», потому, что добавление грани производится пошагово. В худшем случае для предикатов --- использование предикатов может привести к формуле, которая совершенно не похожа на предыдущую. Поэтому,  рекурсивные структуры эффективно описываются рекурсивными описаниями и изменение одного элемента не затрагивает остальные элементы, кроме соседних. В общем случае, это лишь  эвристика и зависит от конкретного алгоритма и от делимости графа на наиболее независимые подграфы. Также как и эвристика: «\textit{граф кучи описывается абстрактными предикатами лучше, чем компактными описаниями}». Делимость проблем остается общей проблемой далеко за пределами этой работы, но возможность делить кучи на «\textit{удобные}» кучи. Это интересно с точки зрения подключения  логических решателей (см. главы \ref{chapter:stricter}, \ref{chapter:APs}, см. \cite{distefano06}).\\

В заключение этого раздела, хотелось бы сказать о  предикатах высшего порядка.  Они не имеют теоретическую значимость для сравнения куч с практической точки зрения. Принципиально предикаты, которые используют предикаты в качестве параметра, интересны с точки зрения  выразимости простых и коротких описаний. Однако, рассматриваемые предикаты  куч имеют  индуктивное определение, а произвольные предикаты высшего порядка в состоянии сломать эти свойства, если не вводить дополнительные ограничения. Поэтому, считается не целесообразно использовать иное, чем  формально-грамматическое представление (см. главу \ref{chapter:APs}). Нужно отметить, что при включении  предикатов высших порядков меняются в общности существенные свойства доказательств, как например, поток и вызовы с продолжениями, порядок вычисления и обработки  подцелей, а также статическая  типизация, но с точки зрения выразимости кучи ничего не меняется, по крайней мере после обсуждения такая необходимость отпадает.

%%%%%%%%%%%%%%%%%%%%%%%%%%%%%%%%%%%%%%%%%%%%%%%%%%%%%%%%%%%%%%%%%%%%%%%%%%%%%%%%%%%%%%%%%%%%%%%%%%%%%%%%%%%%%%%%%%%%%%%%%%%%%%%%%%%%%%%%%%%%%%%%%%% 4 Prolog
%%%%%%%%%%%%%%%%%%%%%%%%%%%%%%%%%%%%%%%%%%%%%%%%%%%%%%%%%%%%%%%%%%%%%%%%%%%%%%%%%%%%%%%%%%%%%%%%%%%%%%%%%%%%%%%%%%%%%%%%%%%%%%%%%%%%%%%%%%%%%%%%%%%
\section*{Логическое программирование и доказательство}

В этой главе рассматривается вопрос, как куча (см. главу \ref{chapter:expression}) может быть представлена в Прологе, а затем теоремы о кучах логически выведены с помощью Пролога. Для этого берётся Пролог и определяется синтаксис термов и правил, обсуждается отсечение и применимость рекурсивных определений. Это способствует логическому выводу, который будет решать задачи верификации. Более детально анализируется декларативный характер абстрактных предикатов, который в главах \ref{chapter:stricter} и \ref{chapter:APs} используется для сближения языков спецификации и верификации. Представление языков с помощью Пролога выявляется на основе выразимости реляций, а также на основе метрик. Предлагается система, основанная на Прологе. Представление и интеграция объектных экземпляров обсуждается.

\subsection*{Пролог как система логического вывода}

Данный раздел не представляет собой введение в  Пролог, а лежит в основе построения архитектуры верификации куч на основе Пролога. Тем не менее, раздел ссылается на первоисточники Пролога, как \cite{sterling94} и \cite{bratko01}, которые рекомендуется изучить досконально, прежде чем, читать далее. Пролог уже оправдался решением трудных теоретических и практических проблем. Например, при доказательстве теоремы Фейгенбаума \cite{koch96}, при доказательстве ошибки с делением  вещественных чисел в процессорах «\textit{Intel Pentium}» первого поколения с помощью логических диалектов  «\textit{ACL2}» \cite{kaufmann00}/«\textit{HOL Light}» \cite{price95} или при обработке и проверке  слабоструктурированных данных \cite{haberland08-1}.

Программа в Прологе задаётся  базой знаний, которая задаётся  правилами Хорна. Запрос к базе знаний можно задавать одной или более  подцелями. Для определения правил и подцелей необходимо ввести определение прологовского выражения терма.

 Символ представляет некоторый логический объект, например: «\textit{я}», «\textit{дед мороз}», число «$33$» или некоторая  локальная переменная. В  Прологе символ начинается всегда с маленькой буквы, а далее следуют любые буквы или цифры.

В отличие от символа,  символьная переменная всегда начинается с большой буквы, либо является зарезервированным знаком  «\_».
Например, переменной $X$ присваивается  (\textit{унифицируется}) некоторое значение, например «$33$», после чего $X$ может быть использован далее в сложных термах или подцелях (см. опр.\ref{def:PrologRule}). Когда используется «\_», тогда ссылаться на это же значение будет не возможно. Поэтому, «\_» используется исключительно в тех случаях, когда должен приниматься ровно один терм, но дальнейшее использование этого значения не предусмотрено, как это обычно бывает в случае  \textit{сопоставления с образцами} термов.  Видимость переменных ограничивается данным правилом.

 Списки, которые определяются с помощью бинарных  функторов «\texttt{,}» или «\texttt{|}» должны иметь хотя бы два компонента. Проверка, является ли список линейным, проводится самими  предикатами, которые принимают термы. Примерами списков являются: \texttt{[12|[]]} или \texttt{[1,2,[4|5]}. Таким образом, базисный тип списков является записью.

Функтор является структурным оператором, который связывает термы и обозначает новое сложное значение. Значение может быть символьным, аналогично  объектному экземпляру или записи. Например, функтор списка конструктор  «\textbf{.}», который, применив к голове $H$ и  списку $Hs$ работает аналогично $[H|Hs]$. Иной пример, это наследник  натурального числа $succ$, который имеет арность $1$, либо принимает терм Чёрча по  арифметике, который опять же снаружи имеет функтор $succ$, либо терм $zero$ с  арностью $0$, т.е. является  константной (см. набл.\ref{obs:SimplificationByGeneralisation}).

На вопрос «\textit{да/нет}» предикат даёт ответ в зависимости от того, связаны ли некоторые термы согласно данному соотношению или нет --- если  предикат  тотален (см. следующие главы). Например,  высказывание $older(plato,aristotle)$ означает   утверждение «\textit{Платон старше Аристотеля}».

Синтаксис правила $pred$ в расширенной форме  Бэккуса-Наура можно определить как:

\begin{center}
\begin{tabular}{l}
  \begin{minipage}[t]{7.3cm}
  \begin{grammar}
<head> ::= <ID> ‘(’ <term> \{ ‘,’ <term> \} ‘)’

<rel> ::= ‘=’ | ‘!=’

<call> ::= <ID> ‘(’ <term> \{ ‘,’ <term> \} ‘)’
  \end{grammar}
  \end{minipage}\\\\
  \begin{minipage}[t]{6.5cm}
  \begin{grammar}
 <goal> ::= <term> <rel> <term> | <call>

 <body> ::= \{ <goal> ‘.’ \} <goal>

 <pred> ::= <head> [ ‘:-’ <body> ] ‘.’
  \end{grammar}
  \end{minipage}
\end{tabular}
\end{center}

Здесь $ID$ идентификатор, который является символом, но не символьной переменной. Бинарные операторы «\texttt{=}» и «\texttt{!=}» обозначают унификацию, либо утверждение о невозможности унификации данных термов. «\textbf{.}» обозначает конец определения правила. Символьные переменные видны только внутри одного определения правила. Разделитель «\textbf{:-}» определяет голову $head$ от тела $body$ правила.

Для иллюстрации, рассмотрим примеры из рисунка \ref{fig:PrologExample1}. Рисунок \ref{fig:PrologExample1} a) содержит два очевидных факта из древнегреческой мифологии: (1) Сократ человек и (2) Цевс бессмертен. Аналогично можно определить иные факты, за достоверность, за что отвечает создатель  базы знаний. Фактом по умолчанию является только неоспоримое  утверждение из проводимой области дискурса. Третье правило гласит: «\textit{любой человек смертный}». Технически более подробно это означает: если некоторый терм $X$ имеет предикат «\textit{человечно}», то терму $X$ безусловно приписывается предикат «\textit{смертно}». «\textit{Безусловно}» подразумевает в прямом смысле отсутствие дальнейших подцелей, кроме подцели $human(X)$.

 Унификация термов из опр.\ref{def:PrologTerm} является в частном случае  инфикс-нотации, опираясь на \texttt{rl}. Более обобщённое соотношение записывается с помощью предиката $p(T_0, \dots, T_n)$. Далее, рисунок \ref{fig:PrologExample1} b) на примере  \textit{функции Аккерманна} демонстрирует, что любая   рекурсия (левая, правая, примитивная,  взаимная, и т.д.) может быть определена в Прологе --- не только примитивная. Оператор  \texttt{is} является арифметическим функтором. Функторы по Карнапу \cite{carnap68} являются арифметическими вычислениями термов и не являются логичными. Функция Аккерманна является определённой и  тотальной, однако, на практике, из-за ограничения памяти и операционных средств вычисление может быть прервано. По этой причине, предикат полученный из функции может оказаться полезной для проверки терминации алгоритмов для индуктивно-определённых  входных данных. Примером могут послужить  натуральные числа или  списки.

 Подцели $goal$ как пересчитываемые подусловия выполнимости предиката обозначаются следующим образом:

 Интроспекция в  Прологе \cite{diaz12} разрешает проверку и определение  типы классов, так называемые «виды», данного терма, которые могут быть:  \texttt{var},  \texttt{atom},  \texttt{number},  \texttt{compound} и иные мало значимые  встроенные  предикаты.
Слияние общих случаев термов к одному виду разрешается, а следовательно, предикаты могут сильно упростить определение предикатов. \texttt{atom} проверяет свойство символа, например \texttt{atom(a)} или \texttt{atom([])} верны, но \texttt{atom([1])} или \texttt{atom([1,2])} не верны.  \texttt{var} проверяет, является ли данный терм  символьной переменной, которая  свободная. Поэтому, \texttt{var(X)} верно, но \texttt{X=1,var(X)} не верно. \texttt{list} проверяет, является ли данный терм  (слабо-типизированным) списком (как это принято считать в  Прологе, ср. \cite{diaz12} с \cite{isocpp14}). Таким образом,  \texttt{list([])} или \texttt{list([1,2,3])} верны, а \texttt{list(a)} не верен. Необходимо отметить, что для определения свойства структуры списка, также как и для  встроенного предиката  \texttt{compound}, используется встроенный предикат  «\texttt{=..}», который разбивает данный  функторный  терм на сам функтор и на список передаваемых термов.

Вызов предикатов рассматривается чуть позже. Однако,  унификация термов по определению не производится в  Прологе по умолчанию из-за необходимости полного анализа всех  подвыражений термов. Поэтому, можно по определению унифицировать \texttt{X=X}, но \texttt{X=f(g)} нельзя. Пролог, в зависимости от реализации, может замечать  рекурсию в термах, т.к. иногда производится унификация исключительно на верхнем термовом уровне. Но, в примере \texttt{X=f(g(X,X))} часто унификация приведёт к провалу или приостановке в лучшем случае, а в обычном случае к выходу  WAM из строя (из-за переполнения стека в версии 1.3.0 «GNU Prolog» \cite{diaz12}). По осторожным оценкам на практике в 95\% всех случаев не требуется проверка  унифицируемости термов, однако, в общем случае необходимо рассматривать именно это, когда речь идёт об обобщённых утверждениях оценок стабильности, например преобразователей. Также нужно рассматривать границы ресурсов, например с помощью  предиката Аккерманна. Поэтому, далее даётся алгоритм, который позволит полностью исключить те 5\%, которые корреспондируют с проблемой, например, с определением объектного экземпляра из раздела \ref{sect:TheoryOfObjects}, см. рисунок \ref{CodeUnificationWithOccursCheck}.

 Термы структуры и списки необходимо определить как изложено в рисунке \ref{CodeUnificationList}.

В главе \ref{chapter:APs} и на рисунке \ref{fig:mapFunctionalExample} дан пример, когда предикат используется в качестве передаваемого терма. Стоит заметить, что аналогично анализу  функтора,  вызов предиката также производится для данного  стекового окна. Разница заключается в представлении вычислительной модели реализации Пролога и синхронизации стека (см. следующий раздел) между  вызванной и  вызывающей сторонами.

Далее рассмотрим пример из рисунка \ref{fig:PrologExample1} b) с  подцелями и их  вызовами. Дана программа и дана подцель «\texttt{?-a2(X,1,3)}» (см. рисунок \ref{fig:PrologExample2}), которая задается интерактивным интерпретатором. Сначала берётся правило № (1) с сопоставлением $\sigma_1$. Это приводит к противоречивости,  т.к. \texttt{Res=2}. На рисунке \ref{fig:PrologExample3} определена система сопоставлений для данного примера дерева вывода). Поиск по данному предикату приостанавливается, т.к. завершился провалом и берется следующий  альтернатив (см. набл.\ref{obs:ProofAsSearching}). Выбрав третье правило с сопоставлением $\sigma_2$, а затем, применив второе правило с $\sigma_5$ и наконец, применив первое правило с $\sigma_9$, успешно находит решение. Так как  отсечение (см. далее) не производится, проверяются также  альтернативы и наконец находится ещё одно положительное решение данной  подцели. Путь от начальной подцели через все новые подцели до успешного решения, которые на рисунке обозначены чёрными квадратиками, назовём «\textit{путь к решению}». Для решения данного примера существуют два пути, при этом, результаты не отличаются друг от друга. Нужно обратить внимание, что если даже пути к решению доказательства различны, то после вычисления первого, решение можно было бы полностью приостановить. Однако, это нельзя обобщать.

Рассмотрев пример из рисунка \ref{fig:PrologExample3}, замечаем, что  стратегия вычисления может кардинально измениться, если в  декларативной парадигме использовать механизм, который позволит  отсекать и менять путь расследования, по крайней мере, по каждому уровню   подцели. Например, если приостановить логический вывод на примере из рисунка \ref{fig:PrologExample2} после подцели «\texttt{?-a2(N1,1,Res2)}» и мы уверены в этом, то всю ветку $(3) \cdot \sigma_3$ можно не проверять. Для этого мы вводим определение  отсечения.

Отсечение не означает, что любой алгоритм станет эффективнее или \textit{лучше}, это вовсе не так. Отсечение даёт лишь возможность, в зависимости от предложенного описания решения проблемы, ускорять нахождение предложенного решения, либо исключать  дубликатные решения. Отсечение вставляется в тело предиката как  подцель. Рассмотрим теперь  семантику «\texttt{!}» на примере  факториала в предикате \texttt{fact} из рисунка \ref{CodeExampleFactorial}.

Первый терм представляет входное  целое число, второй  терм представляет результат факториала. Предикат представляет соотношение между  входным и  выходным вектором, но в данном примере определён только порядок \textit{вход-на-выход}. Предикат \texttt{fact2} почти ничем не отличается от первого предиката, разве что, в базисном случае имеется  отсечение как единственная  подцель бывшего факта. Отсечение приводит к тому, что первый альтернатив \texttt{fact2} только тогда вызывается, когда  входной и  выходной терм  унифицируются с правилом при  запросе согласно опр.\ref{def:QueryToProlog}, например, «\texttt{?-fact(5,R).}» --- это происходит лишь один раз при  индуктивно-определённом спуске, когда входное число равно нулю. Когда эта ситуация возникает,  отсечение гарантирует, что во всех обработанных  подцелях альтернативы далее не обыскиваются. При  рекурсивном подъёме альтернативы дальше не рассматриваются. Если при рекурсивном подъёме нет совпадающей подцели, то альтернативные правила должны рассматриваться, кроме отсечённых альтернатив. Поэтому,  отсечение в примере \texttt{fact2} не лишнее. Оно улучшает понятие вычисления  подцелей и ее оптимизацию. Обратим также внимание на то, что порядок правил сильно влияет на  корректность. Например, поменяв порядок правил или  отсечения, см. рисунок \ref{CodeExampleFactorial2}.

В этом примере  базисный случай не достижим потому, что 0 и 1 также  унифицируются в первом альтернативе. Отсечение приводит к тому, что второй альтернатив не будет рассмотрен. Далее \texttt{N1} присваивается -1 при вызове «\texttt{?-wrong_fact2(0,1).}». В итоге, имеется левая  рекурсия, которая никогда не  остановится. В данном примере  отсечение даже не важно, т.к.  левая рекурсия не разрешает рассмотрение второго альтернатива, хотя все логические  факты и правила явно определены. Дерево вывода имеет ветвь, которая бесконечна. В Прологе эта проблема лучше известна, как  проблема приоризации. Поэтому, ради читаемости правила должны определяться не по любому порядку, а строго по порядку базисного определения: частные и специализированные правила и факты должны определяться первыми, иначе, они могут быть поглощены более общими случаями.

С другой стороны, не поглощённые правила означают, что правила альтернативы могут меняться по позициям. Ради простоты программы, предпочтение всегда должно быть за простыми правилами и одновременно без дополнительных условий касательно порядка. Для достижения этого, необходимо, насколько это возможно, высчитывать входные вектора (см. позже касательно  высчитывания куч).

Нужно отметить, что  отсечение является  синтаксическим сахаром, потому, что отсечение всегда возможно удалить из  предиката (см. дискуссию из главы \ref{chapter:intro}, переписав его). Доказать корректность исключения отсечения возможно средством дополнительного аргумента предиката, который запоминает достижимость результата. Если результат достижим, то имеется  базисный случай, который должен возникать только один раз. Во всех остальных случаях, имеется строго  возрастающая цепочка. Поиск результата может быть обобщён и выражен как  поиск плавающей точки для $\mu$-выражения.
 В логиках, в которых допускается отсечение, соблюдается правило (CUT):

\begin{center}
\begin{tabular}{c}
 \inference[(CUT)]{\Gamma \vdash \Phi,\Omega \quad\Phi,\Delta \vdash \Lambda}{\Gamma, \Delta \vdash \Omega,\Lambda}
\end{tabular}
\end{center}

В контексте  Пролога верхние  консеквенты формой $A \vdash B$ могут рассматриваться как $A$ подцель  тела правила, а $B$  голова предиката, т.е. логический консеквент (следствие). Для дальнейшего обсуждения не будут использоваться консеквенты.
 Отсечение, возможно, рассматривать, как попытку  \textit{детерминации}. При детерминации множество возможностей отсеивается, так именно и происходит при  «!». Отсечение не гарантирует терминации (см. ранее упомянутый пример) и полной детерминации, потому, что всё равно могут возникать альтернативные ответвления, например последующих «\texttt{!}» подцелях. Когда отсечение происходит намеренно и не угрожает потерям годных решений, то речь идет о  «\textit{зелёных отсечениях}»  дубликатных решений, в противном случае о  «\textit{красных отсечениях}».

Обсуждённые решения, ограничения, моделирование правил и подцелей будут применены позже, в частности в главе \ref{chapter:APs}.
В связи с вопросом об  отрицании, необходимо обсудить вопрос представления  литералов, т.е.  утверждений и  отрицание утверждений (если будет необходимо для определения выразимости утверждений о кучах). Когда речь идёт об отрицании, необходимо рассмотреть выражения, которые строятся термами на основании опр.\ref{def:PrologTerm}. 
Отрицание переменной может создать практическую проблему, т.к. оно определяется контекстом. В отличие от единственного случая, отрицание множества (бесконечного, но индуктивно-определённого), в общем случае не решимо. Нужно отметить, что метод представленный в главе \ref{chapter:APs} --- решим, если даже допускаются  индуктивные определения, благодаря обсуждённым свойствам, в частности, из-за \textit{неповторимости} и локальности куч. Очевидно возникает похожая проблема нерешимости для предикатов. Ситуация для функторов в общем однозначно не определена: если дана функция, которая отображает из домена в определённую ко-область, тогда, каким образом можно универсально определить отрицание? --- Это лишь гипотетичный вопрос потому, что отрицание по определению применимо только к  булевым множествам. Следовательно, необходимо определить  отрицание для  предикатов в  Прологе.

Отрицание в Прологе определяется как  провал, т.к.  Пролог ищет решение и должен его найти только отрицательным. Поэтому, в Прологе согласно обработке (см. следующий раздел) вводится на уровне тела  встроенная зарезервированная подцель \texttt{fail}, которая сигнализирует прекращение поиска альтернатив и возврат к предикату вызова. В общих случаях уговаривается явное прологовское определение отрицания предикатов.

%%%%%%%%%%%%%%%%%%%%%%%%%%%%%%%%%%%%%%%%%%%%%%%%%%%%%%%%%%%%%%%%%%%%%%%%%%%%%%%%%%%%%%%%%%%%%%%%%%%%%%%%%%%%%%%%%%%%%%%
\subsection*{Логический вывод как поиск доказательства}

С помощью  Пролога было продемонстрировано множество применений в области обработки  языков, формальных и  естественных \cite{pereira12}, \cite{matthews98}, \cite{kulik01}.
Как будет показано позже, слово данного языка необходимо сгенерировать и проверять согласно правилам.
Естественные языки характеризуются мутациями \cite{pereira12}, \cite{kallmeyer10} в связи с культурно-социологическим развитием. Таким образом, они сильно отличаются от  формальных и формализованных языков. Многозначность  является основной разницей между естественными и формальными языками. Тем не менее, изучены методы распознавания естественных предложений \cite{kallmeyer10}, \cite{pereira80}, которые приводят к интересным итогам в связи с вопросами  искусственного интеллекта и  обработки человеческой речи, в общем. Однако, данная работа, сосредоточена на сравнительно простом уровне  лексики и растолкования. Задача будет заключаться в нахождении максимально простой и адекватной модели  динамической памяти и её проверки. Для реализации прототипа, а также для расследования ключевых задач, будут рассматриваться  правила Хорна. Предложенный  Уорреном подход «\textit{использование предикатной логики для решения логических задач}» воплощается частично в жизнь, хотя Уоррен сам предупреждает о том, что его реализация Пролога не в состоянии полностью отобразить логическую «\textit{мощь}» предикатной \cite{kowalski74} логики, из-за теоретических ограничений в связи с  вычислимостью. Простой пример проиллюстрирует одну фундаментальную сложность --- допустим, имеется функция, где  входной и  выходной вектор записывается как предикат, т.к. это может быть реализовано в Прологе. Тогда, возможно, для данного выходного вектора, не решимо или очень трудно решить обратный результат, например в случае  разовых функций для шифровки сообщений, либо функции, которые не полностью определены (в обе стороны) и имеют, например, двухмерное отображение в качестве эллипса.
 \cite{pereira80} расследует подробно методы  перебора, но, увы, теоретические ограничения синтаксического анализа не принимаются (полностью) к сведению, поэтому перебор оказывается сильно ограничен, в связи с распознавателями категории  LL(1) или слабо расширенные от него анализаторы. Это очень ограничивает выразимость в целом.
\cite{haberland08-2}, \cite{haberland07-1} показывают, что если речь идёт о генерации и проверке данных (в статье представлены исключительно как термы), то:

\begin{enumerate}
 \item  Синтаксический анализ является непосредственно ограничивающим фактором выразимости. Если выражения расширены и искусственные ограничения, а также упомянутых, а также и существующих подходов исключаются, то синтаксис и  семантика всего процесса проверки могут быть сильно упрощены без потери общности.
 \item Упрощение также может быть достигнуто ради обобщённой формы  промежуточного представления. Однако, в статье речь не идёт об  императивных языках.
 \item В частности, из-за выбранной   в статье функциональной парадигмы для сравнения, возникают проблемы с незначительно раздутым описанием. Отметим, что  функциональное описание  состояния вычисления всё равно удобнее и проще любого  императивного представления. В общем используются лишь 5 функционалов  высшего порядка, как например,  \texttt{concatMap},  \texttt{foldl}. Если бы использовалось логическое представление, то можно было бы семантику процесса проверки более близко сформулировать в стиле математико-логических термов, а это более важно в этой работе.
\end{enumerate}

\cite{haberland08-1}, \cite{haberland07-1}, \cite{haberland07-2}, \cite{haberland08-2} и \cite{haberland19-1} показывают, что на первый взгляд нестандартный путь, может привести к очень простому, но эффективному решению проблемы проверки с помощью термов. Причина разносторонняя: представление трансформации термов в виде правил и компактность (проводилась широкая экспертиза по отношению к иным  функциональным языкам, т.к.  императивные языки сильно ухудшают все показательные метрики измерения качества и компактности записи). Причина простоты лежит в сравнении данной программы и её описания, которая вписывается в  \textit{регулярное выражение} в более обобщённом виде, даже в  контекст-зависимое выражение, в частных случаях, как было продемонстрировано на примерах вызовов процедур. Проводятся качественные анализы по каждому из входных операторов и количественные сравнения, основанные на  метриках. Полученный результат неожиданный потому, что почти по всем сравниваемым параметрам представление и решение в  Прологе торжествует к сравнению с   функциональным представлением \cite{haberland08-1}. Далее результаты обсуждаются в разделе \ref{sect:LanguageCompatibility}.

%%%%%%%%%%%%%%%%%%%%%%%%%%%%%%%%%%%%%%%%%%%%%%%

Предпосылкой этому тезису является то, что метод(-ы)  синтаксического перебора позволяет решить проблемы доказательства  динамических куч. Как это решить реально, какие для этого имеются условия и почему должны соблюдаться, объяснения даются в этой работе.

Позже вводится модифицированный диалект  Пролога, который способствует к решению целого ряда проблем в связи с  автоматизацией доказательства  теорем (ср. главу \ref{chapter:intro}). Таким образом, язык похожий на Пролог, квалифицируется не только как язык, а как теоретический метод и инструмент к целому ряду проблем. Особенность  Пролога заключается в том, что доказуемо только то, что выводимо (см. набл.\ref{obs:ModelOfComputationVerification}). Если что-либо не выводимо, то причина может лежать в правилах, в представлении правил (например, в  абстракции), в ограниченности основной выводимости, а также в формулировке  подцелей. Также могут иметься практические ограничения со стороны  синтаксического анализа.

Идея использовать синтаксический анализ, в качестве  решателя куч, появилась после некоторых наблюдений (см. главу \ref{chapter:APs}):
\begin{enumerate}
 \item  Прологовские правила довольно компактны и очень хорошо приспособлены к решению логических задач.
 
 \item Программа Пролога представляется правилами. Вывод правил генерирует  дерево. Доказательство  теорем имеет структуру аналогично программе (изоморфизм  обсуждённый ранее), это же относится к запросу или запуску программ, которые тоже генерируют дерево.
 
 \item При необходимости выражения можно  синтаксически анализировать.  Локальность ссылочных моделей памяти ограничивает  выразимость соответствующих генерируемых  грамматик так, что ранг класса формальной грамматики не должен превышать ранг  контекст-свободных грамматик. Это очень полезно с практической точки зрения.
 
 \item Однажды, на одном мало интересном докладе пропагандист  «\textit{Semantic Web}» пробовал рекламировать свой новый подход, как подход будущего, как минимум в следующие 25 лет. В качестве защиты использовался аргумент, что «\textit{Semantic Web}» решит почти все актуальные проблемы, технические и научные. Также упоминалось, что «\textit{Semantic Web}» высшие по  абстракции и устаревшие методы, как например, синтаксические, давным-давно решены и не имеют будущего. Как я сегодня понимаю и думаю --- докладчик немного преувеличил, если даже считать, что синтаксический перебор можно действительно рассматривать как очень хорошо изученным. Несмотря на это, всё равно синтаксические методы могут пригодиться и сегодня, а также помочь решить реальные проблемы в области автоматизации доказательств.
\end{enumerate}

Теперь необходимо заметить и сравнить цели дисциплин:

\begin{center}
\begin{tabular}{l}
 1. Программирование\\ $\longrightarrow$ нахождение решения\\
 2. Логическое программирование\\ $\longrightarrow$ нахождение (логического) решения\\
 3. Синтаксический анализ\\ $\longrightarrow$ нахождение (синтаксического) решения
\end{tabular}
\end{center}

Если отображения пунктов 2 и 3 типизировать, то получаем характеристики из рисунка \ref{CharacteristicsTyping23}.

То есть, если удастся сблизить входные множества, то решение одним подходом может быть будет сопоставлено решению другим подходом.

Упомянутые ранее подходы для распознавания  естественных языков отличаются от  формальных языков тем, что в них имеется некоторая степень неверности или отклонений. В рассматриваемых языках, отклонения практически отсутствуют.

Согласно лем.\ref{lem:DoubleSemanticsOfProlog}. Семантика (1) по  вызовам  (\textit{декларативное свойство предиката}) означает, что предикат рассматривается строго как процедура  со стеком, которую можно выразить. То есть, создаётся  стековое окно (см. рисунок \ref{fig:ExampleStack2}) при вызове предиката, а после вызова окно  утилизируется. Аналогично  \textit{параметрам по ссылкам} в  императивных языках программирования,  унифицируемые термы, которые используются в различных слоях  предикатов (т.е. которые расположены в последовательных стековых окнах -- нарушение последовательности без глобальных хранилищ приводит к потере термов) передаются вызывающим предикатам. В отличие от императивных языков программирования, переменный  символ может передаваться не только непосредственно, но а также  как частица сложного  унифицируемого терма. Это означает, что передача из одного слоя предиката в другой по умолчанию, может практически привести к построению сложного  объектного экземпляра, за нулевую стоимость и преобразовать довольно сложный процесс трансформации объектов, что очень замечательно. Обратим внимание, что благодаря анонимному оператору  «\_», встроенному предикату  \texttt{concat} и другим  предикатам высшего порядка, нетрудно создать мощную  выразимость. Выразимость резко повышается ради простоты выражения и сравнения данных частей  функтора и структур в общем (включая  списки), потому, что исключается необходимость определять весь перечень элементов списка, но, в отличие от  сопоставления с образцами, остаток не теряется, а лишь остаётся без изменений, что упрощает описание простых операций. Напомним, что  верификация  куч происходит  пошагово после каждого  программного оператора, обычно меняется только маленькая часть всех куч.

 Вызов предикатов из рисунка \ref{fig:BoxModelPredicateCall} осуществляется согласно модели «\textit{чёрного ящика}», который можно отследить в Прологе, включив трассировку. Вызов предиката, т.е.  подцели, обрабатывается независимо от содержимого предиката и имеет четыре возможных результата: (а) $call$, вызов, создаётся  стековое окно, передаются присвоенные  термы, (б) $exit$, возврат при успехе, когда все подцели тела выполнимы и  голова предиката  унифицируема, (в) $fail$, провал, когда хотя бы одна  подцель  невыполнима, или (г) $redo$, повторная попытка, когда (в) происходит в одной  подцели, то находится альтернатив данному предикату, если он существует. Опасность не аккуратно сформулированных предикатов является не эффективной из-за широкого перебора альтернатив. Другой экстремальный случай возникает, когда применяется  отсечение настолько часто, что годные результаты исчезают. Оба случая необходимо устранять. Предикаты \texttt{call} и \texttt{fail} из рисунка \ref{fig:BoxModelPredicateCall} действительно существуют в  «\textit{GNU}»  Прологе \cite{diaz12}, пятый случай  «\textit{исключения}» удалён из-за причин, обсуждённых в главе \ref{chapter:intro}.

Семантика (2) интерпретаций предикатов подразумевает, что предикат является генератором недетерминированности. Общий  предикат может иметь (как обсуждено в главах \ref{chapter:intro} и \ref{chapter:APs}) а) различные определения для различных комбинаций  входных и  выходных векторов, а также б) различные выводимые результаты для одной  подцели. Например, если имеется  предикат \texttt{mortal} из рисунка \ref{fig:PrologExample1} (a) и к нему ещё добавить несколько  фактов, то согласно б) из подцели «\texttt{?-mortal(X)}»  выводимы следующие результаты: \texttt{X=socrates}, \texttt{X=plato} и т.д. Если для  функции Аккерманна вызывать «\texttt{?-acker(0,X,15)}» или «\texttt{?-acker(1,1,Res)}», то в зависимости от реализации, предикат, либо определён (но, возможно, очень долго занят или не подлежит практическому вычислению в связи с ограничением памяти), либо не (полностью) определён, как, например, на рисунке \ref{fig:PrologExample1} (b). В примере для данного  выходного \texttt{Res} предикат не определяет входной вектор, хотя математико-логическое определение явно существует согласно данной схеме, которая практически без изменений стоит в данных прологовских правилах (см. набл.\ref{obs:ProofAsSearching},набл.\ref{obs:ModelOfComputationVerification}).
Дискуссия касательно этого ограничения, также в связи с обратимостью и представлением  реляционной модели, продолжится позже в данной главе.

Как уже упоминалось,  Пролог не является каким-то «\textit{универсальным решением всех задаваемых логических проблем}» априори \cite{warren83}, \cite{warren84}. Это правда, что не всякое доказательство можно решить прологовскими подцелями в связи с ранее обсуждёнными фундаментальными ограничениями. Однако, обратное всегда возможно. Задача заключается в минимизации проблем, которые будут лучше решаться в связи с сближением языков. Подмножество  языков Пролога используется для решения целого ряда проблем  верификации в связи с представлением и методом  логического вывода. Охарактеризуем подцели и аргументы термы:

\begin{itemize}
 \item  Встроенные предикаты, как например  \texttt{concat}, могут быть использованы везде в программах. Кроме того, имеются предикаты, которые могут добавляться в  формальные теории, написанные в  Прологе. Зарезервированные предикаты и функторы могут быть определены мульти-парадигмально.
 \item Отдельное и  перегруженное значение. Отдельные предикаты могут быть перегружены так, что предикат с подходящей  арностью, который определен ранее, используется первым. Этот механизм позволяет подключение  прологовских теорий, в том числе  \textit{мульти-парадигмальность}. Полученный эффект, это  расширяемость,  вариабельность, гибкость и  полиморфизм. Кроме того, перегруженные предикаты, как например \texttt{concat}, позволяют гибкий порядок сопоставлений входных и выходных термов.
 \item Операции над  кучами являются термами. Вместе с  объектными экземплярами, соотношения местоположения куч может быть представлено термами. По сути, граф кучи представляется полностью термом кучи. Термы являются аргументами подцелей, т.е. (под-)\-доказательств теорем и лемм о кучах.
\end{itemize}

Из универсального описания механизма вызовов из рисунка \ref{fig:BoxModelPredicateCall} следует  \textit{поиск с возвратом} (с англ. «\textit{backtracking}»). Эта стратегия просто реализуема   рекурсией над предикатами, и точно описывает деревья  логического вывода (см. контекст-свободность прологовских правил с главой \ref{chapter:APs}). Описанные ранее ограничения необходимо обязательно учесть. Когда поиск завершается успехом, то ветка считается доказанной. Когда ветка  опровержима, то она, либо не доказана или не доказуема в принципе, либо приводит к явному противоречию. Оба варианта достаточны для установления провала (под-)\-доказательства, т.к. всегда необходимо доказывать универсальность данной формулы. Если интерпретировать предикаты как  декларативные, то предикаты можно с легкостью растолковывать как «\textit{процедуры}» в императивных языках программирования, т.е. если даже это формально не совсем совпадает, то интуиция совпадает с подходом в решении задачи. Следовательно, можно в действительности рассматривать логическое программирование в Прологе максимально приближено к доказательству. Хотя специфицируются и доказываются свойства о программе, входная программа при этом не отлаживается и не запускается.
 Как будет показано в более поздних главах, особенность Пролога хорошо распространяется именно для решения постановки проблем  динамической памяти. Представление модели куч, абстракции правил, рекурсивные предикаты, логические  символы в формулах в итоге увеличивают  выразимость. Например, в отличие от обобщённых формул логики  предикатов первого порядка, нам удастся найти компромисс таким образом, чтобы построение графа кучи проводилось последовательно. При этом, описание графа проводится \textit{снизу-вверх}, используя отдельные  простые кучи. Логическое представление выражений в формулах разрешает широкую гибкость самих куч, но а также при реализации стратегий над кучами, в том числе и определения правил вывода.

Структура доказательства является  деревом, если получается цикл внутри доказательства, тогда получается предыдущее состояние вычисления, которое означает ввод регрессии. Регрессия является индикатором того, что применённое правило не привело к новому состоянию прогресса, т.е. все следующие регрессии состояния подлежат к удалению. Ненужные состояния могли появиться только из-за неверных правил. Если доказательство правил завершается успешно, то и соответствующая  подцель доказана, т.к. предикаты в обоих случаях не отличаются (кроме возможно необходимых переименований  предикатов и используемых  символов). Поиск Пролога опирается исключительно на  дедукции (ср. главу \ref{chapter:intro}): выводимо только то, что следует из  фактов и правил, иные выводы, например  \textit{модус толленс} недопустимы. На практике это означает, что  фальсификация предикатов может быть решимой в общности только, если рассматриваемые множества конечны (например, вводя общие символы для определения класса объектов) и полностью определены. То есть, исключаются случаи, которые могли бы сами себе противоречить. Следовательно, доказательства могут быть раздутыми или классической дедукцией просто недоказуемыми (см. пример  исключённого третьего из \ref{chapter:intro}). Предикаты могут быть определены конечным образом, хотя входное множество бесконечное. Тогда,  отсечение (см. опр.\ref{def:CuttingSolutions}), как провал, может представлять собой отсечение предиката, либо иное явное определение.
Идея  поиска с возвратом при провале заключается в  \textit{продолжении} следующего неверного предиката по родительской оси за счёт минимальной записи с тактикой переключения стека. В отличие от метода  «\textit{возврат при провале}», метод «\textit{возврат с направленным фокусированием}» можно считать обобщённым, когда в самом простом случае фокус определяется некоторой целой отметкой и тактика переключения ориентируется на нахождение экстремума этой отметки. Отметка определяет, какая подцель будет доказываться следующей. Если внимательно задуматься, то такого рода задача может быть редуцирована без потерь в общую стратегию вычисления «\textit{ветвей и границ}» \cite{mitchell96}. Обобщением отметки является терм, который также отлично вписывается в головы левых сторон правил с помощью сопоставлений образцами.

 Тройку Хора в  Прологе можно представлять в качестве предиката с минимальной  арностью «$3$», при этом, первые две компоненты записываются как входные термы, а третья --- в качестве  выходного терма. Конечно, в общем, предикат может быть  обратим, но в данный момент мы не рассматриваем вопрос, что же является объяснением данному следствию и правилу (см. раздел \ref{sect:LanguageCompatibility} и следующие). Причины зависимостей в дальнейшем подчёркиваются.\\

\textbf{Обсуждение декларативных стратегий основанных на стеке}

\begin{itemize}
 \item \textbf{Быстродействие.}  В частности, речь идёт об улучшении производительности «\textit{реактивных систем}» (систем под наблюдением и контролем другой системы). Хотя Уоррен \cite{warren83} подчёркивает, что быстродействие не решающий фактор, однако, медлительность может иметь причину в самом алгоритме (проблема \underline{определения} пра\-вил), например в  копировании часто используемых сегментах  стековых окон и построении сложных объектов из более простых. Имеются две причины: во-первых, часто копировать, т.е. использовать  подвыражения --- показывает на (очень) высокую  зависимость данных, чего в общем можно избежать с помощью улучшения алгоритма и \underline{определения правил}. Во-вторых, изменения обработки с поддержкой  быстродействия извне (т.е.  мульти-парадигмальной) адаптацией можно при таком подходе существенно ускорить медленные точки, несмотря на дополнительные затраты в связи с загрузкой. То есть, вторая проблема разрешима с помощью  \underline{вариабельности}. Это означает, что  предикат может быть заменён иным поведением, но одинаковой  спецификацией. Языковые  расширения и вариабельность, тоже входят в поле рассмотрения, но уже в разделе \ref{sect:LanguageCompatibility}, поэтому, их всегда стоит рассматривать, независимо от данных пунктов.
 
 \item \textbf{Порядок вычисления.}   Бесконечные списки (как например  \texttt{take} из главы \ref{chapter:intro}) согласно данной  операционной семантике \cite{warren84} полностью исключаются. Операционная семантика полностью помещает  список в  стек, поэтому, частичное исследование для бесконечных списков отпадает. Частичное помещение  в стеке всё равно, но только, если сравниваемый  терм уже имеется в  стеке, а разница между подвыражениями двух термов минимизируется. Изменение  операционной семантики возможно, но требует полного изменения семантики, хотя порядок изменения и эффект вычисления сам по себе маленький. Надо отметить, что если прологовские правила унифицируются, то сравняются разницы между  термами (точнее  кучами, см. главу \ref{chapter:APs}). Естественным дополнением было бы вычисление снаружи внутрь (см. главу \ref{chapter:intro}) термов параметров, но такое требование можно избежать, благодаря ранее  не полностью вычисленным  символам (под-)термов и не требует никаких изменений или дополнительных затрат, без ограничения общности. Единственный, возможно критический этап, касается обратимости предикатов. Следовательно, единственным фактором изменения  порядка вычисления является само \underline{определение правил}.
 
 \item \textbf{Семантический контекст.}  Предикаты могут быть перегружены или определены  мульти-парадигмально. В обоих случаях проблема зависит исключительно от конкретного данного \underline{определения правил}.
 
 \item \textbf{IR.}  Промежуточное  термовое представление может в ходе разработки оказаться устаревшим или несовместимым требованием к   императивной программе. В этот момент необходимо иметь возможность добавлять и менять существующие термы (в том числе удаление). Поэтому  \underline{вариабельность} и  \underline{расширяемость} считаются ключевыми свойствами, которые  Пролог допускает, в чём можно легко убедиться. Например, если ради локальной оптимизации и  быстрого доступа к памяти, необходимо ввести  B-деревья \cite{cormen09}, \cite{atallah98}, то это возможно с помощью индексации встроенных и/или скрытых предикатов, при этом, реализация предиката может по различным причинам проводиться на другом  языке программирования.
\end{itemize}

Не использовать  стек, не удастся, т.к. языки полностью свободные от  стека и основаны только на  рекурсивных схемах, на практике оказались совершенно не пригодными, тем более нисколько не приспособлены к модуляризации и следовательно не рассматриваются далее по практическим причинам. Современные ЭВМ настолько оптимизированы, что процессоры в состоянии ускорить выделение и утилизацию  стековых окон практически без больших затрат \cite{intel11}, \cite{raman12}.
 Описывающая парадигма состояния вычисления кучи должна быть логической для решения её проблем верификации.\\

\textbf{Обсуждение поиска доказательства над реляциями}\\

Предикаты связывают  термы, которые состоят из  символов, переменных и других термов, в любое определённое соотношение. Не трудно убедиться в том, что количество  отображений между $A$ входных и $B$ выходных термов равно $B^A$. Множество отображений могут быть объединены в одном предикате. Вопрос о том, определено ли и решимо ли данное отображение для данного входного/выходного/смешанного вектора является отдельным и сложным вопросом (смешанный вектор ради простоты не рассматривается). Эта проблема может быть редуцирована на проверку  выполнимости логической формулы. Естественно, с практической точки зрения только разработчик несёт ответственность за перегруженные по количеству  отображений предикаты. Чем больше количество отображений одного предиката, тем выше его применимость. Для перегруженных предикатов вопрос об обратимости становится всё важнее, т.к. не совершение автоматически приведёт к нетерминации, а, следовательно, к не завершению доказательства. Позже это надо обязательно учесть.
 С точки зрения  канторовых множеств предикат является перечисляемым множеством (возможно очень большое, но дискретное), которое объединяет элементы доменов с элементами ко-области --- т.е. похожая концепция. С практической точки зрения это не маловажный универсальный вопрос. Если нам удастся показать универсальные свойства  \textit{реляционной алгебры по Кодду}, то мы её сможем сымитировать. Это фундаментальное свойство, которое позволит достичь уровня выразимости, хотя бы реляционных и объектно-реляционных языков баз данных на практике, что уже покрывает огромную и значимую часть выразимости.

\begin{proof}
 Полное доказательство тому содержится в моей работе \cite{haberland08-1}. Доказательство простое и прямое: свойства Коддской алгебры, как например ---  проекция, связывание и т.д., можно заменить в правилах  Пролога один к одному на каноническую запись без коллизий (см. рисунок \ref{RelationalAlgebraOverProlog}).

Также обратно доказательство производится. При этом могут потребоваться шаги  канонизации и удаления  отсечений, что в общем случае можно всегда совершить, как это было ранее продемонстрировано.
\end{proof}

Здесь нельзя не заметить, что например \cite{laemmer02} задавался частично проблемами более эффективного представления входных и выходных данных. Итог расследования приводит к тому, что удобное представление может быть  дано логическим видом, но не функциональным.
Работа \cite{pitts96} расследует фундаментальные свойства  доменов как реляции. Она посвящается большей частью доказательству существования  инварианта реляций. Использование инварианта остаётся интересным и частично открытым вопросом в связи с улучшением  кэширования  куч в  спецификациях (см. главу \ref{chapter:stricter}). Примеры, где  инварианты реляций пока что имели наиболее широкое поле применения, это  «\textit{нумеральная логика}» \cite{pitts02}, \cite{debruijn72}.

После того, как мы ввели реляции и обосновали их свойства, сразу наблюдаются приятные свойства:

\begin{itemize}
 \item \textbf{Декларативность.} При логическом выводе, арифметические вычисления не столь важны. Важны объекты, т.е.  кучи и их взаимосвязи. Это не только вопрос «\textit{вкуса}», но особенно является фундаментальным свойством именно логической системы. Пролог представляет собой среду представления и обработки знаний. Кучи являются атомами или сложными термами, а правила предикатами куч.  Псевдонимы являются  символьными переменными, которые используются в других местах или более того. Этого достаточно, чтобы определить теории куч.
 \item \textbf{Терм-Дерево.} По  \textit{теореме Биркгоффа} (о термовых продуктах) \cite{pierce91}, \cite{davis94} из абстрактной алгебры следует, что каждый терм корреспондирует с представлением в виде  дерева. Преобразование, увы, не обязательно однозначное, если не проводить  нормализацию. Таким образом, обратное преобразование однозначно определено. Отсюда сразу возникает необходимость, либо определить  канонизацию для выравнивания деревьев, либо устранить, например  ассоциативность, по которой выравнивание деревьев было бы дано неявно (см. главу \ref{chapter:stricter}). Мы не хотим сами себя ограничивать в том, что кучи могли бы быть только деревьями. Мы хотим, чтобы куча могла быть любым графом (см. главу \ref{chapter:expression}), поэтому допускается описывать кучу только одной вершиной. При этом, всё равно, вершину представляет  простая или  сложная куча, т.е. она является обыкновенной кучей или  объектным экземпляром.
 \item \textbf{Генерация термов.} Тематическое исследование \cite{haberland08-1} показывает на то, что Пролог отлично приспособлен для обработки термовых структур, в отличие от функциональных/императивных языков программирования и трансформации. В исследовании использовалось большое количество примеров. Также проводился количественный анализ. Использовались  метрики как компактность,  уровень выразимости и  интеллектуального уровня языка и многое другое. Для широкого объёма примеров, практически без всяких исключений,  Пролог превосходствует чётко с большим отрывом.
 \item \textbf{Проверка термов.} Тематические исследования \cite{haberland08-2}, \cite{haberland19-1} показывают, что если процессы генерации и проверки термов сблизить, то основным твёрдым ограничением является выразимость  языка проверки. В исследовании, без потерь общности, рассматривается обобщённый регулярный язык. Процесс сравнения термов слабо структурируемых данных можно интуитивно понять, либо как сравнение одинарных элементов, либо иерархических элементов с возможными дырами, которые наполняются данными во время запуска.
 Терм, как обобщённое представление, является уникальным IR и может быть широко использовано в различных областях, например для  верификации. Операторы описывают не программу, а структуру генерируемого документа. Статически, цикл не всегда может быть ограничен, поэтому цикл вершины $a$ некоторого дерева в  XML-документе описывает лишь $a^{*}$.
 Отсюда ясно, условия цикла могут быть представлены и проверены только самым общим видом. Естественно,  конечным автоматом так и не удастся распознать $a^nb^n$, но это и не главная цель исследования. Выходит, что главным ограничением  проверки всегда является  выразимость языка утверждений (выражений). Далее, главным результатом вместе с \cite{haberland08-1} является то, что реляции лучше приспособлены для представления знаний с термами и логическими правилами  трансформации, чем функции, которые вычисляют для каждой «\textit{дыры}» необходимые данные. Получается, логические соотношения ссылаются на имеющиеся компоненты.
 В \cite{haberland08-1} под логическими правилами трансформаций в основном подразумеваются $\tau \rightarrow \sigma \rightarrow \tau$, где $\tau$ представляет некоторый терм  IR, а $\sigma$ среда, содержавшая символьные присваивания. Если эту трансформацию расценивать как  реляцию в качестве  предиката, то трансформацию можно расширить как $\tau \rightarrow \sigma \rightarrow \tau \rightarrow \mathbb{B}$, где $\mathbb{B}$ булевое множество.
 Теперь верификация куч может быть представлена таким же семейством трансформаций, как и верификация куч. Её главная разница заключается в использовании  (абстрактных) предикатов и в методах автоматизированного вывода.
 При трансформациях используемые термы --- модели куч, различаются.
 В отличие от  регулярных выражений и возможных  распознавателей \cite{brzozowski64}, схемы  спецификаций в основном контекст-свободные. Предсказывания следующего элемента в обоих методах являются одной из центральных операций, которые могут существенно отличаться.
\end{itemize}

Решение и доказательство этому тезису изложено не только в этой главе, но также и в последующих. Кроме терма можно использовать и другие  промежуточные представления, как например  тетрады,  польско-инверсную запись,  триады и другие. Преимуществом термов при описании состояний куч, как было упомянуто, в первую очередь является простота и максимальная  выразимость. Термовое представление  программных операторов может быть записано в Прологе непосредственно. Преобразование в другие IR отпадает \cite{opaleva05}, \cite{muchnick07}.
Однако, термовое представление в  Прологе имеет то преимущество, что все термы и подтермы не нуждаются в дополнительных контекстах, конвенциях и дополнительных фазах преобразований. --- Их «\textit{можно написать просто так}». Это не только облегчает возможное использование в учебных целях, в быстрой прототипизации, но, а также облегчает преобразование и переписывание термов за счет прямого представления и анализа правил переписывания (см. \cite{baader98}, \cite{dodds08}).   
Аналогичные реализации на более реальном уровне являются, например,  «\textit{LLVM-биткод}» \cite{llvm15}, GCC  «\textit{GIMPLE}» \cite{merrill03} или аннотированные объекты в качестве IR в проекте «\textit{ROSE}» \cite{rose17}. Во всех случаях, которые используют  тетрады, необходимо предварительное IR входной программы преобразовать в  синтаксическое дерево преобразуемое в тетрады.  Синтаксический перебор в  «\textit{LLVM}» производится более гибко, чем в  «\textit{GCC}» с помощью представленного этапа и может быть совершен с помощью среды синтаксического анализа  «\textit{CLANG}» \cite{clang17}. Дерево перебора может теоретически быть введено без  «\textit{clang}», но на практике это совершенно немыслимо, потому, что, даже крайне простое дерево, всё равно может и будет представлено очень большим  промежуточным представлением, если «\textit{LLVM}» заставить вручную ограничиваться неэффективным и полным отсутствием всех дальнейших  трансформаций IR. Несмотря на раздутые наименования и на первый взгляд  «\textit{ненужные}» синтаксические определения, гибкость и  расширяемость сильно повышены в отличие от «\textit{GIMPLE}».

Ради ограничений и простоты, в отличие от  «\textit{биткод}», реализация в  Прологе исключает  безусловные переходы к любому  программному оператору. В реализации не стоит приоритет обеспечить максимальный объём  программных операторов, если в будущем имеется возможность подключения любых иных программных операторов. Более важным вопросом является  расширяемость и  вариабельность модели  кучи: «\textit{можно ли простым образом модифицировать кучу так, чтобы имитировать любую пошаговую манипуляцию кучи?}», «\textit{Можно ли добавлять всё новые фазы и правила логического вывода?}».

Принципиально, нужно заметить, что выбранное  промежуточное представление естественно может быть преобразовано из  Пролога, например в  биткод или  «\textit{GIMP\-LE}», но практическая реализация не стоит вопросом исследования. Прологовские  термы представляют собой  программные операторы (а также спецификацию и правила вывода) и как ранее в этой главе обсуждалось, могут быть представлены в качестве дерева. Да, можно выбрать  тетрады инструкций (например, близки  ассемблеру некоторой  целевой машине), но, первым итогом  синтаксического анализа всегда является  синтаксическое дерево (см. рисунок \ref{fig:PhasesOfCodeGeneration}). Выбрать другую модель означает, что одна и более фазы из упомянутых на рисунке \ref{fig:PhasesOfCodeGeneration} просто пропускаются.
Это простое замечание, но, увы, этот принцип часто (не умышленно) нарушался и нарушается в истории проекта  «\textit{GCC}»,  «\textit{LLVM}», а также в проекте  «\textit{jStar}» \cite{parkinson05-2}, где верификация проводится на уровне оптимизируемых  триад, коротко до и во время  генерации кода и многих других проектов в области  статических анализаторов --- совершенно независимо друг от друга. Хотя замечание простое, последствия могут приводить к необходимости определять точно, в какую фазу анализ  псевдонимов всё-таки нужно включать и это обычно является очень непростым вопросом. Важность заключается в следующем: (1) иметь вообще возможность расширять и менять существующие этапы анализа  динамической памяти, а (2) возможность адекватного представления для того, чтобы избегать раздутые и сложные семантические контексты и множество экстерных хранителей.   Отсутствие возможности №2 является признаком тому, что описание модели сильно усложняется. Простота модели зависит от полного представления всех необходимых данных. Прежде всего, это касается термовых  представлений входных программ и при необходимости  семантических полей содержавшие данные из рисунка \ref{fig:PhasesOfCodeGeneration}. Дополнительные данные не упомянуты на рисунке \ref{fig:PhasesOfCodeGeneration}, но всё равно могут потребоваться на локальных фазах, а следовательно, не вовлекают за собой модификацию общей модели  вычислений для работы с  динамической памятью (см. рисунок \ref{fig:PipelineArchitecture}, см. набл.\ref{obs:TypeCheckingPhases}).

Преимуществом термов является их явное представление и явная манипуляция ими. Отсутствие их явного представления приводит к конструкции вспомогательных подтермов и к введению  семантических полей (например, для  аппроксимации какого бы ни было  лимита).

\subsection*{Совместимость языков}

В введении было описано, как замечают критики  верификации, что часто дисциплина характеризуется «\textit{чисто академической, без практического применения}». Основными причинами тому, являются: перечень конвенций к отдельным моделям, которые часто имеют резкие ограничения при сильно раздутом формализме. Из предыдущего раздела, особенно из тез.\ref{thes:PrologMakesHeapSpecSimpler}, можно заметить сильные обобщения относительно  языков спецификации и  верификации для проблем  динамической памяти (см. главу \ref{chapter:DynMemProblems}, см. набл.\ref{obs:ComparisonDeclarativeParadigms}):

Наблюдение на первый взгляд может казаться малозначимым. Однако, оно означает, что для верификации  троек Хора удобнее использовать логический язык программирования, спецификации и верификации троек, что на первый взгляд является очевидным. В реальности немногие системы верификации были построены на основе логической парадигмы, поэтому они, слишком ограничены или замкнуты (см. главу \ref{chapter:intro}).

 Логические предикаты описывают  состояния вычислений, именно поэтому необходимо проверять. Например, предикаты в первую очередь не нуждаются в  побочных эффектах, чтобы  выразить состояние, а в императивных языках они являются основными. Замысел  предикатов заложен в том, что состояние можно было бы определить непосредственно без манипуляции каких бы то ни было  глобальных переменных. То есть,  диапазон видимости зависит от  символьных переменных предиката, но не от экстерных хранителей памяти.

С другой стороны, на примере модификации  Пролога, единицы логики могут быть сопоставлены один к одному элементами  логического языка программирования. С помощью программирования можно решить вопросы верификации куч.

Не связанное с  динамической памятью, но похожее применение, наблюдается в рукописи Леммера \cite{laemmer02}, в которой предлагается логический аппарат для верификации  сходимости программных компонентов \cite{feijs02}. Идея его работы заключается в предложении перехода на логические  утверждения для  спецификации и верификации, которая опирается на (различную) систему  логического вывода из-за целого ряда проблем в связи с объектно-ориентированным программированием, точнее его спецификации.

Упрощение определений какой бы то ни было математико-логической проблемы, иногда означает упрощение или решение проблемы, а, как правило, бывает --- наоборот. Часто бывает именно так: чем проще определение, тем меньше существует частных случаев, для которых необходимо вводить отдельные определения. Следовательно, объём годных единиц растёт. Например, сопоставление  символам для данного случая $\forall a$ просто, а соблюдение всё новых конвенций усложняет в принципе. Ограничивая $a$, требуется хотя бы один дополнительный предикат. Это означает, использование  квантифицированных переменных может резко увеличить  выразимость даже маленьких формул. Аналогичное действует обратно: при инвариантной длине формулы несоблюдение утвержденного, приводит к спаду  выразимости --- это как раз те проблемы, о которых говорилось в этой и предыдущих главах.

Также можно заметить: если выводить предикаты, которые содержат  символьные переменные вместо конкретного атома, то следовательно, решение будет более общим (в  Прологе это происходит автоматически  унификацией термов, а на уровне вывода, как метод  нормализованной и финитной  резолюции \cite{diaz12}).
Задача верификации заключается в проверке программы генерируемой структуры данных. Как было обсуждено в предыдущем разделе, проверка ограничивает сильнее, чем построение структуры. Оба процесса обычно различаются. Сложность проверки заключается в  выразимости.

\begin{rucorollary}{Минимизация разницы между языками}
 Проверку  куч можно упростить и расширить тогда, когда выражения описываются на одном и том же языке во время (1)  спецификации, (2)  верификации и (3) во  входном языке.
\end{rucorollary}

Использование одного и того же языка является экстремумом по задаче минимизации разницы между языками, которое рассматривается до тех пор, пока не будет обнаружен аргумент нарушавший гипотезу, если таковой имеется.
Когда ликвидируется разница между (1) и (3), либо будь-то  входной язык логический, либо полученное  IR в качестве термов из  входной программы, то  утверждения записываются в качестве  подцелей, а входная программа как термы. Утверждения о программе ссылаются на термы входной программы (см. главы \ref{chapter:APs}, \ref{chapter:stricter}), обратное не допускается. Так как верное предложение является конгруэнтностью, достаточно ликвидировать разницу между (1) и (2). Формулы и любые необходимые им индуктивные определения задаются прологовскими  фактами и  правилами. Верификация полностью упирается на  факты и правила, которые заданы спецификацией. Кроме того, в качестве  тактик, правила спецификации и иных  вспомогательных предикатов, предусмотрена возможность включать новые правила в качестве  прологовской теории.

Задачи (2) и (3) соперничают следующим образом: (2) устанавливает, как должен выглядеть  процесс построения  графа куч, а (3) конкретно, исходя только из данных утверждений, пытается доказать верность. Для этого описание должно быть сфокусировано на граф кучи, где процесс верификации является процессом «\textit{понимания}» и анализа. В главе \ref{chapter:APs} анализ опирается на  синтаксическое определение. Предпосылкой тому являются граф кучи и  унификация синтаксических, семантических и прагматических определений куч и их проверок.

Ли \cite{lee96} справедливо замечает, что  логики высшего порядка очень важны как критерий применимости на практике. Как было обсуждено в начале этой главы,  Пролог поддерживает такую возможность. Правила могут даже меняться, но нет необходимости менять правила во время  интерпретации правил (см. главу \ref{chapter:intro}). Поэтому, определённые правила доступны в обоих  процессах (2) и (3).

Модель кучи описывается спецификацией $S$ и обрабатывается языком верификации $V$. Пользователь использует $V$ для изменения состояния и следит за результатами через $S$ (см. опр.\ref{def:SpecificationLanguage}). Однако, в классическом паттерне по Реенскаугу $V$ непосредственно манипулирует $P$, что здесь варьирует. $P$ может иметь различные графовые представления $S$ с обсуждёнными ранее свойствами. Сближение приводит также к тому, что основные интерфейсы коммуницируют на одном языке.

\subsection*{Представление знаний}

Правила могут в  декларативной парадигме быть представлены, в  функциональном либо логическом виде. На примерах  языков XLS-T и  Пролог в \cite{haberland08-1} проводился количественный анализ по обработке  слабо структурируемых данных (также как и термы в общем, см. предыдущий раздел). Для качественного сравнения эквивалентных программ проводилось множество проверок более 80 выбранных типичных и образцовых примеров, аккуратно вручную подобраны из множества учебников, монографий и онлайн ресурсов. Прямое сравнение показало (см. рисунок \ref{fig:HalsteadMetrics}):

\begin{enumerate}
 \item  Пролог во всех примерах (кроме одного) превосходит в среднем на более чем 30\%  XSL-T, что можно вывести из соотношения  $N_T : N$ и $\eta_1 : \eta_2$.
 \item В среднем описание той же самой функции в Прологе на 50\% короче, а часто даже ещё короче. Обосновано такое решение на $N$, $\lambda$ и $\Delta_N = \| N_T - N \|$.
 \item Функциональный язык страдает от замкнутости, т.к. могут быть использованы только встроенные операторы. В отличие от этого, когда логический язык предусматривает определять  термовое  IR, любые иные операторы и правила.
\end{enumerate}

Более того, качественный анализ показывает и объясняет, почему выражения и правила в логическом представлении можно выразить намного проще. Это в основном из-за логического представления термов и  реляций. В отличие от функциональных языков (имеющие компактную  денотационную семантику), логически основаны на атомах, термах и правил с приоритетами (см. тез.\ref{the:ExpressibilityRelationsInProlog}). Денотационная семантика является  декларативной, однако, полный перебор логических взаимосвязей редуцируется только на одно более оптимальное отображение, а результат высчитывается на основе  входного вектора. Переменные не являются символьными, а лишь переменными, которые даже в  $\lambda$-абстракциях используются только в одностороннем порядке: присваивание значения (даже если по-ленивому) при использовании и вычислении функции. Подставкой параметризованных функций взамен реляций проблему не решить, как это наблюдается, например, в \cite{birkedal07}. Когда необходимо выразить  атомы, выражения и термы в общем, а также определить потенциально любые соотношения между ними, тогда разумнее ориентироваться на  аксиоматическую семантику или семантику, основанную на соотношениях. То есть, успех представления знаний трансформаций и сравнения с существенной частью зависят от выбранной парадигмы, как было продемонстрировано в работах \cite{haberland08-1}, \cite{haberland07-1}.

Кроме ранее упомянутых особенностей, имеются следующие, которые необходимо учесть при верификации  динамической памяти:

\begin{itemize}
  \item \textbf{Правила компактны} и могут быть вызваны из любой позиции интерпретатора. Цели и подцели могут быть любыми. Нетерминация в общем неизбежна и искусственно не ограничивается. Однако, ради более эффективной обработки, в главу \ref{chapter:APs} может вводиться не значимое ограничение.
  
  \item Используется  \textbf{строгое вычисление термов}, ленивое вычисление исключается. Это приводит к тому, что некоторые ситуации, например, передача рекурсивных данных  Аккерманна (см. рисунок \ref{fig:PrologExample1}) не может осуществляться до тех пор, пока не будут вычислены все аргументы. Это ограничение практически не является значимым, т.к. всегда возможно написать программу без употребления незавершённых данных таким образом, что контроль записывается в  тело правила. Более того,  термовая унификация уже приводит к тому, что неопределённые частицы термов, т.е.  символьные переменные присваиваются и при определении нет необходимости вычислять всё заново потому, что присваиваются лишь ранее неизвестные  подтермы, а сам терм не меняется. Это возможно согласно принципу из рисунка \ref{fig:ExampleStack2}.
  
  \item \textbf{Базовый тип} верификации является  термом. Далее он может быть рассмотрен как параметризованный тип   $\lambda$-вычисления третьей степени, т.е. тип из $\Lambda_{T_{\lambda3}}$, который допускает  квантифицированные переменные. Тип, который построен из других типов, это  зависимый тип (см. опр.\ref{def:TypedLambda2ndOrder}). Построение проводится согласно опр.\ref{def:PrologTerm} с помощью  функтора. Таким образом, сравнение термов может быть формализовано как проверка термов, однако, термы могут теоретически содержать сами  себя в Прологе, благодаря  символьным переменным. Практически этого можно избежать, если использовать проверку на самосодержащие термы (см. предикат  \texttt{unify_with_check} вначале главы). Чем выше уровень  $\lambda$-вычисления, тем сильнее действуют ограничения, а это означает, что меньше ожидается  парадоксов.  Функторы могут быть использованы для моделирования любых сложных  структур данных, в том числе списков, деревьев и  объектных экземпляров.
  
  \item \textbf{При успехе подцелей}, выдается подходящее множество сопоставлений, например: $$\texttt{?-H=pointsto(a,2),VC=pointsto(a,X),H=VC.}$$ приводит к результату:
  \texttt{H = pointsto(a,2)}, 
  \texttt{VC = pointsto(a,2)} и \texttt{X = 2}. 
   \textbf{Ге\-не\-рац\-ию контр-примера} можно встроить в процесс  унификации, например  \texttt{unify_with_check}, если в отрицательных возможностях добавить разъяснение, используя \texttt{write} \cite{sterling94}, т.к. прямое присвоение некоторых  символьных значений невозможно и не имело бы смысла. Необходимо выявить сверху-вниз самый ближайший терм, который несопоставимый со сравниваемым  термом. Принципиально, отслеживание модели вызовов можно совершить с помощью встроенной трассировки \cite{diaz12}. Если верификация проходит пошагово, разумно каждый шаг  верификации записывать в виде  «\textit{DOT}»-файла для понятия и документации доказательства.  Контр-пример предоставляет хотя бы один случай, когда верификация отклоняется.  Унификация термов в Прологе приводит к наиболее общему сопоставлению, поэтому несовпадающие  функторы по имени,  арности или неунифицируемые переменные, являются сценариями отказа. В общем, частный пример означает терм, который может быть приведён в качестве отказа, всегда тогда, когда, атом не унифицируем c функтором, либо  атомы или функторы различаются.
  
  \item \textbf{Снятие ограничений символьного характера} при определении куч \cite{berdine05-2} для запуска (см. главу \ref{chapter:DynMemProblems}), в том числе и для анализа правил верификации с обоих сторон при дедуктивном выводе (\textit{(би-)абдукция}) \cite{calcagno09}, \cite{pottier08}, можно осуществить, если правила вывода удобнее моделировать в качестве хорнских правил. Правила Хорна позволят, во-первых,  квантифицировать $\forall, \exists$ переменные в зависимости от того, где символьная переменная определяется и как она связана. Во-вторых, перебор правил и альтернатив может привести к успеху при поиске без дополнительных затрат. Однако, ради применимости, нужно объём  альтернативов ограничивать. Если выбирается правило, то  \textit{абдукцию} можно имитировать следующим образом: дано  правило Хорна «\texttt{b:-a1,a2,...,an.}». Если некоторые из $\texttt{a}_j$ оставлять неопределёнными, т.е. имеют  символьную зависимость, то, таким образом, могут быть выбраны различные \texttt{b}, если имеются альтернативы. Для этого строго требуется, чтобы начальные термы, аргументы  головы \texttt{b} не исключались.
  Необходимо заметить, что Пролог не нуждается при выбранной модели кучи в дополнительных правилах, потому, что  лексикографический порядок по  указателю  простых куч и возможностей уникального выбора (простых куч)  встроенных предикатов по работе со  списками, как например  \texttt{concat}, уже дают широкий круг выбора и преобразований новых термов из старых  подтермов. \texttt{concat} и иные сканирующие одним ходом предикаты эффективны благодаря свойству неповторимости.
  Поэтому, целый набор преобразовательных и вычислительных правил отпадает. Без ограничения общности, правила вывода обратимы, если между всеми подтермами предусловия и постусловиями, либо существует  изоморфизм, либо необратимые встроенные предикаты декларативно в  обратном случае не вызываются. В случае, когда изоморфизм нарушается, то, либо в отображении от входного вектора на выходной, либо в обратном случае производится расширение (см.  абстрактная интерпретация из главы \ref{chapter:intro}). Расширение является проблематичным при обратных отображениях, т.к. относительно  кообласти происходит сужение  домена, т.е. отображение становится прерывным.

  \item \textbf{Перегрузка значений}  правил может осуществляться и пополняться различными телами, которые внутри могут реализоваться, например, на  языке Ява, до тех пор, пока коммуникация осуществляется с помощью  входящих,  выходящих и комбинированных термов. Для анализа  прологовских правил не достаточно использовать  «\textit{DCG}» (см. главу \ref{chapter:intro}), потому, что необходимо изменить контроль и  тактику выбора. Лучший пример может обсуждаться при принятии стратегий восходящих и нисходящих  синтаксических анализаторов (см. главу \ref{chapter:APs}).
\end{itemize}

%%%%%%%%%%%%%%%%%%%%%%%%%%%%%%%%%%%%%%%%%%%%%%%%%%%%%%%%%%%%%%%%%%%%%%%%%%%%%%%%%%%%%%%%%%%%%%%%%%%%%%%%%%%%%%%%%%%%%%%

\subsection*{Архитектура системы верификации}

В \cite{haberland14-2} и \cite{haberland14-1} предлагается архитектура верификации с  динамической памятью на основе  Пролога. Архитектура представляется на рисунке \ref{fig:HeapVerificationArchitecture}.
Архитектура следует ключевым принципам, которые были обсуждены в главе \ref{chapter:intro}, это: (1)  автоматизация, (2)  открытость, (3)  расширяемость и (4)  обоснованность. (1) гласит от том, что доказательство находилось  автоматически. В Прологе решение будет найдено, в зависимости от данных правил, если в правилах исключаются циклы и даны все необходимые правила, иначе, тогда доказательство  не терминирует, либо завершает верификацию преждевременно. Далее, с помощью подхода, в главе \ref{chapter:APs} происходит автоматизация. (2) означает, что искусственные ограничения между термами, правилами и возможными реализациями должны отсутствовать. С одной стороны  Пролог является открытым, т.е. могут быть добавлены всё новые правила, а имеющиеся могут быть обновлены, если их определить ранее. С другой стороны, Пролог замкнут тем, что только то выводимо, что следует из данных правил согласно  дедукции. Пункт (2) относится к архитектуре и используемым моделям представления памяти. (3) означает, что модель памяти может быть  расширена и при необходимости изменены данные правила. Изменения всегда разрешаются благодаря переопределению  правил Хорна. Расширяемость термов и правил также обсуждались детально в предыдущем разделе. (4) означает, что любой шаг  верификации можно отслеживать и проверять, как обоснованные шаги логического вывода. Генерация  «\textit{DOT}»-файла визуализирует вычисление проверки, а при отказе генерация  контр-примеров даёт более подробные результаты и предпосылки. Возможность интерактивно без соблюдений каких бы то ни было конвенций сильно упрощает и способствует проверке на  обоснованность принятых решений.

На рисунке \ref{fig:HeapVerificationArchitecture} на вход поступает данная программа, предпочтительно на  языке Си (или другом императивном), который вместе с утверждениями о программе преобразуется в  прологовские термы и правила. Синтаксический, а затем семантический анализ проверяет и исключает недопустимые  типовые ошибки и основные ошибки нотаций. Термы всегда можно визуализировать в файловом формате  «\textit{DOT}».  Утверждения могут ссылаться на  леммы и  теоремы, которые могут быть записаны непосредственно на языке Пролог и при необходимости могут быть использованы при  верификации. При верификации включаются различные правила, которые задаются в теориях Пролога и загружаются  динамически при интерпретации, например средой \cite{denti01}, \cite{denti05}. Для решения отдельных  теорий могут быть использованы, либо экстерные средства механизмом  мульти-парадигмальной системы, либо подключением некоторых произвольно определённых  SAT-решателей в самом  Прологе. Напомним, Пролог широко используется в области  решателей и  «\textit{Constraint Programming}». Новые этапы решения проблем  динамической памяти могут быть подключены согласно принципу из рисунка \ref{fig:PipelineArchitecture}. Переходы между маленькими и большими фазами совершаются, благодаря передаче  состояний вычислений, т.е. процесс работает согласно  потоку данных (см. рисунок \ref{fig:HeapVerificationArchitecture}) и может быть сравниваемым с общей архитектурой существующих  конвейеров \cite{kennedy02}, \cite{gcc15}. Зависимость данных отличается от классического потока данных \cite{khedker09}, из-за  блока  видимости динамически выделенных данных (ср. главу \ref{chapter:expression}), который обычно не соответствует  автоматически выделенным переменным на  стеке. Однако, инфраструктура предложенного  конвейера принципиально может быть использована при условии, если анализ  псевдонимов установит отдельные  интервалы видимости, см. набл.\ref{obs:VariablesScope}.

Далее в \cite{haberland15-1} обсуждаются и предлагаются основные критерии для  расширяемой и модифицируемой архитектуры. В первую очередь это относится к минимальности  IR  входного языка, который к сравнению с  «\textit{PCF}» \cite{mitchell96} допускает присвоение для основы объектного  вычисления по Абади-Карделли (см. раздел \ref{sect:TheoryOfObjects}), неограниченный цикл и вызов процедур. Так как представленная модель вычисления, согласно  классному виду вычисления, может быть  типизирована по Абади-Карделли (обновления кода во время запуска исключаются, например, передача аргументов производится выборочно, либо  по значению, либо по вызову с особенностью вставленной конструкции более широкого  функторного объекта), то и свойства, согласно второй и третьей степени  $\lambda$-вычисления, (см. опр.\ref{def:TypedLambda2ndOrder} и опр.\ref{def:Lambda2ndOrderTermTypes}) применяются непосредственно.  Расширение может быть применено к  входному языку программирования, к фазам  статического анализа и правилам (включая имеющиеся).

В \cite{haberland15-2} вводятся  фрейм, куча и их  интерпретации, а также даётся предложение практического представления в Прологе. Так как куча предположительно представляется как терм, который принадлежит предикатной интерпретации в Прологе, тогда: либо верно и меняются все неопределённые символьные переменные подцелей, либо даётся отказ с предположительной причиной. Верификация похожа на такой же процесс сравнения, как представленный в \cite{haberland08-2}, \cite{haberland19-1}, т.е. автоматом  сравнения  по образцам и деревьям  (с англ. «\textit{tree graph matcher}» \cite{comon07}). Хотя изначальная модель постановления процесса очень похожа, всё равно её сравнение тяжелее, например: за счёт  вызовов процедур,  абстрактных предикатов,  пред- и  постусловий и моделей графа кучи. Таким образом, из аналогии следует: если утверждения являются схемой/типом, а программа пошагово строит выражение, т.е. граф кучи, тогда верификация является проверкой типов. --- В том случае, если означает содержимое (см. рисунок \ref{fig:HoareCalcVSTypeSystem})?

\begin{rucorollary}{Минимизация входной программы}
 Результатом проверки содержимого является минимальная  входная программа, которая манипулирует  динамической памятью.
\end{rucorollary}

Ответ кажется простым:  входная программа. Однако, соотношение между  «\textit{типом}» и  «\textit{выражением}» не может быть однозначным. Возникает вопрос, а какая действительная программа генерируется? Суть в том, что программа генерируется пошагово только при необходимости  в соответствии с  «\textit{утверждениями}». Программа строится минимальным образом, потому, что каждая  грань графа куч соответствует всё новым программным операторам. Цикл в графе кучи не обязательно равен циклу в  программных операторах, а, например, может быть равен двум операторам.  Бесконечно много граней в графах исключается, поэтому только похожие цепочки могут соответствовать циклу в качестве  программного оператора, условие которое определяется  графом. При конструкции  входной программы  минимальность означает лишь то, что множество операций, которые не имеют прямого отношения к итоговому графу, удаляются и остаются только те операторы, которые действительно необходимы для построения финального графа кучи.     
Таким образом, проверка \textit{содержимого выражения данного типа} может сгенерировать  минимальную программу, которую можно сравнить с данной программой.

Архитектура в Прологе без дополнительных затрат разрешает следующее:

\begin{itemize}
 \item \textbf{Любой входной язык}  допускается, если   термовое IR соблюдается, которое также может быть изменено и добавлено в правилах.  Входная программа может быть даже пуста и процесс  синтаксического анализа пропущен, если IR программы или её части вводятся вручную  для соответствующей части  верификации.
 Архитектура, представленная на рисунке \ref{fig:HeapVerificationArchitecture} позволяет исключить синтаксические и семантические ошибки с помощью гибких фаз из рисунка \ref{fig:PipelineArchitecture}.
 
 \item \textbf{Скромные спецификации} разрешают избегать  полную спецификацию, поэтому специфицируются только те  модули, которые нуждаются в верификации. Кроме полной спецификации, никаких альтернатив не было, что раньше приводило к большим затратам и трудно читаемым утверждениям. Этот подход называется  «\textit{footprint}» и наблюдается в области верификации куч с помощью ЛРП впервые в «\textit{Smallfoot}». 
 Принцип в  Прологе прост, если имеется  утверждение, то оно проверяется на верность, если нет, то верификация по умолчанию продолжается. Верификация применяется только для тех модулей, которые содержат утверждения. Чтобы избежать  полную спецификацию кучи, надо использовать  вспомогательные предикаты, такие как \textit{\underline{true}} или \textit{\underline{false}} (см. главу \ref{chapter:stricter}), а также предположить данные спецификации выборочно и не полностью.
\end{itemize}

 Полиморфизм исключен ради простоты из корневых  программных операторов (см. дискуссии в главе \ref{chapter:intro},\ref{chapter:expression} и далее).  Граф зависимостей заданных правил и  лемм  утверждений анализируются при запросе интерпретатором  Пролога, а также во время  синтаксического анализа, например, при построении  компиляции (см. главу \ref{chapter:APs}).

\subsection*{Объектные экземпляры}

В разделе \ref{sect:TheoryOfObjects} подробно обсуждались два вида вычислений с объектами ---  по Абади-Карделли и  Абади-Лейно. Из-за широкого распространения в  императивных языках программирования с  объектно-ориентированным расширением используется первый вид. Как было продемонстрировано и обсуждено в ранних главах,  объектный вид тяжелее прослеживать в связи с  верификацией  корректности и  полноты. С теоретической точки зрения оба вида имеют одинаковую  выразимость \cite{leino98}, \cite{reus02} и полная абстракция может быть всегда найдена (см. ранее). Однако, написание программ может сильно отличаться и тогда доказательство  полной абстракции, т.е. равенства  операционной и  денотационной семантики между обоими видами очень сильно расходится из-за трудной формализации относительно простых свойств  объектного вида вычисления (см. раздел \ref{sect:TheoryOfObjects}). Так как простота  спецификации имеет наиболее важный эффект на общую простоту доказательств, тогда выбирается именно классный вид вычисления.

 Объект --- это, прежде всего, экземпляр некоторого класса. Ради простоты,  полиморфизм исключается, т.к. он не имеет прямого отношения к корневой функциональности вычисления (настоящий полиморфизм переменных не рассматривается, т.к. в объектно-ориентированных языках полиморфизм выражается \textit{спонтанным} полиморфизмом \cite{cardelli96-2} исключительно с помощью  подклассов \cite{bruce02}), и представляет собой только более удобный способ вызова подходящего метода во время запуска (см. раздел \ref{sect:TheoryOfObjects}). В вычислении Хора полиморфизм выразим, но в наших целях предлагает лишь дополнительный эффект без увеличения выразимости (см. главу \ref{chapter:intro}, \cite{nanevski06}).
Поэтому, объект, моделируемый  кучу, является замкнутым  регионом памяти без дыр, т.е. определённого  типа, который содержит лишь  поля. Методы не сохраняются вместе в  динамической памяти, потому, что они статические и не меняются во время запуска. Изменение кода во время запуска также исключается из-за минимального успеха и огромных проблем свойств  корректности и  полноты (см. дискуссию из раздела \ref{sect:TheoryOfObjects}). Так как объект присвоен данному  типу класса, методы полностью определены. Наследственные поля и методы также полностью определены. Для наиболее легкого сравнения объектных экземпляров, соблюдается конвенция, что пары  (наименование поля $\times$ содержимое) отсортированы по  лексикографическому порядку. Таким образом, проверку и преобразование экземпляра на  подкласс, можно реализовать двумя способами: (1) каждое поле проверяется согласно соотношению  «$>:$» (см. опр.\ref{def:TypeChecking}), при этом, множество полей в верхних и нижних подклассных  экземплярах расходятся или (2) все поля группируются согласно идентификатору  наследованности. Таким образом, в памяти необходимо эффективно сравнивать экземпляры нижних классов только с маленьким подмножеством, которое представляет экземпляр верхнего класса. Чтобы продемонстрировать (2), возьмем пример: «\texttt{SubClass1 s1; SuperClass o1=(SuperClass1)s1;}». Допустим, \texttt{s1} содержит \texttt{[o1,o2,o3]}, тогда вычисление \texttt{o1}, где верхний класс в \texttt{SuperClass1} наследует только поля \texttt{o1} и \texttt{o2}, может производиться преобразование с помощью копирования первых двух полей, т.е. начальным сегментом \texttt{s1}.

Поля  объектных экземпляров записываются в список  кортежом \textit{(наименования, значение)}. Ссылки на объектные экземпляры (в том числе циклические), выражаются  символьными переменными. \texttt{A=object(A,A)} запрещается и может быть исключено с помощью \texttt{unify_with_check}.\\
\texttt{A=object([(a,A),(b,A)])} разрешается. Предполагается использовать упрощенную форму:\\
\texttt{A=object((a,A),(b,A))}, т.к.  функтор \texttt{object} внутри  Пролога уже строит список головы, которая содержит \texttt{object}. Так, как кортеж содержит ровно два элемента, то сложный терм будет всегда чётко определён и однозначен в отношении объектной сети.

Объектные поля доступны с помощью  «\textbf{.}»-оператора и могут быть использованы в  программных операторах и  утверждениях. Исключение неверных доступов обнаруживается во время  семантического анализа. Спецификация (всех) полей данного объекта производится на уровне  абстрактного предиката, которые также могут быть определены, в том числе частично (см. главу \ref{chapter:stricter} и далее).\\\\
Конвенции из главы 1, обсужденные в этой главе, а также кон.\ref{conv:RestrictedObjects} и конв.\ref{conv:HeapAlignment} вводятся, во избежание парадоксов, с целью приближения к «\textit{UML}».

%%%%%%%%%%%%%%%%%%%%%%%%%%%%%%%%%%%%%%%%%%%%%%%%%%%%%%%%%%%%%%%%%%%%%%%%%%%%%%%%%%%%%%%%%%

%%%%%%%%%%%%%%%%%%%%%%%%%%%%%%%%%%%%%%%%%%%%%%%%%%%%%%%%%%%%%%%%%%%%%%%%%%%%%%%%%%%%%%%%%%%%%%%%%%%%%%%%%%%%%%%%%%%%%%%%%%%%%%%%%%%%%%%%%%%%%%%%%%% 5 Stricter
%%%%%%%%%%%%%%%%%%%%%%%%%%%%%%%%%%%%%%%%%%%%%%%%%%%%%%%%%%%%%%%%%%%%%%%%%%%%%%%%%%%%%%%%%%%%%%%%%%%%%%%%%%%%%%%%%%%%%%%%%%%%%%%%%%%%%%%%%%%%%%%%%%%
\section*{Ужесточение выразимости куч}
% (40p.)

В этой главе рассматривается  \textit{ЛРП} и анализируются проблемы в связи с  пространственными операторами. Получается, что один и тот же оператор приспособлен, согласно  \textit{графу куч}, разделять и связывать между собой кучи, в зависимости от состояния указателей и их содержания. Получается, оператор имеет различные аспекты в зависимости от контекста используемой формулы. Другими словами, единый пространственный оператор в классической ЛРП является  \textit{многозначимым} \cite{haberland16-3}. Многозначимость, это удобная запись, но влечёт за собой недостатки. Наиболее важными недостатками являются  \textbf{контекст-зависимость}. Зависимость, прежде всего, означает, необходимость анализировать всю формулу, что может быть (не) связано с данной кучей и все её содержавшие кучи. Чтобы определить независимый граф кучи (т.е. се\-ман\-ти\-чески контекст-независимо), требуется использовать зависимую нотацию того же  графа кучи (т.е. син\-такс\-и\-чески контекст-зависимо). Это не является  парадоксом, однако, желает иметь лучшее. Далее мы покажем, что синтаксически возможно определить граф кучи контекст-независимо без ограничения общности, исключить целый ряд проблем и улучшить процесс  автоматизированного доказательства. Для автоматизации  синтаксический анализ при  интерпретации формул, которые описывают  кучи, является накладным и избыточными, если пространственное отношение между кучами удастся явным образом выразить. Соотношение между кучами может быть связанное или не связанное. Переписывание многозначной формулы кучи в  однозначную (единственную) формулу может быть не тривиально, т.к. необходимо рассматривать все переходы от одной кучи к другой, либо проверить отсутствие любых переходов из одной кучи в другую, на что может потребоваться значительное время. И наоборот, если имеется однозначная формула, то не ожидается сюрпризов в связи с соотношениями куч. (Не-)связанная куча с некоторой другой кучей сохраняется, при этом не имеет значения, какой предшествует контекст или следует иерархически определённой куче. Используя синтаксически контекст-независимую формулу для описания семантически контекст-независимой модели памяти графа кучи объединяет понятие о том, что такое куча.

Глава разбита на семь разделов. В первом разделе анализируются  ЛРП и последствия многозначности. Во втором разделе проблема многозначности локализуется, и обсуждаются подходы к преодолению проблемы. В третьем разделе рассматривается  ужесточение многозначности как решение проблемы. В четвёртом разделе  в частности, рассматриваются объектные экземпляры  классового вычисления.  Объект рассматривается как комплексная единица ЛРП, на который распространяются те же самые свойства ужесточения и для простых ссылок. В связи с ужесточением операторов в пятом разделе обсуждается возможность специфицировать лишь часть  динамической кучи, благодаря свойствам строгого пространственного соотношения  подкуч. В частности, обсуждается  модульность спецификации и улучшения качества программного обеспечения в связи с объектами. В шестом разделе подробнее обсуждаются возможности применения формальных свойств ужесточенной модели памяти. В последнем разделе обсуждаются возможности и ограничения предложенной модели.

\subsection*{Мотивация}

Возьмём следующее синтаксическое определение  термовых выражений $E$ над целыми числами в классической  арифметике целых чисел в качестве рассматриваемой проблемы многозначности:

 \begin{grammar}
<E> ::= <k> | <E> ‘$\otimes$’ <E>
 \end{grammar}
 
Нетрудно убедиться в том, что синтаксис по  Бэккуса-Науру представляет собой ин\-дук\-тив\-но-определяемые термы, где начальное определение любое, но определённое  целое число $k$. Допустим, $\otimes$ является некоторым  бинарным оператором, который  полностью определён для целых чисел, например сложение. Если мы имеем ситуацию, когда для выражения $e_1,e_2,e_3$: $E_0 \cdots \otimes e_1 \otimes e_2 \otimes \cdots E_n$ и $n \in \mathbb{N}_0$ один раз вычисляется как $e_{1,2}$, а при $E'_0 \cdots \otimes e_1 \otimes e_2 \otimes \cdots E'_n$ вычисляется как $e'_{1,2}$, при этом $e_{1,2} \ne e'_{1,2}$, то либо правила вычисления не корректны (возможно, ошибка совершена в стадии разработки; далее исключается), либо вычисление зависит от контекста, т.е. зависит от $E_0$ и $E_n$, либо $E'_0$ и $E'_n$. Необходимо заметить, что если $E_0 \equiv E'_0$ и т.д. до $E_n \equiv E'_n$, то проблема различия всё-таки совпадает с проблемой  (не-)корректности вычисления. Исходя из стандартного случая, т.е.  $e_{1,2} \ne e'_{1,2}$ при $E_0 \ne E'_0$, $E_n \ne E'_n$, можем утверждать, что обе $E_0$ и $E_n$ одновременно не пусты. Следовательно, зависимость означает при $(e_1 \otimes e_2) \otimes e_3$, что $e_3$ содержит синтаксическую информацию, которая влияет на результат $e_1 \otimes e_2$. А поэтому, для каждого $j$ умножения $\otimes_{\forall 0\le j}^{n} e_j$ в худшем случае означает полный синтаксический перебор всех остальных факторов. Сложность ограничивается  рангом полинома $n \choose 2$. Какое отношение эта граница имеет к кучам?

Из набл.\ref{obs:OverloadedSpatialOp} следует, что определение «$\star$» из ЛРП очень близко к определению «$\otimes$» вверху (см. опр.\ref{def:ReynoldsHeapDefinition}). Поэтому при анализе каждую из куч необходимо внимательно проверять (что означает $E$ в верхнем примере). Поэтому, имеется следующее предложение:

Доказательства этому и последующим тезисам будут следовать в этой главе.

Из этого тезиса следует, что  контекст-независимость позволит определить  формальные теории о равенствах и неравенствах куч, которые далее можно будет автоматизировать ради подключения  SMT-решателя.

\begin{proof}
 Идея заключается в расширении пространственным соотношением, ссылаясь на трм.\ref{theo:GeneralizedHeapConjunctionTheorem} и лем.\ref{lem:HeapConjunctionMonoid}, которое может связывать, либо разделять.
 
 Язык моделирования «\textit{UML/OCL}» основан на типизированном лямбда-вычислении второго порядка, следовательно, эквивалентен опр.\ref{def:TypedLambda2ndOrder}, следовательно может быть выражен в типизированном лямбда-вычисление третьего порядка, следовательно может быть представлен в прологовских термах как описано в главе \ref{chapter:logical}.
\end{proof}

Данные наблюдения и тезисы следуют анализам предыдущих глав и замечаний. Из предыдущих наблюдений и анализов можно заметить следующее:

\begin{enumerate}
 \item Простая модель должна быть представлена простым способом.  \textit{Контр-при\-ме\-ром} здесь может послужить \cite{suzuki82}. Там, на первый взгляд неполное множество представляет на самом деле полное множество правил, которое может совершать очень сложные операции с  указателями.
 Самые незначительные изменения могут легко привести к иному или не предсказуемому поведению. То есть, в сильно динамической системе минимальные изменения не должны менять весь характер поведения, особенно не должны менять \textit{далекие} пространственные части куч.
 \item Различные предыдущие модели памяти, точнее, их конвенции, не столь важны, как может показаться на первый взгляд.
 Показано, что ввод всё новых конвенций не расширяет, а наоборот, ограничивает выразимость дополнительными условиями. Предлагаемые новые возможности входных языков или языков спецификации на столько специфичны, что применение и метод верификации не устойчивы к малейшим модификациям и расширениям (см. главу \ref{chapter:logical}).
 С практической точки зрения гораздо важнее описать точно и адекватно ту модель памяти, которая имеется, вводя как можно меньше искусственных ограничений и описывая только основное. Для описания, включая все возможные ограничения, берётся непосредственно  \textit{граф кучи}. Это является  \textit{утилитаристским подходом}. Эпи\-стем\-о\-ло\-ги\-чес\-кое определение термина «\textit{кучи}» дается в главе \ref{chapter:expression}.
\end{enumerate}

В итоге реализации получается  формальная грамматика операторов, с одним бинарным оператором для слияния и одним оператором для разделения, которая контекст-свободная (например,  подграмматика $E$).
Предложенная ужесточённая модель предусмотрена для более эффективного провождения верификации, отделив правила теории куч от общих логических правил. Правила верификации представляются в качестве правил Хорна, и интерпретация правил совершается с помощью Пролога \cite{haberland14-1} (см. тез.\ref{thes:PrologMakesHeapSpecSimpler}). На следующем этапе ужесточённые операторы заменяют оператор $\star$ так, чтобы  абстрактные предикаты могли быть автоматически распознаны при  синтаксическом анализе (см. главу \ref{chapter:APs}).

\subsection*{Многозначимость операторов}

В качестве наиболее точного  языка спецификации в вычислении Хора изначально предлагалось использовать математику, как наиболее точный формализм. Позже математика уточняется  логикой предикатов в самом общем виде. В области  верификации и  спецификации куч можно использовать специальные логики, что довольно успешно применяется на практике (см. раздел \ref{sect:HoareCalc}). Однако, неограниченные формулы могут оказаться более удобными при автоматизации, которые увидим позже. Одно из таких «\textit{более приемлемых}» условий автоматизации может оказаться \textit{однозначное представление пространственных операторов}. Здесь необходимо пояснить возникающий парадокс: ужесточение операторов является условием расширения выразимости, что будет продемонстрировано позже. Ужесточение условий формул куч естественно приводит к ограничениям.

Проблема точности и  выразимости является фундаментальной проблемой не только в области логики, но также в науке в самых различных областях, начиная от восприятия, до записи по согласованным конвенциям до выражений. Надо искать причину возникновения многозначности. Неудивительно, если язык выражений программных операторов и  декларативный язык спецификации различны, то могут появляться разрывы в семиотике (см. главу \ref{chapter:logical}) объектов и их взаимосвязи. Элементарный вопрос о равенстве двух различных представлений о куче (см. главу \ref{chapter:expression}) может существенно усложняться, если не использовать единую, либо согласованную форму --- это причина, которая лежит в определении куч.

Итак, вопрос об  изоморфизме двух графов куч в обобщённом виде обсуждается в отображении на рисунке \ref{fig:GraphIsomorphisms}, и может быть оценён как тяжёлый. Решение изоморфизма может быть оценено, в общем, с экспоненциальной сложностью, для довольно плохих прогнозов. На практике имеются экспоненциальные алгоритмы, которые приближаются к полиному третьего ранга для небольшого объёма входных вершин.
Если куча из рисунка \ref{fig:GraphIsomorphisms} (a) содержит только сплошные линии, а при дальнейшем анализе выходит, что также имеется в частности соединение между вершинами №$0$ и №$5$, когда обе вершины кучи уже были специфицированы, то выявление изоморфизма сильно усложняется. Однако, сложность сужается ради типов и наименований, которые не меняются в итоге. Проблема сложности изоморфизма сохраняется лишь тогда, когда имеется набор указателей вместе с графом кучи, которые можно преобразовать в другой граф, который отличается от предыдущего только множеством наименований вершин. Для рисунка \ref{fig:GraphIsomorphisms} это может быть совершено с помощью пермутации (0 11)(1 8 2 10 3 9)(4 7)(5 6), исходя из графа на рисунке \ref{fig:GraphIsomorphisms} b). С практической точки зрения вопрос изоморфизма стоит только тогда, когда необходимо проверить, может ли в принципе данная структура быть преобразована в другой граф, если допустить, что наименования могут меняться.

В главе \ref{chapter:APs} рассматриваются абстрактные предикаты более детально, однако, идея абстракции предикатов лежит в свёртывании и развёртывании графа. К примеру рассмотрим рисунок \ref{ExampleConnectedHeapGraph}.

Данный граф согласно минимальной зависимости можно разбить на подграфы вдоль мостика $v_1 \mapsto v_3$, см. рисунок \ref{ExampleSplitHeapGraph}.

Граф можно будет описать отдельными предикатами $\pi_0(v_0,v_1), \pi_1(v_3,v_4)$ и предикатом $\pi_2(v_4,v_5)$, либо более абстрактно как: $\pi_0(v_0,v_1), \pi_{1,2}(v_3,v_5)$, при этом, графы представленные предикатами связаны между собой и видимые вершины снаружи появляются в качестве аргументов неявным определением $\pi_j$, см. рисунок \ref{ExampleSplitVariantsHeapGraphs}.

Развёртывание согласно определению $\pi_j$ приводит к обратному.\\

Теперь можно вывести более обобщённые вопросы в связи с адекватным представлением кучи следующим образом: 

\begin{enumerate}
 \item Как специфицировать однозначные формулы и проводить максимально детерминированную верификацию?
 \item Как решить вопросы об изоморфизме, локальности и абстракции графов простым образом (представление кучи)?
 \item Как ограничить принудительную проверку объектов в программных операторах кода (представление объектов)?
 \item Как эффективно решать равенства с кучами, если не все вершины (и грани) графа кучи полностью определены (\textit{частичная спецификация})?
 \item Как можно избавиться от повторного анализа куч (пошаговая верификация)?
\end{enumerate}

В главе \ref{chapter:expression} уже была представлена  модель кучи, которая была предложена Рейнольдсом, Бурсталлом и другими. В этой главе итоги и эффекты определения Рейнольдса рассматриваются и проводятся дискуссии о выводимом графе, а также графах получившие от модификации различных параметров. Наблюдаются свойства и выразимость, что и является главным замыслом этой главы.

Переменные, также как и  указатели, хранятся в стеке, а содержимое указателей хранится в  динамической памяти. Следующее доменное равенство согласно \cite{berdine05} действительно: $Stack=Values \cup Locals$. Утверждения меняются  программными операторами и генерируются при верификации, при проверке куч, исходя из программных операторов. Утверждения о кучах, либо \textit{верны}, либо \textit{ложны}, в зависимости от конкретной кучи. Синтаксис утверждений определяется опр.\ref{def:HeapTermExtendedDefinition}.
Из опр.\ref{def:ReynoldsHeapDefinition} следует, что  бинарный оператор «$\star$» может быть использован двумя способами: для того, чтобы выразить две кучи не пересекаясь, а также, чтобы две кучи делили между собой один общий символ. Оператор «$\star$» используется как логическая конъюнкция для связывания истины о кучах. Кроме того, он является, пространственным оператором, который выражает место нахождения и связанность. Она выражается тем, что связывающие формулы о кучах определяют, как две кучи расположены в некотором адресном пространстве касательно друг друга. Пространство подразумевает, что кучи занимают некоторые поля динамической памяти. Если определить связанность между двумя кучами как  \textit{двудольный граф}  (биграф), то имеется левая сторона указателей и правая сторона множеств содержимого. Для того, чтобы связанный граф можно было полностью описать, высчитав  \textit{максимальное паросочетание}  с целью уменьшения количества  конъюнкций для формулы куч, соответствовала бы полностью \textit{графу куч}. Такой подход на практике очень не практичен, т.к. нет необходимости и желания, со стороны разработчика, описывать \textit{максимально сжатое представление графа куч} целиком (см. раздел \ref{sect:HeapGraph}).
Но, если полученный граф кучи сильно отличается от ожидаемого графа, то неожидаемые  «\textit{лишние части}» кучи являются показателем возможных очагов ошибок в программе.
Важны и другие критерии, например, адекватное соотношение между синтаксическим представлением и графом, с целью нахождения вершин и граней, а также навигации по граням, и т.д. Компактное представление абсолютно не даёт преимуществ в данном случае, а наоборот, трудно читаемо. Необходимо сравнивать конкретные состояния ячеек в динамической памяти. По этой причине \textit{регулярные выражения}, как удобный вариант отпадают. Регулярные выражения страдают проблемой  \textit{нелокальности}: как только граф куч локально меняется в одном месте, то может последовать изменение целого выражения. Желаемое поведение, должно быть таким, что добавление одной грани в граф кучи не должно менять всю формулу, а лишь ту часть подвыражения или части формулы, которая непосредственно связана с меняющейся частью. Не причастные   (под-)кучи не должны меняться.

Сравнив с опр.\ref{def:ReynoldsHeapDefinition}, а также определения и дискуссию из главы \ref{chapter:expression}, можно заметить, что оно довольно трудное, а данные формулы, использующие это определение, могут быть многозначными, если всегда анализировать только часть от формулы. Для полного решения взаимосвязей всегда необходимо полностью анализировать все  конъюнкты. Данное определение одной кучи, является результатом, если попытаться определить отдельную кучу. Напомним, что Рейнольдс определяет только множество куч, а отдельная куча у него, так и не определена.
Увы, другие авторы (см. главу \ref{chapter:intro}) также определяют только множественные кучи, но не определяют  единственную форму кучи, если даже между строками авторы дают неполные и неформальные предпосылки на неё.
Надо обратить внимание на то, что выведенное определение кучи в опр.\ref{def:ReynoldsHeapDefinition} является явным определением одной кучи, при этом, множеству куч не противоречит определению. Отсюда --- берётся мотивация необходимости строже и явным образом решить многозначность «$\star$». Когда имеются однозначные операции, тогда можно обращаться к отдельной куче с помощью одного символа. Фактически модель Рейнольдса (и других) анализирует и подразумевает только смесь куч. Представление об отдельной куче отсутствует. Куча имеет только тогда значимое объяснение, когда оно задаётся вместе с чем-то. Ввод  строгих операторов позволяет выразить семантику и замысел одной кучи, которая как таковой субъект естественно существует независимо от других куч. Следовательно, куча имеет идентичность. Таким образом, определение можно избежать лишь через значение нескольких куч, т.е.  «\textit{структуралистские семантики}» меняются на «\textit{не (строго) структуралистические семантики}».
 Как только два оператора будут определены ($\circ$, $||$), далее свойства и равенства могут быть исследованы, вследствие чего,  термовые алгебры можно будет определить для решения прогрессивной  сходимости верификации.  Термовые алгебры разрешат установить всё новые и новые  формальные теории над  кучами, которые практически можно будет использовать на прямую, в качестве правил Хорна в рамках Пролог программы (см. тез.\ref{thes:PrologMakesHeapSpecSimpler}).

%%%%%%%%%%%%%%%%%%%%%%%%%%%%%%%%%%%%%%%%%%%%%%%%%%%%%%%%%%%%%%%%%%%%%%%
\subsection*{Ужесточение операторов}

Из-за многозначимости, оператор $\star$ может быть использован как конъюнкция (\textit{слияния куч}), а также как дизъюнкция (\textit{деление куч}).
Более того, строгие различия (трм.\ref{thes:StricterOpsExpressibility}) в реализациях  ЛРП часто $\star$ используются равномерно логической конъюнкцией.
Решение этих проблем является задачей данного раздела. Вводится формальное определение \textit{конъюнкции кучи} и свойства этой операции, затем вводится дизъюнкция. Будет показано, что общность выразимости при обеих операциях не ограничивается.

Здесь $\beta' = vertices(\beta) \subseteq V$ определяет все вершины графа куч, которые указываются из $\beta$. В случае если $\beta$ является  объектным экземпляром, тогда также рассматриваются все  поля объекта, которые являются указателями.
Так как $\alpha$ может ссылаться лишь на одну уникальную вершину графа (например,  \textit{путь доступа} к некоторому объектному экземпляру), тогда и существует не более одной совпадающей вершины в $isFreeIn$ для данной кучи $H$.
Общее предположение высшего определения заключается в том, что при  пошаговом построении графа при использовании конъюнкции, всегда существует одна подходящая вершина, иначе, две кучи нельзя связать вмес\-те.

Например, нужно связать три пары  указателей (указатель ссылается на некоторое содержимое) $a,b,c$ (см. рисунок \ref{fig:HeapGraphConjBeforeAfter}).
Сначала необходимо выразить кучу $a$, которая может быть любой, либо как $\underline{emp} \circ a$. Только когда $a$ существует, $a$ ссылается на некоторое содержимое, которое эквивалентно началу кучи $b$ и не связано (на данный момент мы подразумеваем именно такое состояние как начальное), только тогда обе кучи связываются.
Если предположить обратное, т.е. в графе кучи имеются две одинаково содержимые, то по определению это исключено, т.к. допускается только одно по-настоящему уникальное содержимое (см. позже), но с любым количеством псевдонимов. Допустим, имеются одинаковые, но не по-настоящему одинаковые содержимые, тогда путь доступа должен отличаться, иначе, создаётся противоречие. В обоих случаях исключается возможность конъюнкции некорректного связывания.
Итак, мы имеем связанную кучу $a \circ b$. Сейчас можно продолжить конъюнкцию под теми же условиями, как только что было изложено для $c$. При успехе мы получим граф кучи как на рисунке \ref{fig:HeapGraphConjBeforeAfter}.
Так как мы заинтересованы в конъюнкции любых графов, например двоичных деревьев, мы допускаем конъюнкцию в любых частях связанного графа. К примеру, $a \mapsto 5$ является допустимой конъюнкцией графа куч, однако, $a \mapsto 5 \circ b \mapsto 5$ не является. Таким образом, мы сможем выразить  псевдонимы, например, граф кучи $\xymatrix{ x \ \circ \ar[r] & \circ^z &  \ar[l] \circ \ y }$ может быть выражен как терм кучи $x \mapsto z \circ y \mapsto z$.

Кучи могут быть связаны различными способами, когда вершина является объектом. Например, можно договориться, что при присвоении объекта меняется только указатель, либо одно идентифицируемое поле. Можно также договориться о различных присвоениях. Например, когда все поля одновременно присваиваются некоторым входным вектором (например, массив или строка с определённым разделителем), либо  присваивается только одно поле, а все остальные сбрасываются и т.д. (см. главу \ref{chapter:expression}). Ради простоты, более популярный из языков программирования  Си(++), далее рассматривается только присвоение «\textit{один-на-один}» (см. раздел \ref{sect:ClassObjectAsHeap}).\\
\textbf{Замечание:}
Пусть $\varPhi_0$ некоторая куча, тогда следует $\varPhi = \varPhi_0 \circ a_0 \mapsto b_0 \Leftrightarrow \exists (a_m \mapsto b_m) \in \varPhi_0 \wedge (a_m=a_0, b_m\neq b_0 \vee a_m \neq a_0, b_m=b_0)$.

\begin{proof}
Теорема является обобщением опр.\ref{def:HeapConjunctionDefinition}. Обе, $H_1$ и $H_2$ могут быть обыкновенными кучами видом $a_1 \mapsto b_1 \circ a_2 \mapsto b_2 \circ \cdots \circ a_n \mapsto b_n$. Чтобы доказать корректность теоремы, сначала необходимо показать, что если не существует общий элемент в обоих графах, то согласно опр.\ref{def:HeapConjunctionDefinition} получаем $\underline{false}$, что совпадает с ожидаемым от  конъюнкции. В противном случае, если имеется хотя бы один общий элемент, то согласно  индукции, выбираем один элемент и тогда обе кучи связываются. Нужно отметить, что конъюнкция исходит лишь от возможности связывать графы и нас не интересует количество более одного. Все вершины кроме той, которая используется для  слияния графов, могут также гипотетично быть использованы для слияния. В таком случае, полученный граф всё равно остаётся  просто связанным. В противном случае начало, либо конец слитого графа куч находится, только в $H_1$ или только в $H_2$, а также в обоих графах куч $H_1$ и $H_2$ одновременно, но это исключается. Из-за этого противоречия следует годность теоремы.
По определению, $a\circ a$ равно \textit{\underline{false}}, что необходимо фильтровать для всех конъюнктов.
Обсуждение  $a\circ a$-решателя будет проводиться позже и может быть реализован с помощью активного множества, которое содержит статически все успешно обработанные базисные кучи.
\end{proof}

Лево-ассоциативность означает, что терм $object1.field1.field2.field3$ по умолчанию равен:
$$(((object1).field1).field2).field3$$
-- таким образом, части выражения доступа к полю могут сопоставляться символьными переменными. Это повышает модульность, а в частности, повышает гибкость выражений.

\begin{proof}
Чтобы доказать, что $G$ является  \textit{моноидом}, необходимо доказать: (i) $\Omega$ замкнут  под $\circ$, (ii) $\circ$ является ассоциативной операцией, а также (iii) $\exists \varepsilon \in \Omega. \forall m \in \Omega: m \circ \varepsilon = \varepsilon \circ m = m$.

$\omega \in \Omega$ связной граф кучи, полученный с помощью бинарного функтора «$\mapsto$» в соответствии с опр.\ref{def:HeapTermDefinition}. Согласно опр.\ref{def:HeapConjunctionDefinition} $\forall \omega \in \Omega: \omega \circ \omega = \underline{false}$ в силе. В противном случае, для $\omega_1, \omega_2 \in \Omega$ могут быть только два случая: если $\omega_1$ и $\omega_2$ имеют, хотя бы одну, объединяющую вершину, тогда согласно трм.\ref{theo:GeneralizedHeapConjunctionTheorem} соответствующий граф куч определён, иначе, результат $\underline{false}$ (обозначив, $\omega_1$ и $\omega_2$ не пересекается). Таким образом, мы показали, что $\Omega$ замкнуто над $\circ$ и что граф кучи может быть получен в результате конъюнкции. В таком случае, соединение успешно установлено.

Далее, ассоциативность должна быть доказана, а именно, что: $m_1 \circ (m_2 \circ m_3) = (m_1 \circ m_2) \circ m_3$ в силе.

Рассмотрев рисунок \ref{fig:HeapGraphConjBeforeAfter}, можно сразу констатировать верность равенства с обоих сторон, так как не имеет значения $a$ и $b$ связаны первыми, либо $a$ связывается с $b\circ c$, потому, что соединяющая вершина $b$ остаётся  инвариантом, когда порядок  конъюнкций меняется.

$G$ создает  полугруппу. Для этого остаётся доказать существование  нейтрального элемента $\varepsilon$ так, что (iii) остаётся в силе. Однако, это следует из обобщённой теоремы о кучах (трм.\ref{theo:GeneralizedHeapConjunctionTheorem}).
\end{proof}

\textbf{Замечание:} Из (i) следует: $c \not \in b \wedge c \neq a$: $a \mapsto b \ \circ \ c \mapsto d \equiv \underline{false}$, и $a\mapsto b \ \circ \ a \mapsto d \equiv \underline{false}$ в силе. Очевидно, если имеется выбор, то безразлично какие две вершины связывать первыми -- результат тот же самый, благодаря  \textit{сходимости} из-за (ii) свойства, которое доказывается позже в лем.\ref{lem:HeapConjunctionGroupProperty}.

\textbf{Замечание:} Замкнутость (i) показывает на  \textit{свойство неповторимости}  подструктурных логик (ЛРП рассматривается как такова), которое остаётся в силе и будет продемонстрировано позже.

%%%%%%%%%%%%%%%%%%%%%%%%%%%%%%%%%%%%%%%%%%%%%%%%%%%%%%%%%%%%%%%%

\begin{proof}
Из-за лем.\ref{lem:HeapConjunctionMonoid} $G$ является  моноидом. Поэтому нам остаётся показать: (i) существование обратного к любому элементу из множества носителя графа кучи, так, чтобы соблюдалось:

и (ii) $\circ$ является коммутирующим оператором.

Начнём доказательство с (ii): в базисном случае «$loc_1 \mapsto var_1 \circ loc_2 \mapsto var_2 = loc_2 \mapsto var_2 \circ loc_1 \mapsto var_1$» индуктивного опр.\ref{def:HeapTermDefinition} равенство, очевидно, соблюдается. Также соблюдается индуктивный случай до тех пор, пока условия от $\circ$ соблюдаются. Условие индукции можно получить, рассмотрев рисунок \ref{fig:HeapGraphConjBeforeAfter}, если на данный момент предположить, что для любых двух связанных куч оператор $\circ$ коммутирует. Но, как только, речь идёт об абстрактных предикатах, понятие о $\circ$-связанных термах может ограничиваться границами предикатов и не могут быть разбросанными как угодно, как подцель в абстрактных предикатах. С практической точки зрения, это не страшно, а наоборот, призывает к лучшей модульности, но с этим необходимо считаться при написании спецификации.

Доказательство продолжается показом свойства (i). Чтобы доказать  обратимый элемент, это когда куча, существует, нужно задаться вопросом: а что с практической точки зрения означает  «\textit{обратимая куча}»? Когда речь идет о естественных числах, то обратимостью сложения будет вычитание на том же множестве носителя. То же самое происходит для поля  комплексных чисел, которые являются расширением  поля вещественных чисел. И хотя вещественные или комплексные числа неперечислимы, тем не менее, на практике расширение  арифметических полей приводит к значительному упрощению вычислительных задач. Хотя ответ отсутствует для множества вещественных чисел на вопрос: «\textit{а что такое} $i$»?  При решении целого ряда задач всё равно $i$ может быть полезным, зная о равенствах: $i^2=-1$ и $e^{i\pi}=-1$.  В связи с этим, было бы справедливо поставить вопрос: почему бы не предположить, что на данный момент существуют кучи, и мы постулируем урав.\ref{eqn:InverseExists}, хотя бы до тех пор, пока не будет доказано обратное?

Итак, что \textit{интуитивно} подразумевается под  \textit{обратимой кучей} или  «\textit{кучей инверс}»? Если речь идёт об естественных и вещественных числах, то имеется модель числовой оси: числа возрастают/уменьшаются относительно нулю в зависимости от направления оси. Как быть с кучами, например, с обыкновенными, формой $a\mapsto b$? Можно ли обратимой куче (инверсия) присвоить инверсию сопровождаемого предиката? --- Это будет не точно. Может ли быть присвоено инверсии кучи пустое значение --- возможно это будет не правильно, так как любая не  пустая куча будет определяться как пустая. Как же будет определяться  инверсия пустой кучи, и т.д.? Такое наивное определение тоже не целесообразно, так как не полностью определено для всех куч. Что, если стороны граней в графе куч просто поменяют стороны, например, из $a\mapsto b$ становится $b\mapsto a$? Это является лишь интересной идеей, но не практикуемо, потому, что  конъюнкция не противоречит такому определению, а надо, чтобы при конъюнкции получалась пустая куча. При таком подходе не определена левая сторона в случае объекта. Можно представить, что слияние кучи удаляет  «\textit{положительную кучу}», если её соединить вместе с «\textit{отрицательной кучей}». То есть, инверсия означает «\textit{трансцендентную}» операцию удаления кучи из памяти, при этом оператор может быть применен к любой, в том числе и сложной куче. Можно обозначить $(a\mapsto b)^{-1}$ как «$a$ не указывает на $b$» или лучше, как «$a$ инверсно указывает на $b$». Первое объяснение неудачное потому, что  «\textit{не указывает}» по ошибке может расталкиваться, как $a\mapsto c$, при этом $\exists c \in \Omega, c\ne b$ -- это подразумевает, что вершина $a$ всё-таки существует и более того, между $a$ и некоторым элементом, чьё содержимое отличается от $b$ --- всё это ложь, потому, что таких предположений не имеется.
Поэтому, гипотетичный случай «\textit{инверсно указывает}», точнее «\textit{инверсно имеется ссылка на}», требует некоторое удаление, несмотря на странное значение, которое позволило бы чистое удаление всех «\textit{лишних элементов}», которые требуют дальнейшую проверку. На первый взгляд это выглядит странно. Пока что, мы специфицировали только части кучи, которые реально существуют.
Ввод инверсии теперь меняет ситуацию.
 Инверсию также нельзя путать с тем, что не должно быть в куче и это есть отрицание предиката, а инверсия кучи --- это операция. В данный момент, мы допускаем и концентрируемся на урав.\ref{eqn:InverseExists}.

Данное уравнение означает, что «\textit{отрицаемое указывает на}» $a \not \mapsto b$ -- это прежде всего предикатное отношение между $a$ и $b$, а в обобщённой форме  отрицательная куча $H^{-1}$, для которой в силе по определению: $a \mapsto b \circ a \not \mapsto b = \underline{emp}$ а также $a \not \mapsto b \circ a \mapsto b = \underline{emp}$, а более обобщёно $H \circ H^{-1} = H^{-1} \circ H = \underline{emp}$. Это означает, что $\omega \circ \omega^{-1}$ «\textit{чисто}» удаляет кучу, т.е. при необходимости также удаляется ненужная грань и вершина из  графа кучи, если больше не остаётся входящих или выходящих граней касательно ещё существующих вершин графа. Таким объяснением данное равенство об обратимости ссылки становится понятным и сейчас не трудно проверить равенство $H \circ H^{-1} \circ H \equiv H$. 
В примере на рисунке \ref{fig:HeapGraphInversion} до применения  инверсии состояние кучи является таким: $d\mapsto a \ \circ \ a\mapsto b \ \circ \ c \mapsto b$, но когда применяется инверсия $\circ (a\mapsto b)^{-1}$, то получаем $d\mapsto a \ \circ \ a\mapsto b \ \circ \ c \mapsto b \ \circ \ (a\mapsto b)^{-1}$, что равно к $d\mapsto a \ \circ \ a\mapsto b\ \circ \ (a\mapsto b)^{-1} \ \circ \ c \mapsto b$ равно к $d\mapsto a \ \circ \ c \mapsto b$, что не полностью очевидно, т.к. оба  указателя не пересекаются. То есть, такое состояние пока очищено не полностью и нуждается в дополнительных шагах для исправления ненужных элементов. Поэтому, следует пройти ещё два шага нормализации.

 \textbf{Нормализация -- Первый шаг}: Шаг является генеричным (общим). Если между рассматриваемыми графами куч существует  мост, как единственная связь между ними, то оператор должен быть заменён на  дизъюнкцию (см. позже).
 
Теперь, когда обнаружен  \textit{мост} между $a$ и $b$, $\circ$ заменяется на $\|$ в оставшимся терме. Результат можно снова считать обоснованным.
Но, возможно, ради полноты вершины, должны быть полностью удалены из соответствующего  графа кучи. Это требуется тогда, когда речь идёт о  локациях  объектных полей. 
 
 \textbf{Нормализация -- Второй шаг}: Удаление вершины $a$ может быть произведено полностью, когда больше не имеются ссылки на $a$ в оставшимся графе куч.

Применив эти два шага  нормализации, можно избежать проблему упомянутых исключений (ср. обобщённые кучи с набл.\ref{obs:SimplificationByGeneralisation}).
\end{proof}

\textbf{Замечание:}  Обобщённые кучи не были обсуждены. Для доказательства корректности необходимо показать $H \circ H^{-1} \equiv \underline{emp}$. Доказательство нужно проводить индуктивно над $\circ$ при использовании $(g_1 \circ g_2)^{-1} \equiv g_1^{-1} \circ g_2^{-1}$, так, что существует  гомоморфизм для «$.^{-1}$» касательно $\circ$ (см. лем.\ref{lem:HeapInversionHomomorphism}).

\textbf{Объяснение:} $H \circ a\mapsto b \circ (a\mapsto b)^{-1}$ означает:

\begin{enumerate}
 \item удалить грань между $a$ и $b$
 \item удалить вершину $a$, если на неё в графе $H$ не имеется больше ссылок 
 \item также удалить вершину $b$, если на неё в графе $H$ не имеется больше ссылок
\end{enumerate}

Свойства  группы позволяют нам установить равенства над термами куч, например, для ускорения  сходимости доказательства или для преобразования в нормальную форму (см. трм.\ref{thes:SimplificationByDiffingHeaps}). Таким образом, можно будет проводить сокращение раздутых термовых представлений куч. Частичные спецификации разрешат сокращённое представление правил (см. раздел \ref{sect:PartialSpec}). Дальнейшая работа включает в себя подключение  решателей для преобразования простейших термов куч.
Необходимо рассмотреть случай, когда содержимым является указатель, например, $a\mapsto o.f \circ  (a\mapsto o.f)^{-1}$, где $o$ объект содержавший поле $f$. Очевидно, $o.f$ не удаляется из динамической памяти, иначе, целостность объектных экземпляров нарушается --- дальнейшее исследование приветствуется, но находится за пределами этой работы (см. дискуссию в разделе \ref{sect:ClassObjects}). Поэтому, по умолчанию уговаривается, что поле остается в памяти как единица целого объекта, но ссылается на \texttt{nil}. Таким образом, замеченные шаги остаются без изменения в силе.

\begin{proof}
 Если нам удастся доказать более обобщённую форму: $G=g_1\circ g_2\circ \cdots \circ g_n$, то проблема будет решена. Для этого необходимо показать $G \circ G^{-1} = \underline{emp}$. Доказать это можно индуктивно, используя $n$. В базисном случае ($n=1$) получаем $g_1\circ g_1^{-1} \equiv \underline{emp}$, что естественно в силе из-за необходимости существования обратного элемента.
 Для индуктивного случая предположим:
 $$G=\underbrace{(g_1\circ g_2 \circ \cdots \circ g_k)}_{G_k} \circ g_{k+1}$$
 тогда для $$G\circ G^{-1} = (G_k \circ g_{k+1}) \circ (G_k \circ g_{k+1})^{-1}$$ обратная часть кучи является настоящим расширением кучи. Правая часть равенства равна

$$\underbrace{G_k\circ G_k^{-1}}_{\underline{emp}}  \circ \underbrace{g_{k+1} \circ g_{k+1}^{-1}}_{\underline{emp}} \ \equiv \ \underline{emp}$$
(из-за индуктивного свойства  обратимости, соблюдая конв.\ref{conv:EmptyHeapInversion}).
\end{proof}

%%%%%%%%%%%%%%%%%%%%%%%%%%%%%%%%%%%%%%%%%%%%%

Поэтому, $x.b \ \| \ x.c$ не действительно для любого объекта $x$ с полями $b$ и $c$, если существует, хотя бы одна общая вершина для любого пути, начиная с $x.b$ или $x.c$.

Допустим, $\Sigma = X_0 \| X_1 \| \cdots \| X_n$ с $n>0$ и $X_j$ имеет форму $x_j \mapsto y_j$, тогда $\Sigma = \Sigma_0 \ \| \ a_0 \mapsto b_0 \Leftrightarrow \forall (a_j \mapsto b_j) \in \Sigma_0: a_j \neq a_0 \wedge b_j \neq b_0$.

\begin{rutheorem}{Моноид дизъюнкция}
$G=(\Omega, \|)$ является  моноидом и  группой, если $\Omega$ является множеством графов куч и $\|$ является  дизъюнкцией  кучи.
\end{rutheorem}
\begin{proof}
В аналогии к предыдущей лемме, прежде всего, $\forall m_1,m_2 \in \Omega: m_1 \| m_2$, только тогда, когда $m_1$ и $m_2$ не имеют  общую вершину. Это всегда так, когда нет пути от $m_1$ до $m_2$ и не существует граф окружающий оба, $m_1$ и $m_2$. Если $m_1$ и $m_2$ различны, то $m_1 \| m_2$ снова являются годной кучей в $\Omega$, потому, что $m_1$ от другой части графа кучи, чем $m_2$ и наоборот, поэтому следует замкнутость.
 Ассоциативность следует очевидно.
$\underline{emp}$ может послужить нейтральным элементом, тогда $\underline{emp} \| m_1 = m_1 \| \underline{emp} = m_1$. Пусть по определению $\underline{emp} \| \underline{emp} = \underline{emp}$ будет в силе.
Последним, уговаривается $s \| s^{-1} = s^{-1} \| s = \underline{emp}$. Это похоже на $\circ$. В общем, кучи следуют этому правилу.
\end{proof}

%%%%%%%%%%%%%%%%%%%%%%%%%%%%%%%%%%%%%%%%%%
 Конъюнкция и  дизъюнкция частей кучи могут быть выражены правилами с помощью дуальных к нему правил следующим образом:

\begin{center}
\begin{tabular}{lcl}
 \inference[$\circ_{[B,C]}$]{U \circ B \ \| \ C}{U\circ B\circ C} & \qquad \qquad &
 \inference[$\|_{[B,C]}$]{U\circ B\circ C}{U \circ B \ \| \ C}
\end{tabular}
\end{center}

 Операции $\|$ и $\circ$ --- дуальные, и они могут быть преобразованы друг в друга, используя дуальную операцию, получив из урав.\ref{eqn:HeapDisjunctionInvariant} и урав.\ref{eqn:HeapConjunctionInvariant}, где «\textbf{;}»  оператор последовательности.\\\\

Равенства в силе, из-за  самообратимости операции и поставленного утверждения о существовании обоих вершин кучи $B$ и $C$.

\textbf{Замечание:} Так как  нейтральным элементом для операций $\circ$ и $\|$ является $\underline{emp}$, то нельзя определить  \textit{(алгебраическое) поле} (например  поле Галуа) над этими операциями из-за того, что лем.\ref{lem:HeapConjunctionMonoid}, лем.\ref{lem:HeapConjunctionGroupProperty} и лем.\ref{lem:DistributivityForConjDisj} в силе и несмотря на то, что множество носителем $\Omega$ конечное, поэтому любая  куча конечная и все операции над ними также производят конечную кучу.

\textbf{Замечание:} В аналогии к логическим конъюнктам $\wedge$ и $\vee$, нормализованная форма с помощью $||$ всегда существует, применив ранее упомянутые равенства и лем.\ref{lem:HeapInversionHomomorphism} для того, чтобы произвести  инверсию обобщённых куч.\\

Для оптимизации логического вывода с помощью уменьшения объёма графа, необходимо попытаться вывести операцию $\|$ как можно дальше наружи в термах кучи, применив  дистрибутивность, либо переставляя несложные кучи так, чтобы левые части ссылок были отсортированы по имени локации по лексикографическому порядку. Идея заключается в сокращении повторных поисков, например, для  \textit{пошаговой верификации}, так, что только меняющиеся части кучи принадлежат повторному вычислению.

Частично упорядоченное множество  (ЧУМ) можно установить над графами для подмножества графов, операторам слияния ЧУМ   служит $\circ$. Минимальный элемент является  пустой кучей, а максимальный является  полным графом кучи. Например, закон  \textit{абсорбции} не действительный для  куч. ЧУМ не может быть полной решёткой.
На рисунке \ref{fig:HeapGraphPoset} \textit{ЧУМ} $G$ содержит $\{G_1,G_2,G_1',G_2',G_1''\}$, которые соблюдают следующие неравенства по возрастающему порядку $G_1 \sqsubseteq G_1' \sqsubseteq G_1''$, $G_2 \sqsubseteq G_1'$ и $G_2 \sqsubseteq G_2' \sqsubseteq G_1''$. Кучи $\bot$ и $\top$ всегда присутствуют и при необходимости могут быть добавлены, поэтому каждый раз по умолчанию могут быть не указаны. $G_1''$ является максимальным элементом, а $inf(G)=\underline{emp}$ минимальным элементом, где $\sqsubseteq$ определяет подмножественное  соотношение графа. Нетрудно убедиться в том, что две не соединённые кучи при конъюнкции (т.е. равны $\underline{emp}$ из-за опр.\ref{def:HeapConjunctionDefinition}) в соответствующей  диаграмме Хассе всегда остаются несвязанными.
Слияние всегда верное, потому, что: (первое противоречие) $a\mapsto b \ \circ \ a\mapsto d$ не может следовать первой  конъюнкции, либо другой  сложной куче. Из-за (второго противоречия) $a\mapsto b \| b\mapsto d$ противоречит определению $\|$. Однако, необходимо учесть, что порядок ЧУМ может пострадать после ввода  инверсии кучи (ср. ранее), если использовать инверсию без соблюдения ограничений. На данный момент нас вполне устраивает использование  инверсии для «\textit{более удобного}» сравнения куч, в частности с проверяемой спецификацией, поэтому свойство  локальности из главы \ref{chapter:expression} остаётся.

\subsection*{Классовый экземпляр как куча}

Рассмотрим подробнее, как объектные экземпляры классов моделируются как граф куча согласно рисунку \ref{fig:ObjectAttributesScheme}.

\def\gA#1{\save
[].[]!C="gA#1"*[F]\frm{}\restore}

\def\gB#1{\save
[].[dd]!C="gB#1"*[F]\frm{}\restore}

\def\gC#1{\save
[].[ddr]!C="gC#1"*[F]\frm{}\restore}

На рисунке \ref{fig:ObjectAttributesScheme} а) отдельное присвоение к существующему графу представлено (без подробностей), объект в прямоугольной рамке является некоторым  объектным экземпляром. На рисунке \ref{fig:ObjectAttributesScheme} b) только второе  поле объектного экземпляра присваивается некоторое значение  сходимого типа данных, все остальные поля присвоены \texttt{nil}. Рисунки \ref{fig:ObjectAttributesScheme} c) и \ref{fig:ObjectAttributesScheme} d) имеют тоже самое значение как и c), только присвоение производится для всех полей, например по  \textit{лексикографическому порядку}. Далее вводим дополнительные ограничения, но без ограничения общности, которые будут полезными для, позже вводимых, операторов:

Чтобы не противоречить дальнейшим определениям с одной стороны, необходимо иметь возможность быстро проверять соотношения между двумя вершинами. Данная не пустая куча состоит из одной или более того  простых куч. Для проверки, связаны ли две вершины графа кучи, необходимо проверить в худшем случае все левые стороны простых  конъюнктов, т.е. необходим перебор всех граней графа куч.

С другой стороны, чтобы получить для одной вершины все соседние вершины, нам выгоднее иметь модель, которую можно было бы итерировать по вершинам, а не по граням. Для быстрого определения соотношения (например соседства), предлагается например, \texttt{reach\-es(x,y)}, \texttt{reach\-es(x,Y)}, \texttt{reach\-es(X,y)}, \texttt{reach\-es(X,Y)}, где \texttt{x} является определённой вершиной, а \texttt{X} является (под-)множеством вершин (аналогичное для \texttt{y} и \texttt{Y}). Модель, основанная на вершинах, отражена на рисунке \ref{fig:InvertedVertexGraph} маленькими закрашенными квадратиками, соответственно, позволит получить более эффективную итерацию. Обе модели a) и b), основанные на вершинах и гранях, являются  дуальными и они могут быть преобразованы друг в друга: вершины по серединам граней кодируют начало и конец (см. d)), как одна вершина с объединяющим именем и связываются с соседними вершинами (см. рисунок \ref{fig:InvertedVertexGraph} a) и b)). Очевидно, преобразование из одной модели в другую и обратно, согласно схеме преобразования на рисунке \ref{fig:InvertedVertexGraph} c)  возможно без потерь, т.е. отображение между обеими структурами является биективным независимо от направления граней (см. c)).\\

Поэтому, нам разрешается преобразовать графовое представление с конъюнкциями или дизъюнкциями, так нам удобнее для решения.

\subsection*{Частичная спецификация куч}

Как было сказано ранее, в опр.\ref{def:HeapTermDefinition} объектные экземпляры можно рассматривать, как хранители полей данных $obj.f_1 \mapsto .. \circ obj.f_2 \mapsto .. \circ obj.f_n \mapsto .. $. Все поля создают некоторую кучу по отрывкам, но в отличие от абстракции, это преобразование можно охарактеризовать как  \textit{конкретизацию}.
Поля классовых объектов имеют ограничения в связи с наименованием и типами, которые должны совпадать с соответственным классом. Все поля одного объекта существуют независимо друг от друга, поэтому они являются простыми кучами, которые связаны между собой с помощью $\circ$-конъюнкции как это уже было упомянуто. Поля объектов не могут быть утилизированы по отдельности (см. набл.\ref{obs:GraphIRHeap} и последующую дискуссию).
Локальные переменные поля объектов также должны иметь возможность специфицировать части, т.е. подклассовые объекты. В аналогии к необъектным указателям можно определить константные функции над кучами, например,  $\underline{true}(obj)$ или $\underline{false}(obj)$ из опр.\ref{def:HeapTermExtendedDefinition} и опр.\ref{def:HeapTermDefinition}. В отличие от необъектных указателей, \textit{\underline{true}}(obj) вводит объектный терм в качестве дополнительного параметра.
Таким образом, абстрактные предикаты могут быть использованы для дальнейшей модуляризации спецификации, которая синтаксически и семантически не значительно отличается от необъектного случая.\\

В соответствии с предложенным методом распознавания решателя $a \circ a$ на основе  вталкивания и  выталкивания в/из стека согласно данному уровню абстрактного предиката (см. главу \ref{chapter:APs}), теперь входящие и выходящие термы абстрактных предикатов, возможно, отслеживать и пропускать при повторных или ненужных вычислениях. Сравнение может производиться стеком для одинакового уровня, либо отметками между любыми уровнями предикатов с помощью транслирующих правил.

Подробности следуют в главе \ref{chapter:APs}. Нетрудно убедиться в верности $\wedge,\vee,\neg$-конъ\-юнк\-ций для \textit{\underline{true}}(obj) и \textit{\underline{false}}(obj). Эти константные функции могут принимать любое число объектных полей (см. трм.\ref{thes:SimplificationByDiffingHeaps}) и их булевое обозначение не зависит от конкретных полей. Однако, сложные кучи в комбинации с ними, могут иметь неожиданное поведение, которое исключаемо при анализе контекста. Частичные спецификации, которые используют константные функции, частично описывают поля, но могут покрывать гораздо больше. Такми образом, спецификации частичные, но охватывают больше куч (см. трм.\ref{thes:IncompletenessForCompleteness}). Более того, высчитывание позволяет сравнивать имеющие кучи и выявить те кучи, которых не хватает. Это позволяет полностью определить все входные кучи в правилах.

\begin{ruexample}{Неполный предикат №.1}
Дан объект $a$, который имеет три поля $f_1$, $g_1$ и $g_2$. \comp{.}{} означает  семантическую (неявную) функцию над термами куч типом $ET \rightarrow ET \rightarrow \mathbb{B}$, где $ET$ из опр.\ref{def:HeapTermExtendedDefinition}, $\mathbb{B}$ булевое множество, где первая расширенная куча является ожидаемой, а вторая является полученной кучей, то тогда:
 \begin{eqnarray*}
   & \comp{$a.f_1 \mapsto x\circ \underline{true}(a)$}{} = & \comp{$a.f_1 \circ a.g_1 \circ a.g_2$}{}\\   
   = & \comp{$\underline{true}(a) \circ a.f_1 \mapsto x$}{} \neq & \comp{$p(a) \circ a.f_1 \mapsto x$}{}
 \end{eqnarray*} 
где, $p$ абстрактный предикат означает $\underline{true}(a)$ (см. опр.\ref{def:IncompletePredicates}).

Однако, \comp{$a.f_1 \mapsto x \circ p(a)$}{} означало бы равенство потому, что благодаря распознаванию, используя актуальное  стековое окно, находит все оставшиеся поля, если даже ниже спрятано несколько уровней абстрактных вызовов. \comp{.}{} является  гомоморфизмом касательно обсуждённых константных функций и конъюнкции $\circ$.
\end{ruexample}

\begin{ruexample}{Неполный предикат №.2}
 \comp{$\underline{true}(a) \circ \underline{true}(a)$}{} = \comp{$a.f_1 \circ a.g_1 \circ a.g_2$}{} $\circ$ \comp{$\underline{true}(a)$}{} = \comp{$a.f_1 \circ a.g_1 \circ a.g_2$}{} $\circ \ \underline{emp}(a)$.
\end{ruexample}

\subsection*{Обсуждения}

Одна пространственная операция была заменена на две строгие. Изначальные основные свойства ЛРП не поменялись, кроме неограниченной инверсии кучи, которая тщательно обсуждалась ранее. Если в трм.\ref{theo:ReynoldsHeapProperties} заменить «$\star$» на «$\circ$», а в случае дизъюнкции «$\star$» на «$||$», то верность аксиом следует непосредственно, ради исключения аксиомы №5 для конъюнкции как раз из-за ужесточения «$\star$».

Имея форму, которая позволяет нормализовать термы кучи, специфицируемые кучи можно теперь анализировать линейно (см. с разделом \ref{sect:PartialSpec}).
Из-за требования о  неповторимости простых куч, комплексные кучи могут быть эффективно исключены с помощью   «\textit{мемоизатора}» (с англ. «\textit{memoiser}»). На практике, необходимо исключать и обнаруживать повторяющиеся локации (возможно с различными, но одинаковыми значениями), как это было упомянуто в разделе \ref{sect:ConjunctionAndDisjunction}. Иначе, кучи и граф кучи становятся противоречивыми и свойства ЛРП нарушаются. Это приведёт к полной неповторимости теорем. Принцип должен соблюдаться: одна простая куча специфицируется один раз.
Простая куча, специфицируемая в одном месте не должна специфицироваться заново, аналогично принципу локальности: одно изменение локально производится только в одном месте.
Если принципы не верны, то и применение правил не верное, а следовательно, из ложного предусловия могут выводиться любые результаты, поэтому повторные кучи должны исключаться.

Однако, эту проблему в общем удастся решить только динамически, но без запуска программы. «\textit{Динамически}» при статическом анализе означает, что во время верификации,  абстрактные предикаты естественно могут быть произвольными. Абстрактные предикаты могут интерпретироваться процедурально (см. главу \ref{chapter:APs}). Следовательно, используется  \textit{стековая архитектура} для их анализа (ср. главу \ref{chapter:logical}). Она очень близка к  операционной семантике предложенной Уорреном \cite{warren83}, где стек имеет ссылки на предыдущие стековые поля. Эта семантика отличается от семантик классических языков программирования, которая ради исключения, при использовании  параметров по вызову, более близка к реализации логического  языка Пролог. Также представление и переход предикатов более похожи на правила Пролога. Незначительно модифицируемая  семантика машины Уоррена может быть применена к интерпретации абстрактных предикатов, с помощью строгих операций  конъюнкции и  дизъюнкции, с помощью которой удастся распознавать, например, $\forall a \in \Omega.a \circ a$ для соблюдения неповторимости.

 Мемоизатор может  кэшировать только те вызовы  абстрактных предикатов, которые не меняют глобальные состояния. Мемоизатор может запоминать, что имеются (i)  входные или (ii)  выходные или (iii)  входные-выходные переменные термы. Если далее, неограниченные символы внутри одного абстрактного предиката поздним подвызовом ограничиваются, то это необходимо учесть при порядке вычисления подвызовов (определение должно следовать  порядку слева направо, чтобы разрешить подобные конфликты, см. главу \ref{chapter:APs}).

Пролог \cite{sterling94}, как общий логический язык программирования может быть использован (см. тез.\ref{thes:PrologMakesHeapSpecSimpler}), как платформа, основанная на  рекурсивно-процедуральных правилах и термах для  логического вывода, используя теоремы о расширенных термов куч и абстрактных предикатов.
В Прологе  обобщённая схема рекурсивной индукции Пиано может быть всегда определена как «\texttt{p(0).}» для базисных случаев, а «\texttt{p(n):-n1 is n-1,p(n1).}» для  индуктивных случаев, используя  вспомогательный предикат \texttt{is} для высчитывания \texttt{n1}. Не трудно убедиться в том, что в Прологе можно выразить любую обобщённую  $\mu$-рекурсивную схему предикатов. В общем, предикаты могут быть не определены (например, когда процедурный вызов не приостанавливается).
Преимущество  правил Хорна  Пролога, в отличие от классических  разовых функций (как они были использованы, например в \cite{parkinson05-2}), это способность рассматривать термы предикатов как (i),(ii) или даже (iii), объединившие таким образом экспоненциальное количество различных классических разовых функций. Не каждый аспект может быть определён, поэтому вопрос об  обратимости функции требует дополнительного внимания.  Арифметические вычисления, а также  \textit{зелёные} и  \textit{красные отсечения} \cite{sterling94} поисковых пространств, являются одной возможной причиной, почему предикат может быть не обратим. Арифметические выражения в Прологе вычисляются с помощью  оператора \texttt{is}. Обратимость арифметических выражений можно частично восстановить, если заменить  натуральные числа Чёрческими числами (см. \cite{haberland08-2}), а фундаментальные операции обосновать только, если использовать константу,  \textit{одинарный функтор} и  унификацию, как универсальную  \textit{монаду} базовой операции. Говоря обобщённо, необходимо гарантировать сильную корреляцию между входным и выходным результатом, которая должна стать изоморфизмом отображением для полной обратимости. Более того, прологовские отсечения могут быть заменены без потери общности и выразимости, так как отсечения являются лишь  синтаксическим сахаром (см. главу \ref{chapter:expression}).\\\\
Язык  «\textit{Object Constraint Language} (OCL)» \cite{oclspec} является языком спецификации для объектных экземпляров. «\textit{OCL}» является расширением языка «\textit{UML}», и существует, как графическая нотация для утверждений, либо как формулы. Формулы «\textit{OCL}» выражают часть предикатной логики. Имеется: --- поддержка квантификации  переменных, поддержка массивов и абстрактных типов данных/классов и  полиморфизм с помощью подклассов. «\textit{OCL}» разрешает описывать цикл жизни объектов и методов. Однако, «\textit{OCL}» не знает об  указателях или  псевдонимах. Утверждения об указателях отсутствует, поэтому на этапе моделирования и   быстрой прототипизации имеются ограничения, как было описано в главе \ref{chapter:intro} (см. \cite{haberland14-1}, закл.\ref{cor:StrengtheningOCL}).

Предлагается, чтобы указатели записывались в качестве множества наименований объектного экземпляра. Таким образом,  псевдонимы записаны в одном месте, альтернативно вводится множество указателей, которое существует независимо от существующих экземпляров объектов. Состояние объектов совпадает с состоянием вычисления стека/кучи. Сложные объекты имеют поля, которые просты или ссылаются на любые другие объектные экземпляры. Абстрактный предикат предлагается логическим предикатом (см. главу \ref{chapter:logical}), возможно, с символами. Нет обязательства, что в одном предикате описывается только один объект. Но рекомендуется определять одним абстрактным предикатом один объект целиком зависевшие объекты рекомендуется описывать отдельными предикатами. Ситуация, когда одним предикатом описываются два или более того объектных экземпляров не исключена, но не приветствуется ради модульности. Пространственное соотношение между объектами описывается операторами $\circ$ и $||$. Соотношения не требуют обязательного обозначения пространственности, если из контекста ясно, что речь идёт расширении представления «\textit{OCL}».

Будущая работа может заключаться в расширении с  абстрактными предикатами \cite{haberland16-1}. Начатое предложение следует расследовать дальше, особенно, учитывая постановления из конв.\ref{conv:RestrictedObjects} и опр.\ref{def:HeapTermExtendedDefinition}.
Ожидается повышенная выразимость, модульность и абстракция.
Хороший обзор нынешних попыток расширить существующие вычисления с указателями, можно найти в главе \ref{chapter:intro}, а также в \cite{parkinson05-2}.

%%%%%%%%%%%%%%%%%%%%%%%%%%%%%%%%%%%%%%%%%%%%%%%%%%%%%%%%%%%%%%%%%%%%%%%%%%%%%%%%%%%%%%%%%%%%%%%%%%%%%%%%%%%%%%%%%%%%%%%%%%%%%%%%%%%%%%%%%%%%%%%%%%% 6 APs
%%%%%%%%%%%%%%%%%%%%%%%%%%%%%%%%%%%%%%%%%%%%%%%%%%%%%%%%%%%%%%%%%%%%%%%%%%%%%%%%%%%%%%%%%%%%%%%%%%%%%%%%%%%%%%%%%%%%%%%%%%%%%%%%%%%%%%%%%%%%%%%%%%%
\section*{Автоматическая верификация с предикатами}
% (40p.)

Рассмотрим простой пример из области вычислительной геометрии  \cite{deberg08}. \textit{Двусвязный список}  содержавший грань  данного многогранника,  послужит примером, в котором каждая грань  связывает две вершины некоторого 2 или 3-мерного пространства. Сеть многогранника   после триангуляции  распадается на  треугольники. Каждая вершина  дважды связана с двумя соседними вершинами. Каждая грань начинается и заканчивается определёнными гранями. Для определения нормального вектора, достаточно иметь  треугольник.
В простой куче  локализатор  указывает на содержимое.  Для данного примера это могут быть вершины или грани. Согласно упомянутому двусвязному списку, каждый элемент имеет ссылку вперёд к следующей грани и назад к предыдущей грани.
Простая куча содержит связанные между собой грани. Не связанные между собой грани по определению не указываются. Уговаривается, что принудительное отсутствие кучи (см. главу \ref{chapter:stricter}) исключается, либо отдельно не рассматривается в связи с Абельской группой из-за ранее упомянутых ограничений.
Представим себе, что каждый раз как ссылаться на указатели, копии  всех трёх вершин заносятся в память. Нетрудно убедиться в том, что такой подход малоэффективен. Если работать с указателями,  то эта проблема отпадает, особенно когда вводится дополнительный уровень  абстракции. Эти абстрактные предикаты позволяют более интуитивно выражать  сложные кучи. Например, прологовская подцель \texttt{face(p1,p2,p3)} может означать, что три вершины \texttt{p1,p2,p3} связаны между собой в одном треугольнике вместо того, чтобы каждый раз полностью специфицировать $\exists v1.v2.v3$, при \texttt{p1.data $\mapsto$ $v1$} $\star$ \texttt{p2.data $\mapsto$ $v2$} $\star$ \texttt{p3.data $\mapsto$ $v3$} $\star$ \texttt{p1.next $\mapsto$ p2} $\star$ \texttt{p2.next $\mapsto$ p3} $\star$ \texttt{p3.next $\mapsto$ p1} $\star$ \texttt{p1.prev $\mapsto$ p3} $\star$ \texttt{p3.prev $\mapsto$ p2} $\star$ \texttt{p2.prev $\mapsto$ p1}.

Прежде всего, абстракция означает обобщённость, вводя дополнительные параметры. Под абстрактным предикатом подразумевается правило Хорна с произвольным количеством параметров. Хотя некоторые авторы стремятся использовать «\textit{Абстрактный Предикат}»  как новый термин \cite{parkinson05}, нужно заметить, что абстракция  не нуждается в дополнительном определении (см. главу \ref{chapter:expression}), как новой концепции. Аналогично распространяется и на предикаты. Ясны концепции абстракции и предикатов, поэтому считается не целесообразно заново обозначать термины, тем более, имеются случаи, когда к предикатам добавляются  параметры, например, в логике предикатов  первого порядка. Это и является причиной, почему предикаты по-прежнему рассматриваются как классические, а «\textit{абстрактные}» как прилагательные к предикату. «\textit{Абстрактный Предикат}» не отличается от термина предиката, поэтому отдельно взятое семантическое определение просто не существует. Абстрактный предикат имеет любое (включая ноль) количество термовых параметров и может содержать любую последовательность (включая ноль) подцелей ранее декларированных предикатов.
В данном примере \texttt{face(p1,p2,p3)} равен той самой развёрнутой  $\star$-формуле подцелей, которая была указана ранее.
В зависимости от того, в каком состоянии находится вычисление подцели \texttt{face}, либо свёртывание, либо развёртывание, может быть целесообразным. Предикат \texttt{face} может зависеть от других предикатов.
Однако, на данный момент просто не достаточно известно о том, когда необходимо свёртывать или развёртывать определение абстрактного предиката. Если развёртывание проваливается, то это может быть потому, что это в принципе не возможно, а также потому, что развёртывание и свёртывание были совершены в не  правильный момент, либо в неправильной последовательности. 
Далее, будет рассматриваться новый подход, который позволит решить проблему автоматизации с кучами.

Уоррен \cite{warren83} использует термин «\textit{программирования через доказательство}»  для того, чтобы выразить, что Пролог может быть использован как язык программирования, который используется для нахождения решения формулированного запроса  правил Хорна. Изоморфизм  \textit{Карри-Хауарда} \cite{mitchell96} гласит о взаимосвязи между доказательством и программированием. Касательно куч, философский лозунг данного подхода можно охарактеризовать как «\textit{доказательство является проблемой синтаксиса}», это означает, что с помощью синтаксического перебора  можно доказать корректность специфицируемой кучи и представительство кучи близко к моделям по программированию, т.е. к  прологовским правилам. Главным наблюдением этой главы является: абстрактные предикаты описывают на самом деле  формальный язык (см. тез.\ref{thes:PrologMakesHeapSpecSimpler}). Позже мы убедимся в том, что формальный язык является логическим языком программирования. Следовательно, проблемы свойств куч, которые являются семантическими, можно решить синтаксическим распознаванием.

\subsection*{Сжатие и развёртывание}

Представленный в этой главе подход сильно отличается от существующих традиционных.

По Рейнольдсу  \cite{reynolds02,parkinson05-2} пространственный оператор $\star$  связывает две раздельные кучи, при котором основы \textit{подструктурной логики}  остаются в силе, а  \textit{правило сужения} не в силе. Как было изложено в главе \ref{chapter:expression}, классический оператор $\star$ многозначен, поэтому, далее используется только строго соединяющий оператор «$\circ$» в качестве пространственной конъюнкции.

Абстрактные предикаты  задаются пользователем, как это было предложено в «\textit{Verifast}»  \cite{jacobs11}. Вывод  может осуществляться быстрее с помощью «\textit{тактик}», которые могут определяться индуктивно в системе «\textit{Coq}»  \cite{bertot04} и выводиться в  \textit{полуручном режиме}. Системы основаны на принципе «\textit{сжать/распаковать}» (fold/unfold)  \cite{hutton98}, могут при полной подсказке полностью независимо от начала до конца вывод осуществить логически. Подсказки доказательства  показывают на тот  \textit{редекс}, который необходимо предпринять для выхода из неопределённого состояния и для продвижения верификации в целом.

Пролог  здесь используется как  язык утверждения. \cite{kowalski74} наглядно демонстрирует пригодность к доказательству формул в предикатной логике с помощью  правил Хорна. \cite{kowalski74} в прошлом предлагал использовать Пролог в качестве  языка программирования, что, увы, не всегда возможно из-за вычислимости (см. главу \ref{chapter:logical}). \cite{warren83} содержит определение реализации логического вывода, опираясь на  операционную семантику. Каллмайер \cite{kallmeyer10} демонстрирует использование Пролога для распознавания  \textit{морфем} и  \textit{мутаций грамматики} в  естественных языках. В частности,   «\textit{сопряжённые деревья}» предлагаются как механизм обработки мутаций в естественных языках на основе явных определённых $\lambda$-термов, которые исключаются в  формальных языках, в частности, языках программирования во избежание  многозначности. Примеры показывают многозначную перегруженность и трудность анализа из-за экспоненциального роста необходимых проверок посторонних условий. \cite{matthews98} на примерах поэтапно излагает, как Пролог может способствовать к решению проблем многозначности синтаксического перебора. Мэттьюс широко использует  рекурсивные распознаватели, которые работают с деревьями  и в Прологе реализованы эффективно и просто. Реализации являются стековыми автоматами, распознающие  LL(k)-грамматики с модификациями: конечные состояния  выделяются явным образом, а рекурсивные прологовские правила имитируются стеком. Мэттьюс использует «\textit{списки разниц}»   для реализации распознавателя  регулярных языков, который на самом деле основан на «\textit{автомате частичных производимых регулярных выражений}»  \cite{brzozowski64}. Бжозовский предлагает «\textit{прологовские формальные грамматики DCG}»    и встроенные команды Пролога  для изменения баз знаний  во время запуска для увеличения гибкости, которую он считает необходимо расширять. Однако, подход Бжозовского имеет недостаток в том, что выразимость ограничена регулярностью (см. главу \ref{chapter:expression}).
Работу Перейры \cite{pereira12} можно оценивать как классическую монографию по Прологу и обработку естественных языков. Однако, подход Перейры имеет фундаментальные ограничения, которые всё-таки легко можно устранить среди множества образцовых и отличающихся примеров. Невозможность разрешить лево-рекурсию  правил Хорна, хотя решение в принципе существует, т.к.  LL(1)-распознаватель не в состоянии опознать всех предшественников  для решения принадлежности правила. Далее, Перейра вводит \textit{$\lambda$}-вычисления над деревьями для распознавания естественных языков. Таким образом, морфемы  и лексемы  связываются и получают зависимое значение от введённых параметров.
Далее, введём первое прототипное определение кучи, учитывая опр.\ref{def:FirstOrderPredicateLogicFormula} и главу \ref{chapter:expression}.

\subsection*{Предикатное расширение}

Утверждение \textbf{emp} означает  пустую кучу, которая верна, если данная куча пуста. Пустая куча является нейтральным элементом относительно пространственным связям между кучами (см. главы \ref{chapter:expression},\ref{chapter:stricter}). «$\star$» разделяет две кучи на две  независимые кучи. 
В этой главе мы не будем различать ужесточение между «$\circ$» и «$\star$» (см. главу \ref{chapter:stricter}) --- ради общности мы исходим из конъюнкции куч.
Утверждение \textbf{true}  означает, любая куча (в том числе пустая) разрешается, а \textbf{false} означает, любая куча не разрешается. Эти определения близки к определению по Рейнольдсу \cite{reynolds02}. Основой всех определений по Рейнольдсу является  \textit{обыкновенная куча}: $x \mapsto E$, где $x$ некоторый  \textit{локализатор} (например доступ к объектному полю $o1.field1$), а $E$ является некоторым допустимым выражением, которое присваивается ячейке памяти обозначающейся локализатором $x$. Проверка совместимости типов проводится на более раннем этапе \cite{haberland14-1}. На данный момент безразлично, явное значение, либо ссылка на ячейку в памяти содержится в качестве содержимого по выражению (см. \cite{burstall72}).
Рассмотрим две любые сложные кучи  на  Прологе в рисунке \ref{ExampleComplexHeaps1}.

Здесь \texttt{p2} означает некоторый предикат с двумя  символами \texttt{X} и \texttt{Y}, которые представляют собой некоторые значения, которые указываются локализаторами \texttt{loc2} и \texttt{loc3}. В отличие от этого, \texttt{p1} определяется через предикат \texttt{p2}.
Как только мы вызываем \texttt{p2} с двумя синтаксически корректными аргументами, так мы имеем одну форму $a(\vec{\alpha})$.
Вспомним, Пролог не может найти решения для синтаксически корректных, но семантически некорректных термов потому, что «\textit{семантически некорректно}» означает, нахождение не выводимых термов для данных прологовских правил.

Интерпретация  формулы $H$ для данной кучи означает отображение от двух куч, т.е. данной и сравниваемой кучи, в  булевую ко-область. Это означает, что если две данные кучи совпадают, то интерпретация успешна, в противном случае, наоборот. Для данных интерпретаций рассматривается только дедуктивный вывод (см. разделы \ref{Intro:HoareTriple} и \ref{Intro:LogicalReasoning}). Поиск вывода завершается успехом тогда, когда запрос успешен, во всех остальных случаях завершается провалом. Несомненно, это точно то, что ожидается получить от предлагаемого поведения (см. тез.\ref{thes:ProvingEqualsParsing}).\\
Ради простоты согласуем, что формулы куч должны быть  нормализованы в форму
$$a_0 \star a_1 \star \cdots \star a_n   \equiv \prod^n_{\forall j} a_j,  n \ge 0$$
Далее, $\wedge$ и $\vee$-связанные графы куч задаются в  Прологе в виде запросов формой\\
 $s_j, s_{j+1}, \cdots , s_{j+k}$. Альтернативно можно  дизъюнкцию прологовских целей разбить далее на некоторые  альтернативные правила, головы чьи различаются с помощью оператора «\textbf{;}». Отрицание утверждения  рассматривается как  отрицание предиката. В общем, отрицание последовательности не означает отрицание предиката потому, что последовательность для всех может быть просто не определена, кроме некоторого домена, это надо учесть. Двойное отрицание в общем случае не действительно при вызовах подцелей потому, что предикат может быть не тотален. Экзистенциальные переменные могут быть введены в любом месте правила Пролога, однако ожидается, что все введённые переменные когда-то присваиваются и используются.

Константные функции,  как например, \textbf{true} и \textbf{false}, являются «\textit{синтаксическим сахаром}», т.к. они могут быть заменены на любые другие кучи, которые квалифицируются в качестве вставных куч. Использование константных функций упрощает в спецификациях все возможные кучи. \textbf{true} может быть сопоставлена  булевой истине, \textbf{false} наоборот   противоречием.

\begin{rucorollary}{Корректность кучи}
 Любая синтаксически корректная формула кучи описывает соответствующий  граф кучи. Более того, любой граф кучи может быть представлен соответствующей  формулой кучи. В общности обе стороны действительны ради исключения бесконечных куч.
\end{rucorollary}

Опр.\ref{def:NonformalHeapGraph} является уточнением впервые введённого графа кучи из набл.\ref{obs:GraphIRHeap}.

\subsection*{Предикаты как логические правила}

Абстрактные предикаты  позволяют абстрагировать от обыкновенных куч к более сложным кучам, но более интуитивным человеку, используя выражения и формулы.
Например, \cite{parkinson05-2} вводит абстрактные предикаты, которые аннотируют данную входную программу и переводятся вместе с программными операторами до уровня ассемблера.
Представленный подход автоматизации является попыткой преодолеть разрыв между спецификацией и логическим выводом.
Пролог используется в этой главе как язык программирования, в котором утверждения и абстрактные предикаты о кучах специфицируются.
Однако, между программой и языком утверждений существует семантический разрыв: одновременно имеются два параллельных формализма и реализации. Они приводят к всё более различающимся нотациям и представлениям.
Естественно, язык программирования может и должен различаться, когда речь идёт о возможном  императивном языке программирования. В согласии с  вычислением Хора, логические формулы описываются \textit{логически}.
Увы, иногда это нарушается (см. главу \ref{chapter:intro}), а также имеются ограничения вычислимости (см. главу \ref{chapter:logical}).
Так, например, переменные (объекты) используются как  локальные переменные вместо термов, вследствие чего, имеется целый ряд ограничений. Таким образом,  язык спецификации, точнее ее частицы, «\textit{деградируются}» в последовательность команд и больше не имеется ничего общего с изначальным замыслом вычисления Хора. Частицы не имели бы ничего общего с «\textit{декларативно-логической парадигмой}» (см. набл.\ref{obs:EqualityOfProofElements} и набл.\ref{obs:StackBasedCalls}): необходимо описать состояние вычисления, используя символы и предикаты. Символ и его  диапазон годности всё-таки отличается от локальных переменных. Если символы «\textit{вдруг переписываются}», то вряд ли это можно считать символом. 
По определению символам не приписываются новые значения заново, а переменным приписываются.
Разница    может казаться не слишком большой, но для определения и описания это имеет очень серьезные последствия. Вычисления часто (но не всегда) описываются символами, без переменных. При описании центральной концепцией являются  термы,  предикаты и рекурсии, а не цикл или условный переход, это нужно учесть. В обеих моделях можно установить минимальный набор Тьюринг-вычисляемых программ. Для решения верификации необходимо сравнивать состояния вычислений для того, чтобы определить, вычисление было совершенно правильно или нет. Проблема сравнения не исключает возможность и необходимость вычислять  арифметические или  алгебраические выражения, но при этом, подход и замысел вычисления Хора декларативен.

Подход представленный в \cite{berdine05-2} вводит  символы в кучах, но с ограничениями. Например, отсутствует возможность описывать целые кучи, как например $X \star Y$. В отличие от этого, мы допускаем символы без ограничений полностью как они допускаются в Прологе (см. главу \ref{chapter:logical}). В различных вариантах мы вынуждены будем выбирать только между теми правилом, которое имеет более длинное совпадающее предусловие -- как это было реализовано, например, в \cite{berdine05-2}. Мы принципиально следуем поиску Пролога, что в обобщённом случае может быть слишком много, но   «\textit{метод ветвей и границ}» \cite{sterling94}, \cite{bratko01} позволяет нам произвольно сужать поисковое пространство. Таким образом, мы не ограничиваем себя в некоторой методологии, либо \textit{эвристике}, либо  \textit{тактике}. Любая методология может меняться частично, либо полностью --- мы в состоянии всё это учитывать.

Пролог используется для того, чтобы определить, какие существуют подлежащие в кучах и какие связи между ними, для этого мы используем предикаты. Поэтому, можно считать, например \cite{warren83} более близким подходом, чем функциональный или  императивный. Когда речь идёт о проверке состояния куч, в общем легче описать, опираясь на  факты и  правила, чем на последовательность инструкций. С помощью абстрактных предикатов описываются кучи. За компактность, представление фактов и правил, ответственный --- разработчик программы. Следующие формализмы помогут преобразовать абстрактные предикаты в  формальную грамматику для дальнейшего представления. Все перечисленные преимущества характерны для Пролога и используются далее в предложенном методе.

Неслучайно опр.\ref{def:AbstractPredicateRule} близко к опр.\ref{def:PrologRule}. Первое определение используется для описания и вызова предикатов куч.
Уговаривается по умолчанию, что $a$ действительно всегда тогда, когда все  \textit{подцели} $q_{k,j}$ в $a$ в силе для $0 \leq j \leq n$. 
Синтаксис предикатного правила определяется расширенной формой Бэккуса-Наура как показано на рисунке \ref{fig:PrologRulesEBNF}.
\texttt{<number>} определяет любое прологовское число, а \texttt{<atom> '(' <arguments> ')'} определяет некоторый  функтор с простым именем \texttt{<atom>} и любым количеством аргументов. \texttt{<var>} означает некоторый переменный символ, который начинается с большой буквы, например \texttt{X}. Рисунок \ref{fig:mapFunctionalExample} в главах \ref{chapter:expression} и \ref{chapter:logical} представляет пример полезного предиката функционала \texttt{map} на Прологе.

Допустим, имеется некоторый предикат $a$, тогда подцели $q_{k,j}$ вычисляются слева направо для $j\ge 0$. Символьная среда $\sigma$  внутри тела предикатного определения обновляется после каждого вызова, каждого из подцелей  тела согласно рисунку \ref{fig:BoxModelPredicateCall}. Не присвоенные символы остаются, но могут присваиваться после подцели. Термы результаты ранних целей не нуждаются в обновлении, так как символам присвоено неопределённое символьное значение. Семантика  вызова предиката определяется следующим образом:

\begin{center}
\begin{tabular}{c}
$C(a)\llbracket a(\vec{y}):-q(\vec{x}_{k,n})_{k \times n} \rrbracket \sigma = D \llbracket q_{k,n} \rrbracket \sigma(\vec{x}_{k,n})$\\
$ \circ \cdots \circ D \llbracket q_{k,1} \rrbracket \sigma(\vec{x}_{k,1})$
\end{tabular}
\end{center}

По определению  вектор термов $\vec{y}$ может содержать общие элементы с вектором $\vec{x}_{k,n}$, $\forall k,n$ и $C\llbracket . \rrbracket$ имеет тип \texttt{atom} $\rightarrow$ \texttt{predicate} $\rightarrow \ \sigma \ \rightarrow \sigma$, $D\llbracket . \rrbracket$ имеет тип \texttt{subgoal} $\rightarrow \sigma \rightarrow \sigma$, и $\sigma$ имеет тип $\texttt{term}^\star$ $\rightarrow$ \texttt{term}, где $\star$ означает \textit{звезда Клини}  (см. \cite{davis94}).

Подцель $q_{k,j'}$ не обязательно должна определять  связанный граф изначально. Но, если это так, то это признак тому, что подмножество кучи определено целиком, по крайней мере, единая по интуиции структура данных. 
При разработке ПО  модульность и  «\textit{разделение забот}» можно всегда считать хорошим признаком.
Таким образом, абстрактный предикат вынуждает к описанию целой структуры данных. Следовательно, можно вывести следующий лозунг: «\textit{один абстрактный предикат должен корреспондировать с одной кучей}», где под подкучей подразумевается не пустой граф кучи, который содержит настоящее подмножество вершин и представляет собой кучу. Далее, добавляя все более $\star$-конъюнктов, соответствующий граф кучи растёт непрерывно.
 Набор $\star$-конъюнктов образует кучи, возможно связанные между собой, которые соответствуют абстрактным предикатам.
Когда речь идёт о «\textit{сжатии или раскрытии}»   абстрактных предикатов (похоже на вызов метода), существуют параметры, точнее вершины графа куч, которые стоят на обеих сторонах: со стороны вызова и со стороны вызванного предиката. Также могут существовать вершины, которые видны только изнутри предиката, которые не могут быть использованы извне предиката (по крайней мере, явным образом).
 
Без ограничения общности мы согласуем, что  доступ к объектным полям с помощью функтора «\textbf{.}» разрешается, например \texttt{a.b} (см. риcунок \ref{fig:PrologRulesEBNF}) или \texttt{oa(object5, fld123)} \cite{haberland14-1}.
Ради примера и модульности, мы согласуем далее, что объекты, а также объектные поля, передаются в качестве параметров предикатам. Отличие между объектами и простыми  стековыми локальными параметрами отсутствует, подробное объяснение будет дано позже.

По умолчанию согласуется, что для последовательности $q_{k,0}, q_{k,1}, \ldots , q_{k,m}$ из $q_{m \times n}$ любая строка находится в  нормализованном виде, так, что для $s\leq m$ нетривиальных элементов первые $s$ подцели располагаются, а остальные $m-s$ подцели являются  тавтологиями в качестве подцелей, чей  домен полностью определён как истина ($\top$).
Далее, согласуется, что: $$\exists k.a:-q_{k} \preceq a:-q_{k+1}$$ в силе, означая, что предикат, появляющийся ранее, в $\Gamma_a$ имеет выше приоритет, чем предикат, который определён позже.

\textbf{Замечание:} Очевидно, из-за проблемы приостановки, вызов предиката из раздела в общем  не решим. Далее рассматривается выразимость предикатов.

\textbf{Замечание:} Разрешение конфликтов с именами в $\Gamma$ может быть разрешено, если закодировать  локацию предиката в имя, как например класс, для которого предикат предусмотрен, тогда станет возможно различать предикаты. По определению предикаты с одинаковой локацией являются частью предикатного раздела, а следовательно  не конфликтуют.

Ради полноты синтаксического определения из рисунка \ref{fig:PrologRulesEBNF} и перевода представленного в следующем разделе, необходимо задуматься о том, как правильно перевести прологовские  операторы последовательности «\textbf{;}» и  отсечения «\textbf{!}».
Если тело предиката содержит «\textbf{;}», тогда всю последовательность после «\textbf{;}» необходимо перенести в новый предикат с той же левой стороной. Так, например:
 $$b :- a_0, a_1, ..., a_m; a_{m+1}, ... , a_n$$ 
для $\exists m.0\le m \le n$ разбивается на: 
 $$b :- a_0, a_1, ..., a_m. \quad \ b :- a_{m+1}, ... , a_n.$$
Если аналогично встречается оператор отсечения «\textbf{!}» в: $$b :- a_0, a_1, ..., a_m,!, a_{m+1}, ... , a_n$$, то $a_0$, $a_1, ...,  a_m$ может содержать альтернативы, которые будут рассматриваться в случае провала. «\textbf{!}» утверждает, что если только одна  подцель от $a_{m+1}$ до $a_n$ проваливается, то $b$ полностью проваливается без дальнейшего поиска  альтернатив. Все альтернативы могут быть  факторизированы cлева от «\textbf{!}» так, чтобы иные альтернативы исключались. В кратце, это является причиной, почему «\textbf{;}» и «\textbf{!}» могут быть исключены из Пролога в общности без потери выразимости (см. набл.\ref{obs:SimplificationByGeneralisation}). Вопрос обобщения предикатов, безусловно, интересен, но в целях задач поставленных в этой работе далее не рассматривается. Подход Поулсона \cite{paulson93} задаётся вопросом достижения обобщения с помощью введения функционалов на уровне абстракции и логических правил \cite{plaistead79}, \cite{menzies96}, \cite{giunchiglia89}, \cite{degtyarev01} для метода резолюции \cite{kalinina01}, \cite{bratchikov98}, \cite{gast08} --- которые здесь далее не рассматриваются. Модули тактик в системе «\textit{Coq}» \cite{bertot04} основаны на принципе, который также описывается Поулсоном.

Значение лем.\ref{lem:APsAreFOPs} и лем.\ref{lem:APsAreSOPs} таково, что можно выразить любые предикаты эквивалентности в Прологе беспрепятственно и без явной рекурсии. Мы не ограничиваем себя и допускаем  $\mu$-рекурсивные предикаты ценой частичной корректности в связи с не определением приостановки интерпретации предикатов ради беспрепятственной выразимости.

В случае «$\Rightarrow$» верхнего равенства $a(\vec{\alpha})$ предикат развёрнут. В случае «$\Leftarrow$» правая сторона определения предиката свёртывается в  вызов предиката. «$\approx$» означает   \textit{унификацию термов}.

\subsection*{Интерпретация предикатов над кучами}

Предложенный в этом разделе универсальный, но нестандартный, подход зависит от следующих этапов:

\begin{enumerate}
 \item Преобразование входной программы и аннотированные утверждения в прологовские термы, которые затем вписываются в логическую систему на основе Пролога (см. рисунок \ref{fig:HeapVerificationArchitecture}, \cite{haberland14-1}).
 
 \item Определение абстрактных предикатов в предусмотренной части прологовской программы вместе с представлением входной императивной программы. Правила принадлежат интерпретации и следовательно нуждаются в синтаксической корректности. Также как и пункт 1, этот пункт подробнее рассматривается в разделе \ref{sect:ArchitectureVerificationSystem}.
 
 \item Определение формальной грамматики для данных абстрактных предикатов. Обработка грамматики языковым процессором, которая при успехе генерирует конкретный синтаксический анализатор.
 
 \item Во время доказательства использование и подключение ранее преобразованных в синтаксический анализатор абстрактных предикатов при вычислении подцелей.
\end{enumerate}

Простое утверждение «$\mapsto$» становится   терминалом (см. опр.\ref{def:PredicateRuleSetDefinition}). Абстрактный предикат становится нетерминалом, точнее его вызовом. Терминалы могут связываться последовательно с помощью бинарного оператора  $\star$, который коммутативен (см. трм.\ref{theo:ReynoldsHeapProperties}). Терминал получает значение единицы некоторой подграмматики --- это эквивалентно грани некоторой данной куче. Утверждения «$\mapsto$» могут эффективно связываться, когда левые локации сортируются по  лексикографическому порядку. Когда происходит конфликт с одинаковыми именами, то $\alpha$-преобразование, учитывая местонахождения в данном модуле, может разрешиться, например, вводя  префикс.

\begin{rucorollary}{Контекст-свободность выражений куч}
 Раздел абстрактного предиката описывает систему правил, которая является контекст-свободной  формальной грамматикой.
\end{rucorollary}

\begin{proof} 
Левая сторона предиката не может по определению содержать более одного нетерминала, терминалы очевидно также не допускаются. Следовательно, только один нетерминал разрешается на левой стороне.
Отсутствие требования, где правая сторона должна быть строго  право-рекурсивной, т.е. отсутствие требования регулярной грамматики с правилами формой «$S\rightarrow aA$», приводит к контекст-свободности.
Допускается распознавание грамматик, которые эквивалентны к «\textit{скобочным грамматикам}» \cite{grune90} (КС-грамматики, чьи правила могут иметь вид $S\rightarrow ( S )$), а именно, когда имеется $n \in \mathbb{N}_0$ для некоторых терминалов $a$,$b$,$x_0$,$x_1$ и $x_2$, которые могут быть «$\mapsto$»-утверждением, так, чтобы $x_0 a^n x_1 b^n x_2$ и $x_0 \ne a$, $x_0 \ne x_1$ и $x_1 \ne b$, $b \ne x_2$. Если  голова предиката содержит аргументы, то это всё равно не меняет  статическую зависимость между определениями предикатов. Каждому разделу предиката можно приписывать один  начальный нетерминал.
\end{proof}

Из этого следует: выведенная и ожидаемая куча, обе они могут содержать свёрнутые предикатные определения, которые должны быть развёрнуты для окончательного определения равенства. Не трудно увидеть, что этот процесс двунаправленный, так как свёрнутые подкучи могут содержаться в обеих кучах. Важно заметить, что этого рода проблема редуцируется к  «\textit{проблеме корреспонденции Поста}», которая в общем случае теоретически не является  решимой, но только в частных случаях. Общая проблема для куч сформулированная Постом не рассматривается.\\\\
Набл.\ref{obs:FormalLanguageObservation} вместе с набл.\ref{obs:HeapReduction} может быть рассмотрено как предпосылка на формулировку: \textit{дана $\star$-связанная куча. Вопрос: совпадает ли она или нет с данной спецификацией кучи}? А также можно задать вопрос --- \textit{какая куча самая близкая к корректной куче, так, чтобы спецификация соблюдалась}? Решение вопроса способствует к решению проблемы  \textit{контр-примера}.

\subsection*{Перевод правил Хорна}

В этом разделе рассматривается перевод абстрактных предикатов, которые даны в качестве  прологовских правил, в обобщённые правила контекст-свободной грамматики с атрибутами.
Перед этим необходимо преобразовать  обыкновенные утверждения формой $loc \mapsto val$ в  токены, как принудительный лексико-аналитический шаг интерпретации Пролога. Применение  \textit{мульти-парадигмального программирования} \cite{denti05} разрешает интерпретировать прологовские правила во время запуска различными  языковыми процессорами. Это позволяет достичь максимальную расширяемость, например, подключение написанных процедур на различных языках программирования и запуск в третьей загрузочной системе.
Таким образом, процесс верификации может быть инициирован, контролирован и приостановлен входной программой. Процесс трансляции из правил Пролога в формальную грамматику удивительно прост, ради интерпретации правил Пролога. Однако, правила Пролога могут иметь аргументы с обеих сторон от знака определения правила «\textbf{:-}». Параметры и аргументы генерируемой грамматики можно моделировать как  атрибуты  формальной грамматики. Следовательно,   процесс трансляции $C\llbracket \rrbracket$ можно характеризовать как:

В отличие от ранее введённых  нотаций далее вводятся подцели и теперь входной вектор $\vec{x}$ содержит все переменные символы, ради удобной записи внутри каждого предиката. Если некоторая подцель $q_j$ для $j\ge 0$ не нуждается во всех компонентах $\vec{x}$, то подцель не нуждается в них.
$\dot{\cup}$ является множественным объединением, с учётом последовательности и сохранения дубликатов. Нетрудно заметить, что обе записи сопоставимы и сильно похожи друг на друга. Абстрактный предикат описывает кучу. Таким образом, $C\llbracket \rrbracket$ преобразует кучу, т.е. является интерпретацией ассоциативной кучи, см. лем.\ref{lem:HeapAsWord}.
 $C^{-1}\llbracket . \rrbracket$ преобразует атрибутируемую грамматику обратно в  Пролог:

Ещё осталось расследовать $C \llbracket \rrbracket$ и $C^{-1} \llbracket \rrbracket$ касательно корректности и полноты.

\begin{rucorollary}{Корректность и полнота преобразований}
$C \llbracket \rrbracket$ и $C^{-1} \llbracket \rrbracket$ полны и корректны.
\end{rucorollary}

\begin{proof}
 Не трудно убедиться в том, что $C \circ C^{-1} \circ C \equiv C$ и $C^{-1} \circ C \circ C^{-1} \equiv C^{-1}$, сопоставляя верные определения. Обсуждения из раздела \ref{sect:PointsToHeaplets}, ни «\textbf{!}» ни «\textbf{;}» не влияют на выразимость. 
Если для любого элемента из $C\llbracket \rrbracket$ преобразование не приостанавливается, то   кообласть также не полностью определена. Тоже самое распространяется на $C^{-1}\llbracket \rrbracket$.
\end{proof}

Нетрудно заметить, что прологовские правила не единственная форма, но очень близка к формальной грамматике. Более обобщённо можно включить императивные языки программирования с процедурами на основе автоматического запуска со стеком.

\subsection*{Синтаксический перебор как верификация куч}

В целях простого и интуитивно понятного алгоритма, константы из опр.\ref{def:HeapAssertionDefinition} далее пока рассматриваются. Касательно классных объектов, предикат \textbf{true} может означать, например, принимать все «$\mapsto$»-утверждения до отметки, которую необходимо согласовать, в зависимости от данного правила. В данной отметке, которая представляет собой  \textit{безопасный пункт синхронизации} в качестве  «\textit{правил генерации ошибок}»  \cite{grune90}, может продолжаться  синтаксический перебор, если перебор застрянет из-за ошибочной последовательности в  потоке токенов, либо из-за не ожидаемого состояния на стеке при переборе. Предполагается, что входное слово, представляющее кучу, всегда конечное. Предположение основывается по той причине, что память --- это линейно адресуемое конечное пространство. Чисто гипотетично, число совершённых развёртываний и свёртываний стремится к бесконечности. Позже мы покажем, что для $\exists j$ функции $\pi_j$ и $\sigma_j$ ограничиваются полиномиальной сложностью.

Этот раздел предлагает и обсуждает основные конвенции, необходимые для реализации общего подхода синтаксического перебора, в целях верификации куч. Это рассматривается на примере  LL(k)-перебора. LL(k)-анализатор, является синтаксическим анализатором, который может смотреть вперёд любое количество токенов для разрешения многозначности одинаковых правил. LL(k)-анализатор выбирается образцово. Кроме LL(k)-анализатора, могут быть использованы другие анализаторы синтаксиса, как например,  LALR,  SLR,  Эрли и другие. Журдан \cite{jourdan12} с помощью «\textit{Coq}» доказывает корректность LR(1)-анализатора --- вопрос корректности анализатора, безусловно, важен, но в рамках этой работы отпадает, т.к. анализатор конструируется автоматически из данной формальной грамматики.
Далее, первым определяется формальное предложение, как композиция отдельных «$\mapsto$»-утверждений и подцелей как нетерминалы. Вторым и третьим, в аналогии к LL(k)-анализатору, где отдельные терминалы представляются как обыкновенные утверждения операции, вводятся «\textit{first}» и «\textit{follow}». Четвёртым, обе операции \textit{сдвиг}  (SHIFT) и  \textit{свёртка}  (REDUCE) представляются, чтобы иметь представление для более обобщённых анализаторов.

Например, $\alpha::=$ \texttt{[ pointsto(x,nil),} \texttt{pointsto(y,1), member(x,[y])]} описывает актуальное состояние кучи при верификации данной императивной программы. $\star$ заменяется в предыдущем списке запятой. Спецификация правил может зависеть от\\
\texttt{[pointsto(Y,1),member(X,[Y|_]),pointsto(X,_)]}.\\
Поэтому, для проверки абстрактного предложения (спецификация) для данной программы (генерируемая куча), необходимо сравнить, выводима ли одна из сторон из другой или нет.

Абстрактное предложение может также содержать  унификацию термов, как например,\\
\texttt{pointsto(X,5),X=Y}. Унификацию термов необходимо тщательно проверять и отделять от: «$\mapsto$»-утверждений и от вызовов предикатов, т.е. нетерминалов. Надо учесть, что  неограниченная рекурсия термов должна ограничиваться так, чтобы терм сам себя не содержал по определению. Поэтому,  \textit{самосодержащие термы} необходимо  ограничить, так как они часто имеются по определению в Прологе (см. раздел \ref{sect:PrologAsReasoningSystem}), где проверка самосодержимости по умолчанию отсутствует. Анализаторы, которые вынуждены анализировать неограниченные термы, могут попадать в тупиковую ситуацию, если сам терм циклически определяется самим собой. Поэтому, во избежание проблемы принудительной работы, уговаривается, что проверка на присутствие циклов проводится отдельно от верификации, либо по умолчанию исключается.

Рассмотрим теперь «\textit{наивный}» подход для сравнения равенства двух абстрактных предложений. Рассмотрим алгоритм № \ref{AlgorithmEqualityCheck}. $\pi$ означает функцию начальных терминалов для данного правила и данного положения, как было определено ранее. Проблема в этом подходе  (развёртывая фактически принудительно предикаты всё далее и далее, возможно бесконечно) заключается в неопределённости, когда и сколько раз применить точно развёртку и  свёртку, тем более не известно в случае нахождения одного решения, не является ли решение оптимальным, что, безусловно, зависит от данных правил. Предположим, имеется: 

$$\alpha_1 = [\underbrace{a \mapsto b}_{i_0},i_1,\cdots, i_{m_1}, \underbrace{q_1(x)}_{p_0}],$$ 
$$\alpha_2 = [\cdots, \underbrace{a \mapsto b}_{j_3},\cdots, \underbrace{q_1(x)}_{q_7}]$$

необходимо сравнить сходимость термов обоих выражений. Сдвиг термов (SHIFT-TERMs) приводит к унификации $i_0$ и $j_3$ и продолжению сравнения остальных термов.
Сначала свёртка (REDUCE-PREDs) проверит, подходящие и применяемые ли предикаты для расширения.
Поэтому, первый терминал предиката может быть запрошен первым. Развёртывание $expand(p_k,\alpha_1)$ сопоставит подцель телом определения предиката $p_k$ (см. опр.\ref{def:PredicateFolding} и рисунок \ref{fig:PrologRulesEBNF}) и преобразует новое абстрактное предложение $\alpha'$, которое можно описать в Прологе как  \texttt{concat($\alpha$,$[i_7,i_8,i_9],\alpha'$)}, если $q_1(x)$ развёртывается в список $[i_7,i_8,i_9]$.

Независимо от конкретных аргументов, $\pi$ определяет все те терминалы, которые являются «$\mapsto$»-утверждениями, либо являются первым терминалом предикатных  подцелей.
Мы подразумеваем, что подцели унификации и вызовы встроенных (т.е. «\textit{встроенных}») подцелей фильтруются и не рассматриваются при вычислении $\pi$ и $\sigma$.

Неформальное объяснение таково: множество последующих  терминалов определяет все те терминалы, которые могут следовать данному актуальному «$\mapsto$» утверждению или данной подцели в случае, когда терминал находится в конце правила. Согласно \cite{grune90} мы теперь обладаем возможностью определить LL(k)-распознаватель с помощью определений $\pi$ и $\sigma$.

\begin{ruexample}{Многозначимость правил}
Даны следующие правила нетерминалов  $ q_1 \rightarrow a, \  q_2 \rightarrow a q_2 \ | \ q_3 b, \ q_3 \rightarrow \varepsilon \ | \ q_3 a$. Нетрудно заметить, что эти правила  многозначны, например из-за $\pi(q_2)=\{a\}, \sigma(a)=\{ \varepsilon\} \cup \pi(q_2) \cup \pi(q_3) \cup \sigma(q_3)$.
\end{ruexample}

\begin{ruexample}{Корректность}
Дана следующая конечная спецификация: 
$$[ (loc1,v1), p1(loc1,loc2), (loc2,v2) ]$$ и подходящая к нему условная цепочка $[ (loc1,v1), (loc2,v2)  ]$. Цепочка является корректным словом данной спецификации, лишь в том случае, когда $p1$ генерирует только  пустую кучу.
\end{ruexample}

В случае, если данное слово не является просто последовательностью терминалов, но окажется  «\textit{абстрагируемым}» предложением, т.е. содержит произвольную последовательность терминалов и нетерминалов, то проблема сравнения может возникнуть в различных местах последовательности. Одной из этих проблем может оказаться вычисление повторяющихся куч и абстрактных частиц предикатов. Поэтому выгодно, когда вызовы абстрактных предикатов запоминаются в  кэш  (мемоизатор, например в \cite{warren99}), а затем подцели могут быть сравнены гораздо быстрее, тем более, из-за декларативного характера предикатов, где символьные значения не присваиваются заново. Поэтому, предикаты могут быть сравнены чисто процедурно, без «\textit{опасных}»  побочных эффектов и изменения состояния вычисления несколько раз подряд. При ускорении кэширования необходимо запоминать наименование предиката и все термовые аргументы при вызове подцели.

Отрицания предикатов далее не рассматриваются из-за подробностей представленных в разделах \ref{sect:ConjunctionAndDisjunction},\ref{sect:PartialSpec},\ref{sect:StricterDiscussions},\ref{sect:PrologAsReasoningSystem} и \ref{sect:LogicalReasoningAsProof}.

\subsection*{Свойства}

На рисунке \ref{fig:heapConfiguration} показан пример конфигураций одной кучи, в которой любая вершина $v_j$ при $j \ne 2, j \ne 5$ с отходящими гранями для классного экземпляра, более одного указателя в качестве атрибута. 
Конфигурация состоит из  8 треугольников, каждый из которых ограничивается сплошной, пунктирной или извилистой линией. Линия, начиная с середины между $v_0$ и $v_3$ доходящая до $M_1$, выделяет треугольник $\Delta(v_0,M_1,v_1)$. Таким образом, рисунок \ref{fig:heapConfiguration} демонстрирует разновидность одной и той же структуры в динамической памяти, используя различные описания. Предполагается, что два и более выходящих указателей подразумевает соответствующее количество различных полей в объединённом экземпляре (см. главу \ref{chapter:expression} и рисунок\ref{fig:GraphIsomorphisms}).  То есть, вершины могут охватывать те же самые графы различными  абстрактными предикатами, например, треугольниками с пунктирными линиями или извилистыми линиями. Однако, при различных представлениях, необходимо учитывать все вершины --- это является обязательным критерием, но не достаточным. Вопрос равенства двух различных куч, эквивалентен к вопросу  \textit{изоморфизма} двух данных направленных графов.

Когда анализируется входной поток токенов, последовательное состояние анализатора можно определить с помощью проверки нескольких токенов вперёд. В худшем случае, придётся проверить все входные токены, но решение будет принято.
Это в случае, когда куча описывает  объектный экземпляр и необходимо определить границы объекта. Если ограничиться тем, что один объект может быть описан только в одном предикате, либо поля указатели описываются зависимыми предикатами, то такая конвенция позволит избежать описанную проблему и провести генерическую канонизацию по предикатам. Синтаксический перебор, соблюдая данную конвенцию, имеет полиномиальную сложность \cite{grune90}. Проводится  синтаксический перебор без дополнительных условий. Если объектные границы не соблюдаются, то перебор может повлечь за собой пересечение куч, которые могут относиться к другому объекту. На практике это может означать, \textit{частичный и параллельный перебор} различных куч одновременно, что желательно избежать по различным соображениям. Во-первых, проблема плохо делима из-за множества взаимосвязей, определений и из-за иерархии вызовов и фреймов. Во-вторых, надо задаться вопросом о преимуществе такого подхода, который кажется довольно сомнительным на данный момент потому, что не предвидится никакого ощутимого выигрыша, но дополнительные затраты на параллелизацию поглощают все возможные преимущества. Поэтому, такой вопрос с теоретических и практических сторон пока не задаётся. С практической точки зрения необходимы параллельные кучи, которые наполнялись бы последовательно обработкой спецификации (см. \textit{стек Трибера} \cite{bornat00}).

Если все разделы предикатов можно перебрать синтаксически, учитывая упомянутую конвенцию, то перебор «$\mapsto$»-утверждений решим и завершается с ответом или с отказом. При отказе, синтаксическая ошибка содержит неверный токен и/или данный сегмент в нетерминале (т.е. состояние перебора). Таким образом, она показывает на состояние кучи, которая не совпадает с имеющейся кучей.
Также необходимо заметить, что сходимость вычисления соблюдается при сборе LL-/LR-анализаторов потому, что состояние вычисления записывается как кортеж (состояние анализа, входная цепочка) и для каждого состояния переход детерминирован.
В итоге, главная модель памяти не была кардинально изменена, а лишь расширена абстрактными предикатами.

\subsection*{Реализация}

Система реализована на основе  «\textit{GNU Prolog}» \cite{diaz12} при поддержке «\textit{ANTLR}» версии 4 \cite{parr12}. Работа с абстрактными предикатами предусмотрена для работы \cite{haberland14-1}, \cite{haberland15-2}. Изначально, обе работы, основанные на Прологе, были выбраны ради простоты пользования и в преподавательских целях. Для максимальной поддержки пакетов, написанные на Прологе, была использована дистрибуция «\textit{GNU Prolog}», для максимального сходства с генеричной версией Пролога. Для расширения и возможности поддерживать различные библиотеки на различных языках, была использована  \textit{многоцелевая парадигма} \cite{denti05}, для динамического подключения загрузочного кода на разных языках программирования. Таким образом, программная библиотека Денти \cite{denti05} разрешает подключать, например, Ява-библиотеку через соответствующую графическую оболочку. Процедуры на языке Ява или иных языках через соответствующий интерфейс при запуске прологовского запроса могут включить дальние вызовы. Слой вызовов «\textit{tuProlog}» выступает в качестве  «\textit{медиатора}» и  «\textit{адаптера}». При этом, вызов и результаты могут передаваться с обеих сторон с помощью  «\textit{прокси}», как будто процедура дальнего вызова библиотеки является локальным правилом Прологу. Таким образом, не только на языке Ява можно создать всё новые вспомогательные и встроенные предикаты, но также ново\-со\-зданные предикаты могут снова вызывать предикаты. Так вызов при запуске Пролог программы согласно принципу ящика из рисунка \ref{fig:BoxModelPredicateCall}, производится полностью динамически.

Реализация преобразует входную программу в  IR, которое состоит из термов Пролога. Затем, утверждения копируются отдельно в теорию Пролога, а абстрактные предикаты преобразуются в грамматику «\textit{ANTLR}», как это было изложено ранее. Затем, запрашивается перебор, синтаксический анализ с помощью того  распознавателя подключается, который генерируется после определения «\textit{ANTLR}» грамматики (часто на языке Ява). Затем, выход, контроль возвращается вызывающей стороне. При необходимости, абстрактные предикаты могут проверяться автоматически. В случае синтаксической ошибки, выдаётся соответствующая ошибка.

%%%%%%%%
 «\textit{ANTLR}» использует для разрешения синтаксической многозначности две технологии: «\textit{синтаксические предикаты}» и «\textit{семантические предикаты}» \cite{parr12}. Термины очень похожи на проблемы данной работы, но их следует внимательно различать и не путать. Кроме этих двух технологий, всегда имеется возможность переписать правила, так, чтобы многозначность не возникала, например, с помощью \textit{факторизации правил}. На практике  «\textit{ANTLR}», увы, не покрывает полный класс LL(k)-распознавателей. Предикаты ограничены и в случае конфликта часто нуждаются в \textit{переписке (термовых) правил}. «\textit{ANTLR}» часто не в состоянии различать самостоятельно, особенно регулярные, сложные выражения. Следовательно, приводит к полной генерации кода для необходимого распознавателя. Такая же проблема с ограниченностью LL(k) наблюдается во множестве других генераторов компиляции, как например, в  «\textit{yacc}» или  «\textit{bison}» \cite{levine09}. С практической точки зрения это означает, что данная  грамматика, хотя формально и корректна, не может быть перебрана из-за ограничений в реализациях распознавателей. К счастью, на практике переписка грамматики часто разрешает этот конфликт в практических целях. Также необходимо заметить, что более мощные распознаватели, по принципу, «\textit{снизу-вверх}» имеют возможность избежать различные ограничения за счёт сложных реализаций и объёма, генерируемых распознавателем, например,  «\textit{bison}»
работающий по принципу  «\textit{сдвиг-свёртка}». Обзор технологий синтаксического анализа можно найти в \cite{opaleva05}, \cite{grune90}.

В качестве примера, приведём необходимые трансформации на более подробном уровне для обработки лексемов и токенов. Сначала, $bar \mapsto foo$ преобразуется в 
$pt\_3bar\_3foo$, где число «$3$» это длина наименований, либо сложное выражение целой локации, которая требуется для различения последовательных наименований. Если локация сложная, например $b.f.g$, то локация вместе с длиной определяют выражение утверждения в  смежной записи (с англ. «\textit{mangled}» \cite{levine99}). Например, \texttt{pointsto(X,2)} преобразуется в прологовский  атом $p\_X\_2$. Цель преобразований заключается в получении одного атомного символа, который представляет собой простое утверждение о куче, это терминал. При необходимости терминал без потерь может быть полностью восстановлен обратно в утверждение о куче, т.к. для этого все данные имеются и данные компоненты строго различаются в прологовском атоме. Процессы прямого и обратного преобразования могут быть реализованы в качестве встроенного предиката, например на языке  Ява. Затем, используются в главных частях верификации, либо могут быть использованы другими встроенными предикатами на Яве, либо на другом подключенном языке к Прологу (см. предпосылки, \cite{denti05}).

Далее, левая сторона  (де-)канонизации (на Прологе)

$$\texttt{p1(X,[X|Y]):-... .}$$

преобразуется в:

$$\texttt{p1(X1,X2):-X1=X,X2=[X|Y],... .}.$$

Правило «\texttt{p(X,Y):-$\alpha$.}» из Пролога можно перевести в следующую «\textit{ANTLR}»-грамматику:

\begin{center}
\begin{tabular}{c}
\texttt{p[String X,String Y]:} $\alpha$.
\end{tabular}
\end{center}

Таким образом, все  \textit{синтезируемые атрибуты} можно передавать сверху-вниз. \textit{Порождаемые атрибуты} можно приписывать к соответствующему предикату \texttt{p}, добавляя в «\textit{ANTLR}»  ключевое слово \texttt{returns} вместе с наименованием атрибута перед двоеточием. Поэтому, необходимо определить, синтезируемы или порождаемы ли они и затем все атрибуты перевести и вписывать в соответствующее «\textit{ANTLR}»-правило. Правила, конфликты и ограничения Пролога очень похожи на правила «\textit{ANTLR}».

Когда абстрактные предикаты преобразуются в конкретную грамматику, например  «\textit{ANTLR}»-грамматику, то возникает проблема. Унифицированные термы,  «$\mapsto$»-ут\-вер\-жде\-ни\-я (терминалы) и нетерминалы следует преобразовать согласно данному синтаксису в грамматику вместе с аннотациями. Также «\textit{ANTLR}» разрешает  \textit{транслирующие правила}, которые определяют семантику самой программы и записываются в «\textit{ANTLR}» в фигурные скобки. Отрицание предложений и их под-предложений можно сформулировать в «\textit{ANTLR}»  с помощью «$\sim$» и скобок, которые означают, что данное  регулярное выражение из «$\mapsto$»-терминалов должно или не должно следовать. Далее, элементы перевода грамматик не обсуждаются, т.к. с помощью атрибутов и транслирующих правил, можно имитировать все остальные выражения и предложения, в том числе и несовпадение предикатов (см. \cite{opaleva05}, \cite{lavrov01}, \cite{grune90}, \cite{parr12}, \cite{gcc15}, \cite{mitchell96}).

\subsubsection*{Графическая оболочка}

Графическая оболочка (см. рисунок \ref{ScreenshotGUI}) состоит из трёх важных частей: (i) из левой части, которая содержит программу Си-диалекта, (ii) из правой верхней части, которая при выполнении программы содержит актуальное содержание динамической памяти и (iii) из правой нижней части содержавшая программное представление в качестве термов Пролога. Содержание памяти и программное представление можно преобразовать в графическое представление. Кроме того, внизу имеются окна для ошибок и предупреждений, а также запись текущих операций верификации.

Разработка графической оболочки оказалась наиболее полезной в использовании, т.к. для автоматизации преобразования, отслеживания доказательств и представлений динамической памяти и правил верификации, вместе с программой, легче отслеживать, когда рядом имеются все необходимые данные.

\subsubsection*{Use Cases}

Программист в данной оболочке задаёт и редактирует программу (см. рисунок \ref{UMLUseCaseProgrammer}). Для этого имеются синтаксические проверки, а также программное представление может отобразить в графическом представлении в формате ``DOT'' в виде дерева представить.

Роль персоны, которая специфицирует программу (см. рисунок \ref{UMLUseCaseSpecifier}), отличается от программиста --- формальными, точнее логическими, формулами описывается поведение программы. Спецификации  можно в систему заносить, редактировать до и после загрузки программы. При запуске данные спецификации, которые размещены вместе с программой, проверяются. В случае возникновения ошибки, ошибка выделяется в правой нижней части. Проверка спецификации включает в себя при- и постусловия процедур, абстрактные предикаты, классные инварианты, но может также включать проверку полноты набора правил одного правила Пролога, касательно входной кучи.

Верификация кучи (см. рисунок \ref{UMLUseCaseVerifier}) в отличие от редактирования программы заключается в вычислении и проверке имеющийся при выполнении программы состояния памяти с спецификацией. Для верификации необходимо проводить проверку сходимости куч, а также преобразование в нормализованную форму. Для упрощения и объяснения текущего доказательства имеются вспомогательные элементы, прежде всего, это визуализация дерева доказательства, генерация контр-примера, а также использование предыдущих доказательств для более быстрой сходимости.

\textbf{Пример) Реверс линейного списка}

Дан список \{1,2,3\}, имеется указатель $y$ на последний элемент линейного списка. Тогда реверс списка можно описать рекурсивно, отделив последний элемент и присоединив его к началу отставшего списка. Таким образом, получается список как указано на рисунке \ref{HeapExampleAP1}.

Далее, получаем список как указано на рисунке \ref{HeapExampleAP2}.

%%%%%%%%%%%%%%%
\textbf{Классовые определения}
Классовые экземпляры относятся в основном к классам Си++ или Ява и представляют упрощенный образ. Классы имеют разовый идентификатор, поля и методы. Поля и методы определяются в любом порядке и являются разовыми. Видимость для представления универсальности модели необязательна, но разрешается. Встроенные методы запрещаются, доступ к собственному классному экземпляру осуществляется с помощью \texttt{this}.

\begin{grammar}
<class_definition> ::=  [ <class_modifier> ] 'class' <ID> '\{' <class_fields_methods> '\}'

<class_fields_methods> ::= \{ <class_field> | <class_method> \}*

<class_field> ::= [ <class_modifier> ] <variable_declaration> ';'

<class_method> ::=
  [ <class_modifier> ] <function_definition>
  \alt [ <class_modifier> ] <function_declaration> ';'

<class_modifier> ::= ( 'private' | 'public' | 'protected' ) ':'
\end{grammar}

Декларация типов следует в основном декларации по стандарту ISO-C++, которая разрешает неинициализированные и инициированные переменные.

\begin{grammar}
<variable_declaration> ::=
  <type> <ID> \{ ',' <ID> \}*
  \alt <type> <ID> \{ ',' <ID> \}* '=' <expression>
\end{grammar}

Определение функции осуществляется простым образом, т.е. неограниченные синтаксисом (variadic) параметры не допускаются. Конвенция вызовов подпроцедур следует схеме вызовов в языке Си. Тип возврата метода является либо \texttt{void}, либо базисным типом в данный момент. В ближайшее будущее объектные типы будут возвращаться, а до этого все изменения состояния вычисления передаются только с помощью параметров.

\begin{grammar}
<function_definition> ::= <function_declaration> <block>

<function_declaration> ::=
  <type_or_void> <ID> '(' ( <formal_parameters> | ) ')'

<type> ::= 'int'

<type_or_void> ::= <type> | 'void'
\end{grammar}

\textbf{Встроенные правила для автоматизированного логического вывода}

%%%%%%%%%%%
\textbf{Нормализация куч} осуществляется с помощью встроенного предиката Пролога\\
\texttt{simplify/2}. Этот предикат принимает кучу в виде терма и возвращает терм кучи, в котором $||$ оформляется крайней снаружи. Предполагается, для быстрой обработки, что входная куча уже $||$-нормализована, т.е. в виде $q_0 || q_1 || ... || q_k$, где $\forall j.q_i$ в виде $X_i\circ X_{i+1}\circ ... \circ X_{k}$.

\textbf{Высчитывание} (``\textit{множественное сравнение}'') осуществляется с помощью \texttt{sub\-stract}. Этот встроенный предикат принимает список термов Пролога и генерирует список не хватающих термов кучи. Предикат может быть использован для установления полноты данного семейства предиката (например, важно для упрощения, трансформации и для обработки с тройками Хора в общем).

\textbf{Сравнение} производится покомпонентно с помощью двух термов: имеющейся и ожидаемой кучей.

%%%%%%%%%%%

\textbf{Слойная архитектура верификатора}

Предложенная архитектура верификатора ``ProLogika'' \cite{haberland19-1} состоит из 6 главных слоёв (см. рисунок \ref{LayeredArchitectureVerifier}), которые на периферии сотрудничают с библиотекой логического вывода ``\textit{tu\-Pro\-log}'' (2p), а также с генератором синтаксических анализаторов ANTLR в версии 4. Система реализована на языке Ява, хотя одновременно она тесно связана с реализацией Пролога и максимально близка к GNU Prolog. ANTLR может генерировать выходной код не только Ява, но также разных других языков программирования. Таким образом, обеспечены расширяемость и вариабельность системы ``\textit{ProLogika}''. ``\textit{ANTLR}'' принципиально может быть заменено на любой другой синтаксический анализатор эквивалентной мощностью, под условием, что не только стартовый нетерминал может быть распознан, но также подграмматики используемых нетерминалов.%\\%[-2cm]

Система верификатора динамической памяти ``ProLogika'' состоит из следующих пакетов: \texttt{internal}, \texttt{output}, \texttt{parsers}, \texttt{prolog}, \texttt{frontend} и \texttt{gui}. \texttt{gui} предоставляет графическую оболочку и управленческие к ней модули. Результаты, промежуточные представления вычисления и визуализации выявляются в самой оболочке или могут быть запущены с ней.

Пакеты \texttt{frontend} и \texttt{parsers} тесно взаимосвязаны. Пакет \texttt{frontend} содержит интерфейсы и модули различного уровня абстракции для лексического и синтаксических анализаторов. Параметризация синтаксического анализа является одной из основных верификации абстрактных предикатов.
\texttt{parsers} cодержит изначально генерированные лексические и синтаксические анализаторы, а также анализаторы, генерируемые при запуске верификации. Интерпретация абстрактных предикатов входит в пакет \texttt{ProLogika.inter\-nals.aps}.

Пакет \texttt{prolog} обогащает встроенную систему логического вывода с дополнительными предикатами Пролога, например, для обработки термов куч, анализ и выявление ошибки, а также для реализации дальнейших побочных эффектов. Реализация дополнительных встроенных предикатов осуществляется таким дизайном, который гибок и легко может быть расширен. Правила Пролога могут быть всегда добавлены и расширены во время запуска программы верификации.

Пакет \texttt{output} предоставляет вспомогательные функции для писания термовых представлений, например, в DOT-файл или на консоль. На данный момент предусмотренная генерация ОCL-кода не реализована.

Пакет \texttt{internal} производит саму спецификацию и верификацию связанных функций. Пакет \texttt{core} содержит различные модули, они связаны с обработкой термовых представлений куч. Пакет \texttt{aps} содержит все необходимые модули, которые связаны с абстрактными предикатами, следовательно, он связан с пакетом \texttt{core}.
\texttt{checkers} производит проверку куч до, после и во время запуска программного оператора. Остальные пакеты подлежат сильному исправлению и находятся в настоящее время на стадии разработки.
%%%%%%%%%%%%%%%%%%

\textbf{Компоненты}

В соответствии с разделом \ref{sect:Layers} можно выявить 6 главных компонентов верификатора (см. рисунок \ref{VerifierComponents}):
(i) \texttt{ProLogika.GUI} отвечает за графическую оболочку,
(ii) \texttt{Pro\-Log\-ika.par\-sers} отвечает за всё, что связано с анализатором синтаксиса, включая проверки условий синтаксиса и семантики, а также за обработку абстрактных предикатов (пакеты \texttt{parsers}: \texttt{incoming}, \texttt{intermediate} и \texttt{AbstractPredicates}),
(iii) библиотека Java \texttt{SWT} тесно связана с компонетом \texttt{ProLogika.GUI} и введён для поддержки платформы независящей от ОС,
(iv) \texttt{ANTLR} компонент для построения языковых процессоров,
(v) \texttt{alice.tuProlog} представляет библиотечный компонент для реализации генеричного логического вывода, построен на основе Пролога, 
(vi) \texttt{ProLogika.internal} компонент представляющий Прологовские термы (\texttt{PTerm}) и правила (\texttt{PRule},\texttt{PSubgoal}), которые используются другими компонентами.

\textbf{Пакеты}
%%%%%%%%%%%%%%%%

Пакет \texttt{alice.tuProlog} (см. рисунок \ref{PackageAlice2p}) содержит: (i) класс \texttt{Term}, который представляет собой общий Прологовский терм, который не является абстрактным, (ii) \texttt{Library} является Прологовской библиотекой, которая может содержать целую формальную теорию Прологовских правил. \texttt{BuiltinLibrary} является реализацией класса \texttt{Library}. \texttt{Engine} контейнер, который может производить логический вывод правил Хорна согласно данной теории. \texttt{ResultInfo} представляет совокупность любого количества результатов данному запросу (подцелей).\\[0.7cm]

 %%%%%%%%%%%%%%%%%%
\texttt{Пакет parsers:}
 
Пакет \texttt{parsers} (см. рисунок \ref{PackageParsers}) содержит пакеты: \texttt{incoming}, \texttt{intermediate}, \texttt{pro\-log}, а также два класса \texttt{MetaType} и \texttt{MetaTypeManager}. \texttt{MetaType} представляет определение типа во входной программе, например, класс или целое число (integer). \texttt{Me\-ta\-Type\-Ma\-na\-ger} является наблюдающей за определениями инстанций.

\texttt{Пакет internal:}
Пакет содержит пакеты: \texttt{core}, \texttt{checkers}, \texttt{aps}, \texttt{spec}.\\\\

\texttt{Пакет internal.core:}

Пакет \texttt{core} (см. рисунок \ref{PackageCore}) содержит классы, которые представляют ядро модели логического вывода: термы, правила, символьные переменные, подцели, и т.д. Префикс -\texttt{P} ссылается на ``\textit{Prolog}'', классы строго соответствуют элементам ISO-Пролога.

Главными базисными элементами являются: \texttt{PTerm}, \texttt{PRule} и \texttt{PSubgoal}. Остальные P-классы соответствуют Прологу, ради исключения классу \texttt{Unparsed}, который частично вычисляемый терм, который полностью определён на более поздней стадии вычисления.

\texttt{Пакет internal.checkers:}

Проверки проводятся вместе с нормализацией и высчитыванием правил во время верификации, например во время анализа данного абстрактного предиката (\texttt{APCom\-plete\-ness\-Checker}, \texttt{Heap\-Com\-plete\-ness\-Chec\-ker}, \texttt{Pro\-log\-APs\-Chec\-ker}) (см. рисунок \ref{PackageInternalCheckers}). Синтаксические анализы также проводятся во время анализа входной программы и проверки вставленных в программу спецификатором утвреждения (например с помощью \texttt{ANTLR\-Gram\-mar\-Checker}). Все проверки могут осуществляться, если реализовать интерфейс \texttt{Checker\-In\-ter\-face} конкретными методами.

\texttt{Пакет internal.aps:}

Пакет \texttt{aps} (см. рисунок \ref{PackageInternalAPs}) содержит \texttt{ANTLRBuil\-der}, \texttt{ANTLRLauncher}, \texttt{Grammar\-Buil\-der}, а также все остальные классы, которые обрабатывают абстрактные предикаты (APs) согласно формальной атрибутируемой грамматике (как было введено Кнутом и Вегнером). \texttt{ANTLBuilder} является специализацией класса \texttt{GrammarBuilder}. ``\textit{ANTLR}'' используется изначально лишь как один пример генератора синтаксического анализа. Любой генератор или вручную написанный анализатор может послужить заменой, если предоставить необходимый интерфейс.

Аналогично выше сказанному, \texttt{ANTTLRGrammar} является специализацией класса \texttt{Formal\-Grammar}. \texttt{ANTLRBuilder} способствует к построению анализатора анализаторов.

\texttt{ANTLRLauncher} производит вызов драйвера ANTLR-генератора (для произведения различных функций ANT\-LR).

%%%%%%%%%%%%%%%%%%%%%%%
\subsubsection*{Синтаксические анализаторы}
Класс \texttt{Main} (см. рисунок \ref{PackagesMain}, связан с \texttt{Parser} из рисунка \ref{PackagesParser} и \texttt{SyntaxAnalyzer} из рисунка \ref{PackagesSyntaxAnalyzer}) управляет тремя инстанциями синтаксического анализа согласно формальным грамматикам: (1) ClassBased.g4, (2) IRProlog.g4, и (3) AbstractPredicates.g4.
 
Замечения:
\texttt{IRPrologCommand.execute()} возвращает \texttt{IRPrologDescription object}.

Имеется паттерн:

\begin{verbatim}
ADAPTER::(CLASS) ADAPTER:
  SyntaxAnalyzerCommand

ADAPTER::ADAPTEE:
  ClassBasedCommand, IRPrologCommand,
    AbstractPredicatesCommand
\end{verbatim}

\texttt{Package parsers.incoming:}

Следующая структура применима не только для \texttt{parsers.incoming} (см. рисунок \ref{PackagesParsersIncoming}), а также для 
\texttt{parsers.intermediate} и \texttt{parsers.prolog}.

Чтобы модифицировать синтаксический анализ при входе и выходе из правила нетерминала (\textit{срёртка}) были введены классы с суффиксом \texttt{Listener}.

%%%%%%%%%%%%%%%%%%%%%

\texttt{Пакет parsers:} см. рисунок \ref{PackagesParsers}.

Класс \texttt{MetaType} представляет встроенный или само-определяемый тип, который применяется при проверке типов над выражениями для входной программы. Каждый тип разово идентифицируется с помощью поля \texttt{typeID}. Однако, при построении типов \texttt{typeID} не присваивается изначально (только имя, но без полной сигнатуры, т.к. рекурсивные типы также разрешаются). Статически определёнными, встроенными типами являются: \texttt{INTTYPE} и \texttt{VOIDTYPE}.

Элементы пакета \texttt{parsers.intermediate} выполняют паттерн-роль ``\textit{директора}'' для класса \texttt{Me\-ta\-Type\-Ma\-na\-ger}.\\[0.7cm]

\texttt{Пакет frontend} см. рисунок \ref{PackageFrontend2}.

Замечания:

\begin{verbatim}
STRATEGY::TEMPLATE:
  SyntaxAnalyzer.launchMainNT()

STRATEGY::CONCRETE STRATEGY:
  ClassBasedDescription.generateAST()
  IRPrologDescription.generateAST()
  AbstractPredicates.generateAST() 
\end{verbatim}

Метод \texttt{SyntaxAnalyzer.launchMainNT} выглядит следующим образом:
\begin{verbatim}
prelaunch();
this.tree=generateAST();
postlaunch();
\end{verbatim}

В классе \texttt{IRPrologDescription} \texttt{get}-методы означают доступ к синтезированным атрибутам соответствующей формальной грамматики.

Класс \texttt{ParserRuleContext} содержит всё то, что употребляется наследованнымм и синтезируемымм атрибутами.

%%%%%%%%%%%%%%%%%%%%%
\subsubsection*{Прологовские термы}

\texttt{Пакет internal.core} см. рисунок \ref{PackageInternalCore}.

Этот пакет содержит все представления всвязи с термами, подцелями и правилами.

\begin{verbatim}
COMPOSITE::COMPONENT:
 PTerm
COMPOSITE::COMPOSITE:
 PTerm
COMPOSITE::CONCRETE COMPOSITE:
 PAtom, PVariable, PFunctor, PNumber, PList
\end{verbatim}

\texttt{PAtom.getChildren} возвращает \texttt{null}.
\texttt{PAtom.setChildren} игнорирует множество и при необходимости выделяет предупреждение на консоль.
\texttt{PSubgoal} представляет подцель, которая далее может содержать \texttt{PTerm} как составляющее (не подцелей).

%%%%%%%%%%%%%%%%%%%%%

Замечания:

В пакете \texttt{internal.core} (см. рисунок \ref{PackagesFrontendInternalcoreInternalaps}) нет необходимости в \texttt{PSubgoal}, т.к. оно уже содержится в \texttt{PRule}, а также потому, что подцели являются частями запросов (которые не являются частью грамматики построенных абстарктных предикатов).

Метод \texttt{GrammarBuilder.loadRules} читает все \texttt{PRule} и заносит их в формальную грамматику, при этом, атрибуты остаются как в Прологе.
Класс \texttt{GrammarBuilder} реализует оба интерфейса, \texttt{PTermVisitor} и \texttt{PRuleVisitor}, т.к. необходима интерпретация правил и должен иметься один нетерминал в качестве стартового символа грамматики.

\texttt{PTerm} и \texttt{PRule}-экземпляры генерируются в зависимости от экземпляров описаниях (Description-objects).

%%%%%%%%%%%%%%%%%%%%

\subsubsection*{Абстрактные предикаты}

Абстрактные предикаты необходимы для представления и логической обработки в рамках верификации правил Хорна.

Используемые анализаторы имеют один начальный символ нетерминал, однако интерфейс анализатора должен предоставлять возможность распознавать одно или более того нетерминалов данной формальной грамматики, в данном случае (ANTLR-)грамматики.

\texttt{Пакет internal.aps} см. рисунок \ref{PackagesInternalAPs}.

Замечания:

Все классы вовлечены в представление формальной грамматики (далее рассматривается возможность реализации на примере ANTLR).

\texttt{PRule} может быть представлено в виде \texttt{FG\_AP\_Rule}, а также обратно, следующим образом:

\begin{center}
\begin{tabular}{c}
   $\texttt{p1(X,pointsto(X,value1)):- ... .}$\\
   $\Leftrightarrow$\\
   $p1_{x1,x2} \rightarrow \{X_1 \approx X\}\{X_2 \approx pointsto(X,value1)\}, ... .$
\end{tabular}
\end{center}

Далее, \texttt{FG\_NT} имеет атрибуты \texttt{List<PTerm>}, но только не List<FG\_Att>, т.к. оно ссылается только на существующий нетерминал с конкретной последовательностью \texttt{PTerm}, согласно данному порядку определения атрибутов.

%%%%%%%%%%%%%%%%%%%%

\subsubsection*{Термы кучи}

Прологовские термы определены индуктивно над операторами $\{ \circ, || \}$. Полученный терм является итогом предыдущего синтаксического анализа.\\

\texttt{Пакет internal.ht} см. рисунок \ref{PackagesInternalHT}.

Подцели в абстрактных предикатах могут создавать связанный (под-)граф (ради простоты на данный момент, в общем случае с помощью проверок свойства $a\circ a$).

$a_1 \circ a_2 \circ a_3$ может быть записан непосредственно, используя $a_1,a_2,a_3$ (как часть абстрактного предиката) или как абстрактные предикаты
$a_x :- a_1, a_2.$, где оно же далее может опять же выступить в качестве некоторой подцели: $..., a_x, a_3, ...$.

Пример 2) $a_1 \circ (a_2 || a_3)$ можно перезаписать, как абстрактный предикат как: $a_x :- a_2,disj,a_3.$, где имеется некоторая последовательность как часть абстрактного предиката $...,a_1,a_x,...$.

%All marked 'pointsto' heaplets need to be marked as ``\textit{staged}'', meaning it may change to state ``\textit{recognised/accepted}'' in case (top-level) predicates calls as subgoals all match. Alternatively, this may be encoded by integers in order to easily drop failed predicates.\\

Замечания:

\texttt{HeapTerm} является внутренним представлением кучи на основе ЛРП.

Термы кучи над $\{\circ, ||\}$ может также содержать частичные спецификации. Терм кучи строится пошагово при анализе данных абстрактных предикатов.

\texttt{HeapTerm.simplify} ссылается на \texttt{BuiltinLibrary.simplify_1()}.
Упрощения термов кучи осуществляется с помощью перегруженных встроенных предикатов Пролгоа\\
\texttt{simplify\_2}.

Граф кучи представленный термом кучи, может также описывать не полностью специфицированные части с помощью Пролог-оператора ``\texttt{\_}''. Нормализация данных конъюнктов заключается в лексикографической сортировке порядка всех левых сторон базисных куч.

\texttt{DotGenerator} генерирует векторную графику на основе DOT-формата или в качестве его представления в качестве одной строки, которую позже можно записать в файл, например.

%%%

Замечания:

Проводится проверка корректности, например (i) соблюдается ли неповторимость кучи, (ii) связана ли куча, и т.д.\\[0.7cm]

%%%%%%%%%%%%%%%%%%%
\texttt{Пакет prolog} см. рисунок \ref{PackagesProlog}.

Библиотека \texttt{BuiltinLibrary} содержит все Прологовские предикаты, которые становятся доступными для логического анализатора на основе Пролога.
Если распознавание одного абстрактного предиката будет завершено целиком, тогда проводится проверка $a\circ a$.

Замечания:\\

\texttt{BuiltinLibrary:disj\_0} совершает один шаг в верификации динамической памяти. Для этого могут вводиться специальные метки в случае не достижения границ куч.

\texttt{BuiltinLibrary:unify} предоставляет замену классической унификации термов (\texttt{=/2}) с помощью и без проверки повторений.

Класс \texttt{HeapTerm} ссылается на \texttt{BuiltinLibrary::simplify\_1}.

\texttt{BuiltinLibrary::serialize\_1} почти эквивалентна к \texttt{serialize\_2}, ради исключения, что содержимое выдается на консоль.

\texttt{StackState} необходимо для проверки диапозонов годности (например, является ли символ частью данного контекста, когда все остальные семантические проверки уже успешно завершились). Сигнатур имеет тип: $\sigma:: String \rightarrow Type \rightarrow Value$.

%%%%%%%%%%%%%%%%%%%

\subsection*{Выводы}

Представленный подход представляет собой новую технологию для  автоматической верификации «$\mapsto$»-утверждений в абстрактных предикатах синтаксическим путём, избегая ручные предикаты  «\textit{сложить}» и  «\textit{раскрыть}».
Если раздел предикатов может быть представлен как корректное множество правил,  синтаксических правил перебора, то ответ будет найден, совпадает ли данная спецификация с данной конфигурацией кучи или нет. Используется атрибутируемая грамматика \cite{grune90} в качестве соединяющей модели представления правил в Прологе и конкретное множество правил перебора.  Атрибутируемая грамматика поддерживает наследованные и синтезированные атрибуты, которые переводятся как параметры головы правила Хорна. Кроме удобного представления значений, с помощью Пролога обработка сверху-вниз и с символами также совпадает с порядком вычисления и вывода. Без модификации логического вывода Пролога \cite{warren83} совершается процесс верификации.
Расширения, например, представленные преобразования термов и атомов, моделирование куч и т.д. не противоречит корневым свойствам  ЛРП (см. главу \ref{chapter:expression} и \ref{chapter:stricter}). Поэтому, свойства не меняются, если соблюдать  локальность спецификации объектных экземпляров. 
Не ссылочные поля объекта необходимо специфицировать только в одном предикате один раз.
Замечание о том, что прологовские правила могут обрабатываться  WAM и поэтому могут пострадать в скорости, нельзя считать справедливым, т.к. интерпретация в принципе может быть заменена компиляцией оптимизированного кода для любой целевой машины. Однако, вопрос быстрой верификации стоит не на первом месте. На первом месте стоят вопросы  сложности и выразимости. Таким образом, умышленно используется не самый быстрый способ верификации.

В отличие от предыдущих подходов, представлен новый подход, отличающийся резко тем, что  символы больше не являются просто локальными переменными.
Символы как логические используются произвольно в логических термах и предикатах. Язык утверждений  и верификации  теперь совпадают при использовании диалекта «\textit{Пролог}». Теперь дополнительные преобразования для обеих сторон отпадают, также как и имеющиеся ограничения в предыдущих подходах (см. \cite{bertot04}).

Открытым вопросом остаётся, как  \textit{правила ошибки} могут быть эффективно подключены к  генерации  контр-примеров и какие имеются способы для отслеживания ошибок в связи с этим? Из-за приостановки при первой ошибке, все остальные ошибки не рассматриваются, если \textit{настоящей} ошибкой по происхождению является одна из последовательных. Поэтому предлагается рассматривать  «\textit{конечный автомат}» для сравнения и минимизации  «\textit{прописной дистанции}» между имеющейся и ожидаемой кучами, аналогично глобальному  \textit{алгоритму Левенштейна} \cite{levenshtein65}, \cite{wagner74}, \cite{andoni12} для вычитания минимального количества операций замены за $\Theta(n^3)$, удаления и вставления подкуч.
Метод Левенштейна может быть не обобщён, но модифицирован и расширен к деревьям \cite{america89}, \cite{zhang89}, \cite{rutten96}.
«\textit{Метод паники}» \cite{grune90}, \cite{parr12} который при возникновении ошибки во время синтаксического анализа пытается продолжить перебор с помощью обнаружения самого ближнего следующего состояния, гарантированно верное и стабильное при удалении всех «\textit{лишних}» токенов.

Далее предлагается рассматривать вопрос о  частичных спецификациях, как константных функций \textbf{emp}, \textbf{true} или \textbf{false} для достижения упрощённой верификации (см. следующие главы). В частности, задаётся вопрос, можно ли с помощью частичных спецификаций выявить недостижимые мосты, которые предположительно связывают независимые кучи для исправления имеющейся спецификации.

%%%%%%%%%%%%%%%%%%%%%%%%%%%%%%%%%%%%%%%%%%%%%%%%%%%%%%%%%%%%%%%%%%%%%%%%%%%%%%%%%%%%%%%%%%%%%%%%%%%%%%%%%%%%%%%%%%%%%%%%%%%%%%%%%%%%%%%%%%%%%%%%%%% 7 Conclusions
%%%%%%%%%%%%%%%%%%%%%%%%%%%%%%%%%%%%%%%%%%%%%%%%%%%%%%%%%%%%%%%%%%%%%%%%%%%%%%%%%%%%%%%%%%%%%%%%%%%%%%%%%%%%%%%%%%%%%%%%%%%%%%%%%%%%%%%%%%%%%%%%%%%
\section*{Заключение}

«\textit{\textbf{Значимые результаты}}» Во время ра\-бо\-ты были по\-лу\-че\-ны результаты, которые можно классифицировать:
\underline{в области автоматизации доказательств:}
\begin{enumerate}
  \item Анализ разрыва между языками спецификации и верификации динамической памяти и анализ его сближения.
  \item Анализ и измерение выразимости трансформаций термов в логической парадигме на примере диалекта Пролога. Выявление, что логическое описание утверждениями из символов, переменных и термов в общем, а также реляций,
естественно лучше приспосабливает и ближе к верификации задач динамической памяти в Прологе.
  \item Термы Пролога как главное промежуточное представление на протяжении всего конвейера статического анализа проблем динамической памяти.
  \begin{itemize}
    \item Обоснование и использование логического языка в качестве спецификации и верификации. Главная мысль --- логическое программирование для решения верификации теорем.
    \item Предложенная платформа, изначально построена как открытая для расширений и вариаций. Платформа позволяет включение дополнительных экстерных библиотек и логических решателей в Прологе.
    \item Исследование демонстрирует преимущество Пролога.
   \end{itemize}
  \item Предложение об автоматизации верификации куч, как синтаксический перебор абстрактных предикатов, основано на пошаговой обработке граней графа кучи.
\end{enumerate}

\underline{выразимости языков спецификации и верификации:}
\begin{enumerate}
  \item Теперь доказуемые абстрактные предикаты могут без ограничения содержать символы и переменные согласно правилам Хорна, что расширяет сегодняшние ограничения. Теперь предикаты могут иметь потенциально любые рекурсивные определения, где выбранный метод синтаксического перебора может ограничить рекурсию, возможно потребовав переоформление правил.
  \item Повышение выразимости, благодаря ужесточению операций над кучами. Имеется строгое соотношение между выражением кучи и ее графом. Теперь пространственные операторы куч формируют однозначные термы утверждения. Доказываются алгебраические свойства моноида и группы, которые позволяют определение вычислений над кучами. Предложение о расширении и возможной стандартизации «\textit{UML/OCL}» указателями.
  \item Повышение выразимости и полноты, благодаря частичной спецификации куч переменных и объектных экземпляров. Сравнимость теперь позволяет все входные кучи одного набора сравнить и «\textit{ямы}» или «\textit{пересечения}» эффективно
вычислить. Теперь правила могут быть написаны, охватывая потенциально больше случаев.
  \item Теоретическая возможность, исходя из ужесточенного представления, выявить и проверить данную входную программу на минимальность для покрытия граф кучи полностью («\textit{проблема содержимого}») при исключении определений абстрактных предикатов в общем случае.
\end{enumerate}

\underline{Выводы:}
\begin{enumerate}
  \item Единый язык спецификации и верификации. Унификация языков спецификации и верификации к логическому языку программирования диалекту Пролога представляется как одна из крайних мер борьбы с многозначностью и позволяет решить поставленные перед собой проблемы выразимости и полноты. Чем меньше имеется языков и особенностей, тем меньше требуется рассматривать. Особенно, когда языки спецификации и верификации являются по синтаксису и семантике простым единым языком. Процесс верификации представляется как сравнение имеющейся и должной кучи.
  
  По ходу исследований существующих средств верификации динамической памяти и соответствующих проектов в качестве примеров, мною были разработаны «\textit{shrinker}» и «\textit{builder}» в качестве вспомогательных программных систем. Итогом расследования оказалось, что проблемы решённые верификаторами на практике крайне редкие или применимы только в случаях, для которых они были разработаны, а применение в других областях практически исключено.

  \item «Нет» новым и дубликатным определениям. Использование Пролога позволяет избежать введение всё новых конвенций, определений, языков спецификации, программирования, верификаций и даже сред. Отсутствует необходимость введения определений, которые оказывались в прошлом лишними и дубликатными. Выявлены важные свойства и критерии, которые учтены при построении платформы. Теперь кучи и части могут быть выбраны термами и символами. Абстрактные предикаты представляются обыкновенными прологовскими правилами. Используется промежуточное термовое представление, в качестве входного языка, а также разрешается Си-диалект, либо любой другой императивный язык программирования. Также допускается полное отсутствие входного языка, в таком случае, разрешается непосредственная передача промежуточного представления в ядро верификации.

  \item Повышенная выразимость в связи с символами на основе Пролога. Он основан на и реализует предикатную логику в качестве модели и реализации. Предикаты куч стали по-настоящему логически абстрактными (и не «\textit{Абстрактными}»), а именно прологовскими предикатами. В предикате символы допускаются в любых местах и ограничения касательно символьного использования отсутствуют. Логический вывод производится символами, логические выводы не возвращаются, благодаря унификации и реализации стека с возвратом на основе машины Уоррена. Проводилось обсуждение, что граф кучи можно полностью выразить. Объекты представляются как кучи, методы классов специфицируются не обязательно пред- и постусловием аналогично к процедурам. Правило фрейма способствует аналогичному поведению у процедур и позволяет включение в фрейм утверждения по всему циклу существования объекта. Объекты могут быть переданы в качестве символьного параметра абстрактным предикатам, таким образом, повышается модульность и благодаря неполным константным функциям неполный объект необходимо специфицировать.

  \item Реляционная модель правильна и нужна для эффективного описания абстрактных предикатов. Реляционная модель в отличие от других представлений была в итоге выбрана как наиболее удобная для описания и проверки куч. Также упоминают источники из списка литературы. Хотя численный анализ трудный, из-за выбора адекватных примеров и численной значимости, всё равно, проведённая экспертиза показала, что  запись в прологовских термах для проверенного набора примеров всегда лучше, а в среднем арифметическом на 30\% компактнее и более гибкая, чем например эквивалентная функциональная запись. Приведены типичные примеры, которые не могут быть эффективно или точно представлены в функциональных или императивных парадигмах.

  \item В Прологе предложена платформа конвейера. Платформа предлагает крайне простую архитектуру по обработке входного языка и правил верификации. Платформа основывается на Прологе, на конвейере фаз по обработке динамической памяти, которые могут индивидуально меняться и добавляться. Платформа отличается принципами: (а) автоматизацией, (б) открытостью, (в) расширяемостью и (г) понятностью. Выводимо только то, что выводимо из данного набора прологовских правил — это предоставляет известно резкое, но умышленное, ограничение. Правила верификации, леммы и индуктивно определённые абстрактные предикаты, всё записывается на Прологе. Модификации в структуре доказательства получаются задачами программирования и могут быть проверены на каждой фазе.

  \item Упрощение спецификации и верификации ради ужесточения. Нашлось ужесточенное определение одночленной кучи, вместо множественного числа куч. Ужесточение заключается в атомизме куч. Благодаря этому, чётче, проще и короче можно специфицировать кучи. Принудительное сравнение куч со всеми остальными кучами отпадает, основные кучи сохраняются. Сейчас кучи можно эффективно сравнивать. Благодаря групповым свойствам, теперь можно высчитывать, устанавливать в наборе правил не хватающую кучу и таким образом, можно проверять полноту. Далее, благодаря частичным  спецификациям, сейчас можно уменьшать и редуцировать количество и объём специфицируемых правил с кучами —  исключения исключаются, полнота доказывается с помощью неполноты. На практике ужесточение означает, что подвыражения кучи больше не должны полностью анализироваться.
  
  Так как «\textit{UML/OCL}» не содержит указателей, но представленные условия и критерии совместимы, то предлагается соответствующее расширение.

  \item Возможность определения всё новых теорий над кучами. Выявленные алгебраические свойства позволяют определять равенства и неравенства между кучами. Эти формальные теории могут непосредственно задаваться в Пролог, либо соответствующим SAT-решателям (возможно, также на основе Пролога). Теперь правила и теории проверяются проще.

  \item Абстрактные предикаты определяют атрибутируемую грамматику. Аналогично к проверкам схем при слабо структурированных данных, абстрактные предикаты являются шаблонами, которые содержат «дыры». Они наполняются во время верификации. Эта мета-концепция приводит к синтаксическому анализатору. В итоге логический вывод абстрактных предикатов можно интерпретировать как синтаксический анализ. Сравнение актуальной кучи с ожидаемой, можно расталкивать как проблему слов в формальных языках. Естественно, этот универсальный метод имеет ограничения касательно распознавателя, однако, подход очень практичен, несмотря на теоретический формализм, на котором основан процесс перевода.

  \item Универсальный подход к автоматизации утверждений. Предикаты и доказательство производятся на основе Пролога и могут мульти-парадигмально включать любые другие библиотеки. Доказательство приближается к программированию. Необходимость повторно определять процесс перевода и описания правил отпадает потому, что Пролог всё это включает в себя, хотя сам язык минимален.
 Необходимость проводить свёртывание и развёртывание вручную отпадает. Если куча не выводима по данным правилам, то можно их переписать. Синтаксический анализ решим и терминация гарантирована, например, на основе LL-распознавателей за линейное время. Процесс генерации анализатора производится лишь после изменения правил. В зависимости от доступности начальных нетерминалов, распознавание подграмматик избегает необходимости принудительного построения всё новых анализаторов. На практике все нетерминалы могут теоретически быть использованы после построения анализатора. Синтаксические анализаторы не нуждаются в дальнейших исследованиях, т.к. они отлично изучены на сегодняшний день, следовательно, не требуется дальнейшее исследование анализатора, который используется в данной реализации.
 
 В случае ошибки, автоматическая генерация и выдачи контр-примера, производятся анализатором без дополнительных затрат.
\end{enumerate}

«\textit{\textbf{Дальнейшее исследование}}» % Recommendations for FUTURE RESEARCH
В числе дальнейших исследовательских и практических работ можно выявить следующие (не отсортировано по важности):

\begin{enumerate}
 \item \textit{Интеграция решателя}. После определения теории куч, для ускоренного логического вывода необходимо подключение «\textit{SAT}»-решателя. Решатель также может быть написан в самом Прологе.
 \item \textit{I Интеграция в существующие среды}. Помимо упомянутых ограничений, необходимо на практике преобразовать термовое IR в IR пактов «\textit{GCC}», либо «\textit{LLVM}»  для индустриального применения.
 \item \textit{II Интеграция в существующие среды}. Точечный анализ мест «\textit{разрастания}» (образов) может также привести к более выгодному выделению и расположению памяти, подключая аппроксимацию метод абстрактной интерпретации. Также необходимо рассмотреть возможность выражений с статически выявляемыми офсетами ради широкой практической применимости.
 \item \textit{Изоморфизм куч}. Проверка изоморфизма графов в случае возникновения вопроса: может ли куча в принципе быть представлена данными указателями, при этом наименования кучи не уточняются?
 \item \textit{Нахождения минимального отклонения от спецификации}. При синтаксическом переборе может иметь смысл найти глобальный минимум прописной дистанции для проверки и возможной редакции спецификации.
 \item \textit{Анализ зависимостей указателей.} В отличие от анализа псевдонимов, рекомендуется исследовать расширение SSA-формы для динамических переменных, что до этого ещё никем не рассматривалось. Ожидается расхождение, т.к. понятие о зависимости различается. Возможность подключения может также подтвердить более безопасный код при оптимизации, т.к. на первый взгляд несвязанные ячейки связываются, либо точно не связываются. Это также способствует эффективному коду.
 \item \textit{Верификация кодов баз}. В этой работе после устранения технических ограничений в связи с быстрой прототипизацией, предлагается проводить анализ имеющихся кодов баз. Для этого предлагается использовать в открытом доступе пользования, потому, что они также значительно различаются.
 \item \textit{Исследование применимости языков трансформаций на практические нужды}. Обсуждённые языки трансформации могут иметь компактную запись и также часто записываются с помощью правил состояний до и после трансформации, однако ожидается большое расстояние между используемой моделью кучи, моделью выразимости и содержанием динамических куч, входным языком.
 \item \textit{Абстракция вывода.} Некоторые диалекты Пролога поддерживают абдуктивный вывод. В общем, абстракция механизма, который позволяет правила Пролога упростить, сравнить с подходящими пред- и постусловиями и выбрать наиболее подходящее правило. Порядок вывода можно упростить, если соотношения между термами обратимы. Это как раз было предложено и обсуждено в этой работе.
\end{enumerate}

«\textit{\textbf{Рекомендации для применения}}» % Recommendations for practitioners
В числе рекомендаций по итогам работы можно вывести следующие:

\begin{itemize}
 \item Практически спецификация куч стала (немного) проще, если даже не больше. Данная программа аннотируется спецификацией. 
 \item Абстрактные предикаты должны быть синтаксически разборчивыми, тогда подцели соответствующей формальной грамматики можно распознать. В зависимости от мощи анализатора часто переписывание правил приводит к разборчивости.
 \item Использование ужесточенных операторов даёт возможность в спецификациях\\ «\textit{UML\-/OCL}» пространственность объектных экземпляров выражать.
 \item Введение всё новых языков программирования не обязательное. Язык программирования может даже полностью отсутствовать. Расширяемое и модифицируемое промежуточное представление позволяет максимальную гибкость. Архитектура верификации имеет те же самые свойства.
 \item Сравнение куч с данными выражениями в правилах верификации способствует выявить не специфицированные кучи, а более широко способствует к решению вопроса полноты.
 \item Перегрузка значений реляций упрощает гибкость логического вывода, например, в связи с абдукцией, но также может быть использована различными способами.
 \item Использование SAT-решателя для задаваемых теорий куч, будь-то в Прологе или мульти-парадигмально, может привести к резкому ускорению преобразования куч в нормализованную и упрощённую форму.
 \item Использование мемоизатора (в некоторых диалектах Пролога «\textit{tabling}») для абстрактных предикатов может привести к резкому ускорению логического вывода, в частности при пошаговой верификации.
 \item При введении и модификации фаз работы с динамической памятью рекомендуется придерживаться к термовому представлению.
\end{itemize}

 %%%%%%%%%%%%
  % -- Acronyms --
  \section*{Nomenclature} % List of acronyms
\addcontentsline{toc}{section}{Nomenclature}

\begin{tabular}{rcl}
 \textbf{AA} & \qquad & Alias Analysis\\
 \textbf{ABI} & \qquad & Application Binary Interface\\
 \textbf{ACC} & \qquad & Abadi-Cardelli Calculus/(-i)\\
 \textbf{ADT} & \qquad & Abstract Data Type\\
 \textbf{ALC} & \qquad & Abadi-Leino Calculus/(-i)\\
 \textbf{ANTLR} & \qquad & a compiler-compiler\\
 \textbf{AP} & \qquad & Abstract Predicate\\
 \textbf{API} & \qquad & Application Programming Interface\\
 \textbf{AST} & \qquad & Abstract Syntax Tree\\
 \textbf{AT} & \qquad & Automated Testing\\
 \textbf{BB} & \qquad & Basic Block\\
 \textbf{BISON} & \qquad & a compiler-compiler\\
 \textbf{BNF} & \qquad & Backus-Naur Form\\
 \textbf{BSS} & \qquad & Block Started by Symbol\\
 \textbf{CF} & \qquad & Context-free\\
 \textbf{CI} & \qquad & Code Introspection\\
 \textbf{Clang} & \qquad & LLVM's compiler frontend\\
 \textbf{CompCert} & \qquad & Leroy's Compiler Prover\\
 \textbf{Coq} & \qquad & a theorem prover\\
 \textbf{CPU} & \qquad & Central Processing Unit\\
 \textbf{CFG} & \qquad & Control-flow Graph\\
 \textbf{CRT} & \qquad & Church-Rosser Theorem\\
 \textbf{DCG} & \qquad & Definite Clause Grammar\\
 \textbf{DEC} & \qquad & Digital Equipment Corporation\\
 \textbf{DOT} & \qquad & a textual vector graphics format\\
 \textbf{ESC} & \qquad & Extended Static Checker (Java)\\
 \textbf{GC} & \qquad & Garbage Collection\\
 \textbf{GCC} & \qquad & GNU Compiler Compiler\\
 \textbf{GIMPLE} & \qquad & GNU tree IR of Programming units\\
 \textbf{GNU} & \qquad & GNU is Not Unix\\
 \textbf{GUI} & \qquad & Graphical User Interface\\
 \textbf{HDD} & \qquad & Hard Drive Disk\\
 \textbf{EBNF} & \qquad & Extended BNF\\
 \textbf{IR} & \qquad & Intermediate Representation\\
 \textbf{ISO} & \qquad & Intl. Standardization Organization\\
 \textbf{LCF} & \qquad & Logic of Computable Functions (Scott)\\
 \textbf{LISP} & \qquad & List Processing (language)\\
 \textbf{LL} & \qquad & Left-to-Left parser\\
 \textbf{LLVM} & \qquad & Low-level Virtual Machine\\
 \textbf{LR} & \qquad & Left-to-Right parser\\
 \textbf{ML} & \qquad & Meta-Language\\
 \textbf{MVC} & \qquad & Model-View-Controller\\
 \textbf{NT} & \qquad & Non-Terminal\\
 \textbf{OC} & \qquad & Object Calculus\\
 \textbf{Ocaml} & \qquad & Oca-ML\\
 \textbf{OCL} & \qquad & Object Constraint Language\\
 \textbf{OOM} & \qquad & Object-oriented Modelling\\
 \textbf{OS} & \qquad & Operating System\\
 \textbf{PCP} & \qquad & Post's Correspondence Problem\\
 \textbf{PCF} & \qquad & Programming Computable Functions\\ %Plotkin
 \textbf{PL} & \qquad & Programming Language\\
 \textbf{Prolog} & \qquad & Programming in Logic (language)\\
 \textbf{ProLogika} & \qquad & a Prolog-dialect heap verifier\\
 \textbf{PVS} & \qquad & a theorem prover\\
 \textbf{RC} & \qquad & Region Calculus
\end{tabular}

\begin{tabular}{rcl}
 \textbf{RAM} & \qquad & Random-Access Memory\\
 \textbf{ROSE} & \qquad & a compiler-compiler\\
 \textbf{RTTI} & \qquad & Run-time type identification\\
 \textbf{SA} & \qquad & Shape Analysis\\
 \textbf{SL} & \qquad & Separation Logic\\
 \textbf{SAT} & \qquad & SATisfaction of a logical formula\\
 \textbf{SMT} & \qquad & Satisfactority Modulo Theory\\
 \textbf{SSA} & \qquad & Single Static Assignment\\
 \textbf{SSD} & \qquad & Solid State Drive\\
 \textbf{tuProlog} & \qquad & a Prolog backend\\
 \textbf{TIOBE} & \qquad & The Importance of Being Earnest (index)\\
 \textbf{UML} & \qquad & Unified Modeling Language\\
 \textbf{VM} & \qquad & Virtual Machine\\
 \textbf{WAM} & \qquad & Warren's Abstract Machine\\
 \textbf{XOR} & \qquad & eXclusive OR\\
 \textbf{YACC} & \qquad & a compiler-compiler\\
 \textbf{XML} & \qquad & eXtensible Meta Language\\
 \textbf{iff} & \qquad & if and only if\\
 \textbf{algo} & \qquad & algorithm\\
 \textbf{cf} & \qquad & compare\\
 \textbf{conv} & \qquad & convention\\
 \textbf{cor} & \qquad & corollary\\
 \textbf{def} & \qquad & definition\\
 \textbf{eg} & \qquad & for instance\\
 \textbf{eqn} & \qquad & equation\\
 \textbf{fig} & \qquad & figure\\
 \textbf{lem} & \qquad & lemma\\
 \textbf{obs} & \qquad & observation\\
 \textbf{poset} & \qquad & partially-ordered set\\
 \textbf{stderr} & \qquad & console standard error\\
 \textbf{stdout} & \qquad & console standard output\\
 \textbf{theo} & \qquad & theorem\\
 \textbf{thes} & \qquad & thesis\\
 \textbf{wlog} & \qquad & without loss of generality\\
 \textbf{wrt} & \qquad & with respect to
\end{tabular}

  %  % -- Glossary --
  \newpage

%%%%%%%%%%%%%%%%%%%%%%%%%%%%%%%%%%%%%%%%%%
\section*{Glossary}
\addcontentsline{toc}{section}{Glossary}

%%%%%%%%%%%%%%%%%%%%% Hoare Calculus & Proving

\myglsentry{Absorption}
\myglsdesc{
In a boolean algebra, given some \gref{propositions} $p$ and $q$ the following holds: 
$p \wedge (p \vee q) \equiv p \vee (p \wedge q) \equiv p$.
}

\myglsentry{Abstract Data Type}
\myglsdesc{
Precursor concept of the \gref{class} concept where operations define interfaces, \gref{methods} and \gref{fields}.
Defines operations for both, which change its interior behaviour and communication with external objects.
}

\myglsentry{Abstract Interpretation}
\myglsdesc{
\gref{Static analysis method} based on interval calculi and extended denotations.
}

\myglsentry{Abstract Predicate}
\myglsdesc{
In this work, a heap predicate with term parameters.
Heaps are composed of simple heaps that are atomic \gref{points-to} \gref{assertions}.
}

\myglsentry{Alias}
\myglsdesc{
A \gref{pointer}, which points to the same content as another pointer.
}

\myglsentry{(Algebraic) Field}
\myglsdesc{
A discrete structure with two operations (a special ring), for which hold, e.g. rules over $+$ and $\cdot$, with two neutral and inverse elements.
If the field is finite, then the field is Galois.
}

\myglsentry{(Algebraic) Group}
\myglsdesc{
A \gref{monoid} in which for each element the inverse element exists.
}

\myglsentry{Antecedent}
\myglsdesc{
$A$ is the antecedent for any rule of kind: "\textit{if $A$ then $B$}".
}

\myglsentry{Assertion (about a Program Property)}
\myglsdesc{
Assertion about a calculation state described by some \gref{formal language}, e.g. by first-order \gref{predicate} logic.
}

\myglsentry{Atomism}
\myglsdesc{
An analytical (philosophical) concept suggests breaking complex problems down into small individable units, namely atoms.
Has found wide adaption in many branches of science.
Here a discussion of heaps and compositions were discussed.
}

\myglsentry{Attributed Grammar}
\myglsdesc{
A \gref{formal grammar} with rules that are extended by attributes may either be \gref{inherited} or \gref{synthesised}.
Attributes may be considered as parameters in their rule bodies.
}

\myglsentry{Automated Theorem Proving}
\myglsdesc{
Proof of a \gref{theorem} using formal apparatus, s.t. the proof runs through from the beginning till the end without the need for human interaction. This is in contrast to \gref{partially automated proofs} or \gref{interactive proofs}, which both require interaction.
}

\myglsentry{Automatically (Stack-) allocated variable}
\myglsdesc{
See \gref{stack}
}

\myglsentry{Calculation State}
\myglsdesc{
Description of \gref{stack} and \gref{heap}, so a list of all \gref{variables}, which during a calculation step are allocated having some \gref{content} assigned.
}

\myglsentry{Church's (Term) Arithmetics}
\myglsdesc{
Arithmetic of natural numbers over terms instead of numbers.
The mapping between both domains is defined by $0 \mapsto z$, $1 \mapsto s(z)$, $2 \mapsto s(s(z))$, and so on.
So, $m+n$ may be defined as $s^m + s^n = s^{m+n}$.
Minus as plus' inverse can be defined in analogy to that by structural induction and by term unification and \index{pattern matching} pattern matching.
The advantage of terms over natural numbers is \gref{invertibility} of terms due to term unification.
}

\myglsentry{Class}
\myglsdesc{
A record (struct) type uniting \gref{fields} and \gref{methods}.
}

\myglsentry{Class Field}
\myglsdesc{
Memory cell of arbitrary type (e.g. a class) of a class-object.
}

\myglsentry{Class Inheritance}
\myglsdesc{
Inheritance of \gref{fields} and \gref{methods} from a super-class into a sub-class according to defined \gref{visibility levels}.
}

\myglsentry{Class Object}
\myglsdesc{
A variable referring to some allocated, atomic and non-overlapping memory region, instantiated by a corresponding \gref{class} or subclasses.
}

\myglsentry{Class Polymorphism}
\myglsdesc{
An \gref{object instance} is of some \gref{class}.
Type membership may change.
It may be assigned to subclasses which would cause a different runtime behaviour if a subclass' method is overridden and called.
}

\myglsentry{(Class) Visibility Mode}
\myglsdesc{
\gref{Fields} and \gref{methods} of a \gref{class} may have an explicit visibility mode assigned and be inherited.
Unlike C structs, fields and methods are private by default.
The default class visibility is public.
Visibility inherited in Java stays the same or is reduced by one level.
In C, it is similar to Java, except for inheritance modifiers.
The visibility levels are usually: \texttt{private, public, protected}.
}

\myglsentry{Clause}
\myglsdesc{
Disjunctive normal-form in which all \gref{literals} are interconnected.
}

\myglsentry{Co-Routine}
\myglsdesc{
A (sub-)procedure call simultaneously to some main procedure.
\gref{Functionals} may often run as co-routine.
}

\myglsentry{(Codd's) Relational Algebra}
\myglsdesc{
Here relational "\textit{functions}" as predicates, particularly as \gref{Prolog predicates}.
The base operations in the original algebra are defined over sets: union, intersection, Cartesian product, projection, selection and renaming.
}

\myglsentry{Code Introspection}
\myglsdesc{
Modification and reading of any object's \gref{fields} and \gref{methods} during runtime.
Code to be loaded may only be available during runtime, e.g. constructed from external data sources.
}

\myglsentry{Control Inversion}
\myglsdesc{
Usually, control is passed from the calling side (caller) to the called side (callee).
Inversion flips the call direction and is occasionally needed in order to \index{Information Hiding Principle} stick to the information hiding principle.
It is an essential technique for \gref{framework} design, particularly for all callbacks.
}

\myglsentry{Consequent}
\myglsdesc{
$B$ is the consequent for any rule of kind: "\textit{if $A$ then $B$}".
}

\myglsentry{Content Of Pointer}
\myglsdesc{
Data in the memory cell pointed by some \gref{pointer}.
}

\myglsentry{Counter-Example}
\myglsdesc{
A single example showing that a certain \gref{assertion} does universally not hold.
}

\myglsentry{Counter-Example Generation}
\myglsdesc{
See \gref{counter-example}.
Generation is due to Prolog term unification or abstract predicate derivation of calculated and expected heap terms.
}

\myglsentry{Data Dependency Analysis}
\myglsdesc{
Analysis based on the places in code for variables, namely its definition and uses and how these are influencing other variables, e.g. in a SSA-form.
}

\myglsentry{Declarative Programming Language}
\myglsdesc{
Logical or functional programming languages.
"\textit{Imperative}" solves a problem by an instruction list, where "\textit{declarative}" only describes.
}

\myglsentry{Dependent Type (also Recursive or Generalised type)}
\myglsdesc{
Type depending on type(s).
E.g. records (struct).
}

\myglsentry{Dynamic Memory}
\myglsdesc{
The unorganised memory segment as \gref{stack}, allocated by a process spawn by the OS.
Dynamic memory is synonym of \gref{heap} and contains all dynamically allocated variables.
}

\myglsentry{Dynamic Memory Issues}
\myglsdesc{
Problems such as \gref{dangling pointers}, \gref{memory leaks}, \gref{invalid memory access}.
}

\myglsentry{Dynamic Variable}
\myglsdesc{
Allocation and deallocation of \gref{heap} memory are specified by reserved statements, like \texttt{new} or \texttt{free}.
Non-specified deallocations are freed when the OS releases the primary process.
}

\myglsentry{Earley-Parser}
\myglsdesc{
Some CF-parser pushing in rule definitions into the \gref{stack}.
}

\myglsentry{Edge-based Heap Graph}
\myglsdesc{
\gref{Heap graph} in which edges are a list with two columns: source and destination.
}

\myglsentry{Erroneous Code Localization}
\myglsdesc{
An \gref{(erroneous) symptom} search process terminating in the genuine reason a problem occurs.
A manual search is called debugging.
}

\myglsentry{Error-Production Rules}
\myglsdesc{
Rules in case of a syntax error during parsing that define the further steps, including a possible error messaging, recovery or further restoration tasks to be triggered.
}

\myglsentry{Evaluation Ordering}
\myglsdesc{
The ordering in which (sub-)expressions are evaluated, e.g. left-to-right, outermost-first, \gref{lazy evaluation}.
}

\myglsentry{Fail-As-Error}
\myglsdesc{
Prolog built-in \gref{predicate} "\texttt{fail}" in combination with \gref{cut} alternatives.
}

\myglsentry{Finite Automaton}
\myglsdesc{
A simple directed graph representing \gref{calculation states} and whose edges represent state transitions from the source state to the destination state.
An automaton has one starting state and one or more final states.
Transitions are rules which must be right-recursive, so of kind $A \rightarrow a B$, where $a$ denotes a terminal, and $A, B$ are non-terminals.
Non-terminals are replaced with their right-hand side definition.
}

\myglsentry{Fixpoint Combinator}
\myglsdesc{
A syntactic, artificial operator is introduced and is no element of any programming language.
It is to simulate recursion.
Combinators are elements of $\lambda$-calculi.
}

\myglsentry{Formal Grammar}
\myglsdesc{
A formal notation proposed by Noam Chomsky to generate non-ambiguous formal languages, in analogy to natural language's grammars.
Programming, specification, and verification languages are examples of formal languages.
A formal language is defined as $(S,T,NT,P)$, where $S \in NT$ denotes some starting non-terminal, $T$ a set of terminals, $NT$ a set of non-terminals, and $P$ a set of production rules.
The expressibility level of a formal language depends on the $P$.
}

\myglsentry{Formal Language}
\myglsdesc{
Generated language for a given \gref{formal grammar}.
}

\myglsentry{Formal Language Mismatch}
\myglsdesc{
Formal languages may have different purposes, they differ in syntax, semantics and intuitive meaning.
Symbols and alphabet may have different connotations with different level of ambiguity.
The differences cause problems with expressibility, comparability, and further.
}

\myglsentry{Framework}
\myglsdesc{
A special library, also a non-standalone software, with dedicated functions.
Framework functions often are due to a very special purpose only and have to obey dedicated instantiation constraints and are specially designed with software extension and variability points.
A framework defines roles of interaction between several actors as abstract data types.
The interaction points are designed by contract.
A user does not have to know all details functions are implemented, but the user must know how to obtain and modify functionality as well as how to extend behaviour.
}

\myglsentry{Garbage Collection}
\myglsdesc{
E.g. \gref{generational}, \gref{mapped}.
See \gref{heap garbage}.
}

\myglsentry{Generalized Heap}
\myglsdesc{
Is either a \gref{simple heap}, a composition of simple heaps, or a complex \gref{heap} possibly containing \gref{abstract predicates} calls as \gref{subgoals}.
}

\myglsentry{Generational Garbage Collection}
\myglsdesc{
\gref{Garbage collection} where pointers with an older timestamp or content since the last modification is moved or removed entirely from \gref{dynamic memory}.
}

\myglsentry{Global Variable}
\myglsdesc{
Global variables are accommodated once to the \texttt{.bss} memory segment managed by the OS containing uninitialised data.
A global variable may be covered by other non-global variables, e.g. by \gref{local variables} with the same qualified name.
}

\myglsentry{(Graph) Dominator (Vertex)}
\myglsdesc{
Vertex in a directed graph relative to another vertex, so no alternative path exists, except all passing the dominator vertex.
}

\myglsentry{Heap Address Space}
\myglsdesc{
Linearly addressable memory chunk in virtual memory allocated by an OS process occupied by the \gref{heap} segment.
}

\myglsentry{Heap Assertion}
\myglsdesc{
In terms of points-to model, an assertion about \gref{pointers}, their interdependencies and content in \gref{heap}.
}

\myglsentry{Heap Cell}
\myglsdesc{
A contiguous memory leap whose size is defined by the corresponding \gref{pointer's} type.
}

\myglsentry{Heap Configuration}
\myglsdesc{
Some concrete ground \gref{heap graph}.
}

\myglsentry{Heap Conjunction}
\myglsdesc{
Two \gref{heaps} are connected, implying at vertex is shared between both.
}

\myglsentry{Heap Disjunction}
\myglsdesc{
Two \gref{heaps} are separate from each other, so there is no connection between both.
}

\myglsentry{Heap Inversion}
\myglsdesc{
\gref{Heap conjunction} together with its inverse \gref{heap} is defined as the \gref{empty heap}.
}

\myglsentry{Heap Formula}
\myglsdesc{
Some \gref{heap assertion} has no unground free \gref{variables}.
A heap formula is subject to interpretation, so it is true or false for some concrete heap.
}

\myglsentry{Heap Frame}
\myglsdesc{
Invariant \gref{heap} parts used to specify procedure calls.
}

\myglsentry{Heap Garbage}
\myglsdesc{
\gref{Heap} parts that are not pointed by \gref{pointers} nor locations.
Heap parts are subject to \gref{garbage collection}.
}

\myglsentry{Heap Graph}
\myglsdesc{
A graph representation of a \gref{heap}, where edges define dependencies between heaps and vertices are pairs $loc \mapsto val$.
}

\myglsentry{Heap Interpretation}
\myglsdesc{
Comparison between a given \gref{heap} and a specified.
}

\myglsentry{Heap Structure}
\myglsdesc{
Spatial connected data structure located in the \gref{heap} segment.
Its location may either be in \gref{stack} or \gref{heap}, and where its content is in heap.
}

\myglsentry{Higher-Order Function}
\myglsdesc{
Function accepting functions as input and output.
}

\myglsentry{Hindley-Millner Type Systems}
\myglsdesc{
\gref{Type checking} with typing rules, where the major inference rule is: Given some type $a$ and some \gref{functional} (e.g. some unary operator) of type $a \rightarrow b$, then the result will have type $b$.
}

\myglsentry{Hoare Calculus}
\myglsdesc{
A formal method proving the presence of properties for a given program using mathematical-logical formulae.
The effect of a given statement for some imperative programming language is caught by a \gref{calculation state} before execution, called \gref{pre-condition}, and after execution, called \gref{post-condition}.
}

\myglsentry{Hoare Triple}
\myglsdesc{
A triple of kind $\{P\}C\{Q\}$, where (1) \gref{calculation state} $P$ before $C$ is executed, (2) statement $C$, (3) calculation state $Q$ after $C$ is executed.
}

\myglsentry{Homomorphism}
\myglsdesc{
If there is one function (or operation) $g$, then there is a homomorphism w.r.t. to yet another function $h$, if $h(g(x_0,x_1,\cdots , x_n)) \equiv g(h(x_0), h(x_1), \cdots , h(x_n))$ holds, where $\forall j.x_j$ are term arguments.
}

\myglsentry{Howard-Curry Isomorphism}
\myglsdesc{
Postulate according to which proving and programming are interconnected.
}

\myglsentry{I/O-Term (in Prolog)}
\myglsdesc{
A Prolog term that is used as \gref{I-term} and \gref{O-term} at the same time.
This is possible, for instance, in unground composed terms.
}

\myglsentry{I-Term (in Prolog)}
\myglsdesc{
A Prolog term that can, but not necessarily has to, be ground.
Applied to a subgoal, expects term is used but remains untouched.
}

\myglsentry{Incomplete Heap}
\myglsdesc{
Partially defined heap, which may contain unassigned symbols.
Incomplete heaps may match patterns for short specifications.
}

\myglsentry{Inductive Structure Definition}
\myglsdesc{
\gref{Linear lists}, binary tree, \gref{heap}, etc.
}

\myglsentry{Information Hiding Principle}
\myglsdesc{
States by default \gref{abstract data types} need to restrict \gref{visibility levels} as much as possible in order to avoid too lax permission granting, so for instance, communication interfaces may adequately be defined, varied and extended later on.
}

\myglsentry{Inherited Attribute}
\myglsdesc{
An attribute in an \gref{attributed grammar} is passed from the calling rule down to the called rule.
Here \gref{predicates} may act like rules.
}

\myglsentry{Initial Denotation}
\myglsdesc{
Denotation assigned to variables.
Uninitialised variables may lead to a series of problems, especially when reading an uninitialised variable.
\gref{Heap variables} and data on the \texttt{.bss}-segment are not initialised.
}

\myglsentry{Input Token Stream}
\myglsdesc{
A stream of \gref{tokens} obtained during lexical analysis.
}

\myglsentry{Intermediate Representation}
\myglsdesc{
\gref{Terms}, triads, tetrads and further notations representing some incoming program and annotated verification conditions and abstract \gref{predicates}.
}

\myglsentry{Inter-Procedural Alias-Analysis}
\myglsdesc{
Calculation of \gref{alias} parameters and \gref{global variables}, which may change between procedure calls.
Compared to the \gref{intra-procedural} analysis, inter-procedural analyses are more complicated and usually are bound by exponential complexity. 
}

\myglsentry{Intra-Procedural Alias-Analysis}
\myglsdesc{
See \gref{inter-procedural alias analysis}.
It is only inside a procedure and often analyses are bound by polynomial complexity.
}

\myglsentry{Intuitive Reasoning}
\myglsdesc{
Judgements on reasoning that rely only on \gref{facts} and rules of form \gref{if $A$, then $B$} (so-called "\textit{modus ponens}").
}

\myglsentry{Invalid Memory Access}
\myglsdesc{
May read unintended data and cause a page fault.
In consequence, may corrupt execution, crash or even compromise the running process.
}

\myglsentry{Lambda Abstraction}
\myglsdesc{
A theoretical abstraction of a classic procedure, as known in Pascal or C.
It was proposed in the 1930s by Alonzo Church in order to reason about computability.
}

\myglsentry{Lazy Evaluation}
\myglsdesc{
Expression evaluation from outside-in.
Initially, parameters in a procedure call are passed non-interpreted and are interpreted in a procedure only if referenced concretely.
Internally lazy expressions are assigned a symbol and passed that until some arithmetic or logical operation has to be performed.
}

\myglsentry{Levenshtein's Algorithm}
\myglsdesc{
The core algorithm for determining the minimal edit distance, an integer denoting the total edits required to equal two strings.
The algorithm is of cubic complexity and may for the same complexity by enriched by path segment tracing.
}

\myglsentry{Linear List}
\myglsdesc{
A list whose elements have no pointers or are all pointing to nil. In Reynolds' memory model the "," initialises a linear list.
}

\myglsentry{Literal (Assertion)}
\myglsdesc{
An atomic \gref{assertion} or its negation.
}

\myglsentry{LL-Parser}
\myglsdesc{
Parsing where \gref{input tokens} are processed from left to right, and rules of given corresponding \gref{formal grammar} are processed top-down.
}

\myglsentry{Local Variable}
\myglsdesc{
\gref{Automatically managed variable}
}

\myglsentry{Locality Principle}
\myglsdesc{
Principle after which local changes to some \gref{heap} must not affect the whole \gref{heap}.
If modelled correctly, then this may heavily improve heap analyses.
}

\myglsentry{Location}
\myglsdesc{
A variable, pointer or any path to any reachable object field.
}

\myglsentry{Logical Heap Operator}
\myglsdesc{
Logical implication, disjunction, conjunction, negation.
}

\myglsentry{Loop Invariant}
\myglsdesc{
A condition formula, which holds for the time a loop is executed.
The condition does not hold anymore when that loop is quit.
A full and exact invariant refers to all visible \gref{variables} from a loop's body in a single generalised formula.
The loop invariant shall not contain unused variables as well as redundant cases.
An invariant shall not contain superfluous cases.
Hence, invariants are challenging to derive automatically in a constructive way.
}

\myglsentry{LR-Parser}
\myglsdesc{
Parsing where \gref{input tokens} are processed from left to right, and rules of given corresponding \gref{formal grammar} are processed bottom-up.
}

\myglsentry{Mapped Garbage Collection}
\myglsdesc{
Regions of \gref{dynamic memory} are mapped and processed afterwards in parallel on the same principle as z-buffering works.
Regions that often are altered are scattered along axes to the map.
Closer regions are garbage collected depending on the chosen strategy.
}

\myglsentry{Memoization}
\myglsdesc{
Caching \gref{predicate} calls and functions.
In Prolog, it is called "\textit{tabling}".
}

\myglsentry{Memory Cell Overlapping / Sharing}
\myglsdesc{
When some memory cell is entirely or \gref{partially} interpreted differently, e.g. by varying \gref{pointers (types)}, or memory alignment, byte and bit-ordering, or programming language constructs, like \index{\texttt{union}} "\texttt{union}" in C-dialects.
}

\myglsentry{Memory Leak}
\myglsdesc{
\gref{Memory cells} in \gref{dynamic memory} unreachable from \gref{stack variables}.
}

\myglsentry{Memory Location}
\myglsdesc{
Memory addresses replacements of program variables in \gref{stack} and \gref{heap}.
}

\myglsentry{Memory Model (by Burstall)}
\myglsdesc{
\gref{Simple heaps} in SL have the form: $location \mapsto address$.
}

\myglsentry{Memory Model (by Reynolds)}
\myglsdesc{
\gref{Simple heaps} in SL have the form: $location \mapsto value$.
}

\myglsentry{(Memory/Heap) Pointer}
\myglsdesc{
Pointer is of pointer type, where its memory cell contains an address to some content in virtual memory, which is then interpreted according to the pointer type from its declaration.
Dereferencing a valid memory cell accesses the pointer's \gref{content}.
Pointer may be of type \gref{pointer} of some further type.
}

\myglsentry{Meta Scheme}
\myglsdesc{
In contrast, to \gref{fold}/\gref{unfold}, it is a universal schema in proofs similar to, e.g. \gref{analogy}, \gref{induction}, \gref{abduction}.
}

\myglsentry{Method Signature}
\myglsdesc{
Consists of \gref{method name}, if within a class then a prefixed \gref{class} name, as well as all types of formal parameters (maybe further distinguished in incoming and outcoming parameters).
}

\myglsentry{Model Checking}
\myglsdesc{
Model described by (in-)equalities for a given \gref{(formal) theory}.
\gref{Calculation state} is described by a formula that needs to be checked by a \gref{(SMT) solver} to be constructed.
}

\myglsentry{Monoid}
\myglsdesc{
A discrete structure with one \gref{total} operation for which closure, associativity, neutrality are obeyed.
Half-group is a synonym.
}

\myglsentry{Natural Deduction}
\myglsdesc{
Logical proving scheme based on assumptions and consequences or its negations.
The structure of a naturally deducted proof is very close to \gref{Hoare calculi} in general.
An alternative major to natural deduction is Robinson's resolution method.
}

\myglsentry{Non-Interleaving Heap Cells}
\myglsdesc{
A \gref{heap graph} cell is non-interleaving between two \gref{heap} graphs, if both graphs do no share a common vertex at the same time.
}

\myglsentry{O-Term (in Prolog)}
\myglsdesc{
A Prolog term that can, but not necessarily has to, be ground.
Applied to a subgoal, it is expected the term is bound afterwards.
}

\myglsentry{Object Calculi}
\myglsdesc{
(1st kind) class-based calculus after Abadi-Leino using \gref{classes}, very close to classes in Java and C++, or (2nd kind) object-based calculus after Abadi-Cardelli, where everything is an object, a record of methods and fields with possibly nested further objects all without classes, very close to objects in Baby Modula 3 and a few less known mostly experimental scripting languages from Cardelli.
}

\myglsentry{Object Field}
\myglsdesc{
A memory location for a given \gref{class} field will be assigned a memory address during runtime belonging to a surrounding object memory layout.
A field maybe added, inserted and deleted statically by \gref{class inheritance} or dynamically by \gref{code introspection}.
}

\myglsentry{Object Invariant}
\myglsdesc{
An invariant \gref{assertion} (cf. \gref{loop invariant}), which shall always hold and regularly be checked during the lifespan of some \gref{object instance}, particularly before and after construction, destruction, arbitrary \gref{method} calls and \gref{field} modifications.
}

\myglsentry{Object Life-Cycle}
\myglsdesc{
The period between object creation until destruction.
}

\myglsentry{Object Method}
\myglsdesc{
A procedure within the scope of some \gref{class}, which may refer to the class \gref{fields} as well as to the \gref{this} keyword.
In contrast to fields, in classic \gref{object calculi} without \gref{introspecting} code may not be modified during runtime.
}

\myglsentry{(Full) Object Specification}
\myglsdesc{
Includes \gref{object invariant}, \gref{pre-condition}, \gref{post-condition} of all methods and all inner specifications inside methods.
}

\myglsentry{Object Type}
\myglsdesc{
\gref{class}
}

\myglsentry{Partial Correctness}
\myglsdesc{
Correctness without the claim to terminate.
}

\myglsentry{Partial Heap Specification}
\myglsdesc{
\gref{Heap} specification containing special operands, some anonymous placeholder which is existentially quantified for each occurrence.
It is used to summarise unspecified heap fragments. 
}

\myglsentry{Partial(-ly defined) Hoare Triple}
\myglsdesc{
If execution of some program statement for some \gref{Hoare triple} not necessarily terminates, then its \gref{post-condition} is undetermined.
}

\myglsentry{Peano Axioms}
\myglsdesc{
Given the operations "$+$", "$\cdot$" and some \gref{ordering on} (in the common case some lifted and equivalent to) a set of natural numbers, on which popular axioms of addition are established.
For example, equality can be defined as equality of the set's first element and (inductively continued) equality of the following numbers (cf. with \gref{Church's arithmetics}).
}

\myglsentry{Pointer to a Pointer}
\myglsdesc{
\gref{Pointer} whose content is pointer again.
}

\myglsentry{Pointer Rotation}
\myglsdesc{
Approach based on tiny pointer operation within the programming language to be reasoned, e.g. clock-wise permutation.
Rotations seem very easy, but in fact, may become very difficult to maintain and apply due correctly to numerous not apparent exclusions.
}

\myglsentry{Points-To Heap Model}
\myglsdesc{
A \gref{dynamic memory} model consists of base \gref{assertions} of kind: "$location \times content$".
}

\myglsentry{Polymorphism}
\myglsdesc{
A dynamic effect when a \gref{class} method may on runtime be different from the declared class.
In contrast to polymorphism by subclassing in C++ and Java (so-called "\textit{ad-hoc}), some \gref{functional programming languages} provide "\textit{real polymorphism}" evaluating symbols and all-quantified variables.
}

\myglsentry{Post-Condition}
\myglsdesc{
The third component $Q$ of a \gref{Hoare triple}.
}

\myglsentry{Pre-Condition}
\myglsdesc{
The first component $P$ of some \gref{Hoare triple}.
}

\myglsentry{Predicate Cut (in Prolog)}
\myglsdesc{
A reserved built-in \gref{subgoal} "\textbf{!}" with zero arguments for cutting the \gref{proof search} on further alternatives.
If the cut is known to be safe applied to some algorithm, it is called "\textit{green cut}", otherwise, called "\textit{red cut}".
}

\myglsentry{Predicate Fold}
\myglsdesc{
A predicate contains a predicate head with a name and term arguments on the left-hand side, and a predicate's body containing its definition refers to arguments from the left-hand side.
Fold describes within some \gref{assertion} a predicate's body can be unified and replaced by its predicate head.
If the substitution is reversed, it is a \gref{predicate unfold}.
It is called fold since a predicate's body often is longer than its head, and therefore the assertion after a fold is shorter than before.
}

\myglsentry{Predicate Invertibility}
\myglsdesc{
If a Prolog predicate is known to be sound and complete for a certain term arguments placement, then that predicate is invertible when term arguments are flipped, s.t. input term become output, and the predicate will still be terminated.
As input parameters may commute and output parameters too, permutations of flipped arguments also need to be considered.
An example is ISO-Prolog compliant predicate \texttt{append/3}.
}

\myglsentry{Predicate Unfold}
\myglsdesc{
Unfold describes within some \gref{assertion} its predicate body replaces a matching predicate's head.
It is the opposite of \gref{predicate fold}.
}

\myglsentry{Primitive Recursion}
\myglsdesc{
Recursion with apriori bound number of iterations.
}

\myglsentry{(Principle of) Non-Repetitiveness}
\myglsdesc{
Simplification principle after which parts of some \gref{heap} may not spatially re-occur.
}

\myglsentry{Post's Correspondence Problem}
\myglsdesc{
Given two \gref{formal grammars}.
The question of whether the two \gref{formal languages} generated by both are in general undecidable.
}

\myglsentry{(Programming) Paradigm}
\myglsdesc{
E.g. \gref{declarative}, imperative, multi-paradigmal \cite{denti05}.
}

\myglsentry{Programming By Proving}
\myglsdesc{
(Philosophical) concept stating programming can be done by proving.
}

\myglsentry{(Proof) Abduction}
\myglsdesc{
Logical reasoning (proof) \gref{meta-scheme} based on \gref{facts} and implications.
It guesses from a conclusion a potentially "\textit{missing bit}" in the antecedent.
}

\myglsentry{Proof Abstraction}
\myglsdesc{
Parameterising some proof in order to reuse it or parts of it w.r.t. planning.
}

\myglsentry{Proof Analogy}
\myglsdesc{
During \gref{theorem} proving a \gref{meta-scheme} that \gref{abstracts} from already accomplished proofs and tries to apply by tiny modification to a current problem that seems similar to a previously accomplished.
}

\myglsentry{Proof Coherence}
\myglsdesc{
If a proof fibres according to proper proof rules, application and each fibre succeeds with the same result, then the proof rules are coherent.
}

\myglsentry{Proof Deduction}
\myglsdesc{
As logical reasoning (proof) \gref{meta-schema} as \gref{proof abduction}.
It implies from two antecedents a conclusion.
}

\myglsentry{Proof Fact}
\myglsdesc{
A logical fact is believed to be true within some discourse, particularly an axiomatic rules system.
}

\myglsentry{Proof Heuristics}
\myglsdesc{
An assumption made, often due to too high practical complexity, eases and approximates some (sub-)optimal solution, which is often sufficient, however, but not guaranteed.
In terms of formal proofs, heuristics (as one form of a \gref{meta-scheme}) do not accomplish proofs, but they may occasionally be used to transform a non-matching \gref{(calculation) state} into a matching.
}

\myglsentry{Proof Induction}
\myglsdesc{
As logical reasoning (proof) \gref{meta-schema} as a \gref{proof deduction}.
Induction is based on a generalised, often not directly provable, but not contradicting, assumption that is considered a true \gref{assertion}, unless found differently, and used in a proof.
}

\myglsentry{Proof Lemma}
\myglsdesc{
An auxiliary \gref{theorem} without a separate unique meaning often used as an intermediate step in a proof.
Lemmas may also be \gref{folded} or \gref{unfolded}.
}

\myglsentry{Proof Refutation}
\myglsdesc{
A possible negative result of a theorem-\gref{proof} indicating it is not valid.
Often some \gref{counter-example} as a single case explains why a theorem does not hold.
}

\myglsentry{Proof Search}
\myglsdesc{
The search for some positive proof for a given \gref{theorem}.
\gref{(Proof) tactics} may be applied sometimes in order to speed up the search.
}

\myglsentry{(Proof) Tactics}
\myglsdesc{
A \gref{heuristics} of a template sequence of predefined fixed verification steps applied to assist during a \gref{proof search}.
}

\myglsentry{(Proof) Theorem}
\myglsdesc{
An \gref{assertion} that can be proven using a selected set of axioms and rules.
}

\myglsentry{Proof Theory}
\myglsdesc{
A set of axioms and sets for some considered discourse in which proofs are taken out.
}

\myglsentry{Proof Tree}
\myglsdesc{
The structure of some proof. In this work, we are only interested in \index{Hoare calculus} Hoare calculi. Any proof tree is a tree whose vertices represent a \gref{verification state}.
Edges are transitions labelled with the rule applied.
All leaves in a prove tree are the result of axioms.
}

\myglsentry{Prolog Query}
\myglsdesc{
One or more \gref{subgoals} in Prolog to defined \gref{Prolog rules} separated by a comma and ending with a period.
A query may yield zero or more solutions, depending on how many solutions are derivable for the given rules.
}

\myglsentry{Prolog Rules}
\myglsdesc{
See \gref{Prolog predicates}
}

\myglsentry{Prolog Predicate}
\myglsdesc{
A Prolog notation for predicates in first-order predicate logic, which consist of head and body.
It is often called Horn-rules.
It consists of \gref{ingoing} and \gref{outgoing} terms as argument lists in the predicate head.
The body has zero or more subgoals that may refer to the predicate's arguments.
If a predicate has zero subgoals, the Prolog predicate may also be called \gref{fact}.
}

\myglsentry{Prolog Subgoal}
\myglsdesc{
See \gref{Prolog predicate}.
}

\myglsentry{Rapid Prototyping}
\myglsdesc{
A prototype whose only purpose is to document or reject if doable in practice or not.
Its main purpose is to get at least some exploratory estimation.
Often suffers from hard limitations and countless errors.
}

\myglsentry{(Recursive/Non-Recursive Term) Occurs-Check}
\myglsdesc{
Prolog does not check by default if a term contains self-applications, e.g. \texttt{X=s(1,s(2,X)}.
}

\myglsentry{Resolution Method (by Robinson)}
\myglsdesc{
In contrast to \gref{natural deduction}, resolution proves a theorem by showing the falsified theorem refutes.
The given theorem has to be in conjunctive normal-form.
After negation, a disjunctive normal-form is obtained where each disjunct/clause has to be proven false separately.
A modified resolution method is Prolog's default reasoning scheme over predicates.
}

\myglsentry{Safe Pointer Operations}
\myglsdesc{
See \gref{pointer rotation}.
Safe rotations are to distinguish from operations with flaky and hard to predict operational behaviour.
}

\myglsentry{Scope}
\myglsdesc{
A program slice in which \gref{stack-local variables} are pushed to \gref{stack} on entry and popped on exit.
}

\myglsentry{Self-Applicable Terms}
\myglsdesc{
See \gref{recursive terms}
}

\myglsentry{Self-Inverse Operation}
\myglsdesc{
Operation, when being applied twice, returns the origin argument.
Often \gref{totality} helps in showing generalised properties.
}

\myglsentry{Semantic Function}
\myglsdesc{
Defined functions for denoting program properties.
}

\myglsentry{Semi-Automated Theorem Proving}
\myglsdesc{
A proof either terminates successfully or stops unfinished, requiring further input, which might either be missing lemmas, equalities (see \gref{proof theory}).
If nothing is missing and the proof still does not finish, then the considered theorem just cannot be proven.
}

\myglsentry{Separation Of Concerns}
\myglsdesc{
Separation of roles and interfaces (e.g. into objects and procedures) by responsibilities.
}

\myglsentry{Separation Logic}
\myglsdesc{
See \gref{points-to heap model}.
}

\myglsentry{Simple Heap (=Simplex)}
\myglsdesc{
A simple assertion of form: $a \mapsto b$ in Reynolds' \gref{memory model}.
Others are \gref{complex heap}, \gref{abstract predicate}.
}

\myglsentry{ShapeAnalysis}
\myglsdesc{
A \gref{dynamic memory} model in which elements are described by its (geometric) shape.
}

\myglsentry{SMT-solver}
\myglsdesc{
An automated solution to a formula of (in-)equalities.
}

\myglsentry{Spatial Heap Operator}
\myglsdesc{
Describes a binary operator expressing how two heaps are spatial to each other in \gref{heap address space}.
}

\myglsentry{SSA}
\myglsdesc{
\gref{IR} for data dependency, which for each local variable strictly has one defining value and several references per basic block.
}

\myglsentry{Stack Allocation}
\myglsdesc{
See \gref{stack}
}

\myglsentry{Stack}
\myglsdesc{
Organised memory segment allocated by the OS.
It accommodates all \gref{local variables}.
When a procedure is called, local variables are put to the actual reserved \gref{stack window} and freed when leaving.
}

\myglsentry{Stack(-ed) Variable}
\myglsdesc{
See \gref{stack}
}

\myglsentry{Stack Unwinding}
\myglsdesc{
Restoring \gref{stack} required by exception-handling in the event of failure.
}

\myglsentry{Static Analysis}
\myglsdesc{
A code analysis performed without running a program.
}

\myglsentry{Static Variable}
\myglsdesc{
Allocation happens once on process creation.
Deallocation happens once on process release.
}

\myglsentry{Strong Post-condition}
\myglsdesc{
Strengthening some \gref{assertion} $Q'$ of post-condition $Q$ restricts the latter. So, $Q \Rightarrow Q'$ implies $Q' \subseteq Q$. 
}

\myglsentry{Structural Rules}
\myglsdesc{
In contrast to \gref{subtraction rules}, are rules focused on the second component of \gref{Hoare triples}.
}

\myglsentry{Struct-Alignment}
\myglsdesc{
Padding fields in contiguous blocks for structs or \gref{objects} in order to achieve higher access throughput.
}

\myglsentry{Subtraction Rules}
\myglsdesc{
Once applied, rules reduce the remaining proof for a given \gref{theorem} (cf. with \gref{structural rules}).
}

\myglsentry{Subtyping}
\myglsdesc{
A class subtype is any inherited \gref{class} according to its class hierarchy.
}

\myglsentry{Symptom (of an Error)}
\myglsdesc{
See \gref{erroneous code localization}.
}

\myglsentry{Synthesized Attribute}
\myglsdesc{
An attribute in an \gref{attributed grammar} is passed from the called rule up to the calling rule.
See also \gref{inherited attribute}.
}

\myglsentry{Tableaux Method (by Beth)}
\myglsdesc{
Logical reasoning where all derivations may only be applied that are encoded in a given table.
}

\myglsentry{Token}
\myglsdesc{
An atomic unit resulting from lexical analysis.
For example, the assignment symbol \texttt{:=}.
}

\myglsentry{Total Function}
\myglsdesc{
A function that accepts any element from a well-defined domain and terminates for each.
}

\myglsentry{Trieber's Stack}
\myglsdesc{
A \gref{(execution) stack} model supported by \gref{separation logic} in multi-threading environments.
}

\myglsentry{Type Checking}
\myglsdesc{
Checks for a given program all assignments and expressions are valid w.r.t. applied variable types according to its declaration.
}

\myglsentry{Undermined Specification}
\myglsdesc{
A specification containing unbound symbols, an abstract specification.
}

\myglsentry{Variable Mode}
\myglsdesc{
A variable's mode depends on its life-span and its scope of visibility.
Modi might be \gref{automated}, \gref{static}, \gref{dynamic}, \gref{global}, register, thread-local or others.
}

\myglsentry{Warren's Abstract Machine}
\myglsdesc{
In addition to traditional operational semantics WAM defines an abstract machine based on an extended execution stack, where elements are (Prolog) terms with forward and backward references allowed towards (sub-)terms across several stack windows.
Unification makes sure linked terms are interpreted as parameters by reference.
}

\myglsentry{Weak Pre-condition}
\myglsdesc{
Weakening some pre-condition $P$ yields $P'$, s.t. $P'$ is a generalisation of $P$, $P' \Rightarrow P$ implies $P \subseteq P'$.
}

  \glsaddall
  \printglossary[title={Glossary}]

  %  % -- References --
  %%%%%%%%%%%%%%%%%%%%%%%%%%%%%%%%%%%%%%%%%%%%%%%%%%%%%%%%%%%%%%%%%%%%%%%%%%%%%

%  originally from: \addbibresource{references.bib}  % english

%\renewcommand{\refname}{Англо-язычная литература}

  %% (russian references shall not be translated)
  
  \newpage
  \listoftheorems
  \addcontentsline{toc}{section}{List of theorems}

  \listoffigures
  \addcontentsline{toc}{section}{List of figures}

  % -- Index --
  \addcontentsline{toc}{section}{Index}
  \printindex
%%%%%%%%%%%%%%%%%%%%%%%%%%%

\end{document}